\documentclass[11pt,a4paper,openany]{book}

\usepackage{jiasheng-setup-envi}
\usepackage{jiasheng-commands}

\numberwithin{thrm}{chapter}

\hypersetup{citecolor=teal}

\title{Constructive Quantum Field Theory \\ on Curved Surfaces and Related Topics}
\author{Jiasheng Lin\footnote{Sorbonne Universit\'e, email: \href{mailto:jiasheng.lin@imj-prg.fr}{jiasheng.lin@imj-prg.fr}}}
\date{}

\begin{document}

\begin{titlepage}
{\sffamily
    \begin{center}
\includegraphics[scale=0.5]{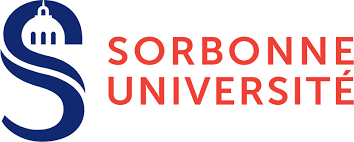}
\end{center}

\begin{center}

\vspace{\stretch{3}}

\hrule
\vspace{\stretch{1}}
{\LARGE \textbf{Constructive Quantum Field Theory}}
\medbreak 

{\Large \textbf{on Curved Surfaces and Related Topics}}

\vspace{\stretch{1}}
\hrule
\vspace{\stretch{2}}

{\large Doctoral Thesis of}
\vspace{\stretch{1}}

\textbf{{\LARGE Jiasheng \textsc{Lin}}}

\vspace{\stretch{2}}

{\Large 
July 2025\\
Sorbonne Universit\'e}

\end{center}


\newpage
\thispagestyle{empty}
\vspace*{\fill}

\noindent\begin{center}
\begin{minipage}[t]{(\textwidth-2cm)/2}
Institut de mathématiques de Jussieu-Paris Rive gauche. UMR 7586. \\
Boîte courrier 247 \\
4 place Jussieu \\
\texttt{75252} Paris Cedex 05
\end{minipage}
\hspace{1.5cm}
\begin{minipage}[t]{(\textwidth-2cm)/2}
 Université de Paris. \\
 École doctorale de sciences\\
 mathématiques de Paris centre.\\
  Boîte courrier 290 \\
 4 place Jussieu\\
 \texttt{75252} Paris Cedex 05
 \end{minipage}
\end{center}
}
\end{titlepage}

\frontmatter

\thispagestyle{empty}
\section*{Résumé}

Dans cette thèse, nous construisons le modèle de Théorie Quantique des Champs (TQC) $P(\phi)_2$ sur des surfaces courbes et montrons qu’il satisfait aux axiomes de Segal. Un ingrédient important de cette construction est l’utilisation d’une procédure de régularisation locale pour définir l’interaction comme une variable aléatoire par rapport au Champ Libre Gaussien (GFF). Nous fournissons un contre-exemple montrant que la régularisation par troncation spectrale viole la localité. Ensuite, nous expliquons en quoi le formalisme de Segal peut être étendu au collage de surfaces avec des coupures, ce qui offre une interprétation géométrique de l’entropie d’intrication. En utilisant cette interprétation, nous exploitons la formule d’anomalie de Polyakov en théorie des champs conformes (CFT) et appliquons une procédure de renormalisation simple pour définir une quantité correspondant à l’entropie d’intrication dans cette interprétation géométrique. Nous montrons alors que cette quantité se comporte comme une fonction de corrélation de CFT. Ceci nous permet de dériver de façon rigoureuse un calcul d'entropie de Cardy et Calabrese. Enfin, le formalisme de Segal est également lié à l'asymptotique de déterminants zeta sur des surfaces de grand genre, où le genre tends vers l'infini. Nous fournissons une démonstration géométrique, indépendante des axiomes de Segal, du résultat correspondant à l’aide des noyaux de la chaleur, en complément d'une autre preuve reposant sur les axiomes de Segal. Les deux preuves apparaissent dans la thèse.

\section*{Abstract}
In this thesis, we construct the $P(\phi)_2$ Quantum Field Theory (QFT) model on curved surfaces and show that it satisfies Segal's axioms. An important ingredient in this construction is the use of a local regularization procedure to define the interaction as a random variable with respect to the Gaussian Free Field (GFF). We provide a counterexample demonstrating that spectral truncation regularization violates locality. We then explain how Segal’s formalism can be extended to the gluing of surfaces with slits, which offers a geometric interpretation of the entanglement entropy. Using this interpretation, we exploit the Polyakov anomaly formula in Conformal Field Theory (CFT) and apply a simple renormalization procedure to define a quantity corresponding to entanglement entropy within this geometric interpretation. We then show that this quantity behaves like a CFT correlation function. This allows us to rigorously derive an entropy calculation of Cardy and Calabrese. Finally, Segal's formalism is also related to the asymptotics of zeta determinants on surfaces of large genus where the genus tends to infinity. We provide a geometric proof—independent of Segal's axioms—of the corresponding result using heat kernels, in addition to another proof based on Segal's axioms. Both proofs are presented in the thesis.

\clearpage
\null
\thispagestyle{empty}
\addtocounter{page}{-1}
\newpage

\chapter*{Acknowledgements}
\addcontentsline{toc}{chapter}{Acknowledgements}

First and foremost, I would like to express my deepest gratitude to my thesis advisor, Nguyen Viet Dang, for guiding me through this incredible journey. Without his support, I cannot imagine having come so far in the vast and challenging field of mathematical Quantum Field Theory. Beyond the immense mathematical knowledge I gained from him, I was truly inspired by his remarkable enthusiasm for scientific research and his unwavering determination in tackling difficult problems in mathematics and mathematical physics. More than just a caring, responsible, and encouraging advisor, he has gradually become a friend and a colleague with whom I enjoy collaborating. I feel truly fortunate to have had him as my advisor.

I would also like to sincerely thank my friend and teacher, Benoit Estienne, with whom I completed a significant portion of this thesis. Most of my understanding of Conformal Field Theory comes from him—either directly or through the excellent lecture notes he co-authored with Yacine Ikhlef. I am deeply grateful for the many insightful discussions on mathematics and physics, as well as the numerous lunches and coffee breaks we shared.

My heartfelt thanks go to Tat Dat Tô, whom I have known since the beginning of my PhD. I thank him for being one of my tutors in the \textit{\'ecole doctorale}, for his encouragements during difficult moments in my research, let alone many fruitful discussions on mathematics. I also extend my gratitude to Colin Guillarmou, whom I met around the same time, and from whom I also learned much of CFT, both directly and indirectly. I deeply appreciate his warmth and willingness to engage with young researchers.

I thank again Colin, together with Ilya Chevyrev, for accepting to be referees of this thesis. I extend my gratitude to Nathana\"el Berestycki, Frédéric Hélein, Pavel Mnev, and Eveliina Peltola for accepting to be on the jury of my defense.

I sincerely thank my collaborators Frédéric Naud, Gaëtan Leclerc, Ismael Bailleul, and Léonard Ferdinand, alongside Viet and Benoit, for also having taught me a lot and having helped me accomplish these valuable projects.

I am grateful to Fathi Ben Aribi, Laurent Charles, Gilles Courtois, Elisha Falbel, Thibault Lefeuvre, Vlerë Mehmeti, Frédéric Naud, Bram Petri, and Shu Shen, as well as many other colleagues in the laboratory, for their enriching discussions on mathematics and culture. A special thanks to Elisha Falbel for being my second tutor. Special thanks to Laurent Charles, Bram Petri, Shu Shen and others for scientifically enriching many of my lunches at RU. It is also my pleasure to thank the wonderful researchers Mikhail Basok, Reda Chhaibi, Nguyen-Bac Dang, Yuxin Ge, Yannick Guedes Bonthonneau, Fei Han, Frédéric Hélein, Luc Hillairet, Paul Laurain, Thierry Lévy, Semyon Klevtsov, Phan Thành Nam, Eveliina Peltola, Gabriel Rivière, Chenmin Sun, Fredrik Viklund, Fabien Vignes-Tourneret, Michał Wrochna, Junrong Yan, as well as many others who have shared their insights with me on several occasions and generously offered their wisdom.

My PhD journey would not have been the same without the camaraderie of my friends and fellow PhD students Xi Chen, Zhe Chen, Léonard Ferdinand, Pierre Godfard, Hedong Hou, Tristan Humbert, Gaëtan Leclerc, Yuxiang Li, Long Liu, Pietro Piccione, Yuan Tian, Tianqi Wang, Zuodong Wang, Yilin Ye, Lyuhui Wu, Sacha Zakharov, Yicheng Zhou and many others for lots of fun and enjoyment with math and other things. I sincerely thank Yilin Wang for playing the special role of a supportive family member as well as an inspiring colleague.

I must also thank our secretary at IMJ, Kahina Bencheikh, for undertaking the incredibly challenging French administrative tasks during my PhD years. I thank also the secretary of \textit{\'ecole doctorale}, Corentin Lacombe, for managing issues related to the defense.

Finally, my warmest gratitude to Yuetong and my parents---for everything else in life.

\clearpage
\null
\thispagestyle{empty}
\addtocounter{page}{-1}
\newpage

\tableofcontents

\chapter*{Résumé Détaillé en Français}
\addcontentsline{toc}{chapter}{Résumé Détaillé en Français}

\counterwithout{equation}{section}

La théorie quantique des champs (TQC dans la suite) est un cadre développé en physique théorique pour décrire et expliquer les phénomènes fondamentaux qui gouvernent la physique à l’échelle nucléaire ou particulaire.
C'est l'un des plus grands succès de la physique du 20ième siècle 
dont
les résultats théoriques
concordent aussi bien que possible avec les résultats expérimentaux -
jusqu’à 14
chiffres significatifs en électrodynamique quantique -
ce qui est un exploit inégalé par la physique
antérieure. Cependant, on peut dire que les fondements mathématiques de la TQC ne sont pas tout à fait définitifs. Il n'existe pas encore de cadre mathématique assez général et riche qui permettrait de décrire de façon non perturbative toutes les TQC qui décrivent des phénomènes physiques observés, comme le modèle standard qui est une théorie de jauge sur l'espace--temps de dimension $4$ (ceci fait l'objet de l'un des problèmes du Millénium). 

 Malgré les difficultés, depuis les années 60, de nombreux cadres ont été proposés pour axiomatiser la théorie quantique des champs. L'une des axiomatiques les plus célèbres est due à G\"arding et Wightman ~:
 on se donne un espace de Hilbert $\mathcal{H}$, un vecteur distingué de $\mathcal{H}$ appelé le vecteur vide $\Omega_0\in \mathcal{H}$, un opérateur sur $\mathcal{H}$ à valeur distributions $\Phi$ ainsi qu'une repésentation unitaire $U$ du groupe de Poincaré $\mathcal{P}_+^{\uparrow}$ satisfaisant une certaine liste d'axiomes. Par exemple, l'une des règles les plus importantes est la covariance:
 $$ U(a,\Lambda) \Phi(x) U(a,\Lambda)^{-1}=\Phi(\Lambda x+a) .$$
Suite aux travaux novateurs de Symanzik et Nelson, le sujet a été révolutionné par les travaux profonds de Osterwalder--Schrader dans les années 70. Ces auteurs ont donné un r\^ole fondamental à l'approche probabiliste reposant sur l'intégrale fonctionnelle Euclidienne. On construit une mesure de probabilité sur l'espace de configuration des champs et ils ont décrit une méthode permettant de reconstruire la théorie quantique des champs relativiste à partir d'axiomes venant du monde Euclidien.

Plus précisément, pour les bosons scalaires réels, l'espace des configurations est celui des distributions à valeur réelles, et la mesure $\mu$ est définie et reliée aux objets de G\"arding-Wightman (pour $t_1<\dots<t_n $) par
\begin{align}
	  \mb{E}_{\phi\in \mathcal{D}'(\mb{R}^4)}^{\mu}\big[\phi(\mn{x}_1,t_1)\cdots \phi(\mn{x}_n,t_n)\big] &:\heueq\frac{\int_{\mathcal{D}'(\mb{R}^4)}^{}\phi(\mn{x}_1,t_1)\cdots \phi(\mn{x}_n,t_n)\me^{-S(\phi)}[\dd \mathcal{L}(\phi)]}{\int_{\mathcal{D}'(\mb{R}^4)}^{}\me^{-S(\phi)}[\dd \mathcal{L}(\phi)]} \label{eqn-corr-func-fr}\\
	  &\heueq \bank{\Omega_0,\phi_{\mm{GW}}(-\ii t_n,\mn{x}_n)\cdots \phi_{\mm{GW}}(-\ii t_1,\mn{x}_1)\Omega_0}_{\mathcal{H}}.
	  \label{eqn-corr-wightman-fr}
	\end{align}
Ici, $[\dd \mathcal{L}(\phi)]$ est la mesure de Lebesgue formelle (il n'existe pas de telle mesure invariante par translation en dimension infinie) sur l'espace des champs et $S$ est la fonctionnelle d'action. Autrement dit, les fonctions de corrélation Euclidiennes sont obtenues à partir des fonctions de corrélation relativistes (Minkowskiennes) correspondantes en rendant le temps comme imaginaire.

Une des propriétés centrale permettant une la reconstruction de la théorie relativiste est la propriété de réflection positivité.
 Cette formulation a permis la construction rigoureuse du premier exemple d'une théorie quantique des champs en interaction en dimension $2$. En combinant avec les travaux antérieurs de Nelson \cite{GJ, Sim2}, Glimm--Jaffe sont parvenus à construire une mesure de probabilité chargeant des distributions de $\mathcal{D}^\prime(\Lambda)$ pour $\Lambda$ une boite de volume fini dans $\mathbb{R}^2$.
Cette construction correspond au célèbre modèle $P(\varphi)_2$. Notons qu'a priori, une distribution aléatoire tirée au hasard suivant la loi du GFF ou du $P(\varphi)_2$ a une très faible régularité, c'est une distribution qui vit dans les espaces de Sobolev strictements négatifs et ni les puissances $\varphi^n$ de $\varphi$ ni la densité de fonctionnelle $\vert\nabla \phi\vert^2  $ ne sont bien définis. L'observation clé est d'isoler la partie libre (qaudratique) de l'action, au lieu de prendre comme mesure de référence la mesure de Lebesgue en dimension infinie qui n'existe pas, on prend comme mesure de référence la mesure du champs libre Gaussien appelée aussi GFF.
La partie non quadratique de la fonctionnelle d'action devient ainsi intégrable contre la mesure gaussienne après une certaine procédure limite de soustraction de contre-termes divergents appelée \textbf{renormalisation} (ici en $2$--diemensions une simple renormalisation de Wick suffit). Récemment gr\^ace aux travaux révolutionnaires de Hairer et Gubinelli, une approche alternative à la construction de la mesure $P(\phi)_2$, différente de celle de Nelson, est apparue, en considérant $\mu$ comme la mesure invariante d'une EDP stochastique de Langevin et cette méthode a permis de donner de nouvelles constructions élégantes des mesures $\Phi^4_3$.

Dans les années 90, une approche différente d'axiomatisation a émergé, inspirée par la théorie des cordes et motivée par des considérations algébro-topologiques \cite{Segal}. Cette approche, initiée par Atiyah et Segal, se concentre sur les dimensions $2=1+1$ mais inclut naturellement comme espaces-temps toutes les surfaces courbes compactes, Lorentziennes ou Riemanniennes, avec ou sans bord. Rappelons que dans la théorie de G\"arding-Wightman, les translations en temps faisaient partie des symétries du groupe de Poincaré, et que construire une théorie des champs consistait en partie à construire une représentation unitaire du groupe de Poincar\'e incluant ces translations. Dans le cadre d'Atiyah-Segal, cependant, l'évolution temporelle d'une « tranche » est représentée par la portion de l'espace-temps elle-même qui relie la tranche à l'entrée et à la sortie (un cobordisme), et qui a priori peut ne posséder aucune symétrie. Néanmoins, les évolutions temporelles doivent se « composer » comme dans une représentation lorsque les espaces-temps sont collés. Le résultat est qu'une théorie quantique des champs (2D) se comprend comme un foncteur de la catégorie des cobordismes (métrisés) vers la catégorie des espaces de Hilbert où les morphismes correspondent à des opérateurs de Hilbert--Schmidt.

Le point de vue d'Atiyah-Segal présente l'avantage d'être étroitement lié à de nombreuses structures intéressantes en algèbre, géométrie et topologie. De plus, la métrique sur l'espace-temps peut varier et joue un rôle spécial. En particulier, lorsque les quantités physiques comme (\ref{eqn-corr-func-fr}) se transforment de manière covariante sous des changements conformes de la métrique, on obtient une théorie conforme des champs (CFT), dotée de riches structures du point de vue de la théorie des représentations. Cependant, contrairement au cas de Nelson-Osterwalder-Schrader, très peu d'exemples de TQC satisfaisant les axiomes ont été explicitement et rigoureusement construits \cite{GKR, KMW, Lin}, encore moins si l'on vise une mesure fonctionnelle $\mu$ définie de façon non perturbative, comme dans les travaux de Glimm--Jaffe, Nelson. En effet, de nombreuses techniques développées traditionnellement en théorie constructives reposent sur le fait que la métrique est plate et que tout est invariant par translation, et
ce n'est pas forcément clair comment de telles méthodes pourraient \^etre adaptées aux métriques courbes ou au découpage collage de l'espace ambiant. Un défi général dans la poursuite du programme constructif réside donc dans la fusion des techniques d'analyse fonctionnelles avec le cadre d'Atiyah-Segal, et dans la déduction, par la suite, de phénoménologies intéressantes (comme les transitions de phase) ou de conséquences mathématiques à partir de toute construction particulière.

\section*{Description des Contributions Principales  de la thèse}
\addcontentsline{toc}{section}{Description des Contributions Principales  de la thèse}

\paragraph{Axiomes de Segal pour la théorie $P(\varphi)_2$ et Conséquences.}
Dans la prépublication \cite{Lin}, j'obtiens une théorie quantique des champs $P(\phi)_2$ en interaction sur une surface riemannienne périodique de volume et genre infini. 
Ceci résulte précisément de la fusion des points de vue d'Atiyah-Segal et de Nelson-Osterwalder-Schrader, comme discuté ci-dessus. De plus, je montre que cette TQC possède un gap de masse, un exemple de phénoménologie découlant de la construction, ce qui implique la décroissance exponentielle des corrélations. 

Comme sous-produit, on retrouve également des résultats connus sur l'asymptotique de certains déterminants zêta de surfaces lorsque le genre tend vers l'infini, et le résultat est généralisé aux fonctions de partition de $P(\phi)_2$. Cela éclaire d'un jour nouveau ce type de résultats en analyse globale.

L'argument principal repose sur la vérification des axiomes de Segal à partir de la construction classique de Nelson des mesures fonctionnelles $P(\phi)_2$, qui s'appuie elle-même sur les mesures du champ libre gaussien. En fait, une difficulté majeure est déjà présente pour démontrer l'axiome de composition de Segal pour le champ libre, et je propose une nouvelle approche basée sur une idée conceptuelle simple impliquant une symétrie dans le conditionnement successif de variables aléatoires. Une seconde difficulté réside dans la définition elle-même des amplitudes de Segal en interaction pour une surface riemannienne avec bord. Un ingrédient clé pour y parvenir est de montrer que la fonctionnelle d'interaction est locale, pour laquelle aucune analyse suffisamment précise et explicite n'existait dans la littérature. Ce chapitre fournit un résultat rigoureux basé sur un renforcement de l'argument de Nelson et montre que la renormalisation de Wick peut être rendue compatible avec la localité.

\paragraph{Vérification Explicite que le GFF Régularisé n'est pas Réflexion positif.}
Un ingrédient clé dans la renormalisation d'une TQC du point de vue des mesures fonctionnelles est de régulariser le champ, qui est généralement une distribution de très faible régularité imposée par le support de la mesure. Cependant, la plupart des méthodes traditionnelles de régularisation (par exemple, la troncature dans l'espace de Fourier) ne sont pas compatibles avec la propriété importante de localité, car les supports des distributions deviennent fortement délocalisés. Dans cette courte note \cite{BDFL2}, nous montrons que l'exigence que certaines fonctionnelles du GFF spectralement tronqué soient mesurables par rapport à certaines $\sigma$-algèbres locales les force à être constantes, et nous construisons des fonctionnelles tests explicites montrant que le GFF spectralement régularisé n'est pas réflexion positif. Ces deux résultats sont des conséquences du fait que la régularisation spectrale est non locale.

\paragraph{Entropie d'intrication par Renormalisation de Hadamard.}
Dans ce travail en collaboration avec B. Estienne \cite{BL}, nous présentons une construction purement mathématique qui retrouve certains résultats de Cardy et Calabrese \cite{Cardy_Calabrese} sur l'entropie d'intrication. Physiquement, il s'agit d'une quantité décrivant comment des sous-systèmes (typiquement localisés dans une région spatiale) d'un système quantique s'influencent mutuellement. Pour les systèmes de dimension $1$ et en supposant qu'on puisse les décrire par une CFT en deux dimensions d'espace-temps, cela correspond mathématiquement à des revêtements ramifiés de surfaces de Riemann munies de métriques à singularités coniques. Dans ce travail, nous définissons d'abord la fonction de partition de la CFT sur une surface avec singularité conique en utilisant une simple renormalisation de Hadamard de la formule d'anomalie de Polyakov. Ensuite, le résultat principal est que pour un revêtement ramifié de degré $d$, on peut étudier comment un certain rapport des fonctions de partition se transforme par changements conformes de la métrique. On montre que le comportement ressemble à une fonction à deux points d'opérateurs primaires d'une CFT avec des poids conformes spécifiques. C'est le résultat clé qui conduit à des expressions explicites de l'entropie d'intrication retrouvant la formule de Cardy et Calabrese.

Dans un futur travail, nous prévoyons de construire un champs libre à valeurs dans un fibré sur la surface de base munie d'une connexion plate singulière. Nous montrons que ce modèle vivant sur la base permet de décrire de façon équivalente la théorie conforme vivant sur le rev\^etement ramifié, cette construction donnerait une justification rigoureuse des champs de twist apparaissant en physique théorique.
 Mathématiquement, il s'agit d'une version singulière des « laplaciens tordus » utilisés par exemple par Phillips et Sarnak \cite{PSa} pour compter les géodésiques fermées avec contraintes homologiques.

\paragraph{Déterminants Zêta des Laplaciens sur les Revêtements Cycliques.}
Dans ce travail avec Nguyen-Viet Dang et Frédéric Naud, nous obtenons l'asymptotique des déterminants zêta sur les revêtements abéliens d'une variété riemannienne fermée lorsque le degré du revêtement tend vers l'infini. Plus précisément, si $ (M_N\mapsto M)_N $ est une suite croissante de revêtements cycliques d'une variété riemannienne fermée $M$ qui converge quand $N\lto \infty$, dans un sens approprié, vers un certain revêtement  limite $M_\infty$ sur $M$, nous montrons la convergence de la suite 
$ \log\det_\zeta(\Delta_{N})/\text{Vol}(M_N)$ quand $N\lto \infty$
où $\Delta_N$ est l'opérateur de Laplace-Beltrami sur $M_N$. Pour ce faire, nous utilisons des résultats sur l'asymptotique du noyau de la chaleur combinées aux estimations de Li-Yau, ainsi qu'une analyse des « laplaciens tordus » développée par Phillips et Sarnak \cite{PSa}.

\counterwithin{equation}{section}

\mainmatter

\chapter{Introduction}

\section{Brief General Introduction}

Quantum Field Theory (QFT) is a theoretical framework developed by theoretical physicists for describing and analyzing the elementary particles and their fundamental interactions. It has been highly successful in standards of physics and is widely held to be the most accurate theory of nature in terms of producing experimental predictions. However, the mathematical foundation of QFT is not yet definitive, and there exists, so far, no mathematical framework which is both general and rich enough to include all existing models of QFT, particularly the Gauge theory in 4 dimensions (a Millennium problem). 

Despite the huge difficulty, there has been numerous attempts since 1960s at axiomatics. One earliest and most widely accepted axiom system proposed by \textbf{G\aa rding} and \textbf{Wightman} defines a QFT (on the Minkowski space $\mb{R}^{1,3}$) as the data of a Hilbert space $\mathcal{H}$, a distinguished vector~$\Omega_0\in \mathcal{H}$ called the \textsf{vacuum}, an~$\mathcal{H}$-acting operator-valued distribution~$\phi$ on~$\mb{R}^{1,3}$, and a unitary representation~$U$ of the proper Poincar\'e group~$\mathcal{P}_+^{\uparrow}$, satisfying a certain set of rules \cite{GJ, Sim2}. One example is \textsf{Poincar\'e covariance}, that is, for each~$(a,\Lambda)\in \mathcal{P}_+^{\uparrow}$,~$a\in\mb{R}^{1,3}$ and~$\Lambda\in\so^{\uparrow}(1,3)$, we have formally~$U(a,\Lambda)\phi(x)U(a,\Lambda)^{-1}=\phi(\Lambda x+a)$. A revolution came when \textbf{Osterwalder} and \textbf{Schrader}, exploiting the so-called \textsf{Euclidean path integral} approach inspired by \textbf{Nelson} and \textbf{Symanzik}, obtained another set of axioms which emphasized a probability measure~$\mu$ defined on the \textsf{field configuration space}, and proved the \textsf{reconstruction theorem} allowing one to go back to the former system. 
More precisely, for \textsf{real scalar bosons} the configuration space is the real distributions~$\mathcal{D}'(\mb{R}^4)$, and the measure~$\mu$ is defined and related to the G\aa rding-Wightman objects (provided $t_1<\cdots <t_n$) by 
	\begin{align}
	  \mb{E}_{\phi\in \mathcal{D}'(\mb{R}^4)}^{\mu}\big[\phi(\mn{x}_1,t_1)\cdots \phi(\mn{x}_n,t_n)\big] &:\heueq\frac{\int_{\mathcal{D}'(\mb{R}^4)}^{}\phi(\mn{x}_1,t_1)\cdots \phi(\mn{x}_n,t_n)\me^{-S(\phi)}[\dd \mathcal{L}(\phi)]}{\int_{\mathcal{D}'(\mb{R}^4)}^{}\me^{-S(\phi)}[\dd \mathcal{L}(\phi)]} \label{eqn-corr-func}\\
	  &\heueq \bank{\Omega_0,\phi_{\mm{GW}}(-\ii t_n,\mn{x}_n)\cdots \phi_{\mm{GW}}(-\ii t_1,\mn{x}_1)\Omega_0}_{\mathcal{H}}.
	  \label{eqn-corr-wightman}
	\end{align}
Here~$\dd \mathcal{L}$ is the non-existent Lebesgue measure on~$\mathcal{D}'(\mb{R}^4)$ and~$S$ is the \textsf{action functional}. In other words, the \textsf{Euclidean correlation functions} are obtained from the corresponding \textsf{relativistic (Minkowskian) correlation functions} by setting time to be imaginary.

Importantly, one may also ask about the reverse direction: starting from a probability measure $\mu$ on $\mathcal{D}'(\mb{R}^d)$ (for a theory in $d$ dimensions), can one recover from it the G\aa rding-Wightman objects? A crucial axiom on the measure $\mu$ which ensures this can be done rigorously is \textsf{reflection positivity}. This formulation allowed the first example of an \textit{interacting} quantum field theory to be constructed rigorously for $d=2$.
It included the case~$S(\phi):=\int_{\Lambda}^{}\frac{1}{2}|\nabla \phi(x)|^2+\frac{1}{2}m^2 \phi(x)^2+\phi(x)^4\,\dd x$ which corresponds to the celebrated $\Phi^4_2$ model. Note that a priori a distribution~$\phi$ has very low regularity and has neither gradient squared~$|\nabla\phi|^2$ nor products~$\phi^k$. The key observation of Nelson was that, after singling out the \textsf{free field action}~$S_0(\phi):=\int_{\Lambda}^{}\frac{1}{2}|\nabla \phi(x)|^2+\frac{1}{2}m^2 \phi(x)^2\,\dd x$ so that $\exp(-S_0)[\dd \mathcal{L}(\phi)]$ is the well-defined \textsf{Gaussian Free Field} (GFF) measure,
the functional~$S_{\Lambda}(\phi)=\int_{\Lambda}^{}\phi(x)^4\,\dd x$ becomes integrable against the Gaussian measure after a certain limiting procedure with subtraction of diverging counterterms (\textsf{Wick renormalization}). Recently an alternative approach to the measure $\mu$ different from Nelson's has appeared which views $\mu$ as the invariant measure of a \textsf{Langevin Stochastic PDE} \cite{BDFT}, particularly effective in 3 dimensions.

Then in the 1990s there arose a different approach to axiomatization from the algebraic-topological point of view, partly motivated by string theory \cite{Segal}. This approach, initiated by \textbf{Atiyah} and \textbf{Segal}, focuses on~$2=1+1$ dimensions but naturally includes as space-times all compact curved surfaces, Lorentzian or Riemannian, with or without boundary. Recall that in the G\aa rding-Wightman theory, time translations were part of the Poincar\'e symmetries and having a unitary representation meant~$U(t_1+t_2,\one)=U(t_2,\one)U(t_1,\one)$. In the Atiyah-Segal picture, however, time evolution from an initial to a terminal ``time-slice'' is represented by the piece of the space-time itself which connects them (a \textsf{cobordism}), and which a priori may not possess any symmetry at all. Nevertheless time evolutions must ``compose'' as in a representation when the space-times get glued. The \textbf{upshot} is that \textit{a (2d) Quantum Field Theory} could then be nothing but a functor from the (metrized)~$(1,2)$-bordism category to the category of Hilbert spaces and special operators.

\begin{figure}[h]
    \centering
    \includegraphics[width=0.8\linewidth]{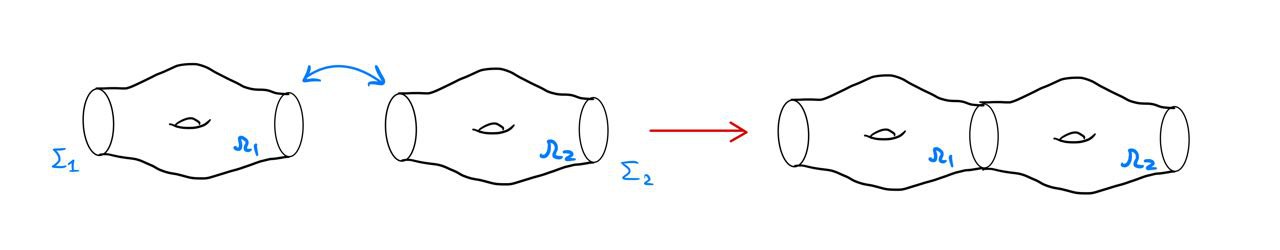}
    \caption{Segal Gluing}
    \label{fig-segal-glu}
\end{figure}

The Atiyah-Segal point of view has the advantage of being closely connected to many interesting structures in algebra, geometry and topology. Also the metric on the space-time is allowed to vary and plays a special role (hence making quantum gravity amenable). Especially when the physical quantities like (\ref{eqn-corr-func}) transform covariantly under conformal changes of the metric, one obtains a \textsf{Conformal Field Theory} (CFT), having rich representation theoretical structures. However unlike the case with Nelson-Osterwalder-Schrader, very few examples have been explicitly and rigorously constructed \cite{GKR, KMW, Lin}, even less if we aim ``non-perturbatively'' for a functional measure~$\mu$ as with Nelson. Indeed, many of the techniques developed traditionally in CQFT have been ``inherently Euclidean'' and it remains unclear how they could be adapted to curved metrics and further to cutting-gluing of space-times. One general challenge in continuing the constructive program lies henceforth in merging the functional measure techniques with the Atiyah-Segal picture, and in deducing, thereafter, interesting phenomenology (such as phase transition) or mathematical consequences out of any particular construction.

\subsection{Description of Main Contributions}

\paragraph{Segal Axioms for $P(\phi)_2$ and Consequences (Chapter \ref{chap-segal}).}
In the preprint \cite{Lin} I obtain an interacting $P(\phi)_2$ Quantum Field Theory on a curved Riemannian surface with infinite volume and genera. This is precisely a product of fusing the Atiyah-Segal and the Nelson-Osterwalder-Schrader points of view as discussed above. Moreover, I show that this QFT possesses a \textit{mass gap}, an example of phenomenology that comes out of the construction. As a by-product, one also re-obtains previously known results on the asymptotic of certain zeta-determinants of surfaces when the genus of the surface goes to infinity, and the result is generalized to QFT partition functions. This sheds new light on this type of results of pure geometry.

The main argument lies in verifying Segal's axioms starting from Nelson's classical construction of the $P(\phi)_2$ functional measures, which in turn builds upon the GFF measures. In fact a major difficulty is already present in showing Segal's composition axiom for the free field and I offer a novel approach based on a simple conceptual idea involving a symmetry in the successive conditioning of random variables. A second difficulty lies as a matter of fact in \textit{defining}, in the first place, the $P(\phi)_2$-interacting \textit{Segal amplitudes} for a Riemannian surface with boundary. A key ingredient in doing so is to show that the interaction functional is \textit{local}, of which no sufficiently precise and explicit analysis has been given in the literature. This chapter gives a precise and rigorous result based on a strengthening of Nelson's argument and shows that Wick renormalization can be made compatible with locality.

\paragraph{Explicit Verification that Regularized GFF is not Reflection Positive (chapter \ref{chap-cut-off}).}
A key ingredient in the renormalization of a QFT from the functional measure point of view is to \textit{regularize} a typical field configuration which generally has very low regularity as enforced by the support of the functional measure. However, most traditional methods for regularization (e.g.\ Fourier cut-off) are not compatible with the important property of \textit{locality} in that supports of distributions become tremendously de-localized. In this short note \cite{BDFL2} with many co-authors we argue that requiring certain functionals of the spectrally cut-off GFF to be measurable against some local $\sigma$-algebras forces them to be constant, and we construct explicit test-functionals showing that the spectrally cut-off GFF is not reflection positive. Both of these are consequences of spectral cut-off regularization being non-local.

\paragraph{Entanglement Entropy from Hadamard Renormalization (Chapter \ref{chap-entangle}).}
In this work in collaboration with B.\ Estienne \cite{BL} we demonstrate a purely mathematical construction which recovers some results of Cardy and Calabrese \cite{Cardy_Calabrese} on the so-called \textsf{entanglement entropy}. Physically this is a quantity which describes how subsystems (typically localized in a space region) of a quantum system influences each other. For systems of one spatial dimension and assuming the so-called ``replica trick'' that views it under a CFT in two space-time dimensions, it corresponds mathematically to branched covers of Riemann surfaces and metrics on such covers with conical singularities. In this work we first define the CFT ``partition function'' (denoted $\mathcal{Z}$, cf.\ denominator of (\ref{eqn-corr-func})) on a surface with conical singularity using a simple Hadamard renormalization (removing disks) of the Polyakov anomaly integral. Then the \textbf{main result} is that for a branched cover $f:\Sigma_d\lto \Sigma$ of degree $d$, the ratio $\mathcal{Z}(\Sigma_d,f^*g)/\mathcal{Z}(\Sigma,g)^d$ of partition functions transforms under conformal changes of $g$ like a \textit{two-point function of CFT primary operators} of specific conformal weights. This is the key result that leads to explicit expressions of the entanglement entropy from the work of Cardy and Calabrese.

  In a future work, we plan to construct a bundle-valued free field on the base surface with the critical values of $f$ removed, whose action is given by a canonical flat connection on the bundle. We show that this is an equivalent way of viewing the model on $\Sigma_d$ that gives the partition function $\mathcal{Z}(\Sigma_d,f^*g)$ (called ``twist fields'' in physics). Mathematically this is a singular version of the ``twisted Laplacians'' used e.g.\ by Phillips and Sarnak \cite{PSa} for counting closed geodesics with homological constraints.

\paragraph{Zeta Determinants of Laplacians on Cyclic Covers (Chapter \ref{chap-covers}).} In this work with Nguyen-Viet Dang and Fr\'ed\'eric Naud we obtain asymptotic of zeta-determinants on Abelian covers of a closed Riemannian manifold when the degree of the cover goes to infinity.
More precisely if $(M_N\lto  M)_N$ is an increasing sequence of cyclic covers of a closed Riemannian manifold $M$ that converges when $N\lto \infty$, in a suitable sense, to some limit $\mathbb{Z}$-cover $M_\infty$ over $M$,
we show the convergence of the sequence $ \log\det_\zeta(\Delta_{N})/\text{Vol}(M_N)$ when $N\rightarrow \infty$ where $\Delta_N$ is the Laplace-Beltrami operator on $M_N$. To do this we employ certain heat kernel techniques in combination with Li-Yau bounds, as well as an analysis of ``twist Laplacians'' as developed by Phillips and Sarnak \cite{PSa}.

\subsection{Road Map to the Rest of Introduction}

\begin{figure}[h]
    \centering
    \includegraphics[width=0.85\linewidth]{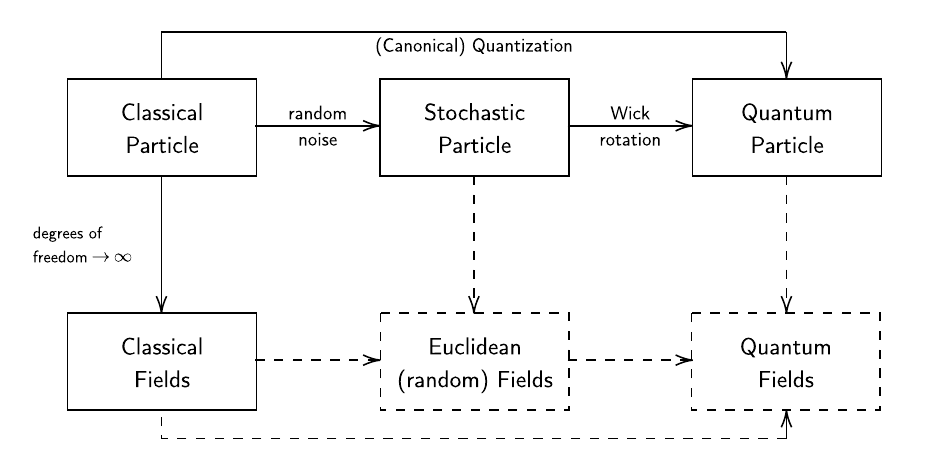}
    \caption{Method of the ``Missing Box(es)''}
    \label{fig-intro-road-map}
\end{figure}

In the rest of this introduction, rather than introducing background knowledge in the strict sense chapter-by-chapter for this thesis, we have chosen to accomplish a certain expository task which may also stand by its own. Certainly, it will provide enough motivation and background for proceeding to the remaining chapters of this thesis, each having also its own introduction. In fact, we will describe a main body of ideas which underlies all individual chapters although these may at first sight appear to belong to different subjects.

We borrow from Sidney Coleman (\cite{Coleman} pp.\ 57) the organizing principle of this exposition---he calls it the ``Method of the Missing Box''. However, we have one more box: the stochastic box. One of the main efforts in accomplishing our task has hence been to fit the new box(es) into Coleman's picture, and the result is roughly depicted in figure \ref{fig-intro-road-map}. In section \ref{sec-intro-1d} we traverse through the first line. In section \ref{sec-intro-field-2d} we traverse through the second line. In section \ref{sec-intro-fock} we present some necessary materials to deal with the stochastic box on the second line and to connect it with the quantum box; we refer to the more detailed description at the beginning of that section.

\section[The Harmonic Oscillator]{The Harmonic Oscillator: Perspectives on Time Evolution}\label{sec-intro-1d}

This section aims to introduce some key concepts in the mathematical treatment of quantum fields which are already relevant on the level of quantum mechanics. We discuss them in the simpler setting of quantum mechanics in order to bring conceptual clarity. We also make connections and comparison to probability theory on this level centering around the Feynman-Kac formula. In the process we illustrate all concepts with the key example of the harmonic oscillator. Persistent emphasis will be put on \textit{time evolution}, except perhaps in subsection \ref{sec-intro-stat-mec} where, loosely speaking, we have taken the limit $t\to +\infty$.

\subsection{Classical Mechanics}\label{sec-intro-class-mec}
We shall start with classical mechanics, recalling briefly the basic concepts of Hamiltonian and Lagrangian formalisms (both are important in this thesis). The three basic theoretical ingredients of mechanics are: the \textsf{states}, the \textsf{observables} and the \textsf{time evolution} (dynamics). The last one is our primary focus.

      Suppose we want to describe the motion of \textit{one} particle moving in $\mb{R}^d$, possibly under some external potential $V:\mb{R}^d\lto \mb{R}$. In other words, $\mb{R}^d$ is understood here as the space of \textit{all possible positions} of the particle, also called the \textsf{configuration space}. Suppose~$\mn{x}:(-\varepsilon,\varepsilon)\lto \mb{R}^d$, $\varepsilon>0$, describes the trajectory of this particle in a small interval of time. The basic idea of \textsf{Hamiltonian mechanics} is that one should lift the trajectory to the \textsf{phase space}~$T^*\mb{R}^d\cong \mb{R}^{2d}$, that is, to consider
\begin{equation}
  \tilde{\mn{x}}(t)\defeq (\mn{x}(t),\mn{p}(t))\defeq (\mn{x}(t),m\ank{\dot{\mn{x}}(t),-}_{\mb{R}^d}).
  \label{eqn-class-mec-phase-trajec}
\end{equation}
Here the 1-form $\mn{p}(t):=m\ank{\dot{\mn{x}}(t),-}_{\mb{R}^d}$ is called the \textsf{momentum} of the particle\footnote{The momentum relates to the velocity~$\dot{\mn{x}}(t)$ in this way only in some cases (though many). In general the relation is given by the \textit{Lagrangian} using the \textit{Legendre transform}, see (\ref{eqn-classical-canonical-momentum}) below and subsequent references.}, and the parameter $m>0$ the \textsf{mass} of the particle. 
The fundamental ingredients are described as follows.
\begin{enumerate}[(i)]
  \item The \textsf{classical state} of the particle at time~$t$ is given by the point~$(\mn{x}(t),\mn{p}(t))\in T^*\mb{R}^d$; thus the \textsf{phase space} is the space of \textit{all possible (classical) states} of the particle;
  \item An \textsf{observable} is given by a real function~$O\in C^{\infty}(T^*\mb{R}^d)$; its value at~$(\mn{x},\mn{p})$ corresponds to the outcome of an experimental observation made of the particle in the state~$(\mn{x},\mn{p})$;
  \item The \textsf{time evolution}, that is, the phase space trajectory~$t\mapsto \tilde{\mn{x}}(t)$, is an integral curve of a vector field $X_H\in C^{\infty}(T^*\mb{R}^d,TT^*\mb{R}^d)$ associated to a special observable~$H$ called the \textsf{Hamiltonian}, which gives the \textsf{energy} of each classical state $(\mn{x},\mn{p})$.
\end{enumerate}

Thus~$H$ is thought to specify the physics of any particular mechanical model at hand, and the vector field $X_H$, called the \textsf{Hamiltonian vector field} of $H$, is determined uniquely by $H$ through
\begin{equation}
  X_H=\sum_{j=1}^d \Big(
  \frac{\partial H}{\partial p_j}\frac{\partial}{\partial x^j}-\frac{\partial H}{\partial x^j}\frac{\partial}{\partial p_j} \Big),
  \label{}
\end{equation}
where~$(x^1,\dots,x^d,p_1,\dots,p_d)$ are standard coordinates on~$T^*\mb{R}^d$. An integral curve of $X_H$, therefore, solves the system of ODEs
\begin{equation}
  \frac{\dd x^j}{\dd t}(t)=\frac{\partial H}{\partial p_j}(\mn{x}(t),\mn{p}(t)),\quad\quad
  \frac{\dd p_j}{\dd t}(t)=-\frac{\partial H}{\partial x^j}(\mn{x}(t),\mn{p}(t)).
  \label{}
\end{equation}
These are called \textsf{Hamilton's equations of motion}.

Next we turn to the \textsf{Lagrangian formalism}. This formalism, while putting less emphasis on states and observables (traditionally), describes the time evolution by dictating that the full trajectory~$\mn{x}:[0,T]\lto \mb{R}^d$ of the particle over any \textit{fixed time interval}~$[0,T]$ solves a variational problem. It focuses on a function~$L$, called the \textsf{Lagrangian}, defined on the \textit{tangent bundle}~$T\mb{R}^d$ instead of the cotangent bundle. Precisely, the \textsf{least action principle} says that among all trajectories with fixed end points~$\mn{x}(0)=\mn{x}_0$ and~$\mn{x}(T)=\mn{x}_T$, the particle chooses the one which minimizes the \textsf{action}
\begin{equation}
  S_{[0,T]}(\mn{x})\defeq \int_{0}^{T}L(\mn{x}(t),\dot{\mn{x}}(t))\,\dd t.
  \label{}
\end{equation}
One could then deduce that the minimizer also satisfy a set of equations of motion, called the \textsf{Euler-Lagrange equations} \cite{Arnold-1989, Takh}. Now we point out that in \textit{regular}\footnote{See \cite{baez-wise-2005} section 4.1.} circumstances the Lagrangian and Hamiltonian formalisms are related as follows. Given a Lagrangian~$L$ one can define for each~$\mn{x}\in\mb{R}^n$ a map $\lambda_{L,\mn{x}}:T_\mn{x}\mb{R}^d\lto T^*_{\mn{x}}\mb{R}^d$ so that the coordinates~$p_j$,~$1\le j\le d$, of~$\lambda_{L,\mn{x}}(\mn{v})$ are given by
\begin{equation}
  p_j=\frac{\partial L}{\partial v^j}(\mn{v}),
  \label{eqn-classical-canonical-momentum}
\end{equation}
where~$(v^j)$ are standard coordinates on~$T_{\mn{x}}\mb{R}^d$. This~$p_j$ is sometimes called the \textsf{canonical momentum conjugate to the position coordinate~$x^j$} (with respect to~$L$). Moreover, the Hamiltonian function~$H$ on~$T^*\mb{R}^d$ is given by the \textit{Legendre transform} of~$L$. Now if~$t\mapsto \mn{x}(t)$ is a trajectory minimizing the action given by~$L$, then~$t\mapsto (\mn{x}(t),\lambda_{L,\mn{x}(t)}(\dot{\mn{x}}(t)))$ solves Hamilton's equations for~$H$.\footnote{A vast generalization of ``classical mechanics'' that deals with general maps between Riemannian manifolds (and sections of fiber bundles) has been put forward in \cite{Ch}.}

\begin{exxx}
  We look at
  \begin{equation}
    H(\mn{x},\mn{p})=\frac{|\mn{p}|_{T^*_{\mn{x}}\mb{R}^d}^2}{2m}+V(\mn{x})
    \label{eqn-ham-class-part-potential}
  \end{equation}
  with~$m>0$ and~$V:\mb{R}^d\lto \mb{R}$ a function. For this~$H$, we have
  \begin{equation}
    X_H|_{(\mn{x},\mn{p})}=(m^{-1}\mn{p},-\nabla_{\mb{R}^d} V(\mn{x})),
    \label{}
  \end{equation}
  and hence the Hamilton's equation becomes
  \begin{equation}
    \left\{
    \begin{array}{l}
      \partial_t\mn{x}(t)=m^{-1}\mn{p}(t),\\
      \partial_t\mn{p}(t)=-\nabla V(\mn{x}(t)).
    \end{array}
    \right.
    \label{eqn-HJ-equation-particle-potential}
  \end{equation}
  This implies~$m\cdot \ddot{\mn{x}}(t)=-\nabla V(\mn{x}(t))$ which is Newton's law. The Lagrangian is accordingly
  \begin{equation}
    L(\mn{x},\mn{v})=\frac{1}{2}m|\mn{v}|_{T_{\mn{x}}\mb{R}^d}^2 -V(\mn{x}).
    \label{eqn-lag-general-potential}
  \end{equation}
In this case indeed $\lambda_{L,\mn{x}}(\mn{v})=m\ank{\mn{v},-}_{\mb{R}^d}$.   For a classical 1-dimensional harmonic oscillator we have
  \begin{equation}
    V_{m,\omega}(x)=\frac{1}{2}m\omega^2x^2, \quad\quad\textrm{for }x\in\mb{R}\textrm{ with }\omega>0.
    \label{}
  \end{equation}
  The solutions are~$x(t)=A\sin(\omega t)+B\cos(\omega t)$.
\end{exxx}

\subsection{Quantum Mechanics}
Now we come to quantum mechanics, again of the same particle moving in~$\mb{R}^d$. To set the stage, now we remember that the configuration space~$\mb{R}^d$ comes equipped with a \textit{measure}~$\mathcal{L}_{\mb{R}^d}$, the \textsf{Lebesgue measure}.
The basic framework of quantum mechanics now dictates that, in the so-called \textsf{Schr\"odinger picture},
  \begin{enumerate}[(i)]
    \item the \textsf{state} of the particle is given by a complex function~$\psi\in \mathcal{H}_{\mm{pos}}:= L^2(\mb{R}^d,\mathcal{L}_{\mb{R}^d})$ with~$\nrm{\psi}_{L^2}=1$, called the \textsf{wave function}\footnote{More precisely, two wave functions~$\psi$ and~$\tilde{\psi}$ deferring by a phase factor,~$\tilde{\psi}=\me^{\ii \theta}\psi$, represent the same state. Thus the \textsf{space of states} is the \textsf{projective Hilbert space}~$\mathcal{PH}_{\mm{pos}}:=\mathcal{H}_{\mm{pos}}/\sim$ where~$\psi\sim \lambda \psi$ for all~$\lambda\in \mb{C}$,~$\lambda\ne 0$. A wave function specifies and over-specifies a state. While the projective point of view remains important, we will work with wave functions in this thesis.}; the basic interpretation is that~$|\psi(\mn{x})|^2$ for~$\mn{x}\in\mb{R}^d$ gives the \textit{probability density} that the particle be found at~$\mn{x}\in \mb{R}^d$ in an attempt of observation (on the \textit{position} of the particle); 
    \item an \textsf{observable} of the mechanical system is given by a (usually unbounded) self-adjoint operator on~$\mathcal{H}_{\mm{pos}}$.
    \item the \textsf{time evolution} is represented by the 1-parameter unitary group~$\me^{-\ii t\mn{H}}$ generated by a special observable~$\mn{H}$ called the \textsf{Hamiltonian operator}; in other words, if the system starts with an initial wave function~$\psi$, then the wave function at time~$t$ is~$\me^{-\ii t\mn{H}}\psi$.
  \end{enumerate}

  The \textsf{probabilistic interpretation} in general says that given an observable~$\mn{O}$ on~$\mathcal{H}$ and a wave function~$\psi\in \mathcal{H}$, then a large number of independent measurements of~$\mn{O}$ of the system ``ideally'' prepared in state~$\psi$ will return results following the probability law
\begin{equation}
  \left.
  \begin{array}{rcl}
    \ank{f,E_{\mn{O}}(-)f}_{\mathcal{H}}: \mathcal{B}_{\mb{R}}&\lto & [0,1],\\
    A&\longmapsto & \ank{f,E_{\mn{O}}(A)f}_{\mathcal{H}},
  \end{array}
  \right.
  \label{}
\end{equation}
where~$\mathcal{B}_{\mb{R}}$ denotes the Borel~$\sigma$-algebra of~$\mb{R}$, and~$E_{\mn{O}}$ the projection valued spectral measure of~$\mn{O}$. The interpretation of~$|\psi(\mn{x})|^2$ as a probability density mentioned in (i) is indeed a special case of this general interpretation by taking the multiplication operators~$\mn{X}_j:\psi\mapsto x_j\cdot \psi$,~$1\le j\le n$, as observables (called the \textsf{position operators}).

As with classical mechanics, the physics of a particular model is specified by giving the Hamiltonian~$\mn{H}$. Therefore, in view of the rule (iii) and the probabilistic interpretation, knowing the spectral resolution of~$\mn{H}$ is especially important. In the sequel, we focus on 3 specific models: the 1D harmonic oscillator, the 1D free particle, and the 2D harmonic oscillator. To indicate rigor, we state results on the first model as lemmas \ref{lemm-spec-har-osc-1d} and \ref{lemm-crea-ann-har-osc-1d}.

\begin{lemm}[\cite{Dimock2} theorem 4.5]\label{lemm-spec-har-osc-1d}
  The \textsf{harmonic oscillator Hamiltonian}
  \begin{equation}
    \mn{H}_{m,\omega}\defeq -\frac{1}{2m}\frac{\dd ^2}{\dd x^2}+\frac{1}{2}m\omega^2 x^2,\quad\quad m>0,~\omega>0,
    \label{eqn-def-1d-harm-osc-hamil}
  \end{equation}
  acts on~$L^2(\mb{R},\mathcal{L}_{\mb{R}})$ as an unbounded essentially self-adjoint operator with core $\mathcal{S}(\mb{R})$, the Schwartz functions, and its spectrum consists of \textsl{simple} eigenvalues
  \begin{equation}
    E_n(\mn{H}_{m,\omega})\defeq \Big( n+\frac{1}{2} \Big)\omega,\quad\quad n\in\mb{N},
    \label{}
  \end{equation}
  with a unique \textsf{ground state}
  \begin{equation}
    \Omega_0(x)=\left( \frac{m\omega}{\pi} \right)^{1/4}\me^{-\frac{1}{2}m\omega x^2}
    \label{eqn-har-osc-ground-state-1d}
  \end{equation}
  corresponding to~$E_0=\frac{1}{2}\omega$.
\end{lemm}

\begin{lemm} \label{lemm-crea-ann-har-osc-1d}
  Let~$\mn{H}_{m,\omega}$ be as in lemma \ref{lemm-spec-har-osc-1d}. Define the operators
  \begin{equation}
    \mn{A}_{m,\omega}\defeq \frac{1}{\sqrt{2m\omega}}\Big( \frac{\dd}{\dd x}+m\omega x \Big),\quad\quad \mn{A}_{m,\omega}^{\dagger}\defeq \frac{1}{\sqrt{2m\omega}}\Big( -\frac{\dd}{\dd x}+m\omega x \Big),
    \label{eqn-1d-har-osc-crea-ann-expres}
  \end{equation}
  Then we have~$\mn{A}_{m,\omega}:\mss{E}_n\lto \mss{E}_{n-1}$,~$n\ge 1$,~$\mn{A}_{m,\omega}^{\dagger}:\mss{E}_{n}\lto \mss{E}_{n+1}$,~$\mss{E}_n$ the eigenspace of~$E_n$,
  \begin{align}
    [\mn{A}_{m,\omega},\mn{A}_{m,\omega}^{\dagger}]&=\one,\\
    \mn{H}_{m,\omega}&=\omega\Big( \mn{A}_{m,\omega}^{\dagger}\mn{A}_{m,\omega}+\frac{1}{2} \Big), \label{eqn-1d-har-osc-hamil-by-crea-ann}\\
    \mn{A}_{m,\omega}\Omega_0 &=0,
    \label{}
  \end{align}
  with~$\Omega_0$ given by (\ref{eqn-har-osc-ground-state-1d}), and that
  \begin{equation}
    \Omega_n\defeq \frac{1}{\sqrt{n!}}(\mn{A}_{m,\omega}^{\dagger})^n\Omega_0
    \label{eqn-harm-osc-nth-eigenfunc-def}
  \end{equation}
  is the normalized eigenfunction generating the eigenspace of~$E_n(\mn{H}_{m,\omega})$. Moreover,~$\{\Omega_n\}_{n=0}^{\infty}$ forms a complete orthonormal basis of~$L^2(\mb{R},\mathcal{L}_{\mb{R}})$.
\end{lemm}

 \begin{def7}[anticipating subsections \ref{sec-intro-fock-rep}, \ref{sec-intro-free-field-cyl-canon-quant}] \label{rem-necessary-crea-ann-har-osc-1d}
    Since knowing the spectral resolution of~$\mn{H}_{m,\omega}$ already allows us to describe explicitly the dynamics, of what use are the operators~$\mn{A}_{m,\omega}$,~$\mn{A}_{m,\omega}^{\dagger}$ in (\ref{eqn-1d-har-osc-crea-ann-expres})? While they could still be dispensed here, they become crucial in QFT in two respects. Firstly (\ref{eqn-harm-osc-nth-eigenfunc-def}) shows that the finite linear spans of the vectors~$\{(\mn{A}_{m,\omega}^{\dagger})^n\Omega_0\}$ are \textit{dense} in the Hilbert space; secondly we note that many important observables can actually be expressed in terms of~$\mn{A}$,~$\mn{A}^{\dagger}$, such as the Hamiltonian (\ref{eqn-1d-har-osc-hamil-by-crea-ann}), the \textsf{position} and \textsf{momentum}
    \begin{equation}
      \mn{X}=\frac{1}{\sqrt{2m\omega}}(\mn{A}_{m,\omega}+\mn{A}_{m,\omega}^{\dagger}),\quad\quad \mn{P}=-\ii\frac{\dd}{\dd x}=\ii\sqrt{\frac{m\omega}{2}}(\mn{A}_{m,\omega}^{\dagger}-\mn{A}_{m,\omega}).
      \label{eqn-1d-har-osc-pos-mom}
    \end{equation}
    In infinite degrees of freedom, we will have an infinite family of operators~$\{\mn{A}_k\}$ indexed by some~$k$, as well as their adjoints. In that case, the above two points combined will tell us that if we know just the vacuum state~$\Omega_0$ and how the~$\mn{A}_k$'s commute with each others (and each other's adjoints), then \textit{essentially} we know the dynamics of the whole system, since the Hamiltonian and the important observables are all expressed in terms of~$\mn{A}_k$,~$\mn{A}_k^{\dagger}$!
  \end{def7}

\begin{exxx}
    [1D free particle] \label{exam-1d-free-part} The \textsf{free Hamiltonian} with mass~$m>0$ is
    \begin{equation}
      \mn{H}_{m,0}\defeq -\frac{1}{2m}\frac{\dd^2}{\dd x^2}.
      \label{}
    \end{equation}
    In other words, it is the (positive) Laplacian on~$\mb{R}$ with a coefficient in front. We recall some of its properties:
  \begin{enumerate}[(i)]
    \item it is self-adjoint and nonnegative on~$L^2(\mb{R},\mathcal{L}_{\mb{R}})$ with domain being the Sobolev space~$W^2(\mb{R})\subset L^2(\mb{R})$ of order~$2$ (operator closure of~$\mn{H}_{m,0}|_{\mathcal{S}(\mb{R})}$);
    \item its spectrum is~$\sigma(\mn{H}_{m,0})=[0,+\infty)$;
      \item none of~$\lambda\in \sigma(\mn{H}_{m,0})$ is an eigenvalue.
  \end{enumerate}
  The spectral properties of~$\mn{H}_{m,0}$ are thus drastically different from those of~$\mn{H}_{m,\omega}$ with~$\omega>0$. For proofs see \cite{Lewin} section 2.7. Although~$\mn{H}_{m,0}$ has no eigenfunctions in~$L^2$, the functions~$\mathcal{E}_k(x):=\me^{\ii kx}$,~$k\in \mb{R}$, satisfy
  \begin{equation}
    \mn{H}_{m,0}\mathcal{E}_k=\frac{|k|^2}{2m}\mathcal{E}_k.
    \label{}
  \end{equation}
These functions could be interpreted as ``traveling waves''. Indeed, applying time evolution we find
  \begin{equation}
    (\me^{-\ii t\mn{H}_{m,0}}\mathcal{E}_k)(x)=\me^{\ii k(x-\frac{k}{2m}t)}=(\tau_{vt}\mathcal{E}_k)(x),
    \label{}
  \end{equation}
  where~$(\tau_y f)(x):=f(x-y)$ defines the \textsf{shift operator} and~$v:=k/2m$ is the \textsf{(phase) velocity}. Therefore, heuristically,~$\mathcal{E}_k$ represents a wave ``traveling to the right''. Likewise,~$\mathcal{E}_{-k}$ ``travels to the left''.
  \end{exxx}

  \begin{exxx}
    [2D harmonic oscillator] \label{exam-2d-har-osc} The 2-dimensional \textsf{harmonic oscillator Hamiltonian} with mass~$m>0$ and angular frequency~$\omega>0$ is
    \begin{equation}
      \mn{H}^{\mb{R}^2}_{m,\omega}\defeq -\frac{1}{2m}\Big(
      \frac{\partial^2}{\partial x^2}+\frac{\partial^2}{\partial y^2}\Big) +\frac{1}{2}m\omega^2(x^2+y^2).
      \label{}
    \end{equation}
    Now by separation of variables we find that it acts on~$L^2(\mb{R}^2,\mathcal{L}_{\mb{R}^2})=L^2(\mb{R},\mathcal{L}_{\mb{R}})\otimes L^2(\mb{R},\mathcal{L}_{\mb{R}})$, the latter being the completed (Hilbert) tensor product, by
    \begin{equation}
      \mn{H}^{\mb{R}^2}_{m,\omega}=\mn{H}_{m,\omega}\otimes \one_{L^2(\mb{R})}+\one_{L^2(\mb{R})}\otimes \mn{H}_{m,\omega}.
      \label{eqn-2d-har-osc-hamil-tens-prod}
    \end{equation}
    Its spectrum consists again only of eigenvalues $E_n(\mn{H}^{\mb{R}^2}_{m,\omega}):=(n+1)\omega$, $n\in\mb{N}$, 
  and one could simply check that $\Omega_{p,q}(x,y):= \Omega_p(x)\Omega_q(y)$,
  with~$\Omega_n$ given by (\ref{eqn-harm-osc-nth-eigenfunc-def}), is now an eigenfunction of~$\mn{H}^{\mb{R}^2}_{m,\omega}$ with eigenvalue~$(p+q+1)\omega$. Since~$\{\Omega_n\}_{n=0}^{\infty}$ is a complete orthonormal basis of~$L^2(\mb{R},\mathcal{L}_{\mb{R}})$, we know that~$\{\Omega_{p,q}\}_{p,q=0}^{\infty}$ is then a complete orthonormal basis of~$L^2(\mb{R}^2,\mathcal{L}_{\mb{R}^2})$. Moreover, we know that the eigenvalue~$E_n$ now has multiplicity~$d(n)=n+1$, which is the number of~$(p,q)$ with~$p+q=n$. The \textsf{ground state}
  \begin{equation}
    \Omega_{0,0}(x,y)=\left( \frac{m\omega}{\pi} \right)^{1/2}\me^{-\frac{1}{2}m\omega (x^2+y^2)}
    \label{eqn-har-osc-ground-state-2d}
  \end{equation}
  will be of particular importance. By separation of variables, we could define similarly as in the 1D case operators
\begin{equation}
  \mn{A}_{m,\omega}^{(x)}\defeq \frac{1}{\sqrt{2m\omega}}\Big( \frac{\partial}{\partial x}+m\omega x\Big),\quad\quad
\mn{A}_{m,\omega}^{(y)}\defeq \frac{1}{\sqrt{2m\omega}}\Big( \frac{\partial}{\partial y}+m\omega y\Big),
  \label{}
\end{equation}
as well as their adjoints, and we would find
\begin{align}
  [\mn{A}_{m,\omega}^{(i)}, \mn{A}_{m,\omega}^{(j)\dagger}]&=[\mn{A}_{m,\omega}^{(i)}, \mn{A}_{m,\omega}^{(i)}]=0,\quad\quad\textrm{for }i\ne j,\label{eqn-2d-har-osc-ccr-1}\\
  [\mn{A}_{m,\omega}^{(i)},\mn{A}_{m,\omega}^{(i)\dagger}]&=\one,
  \label{eqn-2d-har-osc-ccr-2}
\end{align}
where~$i$,~$j\in \{x,y\}$. In the 2D case, since~$\mb{R}^2\cong\mb{C}$, it turns out useful to also introduce the \textit{complex coordinate}~$z=x+\ii y$. The corresponding position operator is then~$\mn{Z}:f\mapsto (x+\ii y)f$. This is no longer self-adjoint, and~$\mn{Z}^{\dagger}:f\mapsto (x-\ii y)f$. The momentum operators are $\mn{P}_z=\frac{1}{2}(\mn{P}_x-\ii \mn{P}_y)=-\ii\partial_z$, $\mn{P}_z^{\dagger}=-\ii \partial_{\ol{z}}$, where $\partial_z$, $\partial_{\ol{z}}$ are the Wirtinger operators. We thus put accordingly
\begin{equation}
  \mn{A}_{m,\omega}^{(z)}\defeq \frac{1}{\sqrt{m\omega}}\Big( \frac{\partial}{\partial \ol{z}}+\frac{1}{2}m\omega z\Big)=\frac{1}{\sqrt{2}}\big(\mn{A}_{m,\omega}^{(x)}+\ii\mn{A}_{m,\omega}^{(y)}\big),
  \quad\quad
  \mn{A}_{m,\omega}^{(\ol{z})}\defeq \frac{1}{\sqrt{m\omega}}\Big( \frac{\partial}{\partial z}+\frac{1}{2}m\omega \ol{z}\Big).
  \label{}
\end{equation}
Note here that~$\mn{A}_{m,\omega}^{(\ol{z})}=\ol{\mn{A}}_{m,\omega}^{(z)}\ne \mn{A}_{m,\omega}^{(z)\dagger}$. One could check that~$\mn{A}_{m,\omega}^{(z)}$,~$\mn{A}_{m,\omega}^{(\ol{z})}$ and their adjoints would satisfy the same relations (\ref{eqn-2d-har-osc-ccr-1}), (\ref{eqn-2d-har-osc-ccr-2}) now with~$i$,~$j\in\{z,\ol{z}\}$.
  \end{exxx}

\subsection{Stochastic Processes and Invariant Measures}

  For this subsection, we refer to appendix \ref{sec-app-proba-background} for more background and intuitive discussions on basic probability theory. 
  Let $(X_t)_{t\ge 0}$ be a (real-valued) \textsf{stochastic process} indexed by~$\mb{R}_+=[0,+\infty)$ on the probability space~$(Q,\mathcal{O},\mb{P})$.
  That is, let each $X_t$, $t\ge 0$, be a random variable on $(Q,\mathcal{O},\mb{P})$.
    For example,~$X_t$ could be the position (in~$\mb{R}$) of a Brownian particle at time~$t$. 
    Given any collection~$\{Y_i\}_{i\in I}$ of random variables defined on~$(Q,\mathcal{O},\mb{P})$, they generate a sub-$\sigma$-algebra of~$\mathcal{O}$ which is denoted~$\sigma(Y_i|i\in I)$. 
For any sub-$\sigma$-algebra~$\mathcal{A}\subset \mathcal{O}$, one could take the \textsf{conditional expectation} of a random variable~$Y$ with respect to~$\mathcal{A}$ which is an $\mathcal{A}$-measurable random variable denoted~$\mb{E}[Y|\mathcal{A}]$. 
In this thesis, we will focus on the following class of stochastic processes.

\begin{deef}
  We say that a stochastic process~$(X_t)_{t\ge 0}$ on~$(Q,\mathcal{O},\mb{P})$ is \textsf{Markov} if
  \begin{equation}
    \mb{E}\big[f(X_t)\big|\sigma(X_r|r\le s) \big]=\mb{E}\big[f(X_t) \big|\sigma(X_s) \big],
    \label{eqn-cond-gen-markov-proces}
  \end{equation}
  for all~$s\le t$ and Borel measurable functions~$f:\mb{R}\lto\mb{R}$.
\end{deef}
We are interested in describing the evolution of~$X_t$ through time. 
For $s\le t$, we define the \textsf{transition kernel}~$P_{s\to t}:\mb{R}\times \mathcal{B}_{\mb{R}}\lto [0,1]$,~$\mathcal{B}_{\mb{R}}$ being the Borel~$\sigma$-algebra on~$\mb{R}$, so that intuitively,
\begin{equation}
   P_{s\to t}(x,A)\heueq \mb{P}(X_t\in A|X_s=x),
 \end{equation}
namely,~$P_{s\to t}(x,A)$ is the probability for~$X_t$ to land in~$A$ if~$X_s=x$ (\cite{Comets-Meyre-2020} pp.\ 112). This kernel moreover has the property that for any measurable function $h:\mb{R}\lto \mb{R}$,
\begin{equation}
   \int_{}^{}h(y)P_{s\to t}(X_s,\dd y) = \mb{E}[h(X_t)|\sigma(X_s)],
   \label{eqn-trans-kernel-measura-func-property}
\end{equation}
almost surely as random variables.\footnote{We remark that the transition kernel is an object that could be defined for two general random variables $X$ and $Y$ regardless of whether they come from a stochastic process, see \cite{Dudley} section 10.2.}
 \begin{exxx}\label{exam-transition-brownian}
  The transition kernel for Brownian motion (in~$\mb{R}$) is 
  \begin{equation}
    P_{+t}^{\mm{BM}}(x,\dd y)\defeq \frac{1}{(2\pi t)^{1/2}} \me^{-\frac{1}{2t}|x-y|^2}\cdot\dd\mathcal{L}_{\mb{R}}(y),
    \label{}
  \end{equation}
  where we denote by~$P^{\mm{BM}}_{+t}$ the law~$P^{\mm{BM}}_{s\to s+t}$ for any~$s\ge 0$.
\end{exxx}
For a Markov process, there is the slogan that \textit{given the present, the future is independent from the past}. Indeed, as one could see from the condition (\ref{eqn-cond-gen-markov-proces}), regarding the future of any particular instant~$s$, knowing the whole history up-to~$s$ amounts to knowing just~$X_s$.
This has an important implication which we state as follows.

\begin{lemm}
    [Chapman-Kolmogorov]
    Suppose we have a Markov process~$(X_t)_{t\ge 0}$ on~$(Q,\mathcal{O},\mb{P})$ and let~$P_{s\to t}$ denote the transition kernels as above. Then for every~$r\le s\le t$, every~$x\in\mb{R}$ and Borel set~$A\subset\mb{R}$ we have
  \begin{equation}
    P_{r\to t}(x,A)=\int_{}^{} P_{s\to t}(y,A)P_{r\to s}(x,\dd y).
    \label{eqn-chap-kol}
  \end{equation}
\end{lemm}

\begin{proof}
  Indeed, note that~$\sigma(X_r)\subset \sigma(X_{r_1}|r_1\le r)\subset \sigma(X_{r_1}|r_1\le s)$. Therefore
  \begin{align*}
    \mb{E}[1_A(X_t)|\sigma(X_r)]&=\mb{E}\big[~\mb{E}\big[1_A(X_t)\big|\sigma(X_{r_1}|r_1\le s)\big]~\big|~\sigma(X_r)\big] \tag{law of total probability, lemma \ref{lemm-total-prob}} \\
    &=\mb{E}\big[~\mb{E}\big[1_A(X_t)\big| \sigma(X_s)\big]~\big|~\sigma(X_r)\big] \tag{Markov property} \\
    &=\mb{E} \big[P_{s\to t}(X_s,A)\big| \sigma(X_r)\big] \tag{definition of~$P_{s\to t}$} \\
    &=\int_{}^{} P_{s\to t}(y,A)P_{r\to s}(X_r,\dd y), \tag{property (\ref{eqn-trans-kernel-measura-func-property}) of~$P_{r\to s}$}
  \end{align*}
  giving the result.
\end{proof}

\begin{deef}
  A Markov process is called \textsf{homogeneous} if the transition kernel~$P_{s\to t}$ depends only on~$t-s$.
\end{deef}

Like the case of Brownian motion, if the Markov process~$(X_t)_{t\ge 0}$ is homogeneous, it makes sense then to define the transition kernel~$P_{+t}$ to be~$P_{s\to s+t}$ for any~$s\ge 0$. Now we are ready to introduce two important evolution operators.

\begin{deef}
  Let~$(X_t)_{t\ge 0}$ be a (real-valued) stationary Markov process and denote by~$P_{+t}$ its transition kernel. 
  \begin{enumerate}[(i)]
    \item Define the \textsf{Feller operator} (or \textsf{backward Kolmogorov}), for~$f\in L^{\infty}(\mb{R})$, by
    \begin{equation}
        (P_t^*f)(\cdot)\defeq \mb{E}\big[f(X_{r+t})\big|X_r=(\cdot)\big]=\ddp \int_{}^{}f(y)P_{+t}(\cdot,\dd y),\quad\quad \textrm{for any }r\ge 0.
    \end{equation}
    \item Define the \textsf{Fokker-Planck operator} (or \textsf{forward Kolmogorov}), for a Borel probability measure~$\mu\in \mathcal{M}_{\mm{P}}(\mb{R})$ on~$\mb{R}$, by
    \begin{equation}
      ({P_t}_* \mu)(A)\defeq \int_{}^{}P_{+t}(x,A)\dd\mu(x),
      \label{}
    \end{equation}
    namely $P_{t*}\mu$ is the law of $X_t$ when $\mu$ is the law of $X_0$.
  \end{enumerate}
\end{deef}

The Chapman-Kolmogorov lemma now translates precisely into saying that both~$(P_t^*)_{t\ge 0}$,~$(P_{t*})_{t\ge 0}$ form semi-groups:\begin{equation}
  P_{s+t}^*=P_t^*\circ P_s^*,\quad\quad\textrm{and}\quad\quad P_{s+t*}=P_{t*}\circ P_{s*},\quad\quad\textrm{for all }t,~s\ge 0.
  \label{}
\end{equation}
One also has straightforwardly~$P_0^*=\one_{L^{\infty}(\mb{R})}$ and~$P_{0*}=\one_{\mathcal{M}_{\mm{P}}(\mb{R})}$.
Before specializing to the main example of this subsection, we point out one last lemma which results from an iterated application of the proof of the Chapman-Kolmogorov lemma.

\begin{lemm}
  [multiple correlation] \label{lemm-multiple-corr-intro} Let~$(X_t)_{t\ge 0}$ be a stationary Markov process and denote by~$(P_t^*)_{t\ge 0}$ its Feller semi-group. Then we have, for any sequence of times~$0<t_1<t_2<\cdots<t_n$, and functions $f_j\in L^{\infty}(\mb{R})$, $0\le j\le n$,
  \begin{align}
  \mb{E}\big[ f_1(X_{t_1})\cdots f_n(X_{t_n}) \big|X_0=x\big]&=(P_{t_1}^*f_1\cdots P^*_{t_{n-1}-t_{n-2}}f_{n-1} P^*_{t_n-t_{n-1}}f_n )(x),\quad\textrm{and} \\
    \mb{E}\big[ f_0(X_0)f_1(X_{t_1})\cdots f_n(X_{t_n}) \big]&=\bank{f_0, P_{t_1}^*f_1\cdots P^*_{t_{n-1}-t_{n-2}}f_{n-1} P^*_{t_n-t_{n-1}}f_n }_{L^2(\gamma_0)}.
    \label{eqn-mult-cor-fey-kac}
  \end{align}
  Here~$\gamma_0$ denotes the law of~$X_0$.
\end{lemm}

\begin{exxx}
  [Ornstein-Uhlenbeck process]\label{exam-orns-uhl} We consider the process~$(X_t)_{t\ge 0}$ which solves the \textsf{stochastic differential equation} (SDE), called the \textsf{Langevin equation},
  \begin{equation}
    \frac{\dd X_t}{\dd t}=-X_t +\xi_t,\quad\quad \textrm{say, with }X_0\defeq \zeta\textrm{ deterministic},
    \label{eqn-langevin-orn-uhl-1d-simple}
  \end{equation}
  where~$\xi$ is a random distribution on~$\mb{R}$ called the \textsf{Gaussian white noise}. 
  Informally,~$\xi$ can be thought of as the velocity~$\dd B_t/\dd t$ of a standard Brownian motion~$(B_t)_{t\ge 0}$.\footnote{Although the law of~$B_t$ itself evolves with time and depends on an initial condition,~$\xi$ is ``uniformly'' random through~$\mb{R}$ (against the Lebesgue measure).} Note also that $X_t$ is the value of the gradient~$\nabla(\frac{1}{2}x^2)$ at~$x=X_t$. Thus (\ref{eqn-langevin-orn-uhl-1d-simple}) also describes a \textsf{noisy gradient flow}. One could verify for (\ref{eqn-langevin-orn-uhl-1d-simple}) that the solution is given (up-to negligible modification) by the \textsf{Duhamel formula}
\begin{equation}
  X_t=\zeta\me^{-t}+\int_{0}^{t}\me^{-(t-s)}\,\dd B_s,\quad\quad \textrm{for all }t\ge 0.
  \label{}
\end{equation}
Here the latter integral, which could also be written as~$\int_{0}^{t}\me^{-(t-s)}\xi(s)\,\dd s$, is interpreted as a Wiener integral. From this formula we deduce
\begin{equation}
  X_t\quad\xlongrightarrow[t\to +\infty]{\mm{law}} \quad\mathcal{N}(0,\tts\frac{1}{2}).
  \label{eqn-orns-uhlen-equil-measure}
\end{equation}
Now we could also see that if we randomize the initial condition~$\zeta$ to have law~$\mathcal{N}(0,\frac{1}{2})=:\gamma_{1,1}$, then the law of~$X_t$ remains~$\gamma_{1,1}$ for all~$t\ge 0$. In other words,~$P_{t*}\gamma_{1,1}\equiv \gamma_{1,1}$ for all~$t\ge 0$. When this happens $(X_t)_{t\ge 0}$ is called the \textsf{Ornstein-Uhlenbeck process} and we say that~$\gamma_{1,1}$ is an \textsf{invariant measure} of the equation (\ref{eqn-langevin-orn-uhl-1d-simple}). Equivalently, $(X_t)_{t\ge 0}$ is also characterized as the unique \textsf{Gaussian process} (each $X_t$ is Gaussian) satisfying
\begin{equation}
  \mb{E}[X_t]\equiv 0,\quad\quad \mb{E}[X_s X_t]=\frac{1}{2}\me^{-|t-s|},\quad\quad\textrm{for all }t,s\ge 0.
  \label{eqn-orns-uhl-cov-1d}
\end{equation}
By the definition of $P_t^*$ we have, for~$f$,~$h\in L^{\infty}(\mb{R})$,
\begin{equation}
  \mb{E}[f(X_0)h(X_t)]=\ank{f,P_t^* h}_{L^2(\gamma_{1,1})}.
  \label{eqn-orns-uhl-feynman-kac-1d-basic}
\end{equation}
At this stage a coincidence manifests itself with quantum mechanics. We note that
\begin{equation}
  \dd \gamma_{1,1}(x)=|\Omega_0(x)|^2\,\dd \mathcal{L}_{\mb{R}}(x),
  \label{eqn-orns-uhl-ground-state-gauss-meas}
\end{equation}
where we have taken~$\Omega_0$ to be the ground state of the quantum harmonic oscillator (\ref{eqn-har-osc-ground-state-1d}) with mass~$m=1$ and angular frequency~$\omega=1$. We could therefore rewrite the inner product on the r.h.s.\ of (\ref{eqn-orns-uhl-feynman-kac-1d-basic}) against the Lebesgue measure as
\begin{equation}
  \bank{f,P_t^* h}_{L^2(\gamma_{1,1})}=\bank{f\Omega_0,\tilde{P}_t^* (h\Omega_0)}_{L^2(\mathcal{L}_{\mb{R}})},
  \label{}
\end{equation}
where~$\tilde{P}_t^*:=\Omega_0 P_t^* \Omega_0^{-1}$ with~$\Omega_0$ interpreted as a multiplication operator. The unitary transformation
  \begin{equation}
    \left.
    \begin{array}{rcl}
       \cdot \Omega_0:L^2(\mb{R},\gamma_{1,1})&\lto& L^2(\mb{R},\mathcal{L}_{\mb{R}}),\\
       f&\longmapsto & f\Omega_0,
    \end{array}
    \right.
    \label{eqn-def-ground-state-trans}
  \end{equation}
  is usually called the \textsf{ground state transform}. The \textsf{generator} of the semi-group~$(P_t^*)_{t\ge 0}$ can be computed (usually by It\^o's formula \cite{Comets-Meyre-2020} theorem 6.5) for~$f\in C_c^{\infty}(\mb{R})\subset L^{\infty}(\mb{R})$ to be
\begin{equation}
  Lf\defeq \lim_{t\to 0^+}\frac{1}{t}[(P_t^* f)-f]=\Big( \frac{1}{2}\frac{\dd^2}{\dd x^2} -x\frac{\dd}{\dd x}\Big)f.
  \label{}
\end{equation}
Accordingly, the generator of~$(\tilde{P}_t^*)_{t\ge 0}$ is exactly~$\Omega_0 L \Omega_0^{-1}$ and by computation we find,
\begin{equation}
  \Omega_0 L \Omega_0^{-1}=-\mn{H}_{1,1}+\frac{1}{2}=-\mn{H}_{1,1}+E_0(\mn{H}_{1,1}),
  \label{}
\end{equation}
where~$\mn{H}_{1,1}$ is defined as in (\ref{eqn-def-1d-harm-osc-hamil}). Therefore (\ref{eqn-orns-uhl-feynman-kac-1d-basic}) can be rewritten again as
\begin{equation}
  \boxed{\mb{E}\big[f(X_0)h(X_t)\big]=\bank{f\Omega_0,\me^{-t(\mn{H}_{1,1}-E_0(\mn{H}_{1,1}))} (h\Omega_0)}_{L^2(\mathcal{L}_{\mb{R}})}.}
  \label{eqn-funda-relation-har-osc-orn-uhl}
\end{equation}
We find here a precise relation between the Ornstein-Uhlenbeck process~$(X_t)_{t\ge 0}$ and the quantum harmonic oscillator. This relation is going to be exploited again and again throughout the thesis. Finally, we remark that conjugation by the ground state is the analogue on~$\mb{R}$ of the \textsf{Witten deformation} by a Morse function on a general manifold. Indeed, since by Morse's lemma any Morse function~$\varphi$ is quadratic in coordinates near a critical point,~$\me^{-\varphi}$ is then approximately (square-root of) a ``Gaussian density'' (assuming the Hessian positive), and conjugation by~$\me^{-\varphi}$ an ``approximate'' ground state transform. 
\end{exxx}

\subsection{Statistical Mechanics and Equilibria}\label{sec-intro-stat-mec}
If we drip a few drops of dye into a glass of water, and wait, we see that the dye ``diffuses'' and after a while the whole glass becomes evenly colored. Why? One naive attempt is to assume that water and dye consist in total of~$N\gg 1$ molecules obeying \textit{classical laws of motion}, write down the Newtonian equations on the configuration space~$\mb{R}^{3N}$ (or Hamilton's equations on the phase space~$\mb{R}^{6N}$), and solve for the complete trajectories. However, usually~$N\sim 10^{23}$ and it's practically impossible to solve a system of ODEs with~$10^{23}$ mutually coupled components. In fact, since the molecules are so small and so many, it is simply impossible to decide which state in~$\mb{R}^{6N}$ the system is at, for any given instant. It has proved more fruitful, on the other hand, to consider the motion of the dye particles to be \textit{random}. 
Indeed, it turns out that the diffusion of dye in water could be modeled\footnote{We need to assume, for example, that the dye is \textit{dilute} and no \textit{convection} takes place.} by a stochastic process similar to the \textit{Brownian motion} mentioned in the previous subsection and, therefore, is now very well-understood (\cite{Lavenda-Scientific-American}, \cite[section 2.2]{Sethna-2021}).

In the case above, the collection of all the classical states of the individual particles in question (a point in~$\mb{R}^{6N}$) gives what is usually called a \textsf{microstate} of the system. The moral of the above story, and of statistical mechanics, on studying large (macroscopic) systems is thus: one should give up having complete deterministic descriptions of the system in terms of particular microstates, but instead work with \textit{probability measures} over the set of all microstates \cite[pp.\ 19]{FV}. Moreover, there might be a particular measure that would describe the given system when it has ``reached equilibrium''.\footnote{Heuristically and informally, say for a box of gas or a glass of liquid, this corresponds to the ``state of affairs'' when one has isolated the system and waited ``long enough'' ($t\to +\infty$). It could be compared to a limit measure such as (\ref{eqn-orns-uhlen-equil-measure}). Precise physical interpretations, however, are irrelevant for us.} This is usually called an \textsf{equilibrium measure} or \textsf{Gibbs measure} (or \textsf{Gibbs ensemble}). \textsf{Equilibrium statistical mechanics} focus on the study of these measures \cite[pp.\ 1]{FV}.

\begin{exxx}[Ising Model]
One natural phenomenon that the above line of thought has proved particularly successful in explaining is that of \textit{magnets}. 
The model starts by assuming that a (solid) magnet consists of ``dipoles'' located at the vertices of a finite lattice\footnote{Assume also that we have normalized the spacing between nearby atoms to be~$1$.}~$\Lambda\subset \mb{Z}^d$, which are themselves ``small magnets'' each having two possible orientations denoted~$+1$ and~$-1$ called the \textsf{spin}. 
Here a \textsf{microstate} of the system will be specified by the knowledge of the individual spin values at all vertices,
  namely a function~$\phi:\Lambda\lto \{\pm 1\}$, also called a \textsf{configuration}. It turns out that the \textsf{Gibbs measure}, now a probability measure on the \textsf{configuration space} which is the space of maps~$\{\pm 1\}^{\Lambda}$, is given by
  \begin{equation}
    \dd\mu_{\Lambda,\mss{H}}(\phi)\defeq \frac{1}{Z_{\Lambda,\mss{H}}}\me^{-\mss{H}_{\mn{b}}(\phi)}\prod_{x\in\Lambda}\dd \mu_x(\phi(x)),
    \label{eqn-def-lattice-gibbs-measure}
  \end{equation}
  where each~$\mu_x$ is taken as the uniform measure on~$\{\pm 1\}$ with~$\mu_x(\{+1\})=\mu_x(\{-1\})=\frac{1}{2}$,~$\mss{H}_{\mn{b}}$ a specific functional on~$\{\pm 1\}^{\Lambda}$ called the \textsf{Hamiltonian}\footnote{The specific Hamiltonian for the Ising model is written on \cite[pp.\ 42]{FV}. One can also add \textit{boundary conditions} which are discussed in \cite[section 3.1]{FV}.} giving the ``energy'' of a configuration, with several parameters denoted~$\mn{b}$ (e.g.\ the inverse temperature), and~$Z_{\Lambda,\mss{H}}$ a normalization constant making~$\mu_{\Lambda,\mss{H}}$ a probability measure, called the \textsf{partition function}. 
  It is important that~$\mss{H}_{\mm{b}}$ contains joint expressions of spins over different sites such as~$\phi(x)\phi(y)$ for some~$x\ne y$, so that (\ref{eqn-def-lattice-gibbs-measure}) is not merely a product measure. This captures the fact that the spins interact. 
  Having defined~$\mu_{\Lambda,\mss{H}}$, the macroscopic properties of this magnet in equilibrium is then reflected in the expected values of various random variables (functionals) defined on~$\{\pm 1\}^{\Lambda}$. 
  One natural and useful functional is the \textsf{magnetization density}
\begin{equation}
  \mss{M}_{\Lambda}(\phi)\defeq \frac{1}{\#\Lambda}\sum_{x\in \Lambda}\phi(x).
  \label{}
\end{equation}
This measures whether the orientations of the individual atoms add up to an overall orientation of the whole system. If it happens that for certain values of parameters in $\mn{b}$ (such as external magnetic field) we have
\begin{equation}
    \mb{E}_{\mu}[\mss{M}_{\Lambda}]>0, \quad\quad (\textrm{or }<0)
\end{equation}
then it would be reasonable to say that the magnet is ``magnetized''.
\end{exxx}

The recipe given in the form of (\ref{eqn-def-lattice-gibbs-measure}) is in fact generic and gives the Gibbs measure for a wide class of lattice spin models. 
For other models the spin at each vertex may value in spaces other than~$\{\pm 1\}$, such as~$\mb{R}$ or~$\mb{S}^k$. 
Accordingly the \textsf{single spin measures}~$\mu_x$ are adjusted to be the \textit{uniform} measures on these spaces. For example, the Lebesgue measure~$\mathcal{L}_{\mb{R}}$ on~$\mb{R}$. The physics of the model is specified by the  Hamiltonian~$\mss{H}_{\mn{b}}$. 
The general intuition behind (\ref{eqn-def-lattice-gibbs-measure}) is that the measure would favor ``low energy'' configurations and penalize ``high energy'' ones.
In particular, if the measure favors \textit{alignment} of nearby spins, then we call the model \textsf{ferromagnetic} \cite[pp.\ 42]{FV}.\footnote{Although this does not imply that the model \textit{behaves} like a ferromagnet.} These include both the (ordinary) Ising model and the discrete GFF introduced below.
Generally, the products of spin values give an important family of functionals whose expected values
\begin{equation}
  \mb{E}_{\mu}[\phi(x_1)\cdots\phi(x_n)],\quad\quad x_i\in\Lambda,~x_i\ne x_j,~i\ne j
  \label{}
\end{equation}
are called \textsf{correlation functions} or~$n$\textsf{-point functions}.

\begin{exxx}[discrete GFF]\label{exam-intro-discre-gff}
    Another important model is the \textsf{discrete (massive) Gaussian Free Field (GFF)} with~$\phi(x)\in\mb{R}$, and with Hamiltonian
\begin{equation}
  \mss{H}_{\beta,m}^{\mm{GFF}}(\phi)\defeq \frac{\beta}{4d}\sum_{\substack{x,y\in \Lambda \\ x\sim y}} |\phi(x)-\phi(y)|^2 +\frac{m^2}{2}\sum_{x\in \Lambda}\phi(x)^2,
  \label{}
\end{equation}
where~$\beta>0$ is the \textsf{inverse temperature},~$m>0$ is the \textsf{mass},~$\Lambda\subset \mb{Z}^d$ finite, and~$x\sim y$ means~$x$ and~$y$ are nearest neighbors. In this case~$\mu_x\equiv \mathcal{L}_{\mb{R}}$.
\end{exxx}

We end this section with an example which gives a good hint on the relationship between statistical mechanics and (Euclidean) quantum field theory.

\begin{exxx}
  [1d Spring system and Brownian motion] Let~$(B_t)_{t\ge 0}$ be the standard Brownian motion on~$\mb{R}$ starting from~$B_0\equiv 0$. We start by recalling the standard \textit{cylindrical measure} formula for~$(B_t)_{t\ge 0}$, which can be viewed as a special case of lemma \ref{lemm-multiple-corr-intro}:
  \begin{equation}
    \mb{P}(B_{t_1}\in A_1,\cdots,B_{t_n}\in A_n)=\int_{A_1\times\cdots\times A_n}^{}\frac{\me^{-\phi_1^2/2t_1}}{\sqrt{2\pi t_1}}\prod_{j=1}^{n-1}\bigg[ \frac{\exp\big(-\frac{(\phi_{j+1}-\phi_j)^2}{2(t_{j+1}-t_j)}\big)}{\sqrt{2\pi (t_{j+1}-t_j)}} \bigg]\,\dd \phi_1\,\dd \phi_2\cdots \dd \phi_n,
    \label{eqn-brownian-cyl-measure}
  \end{equation}
  where~$0<t_1<\cdots <t_n$ and~$A_j\subset \mb{R}$. For simplicity, we let~$t_j:=j\in\mb{N}$. Formula (\ref{eqn-brownian-cyl-measure}) then tells us that the joint distribution of the tuple~$(B_{t_1},\dots,B_{t_n})$ is
  \begin{equation}
    \dd \mu_{\Lambda_n,\mm{Br}}(\phi)=\frac{1}{(2\pi)^{n/2}}\me^{-\frac{1}{2}\sum_{j=1}^n (\phi_j-\phi_{j-1})^2}\prod_{j=1}^n \dd \mathcal{L}_{\mb{R}}(\phi_j),
    \label{}
  \end{equation}
  with~$\phi_0\equiv 0$ and setting $\phi_j:=B_{t_j}$. This measure is otherwise viewed as a Gibbs measure of a lattice spin system on the 1d finite lattice~$\Lambda_n=\{1,\dots,n\}$, with spin values in~$\mb{R}$, and Hamiltonian
  \begin{equation}
    \mss{H}_{\mm{Br}}(\phi)\defeq \frac{1}{2}\sum_{j=1}^n (\phi_j-\phi_{j-1})^2
    \label{}
  \end{equation}
  which is akin to that of the discrete GFF mentioned above. Moreover we have now a ``boundary condition'' $\phi_0=0$. This means that if we are only interested in the position of the Brownian motion at the time points~$t_j$, then we would do equally well by simulating it on a discrete system consisting of a ``comb'' with beads on it connected by springs (see figure \ref{fig-intro-stat-mec-bead}).
\end{exxx}

\begin{figure}[t]
    \centering
    \includegraphics[width=0.6\linewidth]{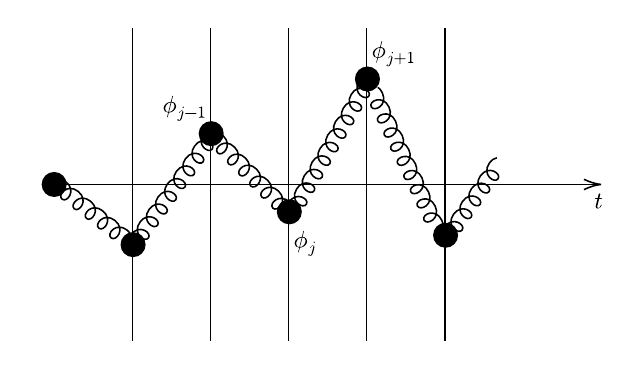}
    \caption{Spring system for Brownian Motion}
    \label{fig-intro-stat-mec-bead}
\end{figure}

The question thus arises: if one is allowed to ``subdivide'' the lattice finer and finer, does it make sense to say that the behavior of the random configuration (discrete field) under a sequence of Gibbs measures would actually ``converge'' to that of a random \textit{field} defined on a continuum domain~$D\subset \mb{R}^d$, with a certain probability law? The answer is yes in the above example as one already has the well-defined Wiener measure. But the question is highly difficult in general.
In fact, in many cases this would prove more difficult than constructing \textit{random fields} directly, which we will do later in the thesis.
Nevertheless, random fields and lattice spin models share many phenomena in common, and they benefit from each other in understanding.

\subsection{Path Integrals: Feynman-Hibbs versus Feynman-Kac}\label{sec-intro-path-integral}

So far we have encountered two points of view on describing the time evolution of systems: the ``Hamiltonian'' point of view which focuses on ``states'' of the system at individual instants, and the ``Lagrangian'' point of view which considers the whole \textit{trajectory} over a period of time as a single object. In probability theory (stochastic processes), these two points of view are in harmony.  For example, for the Brownian motion, the \textsf{Wiener measure}~$\mb{P}^{\mm{W}}$ is a single probability measure defined on the space~$C([0,\infty),\mb{R})$ of \textit{all possible trajectories}, and the random variable~$X_t$ at the instant~$t$ is realized as the evaluation map sending~$\omega\in C([0,\infty),\mb{R})$ to $ \omega(t) \in \mathbb{R}$. The law of~$X_t$ is accordingly the push-forward of~$\mb{P}^{\mm{W}}$ under this evaluation map.
  So far, we have been describing quantum mechanics from the Hamiltonian viewpoint and the Lagrangian description is still missing. Now we fill this gap by introducing the \textsf{path integral} approach. 

  \begin{exxx}[Feynman and Hibbs \cite{Feynman-Hibbs-Styer2010}] \label{exam-feynman-hibbs}
    We start with the 1D free particle introduced in example \ref{exam-1d-free-part}. We set~$m=1$ so that~$\mn{H}_{1,0}=-\frac{1}{2}\Delta_{\mb{R}}$. Using Fourier analysis, we express the unitary evolution~$\me^{-\ii t\mn{H}_{1,0}}$ as an integral operator:
    \begin{equation}
      (\me^{-\ii t\mn{H}_{1,0}}\psi)(x)=\int_{\mb{R}}^{}K_t^0(x,y)\psi(y)\,\dd y,
      \quad \textrm{with }\quad
      K_t^0(x,y)\defeq \frac{1}{\sqrt{2\pi\ii t}}\me^{\frac{\ii}{2t}|x-y|^2},
      \label{}
    \end{equation}
    called the \textsf{free propagator}. We note that~$(\me^{-\ii t \mn{H}_{1,0}})_{t\in\mb{R}}$ is a 1-parameter group which means
    \begin{equation}
      K_{t+s}^0(x,y)=\int_{\mb{R}}^{}K_t^0(x,z)K_s^0(z,y)\,\dd z,\quad\quad\textrm{for all }t,s\in\mb{R}.
      \label{eqn-composition-free-propa-1d}
    \end{equation}
    In particular, given a time interval~$[0,T]$ we divide it into~$N$ equal subintervals and write~$\me^{-\ii T \mn{H}_{1,0}}=(\me^{-\ii \frac{t}{N} \mn{H}_{1,0}})^N$. Then by the composition property (\ref{eqn-composition-free-propa-1d}) of the kernels we obtain
    \begin{align}
      K_T^0(x,y)&=\underbrace{\idotsint}_{N-1}K_{T/N}^0(x,\phi_{N-1})K_{T/N}^0(\phi_{N-1},\phi_{N-2})\cdots K_{T/N}^0(\phi_1,y)\,\dd \phi_1\cdots \dd \phi_{N-1} \\
      &=\idotsint \prod_{j=0}^{N-1} \bigg[
      \frac{\exp\big( \frac{\ii(\phi_{j+1}-\phi_{j})^2}{2(T/N)}\big)}{\sqrt{2\pi \ii(T/N)}} \bigg]\,\dd \phi_1\cdots \dd \phi_{N-1} \label{eqn-1d-free-path-integral-brownian-expres}\\
      &=\idotsint \exp\bigg( \sum_{j=0}^{N-1} \frac{\ii}{2}\frac{(\phi_{j+1}-\phi_j)^2}{(\delta t)^2}(\delta t) \bigg)
      \frac{\dd \phi_1\cdots \dd \phi_{N-1}}{(2\pi\ii (\delta t))^{N/2}}, 
      \label{eqn-1d-free-path-integral-riemann-sum}
    \end{align}
    where~$\delta t:=T/N$, $\phi_0\equiv y$ and $\phi_N\equiv x$. We remember from quantum mechanics that the value~$\psi(x)$ at~$x\in\mb{R}$ of a wave function~$\psi$ is interpreted as the ``probability amplitude'' that the particle be found at~$x$. Accordingly, the value of the propagator~$K_t^0(x,y)$ at~$x$ and~$y\in\mb{R}$ gives the ``probability amplitude'' that the particle has traveled from~$y$ to~$x$ over the time interval~$[0,t]$. Now it is also reasonable to interpret the integration variables~$\phi_j$ above as the positions of the particle at~$t_j=Tj/N$,~$1\le j\le N-1$, respectively. 
  We recognize the term in the exponential of (\ref{eqn-1d-free-path-integral-riemann-sum}) as a Riemann sum and
  \begin{equation}
    \frac{\ii}{2}\sum_{j=0}^{N-1} \frac{(\phi_{j+1}-\phi_j)^2}{(\delta t)^2}(\delta t) \quad \xlongrightarrow[N\to +\infty]{\mm{heu}}\quad  \frac{\ii}{2}\int_{0}^{T}|\phi'(t)|^2\,\dd t,
    \label{eqn-heuristic-limit-1d-free-path-int}
  \end{equation}
  where~$t\mapsto \phi(t)$ now denotes the \textit{trajectory} of the particle over~$[0,T]$, with~$\phi(0)\equiv y$ and~$\phi(T)\equiv x$. Since (\ref{eqn-1d-free-path-integral-riemann-sum}) is valid for all~$N$, we deduce
  \begin{equation}
    K_T^0(x,y) \heueq \int_{\left\{  \substack{\textrm{paths} \\ \phi:[0,T]\to\mb{R} }~\middle|~ \substack{
	     \phi(0)=y,\\ 
	     \phi(T)=x
	   }
      \right\}}^{}
      \me^{\ii\int_{0}^{T}\frac{1}{2}|\phi'(t)|^2\,\dd t} \,\dd \mathcal{L}(\phi),
    \label{eqn-1d-free-path-int-heu-kernel}
  \end{equation}
  where
  \begin{equation}
    \dd \mathcal{L}(\phi) \heueq \lim_{N\to +\infty}\frac{\dd\mathcal{L}_{\mb{R}} (\phi(t_1))\cdots \dd\mathcal{L}_{\mb{R}} (\phi(t_{N-1}))}{(2\pi\ii (T/N))^{N/2}}
    \label{}
  \end{equation}
  is interpreted as the ``Lebesgue measure'' on the space of all paths (with $\phi(0)=y$ and $\phi(T)=x$). We have hence expressed the propagator~$K_T^0(x,y)$ in terms of a so-called \textsf{path integral}.
  \end{exxx}

We have a similar (heuristic) formula also in the more general case with an external potential, that is, for~$\mn{H}_V=\mn{H}_{1,0}+V(x)$. Here~$V$ acts as a multiplication operator~$\psi\mapsto V\psi$. 
The formula is
  \begin{equation}
    K_T^V(x,y) \heueq \int_{\left\{  \substack{\textrm{paths} \\ \phi:[0,T]\to\mb{R} }~\middle|~ \substack{
	     \phi(0)=y,\\ 
	     \phi(T)=x
	   }
      \right\}}^{}
      \me^{\ii\int_{0}^{T}\big(\frac{1}{2}|\phi'(t)|^2-V(\phi(t))\big)\,\dd t} \,\dd \mathcal{L}(\phi),
    \label{eqn-feynman-hibbs-with-pot}
  \end{equation}
  with~$K_t^V$ being the integral kernel of~$\me^{-\ii t\mn{H}_V}$.
  However, due to the non-commutativity between~$\mn{P}=-\ii \partial_x$ and~$\mn{X}$ (cf.\ (\ref{eqn-1d-har-osc-pos-mom}), and hence with~$V$), much more caution is needed. For instance, generally
  \begin{equation}
    \me^{-\ii T(\mn{H}_{1,0}+V)}\ne (\me^{-\ii \frac{T}{N}\mn{H}_{1,0}}\me^{-\ii \frac{T}{N}V})^N.
    \label{}
  \end{equation}
   Along these lines, the key rigorous result available is lemma \ref{lemm-lie-trotter-product} below. 

We now notice the correspondence between (\ref{eqn-1d-free-path-integral-brownian-expres}) and (\ref{eqn-brownian-cyl-measure}), where the former becomes precisely the latter if we substitute $t\mapsto -\ii t$. This is called the \textsf{Wick rotation} which will be discussed in the next subsection. Indeed, from example \ref{exam-transition-brownian} we have
  \begin{equation}
    \mb{E}[h(B_t)|B_0=x]=(\me^{-t\mn{H}_{1,0}}h)(x).
    \label{}
  \end{equation}
Thus lemma \ref{lemm-multiple-corr-intro} gives for a sequence~$0<t_1<\cdots<t_n$,
  \begin{equation}
    \mb{E}\big[ f_1(B_{t_1})\cdots f_n(B_{t_n}) \big|B_0=x\big]=(\me^{-t_1\mn{H}_{1,0}}f_1\cdots \me^{-(t_{n-1}-t_{n-2})\mn{H}_{1,0}}f_{n-1} \me^{-(t_n-t_{n-1})\mn{H}_{1,0}}f_n )(x).
    \label{}
  \end{equation}
  Now let~$V:\mb{R}\lto\mb{R}$ be a potential,~$[0,T]$ a time interval and~$t_j:= Tj/N$ as above. Setting~$f_j:=\me^{-\frac{T}{N}V}$,~$1\le j\le N-1$, and~$f_N:=\me^{-\frac{T}{N}V}h$ we obtain
  \begin{equation}
    \mb{E}\big[ \me^{-\frac{T}{N}V}(B_{T/N})\cdots \me^{-\frac{T}{N}V}(B_{T}) h(B_T) \big|B_0=x\big]=(\me^{-\frac{T}{N}\mn{H}_{1,0}}\me^{-\frac{T}{N}V}\cdots  \me^{-\frac{T}{N}\mn{H}_{1,0}}\me^{-\frac{T}{N}V}h )(x),
    \label{}
  \end{equation}
that is,
\begin{equation}
  \mb{E}\big[ \me^{-\sum_{j=1}^N V(B(t_j))(\delta t)}h(B_T) \big|B_0=x\big]=\big[(\me^{-\frac{T}{N}\mn{H}_{1,0}}\me^{-\frac{T}{N}V} )^Nh\big](x).
    \label{}
  \end{equation}
Minimizing technicality, we assume now~$V$ is bounded and continuous. Since the Brownian path~$t\mapsto B_t$ is (almost surely) continuous, we obtain in combining with lemma \ref{lemm-lie-trotter-product},
\begin{prop}
  [Feynman-Kac] \label{prop-1d-feynman-kac} Let~$V$,~$h\in L^{\infty}(\mb{R})$ and let~$V$ be continuous and non-negative. Then we have the representation
  \begin{equation}
    (\me^{t(\frac{1}{2}\Delta_{\mb{R}}-V)}h)(x)=\mb{E}\big[ \me^{-\int_{0}^{t}V(B_s)\,\dd s}h(B_t) \big|B_0=x\big],
    \label{eqn-intro-feynman-kac-main}
  \end{equation}
  for all $t\ge 0$. \hfill $\Box$
\end{prop}

\begin{lemm}
  [Lie-Trotter] \label{lemm-lie-trotter-product} Let~$A$ and~$B$ be positive self-adjoint operators on a Hilbert space $\mathcal{H}$ with domains~$\dom(A)$, $\dom(B)$ and suppose that~$A+B$ is essentially self-adjoint on the domain~$\dom(A)\cap \dom(B)$. Then for~$t>0$ we have
  \begin{equation}
    \slim_{N\to \infty}\big( \me^{-\frac{t}{N}A}\me^{-\frac{t}{N}B}\big)^N =\me^{-t(A+B)},
    \label{}
  \end{equation}
  where ``$\slim$'' denotes the limit in strong operator topology.
\end{lemm}

Written more suggestively in analogy with (\ref{eqn-feynman-hibbs-with-pot}), (\ref{eqn-intro-feynman-kac-main}) reads
\begin{equation}
  \bank{f,\me^{t(\frac{1}{2}\Delta_{\mb{R}}-V)}h}_{L^2(\mb{R},\mathcal{L}_{\mb{R}})} \heueq 
  \iint f(y)h(x)\int_{\left\{  \substack{\textrm{paths} \\ \phi:[0,T]\to\mb{R} }~\middle|~ \substack{
	     \phi(0)=y,\\ 
	     \phi(T)=x
	   }
      \right\}}^{}
      \me^{-\int_{0}^{T}\left(\frac{1}{2}|\phi'(t)|^2+V(\phi(t))\right)\,\dd t} \,\dd \mathcal{L}(\phi)\,\dd x\,\dd y.
  \label{eqn-1d-feynman-kac-heuristic}
\end{equation}
Proposition \ref{prop-1d-feynman-kac} gives a rigorous sense to this equality if we \textit{define}
\begin{equation}
  \me^{-\frac{1}{2}\int_{0}^{T}|\phi'(t)|^2\,\dd t} \,\dd \mathcal{L}(\phi)\defeq \frac{1}{\sqrt{2\pi T}}\me^{-\frac{1}{2T}|x-y|^2}\cdot \dd\mb{W}^{y,x}_{[0,T]}(\phi),
  \label{}
\end{equation}
where~$\mb{W}^{y,x}_{[0,T]}$ denotes the \textsf{Brownian bridge measure} on~$C([0,T])$ with boundary conditions~$\phi(0)=y$,~$\phi(T)=x$.

\subsection{Wick Rotation}\label{sec-intro-wick-rot}

 In this subsection we consider~$\ms{H}:=\mn{H}_{1,1}-E_0$, the ``shifted Hamiltonian'' of the harmonic oscillator, appeared in (\ref{eqn-funda-relation-har-osc-orn-uhl}). The main concern of this subsection is to relate the unitary group~$\me^{-\ii t\ms{H}}$ to the semi-group~$\me^{-t\ms{H}}$ by ``analytic continuation''. Formally, the former becomes the latter upon a substitution~$t\mapsto -\ii t$ which is usually called the \textsf{Wick rotation}. Rigorously, the question could be asked at different levels. Denote~$\mathcal{H}:=L^2(\mb{R},\mathcal{L}_{\mb{R}})$.
  \begin{enumerate}[(i)]
    \item For which vectors~$\psi\in \mathcal{H}$ (and which~$z\in\mb{C}$) does~$z\mapsto \me^{-z\ms{H}}\psi$ make sense as an~$\mathcal{H}$-valued analytic function of~$z$? Clearly, if~$\psi$ is a finite linear combination of the eigenstates (\ref{eqn-harm-osc-nth-eigenfunc-def}), then~$\sum_n \frac{|z|^n}{n!}\nrm{\ms{H}^n\psi}_{\mathcal{H}}<\infty$ for all~$z\in \mb{C}$ and~$z\mapsto\me^{-z\ms{H}}\psi$ is entire. In particular these~$\psi$ are \textsf{analytic vectors} of~$\ms{H}$, see \cite{RSim2} section X.6;
    \item For which~$z\in\mb{C}$ does~$z\mapsto \me^{-z\ms{H}}$ define an \textit{analytic family} of bounded operators (that is, an~$\mathcal{L}(\mathcal{H})$-valued analytic function of~$z$), which then necessarily satisfies~$\me^{-(z+z')\ms{H}}=\me^{-z\ms{H}}\me^{-z'\ms{H}}$ for admissible~$z$ and~$z'\in\mb{C}$? This leads to the notion of a \textsf{bounded holomorphic semi-group} (\cite{RSim2} pp.\ 248), typically defined on a ``right opening'' sector in~$\mb{C}$ centered at~$0$. Now since~$\ms{H}$ is self-adjoint and positive, it is possible to define~$\me^{-z\ms{H}}$ as such for~$\fk{Re}(z)>0$ by holomorphic functional calculus (Cauchy integral), see \cite{RSim2} theorem X.52. However we barely miss the imaginary axis. 
    \item The ``path measure'' of the Ornstein-Uhlenbeck process can more generally be seen as a probability measure on~$\mathcal{S}'(\mb{R})$, so that~$\omega\in \mathcal{S}'(\mb{R})$ gives a typical path and the ``evaluation''~$\omega\mapsto \omega(t)$ gives the random variable~$X_t$ (see example \ref{exam-orns-uhl-cont}). Can we characterize more generally properties of such a measure which, like our case, enables one to define from it a Hilbert space based on the ``law of~$\omega(0)$'' (cf.\ $L^2(\mb{R},\gamma_{1,1})$), a ``stochastic evolution'' on this Hilbert space, and extract from it a positive self-adjoint generator? In this way, we say that we have \textit{reconstructed} quantum mechanics out of the given probability measure on~$\mathcal{S}'(\mb{R})$.
  \end{enumerate}
  
  More could be said about question (ii). As the imaginary axis is only ``barely missed'' by the region of analyticity of~$\me^{-z\ms{H}}$,  it is still reasonable that for (general) fixed~$\varphi$,~$\psi\in \mathcal{H}$,~$\sank{\varphi,\me^{-\ii t\ms{H}}\psi}$ can be taken as a ``boundary value'' of the analytic function~$\sank{\varphi,\me^{-z\ms{H}}\psi}$. Let~$z=\ii w$ so that~$w\mapsto \me^{-\ii w\ms{H}}$ is analytic on the lower half-plane~$\mb{C}_-$. Note
  \begin{equation}
    \sank{\varphi,\me^{-\ii(t-\ii\eta)\ms{H}}\psi}=\int_{\mb{R}}^{} \me^{-\ii(t-\ii\eta)\lambda}\sank{\varphi,\dd E_{\ms{H}}(\lambda)\psi},
    \label{}
  \end{equation}
  where~$E_{\ms{H}}$ is spectral measure of~$\ms{H}$. In our case, since~$\ank{\varphi,\dd E_{\ms{H}}(\bullet)\psi}$ is a complex measure (with finite total variation) and~$\ms{H}$ is \textit{positive} (\textit{a fortiori} bounded from below),~$ \sank{\varphi,\me^{-\ii t\ms{H}}\psi}$ is indeed the limit of~$ \sank{\varphi,\me^{-\ii(t-\ii\eta)\ms{H}}\psi}$ as~$\eta\to 0^+$ by dominated convergence. More generally, the Fourier transform of any tempered distribution supported on~$\mb{R}_+$ is the boundary value of an analytic function on~$\mb{C}_-$ (\cite{RSim2} theorem IX.16). For example,
  \begin{equation}
    \mathcal{F}(1_{(0,\infty)})=\frac{1}{\ii(x- \ii 0)}=\frac{1}{\ii}\mm{p.v.}\Big( \frac{1}{x} \Big)+ \pi \delta_0,
    \label{}
  \end{equation}
  that is,~$\mathcal{F}(1_{(0,\infty)})$ is the boundary value of~$1/\ii w$ on~$\mb{C}_-$. In QFT, it is important to make sense of the (higher dimensional) analogue of
\begin{equation}
  W_n(t_1,\dots,t_n)\defeq \bank{\Omega_0,\mn{X}\me^{-\ii(t_n-t_{n-1})\ms{H}}\mn{X}\cdots\me^{-\ii (t_2-t_1)\ms{H}}\mn{X}\Omega_0},
  \label{eqn-wightman-1d-express}
\end{equation}
the so-called \textsf{Wightman distributions}, as the boundary value of an analytic function of \textit{several variables} defined on the \textit{forward tube} (\cite{RSim2} theorem IX.32, \cite{Sim2} pp.\ 60). Here~$\mn{X}$ is the position operator (cf.\ (\ref{eqn-1d-har-osc-pos-mom})). In that case,~$W_n$ must be taken as a distribution and (\ref{eqn-wightman-1d-express}) is understood formally. Generally, the possibility of doing the analytic continuation comes from the requirement that the \textit{joint spectrum} (\cite{Sim3} section 5.6) of the generators of space-time translations\footnote{from the unitary representation of the Poincar\'e group.} lie in the \textit{closed forward light cone} (\cite{RSim2} pp.\ 63), in particular, that the Hamiltonian be positive. That~$W_n$ could be analytically continued to the forward tube means if~$s_1<\cdots <s_n$ then
\begin{equation}
  W_n(-\ii s_1,\dots,-\ii s_n)= \bank{\Omega_0,\mn{X}\me^{-(s_n-s_{n-1})\ms{H}}\mn{X}\cdots\me^{- (s_2-s_1)\ms{H}}\mn{X}\Omega_0}.
  \label{}
\end{equation}
The function~$S_n(s_1,\dots,s_n):=W_n(-\ii s_1,\dots,-\ii s_n)$ is usually called a \textsf{Schwinger function}\footnote{This definition only works for~$s_1<\cdots <s_n$ by the requirements of analytic continuation. It could be extended symmetrically to larger domains, see \cite{Sim2} theorem II.12.}. We observe a peculiarity of the Schwinger function: let~$s_2>s_1>0$, since the semi-group~$\me^{-s\ms{H}}$ is now self-adjoint we have
\begin{equation}
  \boxed{S_4(-s_2,-s_1,s_1,s_2)=\bank{\me^{-s_1\ms{H}}\mn{X}\me^{-(s_2-s_{1})\ms{H}}\mn{X}\Omega_0,\me^{-s_1\ms{H}}\mn{X}\me^{-(s_2-s_{1})\ms{H}}\mn{X}\Omega_0}\ge 0.}
  \label{}
\end{equation}
This is an instance of a property called \textsf{reflection positivity} which we discuss briefly below.

Question (iii) underlies the so-called \textsf{Euclidean approach} to constructive QFT. It goes the other way round: start with the Euclidean, or probabilistic, objects, in what circumstances can we recover from them a ``genuine'' Minkowski QFT? This leads to the \textsf{Osterwalder-Schrader (OS) axioms}. Let~$\mu$ be a probability measure on~$\mathcal{S}'(\mb{R})$. Assume that the random variables~$(X_t)_{t\in\mb{R}}$ are well-defined by evaluation~$\omega\mapsto \omega(t)$, giving accordingly the~$\sigma$-algebras~$\mathcal{O}_{\ge t}=\sigma(X_s|s\ge t)$,~$\mathcal{O}_{\le t}=\sigma(X_s|s\le t)$.\footnote{We will use the notations~$X_t$ and~$\omega(t)$ interchangeably to avoid heavy subscripts.} We present an incomplete and simplified version for~$\mu$, referring to \cite{GJ} chapter 6 for the complete version:
\begin{description}
    \item[(OS0)] The variables~$X_t$ have finite joint moments. 
    \item[(OS2)] $\mu$ is Euclidean (translation and reflection) invariant (e.g.\ $\mb{E}_{\mu}[X_s X_t]=\mb{E}_{\mu}[X_{s-a}X_{t-a}]=\mb{E}_{\mu}[X_{-s}X_{-t}]$).
    \item[(OS3)] \textsf{Reflection Positivity (RP)}. Let~$F=F(\omega(t_1),\dots,\omega(t_N))$ be a polynomial functional of finitely many~$X_t$'s. Denote by~$\Theta F$ the ``reflection across~$t=0$'' of~$F$, namely~$(\Theta F)(\omega)=F(\omega(-t_1),\dots,\omega(-t_N))$. Then we have
\begin{equation}
  \mb{E}_{\mu}[\ol{F}(\Theta F)]\ge 0,\quad \textrm{provided}\quad F=F(\omega(t_1),\dots,\omega(t_N))\textrm{ with }t_j\ge 0.
  \label{eqn-intro-rp-inner-prod}
\end{equation}
\end{description}
\begin{exxx}[Ornstein-Uhlenbeck, following \ref{exam-orns-uhl}]
  We can characterize the Hilbert space~$L^2(\mb{R},(\tau_0)_*\mu)$ of the Ornstein-Uhlenbeck process in this new framework. Here~$(\tau_0)_*\mu=\gamma_{1,1}$ is the law of~$X_0$. Consider the map~$J_0:L^2(\mb{R},(\tau_0)_*\mu)\lto L^2(\mathcal{S}'(\mb{R}),\mu)$ defined by~$(J_0 F)(\omega):=F(\omega(0))$. Then by definition the functionals in~$\ran J_0$ are~$\sigma(X_0)$-measurable. In fact, by the Gaussianity of the Ornstein-Uhlenbeck process and the Wiener Chaos decomposition (proposition \ref{prop-wiener-chaos}),~$\ran J_0$ are \textit{precisely} the~$\sigma(X_0)$-measurable functionals in~$L^2(\mathcal{S}'(\mb{R}),\mu)$. Next we note that the latter condition in (\ref{eqn-intro-rp-inner-prod}) says~$F$ is~$\mathcal{O}_{\ge 0}$-measurable. Accordingly~$\Theta F$ is~$\mathcal{O}_{\le 0}$-measurable. Applying the \textit{Markov property} and reflection invariance of the Ornstein-Uhlenbeck process we have
  \begin{equation}
    \mb{E}_{\mu}[\ol{F}(\Theta F)]=\mb{E}_{\mu}[~\mb{E}_{\mu}[\ol{F}|\mathcal{O}_{\le 0}]\cdot \Theta F~]\overset{\mm{Markov}}{=}\mb{E}_{\mu}[~\mb{E}_{\mu}[\ol{F}|\sigma(X_0)]\cdot \Theta F~]\overset{\substack{\textrm{reflection}\\ \mm{invariance}}}{=}\mb{E}_{\mu} \bigl|\mb{E}_{\mu}[F|\sigma(X_0)]\bigr| ^2\ge 0.
    \label{}
  \end{equation}
   This shows that the Ornstein-Uhlenbeck process is reflection positive and more importantly, that for an~$\mathcal{O}_{\ge 0}$-measurable functional~$F$,~$\mb{E}_{\mu}[\ol{F}(\Theta F)]=0$ iff~$\mb{E}_{\mu}[F|\sigma(X_0)]=0$ almost surely and these are precisely the functionals in~$(\ran J_0)^{\perp}$. From the above discussions we conclude
  \begin{equation}
    L^2(\mb{R},(\tau_0)_*\mu)\cong \ol{\mathcal{E}_+/\mathcal{N}},
    \label{eqn-intro-rp-reconstruct-hilb-space}
  \end{equation}
  where~$\mathcal{E}_+:=L^2(\mathcal{S}'(\mb{R}),\mathcal{O}_{\ge 0},\mu)$,~$\mathcal{N}$ is the kernel of the semi-definite inner product (\ref{eqn-intro-rp-inner-prod}), and the closure taken w.r.t.\ (\ref{eqn-intro-rp-inner-prod}) as well. 
\end{exxx}

The recipe on the r.h.s.\ of (\ref{eqn-intro-rp-reconstruct-hilb-space}) applies more generally when neither Gaussianity nor Markov property is available. It gives a general reconstruction of the ``quantum mechanical'' Hilbert space of a QFT out of a probability measure on~$\mathcal{S}'(\mb{R}^{d+1})$ satisfying (OS2-3) (\cite{GJ} propositions 6.1.1-2). Now following \cite{GJ} section 6.1, we indicate how to reconstruct time evolution again with~$d=0$. The crucial fact to note is that if a Euclidean transformation~$\mss{T}$\footnote{Like~$\Theta$, we define the action of~$\mss{T}$ on~$L^2(\mathcal{S}'(\mb{R}^{d+1}))$ by~$(\mss{T}F)(\phi)=F(\mss{T}_*\phi)$ for~$\phi\in \mathcal{S}'(\mb{R}^{d+1})$, and~$(\mss{T}_*\phi)(f)=\phi(\mss{T}^* f)$.} maps~$\mathcal{E}_+\lto \mathcal{E}_+$ and~$\mathcal{N}\lto \mathcal{N}$, then it descends to~$\mathcal{H}:=\ol{\mathcal{E}_+/\mathcal{N}}$. A \textit{positive} time translation~$a\mapsto a+t$, denoted~$\mss{T}(t)$ for~$t\ge 0$, has this property. We denote by~$U(t):\mathcal{H}\lto \mathcal{H}$ the operator induced by~$\mss{T}(t)$. Since~$\mss{T}(t)\mss{T}(s)=\mss{T}(t+s)$,~$(U(t))_{t\ge 0}$ defines a semi-group and more precisely it will be strongly continuous and self-adjoint (\cite{GJ} theorem 6.1.3). There will be an extra, most noticeable feature: each~$U(t)$ will be a \textit{contraction} (i.e.~$\nrm{U(t)}\le 1$).
Indeed, by semi-group property, self-adjointness, and Cauchy-Schwarz on $\mathcal{H}$,
\begin{equation}
  \nrm{U(t)\psi}_{\mathcal{H}}=\ank{\psi,U(2t)\psi}^{1/2}_{\mathcal{H}}\le \nrm{\psi}^{1/2}_{\mathcal{H}}\nrm{U(2t)\psi}^{1/2}_{\mathcal{H}}\le\cdots\le \nrm{\psi}^{1-2^{-n}}_{\mathcal{H}}\nrm{U(2^n t)\psi}^{2^{-n}}_{\mathcal{H}}.
  \label{}
\end{equation}
Now observe (i) the projection~$\Pi:\mathcal{E}_+\lto \mathcal{H}$ is itself a contraction, (ii) the Euclidean transformations act unitarily on~$L^2(\mathcal{S}'(\mb{R}))$, and (iii) by definition of~$U(t)$, if~$\psi=\Pi(F)$ then~$U(2^n t)\psi=\Pi(\mss{T}(2^n t)F)$. Therefore~$\nrm{U(2^n t)\psi}_{\mathcal{H}}\le \nrm{F}_{\mathcal{E}_+}$ and we obtain the result by sending~$n\to \infty$. This contraction property of~$U(t)$, with self-adjointness, tells that its generator~$\ms{H}$, the reconstructed Hamiltonian, must be \textit{positive}. As discussed above, this is a key fact which suggests that one may proceed to reconstruct a ``genuine'' quantum theory out of~$\mathcal{H}$ and~$\ms{H}$.

\section{Fock Space and Gaussian Measures}\label{sec-intro-fock}

We start by recalling the formula for a Gaussian density on~$\mb{R}^N$. That is, given a real positive symmetric~$N\times N$ matrix~$C^{-1}$ on~$\mb{R}^N$, one has
\begin{equation}
  \int_{\mb{R}^N}^{}\me^{-\frac{1}{2}\sank{\mn{x},C^{-1}\mn{x}}} \dd\mathcal{L}_{\mb{R}^N}(\mn{x})=(2\pi)^{\frac{N}{2}}(\det C)^{\frac{1}{2}},
  \label{}
\end{equation}
and the measure
\begin{equation}
  \dd\mu_C(\mn{x})\defeq (2\pi)^{-\frac{N}{2}}(\det C^{-1})^{\frac{1}{2}}\me^{-\frac{1}{2}\sank{\mn{x},C^{-1}\mn{x}}} \dd\mathcal{L}_{\mb{R}^N}(\mn{x})
  \label{eqn-fini-dim-gaussian-measure-expr}
\end{equation}
is a Gaussian probability measure with~$C$ being its covariance matrix, namely~$C_{ij}=\mb{E}_{\mu_C}[x_ix_j]$. In the case where~$C^{-1}$ is diagonal with entries (eigenvalues)~$\lambda_1$, \dots,~$\lambda_N$, the expression becomes
\begin{equation}
  \dd\mu_C(\mn{x})=\bigg[ \prod_{i=1}^N \Big(\frac{\lambda_i}{2\pi}\Big)^{\frac{1}{2}} \bigg]
  \me^{-\frac{1}{2}\sum_i \lambda_i x_i^2}\,\dd \mathcal{L}_{\mb{R}^N}(\mn{x}).
  \label{}
\end{equation}

The purpose of this section is two-fold. In subsection \ref{sec-intro-from-bm-to-gff} we motivate the notion of a \textit{Gaussian random field} which is the infinite dimensional generalization of a \textsf{Gaussian random vector}, i.e.\ a random vector in~$\mb{R}^N$ with law (\ref{eqn-fini-dim-gaussian-measure-expr}). Examples \ref{exam-cons-brownian-motion} and \ref{exam-orns-uhl-cont} show that these are ubiquitous and natural. Then in subsection \ref{sec-intro-general-gauss-meas} we explain the general characterization of (infinite dimensional) Gaussian measures, i.e.\ laws of Gaussian random fields, by their \textit{covariance structure} and indicate one of their general constructions. They show up naturally in considerations of both the Quantum and Euclidean free fields discussed in subsections \ref{sec-intro-free-field-cyl-canon-quant} and \ref{sec-intro-dyna-massive-gff}. These are the first focus. Subsections \ref{sec-intro-fock-wiener-chaos} and \ref{sec-intro-fock-rep} constitute the second focus devoted to explaining one fact, that is, informally, any probability~$L^2$ space containing a Gaussian variable automatically ``contains'' infinitely many Bosons. The precise meaning of this statement is propositions \ref{prop-wiener-chaos}, \ref{prop-intro-fock-rep-on-l2} and definition \ref{def-fock-rep-abs} combined.

\subsection{From the Brownian Motion to the GFF}\label{sec-intro-from-bm-to-gff}
The Brownian motion (on $\mb{R}$) is a family~$(B_t)_{t\ge 0}$ of Gaussian random variables defined on some probability space~$(Q,\mathcal{O},\mb{P})$ such that~$\mb{E}[B_t]\equiv 0$ and~$\mb{E}[B_t B_s]=\min\{s,t\}$ for~$s$,~$t\ge 0$. To \textit{construct} the Brownian motion means to find a suitable probability space~$(Q,\mathcal{O},\mb{P})$ and to define the random variables~$B_t$. A natural choice for~$Q$ is the space of \textit{all possible trajectories}. Below, we construct $B_t$ up to $t=\pi/2$.

\begin{exxx}
  [construction of Wiener measure on~$C({[}0,\pi/2{]},\mb{R})$] \label{exam-cons-brownian-motion} We start with the Laplacian~$\Delta_{\mb{R}}=-\dd^2/\dd x^2$ on the open interval~$(0,\frac{\pi}{2})$, initially acting on~$C_{c}^{\infty}((0,\frac{\pi}{2}))$. Consider its self-adjoint extension~$\Delta_{\mm{W}}$ with domain
  \begin{equation}
    \dom(\Delta_{\mm{W}})\defeq \big\{ f\in W^2((0,\tts \frac{\pi}{2})) ~\big|~ f(0)=f'(\frac{\pi}{2})=0\big\},
    \label{eqn-wiener-laplacian-domain}
  \end{equation}
  where~$W^2$ denotes the second order Sobolev space (\cite{Lewin} section 2.8.3). This~$\Delta_{\mm{W}}$ is then positive, has compact resolvent (\cite{Lewin} theorem 5.25), and~$L^2$-diagonalized by the complete eigenbasis~$\{e_n\}_{n=0}^{\infty}$ where
  \begin{equation}
    e_n(t)\defeq \frac{2}{\sqrt{\pi}}\sin((2n+1)t),\quad\quad x\in (0,\tts\frac{\pi}{2}),
    \label{}
  \end{equation}
  corresponding to eigenvalues~$\lambda_n=(2n+1)^2$. We now define a random distribution~$\varphi\in \mathcal{D}'((0,\frac{\pi}{2}))$ by writing formally a Fourier series
  \begin{equation}
    B\defeq\sum_{n=0}^{\infty}\frac{\xi_n}{2n+1}e_n,
    \label{eqn-def-brownian-with-basis}
  \end{equation}
  with the random coefficients~$\{\xi_n\}$ i.i.d.\ following the standard Gaussian~$\mathcal{N}(0,1)$. What does~$\varphi$ look like? We could compute the covariances. Indeed, since~$\{\xi_n\}$ are i.i.d.\ we have formally
  \begin{align*}
    \mb{E}[B(t)B(s)]&=\sum_{n=0}^{\infty}\frac{e_n(t)e_n(s)}{(2n+1)^2}=\sum_{n=0}^{\infty}\frac{4}{\pi}\frac{\sin((2n+1)t)\sin((2n+1)s)}{(2n+1)^2} \\
    &=\sum_{n=0}^{\infty}\frac{4}{\pi}\int_{0}^{t}\cos((2n+1)x)\,\dd x\int_{0}^{s}\cos((2n+1)y)\,\dd y\\
    &=\ank{1_{(0,t)},1_{(0,s)}}_{L^2((0,\pi/2))}=\min\{t,s\},
  \end{align*}
  where the second last step is true by Plancherel's theorem and realizing that~$\tilde{e}_n:=2\pi^{-1/2}\cos((2n+1)x)$ form another basis\footnote{
  In fact, a closer scrutiny shows that~$\sum_n \int_{0}^{t}e_n'\int_{0}^{s}e_n'=\min\{t,s\}$ for any o.n.\ basis~$\{e_n'\}$ of~$L^2((0,\pi/2))$. This reflects the fact that the \textsf{Gaussian white noise}~$B'=\sum_n \xi_n e_n'$ is ``intrinsically defined'', whose law being independent of the basis, and accordingly so is~$B$ when written as~$\sum_n \xi_n \int_{0}^{t}e_n'$. However, our basis~$\{e_n\}$ is special in giving the equality (\ref{eqn-cov-brownian-test-function}).}
  of~$L^2((0,\frac{\pi}{2}))$.
  Therefore we recovered the covariance structure of Brownian motion. In fact, the expression (\ref{eqn-def-brownian-with-basis}) also shows that~$\mb{E}\snrm{\Delta_{\mm{W}}^{s/2}B}_{L^2}^2=\sum_n(2n+1)^{-2+s}<\infty$ if~$s<1$. Since~$\Delta_{\mm{W}}$ is an elliptic positive differential operator,~$\snrm{\Delta_{\mm{W}}^{s/2}B}_{L^2}\approx \nrm{B}_{W^s}$, and we see that~$B$ is almost surely in~$W^s$ for all~$s<1$. By the Sobolev embedding theorem, therefore,~$B$ is almost surely in~$C^{\frac{1}{2}-\varepsilon}$ for all~$\varepsilon>0$. We have refound Brownian motion!
  We remark that the construction above is adapted form an original construction by Paul Lévy (1948). See \cite{Comets-Meyre-2020} section 2.3 and also \cite{Kahane-1985} chapter 16.

  Pick~$f\in L^2((0,\frac{\pi}{2}))$. The expression (\ref{eqn-def-brownian-with-basis}) also enables us to work more generally with the random variables~$\ank{B,f}_{L^2}$, and~$B(t)$ could be thought of as taking~$f=\delta_t$. Indeed, its variance can be computed as
\begin{equation}
  \mb{E} \bigl|\bank{B,f}\bigr|^2
  =\sum_{n=0}^{\infty}\frac{\left|\ank{f,e_n}\right|^2}{(2n+1)^2}=\bank{f,\Delta_{\mm{W}}^{-1}f}_{L^2}.
  \label{eqn-cov-brownian-test-function}
\end{equation}
This \textit{a posteriori} shows that~$\ank{B,f}_{L^2}$ is also the probability~$L^2$-limit of the Gaussian variables~$\sum_{n=1}^N \xi_n f_n(2n+1)^{-1}$ where~$f_n:=\ank{f,e_n}$, and hence is also \textit{Gaussian}. Since~$\sank{f,\Delta_{\mm{W}}^{-1}f}\le \nrm{f}^2$,~$\ank{B,f}$ is as well the limit of~$\ank{B,\varphi_n}$ for some sequence of test functions~$\varphi_n\in C_c^{\infty}((0,\frac{\pi}{2}))$ approaching~$f$ in~$L^2$. We will see in subsection \ref{sec-intro-general-gauss-meas} that this perspective of characterizing \textit{Gaussian} random distributions, by looking at the random variables produced by testing against test functions, generalizes.
\end{exxx}

It is now not difficult to generalize the recipe (\ref{eqn-def-brownian-with-basis}) and construct Gaussian random distributions on an arbitrary closed Riemannian manifold~$(M,g)$ with the help of the discrete spectrum of~$\Delta_g$. Indeed, if~$\{e_n\}_{n=0}^{\infty}$ is the o.n.\ basis (abusing notation) of~$L^2(M)$ consisting of eigenfunctions of~$\Delta_g$ with eigenvalues~$0=\lambda_0< \lambda_1\le \lambda_2\le\cdots \le \lambda_n\le \cdots$ counted with multiplicity, we define
\begin{equation}
  \phi\defeq \sum_{n=1}^{\infty}\frac{\xi_n}{\sqrt{\lambda_n}}e_n,
  \label{eqn-intro-def-gff-eigenmodes}
\end{equation}
with~$\{\xi_n\}$ again being i.i.d.\ Gaussians following~$\mathcal{N}(0,1)$. One could then show that~$\phi$ defines a Gaussian random distribution with~$\int_{M}^{}\phi=0$ almost surely (since we removed the zero eigenvalue), and the random variables~$\ank{\phi,f}_{L^2}$,~$f\in C^{\infty}(M)$, satisfy
\begin{equation}
  \mb{E}\bigl| \bank{\phi,f}_{L^2} \bigr|^2 =\bank{f,\Delta_{g,0}^{-1}f}_{L^2},\quad\quad \mb{E} \big[ \bank{\phi,f}_{L^2} \big]=0,
  \label{eqn-intro-covariance-massless-gff}
\end{equation}
where~$\Delta_{g,0}^{-1}=\Delta_g^{-1}(\one-P_{\ker \Delta})$ denotes the Green operator. This random distribution, or random \textit{field} (and accordingly the measure on~$\mathcal{D}'(M)$), is called the \textsf{(massless) Gaussian Free Field} (GFF). The bottom eigenvalue may also be included if we consider~$\Delta_g+m^2$ instead for some~$m>0$. The result will be called the \textsf{(massive) Gaussian Free Field with mass}~$m>0$. Moreover, the almost sure Sobolev regularity of~$\phi$ could also be deduced from the same argument as in example \ref{exam-cons-brownian-motion} by resorting to the Weyl asymptotics.

\begin{exxx}
    [Ornstein-Uhlenbeck, continued] \label{exam-orns-uhl-cont} We notice the following workhorse formula:
\begin{equation}
  \frac{1}{2\pi}\int_{-\infty}^{\infty}\frac{\me^{\ii \rho s}}{\rho^2+a^2}\dd \rho =\frac{1}{2a}\me^{-a|s|},\quad\quad a>0,~s\in\mb{R}.
      \label{eqn-massive-green-function-formula}
    \end{equation}
    Comparing with (\ref{eqn-orns-uhl-cov-1d}), we notice that $\mb{E}[X_s X_t]=\sank{\delta_s,(-\frac{\dd^2}{\dd x^2}+1)^{-1} \delta_t}$. Indeed, (\ref{eqn-orns-uhl-cov-1d}) already suggests that unlike the Brownian motion, the Ornstein-Uhlenbeck process is allowed to go ``backwards in time''. By Minlos' theorem, one can indeed define a (unique) global Gaussian measure on~$\mathcal{S}'(\mb{R})$ with covariance~$(-\frac{\dd^2}{\dd t^2}+1)^{-1}$ (playing the role of $\Delta_{\mm{W}}^{-1}$ above), and the random variables~$X_t$ realized as evaluation at~$t$.
  \end{exxx}

\subsection{Fock Space and Wiener Chaos}\label{sec-intro-fock-wiener-chaos}

In this subsection we show how Gaussian random variables could be related to a construction in mathematical physics called the (Boson) \textit{Fock space}.
This elegant relation will be exploited in the next subsection and then in subsection \ref{sec-intro-dyna-massive-gff} to obtain the semi-group of the Euclidean free field as well as its Markov property. 
We mention that, among the other references cited below, a detailed development that parallels this subsection and the next could be found in \cite{Malliavin-1997} chapter I, but in a different language.

Any closed linear subspace of~$L^2(Q,\mathcal{O},\mb{P})$ of a probability space~$(Q,\mathcal{O},\mb{P})$, consisting of centered \textit{Gaussian} variables will be called a \textsf{Gaussian Hilbert space}.

\begin{deef}
   Let~$\mathcal{H}$ be a Hilbert space. Then the \textsf{Boson Fock space~$\Gamma(\mathcal{H})$ associated to~$\mathcal{H}$} is defined to be the (completed) Hilbert space direct sum
   \begin{equation}
     \Gamma(\mathcal{H})\defeq\bigoplus_{n=0}^{\infty}\sym^{\otimes n}(\mathcal{H}),\quad\quad
     \nrm{(\alpha_0,\alpha_1,\alpha_2,\dots)}_{\Gamma(\mathcal{H})}^2\defeq\sum_{i=0}^{\infty} \nrm{\alpha_i}_{\odot i}^2,
     \label{}
   \end{equation}
   where $\sym^{\otimes n}(\mathcal{H})$ denotes the $n$-fold (completed) symmetric tensor product of $\mathcal{H}$ with norm $\nrm{\cdot}_{\odot n}$, and we put~$\mathcal{H}^{\otimes 0}:=\mb{C}$ (or $\mb{R}$), and the unit length generator of $\mathcal{H}^{\otimes 0}$ is called the \textsf{vacuum}, denoted~$|0\rangle$ or~$\Omega_0$.
 \end{deef}

In concrete situations,~$\mathcal{H}$ is the Hilbert space of wave functions describing \textit{1 particle}. If the particle is a \textit{Boson}, then~$\sym^{\otimes n}(\mathcal{H})$ is the Hilbert space of the joint wave functions of~$n$ indistinguishable particles of the same type. The next two propositions suggest that each probability~$L^2$ space having Gaussian variables automatically ``contains'' infinitely many Bosons. 
 
 \begin{prop}
  [Wiener Chaos decomposition, \cite{Janson} theorem 2.6]\label{prop-wiener-chaos}
  Let~$(Q,\mathcal{O},\mb{P})$ be a probability space and~$\mathcal{H}\subset L^2(Q,\mathcal{O},\mb{P})$ a Gaussian Hilbert space. Then there is an orthogonal decomposition
  \begin{equation}
    L^2(Q,\mathcal{O}(\mathcal{H}),\mb{P})\cong \bigoplus_{n=0}^{\infty}\mathcal{H}^{{:}n{:}},
    \label{}
  \end{equation}
  where~$\mathcal{O}(\mathcal{H})$ is the~$\sigma$-algebra generated by variables in~$\mathcal{H}$,~$\mathcal{H}^{{:}n{:}}=\ol{\mathcal{P}}_n(\mathcal{H})\cap \ol{\mathcal{P}}_{n-1}(\mathcal{H})^{\perp}$, where~$\mathcal{P}_j(\mathcal{H})$ denotes the span of polynomials of random variables in~$\mathcal{H}$ of degree~$\le j$; in particular~$\mathcal{H}^{{:}0{:}}$ denotes the constants.\hfill~$\Box$
\end{prop}

 \begin{prop}
   [It\^o-Wiener-Segal isomorphism] \label{prop-ito-wiener-segal-iso}
   Let~$(Q,\mathcal{O},\mb{P})$ be a probability space and~$\mathcal{H}\subset L^2(Q,\mathcal{O},\mb{P})$ a Gaussian Hilbert space. Then the map, denoting by $\odot$ the symmetric tensor product,
   \begin{equation}
     X_1\odot \dots\odot X_n\longmapsto {:}X_1\cdots X_n{:}
     \label{eqn-ito-wiener-segal-map}
   \end{equation}
   gives a Hilbert space isomorphism~$\sym^{\otimes n}(\mathcal{H})\cong\mathcal{H}^{{:}n{:}}$. Moreover, the direct sum of these maps for each~$n$ extends to a Hilbert space isomorphism
   \begin{equation}
     \Gamma(\mathcal{H})\cong L^2(Q,\mathcal{O}(\mathcal{H}),\mb{P}).
     \label{eqn-iso-ito-wiener-segal}
   \end{equation}
   Consequently for~$X_1$, \dots, $X_n\in \mathcal{H}$,~$X_1\odot \dots\odot X_n$ and~${:}X_1\cdots X_n{:}$ are indistinguishable. 
 \end{prop}

 \begin{proof}
  It follows from the observation that the~$L^2(Q)$ inner product restricted to~$\mathcal{H}^{{:}n{:}}$ coincides with the symmetric one, by Wick's theorem (\ref{eqn-wick-inner-prod}) which we discuss below.
\end{proof}

If~$F\in \ol{\mathcal{P}}_n(\mathcal{H})$, denote by~${:}F{:}$ the projection of~$F$ onto~$\mathcal{H}^{{:}n{:}}$, and is called a \textsf{Wick ordered polynomial}; for~$F\in L^2(Q,\mathcal{O}(\mathcal{H}),\mb{P})$, denote by~$I_n(F)$ its projection onto~$\mathcal{H}^{{:}n{:}}$. Computation rules for Wick ordering generally go under the name \textsf{Wick's theorem} or \textsf{Feynman rules}. The terms in those formulas can be systematically represented by \textsf{Feynman diagrams}, and are generally expressed in terms of the \textsf{Hermite polynomials}~$h_n(x)$, defined by
\begin{equation}
  \exp\Big( zx-\frac{1}{2}z^2 \Big)=\sum_{n=0}^{\infty}\frac{z^n}{n!}h_n(x).
  \label{eqn-gen-func-hermite}
\end{equation}
We have~$h_0(x)=1$,~$h_2(x)=x^2-1$,~$h_4(x)=x^4-6x^2+3$, etc. Useful formulas include
\begin{align}
    & \mb{E}\big[{:}X_1\dots X_n{:}~{:}Y_1\dots Y_n{:}\big]=\sum_{\sigma\in\fk{S}_n}\prod_{i=1}^n \mb{E}[X_i Y_{\sigma(i)}],\label{eqn-wick-inner-prod}\\
    & {:}X^n{:}=\mb{E}\big[X^2\big]^{\frac{n}{2}} h_n\big(X\big/\mb{E}[X^2]^{\frac{1}{2}}\big),\label{eqn-hermite-wick-power}\\
    &Y{:}X_1\cdots X_n{:} ={:}YX_1\cdots X_n{:}+\sum_{j=1}^n \mb{E}\big[YX_j\big]{:}X_1\cdots\wh{X_j}\cdots X_n{:}. \label{eqn-wick-fock-rep-lemma}
\end{align}
Here~$\fk{S}_n$ is the symmetric group. Next we indicate how \textit{operators} correspond under (\ref{eqn-iso-ito-wiener-segal}).

 \begin{prop}[\cite{Janson} theorem 4.5]\label{cor-IWS-induced}
   Let~$(Q,\mathcal{O},\mb{P})$ be a probability space and~$\mathcal{H}$,~$\mathcal{K}\subset L^2(Q,\mathcal{O},\mb{P})$ two Gaussian Hilbert spaces. Denote respectively by~$\mathcal{O}(\mathcal{H})$ and~$\mathcal{O}(\mathcal{K})$ the sub-$\sigma$-algebra generated by variables in~$\mathcal{H}$ and in~$\mathcal{K}$. Let~$A:\mathcal{H}\lto \mathcal{K}$ be an operator which is a \textsf{contraction}, namely~$\nrm{A}\le 1$. Then
   \begin{equation}
     \left.
     \begin{array}{rcl}
       \Gamma(A):L^2(Q,\mathcal{O}(\mathcal{H}),\mb{P})&\lto & L^2(Q,\mathcal{O}(\mathcal{K}),\mb{P}),\\
       {:}X_1\cdots X_n{:}&\longmapsto & {:}A(X_1)\cdots A(X_n){:},
     \end{array}
     \right.
     \label{}
   \end{equation}
   for any~$X_1$, \dots,~$X_n\in \mathcal{H}$, is a bounded operator with norm~$1$. 
 \end{prop}

Lastly we point out that the operation ``$\Gamma$'' has a nice functoriality property:

 \begin{prop}
   [\cite{Janson} page 45] Following the previous proposition let~$\mathcal{L}\subset L^2(Q,\mathcal{O},\mb{P})$ be another Gaussian Hilbert space and~$B:\mathcal{K}\lto \mathcal{L}$ another operator with~$\nrm{B}\le 1$. Then
   \begin{enumerate}[(i)]
     \item $\Gamma(BA)=\Gamma(B)\Gamma(A)$ as operators~$\Gamma(\mathcal{H})\to \Gamma(\mathcal{L})$;
     \item $\Gamma(A^*)=\Gamma(A)^*:\Gamma(\mathcal{K})\to \Gamma(\mathcal{H})$,~$(\bullet)^*$ being the adjoint;
     \item $\Gamma(\one_{\mathcal{H}})=\one_{\Gamma(\mathcal{H})}$ for any~$\mathcal{H}$.
   \end{enumerate}
 \end{prop}

\subsection{Fock Representation of the CCR}\label{sec-intro-fock-rep}

 Let~$(Q,\mathcal{O},\mb{P})$ be a probability space and~$\mathcal{H}\subset L^2(Q,\mathcal{O},\mb{P})$ a Gaussian Hilbert space. We have seen above that in the case~$\mathcal{H}$ generates~$\mathcal{O}$,~$L^2(Q,\mathcal{O},\mb{P})$ is naturally identified with the (Boson) Fock space~$\Gamma(\mathcal{H})$ over~$\mathcal{H}$, and algebraic operations on Fock spaces are thus carried over to~$L^2(Q,\mathcal{O},\mb{P})$. In this subsection, we present one of the most important operations. We construct (Boson) \textsf{creation} and \textsf{annihilation operators} acting on~$L^2(Q)$, satisfying a set of \textsf{canonical commutation relations} (CCR), that generalizes (\ref{eqn-2d-har-osc-ccr-1}) and (\ref{eqn-2d-har-osc-ccr-2}). 
 In example \ref{exam-1d-fock-ccr} we apply the construction to the 1D case, rebuilding the operators (\ref{eqn-1d-har-osc-crea-ann-expres}, with~$m=\omega=1$) of the harmonic oscillator from the invariant measure~$\mathcal{N}(0,\frac{1}{2})$ of the Ornstein-Uhlenbeck process. This provides one more link between the two objects, furnishing (\ref{eqn-funda-relation-har-osc-orn-uhl}).

  Denote~$\Omega_0:=1_Q\in L^2(Q,\mathcal{O},\mb{P})$, and the (abstract) random variable corresponding to~$f\in \mathcal{H}$ by~$\phi(f)$.

      \begin{deef}
	[Fock representation on $L^2(Q)$] \label{def-fock-rep-crea-ann-ops} For~$f\in \mathcal{H}$ and each~$n\ge 0$ define\footnote{ Our conventions on the symmetrization operator~$\sym^{\otimes n}$ and inner product on~$\sym^{\otimes n}(\mathcal{H})$ are such as to make (\ref{eqn-ito-wiener-segal-map}) an isometry and the following definitions coincide with the usual one on Fock space. This convention is notably adopted in \cite{Jaffe-lec-note-2005}.}
	 \begin{align}
	\mn{a}(f)\big( {:}\phi(f_1)\cdots \phi(f_n){:} \big)&\defeq \sum_{j=1}^n \ank{f,f_j}_{\mathcal{H}}{:}\phi(f_1)\cdots\wh{\phi(f_j)}\cdots\phi(f_n){:}, \label{eqn-def-fock-ann}\\
	\mn{a}^{\dagger}(f)\big( {:}\phi(f_1)\cdots \phi(f_n){:} \big) &\defeq {:}\phi(f)\phi(f_1)\cdots\phi(f_n){:}.
	\label{eqn-def-fock-crea}
      \end{align}
      Moreover,~$\mn{a}^{\dagger}(f)\Omega_0:=\phi(f)={:}\phi(f){:}$ and~$\mn{a}(f)\Omega_0:=0$.
      \end{deef}

         \begin{exxx}\label{exam-1d-fock-ccr}
  We will consider~$Q=\mb{R}$ with the Borel~$\sigma$-algebra and~$\mb{P}=\gamma_{1,1}=\mathcal{N}(0,\frac{1}{2})$, that is
  \begin{equation}
    \dd\gamma_{1,1}(x)=\sqrt{\frac{1}{\pi}}\me^{-x^2}\dd \mathcal{L}_{\mb{R}}(x),
    \label{}
  \end{equation}
  here~$\mathcal{L}_{\mb{R}}$ being the Lebesgue measure on~$\mb{R}$. Set $\mathcal{H}=\spn\{x\}$.
  Since~$x$ generates the Borel~$\sigma$-algebra, we have~$L^2(\mb{R},\gamma_{1,1})=\Gamma(\mathcal{H})$. The normalized basis of~$\mathcal{H}$ consists of~$X=\sqrt{2}x$ and thus by (\ref{eqn-hermite-wick-power}),
  \begin{equation}
    {:}X^n{:}=h_n(X)=h_n(\sqrt{2}x).
    \label{}
  \end{equation}
  In this setting, let
  \begin{equation}
    A\defeq \mn{a}(X)=\mn{a}(\sqrt{2}x),\quad\quad A^{\dagger}\defeq \mn{a}^{\dagger}(X)=\mn{a}^{\dagger}(\sqrt{2}x),
    \label{}
  \end{equation}
  then~$[A,A^{\dagger}]=\one$ by computing with Hermite polynomials. Indeed, from (\ref{eqn-gen-func-hermite}) we have recurrence formulas
  \begin{equation}
    \left\{
    \begin{array}{l}
      h_n'(\theta)=nh_{n-1}(\theta),\\
      h_{n+1}(\theta)=\theta h_n(\theta)-h_n'(\theta).
    \end{array}
    \right.
  \end{equation}
  This, via the prescriptions (\ref{eqn-def-fock-ann}) and (\ref{eqn-def-fock-crea}), gives
  \begin{equation}
    Ag=\frac{1}{\sqrt{2}}\frac{\dd g}{\dd x},\quad\quad  A^{\dagger}g=\sqrt{2}xg-\frac{1}{\sqrt{2}}\frac{\dd g}{\dd x},
    \label{}
  \end{equation}
  with~$g$ being any finite linear combination of~$\{h_n(\sqrt{2}x)\}$. To fully recover (\ref{eqn-1d-har-osc-crea-ann-expres}), we need to conjugate by the ground state transform defined in (\ref{eqn-def-ground-state-trans}) since the natural Hilbert space for quantum mechanics is $L^2(\mb{R},\mathcal{L}_{\mb{R}})$. Doing this, we find
  \begin{equation}
    \Omega_0 A\Omega_0^{-1}=\frac{1}{\sqrt{2}}\Big( \frac{\dd}{\dd x}+x \Big)=\mn{A}_{1,1},\quad\quad
    \Omega_0 A^{\dagger}\Omega_0^{-1}=\frac{1}{\sqrt{2}}\Big( -\frac{\dd}{\dd x}+x \Big)=\mn{A}_{1,1}^{\dagger},
    \label{}
  \end{equation}
  acting on~$L^2(\mb{R},\mathcal{L}_{\mb{R}})$.
\end{exxx}

In general, by using (\ref{eqn-wick-inner-prod}) and the definitions one can show the following.

      \begin{prop}\label{prop-intro-fock-rep-on-l2}
	With the above definitions (\ref{eqn-def-fock-ann}), (\ref{eqn-def-fock-crea}) we have
	\begin{equation}
	  \bank{\mn{a}^{\dagger}(f) F,G}_{L^2(Q)}=\bank{F,\mn{a}(f) G}_{L^2(Q)}
	  \label{eqn-ccr-conjugate}
	\end{equation}
	for~$F$,~$G\in \ol{\mathcal{P}}_N(\mathcal{H})$ for some~$N<\infty$, and
	\begin{align}
	  [\mn{a}(f),\mn{a}(h)]&=[\mn{a}^{\dagger}(f),\mn{a}^{\dagger}(h)]=0, \label{eqn-ccr-commutator-0}\\
	  [\mn{a}(f),\mn{a}^{\dagger}(h)]&=\ank{f,h}_{\mathcal{H}},
	  \label{eqn-ccr-commutator}
	\end{align}
	for~$f$,~$h\in \mathcal{H}$.\hfill~$\Box$
      \end{prop}

 By (\ref{eqn-wick-fock-rep-lemma}) the random variable~$\phi(f)$ acting as a multiplication operator on each~$\mathcal{H}^{{:}n{:}}$ satisfies
  \begin{equation}
    \phi(f)=\mn{a}^{\dagger}(f)+\mn{a}(f).
    \label{}
  \end{equation}
  Using further Wick formulas, one can show (\cite{Sim2} pp.\ 24) that the multiplication by~${:}\phi(f)^n{:}$ is accordingly
  \begin{equation}
    {:}\phi(f)^n{:}=\sum_{p=0}^n \binom{n}{p}\mn{a}^{\dagger}(f)^{n-p}\mn{a}(f)^p.
    \label{eqn-wick-binomial-crea-left-ann-right}
  \end{equation}
  Therefore, the Wick ordering~${:}\bullet{:}$ amounts to \textit{moving annihilation operators to the right}, and \textit{creation operators to the left}. In regard to the important remark \ref{rem-necessary-crea-ann-har-osc-1d}, we notice that a common domain for all the~$\mn{a}$,~$\mn{a}^{\dagger}$'s is
\begin{equation}
  \mathcal{D}_{\mm{Fk}}\defeq \spn\big\{f\in L^2(Q)~\big|~f\in \mathcal{H}^{{:}N{:}}\textrm{ for some }N<\infty \big\},
  \label{}
\end{equation}
which by proposition \ref{prop-wiener-chaos} is \textit{dense} in~$L^2(Q)$. We end our discussion with the following definition.

\begin{deef}[\cite{GJ} definition 6.3.2] \label{def-fock-rep-abs}
  Let~$\mathcal{S}$ be a real pre-Hilbert space with inner product~$\ank{\cdot,\cdot}_{\mathcal{S}}$ and~$\mss{H}$ a complex Hilbert space. Then a \textsf{Fock representation of the CCR} over~$\mathcal{S}$ on~$\mss{H}$ is a pair of linear maps
  \begin{equation}
    \left.
    \begin{array}{rcl}
      \mn{a},\mn{a}^{\dagger}:\mathcal{S} &\lto & 
      \left\{  
      \begin{array}{c}
	(\textrm{\small Unbounded}) \\ \textrm{\small Operators on }\mss{H}
      \end{array}
    \right\} \\ [+15pt]
    f&\longmapsto & \mn{a}(f),\mn{a}^{\dagger}(f).
    \end{array}
    \right.
    \label{}
  \end{equation}
  such that
  \begin{enumerate}[(i)]
    \item there is a unit vector~$\Omega\in \mss{H}$ with~$\mn{a}(f)\Omega=0$ for all~$f\in \mathcal{S}$;
    \item all of~$\mn{a}(f)$ and~$\mn{a}^{\dagger}(f)$,~$f\in \mathcal{S}$, have the common domain
      \begin{equation}
	\mss{D}_{\mm{Fk}}=\spn \big\{\mn{a}^{\dagger}(f_1)\cdots \mn{a}^{\dagger}(f_n)\Omega~\big|~f_1,\cdots,f_n\in \mathcal{S},~n\in\mb{Z}_{\ge 0} \big\}
	\label{}
      \end{equation}
      which is dense in~$\mss{H}$;
    \item we have~$\mn{a}(f)\mss{D}_{\mm{Fk}}\subset \mss{D}_{\mm{Fk}}$,~$\mn{a}^{\dagger}(f)\mss{D}_{\mm{Fk}}\subset \mss{D}_{\mm{Fk}}$, (\ref{eqn-ccr-conjugate}), (\ref{eqn-ccr-commutator-0}), and (\ref{eqn-ccr-commutator}) for all~$f$,~$h\in \mathcal{S}$,~$F$,~$G\in \mss{D}_{\mm{Fk}}$, with~$\mss{H}$ in place of~$L^2(Q)$ and~$\mathcal{S}$ in place of~$\mathcal{H}$.
  \end{enumerate}
\end{deef}

  \begin{def7}\label{rem-fock-rep-unique}
  The Fock representation over the same real pre-Hilbert space is unique up to unitary equivalence (\cite{GJ} theorem 6.3.4).
\end{def7}

\begin{def7}
For parallel results on Fermions, see \cite{CHP} section 2. These are carried out in the context of \textit{non-commutative} probability theory.
\end{def7}

\subsection{General Gaussian Measures}\label{sec-intro-general-gauss-meas}

In this subsection we discuss general Gaussian measures on (real) separable Fr\'echet spaces: they are uniquely determined by their \textit{covariance structures} (cf.\ (\ref{eqn-intro-covariance-massless-gff}), assuming they are all centered), and on the other hand, on a special class of spaces we could always construct a Gaussian measure with a given covariance structure (Minlos' theorem). Readers focused on quantization may take these two statements for granted and skip this subsection. The main reference is \cite{Bogachev}, but see also \cite{HairerSPDE} section 4 for a shorter and lucid account.

\begin{deef}\label{def-gaussian-measure}
  A Borel probability measure~$\gamma$ on a locally convex topological vector space~$\mathcal{X}$ is called (centered) \textsf{Gaussian} if every linear functional~$f\in \mathcal{X}^*$,~$f:\mathcal{X}\lto \mb{R}$, is a (centered) Gaussian random variable on~$\mathcal{X}$. 
\end{deef}

In the general situation we will also denote the random variable represented by the linear functional~$f$ by~$\phi(f)$ or~$\ank{\phi,f}$, thinking of~$\phi$ as a random vector following~$\gamma$ in~$\mathcal{X}$. An important tool for identifying Gaussian measures is the (inverse) Fourier transform, more frequently called the \textsf{charateristic function},~$\wh{\gamma}:\mathcal{X}^*\lto \mb{C}$, defined as
\begin{equation}
  \wh{\gamma}(f)\defeq \int_{\mathcal{X}}^{}\me^{\ii \phi(f)}\dd\gamma(\phi).
  \label{}
\end{equation}
This is useful because of the following property.
\begin{lemm}[\cite{Bogachev} theorem 2.2.4, proposition A.3.18]  \label{lemm-gaus-mea-char-func}
  Let~$\mathcal{X}$ be a real separable Fr\'echet (or Banach) space, then any two Borel measures on~$\mathcal{X}$ with equal characteristic functions coincide. In particular, a Borel probability measure~$\gamma$ on~$\mathcal{X}$ is a centered Gaussian measure if and only if
  \begin{equation}
    \wh{\gamma}(f)=\me^{-\frac{1}{2}\mb{E}_{\gamma}[\phi(f)^2]}
    \label{eqn-gaus-mea-char-func}
  \end{equation}
  for all~$f\in \mathcal{X}^*$, where~$\mb{E}_{\gamma}$ is the expectation with respect to~$\gamma$.
\end{lemm}

\begin{corr}
    Let~$\mathcal{X}$ be a real separable Fr\'echet (or Banach) space, then any two centered Borel Gaussian measures on~$\mathcal{X}$ with the same \textsf{covariance structure} must coincide; that is, if
   \begin{equation}
     \mb{E}_{\gamma}[\phi(f)\phi(h)]=\mb{E}_{\tilde{\gamma}}[\phi(f)\phi(h)]
     \label{}
   \end{equation}
   for all~$f$,~$h\in \mathcal{X}^*$, then~$\gamma=\tilde{\gamma}$. \hfill~$\Box$
\end{corr}

 Now we turn to the existence result which applies to a class of Fr\'echet topological vector spaces called \textsf{nuclear spaces}. While postponing their definition to appendix \ref{sec-app-nuclear-space}, we point out that the two most typical examples of nuclear spaces are $\mathcal{S}(\mb{R}^d)$, the Schwartz functions on $\mb{R}^d$, and $C^{\infty}(M)$, the smooth functions on a closed smooth manifold. Their topological duals are accordingly $\mathcal{S}'(\mb{R}^d)$ and $\mathcal{D}'(M)$. The main reference is \cite{GV} section IV.3.2 theorem 2 (chapter I for nuclear space). Helpful information can be found in \cite{BHL2} section 1.4.2, \cite{BS} sections 5.10-11, \cite{GJ} appendix A.4.

\begin{prop}
  [Minlos] \label{prop-minlos-theo} Let~$\mathcal{H}_{\infty}$ be a nuclear space and~$\mathcal{H}_{-\infty}$ its topological dual space. Then given a continuous positive definite Hermitian bilinear form (indeed, an \textsf{inner product})~$\mathcal{C}$ on~$\mathcal{H}_{\infty}\times \mathcal{H}_{\infty}$, there exists a unique (centered) Borel Gaussian probability measure~$\mu_{\mathcal{C}}$ on~$\mathcal{H}_{-\infty}$ such that
  \begin{equation}
    \mb{E}_{\mu}[\phi(f)\phi(h)]=\mathcal{C}(f,h)
    \label{}
  \end{equation}
  for every~$f$,~$h\in \mathcal{H}_{\infty}$.
\end{prop}

\begin{proof}
    See \cite{BHL2} section 1.4.2, \cite{GJ} appendix A.4.
\end{proof}

\section{Quantization on the Cylinder $\mb{R}\times \mb{S}^1$}\label{sec-intro-field-2d}

In this section we traverse through the second line of development in the map \ref{fig-intro-road-map}, for the free real scalar field over the circle, taken as space, and the cylinder, taken as space-time. This is the simplest and most fundamental type of field theory. Most of subsections \ref{sec-intro-class-field-theory} and \ref{sec-intro-free-field-cyl-canon-quant} are standard introductory graduate-level materials in the physics department, see e.g.\ \cite{PS} and \cite{Coleman}. Our presentation, however, will have slightly different flavor and emphasis, oriented towards the aspects that could be mathematically ``rigorized''. In addition, in subsection \ref{sec-intro-free-field-cyl-canon-quant} (cf.\ (\ref{eqn-slice-gauss-measure-expression})) we point out a precise relationship between the ``ground state'' of the free field on the Minkowski cylinder with a Gaussian measure on $\mathcal{D}'(\mb{S}^1)$ analogous to the relationship (\ref{eqn-orns-uhl-ground-state-gauss-meas}) in the 1D case.

\subsection{From Particles to Fields}\label{sec-intro-class-field-theory}

In this subsection we describe the classical field theory of a free real scalar field over the circle~$\mb{S}^1$. We shall say that a typical \textsf{field configuration}~$\varphi$ belongs to $\mm{Map}(\mb{S}^1,\mb{R})$, where by writing ``Map'' we mean functions on which we deliberately do not specify the regularity, as the purpose of this subsection is rather heuristic. More precisely, we will proceed by the following analogy with the point particle in~$\mb{R}^d$ described in subsection \ref{sec-intro-class-mec}.
\begin{center}
  \begin{tabular}[]{|c|c|c|}
    \hline 
    & \textbf{\textsf{real scalar field}} & \textbf{\textsf{particle in~$\mb{R}^d$}} \\ \hline
    configuration space & $\mm{Map}(\mb{S}^1,\mb{R})$ & $\mb{R}^d$ \\ \hline
    \begin{tabular}[]{c}
      ``position-velocity''\\
       space
    \end{tabular} & $T\mm{Map}(\mb{S}^1,\mb{R})\cong \mm{Map}(\mb{S}^1,\mb{R})\times \mm{Map}(\mb{S}^1,\mb{R})$ & $T\mb{R}^d$ \\ \hline
    phase space
    & $T^*\mm{Map}(\mb{S}^1,\mb{R})\cong \mm{Map}(\mb{S}^1,\mb{R})\times \mm{Map}(\mb{S}^1,\mb{R})$ & $T^*\mb{R}^d$ \\ \hline 
\begin{tabular}[]{c}
      index for\\
       ``coordinate''
     \end{tabular} & $\theta\in\mb{S}^1$ or $n\in\mb{N}$ (see example \ref{exam-intro-class-field-fourier}) & $1\le i\le d$ \\ \hline
  \end{tabular}
\end{center}

\begin{def7}
The real scalar field is a simplest type of field that one can consider. The ``configuration'' spaces of other types of fields (over some manifold~$M$) include the section space of a complex vector bundle over~$M$ (e.g.\ for \textsf{spinor fields}), the space~$\mm{Map}(M,N)$ of maps into a certain Riemannian manifold~$N$ ($\sigma$\textsf{-models}), or space of connections on a principal bundle associated to a compact Lie group~$G$ over~$M$ (\textsf{gauge fields}).
 Moreover, different types of fields can be ``coupled'' (e.g.\ spinor field to a gauge field). We refer to \cite{Rudolph-Schmidt-2017-ii} chapter 7 for some detailed discussion of modern field theory models in mathematical language. 
 \end{def7}
 
Now we give the defining data of the real scalar free field as follows.
\begin{align}
&\textrm{\textsf{Lagrangian density}}: &\mss{L}(\varphi,\dot{\varphi})(\theta)&\defeq \frac{\beta}{4\pi} \big( \dot{\varphi}(\theta)^2-(\partial_{\theta}\varphi)(\theta)^2-m^2\varphi(\theta)^2 \big),\\
     &\textrm{\textsf{Lagrangian}}: & L(\varphi,\dot{\varphi})&\defeq \int_{\mb{S}^1}^{}\mss{L}(\varphi,\dot{\varphi})(\theta)\,\dd \theta=\frac{\beta}{4\pi}\int_{\mb{S}^1}^{}\big( \dot{\varphi}(\theta)^2-(\partial_{\theta}\varphi)(\theta)^2-m^2\varphi(\theta)^2 \big)\,\dd\theta\\
  &\textrm{\textsf{Action over }}[0,T]: & S_{[0,T]}^{\mm{free}}(\phi)&\defeq \frac{\beta}{4\pi}\int_{0}^{T}\int_{\mb{S}^1}^{}\big( (\partial_t\phi)^2-(\partial_{\theta}\phi)^2 -m^2\phi^2\big)\,\dd\theta\,\dd t.
  \label{}
\end{align}
Here $m\ge 0$, $\beta>0$ are some fixed parameters, and $\phi:[0,T]\lto \mm{Map}(\mb{S}^1,\mb{R})$ or equivalently $[0,T]\times\mb{S}^1\lto\mb{R}$, is the trajectory through which a configuration in $\mm{Map}(\mb{S}^1,\mb{R})$ evolves over $[0,T]$. 

Below, we discuss this model for two choices of $m$, $\beta$ with different purposes: in example \ref{exam-intro-class-field-klein-gordon} we let $m>0$, $\beta=2\pi$ and obtain the Hamiltonian picture and Newtonian equation of motion, which results in an analogue of the \textit{Klein-Gordon} equation on the circle; then in example \ref{exam-intro-class-field-fourier} we describe this model in a special set of coordinates on $\mm{Map}(\mb{S}^1,\mb{R})$, in the case $m=0$, which prepares for the next subsection.

\begin{exxx}
    [$m>0$, $\beta=2\pi$, Klein-Gordon]\label{exam-intro-class-field-klein-gordon}
 To transit to the Hamiltonian picture, the first step is to find an analogue of the ``canonical momentum'' of the point particle in $\mb{R}^d$. Before doing that, it would be beneficial to go back and make the analogy more precise. The convenient point of view here is to treat~$\varphi(\theta)$ as an ordinary coordinate for~$\varphi$ just as~$x^i$ is for~$\mn{x}\in\mb{R}^d$, and taking the~$\theta$-derivative on~$\mm{Map}(\mb{S}^1,\mb{R})$ could be compared to a linear map~$T$ on~$\mb{R}^d$. Now a precise analogy of our Lagrangian for a particle in~$\mb{R}^d$ would be
 \begin{equation}
   L(\mn{x},\mn{v})=\frac{1}{2}\sum_i \big[(v^i)^2 -((T\mn{x})^i)^2 -m^2(x^i)^2\big].
   \label{}
 \end{equation}
 The good news is that whatever~$T$ may be,~$T\mn{x}$ does not involve any~$v^i$. So the canonical momentum conjugate to the coordinate~$x^i$ is as usual~$p_i=\partial L/\partial (v^i) =v^i$. Going back to our scalar field, we obtain by direct analogy the \textsf{momentum density}~$\pi(\theta)=\dot{\varphi}(\theta)$. Now by comparing (\ref{eqn-ham-class-part-potential}) with (\ref{eqn-lag-general-potential}), and then with our case, we also write down by direct analogy the \textsf{Hamiltonian}
\begin{equation}
  H(\varphi,\pi)=\frac{1}{2}\int_{\mb{S}^1}^{}\big( 
  |\pi(\theta)|^2+|\partial_{\theta}\varphi(\theta)|^2 +m^2|\varphi(\theta)|^2 \big) \,\dd\theta.
  \label{}
\end{equation}
Next, what is the Hamiltonian equation of motion? To continue with our analogy, we now need to find the infinite dimensional gradient of the potential
\begin{equation}
  V(\varphi)=\frac{1}{2}\int_{\mb{S}^1}^{}(|\partial_{\theta}\varphi (\theta)|^2+m^2\varphi^2(\theta))\,\dd \theta.
  \label{}
\end{equation}
To proceed, we notice that the~$L^2$ inner product on~$\mm{Map}(\mb{S}^1,\mb{R})$ now plays the role of the standard inner product on~$\mb{R}^d$ (more precisely, the Riemannian metric on the corresponding tangent spaces). Therefore, by definition of the ``gradient'',
\begin{equation}
  \ank{\nabla V(\varphi),h}_{L^2(\mb{S}^1)}\defeq (\DD_h V)(\varphi)= \lim_{\varepsilon\to 0}\frac{1}{\varepsilon}(V(\varphi+\varepsilon h)-V(\varphi)),
  \label{}
\end{equation}
where~$\DD_h V$ denotes the ``directional derivative'' in the direction of the ``test function''~$h\in \mm{Map}(\mb{S}^1,\mb{R})$. Applying the Green-Stokes formula we get
\begin{equation}
  V(\varphi+\varepsilon h)-V(\varphi) =\varepsilon\int_{\mb{S}^1}^{}\big((-\partial_{\theta}^2 \varphi)h +m^2 \varphi h\big)\,\dd \theta+ \mathcal{O}(\varepsilon^2).
  \label{}
\end{equation}
Therefore,~$\nabla V(\varphi)=(-\partial_{\theta}^2 +m^2)\varphi$. We could then proceed to write down the Hamiltonian as well as Newtonian equations. The Newtonian equation is
 \begin{equation}
   \Big( \frac{\partial^2}{\partial t^2} - \frac{\partial^2}{\partial \theta^2}+m^2 \Big) \varphi=0,
   \label{}
 \end{equation}
 which is an analogue of the \textsf{Klein-Gordon} equation on the circle. Lastly, we mention that the ``value'' of~$\nabla V(\varphi)$ at~$\theta\in \mb{S}^1$ corresponds heuristically to the case~$h=\delta_{\theta}$, the delta function supported at~$\theta$. The commonly used notation is
 \begin{equation}
   \frac{\delta V}{\delta \varphi(\theta)}\defeq (\DD_{\delta_{\theta}} V)(\varphi).
   \label{}
 \end{equation}
 The~$\DD_h V$ and~$\delta V/\delta \varphi(\theta)$ are usually called \textsf{functional derivatives} and could be made rigorous under various assumptions. See \cite{GJ} section 9.1, \cite{Dimock2} section 12.1.3 or \cite{Folland-2008} pp.\ 268-269 for details.
 \end{exxx}

\begin{exxx}
    [$m=0$, Fourier coordinates]\label{exam-intro-class-field-fourier}
In the previous example we exploited an analogy with particle mechanics to analyze the real scalar field, by treating~$\varphi(\theta)$ as a ``coordinate'' for~$\varphi\in \mm{Map}(\mb{S}^1,\mb{R})$. For the purpose of quantization which we do in the next subsection, another set of coordinates will prove convenient. These are given by the Fourier decomposition. For precision we now stress that our circle has perimeter $2\pi$, and $\theta$ is the arc length parameter. Given the trajectory~$\phi:[0,T]\lto \mm{Map}(\mb{S}^1,\mb{R})$, we write
\begin{equation}
  \phi(t,\theta)=\sum_{n=-\infty}^{\infty}\phi_n(t)\me^{\ii n\theta},\quad\quad \textrm{with }\phi_{-n}=\ol{\phi_n},
  \label{eqn-fourier-decomp-field-cylinder}
\end{equation}
as~$\phi$ is real-valued\footnote{
  Alternatively, we could also write~$\phi(t,\theta)=\phi_0(t)+2\sum_{n>0}(x_n(t)\cos(n\theta)-y_n(t)\sin(n\theta))$ with~$\phi_0(t)\in\mb{R}$,~$x_n(t)=\fk{Re}(\phi_n(t))$ and~$y_n(t)=\fk{Im}(\phi_n(t))$.}. Now we make an important observation: thanks to the fact that~$\{\me^{\ii n\theta}\}_{n\in\mb{Z}}$ is orthogonal\footnote{Note $\snrm{\me^{\ii n\theta}}_{L^2}=2\pi$.} in~$L^2(\mb{S}^1)$, the \textit{action} now also decomposes into a sum:
  \begin{align}
    S_{[0,T]}^{\mm{free}}(\phi)&=\frac{\beta}{2}\int_{0}^{T}|\phi_0'(t)|^2\,\dd t +\beta\sum_{n>0}\int_{0}^{T}\big(|\phi_n'(t)|^2 -n^2|\phi_n(t)|^2\big)\,\dd t \label{eqn-class-free-field-fourier-decomp-action}\\
    &\defeq S_{[0,T]}^{(0)}(\phi_0)+\sum_{n>0}S_{[0,T]}^{(n)}(\phi_n),
    \label{}
  \end{align}
where by the latter relation in (\ref{eqn-fourier-decomp-field-cylinder}) the two terms involving the modes $\phi_{-n}$ and $\phi_n$ are combined into one in (\ref{eqn-class-free-field-fourier-decomp-action}). Therefore,~$\phi$ minimizes~$S_{[0,T]}^{\mm{free}}$ if and only if each~$\phi_k$,~$k\ge 0$, minimizes~$S_{[0,T]}^{(k)}$ individually. This means that for the \textit{classical} field $\phi$, the evolution over time of the space-slice configuration~$\phi(t,\bullet)$ under~$S^{\mm{free}}$ is equivalent to the evolution of the sequence of numbers~$\{\phi_n(t)\}_{n\ge 0}$ in which each~$\phi_n(t)$ evolves under~$S^{(n)}$ separately. 
  Note, in comparing with (\ref{eqn-lag-general-potential}), that~$S^{(0)}$ describes a free particle with mass~$\beta$ moving in~$\mb{R}$, and~$S^{(n)}$,~$n\ge 1$, the harmonic oscillator with mass~$2\beta$ and angular frequency~$\omega_n=n$ on $\mb{C}=\mb{R}^2$.
  \end{exxx}

\subsection{Free Field on the Cylinder = Infinitely Many Harmonic Oscillators}\label{sec-intro-free-field-cyl-canon-quant}
Having described the classical mechanics of the (massless) real scalar free field~$\phi$ on~$[0,T]\times\mb{S}^1$, we will now try to quantize it. To begin with, we emphasize that we have identified heuristically\footnote{We remark here that $\varphi_n$ are not the actual Fourier coefficients, $\sqrt{2\pi}\varphi_n$ are, and we have~$\ank{\varphi,f}_{L^2(\mb{S}^1)}=(2\pi)[\varphi_0 f_0+2\sum_{n\ge 1}\fk{Re}(\ol{\varphi_n} f_n)]$ for another real ``test function'' $f$ written in the form (\ref{eqn-fourier-decomp-field-cylinder}).}
  \begin{equation}
    \left.
    \begin{array}{rcl}
     \mathcal{F}_\mb{R}: \mm{Map}(\mb{S}^1,\mb{R}) &\xlongrightarrow{\sim} & \mb{R}\times\mb{C}^{\mb{N}},\\
      \varphi &\longmapsto& (\varphi_0,\varphi_1,\varphi_2,\dots),
    \end{array}
    \right.
    \label{}
  \end{equation}
  where~$\varphi_i$ are the Fourier coefficients of~$\varphi$ as in (\ref{eqn-fourier-decomp-field-cylinder}). To quantize a system, two basic ingredients are the Hilbert space and the operators of time evolution. In analogy with the 1D case, the Hilbert space should be
  \begin{align}
    \mathcal{H}_{\mb{S}^1}&\defeq L^2(\mm{Map}(\mb{S}^1,\mb{R}),\mathcal{L}_{\mm{Map}(\mb{S}^1,\mb{R})})\heueq L^2(\mb{R}\times\mb{C}^{\mb{N}},\mathcal{L}_{\mb{R}\times\mb{C}^{\mb{N}}})\\
    &\heueq L^2(\mb{R},\mathcal{L}_{\mb{R}})\otimes \bigotimes_{i=1}^{\infty}L^2(\mb{C},\mathcal{L}_{\mb{C}})_i\defeq \mathcal{H}_0\otimes \bigotimes_{i=1}^{\infty}\mathcal{H}_i
    \label{eqn-circle-hilb-space-leb-repre}
  \end{align}
  where~$\mathcal{L}_{\mb{R}\times\mb{C}^{\mb{N}}}$ denotes the non-existent Lebesgue measure on~$\mb{R}\times\mb{C}^{\mb{N}}$, and~$L^2(\mb{C},\mathcal{L}_{\mb{C}})_i$ is the individual Hilbert space associated to the component~$\varphi_i$. The recipe for time evolution, in one way, is given by the so-called \textsf{canonical quantization}:  
  \begin{framed}\noindent We assume that the coefficient~$\phi_0(t)$ evolves as a \textit{quantum} free particle of mass~$\beta$ in~$\mb{R}$, and~$\phi_n(t)$,~$n\ge 1$, as a \textit{quantum} harmonic oscillator of mass~$2\beta$ and angular frequency~$\omega_n=n$ in~$\mb{C}=\mb{R}^2$, and that $\phi_i$ evolves independently from $\phi_j$ for $i\ne j$.\end{framed}

   In analogy with (\ref{eqn-2d-har-osc-hamil-tens-prod}), we find that the \textsf{Hamiltonian} of~$\phi$ on~$\mathcal{H}_{\mb{S}^1}$ should be
  \begin{equation}
    \mn{H}_{\mm{fr},\beta}\defeq \mn{H}_{\beta,0}+\sum_{n=1}^{\infty}\mn{H}_{2\beta, n},
    \label{eqn-sum-expression-free-ham}
  \end{equation}
  where each~$\mn{H}_{2\beta, n}$ is defined as in (\ref{eqn-def-1d-harm-osc-hamil}), and its action extended to~$\mathcal{H}_{\mb{S}^1}$ by imposing
  \begin{equation}
    \mn{H}_{2\beta, n}\defeq \one_{\mathcal{H}_0}\otimes \cdots\otimes \one_{\mathcal{H}_{n-1}}\otimes (\mn{H}_{2\beta, n}|_{\mathcal{H}_n})\otimes \one_{\mathcal{H}_{n+1}}\otimes \cdots,
    \label{}
  \end{equation}
  and likewise for $\mn{H}_{\beta,0}$.

   Next, we try to postulate a ``ground state'' for~$\mn{H}_{\mm{fr},\beta}$ and relate it to a Gaussian measure on $\mm{Map}(\mb{S}^1,\mb{R})$, preliminarily, in the spirit of (\ref{eqn-orns-uhl-ground-state-gauss-meas}). First we consider the components $\mn{H}_{2\beta, n}$ with $n>0$. From example \ref{exam-2d-har-osc} we know that~$\mn{H}_{2\beta, n}$ has the ground state
  \begin{equation}
    \Omega_{(n)}(x,y)\defeq \Big(\frac{2n\beta}{\pi}\Big)^{\frac{1}{2}}\me^{-n\beta (x^2+y^2)}.
    \label{}
  \end{equation}
In analogy with the case of finite tensor product, the reasonable ``ground state'' for~$\mn{H}_{\mm{fr},\beta}$ restricted to~$\bigotimes_{i\ge 1}\mathcal{H}_i$ is therefore
  \begin{align}
    \Omega_{0,\beta}^{\mm{os}}(\varphi_1,\varphi_2,\cdots)&\defeq \prod_{n=1}^{\infty}\Big(\frac{2n\beta}{\pi}\Big)^{\frac{1}{2}}\me^{-n\beta (x_n^2+y_n^2)}=\bigg[ \prod_{n=1}^{\infty}\Big(\frac{2n\beta}{\pi}\Big)^{\frac{1}{2}}\bigg]
    \exp\Big(-\frac{\beta}{2}\sum_{n\in\mb{Z}}|n|\cdot|\varphi_n|^2 \Big)\\
    &=\bigg[ \prod_{n=1}^{\infty}\Big(\frac{2n\beta}{\pi}\Big)^{\frac{1}{2}}\bigg]
    \me^{-\frac{1}{2}\cdot \frac{\beta}{2\pi}\sank{\varphi,\mn{D}_{\mb{S}^1}\varphi}_{L^2(\mb{S}^1)}}
    \label{}
  \end{align} 
  where~$\mn{D}_{\mb{S}^1}=(\Delta_{\mb{S}^1})^{\frac{1}{2}}$, and~$\varphi_n=x_n+\ii y_n$. This suggests the following expression for a Gaussian probability measure on~$\mm{Map}(\mb{S}^1,\mb{R})$:
\begin{equation}
  \dd\mu_{\mn{D},\mb{S}^1}(\varphi)\defeq \dd\Big[\bigotimes_{n=1}^{\infty} \gamma^{\mb{C}}_{2\beta,n}\Big](\varphi)=\big(\textrm{``}\det\textrm{''}(4\beta \mn{D}_{\mb{S}^1})\big)^{\frac{1}{2}}\me^{-\frac{1}{2}\cdot \frac{\beta}{2\pi}\sank{\varphi,2\mn{D}_{\mb{S}^1}\varphi}_{L^2(\mb{S}^1)}}
  \prod_{n=1}^{\infty}\frac{\dd \mathcal{L}_{\mb{C}}(\varphi_n)}{2\pi}.
  \label{eqn-slice-gauss-measure-expression}
\end{equation}
where we write in short hand~$\mathcal{N}_{\mb{C}}(0,(4n\beta)^{-1}):=\gamma^{\mb{C}}_{2\beta,n}$.
This measure has covariance\footnote{
  The discrepancy of a factor of 2 between what's in the exponential of (\ref{eqn-slice-gauss-measure-expression}) and (\ref{eqn-slice-field-covariance}) as compared to the finite dimensional version (\ref{eqn-fini-dim-gaussian-measure-expr}) comes from the fact that the quadratic form~$\ank{\varphi,2\mn{D}_{\mb{S}^1}\varphi}_{L^2(\mb{S}^1)}$ is a priori defined on~$\mm{Map}(\mb{S}^1,\mb{C})$, but only restricted to~$\mm{Map}(\mb{S}^1,\mb{R})$.}
\begin{equation}
  \mb{E}\bigl|\ank{\varphi,f}_{L^2(\mb{S}^1)}\bigr|^2 =\frac{4\pi}{\beta}\bank{f,(2\mn{D}_{\mb{S}^1})^{-1}f}_{L^2(\mb{S}^1)},
  \label{eqn-slice-field-covariance}
\end{equation}
for real test functions~$f\in C^{\infty}(\mb{S}^1)$. Accordingly from example \ref{exam-1d-free-part} we know that the ``ground state'' of the free Hamiltonian~$\mn{H}_{\beta,0}$ on the zero mode is simply the function~$1$. Thus~$|1|^2$ gives nothing but the Lebesgue measure. Therefore an infinite dimensional analogue of the \textsf{ground state transform} (cf.\ (\ref{eqn-def-ground-state-trans})) gives a formal isomorphic representation
\begin{equation}
  \mathcal{H}_{\mb{S}^1}\cong L^2(\mb{R},\mathcal{L}_{\mb{R}})\otimes \bigotimes_{n=1}^{\infty}L^2(\mb{C},\gamma^{\mb{C}}_{2\beta,n}).
  \label{eqn-circle-hilb-space-gaussian-rep}
\end{equation}
 Note that unlike (\ref{eqn-circle-hilb-space-leb-repre}), (\ref{eqn-circle-hilb-space-gaussian-rep}) can actually be given a rigorous meaning since the second component makes sense as the~$L^2$ space over $\mb{C}^{\mb{N}}$ equipped with a Gaussian measure.

Each~$L^2(\mb{C},\gamma^{\mb{C}}_{2\beta,n})$ now carries a Fock representation of the CCR as suggested by example \ref{exam-1d-fock-ccr} in the 1D case. For conformal purposes it turns out more convenient to write the representation in complex coordinates as introduced in example \ref{exam-2d-har-osc}. We thus define the (ground state conjugated) operators
\begin{equation}
  \mn{a}_{2\beta}(n)\defeq \frac{1}{\sqrt{2\beta n}}\Big( \frac{\partial}{\partial \ol{\varphi}_n}+n\beta \varphi_n\Big),\quad\quad
  \ol{\mn{a}}_{2\beta}(n)\defeq \frac{1}{\sqrt{2\beta n}}\Big( \frac{\partial}{\partial \varphi_n}+n\beta \ol{\varphi}_n\Big).
  \label{eqn-free-field-crea-anni-complex}
\end{equation}
Then we would have~$\mn{H}_{2\beta,n}=n(\mn{a}^{\dagger}_{2\beta}(n)\mn{a}_{2\beta}(n)+\ol{\mn{a}}^{\dagger}_{2\beta}(n)\ol{\mn{a}}_{2\beta}(n)+1)$ and
\begin{equation}
  \mn{H}_{\mm{fr},\beta}=-\frac{1}{2\beta}\frac{\partial^2}{\partial\varphi_0^2}+\sum_{n=1}^{\infty}n(\mn{a}^{\dagger}_{2\beta}(n)\mn{a}_{2\beta}(n)+\ol{\mn{a}}^{\dagger}_{2\beta}(n)\ol{\mn{a}}_{2\beta}(n)+1).
  \label{}
\end{equation}
Note here that we have an infinite ground energy~$E_0=\sum_{n=1}^{\infty}n$, where it is customary to adopt~$\sum_{n=1}^{\infty}n=:\zeta(-1)=-\frac{1}{12}$ (\cite{DiFrancesco-Mathieu-Senechal-1997} pp.\ 172). 
Given a smooth real test function~$f\in C^{\infty}(\mb{S}^1)$ written in the form~$f=\sum_{n\in\mb{Z}}f_n\me^{\ii n\theta}$ again with~$f_{-n}=\ol{f}_n$ we have from (\ref{eqn-free-field-crea-anni-complex}) the following representation of the multiplication operator~$\mf{\varphi}(f):F\mapsto \ank{\varphi,f}_{L^2}F$ on~$\mathcal{H}_{\mb{S}^1}$:
\begin{equation}
  \mf{\varphi}(f)=2\pi\varphi_0 f_0+\frac{2\pi}{\sqrt{2\beta}}\sum_{n>0}\frac{1}{\sqrt{n}}\big[
    f_{-n}(\mn{a}_{2\beta}(n)+\ol{\mn{a}}^{\dagger}_{2\beta}(n))
  +f_n(\mn{a}^{\dagger}_{2\beta}(n)+\ol{\mn{a}}_{2\beta}(n)) \big].
  \label{eqn-field-op-test-func}
\end{equation}
Actually, since~$\varphi_{-n}=\ol{\varphi}_n$, namely~$x_{-n}=x_n$,~$y_{-n}=-y_n$, we have~$\ol{\mn{a}}_{2\beta}(n)=\mn{a}_{2\beta,-n}$. Thus (\ref{eqn-field-op-test-func}) shows the operator valued distribution~$\mf{\varphi}$ is also given formally by
\begin{equation}
  \mf{\varphi}(\theta)=\varphi_0 +\frac{1}{\sqrt{2\beta}}\sum_{n\ne 0}\frac{1}{\sqrt{|n|}} \big( \me^{\ii n\theta}\mn{a}_{2\beta}(n)+\me^{-\ii n\theta}\mn{a}^{\dagger}_{2\beta}(n)\big).
  \label{}
\end{equation}
Here~$\{\mf{\varphi}(\theta)\}_{\theta}$ could be thought of as an~$\mb{S}^1$-indexed family of ``position operators''~$\mn{X}_{\theta}:F\mapsto \varphi(\theta)\cdot F$ on~$\mathcal{H}_{\mb{S}^1}$.

\subsection{Dynamics of the Massive Gaussian Free Field}\label{sec-intro-dyna-massive-gff}
    In this subsection we are interested in describing the ``slice dynamics'' of the massive Gaussian Free Field (GFF) on the infinite cylinder~$\mb{R}\times\mb{S}^1$. In the previous subsection we have seen that the quantum (massless) free field on the cylinder consists of infinitely many harmonic oscillators and a free particle; here we will see that the Euclidean (probabilistic) free field accordingly consists of infinitely many Ornstein-Uhlenbeck processes, while the counterpart of the free particle disappears due to the mass.

    Let~$\mathcal{S}(\mb{R}\times \mb{S}^1)$ denote the space of smooth functions on~$\mb{R}\times \mb{S}^1$ which are rapidly decaying (Schwartz) in the~$\mb{R}$-direction, and accordingly~$\mathcal{S}'(\mb{R}\times \mb{S}^1)$ the dual space of distributions tempered in the~$\mb{R}$-direction. Let~$\Delta:=\Delta_{\mb{R}\times \mb{S}^1}$ denote the positive Laplacian. We have by Fourier transform\footnote{Now we use~$\wh{\chi}(k):=(2\pi)^{-\frac{1}{2}}\int_{}^{}\me^{-\ii k\theta}\chi(\theta)\,\dd\theta$ on~$\mb{S}^1$ and~$\wh{f}(E)=\int_{}^{}\me^{-\ii E t}f(t)\,\dd t$ on~$\mb{R}$.}
    \begin{equation}
      \ank{f,(\Delta+\one)^{-1}f}_{L^2}=\frac{1}{2\pi}\sum_{k=-\infty}^{\infty}\int_{-\infty}^{\infty}\frac{|\wh{f}(k,E)|^2}{k^2+E^2+1}\,\dd E
      \label{}
    \end{equation}
    for~$f\in\mathcal{S}(\mb{R}\times\mb{S}^1)$ and one sees that~$(\Delta+\one)^{-1}$ extends to a bounded positive self-adjoint operator on~$L^2(\mb{R}\times \mb{S}^1)$. By Minlos' theorem (proposition \ref{prop-minlos-theo}) there exists a unique Gaussian measure~$\mu_{\mm{GFF}}^{\mb{R}\times \mb{S}^1}$ (abb.\ $\mu$) on~$\mathcal{S}'(\mb{R}\times \mb{S}^1)$ with
    \begin{equation}
      \mb{E}_{\mu}\big[\phi(f)\phi(h)\big]=\bank{f,(\Delta+\one)^{-1}h}_{L^2},\quad\quad \mb{E}_{\mu}\big[\phi(f)\big]\equiv 0,
      \label{eqn-cov-massive-gff}
    \end{equation}
    for all~$f$,~$h\in \mathcal{S}(\mb{R}\times \mb{S}^1)$,~$\phi$ being the random distribution following~$\mu$, called the \textsf{massive GFF} with mass~$=1$. We will proceed by the following (heuristic) analogy with the Ornstein-Uhlenbeck process:
\begin{center}\def\arraystretch{1.2}
  \begin{tabular}{|c|c|c|}
    \hline
    &\textsf{Ornstein-Uhlenbeck on $\mb{R}$} & \textsf{massive GFF on $\mb{R}\times \mb{S}^1$} \\ \hline
    sample path & $\omega\in \mathcal{S}'(\mb{R})$ & $\phi\in \mathcal{S}'(\mb{R}\times \mb{S}^1)$ \\ \hline
    configuration at time $t$ &$X_t=\omega(t)\in\mb{R}$ &  $\phi|_{\{t\}\times\mb{S}^1}\in \mathcal{S}'(\mb{S}^1)$  \\ \hline
    covariance & $\mb{E}[X_s X_t]=\sank{\delta_s,(-\frac{\dd^2}{\dd x^2}+1)^{-1} \delta_t}$ & (\ref{eqn-cov-massive-gff}) \\ \hline
    invariant measure & $\mathcal{N}(0,\frac{1}{2})=\gamma_{1,1}$ & ? \\ \hline
    Hilbert space & $L^2(\mb{R},\gamma_{1,1})$ & $L^2(\mathcal{S}'(\mb{S}^1),~?~)$ \\ \hline
    Feller semi-group & (\ref{eqn-orns-uhl-feynman-kac-1d-basic}) or (\ref{eqn-funda-relation-har-osc-orn-uhl}) & ? \\ \hline
  \end{tabular}
\end{center}

 Our goal now is to find the question marks. Recall that due to the Gaussianity of the Ornstein-Uhlenbeck process, we were able to read-off the fact that the law of each~$X_t$ is invariant for all~$t$ from the covariance structure, as all~$X_t$'s have the same variance (and are centered). In the GFF case, we are led naturally to ask: pretend that the restriction map~$\tau_t:\mathcal{S}'(\mb{R}\times\mb{S}^1)\lto \mathcal{S}'(\mb{S}^1)$,~$\phi\mapsto \phi|_{\{t\}\times\mb{S}^1}$ is well-defined, what would be the image measure~$(\tau_t)_*\mu$ living on~$\mathcal{S}'(\mb{S}^1)$? Is it still Gaussian? If so, what is the covariance? Is it also invariant for all $t$?

    Many of these questions have simple answers intuitively. Since~$\tau_t$ is \textbf{linear}, it follows from the definition \ref{def-gaussian-measure} that~$(\tau_t)_*\mu$, which is the linear image of a Gaussian measure, should also be Gaussian, provided that we are allowed to apply the ``change of variables'' formula\footnote{That is,~$\int_{}^{}F(\varphi)\,\dd(\tau_t)_*\mu=\int_{}^{}F(\tau_t\phi)\,\dd\mu$.}. Moreover, since the covariance operator $(\Delta+\one)^{-1}$ commutes with time translations, $\mu_{\mm{GFF}}$ is time translation invariant, and therefore $(\tau_t)_*\mu$ should indeed be the same for all $t$, once we are able to define it. Along these lines, the main technical point thus boils down to making sense of the restriction map $\tau_t$ in our low regularity scenario where it is ill-defined ordinarily. Note that in the probabilistic setting we need not make $\tau_t$ \textit{continuous} of any sort; all we need is \textit{measurability}. This will turn out to allow much more freedom; we leave the details to more specialized discussions (e.g.\ \cite{HairerSPDE} section 4.3).

    \begin{figure}[h]
     \centering
     \begin{subfigure}[t]{0.4\textwidth}
         \centering
         \includegraphics[width=\textwidth]{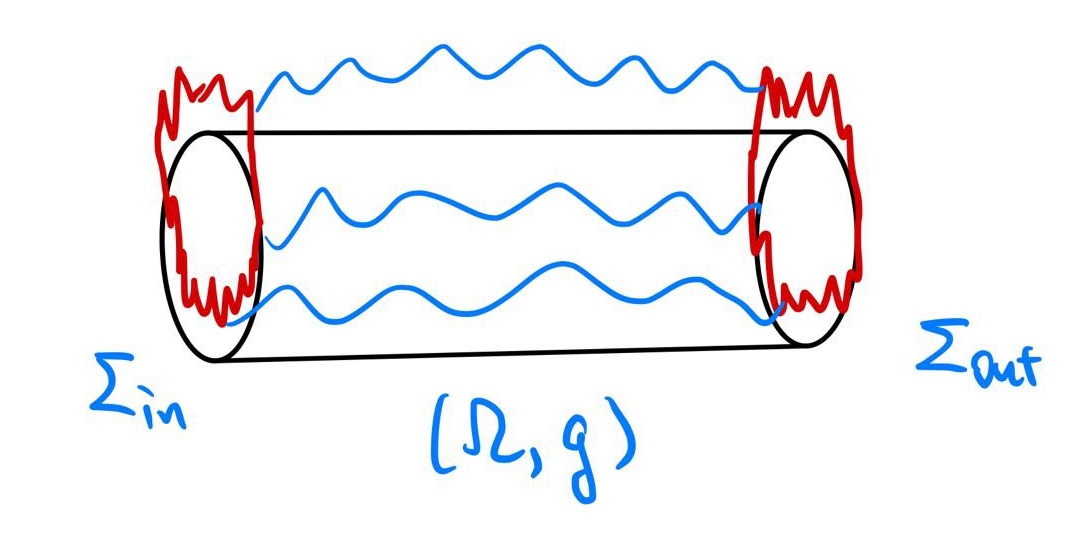}
         \caption{GFF on a finite piece of Cylinder}
         \label{fig-cylinder-propagat}
     \end{subfigure}
     \hspace{1cm}
     \begin{subfigure}[t]{0.4\textwidth}
         \centering
         \includegraphics[width=\textwidth]{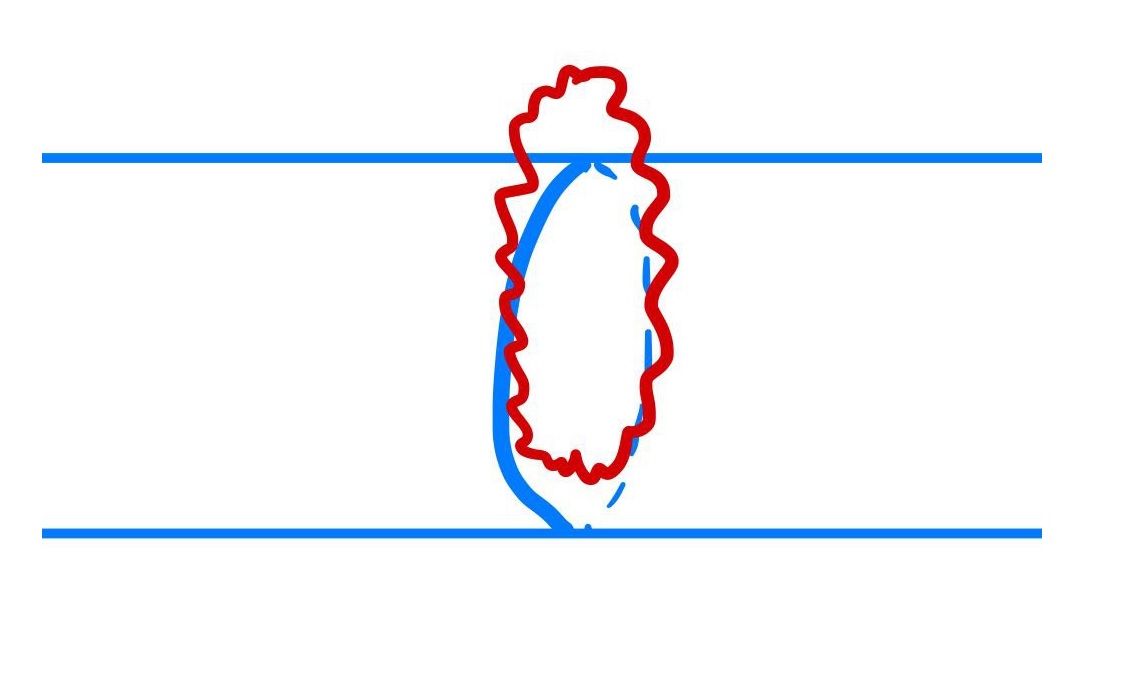}
         \caption{Sharp-time Test Function}
         \label{fig-sharp-time-field}
     \end{subfigure}
        \caption{GFF on the Cylinder}
        \label{fig-gff-cylinder-massive}
\end{figure}
    
     To actually ``know'' the measure $(\tau_t)_*\mu$, we still need to find the covariance. This, in other words, is to determine the bilinear form\footnote{Again having applied the ``change of variables'' formula above.}~$\mb{E}_{\mu}[(\tau_t \phi)(\chi)(\tau_t \phi)(\eta)]$ for two (real) test functions~$\chi$,~$\eta\in C^{\infty}(\mb{S}^1)$. In the first place, we need to have a good definition of the random variable~$(\tau_t \phi)(\chi)$, and again, the problem is the low regularity of $\phi$. However, we see from the covariance structure (\ref{eqn-cov-massive-gff}) that as long as~$f\in W^{-1}(\mb{R}\times \mb{S}^1)$, the Sobolev space of order~$-1$, we would be able to define the random variable~$\phi(f)$ uniquely as a limit in~$L^2(\mu)$. Now observe that (by definition)
    \begin{equation}
      \ank{\tau_t\phi,\chi}_{L^2(\mb{S}^1)}=\ank{\phi,J_t \chi}_{L^2(\mb{R}\times \mb{S}^1)},\quad\quad\textrm{with }J_t(\chi)\defeq \delta_t\otimes \chi,
      \label{}
    \end{equation}
    and by Fourier transform,~$\delta_t\otimes \chi$ is actually in~$W^{-1}$ (see figure \ref{fig-sharp-time-field})!  We thus find by (\ref{eqn-massive-green-function-formula})
    \begin{align}
      \mb{E}_{\mu}\big[(\tau_t \phi)(\chi)(\tau_t \phi)(\eta)\big]&=\bank{J_t \chi,(\Delta+\one)^{-1}J_t \eta}_{L^2(\mb{R}\times \mb{S}^1)}
      =\frac{1}{2\pi}\sum_{k}^{}\int_{\mb{R}}^{}\frac{\ol{\wh{\chi}}(k)\wh{\eta}(k)}{k^2+E^2+1}\,\dd E\\
      &=\sum_k \frac{\ol{\wh{\chi}}(k)\wh{\eta}(k)}{2\sqrt{k^2+1}}=\bank{\chi,(2(\Delta_{\mb{S}^1}+1)^{\frac{1}{2}})^{-1}\eta}_{L^2(\mb{S}^1)}.
      \label{eqn-inner-prod-gauss-hilb-space-slice-meas}
    \end{align}
    This shows the covariance operator of~$(\tau_t)_*\mu$ is~$(2\mn{D}_{\mb{S}^1+1})^{-1}$ setting $\mn{D}_{\mb{S}^1+1}:=(\Delta_{\mb{S}^1}+1)^{\frac{1}{2}}$ (cf.\ (\ref{eqn-slice-field-covariance})). Moreover, we are reassured that~$(\tau_t)_*\mu$ is indeed independent of~$t$. Note that the above calculations does not prevent us from taking two \textit{different} time instants~$s$ and~$t$ for~$\chi$ and~$\eta$ respectively. Doing that, we find
    \begin{equation}
      \mb{E}_{\mu}\big[(\tau_s \phi)(\chi)(\tau_t \phi)(\eta)\big]=\bank{\chi,\frac{\me^{-|s-t|\cdot\mn{D}_{\mb{S}^1+1}}}{2\mn{D}_{\mb{S}^1+1}}\eta}_{L^2(\mb{S}^1)}.
      \label{eqn-slice-gff-covariance-diff-time}
    \end{equation}
    This is the ``1-particle'' analogue of (\ref{eqn-orns-uhl-feynman-kac-1d-basic}). If we switch to~$L^2((\tau_0)_*\mu)$ on the r.h.s., that is, we regard~$\chi$ and~$\eta$ as random variables defined by pairing with~$\varphi\in \mathcal{S}'(\mb{S}^1)$, then we get
    \begin{equation}
      \mb{E}_{\mu}\big[(\tau_s \phi)(\chi)(\tau_t \phi)(\eta)\big]=\bank{\chi,\me^{-|s-t|\cdot\mn{D}_{\mb{S}^1+1}}\,\eta}_{L^2((\tau_0)_*\mu)}.
      \label{eqn-gff-semi-group-express-1-part}
    \end{equation}
    The inner product (\ref{eqn-inner-prod-gauss-hilb-space-slice-meas}) suggests that the \textit{Gaussian Hilbert space} of~$(\tau_0)_*\mu$ is~$W^{-\frac{1}{2}}(\mb{S}^1)$ which corresponds to the ``1-particle'' Hilbert space. By proposition \ref{prop-wiener-chaos} its Fock space is exactly~$L^2((\tau_0)_*\mu)$ and we can further promote (\ref{eqn-gff-semi-group-express-1-part}) to general functionals~$F$,~$G\in L^2((\tau_0)_*\mu)$ by applying the ``Fock functor''~$\Gamma$:
    \begin{equation}
      \mb{E}_{\mu}\big[\ol{F(\tau_s \phi)}G(\tau_t \phi)\big]=\bank{F,\Gamma(\me^{-|s-t|\cdot\mn{D}_{\mb{S}^1+1}})\,G}_{L^2((\tau_0)_*\mu)}.
      \label{eqn-intro-gff-feller-semigroup-fock}
    \end{equation}
    Since $\me^{-|s-t|\cdot\mn{D}_{\mb{S}^1+1}}$ acting on $W^{-\frac{1}{2}}(\mathbb{S}^1)$ is \textit{contracting},
    further Fock space theory tells that 
\begin{equation}
  \Gamma(\me^{-|s-t|\cdot\mn{D}_{\mb{S}^1+1}}) = \me^{-|s-t|\cdot \dd \Gamma(\mn{D}_{\mb{S}^1+1})},
  \label{eqn-intro-gff-feller-semigroup-fock-generator}
\end{equation}
with~$\dd\Gamma(\mn{D}_{\mb{S}^1+1})|_{\sym^{\otimes n}}:=\mn{D}_{\mb{S}^1+1}\otimes \one\otimes\cdots\otimes \one +\cdots +\one\otimes\cdots\otimes \one\otimes \mn{D}_{\mb{S}^1+1}$ called the \textsf{second quantization} of~$\mn{D}_{\mb{S}^1+1}$ (see \cite{Sim2} pp.\ 32). It \textit{seems} that we have found the Feller semi-group of the massive GFF, generated by~$\mn{H}_0:=\dd\Gamma(\mn{D}_{\mb{S}^1+1})$, and completed the program of determining the question marks!

Actually, we missed the \textit{most important} point. Namely, we did not show that the GFF has a \textit{Markov} property analogous to (\ref{eqn-cond-gen-markov-proces}) and hence we do not yet have the Chapman-Kolmogorov lemma!  One way to peek into this issue is to put~$\chi=\eta=\me^{\ii k\theta}/\sqrt{2\pi}$ and see from (\ref{eqn-slice-gff-covariance-diff-time}) that the Fourier coefficients~$\wh{\phi}_k(t)$ of~$\tau_t\phi$ follow pairwise independent Ornstein-Uhlenbeck processes with~$\mb{E}|\wh{\phi}_k(t)|^2=(2\sqrt{k^2 +1})^{-1}$. As the Ornstein-Uhlenbeck process is Markov, it is reasonable that the ``distribution valued'' process~$t\mapsto \tau_t\phi$ is also Markov in time. Nelson \cite{Nel2}, nevertheless, found the following beautiful generalization of the Markov property that is \textbf{covariant} and does not depend on the choice of a time coordinate, or a ``foliation'' of our space-time by timelike Cauchy hypersurfaces.

One important ingredient is to define ``localized''~$\sigma$-algebras~$\mathcal{O}_A$ associated to \textit{closed subsets}~$A\subset\mb{R}\times \mb{S}^1$ which generalizes the~$\sigma$-algebra~$\mathcal{O}_{\le t}=\sigma(X_s|s\le t)$ in 1D. Restricting to the GFF case, we exploit the fact that the random variable~$\phi(f)$ is well-defined for~$f\in W^{-1}(\mb{R}\times \mb{S}^1)$ and put~$\mathcal{O}_A:=\sigma(\phi(f)|\supp f\subset A,f\in W^{-1})$. 
\begin{prop}
  [Nelson's Markov property] Let~$A$,~$B\subset \mb{R}\times \mb{S}^1$ be closed subsets such that~$A^{\circ}\cap B=\varnothing$. Then
  \begin{equation}
    \mb{E}_{\mu}[ F|\mathcal{O}_A]=\mb{E}_{\mu}[F|\mathcal{O}_{\partial A}]
    \label{eqn-main-markov-property-gff}
  \end{equation}
  for all functionals~$F\in L^2(\mu)$ which are~$\mathcal{O}_B$-measurable. Here~$\partial A$ denotes the boundary points of~$A$.
\end{prop}

\begin{proof}
  [Idea of proof.] First recall that for~$L^2$ random variables conditional expectation corresponds to orthogonal projection. For \textit{Gaussian} variables, we can use the Fock functor~$\Gamma$ and see that (\ref{eqn-main-markov-property-gff}) boils down to a fact about projections in the (space-time) ``1-particle'' Hilbert space, that is,~$W^{-1}$. This is to say, if $f\in W^{-1}$ and~$\supp f\subset B$ then~$\supp(P_{A}f)\subset \partial A$ and hence~$P_A f=P_{\partial A}f$. Here~$P_A$ denotes the orthogonal projection in~$W^{-1}$ onto the closed subspace supported in~$A$. This latter property of~$W^{-1}$ in turn boils down to the fact that~$\Delta+\one$, inverse of the covariance which gives the \textit{action}~$\int_{}^{}|\nabla\phi|^2+\phi^2$ of the GFF, is a \textit{local} operator.
\end{proof}

With this Markov property in hand, it would be possible to proceed exactly as in the 1D case (using the Lie-Trotter formula) to derive a generalized ``Feynman-Kac'' type formula in our setting, modulo technicalities in defining rigorously the interaction potential, a version of which will be treated in section \ref{sec-nelson-main}. Heuristically, for a (real) test function~$\chi\in C^{\infty}(\mb{S}^1)$ we define, with $\varphi$ following $(\tau_0)_*\mu$, the random variable
\begin{equation}
  V_{\mm{int}}(\chi)(\varphi)\defeq \int_{\mb{S}^1}^{}\chi(\theta){:}\varphi^{2n}(\theta){:}\,\dd \theta.
  \label{}
\end{equation}
Acting on~$\mathcal{H}^{{:}m{:}}$ by multiplication, we have heuristically (\cite{GRS1} pp.\ 142), by the slogan below (\ref{eqn-wick-binomial-crea-left-ann-right}),
\begin{equation}
  V_{\mm{int}}(\chi)=\sum_{k_1,\cdots,k_{2n}}\ol{\wh{\chi}}(k_1+\cdots+k_{2n})\sum_{j=0}^n \binom{2n}{j}
  \frac{\mn{a}^{\dagger}(-k_1)\cdots\mn{a}^{\dagger}(-k_j)\mn{a}(k_{j+1})\cdots \mn{a}(k_{2n})}{2^{\frac{n}{2}}(k_1^2+1)^{\frac{1}{4}}\cdots (k_{2n}^2+1)^{\frac{1}{4}}},
  \label{}
\end{equation}
where~$\mn{a}^{\dagger}(k)$ and~$\mn{a}(k)$, like (\ref{eqn-free-field-crea-anni-complex}), are the ground state conjugated complex creation and annihilation operators corresponding to the~$k$-th Fourier mode following the measure~$\mathcal{N}_{\mb{C}}(0,\frac{1}{2}(k^2+1)^{-\frac{1}{2}})$.
We conclude the present subsection with the 2D analogue of (\ref{eqn-1d-feynman-kac-heuristic}) which summarizes in a sense the probabilistic interpretation of the Euclidean field theory.

\begin{prop}
    [Feynman-Kac-Nelson, \cite{GRS1} Theorem II.14] Let~$\mu=\mu_{\mm{GFF}}$, and~$\chi\in C^{\infty}(\mb{S}^1)$ as above. Then
    \begin{equation}
      \mb{E}_{\mu}\Big[ 
      \ol{F(\tau_0\phi)}\exp\Big( -\int_{0}^{t}\int_{\mb{S}^1}^{}\chi(\theta){:}\phi^{2n}(s,\theta){:}\,\dd \theta\,\dd s \Big) G(\tau_t\phi) \Big]=\bank{F,\me^{-t(\mn{H}_0+V_{\mm{int}}(\chi))}G}_{L^2((\tau_0)_*\mu)},
      \label{}
    \end{equation}
    for any~$F$,~$G\in L^2((\tau_0)_*\mu)$ and~$t>0$.\hfill~$\Box$
  \end{prop}

\section{The Functorial Framework}
In the previous sections we have been describing quantum evolution over space-times where there is always a distinguished time coordinate (space-time~$=$ space~$\times$ time). For \textit{general} curved space-times this is no longer the case. To deal with the more general situation Atiyah and Segal \cite{Segal} proposed the following \textit{functorial} framework, motivated partly by gauge theories as well as string theory. In this subsection we describe a very rough version of this framework (as multiple precise versions exist oriented towards different purposes), and finally relate it to the more traditional point of view (Feynman-Kac) established in the previous sections.

Let~$\mathcal{C}_{d+1}^{\fk{s}}$ denote the category whose objects are closed~$d$-dimensional manifolds with a specified ``structure'' denoted~$\fk{s}$ (e.g.\ topological, smooth, Riemannian, conformal, \dots), and for two such objects~$\Sigma_1$,~$\Sigma_2$ a morphism from~$\Sigma_1$ to~$\Sigma_2$ given by an \textit{oriented cobordism} from~$\Sigma_1$ to~$\Sigma_2$ (that is, a compact~$(d+1)$-manifold~$\Omega$ with boundary~$\partial\Omega=\Sigma_1\sqcup \Sigma_2$, with~$\Sigma_1$ labelled as ``in'', and~$\Sigma_2$ ``out'') also with the structure~$\fk{s}$. We shall write $\Omega:\Sigma_1\rsa \Sigma_2$. The composition of morphisms is concatenation (or sewing/gluing) of cobordisms\footnote{If $\fk{s}$ is ``topological'', then the identities are finite cylinders; however in other cases the identity might be something degenerate, e.g.\ an ``empty'' cobordism.}. We denote very loosely by~$\ms{Vect}_{\mb{C}}$ the category of (topological) complex vector spaces and linear maps (different subcategories would be considered for different~$\fk{s}$).

\begin{deef}
  [Atiyah \cite{Atiyah}, Segal \cite{KS, Segal}, Kontsevich \cite{KS}] A~$(d+1)$-dimensional QFT (with structure~$\fk{s}$) is a monoidal\footnote{Loosely speaking, a \textsf{monoidal category} is a category where tensor products of objects make sense. In~$\mathcal{C}^{\fk{s}}_{d+1}$ the ``tensor product'' is given by disjoint union. Moreover, there is a ``unit'' with respect to the tensor product ($\varnothing$ and~$\mb{C}$ respectively in~$\mathcal{C}^{\fk{s}}_{d+1}$ and~$\ms{Vect}_{\mb{C}}$). Therefore, rule (i) is saying simply that our functor respects the monoidal structure, i.e.\ a \textsf{monoidal functor}. See \cite{MacL} chapter XI.} functor~$\mathcal{C}^{\fk{s}}_{d+1}\lto \ms{Vect}_{\mb{C}}$. This means effectively that it is a rule which associates
  \begin{enumerate}[(a)]
    \item to each closed~$d$-dimensional manifold~$\Sigma$ a complex vector space~$\mathcal{H}_{\Sigma}$, the ``space of states'';
    \item to each~$\Omega:\Sigma_{\mm{in}}\rsa \Sigma_{\mm{out}}$ a linear map~$U_{\Omega}:\mathcal{H}_{\Sigma_{\mm{in}}}\lto \mathcal{H}_{\Sigma_{\mm{out}}}$, the ``evolution'';
  \end{enumerate}
  satisfying the following \textit{axioms}:
  \begin{enumerate}[(i)]
    \item (finite) disjoint unions are taken to (finite) tensor products (e.g.\ $\Sigma_1\sqcup \Sigma_2$ to~$\mathcal{H}_{\Sigma_1}\otimes \mathcal{H}_{\Sigma_2}$), and the empty $d$-manifold~$\varnothing$ is taken to~$\mathcal{H}_{\varnothing}:=\mb{C}$;
    \item let~$\Omega_{12}:\Sigma_1\rsa \Sigma_2$ and~$\Omega_{23}:\Sigma_2\rsa \Sigma_3$ be two cobordisms, and consider~$\Omega_{23}\circ \Omega_{12}:=\Omega_{12}\cup_{\Sigma_2}\Omega_{23}$ the cobordism obtained by gluing~$\Omega_{12}$ and~$\Omega_{23}$ along~$\Sigma_2$. Then we have
      \begin{equation}
	U_{\Omega_{23}\circ \Omega_{12}}=U_{\Omega_{23}}\circ U_{\Omega_{12}}:\mathcal{H}_{\Sigma_1}\lto \mathcal{H}_{\Sigma_3}
	\label{eqn-intro-segal-compos-axiom}
      \end{equation}
      as linear maps.
  \end{enumerate}
\end{deef}

\begin{figure}[h]
    \centering
    \includegraphics[width=0.8\linewidth]{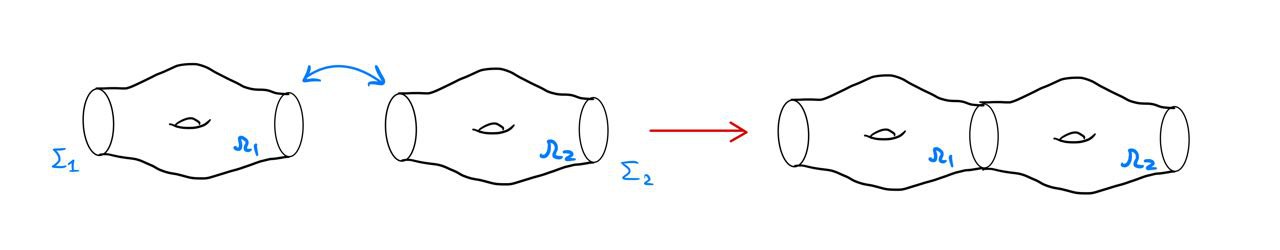}
    \caption{Segal Gluing}
    \label{fig-segal-glu-1}
\end{figure}

Unfortunately, much more needs to be said to make the above definition precise (and complete, in many situations). We list a few (far from exhaustive) of these technical points in remark \ref{rem-intro-segal-general} below, and refer to \cite{KS} section 3, \cite{Mnev} section 1.2, \cite{Segal} and the introduction of \cite{KMW} for further discussions.

In chapter \ref{chap-segal} we work with \textit{a} 2D QFT over the Riemannian category. For us objects of~$\mathcal{C}^{\mm{riem}}_{1+1}$ are finite disjoint unions of Riemannian circles characterized by perimeters (moreover they need to be ``enhanced'' by two-sided collars), morphisms are Riemannian surfaces with geodesic boundaries,~$\mathcal{H}_{\Sigma}$ are Hilbert spaces and~$U_{\Omega}$ are \textit{Hilbert-Schmidt} operators. See section \ref{sec-segal-descript} for a precise re-statement.

\begin{def7}\label{rem-intro-segal-general}
  \begin{enumerate}[(1)]
    \item Regarding rule (ii), when one glues~$\Omega_{12}$ and~$\Omega_{23}$ one must at the same time glue the geometric structure~$\fk{s}$.\footnote{A succinct way to say this is that one takes the \textit{fiber product}~$\Omega_{12}\times_{\Sigma_2}\Omega_{23}$ in the corresponding category of manifolds with structure~$\fk{s}$ (\cite{Mnev} section 1.2.2).} Also, one should appropriately identify cobordisms~$\Omega$,~$\Omega':\Sigma_{\mm{in}}\lto \Sigma_{\mm{out}}$ which are ``isomorphic'' via a boundary-preserving~$\fk{s}$-isomorphism; equivalently, one requires a ``naturality'' condition in terms of a commutative diagram (\cite{Mnev} section 1.2.2).
    \item As~$d$-manifolds are not \textit{a priori} embedded in~$(d+1)$-manifolds, a careful definition of the cobordism requires the data of the two $\fk{s}$-embeddings (a.k.a.\ parametrizations)~$\xi_{\mm{in},\mm{out}}:\Sigma_{\mm{in,out}}\hooklto \partial\Omega$, and gluing cobordisms as in rule (ii) essentially glues the parametrizations of~$\Sigma_2$. With this data, the ``in'' and ``out'' labels on the boundaries are identified by checking whether the induced orientations by the embeddings agree with the co-orientations coming from~$\Omega$. The parametrizations naturally go into the definition of~$U_{\Sigma}$ as we need them to relate things ``happening on~$\Omega$'' to the ``model'' vector spaces~$\mathcal{H}_{\Sigma_{\mm{in}}}$,~$\mathcal{H}_{\Sigma_{\mm{out}}}$. They are especially important in (2D) CFT regarding the next point below.
    \item The vector space~$\mathcal{H}_{\Sigma}$ is usually \textit{constructed out of}~$\Sigma$ (see path-integral representation below) and hence carries a natural representation of~$\mm{Diff}(\Sigma)$, assuming~$\fk{s}$ at least smooth. In 2D CFT the operators~$U_{\Omega}$ interplay with the representations of~$\mm{Diff}(\mb{S}^1)$ on~$\mathcal{H}_{\Sigma_{\mm{in}}}$,~$\mathcal{H}_{\Sigma_{\mm{out}}}$ through the parametrizations of~$\Sigma_{\mm{in}}$ and~$\Sigma_{\mm{out}}$ by copies of~$\mb{S}^1$.
    \item (\textsf{unitarity assumption}) For \textit{unitary} QFTs we will make the following assumptions \textit{in this thesis}: we let~$\mathcal{H}_{\Sigma}$ be real Hilbert spaces;\footnote{Equivalently, complex Hilbert spaces equipped with an antilinear involution making it the complexification of an underlying real Hilbert space.} 
    if~$\Omega:\Sigma_1\rsa \Sigma_2$ is a cobordism, it is seen equally as a cobordism~$\Omega^*:\Sigma_2^*\rsa \Sigma_1^*$, and also~$\Omega:\varnothing\rsa \Sigma_1^*\sqcup \Sigma_2$ (abusing notation), where~$\Sigma_j^*$ denotes~$\Sigma_j$ with an oppositely oriented parametrization. In this case, we require~$U_{\Omega^*}=U_{\Omega}^{\dagger}$, the \textit{real adjoint} (or transpose) of~$U_{\Omega}$. This is not yet entirely precise as $\Sigma$ needs to be enhanced by symmetric 2-sided collars, see \cite{KS} section 3 under ``unitarity''.
    \item A \textit{closed}~$(d+1)$-manifold is considered as a cobordism~$\varnothing \lto\varnothing$. By rule (i) we get in this case an operator~$U_M:\mb{C}\lto\mb{C}$ identified with a single complex number. We denote this number otherwise by~$\mathcal{Z}_M$, usually called the \textsf{partition function}. Similarly for a ``1-sided'' cobordism~$\Omega$ with~$\partial\Omega=\Sigma_{\mm{out}}$ we get a vector~$U_{\Omega}$ in~$\mathcal{H}_{\Sigma_{\mm{out}}}$ ($\partial\Omega=\Sigma_{\mm{in}}$, a dual vector in~$\mathcal{H}^*_{\Sigma_{\mm{in}}}$). Often a third axiom is mentioned, that if we glue a cobordism~$\Omega$ with itself by identifying~$\Sigma_{\mm{in}}$ with~$\Sigma_{\mm{out}}$, getting the closed manifold~$\check{\Omega}$, then~$\ttr(U_{\Omega})=\mathcal{Z}_{\check{\Omega}}$. Technically, this could be implied by rule (ii) using the aforementioned interpretations.\footnote{If one glues two oppositely 1-sided cobordisms with boundaries of both identified with~$\Sigma$ (one labelled ``in'', the other ``out''), then one effectively pairs a vector in~$\mathcal{H}_{\Sigma}$ with a dual vector in~$\mathcal{H}_{\Sigma}^*$, and gets the number~$\mathcal{Z}_M$,~$M$ being the resulting closed manifold. One could then recover the ``trace axiom'' by taking one of these cobordisms to be a very thin cylinder and taking the length of the cylinder to zero, through an appropriate limiting procedure. See ``foreword and postscript'' to section 4 in \cite{Segal}.}
    On the other hand, rule (ii) is also a special case of a ``partial'' version of this trace axiom, namely, composition of linear maps~$U\lto V$,~$V\lto W$ of vector spaces corresponds to contracting (tracing)~$V$ with~$V^*$ in~$W\otimes V^*\otimes V\otimes U^*$.
  \end{enumerate}
\end{def7}

As mentioned by the authors in \cite{KS} section 3, ``the guiding principle of this approach is to preserve as much as possible of the path-integral intuition''. Let us see now that the path-integral formulation introduced in subsection \ref{sec-intro-path-integral} applies without essential modification to this new setting.

A cobordism~$\Omega:\Sigma_{\mm{in}}\rsa \Sigma_{\mm{out}}$ is considered as a piece of space-time connecting the two ``time slices''~$\Sigma_{\mm{in}}$ and~$\Sigma_{\mm{out}}$. As explained in section \ref{sec-intro-class-field-theory}, each field theory comes with specified \textsf{field configuration spaces} over space and space-time, denoted~$\conf(\Sigma)$ and~$\conf(\Omega)$. Typically~$\conf(X)=\mm{Map}(X,\mss{V})$, the space of \textit{maps} from~$X$ to a \textsf{spin value space}/\textsf{target space}~$\mss{V}$.
The path-integral formulation gives a recipe for $\mathcal{H}_{\Sigma}$ and $U_{\Omega}$ as follows. We define
	\begin{equation}
	  \mathcal{H}_\Sigma\defeq L^2(\conf(\Sigma),\mathcal{L}),
	  \label{eqn-intro-heu-hilb-space}
	\end{equation}
	where~$\mathcal{L}$ denotes the non-existent \textit{Lebesgue measure} on the configuration space~$\conf(\Sigma)$. Now suppose our field theory (over Euclidean signature, for definiteness) comes with the action~$S_{\Omega}:\conf(\Omega)\lto\mb{R}$. Define
	\begin{equation}
	  \left.
	  \def\arraystretch{1.3}
	  \begin{array}{rcl}
	     U_{\Omega}: \mathcal{H}_{\Sigma_{\mm{in}}} &\lto &\mathcal{H}_{\Sigma_{\mm{out}}},\\
	     F&\longmapsto & (U_{\Omega} F)(\psi)\defeq \ddp\int_{\conf(\Sigma_{\mm{in}})}^{}\mathcal{A}_\Omega(\psi,\varphi)F(\varphi)\,\dd \mathcal{L}(\varphi),
	  \end{array}
	  \right.
	  \label{}
	\end{equation}
	with the integral kernel
	\begin{equation}
	  \mathcal{A}_\Omega(\psi,\varphi)\defeq \int_{
	     \left\{  \phi\in \conf(\Omega)\middle| \substack{
	     \phi|_{\mm{in}}=\varphi,\\ 
	     \phi|_{\mm{out}}=\psi
	   }
      \right\}} \me^{-S_{\Omega}(\phi)}\,\dd \mathcal{L}(\phi),
	  \label{}
	\end{equation}
	where now we integrate against the still non-existent Lebesgue measure on~$\conf(\Sigma)$ (with the indicated boundary conditions). At this level, (\ref{eqn-intro-segal-compos-axiom}) reads
    \begin{equation}
        \mathcal{A}_{\Omega_{13}}(\psi,\varphi)=\int_{\conf(\Sigma_2)} \mathcal{A}_{\Omega_{23}}(\psi,\eta)\mathcal{A}_{\Omega_{12}}(\eta,\varphi)\, \dd\mathcal{L}(\eta),
    \end{equation}
    and corresponds to a ``formal Fubini theorem''. It requires $S_{\Omega_{13}}(\phi)=S_{\Omega_{12}}(\phi|_{\Omega_{12}})+S_{\Omega_{23}}(\phi|_{\Omega_{23}})$, namely, the action is \textit{local}.

\begin{exxx}[connection with Feynman-Kac for cylinders]
  Now we try to connect the Atiyah-Segal picture to the traditional Feynman-Kac-Nelson picture introduced previously. We consider the massive GFF in 2D discussed in subsection \ref{sec-intro-dyna-massive-gff}. This is a unitary QFT so we work under (4) of remark \ref{rem-intro-segal-general}, and since the massive GFF depends on a metric we work in the Riemannian category. Here the configuration space for a circle~$\mb{S}^1$ is~$\mathcal{D}'(\mb{S}^1)$ and the formal Hilbert space (\ref{eqn-intro-heu-hilb-space}) can be replaced by the rigorous probabilistic Hilbert space~$L^2(\mathcal{D}'(\mb{S}^1),(\tau_0)_*\mu)$ obtained in subsection \ref{sec-intro-dyna-massive-gff} (cf.\ (\ref{eqn-gff-semi-group-express-1-part})). 
  On a naive level, for a flat cylinder~$\Omega= [0,t]\times \mb{S}^1$, it is natural to impose~$U_{\Omega}:=\me^{-t\mn{H}_0}$ with~$\mn{H}_0$ defined below (\ref{eqn-intro-gff-feller-semigroup-fock-generator}). On this level (\ref{eqn-intro-segal-compos-axiom}) holds trivially between standard cylinders. Going one step further, the fact that~$\me^{-t\mn{H}_0}$ is the ``Feller semi-group'' of the GFF means that it comes from the underlying path-integral representation for the GFF in the spirit of (\ref{eqn-1d-feynman-kac-heuristic}). 
  On this level (\ref{eqn-intro-segal-compos-axiom}), also for cylinders, is the Chapman-Kolmogorov lemma for GFF which is a Markov process in the sense of (\ref{eqn-main-markov-property-gff}). However, more questions remain to be asked:
  \begin{enumerate}[(a)]
    \item How is this definition of~$U_{\Omega}$ for cylinders compatible with gluing other general curved surfaces?
    \item Postulating a definition for~$U_{\Omega}$ that satisfies (\ref{eqn-intro-segal-compos-axiom}) gives one \textit{possibility} for an Atiyah-Segal model for the (massive) GFF. In general this does not imply \textit{necessity}. How \textit{necessary} is this definition of~$U_{\Omega}$, even for cylinders? 
  \end{enumerate}
  As an immature attempt at answering these questions (especially (b)), we single out the following assumptions which would be made for the GFF. In this more general setting the data of the GFF includes a probability measure~$\mu_M$ defined on~$\mathcal{D}'(M)$ for every closed manifold~$M$ following (\ref{eqn-intro-def-gff-eigenmodes}), which has a Markov property in the sense of (\ref{eqn-main-markov-property-gff}) due to the locality of its action. Moreover denote by~$\mu$ the measure for the infinite flat cylinder defined in (\ref{eqn-cov-massive-gff}). We postulate that
\begin{enumerate}[(i)]
  \item for a 1-sided cobordism~$\Omega:\varnothing \rsa \Sigma$, the vector or state~$U_{\Omega}$ is represented by a nonnegative~$L^2$ function\footnote{This would mean that~$U_{\Omega}$ could be interpreted as a ``perturbed ground state'' in view of the Perron-Frobenius theorem (cf.\ proposition \ref{prop-perron-frob} and \cite{GJ} theorem 3.3.2).}, note that under our unitarity assumption (remark \ref{rem-intro-segal-general}(4)),~$U_{\Omega}$ is identified with~$U_{\Omega^*}$, the dual vector representing~$\Omega^*:\Sigma\rsa \varnothing$;
  \item for 2 oppositely 1-sided cobordisms~$\Omega_1:\varnothing\rsa \Sigma$ and~$\Omega_2:\Sigma\rsa \varnothing$ such that~$\Omega_2\circ\Omega_1=M$ with~$\Sigma\subset M$ embedded as a hypersurface, then~$U_{\Omega_1} U_{\Omega_2}\cdot(\tau_0)_*\mu=(\tau_{\Sigma})_*\mu_M$. Here~$U_{\Omega_1} U_{\Omega_2}$ is the product of the 2 functions in~$L^2$ space, and~$\tau_{\Sigma}:\mathcal{D}'(M)\lto \mathcal{D}'(\Sigma)$ denotes the restriction.
\end{enumerate}
\begin{figure}[h]
    \centering
    \includegraphics[width=0.6\linewidth]{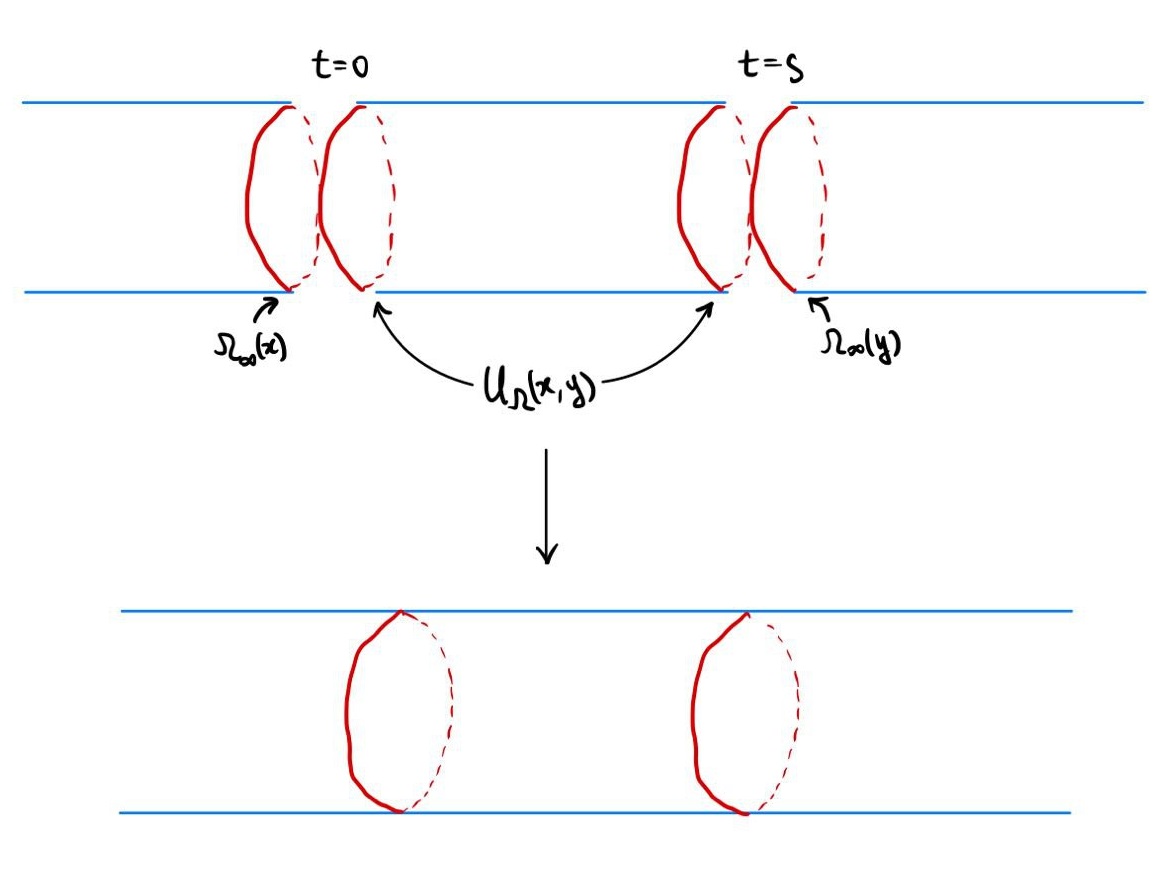}
    \caption{The Cylinder Case}
    \label{fig-segal-cyl}
\end{figure}
From (i) and (ii) one can deduce that for an infinite half-cylinder~$(-\infty,0]\times \mb{S}^1$, the associated state~$U_{(-\infty,0]\times \mb{S}^1}$ is the function~$1$, the ``ground state'', somehow tautologically since we used~$(\tau_0)_*\mu$ to define the Hilbert space. Accordingly the operator~$U_{[0,t]\times\mb{S}^1}$ associated to the finite cylinder~$[0,t]\times\mb{S}^1$ is the one represented by the Radon-Nikodym density
\begin{equation}
  U_{[0,t]\times\mb{S}^1}=\frac{\dd (\tau_{0\sqcup t})_* \mu}{\dd [(\tau_0)_*\mu\otimes (\tau_0)_*\mu]} \in \mathcal{H}_{\mb{S}^1}\otimes \mathcal{H}_{\mb{S}^1},
  \label{eqn-intro-segal-amp-cylinder}
\end{equation}
where~$\tau_{0\sqcup t}$ is the 2-component restriction map~$\mathcal{S}'(\mb{R}\times\mb{S}^1)\lto \mathcal{D}'(\{0\}\times\mb{S}^1)\times  \mathcal{D}'(\{t\}\times\mb{S}^1)$. By (\ref{eqn-intro-gff-feller-semigroup-fock}), we justify that~$U_{[0,t]\times\mb{S}^1}$ given by (\ref{eqn-intro-segal-amp-cylinder}) indeed coincides with~$\me^{-t\mn{H}_0}$. Moreover, one can proceed to show that (\ref{eqn-intro-segal-amp-cylinder}) is compatible with gluing other surfaces under assumption (ii), hence answering question (i). For example, let~$\Omega_1:\varnothing \rsa \mb{S}^1$ and~$\Omega_2:\mb{S}^1\rsa \varnothing$ again be two 1-sided cobordism. Put~$\tilde{M}:=\Omega_2\circ [0,t]\times \mb{S}^1\circ\Omega_1$. Then, following essentially the arguments in section \ref{sec-segal-main} using the Markov property, one may show that
\begin{equation}
  U_{[0,t]\times\mb{S}^1}(U_{\Omega_1}\otimes U_{\Omega_2})\cdot (\tau_0)_*\mu\otimes (\tau_0)_*\mu =(\tau_{0\sqcup t})_*\mu_{\tilde{M}},
  \label{}
\end{equation}
where we abused the notation~$\tau_{0\sqcup t}$ now to denote the restriction map onto the two boundaries of~$[0,t]\times\mb{S}^1$ seen as embedded in~$\tilde{M}$. Last but not least, we mention that postulate (ii) could be formulated in a ``ground-state-free'' manner using the infinite dimensional half-densities introduced by Pickrell \cite{Pickrell}.
\end{exxx}

To end this section, we mention that \cite{KS} has recently proposed a framework for doing Wick rotation that fits into the functorial picture, based on the idea of working with complex metrics on curved surfaces. A remarkable feature of this setting is that for any fixed surface with boundary there's a space of admissible metrics where Lorentzian ones are on the boundary and Euclidean ones in the interior. Remembering from subsection \ref{sec-intro-wick-rot} that Minkowski evolution operators were the boundary value of (complex) Euclidean ones, one plausibly imagines that Wick rotation may also be made ``functorial'', accounting for the ``boundary-bulk'' correspondence both on the metric side and on the operator side. However, making this statement precise lies far beyond the scope of the present thesis. We also mention that complex metrics have also been considered in \cite{Wrochna-2019-wick-analytic} and \cite{gérard-2021-hartle-hawking-israel-state-stationary-black}. Lastly, see also \cite{witten-2022-notecomplexspacetimemetrics}.

\chapter{Segal Axioms and Periodic Covers}\label{chap-segal}

This Chapter is adapted from \cite{Lin}.

     \section{Introduction}

In the classical approach to Quantum Field Theory (QFT),
a central role is played by the representations of the Poincar\'e group. For example, the Wightman axioms postulate such a representation on a Hilbert space $\mathcal{H}$ and assign operators on $\mathcal{H}$ to test functions on Minkowski space-time respecting the representation (\cite{RSim2} X.7).
It was later realized by Feynman, Symanzik \cite{Symanzik}, Nelson \cite{NelsonEuclidean} in the 60's that one could also describe many models of QFT using the \textit{functional integral}, and considering it in \textit{imaginary time} leads to the \textbf{Euclidean approach} to QFT (EQFT). This approach allowed many interesting models to be constructed rigorously in Minkowski space-time, for which the Wightman axioms were explicitly checked. They constitute a significant part of the subject of Constructive Quantum Field Theory (CQFT). However, the traditional techniques relied heavily on the space-time being flat and the results were likewise limited.

In the 80's, motivated by works of Atiyah and Witten on the relation of QFT with geometry and topology, there was an attempt to give an axiomatic definition of QFT that would incorporate arbitrary curved space-times, while still keeping the main properties of the
functional integral representation. Atiyah, Segal, Kontsevich \cite{Segal, KS} then arrived at a \textbf{functorial framework}. Loosely speaking, this framework defines a QFT as a linear representation of a certain bordism category of ``space-times''. At a preliminary level, a manifest difference between the traditional axiomatics and the functorial approach is the way they view time evolution. Traditionally time evolution is represented by a symmetry transformation of space-time (i.e.\ a Poincar\'e group element), whereas in the functorial picture it is represented by the piece of space-time cobordism connecting the initial and final ``time-slices''.
One important postulate, for example, says that \textbf{composition} of time evolutions corresponds to \textbf{gluing} of space-time cobordisms.

This paper is an attempt to reconstruct in the new functorial framework a celebrated interacting QFT model (the $P(\phi)_2$-model) previously constructed using traditional CQFT techniques. We give the precise sense and a rigorous proof that the $P(\phi)_2$-model satisfies Segal's axioms (Section \ref{sec-segal-main}). Then we apply the functorial QFT framework to certain purely mathematical questions on surface geometry (Section \ref{sec-period-cover}).

\subsection{Example of 1D Spin Chain and its Transfer Operator}
\label{sec-intro-spin-chain}

\noindent Our example illustrates the gluing properties by the transfer operator for discrete spin systems.
Consider $\mb{Z}_N:=\mathbb{Z}/N\mathbb{Z}$ as a $1$-dimensional chain of size $N$, consider the space $\mathbb{R}^N$ now as the space of maps $\mb{Z}_N \lto \mathbb{R}$, namely a discrete \textsf{path space}. The lattice site is denoted by $i\in \{1,\dots,N\}$ and call $\sigma\in\mb{R}^N$ a \textsf{spin configuration}, whose value at the site $i$ reads $\sigma(i)\in \mathbb{R}$. We impose \textsf{periodic boundary condition}, which means~$\sigma(N+1)\equiv \sigma(1)$. Thus the chain could be considered as circular.
We are given an action functional on the configuration space $\mb{R}^N$ that reads
\begin{eqnarray}
S_N(\sigma)\defeq\sum_{i=1}^{N} \vert \sigma(i+1)-\sigma(i)\vert^2 + \sum_{i=1}^N P(\sigma(i))
\label{eqn-spin-chain-inter}
\end{eqnarray}
where $P$ is a polynomial bounded from below, the interaction is nearest neighbour.
Given a configuration~$\sigma\in \mb{R}^N$,~$S_N(\sigma)$ may be thought of as its ``energy'' and the statistical behaviour of the system is described by the probability measure called \textsf{Gibbs measure},
\begin{eqnarray}
 \dd\mu_{P(\sigma)}(\sigma)\defeq \frac{1}{\mathcal{Z}(N)}\me^{-S_N(\sigma)} \dd^N\sigma, \quad\textrm{with}\quad
\mathcal{Z}(N)=\int_{\mathbb{R}^N} \me^{-S_N(\sigma)} \dd^N\sigma 
\label{eqn-spin-chain-gibbs-part-func}
\end{eqnarray}
called the \textsf{partition function} of the system.

Now we would like to express this partition function $\mathcal{Z}(N)$ in terms of elementary building blocks.
The main idea is to slice the action functional as
\begin{eqnarray*}
S_N(\sigma)=\sum_{i=1}^{N} \Big[ \vert \sigma(i+1)-\sigma(i)\vert^2 + \frac{1}{2}(P(\sigma(i+1))+P(\sigma(i))) \Big]
\end{eqnarray*} 
so exponentiating gives
\begin{eqnarray}
\exp\left(-S_N(\sigma) \right)=  \prod_{i=1}^{N} K(\sigma(i+1),\sigma(i)) 
\label{eqn-spin-chain-inter-expression}
\end{eqnarray}
where 
\begin{equation}
    K(x,y)= \me^{-\vert x-y\vert^2-\frac{1}{2}(P(x)+P(y))} 
    \label{eqn-spin-chain-trans-kernel}
\end{equation}
which is the Schwartz kernel of an operator on $L^2(\mathbb{R})$ that is smoothing.

\begin{deef}
  Define the \textsf{transfer operator} $T$ to be exactly the operator with kernel~$K(x,y)$, that is,
  \begin{equation}
    (TF)(x)=\int K(x,y)F(y) \dd y,
    \label{}
  \end{equation}
  for any function(al)~$F\in L^2(\mb{R})$.
\end{deef}

Then we see immediately from (\ref{eqn-spin-chain-inter-expression}) that
\begin{eqnarray}
  \mathcal{Z}(N)=\int_{}^{}\prod_{i=1}^{N} \big[K(\sigma(i+1),\sigma(i)) \dd\sigma(i) \big]= \ttr_{L^2(\mathbb{R})}( T^{N} ),
  \label{eqn-spin-chain-part-func-trace}
\end{eqnarray}
remembering that~$\sigma(N+1)\equiv \sigma(1)$, since for a smoothing operator~$A: L^2(\mb{R}) \lto \mathcal{S}(\mathbb{R})$ we have~$\ttr_{L^2(\mb{R})}(A)=\int_{}^{}K_A(x,x)\dd x$ with~$K_A$ being the integral kernel.

More generally we would like to express the kernel of~$T^N$ in terms of~$\exp(-S(\sigma))$ using the relation (\ref{eqn-spin-chain-inter-expression}). This means instead of letting the boundary condition be periodic we let $\sigma(1)=\sigma_{\mm{in}}$, $\sigma(N+1)=\sigma_{\mm{out}}$ given two boundary conditions $(\sigma_{\mm{in}},\sigma_{\mm{out}})\in \mathbb{R}^2 $. Then the kernel~$K_N$ of~$T^N$ is\footnote{Formula (\ref{eqn-intro-kernel-compo-transfer}) below is an analogue of what is called in the quantum information context the ``\textsf{replica trick}''.}
\begin{align}
  K_N(\sigma_{\mm{out}},\sigma_{\mm{in}})&=\int_{}^{}K(\sigma_{\mm{out}},\sigma(N))\cdots K(\sigma(2),\sigma_{\mm{in}}) \prod_{i=2}^N \dd\sigma(i) \label{eqn-intro-kernel-compo-transfer} \\
  &=\int_{\mb{R}^{N-1}}^{} \me^{-S_N(\sigma|\sigma_{\mm{in}},\sigma_{\mm{out}})} \prod_{i=2}^N \dd\sigma(i), \label{eqn-intro-kernel-path-integral}
\end{align}
with the \textsf{conditioned} interaction~$S_N(\sigma|\sigma_{\mm{in}},\sigma_{\mm{out}})$ defined as
\begin{equation}
  S_N(\sigma|\sigma_{\mm{in}},\sigma_{\mm{out}})\defeq \sum_{i=1}^{N} \vert \sigma(i+1)-\sigma(i)\vert^2 + \sum_{i=2}^N P(\sigma(i)) +\frac{1}{2}(P(\sigma_{\mm{in}})+P(\sigma_{\mm{out}})),
  \label{eqn-intro-interact-condition}
\end{equation}
with~$\sigma(1)\equiv \sigma_{\mm{in}}$ and~$\sigma(N+1)\equiv\sigma_{\mm{out}}$. Now if~$N_1$,~$N_2$ are two integers and we define the kernels~$K_{N_1}$,~$K_{N_2}$ using (\ref{eqn-intro-kernel-path-integral}) and (\ref{eqn-intro-interact-condition}) with~$N$ replaced respectively by~$N_1$,~$N_2$, then it follows ``trivially'' from the composition property~$T^{N_2}\circ T^{N_1}=T^{N_2+N_1}$ that
\begin{equation}
  K_{N_2+N_1}(\sigma_{\mm{out}},\sigma_{\mm{in}})=\int_{}^{}K_{N_2}(\sigma_{\mm{out}},\sigma)K_{N_1}(\sigma,\sigma_{\mm{in}}) \dd\sigma.
  \label{eqn-intro-composition}
\end{equation}

However, such a relation becomes remarkable (rather than trivial) if we do not have the ``unit'' transfer operator~$T$ to start with; that is, if we do not have (\ref{eqn-intro-kernel-compo-transfer}) but define~$K_{N_1}$ and~$K_{N_2}$ \textit{directly} with an expression of the form (\ref{eqn-intro-kernel-path-integral}) and (\ref{eqn-intro-interact-condition}). 
This corresponds to the idea of a \textsf{path integral} in quantum mechanics and quantum field theory. Alternatively one could consider the Gibbs measure (\ref{eqn-spin-chain-gibbs-part-func}) and take~$K_N$ as the \textit{transition probability} of a certain stochastic process. Then (\ref{eqn-intro-composition}) is the \textsf{Chapman-Kolmogorov} equation which relies heavily on the fact that the underlying process is \textsf{Markovian}. In a sense, for both interpretations a crucial condition is that the interaction~$S(\sigma)$ be \textsf{local}; that is, very roughly speaking, if one chops the sites~$[1,N_1+N_2]:=\{1,2,\dots,N_1+N_2\}$ into~$[1,N_1]\sqcup [N_1+1,N_1+N_2]$ then~$S_{[1,N_1+N_2]}(\sigma)\approx S_{[1,N_1]}(\sigma|_{[1,N_1]})+S_{[N_1+1,N_1+N_2]}(\sigma|_{[N_1+1,N_1+N_2]})$.

The main result of this article concerns a 2-dimensional ``continuum'' version of this story where lattice sites are replaced by the continuum of points on a 2D surface (considered as space-time) and a configuration is replaced by a distribution. See the section below for a more precise description. In the final section, we also show that when the space-time admits a periodic translation symmetry then a more precise analogy with the spin chain described above can be restored, in particular, there exists a \textsf{Gibbs state} in the \textsf{thermodynamic limit}. Further discussion of the above example continues in section \ref{sec-spin-chain-cont}.

\subsection{Main Results}

This article will prove three main results whose preliminary versions are stated as theorems \ref{thrm-intro-main-1}, \ref{thrm-intro-main-2}, and \ref{thrm-intro-main-3} below, among which the first is the main theorem and the second and third are two main consequences.

\begin{thrm}
  [Segal Axioms for~$P(\phi)_2$, partial statement] \label{thrm-intro-main-1} For each finite disjoint union~$\Sigma$ of Riemannian circles (each component circle characterised by its radius) there exists a finite measure~$\mu_{\Sigma}$ on the space~$\mathcal{D}'(\Sigma)$ of real distributions\footnote{By a \textsf{real distribution} $\phi$ we mean that if we pair~$\phi$ with a potentially complex test function~$f$, then $\ol{\phi(f)}=\phi(\ol{f})$,
  the bar denoting complex conjugation.} over~$\Sigma$, giving the Hilbert space~$\mathcal{H}_{\Sigma}:= L^2(\mathcal{D}'(\Sigma),\mu_{\Sigma})$, such that the following holds: 
  if~$\Omega$ is a Riemannian surface whose boundary has two components~$\partial\Omega=\Sigma_{\mm{in}}\sqcup \Sigma_{\mm{out}}$, denote by~$\mu_{\mm{in}}$,~$\mu_{\mm{out}}$ respectively the measures on~$\mathcal{D}'(\Sigma_{\mm{in}})$,~$\mathcal{D}'(\Sigma_{\mm{out}})$ and put~$\mathcal{H}_{\mm{in}}:=L^2(\mathcal{D}'(\Sigma_{\mm{in}}),\mu_{\mm{in}})$,~$\mathcal{H}_{\mm{out}}:=L^2(\mathcal{D}'(\Sigma_{\mm{out}}),\mu_{\mm{out}})$, then there exists an operator~$U_{\Omega}:\mathcal{H}_{\mm{in}}\lto \mathcal{H}_{\mm{out}}$ given by
  \begin{equation}
    (U_{\Omega}F)(\varphi_{\mm{out}})=\int_{}^{}\mathcal{A}_{\Omega}(\varphi_{\mm{in}},\varphi_{\mm{out}}) F(\varphi_{\mm{in}})\,\dd\mu_{\mm{in}}(\varphi_{\mm{in}}),\quad F\in \mathcal{H}_{\mm{in}},
    \label{}
  \end{equation}
  where $\mathcal{A}_{\Omega}(\varphi_{\mm{in}},\varphi_{\mm{out}})$ is given by the Radon-Nikodym density between two mutually absolutely continuous finite measures on~$\mathcal{D}'(\Sigma_{\mm{in}})\times \mathcal{D}'(\Sigma_{\mm{out}})$ defined rigorously in Eqn. (\ref{eqn-def-amp-omega}), Subsection \ref{sec-amplitude}.

Moreover, the correspondence~$\Omega\longmapsto U_{\Omega}$ between surfaces and operators satisfy the following \textsf{composition/functorial property}. Let~$\Omega_1$,~$\Omega_2$ be two Riemannian surfaces of the kind as above where the ``out'' boundary of~$\Omega_1$ is isometric to the ``in'' boundary of~$\Omega_2$ via an isometry~$\rho$, we glue them along~$\rho$ and obtain the surface~$\Omega_2\cup_{\rho}\Omega_1$. Then
\begin{equation}
  U_{\Omega_2\cup_{\rho}\Omega_1}=U_{\Omega_2}\circ U_{\Omega_1},
  \label{eqn-intro-state-compo}
\end{equation}
for the operators~$U_{\Omega_2\cup_{\rho}\Omega_1}$,~$U_{\Omega_2}$ and~$U_{\Omega_1}$ defined as above.
\end{thrm}

See Subsection \ref{sec-segal-descript}, Theorem {\color{blue} 1$'$} for a full statement.

The first application of Theorem \ref{thrm-intro-main-1} generalizes to the interacting cases, certain results of Naud~\cite{Nauddet} on the asymptotics of the free energy of free bosons on certain large degree covers.

\begin{thrm}[asymptotics of the partition function on large degree cyclic covers]\label{thrm-intro-main-2}

 Let $(M,g)$ denotes a compact Riemannian surface, $M_\infty\lto M$ a Riemannian $\mathbb{Z}$-cover of $M$ and $\gamma$ the generator of the deck group.
We can think of $M_\infty$ as the infinite composition of some given cobordism $\Omega$, which serves as the ``fundamental domain'' of the deck transformations on $M_\infty$. 
For every $N\in \mb{N}$, denote by $M_N:=M_\infty/\gamma^N $ the cyclic cover of degree $N$ of $M$ and $\Delta_N$ is the Laplace-Beltrami operator on $M_N$.
Let $Z_N$ be the partition function of the $P(\phi)_2$ theory on $M_N$ obtained in Theorem \ref{thrm-nelson}.

Then the renormalized sequence of free energies
$
\frac{1}{N}\log\left(Z_N \right)    
$ has a limit $\lambda_0$ when $N\rightarrow +\infty$, moreover this limit $\lambda_0$ can be interpreted as the leading eigenvalue of some transfer operator $U_\Omega$ which is the quantization of the cobordism $\Omega$ mentioned in Theorem \ref{thrm-intro-main-1}.
\end{thrm}

In subsection~\ref{sec-perron-frob-gibbs}, we use the above $P(\phi)_2$ measures defined on the towers of cyclic covers $M_N$ to produce a $P(\phi)_2$-Gibbs state on the periodic surface $M_\infty$ of infinite volume. 

\begin{thrm}[Mass gap for $P(\phi)_2$ on $M_\infty$]\label{thrm-intro-main-3}
In the notations from Theorem~\ref{thrm-intro-main-2},
for compactly supported bounded observables $F$ on $M_\infty$, one can define expectations with respect to the $P(\phi)_2$ Gibbs state as
\begin{eqnarray*}
\mathbb{E}\left[F \right]  \defeq \lim_{N\rightarrow +\infty} \frac{1}{Z_N\cdot{\tts \det_{\zeta}}(\Delta_N+m^2)^{-\frac{1}{2}}} \int_{\mathcal{D}'(M_N)}^{} F(\phi)\me^{-\int_{M_N}^{}{:}P(\phi(x)){:}\dd V(x)}\dd\mu_{\mm{GFF}}^{M_N}(\phi),  
\end{eqnarray*}
where the construction of $Z_N$ and the sense of the latter integral is made in Theorem \ref{thrm-nelson}. The generator of the deck group $\gamma$ acts on $M_\infty$ by diffeomorphism which induces a shift map $\tau$ on observables which moves the support by $\gamma$.
The Gibbs state defined above is exponentially mixing under the shift map, for compactly supported $L^2$ observables~:
\begin{eqnarray*}
\mathbb{E}[ \tau^kF G]=\mathbb{E}[ F ]\mathbb{E}[ G]+ \mathcal{O}(\alpha^k)    
\end{eqnarray*}
for some $\alpha<1$ (see corollary~\ref{cor-segal-trans-exp-decay}).
\end{thrm}

\subsection{Main Novelties, Organization}

We would like to point out \textbf{four main novelties} of this article, regarding the earlier work \cite{Pickrell} on the subject as well as the related work \cite{GKRV}. See the published version of \cite{Lin} for a more detailed comparison with literature.

\paragraph{Segal Gluing Gives a Gapped Theory.}
Firstly, the relationship between theorems \ref{thrm-intro-main-2}, \ref{thrm-intro-main-3} and \ref{thrm-intro-main-1} above offers a new perspective that certain asymptotic results in pure geometry can be viewed as a consequence of an underlying Segal Quantum Field Theory whose ``transfer operators''~$U_{\Omega}$ satisfy certain nice properties. Besides, of course, theorem \ref{thrm-intro-main-2} itself is new in that the numbers~$Z_N$ are now partition functions of \textit{interacting} field theories instead of regularized determinants of geometric operators, which were studied previously and in our setting corresponds to free field theories. Moreover, in the process of proving theorem \ref{thrm-intro-main-2} one also obtains an infinite-volume~$P(\phi)_2$-theory (in one direction) over a curved periodic surface which has never been considered before in the literature. Moreover, the Segal view point gives a simple conceptual proof of the mass gap for the $P(\phi)_2$ state for any polynomial $P$ of degree $\geqslant 4$ bounded from below on a space of infinite volume, infinite genus. Theorem \ref{thrm-intro-main-3} seems to be one of the first mass gap results for interacting QFT on curved spaces.

\paragraph{A Conceptual Proof for Gluing GFFs.}

We also offer a novel treatment on Segal gluing of the GFF ($P=0$), which is natural and based on first principles, but nevertheless has not been followed by existing literature. We describe it briefly via a finite dimensional analogy: suppose two real random variables~$X$ and~$Y$ have joint probability law~$\mb{P}_{(X,Y)}$; then we have two equal expressions for their joint density~$p(x,y)=(\dd\mb{P}_{(X,Y)}/\dd \mathcal{L}_{\mb{R}}^{\otimes 2}) (x,y)$, namely
\begin{equation}
  p(x,y)=\frac{\dd(\pi^x_*\mb{P}_{(X,Y)})}{\dd \mathcal{L}_{\mb{R}}}(x)\frac{\dd \mb{P}_{Y|X=x}}{\dd \mathcal{L}_{\mb{R}}}(y)=\frac{\dd(\pi^y_*\mb{P}_{(X,Y)})}{\dd \mathcal{L}_{\mb{R}}}(y)\frac{\dd \mb{P}_{X|Y=y}}{\dd \mathcal{L}_{\mb{R}}}(x),
  \label{eqn-bayes-ordinary}
\end{equation}
where~$\pi^x$,~$\pi^y$ are respectively the projections onto the~$x$- and~$y$-axes,~$\mb{P}_{Y|X=x}$ denotes the conditional law of~$Y$ knowing ``$X=x$'', and vice versa for~$\mb{P}_{X|Y=y}$. In other words, one could evaluate~$p(x,y)$ by \textit{conditioning} in two alternative ways: on~$X$ or on~$Y$. 
In Proposition \ref{prop-bayes-gff} we prove a version of (\ref{eqn-bayes-ordinary}) for the GFF measure on~$\mathcal{D}'(M)$ for~$M$ a closed Riemannian surface, where the role of~$\pi^{x,y}$ are played by trace maps (restrictions) onto embedded Riemannian circles. 
Since the Segal amplitudes (Eqn.\ (\ref{eqn-def-amp-omega})) are related in a simple way to the trace-image measures on distributions over the circles (corresponding to~$p(x,y)$ above), we may prove Segal's gluing for the GFF essentially from just the definitions, the Markov property, and symmetry considerations (Proposition \ref{prop-bayes-free-end}). Then by the \textit{locality} of the~$P(\phi)_2$-interaction, the result extends to the interacting case. 

\paragraph{A Transparent Treatment of Locality.}

To go from~$P=0$ to~$P\ne 0$, a key ingredient is to show that the~$P(\phi)$-interaction functional of the field ---~$\int_{M}^{}P(\phi(x))\,\dd x$ heuristically --- is \textit{local}, which in fact enables the amplitude (cf.\ Eqn.\ (\ref{eqn-def-amp-omega})) to be defined rigorously in the first place. No precise description of this locality existed in the literature and no explicit construction was given to prove it except for some remarks in~\cite{Abdesselam, Duch}. The present article fills this gap in Subsection \ref{sec-locality}, based on a strengthening of Nelson's argument presented in Subsection \ref{sec-nelson-reg-indep} that allows a freedom in choosing the regulators in the renormalization process, producing the same measure in the limit (see also the introduction of Section \ref{sec-nelson-main}). This latter result is also new by itself.

\paragraph{Dirichlet-to-Neumann Operator and Reflection Positivity.}
 As a tool in our treatment of the law of the restriction of GFF on circles, we exhibit an elementary relation between the operators~$\tau_{\Sigma}$,~$\Delta_M+m^2$,~$\PI_M^{\Sigma}$ and~$\DN_M^{\Sigma}$ (see ``notations'' below):
\begin{equation}
  j_{\Sigma}\defeq \tau_{\Sigma}^*=(\Delta_M+m^2)\PI_M^{\Sigma}(\DN_M^{\Sigma})^{-1},
  \label{}
\end{equation}
where~``$*$'' denotes the distributional adjoint. In Lemma \ref{lemm-DN-trick} we derive this from integration by parts (Green-Stokes theorem). Though elementary, we do not find this formula elsewhere. Via similar considerations, one also has (see Appendix \ref{sec-RP-first})
 \begin{equation}
    \PI_{\Omega}^{\partial\Omega}(\DN_{\Omega}^{\partial\Omega})^{-1}(\PI_{\Omega}^{\partial\Omega})^*=C_N-C_D,
    \label{}
  \end{equation}
  where~$C_D$,~$C_N$ respectively are the Green operators (of the massive Laplacian) with Dirichlet and Neumann conditions. This gives a \textbf{direct relation} between the positivity of the Dirichlet-to-Neumann map and the reflection positivity (RP) of the GFF on closed Riemannian manifolds with reflection symmetry. It has never been made explicit previously.

\paragraph{Organisation.} In Section \ref{sec-gaus-fields} we recall the preliminaries on which our construction will be based; these include the definition and properties of the (massive) Gaussian Free Field (GFF) measure on $\mathcal{D}'(M)$ (Subsection \ref{sec-mas-GFF}), the functional determinants (Subsection \ref{sec-det}) and a formula computing the Radon-Nikodym densities between mutually absolutely continuous Gaussian measures on spaces of distributions (Subsection \ref{sec-quad-pert-rad-niko-main}), of which we include a full proof based on the method of \cite{GJ} section 9.3. 

In Section \ref{sec-nelson-main} we retrace the classical argument of Nelson \cite{Nel66} and construct the interacting $P(\phi)$ functional measure on $\mathcal{D}'(M)$ for $M$ a closed Riemannian surface, giving the partition function $Z_M$; the new result which says the interaction functional (and hence the limiting measure) is independent of the regulator for the renormalisation process, when the regulator is picked from a specific class, is Proposition \ref{prop-nelson-main}.

Section \ref{sec-trace-poisson} aims to derive the behavior of a Gaussian field under the trace map (restriction to hypersurface); our method is based on elementary relations between several geometric operators as mentioned above (Lemma \ref{lemm-DN-trick}), and as above it leads to the somewhat new perspective on reflection positivity (Subsection \ref{sec-RP-first}). Section \ref{sec-markov-main} discusses the Markov decomposition of the GFF (Subsection \ref{sec-markov-main}) that culminates in a ``Bayes formula'' (Proposition \ref{prop-bayes-gff}) for the GFF and proving locality of the $P(\phi)$ interaction (Subsection \ref{sec-locality}). 

Section \ref{sec-segal-main} formulates precisely and proves theorem \ref{thrm-intro-main-1} while Section \ref{sec-period-cover} does it for theorem \ref{thrm-intro-main-2} and \ref{thrm-intro-main-3}.

\subsection{Notations}

\noindent In this paper, unless stated otherwise,

``$\heueq$'' means ``heuristically equal to'';

$(M,g)=$ closed Riemannian surface with metric~$g$;

$(\Omega,g)=$ compact Riemannian surface with metric~$g$, boundary~$\partial\Omega$, seen as isometrically embedded in~$M$, the induced metric on~$\partial\Omega$ is still denoted~$g$;

$\Sigma =$ disjoint union of Riemannian circles (of different radii), usually~$\Sigma =\partial\Omega$; hypersurface in $M$;

$\Omega^{\circ}=$ interior of $\Omega$;

$m=$ mass parameter,~$m>0$, \textit{fixed throughout paper};

$\mathcal{D}'(M)$,~$\mathcal{D}'(\Omega^{\circ})=$ \textsf{real} distributions on~$M$ and~$\Omega$;

$C^{\infty}(M)$,~$C_c^{\infty}(\Omega^{\circ})=$ smooth functions on~$M$, smooth compact support functions on~$\Omega^{\circ}$;

$\Delta=$ Laplacian on~$(M,g)$, $\Delta f:=-\ddiv (\nabla f)$, therefore it is defined to be \textsf{nonnegative}, the same applies below;

$\Delta_{\Sigma}=$ Laplacian on~$\Sigma$;

$\Delta_{\Omega,D}=$ Laplacian on~$\Omega$ with (zero) Dirichlet boundary condition (see also remark \ref{rem-Dir-Neu-cond-mean});

$\mn{D}_{\Sigma}=(\Delta_{\Sigma}+m^2)^{1/2}$;

$\tau_{\Sigma}=$ trace operator (restriction) onto the hypersurface $\Sigma$;

$(Q,\mathcal{O})=$ general probability sample space with~$\sigma$-algebra~$\mathcal{O}$;

$C=$ general positive self-adjoint elliptic pseudodifferential operator on~$M$,~$\Omega$,~$\Sigma$, which is Hilbert-Schmidt on the corresponding~$L^2$ spaces;

$\mu_C^{\mm{cond}}=\mu_{C^{-1}}^{\mm{cond}}=$ Gaussian measure on~$\mathcal{D}'(M)$,~$\mathcal{D}'(\Omega^{\circ})$, or~$\mathcal{D}'(\Sigma)$, equipped with their Fr\'echet Borel~$\sigma$-algebra, with covariance~$\mb{E}_C[\phi(f)\phi(h)]=\ank{f,C h}_{L^2}$, under some conditions;

$\mu_{\mm{GFF}}^M=$ massive GFF measure on $\mathcal{D}'(M)$; $\mu_{\mm{GFF}}^{\Omega,D}=$ massive Dirichlet GFF measure on $\mathcal{D}'(\Omega^{\circ})$;

$\mb{E}_{B}^A=$ expectation under $\mu_{B}^A$;

$\ank{-,-}_{L^2}=L^2$-inner product, pairing between $\mathcal{D}'(M)$ and $C^{\infty}(M)$, or between $W^s(M)$ and $W^{-s}(M)$, or any other pairing which is an extension of the $L^2$-inner product;

$\phi(f)=$ the random variable~$\phi\mapsto \ank{\phi,f}_{L^2}$ indexed by~$f\in C^{\infty}(M)$,~$C_c^{\infty}(\Omega^{\circ})$ or~$C^{\infty}(\Sigma)$;

$W^s(M)=$ the~$L^2$ Sobolev space on~$M$, with inner product~$\ank{-,-}_{W^s(M)}:=\ank{-,(\Delta+m^2)^{s} -}_{L^2}$;

$W^s_{A}(M)$, $W^s_U(M)$, $W^s(U)=$ see appendix \ref{sec-app-sobo};

$p_{M\setminus A}^{\perp}$,~$p_{M\setminus A}$,~$P_{ A}$, and~$P_{ A}^{\perp}=$ Sobolev projections defined in lemma \ref{lemm-sobo-decomp};

$\Psi^r(M)$, $\Psi^r(\Sigma)=$ pseudodifferential operators ($\Psi$DOs) on~$M$ or $\Sigma$ with order~$r$;

$[\mathcal{L}\phi]=$ hypothetical Lebesgue measure on a space of distributions;

$\nrm{\cdot}_{\mm{tr}}$, $\nrm{\cdot}_{\mm{HS}}$, $\nrm{\cdot}_{L^2}=$ trace norm, Hilbert-Schmidt norm, operator norm acting on $L^2$ or $L^2$-norm on function;

$\PI_{\Omega}^{\Sigma,B}=$ Poisson integral operator extending from~$\Sigma$ to~$\Omega$ (definition \ref{def-pi-more-general}), with boundary conditions~$B$ imposed on boundary components \textit{other than}~$\Sigma$ (see also remark \ref{rem-Dir-Neu-cond-mean});

$\DN_{\Omega}^{\Sigma,B}=$ Dirichlet-to-Neumann operator or the jumpy version defined using~$\PI_{\Omega}^{\Sigma,B}$ (jumpy version is understood when~$\Sigma$ is in the interior of~$\Omega$ rather than a boundary component);

$\PI_{M}^{\Sigma}$, $\DN_{M}^{\Sigma}=$ Poisson integral operator extending from~$\Sigma$ to~$M$ (see (\ref{eqn-def-pi-emb-hyp-in-closed-case})), and jumpy Dirichlet-to-Neumann operator defined using $\PI_M^{\Sigma}$ (definition \ref{def-jp-DN-map}).

$\mu_{\DN}^{\Sigma,\Omega,B}=$ Gaussian measure on~$\mathcal{D}'(\Sigma)$ with covariance~$(\DN_{\Omega}^{\Sigma,B})^{-1}$;

$\mu_{2\DN}^{\partial\Omega,\Omega}=$ Gaussian measure on~$\mathcal{D}'(\partial\Omega)$ with covariance~$\frac{1}{2}(\DN_{\Omega}^{\partial\Omega})^{-1}$;

$\mu_{2\mn{D}}^{\Sigma}=$ Gaussian measure on~$\mathcal{D}'(\Sigma)$ with covariance~$\frac{1}{2}(\mn{D}_{\Sigma})^{-1}$;

$\mathcal{M}_{M,2}^1= \tau_{\Sigma_2}\PI_{M}^{\Sigma_1}$ is the \textsf{transition operator} defined in section \ref{sec-bayes-gff}.

  \section{Gaussian Fields on a Riemannian Manifold}
  \label{sec-gaus-fields}

\noindent \textbf{In this and the next section} we define rigorously a probability measure on~$\mathcal{D}'(M)$, as well as certain variants on~$\mathcal{D}'(\Omega^{\circ})$, which heuristically bears the form
\begin{equation}
  \dd\mu_P(\phi)\heueq\frac{1}{Z}\overbrace{\me^{-\int_{M}^{}P(\phi(x))\dd V_M (x)}}^{A} \underbrace{\me^{-\frac{1}{2}\int_{M}^{}(|\nabla \phi|_g^2 +m^2\phi^2 )\dd V_M}[\mathcal{L}\phi]}_{B},
  \label{eqn-def-mes-gibbs-heu}
\end{equation}
where~$P$ is a polynomial bounded from below, such as $P(\phi)=\phi^4-\phi^2$,~$[\mathcal{L}\phi]$ denotes the nonexistent Lebesgue measure on~$\mathcal{D}'(M)$ and
\begin{equation}
  Z_M\heueq\int_{\mathcal{D}'(M)}^{} \me^{-\int_{M}^{}P(\phi(x))\dd V_M(x)}\me^{-\frac{1}{2}\int_{M}^{}(|\nabla \phi|_g^2 +m^2\phi^2 )\dd V_M}[\mathcal{L}\phi]
  \label{}
\end{equation}
is the normalization factor, also called the \textsf{partition function}. 

The idea is that while~$[\mathcal{L}\phi]$ is nonexistent, together with the factor~$\exp(-\frac{1}{2}\int_{M}^{}(|\nabla \phi|_g^2 +m^2\phi^2 )\dd V_M)$ the expression~$B$ can be given a rigorous meaning as a Gaussian probability measure scaled by a (finite) volume constant, and the measure (\ref{eqn-def-mes-gibbs-heu}) can be constructed if, after defining part $A$ rigorously, one proves that it is~$L^1$ with respect to the measure~$B$. 

\textbf{In this section} we construct the measure $B$. In analogy with the expression of Gaussian measures on $\mb{R}^n$, one should define heuristically
\begin{equation}
  \me^{-\frac{1}{2}\int_{M}^{}(|\nabla \phi|_g^2 +m^2\phi^2 )\dd V_M}[\mathcal{L}\phi] \heueq \textrm{``}\det\textrm{''}(\Delta+m^2)^{-\frac{1}{2}} \dd\mu_{\mm{GFF}}^M (\phi)
  \label{eqn-heu-gaussian-mea-det-volume}
\end{equation}
where~$\mu_{\mm{GFF}}^M$ is a Gaussian measure on~$\mathcal{D}'(M)$ with \textit{covariance operator}~$(\Delta+m^2)^{-1}$, and ``$\det$'' is an infinite dimensional generalization of the determinant of a matrix. The Gaussian measure and the determinant are two issues to be treated separately (in sections \ref{sec-mas-GFF} and \ref{sec-det}), both being rather classical, the former called \textsf{Gaussian Free Field (measure)} and the latter called the \textsf{$\zeta$-regularized determinant}.

  \label{sec3.1}

 \subsection{The Massive Gaussian Free Field}\label{sec-mas-GFF}

 \subsubsection{Definition and Representations}
In this subsection we explain the point of view adopted in this article of the massive Gaussian Free Field (GFF), and refer to Sheffield \cite{Shef} and Powell and Werner \cite{powell-werner} for more information. 

Let~$(M,g)$ be a closed Riemannian manifold \textsf{of dimension~$d$} with metric~$g$, and~$\Omega \subset M$ an open domain with smooth boundary~$\partial\Omega$, both equipped with the metric induced from~$g$ (same notation). Fix~$m>0$ as the mass parameter. We say that
  the \textsf{massive Gaussian Free Field (GFF)} with mass~$m$ on~$M$ is the \textit{Gaussian random process indexed by}~$C^{\infty}(M)$, consisting of random variables~$\{\phi(f)~|~f\in C^{\infty}(M)\}$ such that
  \begin{equation}
    \mb{E}[\phi(f)\phi(h)]=\ank{f,(\Delta+m^2)^{-1}h}_{L^2(M)},\quad \mb{E}[\phi(f)]\equiv 0,
    \label{eqn-gff-cov-closed}
  \end{equation}
  for any~$f$,~$h\in C^{\infty}(M)$. Similarly, the \textsf{Dirichlet massive Gaussian Free Field} with mass~$m$ on~$\Omega$ is the \textit{Gaussian random process indexed by}~$C_c^{\infty}(\Omega^{\circ})$, consisting of random variables~$\{\phi(f)~|~f\in C_c^{\infty}(\Omega^{\circ})\}$ such that
  \begin{equation}
    \mb{E}[\phi(f)\phi(h)]=\bank{f,(\Delta_{\Omega,D}+m^2)^{-1}h}_{L^2(\Omega)}=\bank{P_{M\setminus \Omega^{\circ}}^{\perp}f, P_{M\setminus \Omega^{\circ}}^{\perp}h}_{W^{-1}},\quad \mb{E}[\phi(f)]\equiv 0,
    \label{}
  \end{equation}
  for any~$f$,~$h\in C_c^{\infty}(\Omega^{\circ})$. See appendix \ref{sec-app-sobo} and in particular lemma \ref{lemm-diri-green-op-quad-form}.

There exist many choices of sample spaces (``$Q$-spaces'' in the terminology of \cite{Sim2}) on which to realize those Gaussian processes. These realizations are all equivalent in the sense of \textsf{isomorphism of measure algebras} (see \cite{Sim2} section I.2). For the sake of concreteness, we point out that one choice for the sample space is~$\mathcal{D}'(M)$ (or~$\mathcal{D}'(\Omega^{\circ})$), as formulated in the following proposition, whose proof parallels the case on Euclidean space with minor modification.

\begin{prop}
  [Bochner-Minlos, \cite{BHL} theorem 5.11, page 266] \label{prop-boch-min} There exists a Borel probability measure~$\mu_{\mm{GFF}}^M$ on the Fr\'echet space~$\mathcal{D}'(M)$ such that~$\phi\mapsto \ank{\phi,f}_{L^2(M)}=:\phi(f)$,~$f\in C^{\infty}(M)$, realizes the random variable~$\phi(f)$ of the massive GFF on~$M$. 
  
  Similarly, there exists a Borel probability measure~$\mu_{\mm{GFF}}^{\Omega,D}$ on the Fr\'echet space~$\mathcal{D}'(\Omega^{\circ})$ such that~$\phi\mapsto \ank{\phi,f}_{L^2(\Omega)}$,~$f\in C_c^{\infty}(\Omega^{\circ})$, realizes the random variable~$\phi(f)$ of the Dirichlet massive GFF on~$\Omega$.\hfill~$\Box$
\end{prop}

\begin{notation}
  Denote by~$\mb{E}_{\mm{GFF}}^M$ and~$\mb{E}_{\mm{GFF}}^{\Omega,D}$ the expectations under~$\mu_{\mm{GFF}}^M$ and~$\mu_{\mm{GFF}}^{\Omega,D}$ respectively.
\end{notation}

\begin{def7}\label{rem-mass-gau-field-four}
  Alternatively, based on the spectral theory of~$\Delta$, let~$\left\{ \varphi_j \right\}_{j=0}^{\infty}$ be its complete orthonormal eigenfunctions with (real nonnegative) eigenvalues~$\left\{ \lambda_j \right\}_{j=0}^{\infty}$,~$0=\lambda_0<\lambda_1\le \lambda_2\le \dots\le \lambda_j\le \cdots$, counted with multiplicity. Then the \textsf{GFF with mass~$m$} on~$M$ could also be represented as the random formal series
  \begin{equation}
    \phi=\sum_{j=0}^{\infty}\xi_j \varphi_j
    \label{eqn3.3}
  \end{equation}
  where the sequence~$(\xi_j)_j$ consists of i.i.d.\ real-valued random variables with $\xi_j$ being the standard centered Gaussian on~$\mb{R}$ with variance~$(\lambda_j+m^2)^{-1}$. Similarly, the Dirichlet GFF on~$\Omega$ could also be so represented using eigenfunctions of~$\Delta_{\Omega,D}$, which are complete for~$L^2(\Omega)$. One way to see that this construction is equivalent to the previous one is by appealing to \cite{Shu} page 92 proposition 10.2.
\end{def7}

We also consider a slightly more general situation  where~$(\Delta+m^2)^{-1}$ is replaced by an operator $C$.

\begin{deef}
  We say that a bounded self-adjoint positive elliptic pseudodifferential operator~$C$ of order~$-s$ on~$M$ ($s>0$) is a \textsf{Gaussian covariance operator} of order~$-s$.
\end{deef}

Following similar reasoning as proposition \ref{prop-boch-min} or remark \ref{rem-mass-gau-field-four}, one obtains a measure~$\mu_C$ on~$\mathcal{D}'(M)$.

\begin{deef}
  The \textsf{Gaussian Field} on~$M$ with \textsf{covariance operator}~$C$ is the \textit{Gaussian random process indexed by}~$C^{\infty}(M)$, consisting of random variables~$\{\phi(f)~|~f\in C^{\infty}(M)\}$ such that
  \begin{equation}
    \mb{E}[\phi(f)\phi(h)]=\ank{f,C h}_{L^2(M)},\quad \mb{E}[\phi(f)]\equiv 0,
    \label{}
  \end{equation}
  for any~$f$,~$h\in C^{\infty}(M)$. We denote the corresponding measure on~$\mathcal{D}'(M)$ by~$\mu_C$ and the expectation with respect to this measure by~$\mb{E}_C$.
\end{deef}

\begin{def7}
  For~$C$ satisfying the assumptions, the inner product~$(f,h)\mapsto \ank{f, Ch}_{L^2(M)}$ defines an equivalent norm for~$W^{-s}(M)$.
\end{def7}

\subsubsection{Essential Properties}

Now we collect some properties of the Gaussian random fields constructed above. We start with the following classical fact (the proofs are in parallel with the Euclidean case for which one could refer to \cite{Shef}).
\begin{lemm}\label{lemm-ghs-for-gff}
  The \textsf{Gaussian Hilbert space} of~$\mu_{\mm{GFF}}^M$ is~$W^{-1}(M)$, and that of~$\mu_{\mm{GFF}}^{\Omega,D}$ is~$W^{-1}(\Omega)$. The \textsf{Cameron-Martin space} of~$\mu_{\mm{GFF}}^M$ is~$W^{1}(M)$, and that of~$\mu_{\mm{GFF}}^{\Omega,D}$ is~$W^{1}_{\Omega}(M)$.\hfill~$\Box$
\end{lemm}

\begin{def7}\label{rem-use-cam-mar-pairing-not-l2}
  There is a tacit assumption in the way we defined our fields: we expected the random variable~$\phi(f)$ to come from the distributional pairing ($L^2$-pairing) between~$\mathcal{D}'(M)$ and~$C^{\infty}(M)$ (respectively,~$\Omega$). Other pairings may also be used. For example, let~$\phi(f)$ be~$\ank{-,f}_{W^1(M)}$ instead of~$\ank{-,f}_{L^2(M)}$. One then needs to alter the covariances accordingly. Indeed, they are related by
  \begin{equation}
    \mb{E}\big[\bank{\phi,f}_{W^1}\bank{\phi,h}_{W^1}\big]=\mb{E}\big[
    \bank{\phi,(\Delta+m^2)f}_{L^2}\bank{\phi,(\Delta+m^2)h}_{L^2}\big]=\bank{f,h}_{W^1(M)}.
    \label{}
  \end{equation}
  This way of definition is noticeably used by Sheffield \cite{Shef}. The Gaussian Hilbert space in this case is~$W^1(M)$ (respectively,~$W^1_{\Omega}(M)$), and the Cameron-Martin spaces are the same. See also remark \ref{rem-natural-sobo-cam-mar}.
\end{def7}

Similarly,

\begin{lemm}\label{lemm-GHS-cam-mar-of-cov-C}
  The Gaussian Hilbert space of~$\mu_C$ is~$W^{-s}(M)$, equipped with~$\ank{-,C-}_{L^2}$, and the Cameron-Martin space is~$W^s(M)$, equipped with~$\sank{-,C^{-1}-}_{L^2}$. \hfill~$\Box$
\end{lemm}

Finally we say a word about the supports of the measures in the closed manifold case. This can be proved either using lemma \ref{lemm-quadratic-pert-trace} and evaluating the expectation of $\snrm{\phi}_{W^s(M)}^2$ or a spectral representation like in remark \ref{rem-mass-gau-field-four}.

\begin{lemm}\label{lemm-supp-gaus-meas}
  We have~$\mu_C(W^{-\delta}(M))=1$ for any~$\delta>\frac{1}{2}(d-s)$.\hfill~$\Box$
\end{lemm}

\begin{def7}
  We point out that~$\mu_C(W^{-\delta}(M))=1$ for any~$\delta>0$ in the following two cases:
  \begin{enumerate}[(i)]
    \item $\dim M=2$ and~$\mu_C=\mu_{\mm{GFF}}^M$;
    \item $\dim M=1$ and~$C$ has order~$-1$.
  \end{enumerate}
\end{def7}

Last but not least, we make the following innocent but useful observation.

\begin{lemm}\label{lemm-gaus-field-disj-indep}
  If~$\Omega=\Omega_1\sqcup \Omega_2$ (possibility of non-empty boundary in either or both components), then GFFs (indeed, Gaussian fields) over~$\Omega_1$ and~$\Omega_2$ are independent and~$\mu_{\mm{GFF}}^{\Omega,B}=\mu_{\mm{GFF}}^{\Omega_1,B}\otimes \mu_{\mm{GFF}}^{\Omega_2,B}$ where~$B=D$ when the corresponding~$\Omega$,~$\Omega_1$ or~$\Omega_2$ has boundary and~$B=\varnothing$ when either is closed. \hfill~$\Box$
\end{lemm}

 \subsection{Determinants}\label{sec-det}

 \noindent In this section we discuss (two) generalizations of the notion of the determinant (of a matrix) to infinite dimensional operators. General references include Kontsevich and Vishik \cite{KV}, Shubin \cite{Shu} sections 9-13, Simon \cite{Sim1} chapter 3 and finally Gohberg, Goldberg and Krupnik \cite{GGK}. See also Dang \cite{Dang} for a quick acquaintance of the physical-geometric context and Quine, Heydari and Song \cite{QHS} for an interesting discussion of zeta-regularization of infinite products.

\subsubsection{Zeta-regularized and Fredholm Determinants}\label{sec-det-def}

\noindent The zeta-regularized determinant was first introduced by Ray and Singer \cite{RS}. The first step is to define the \textsf{zeta function} of a (rather special) pseudodifferential operator~$A$ over a manifold $M$ with or without boundary,
\begin{equation}
  \zeta_A(z)\defeq \ttr_{L^2}(A^{-z})
  \label{eqn-zeta-def}
\end{equation}
as a function of the complex variable~$z$ and study its meromorphic extension over a region that includes~$z=0$. For our purposes ~$A$ is~$\Delta+m^2$,~$(\Delta+m^2)^{1/2}$ or a Dirichlet-to-Neumann operator. These are positive elliptic~$\Psi$DOs of positive order such that the principal symbol~$\sigma_A(x,\xi)$ is \textsf{strictly positive} whenever~$\xi\ne 0$. In particular their spectra are in~$\mb{R}_+$ and does not intersect~$\mb{B}_{\rho}(0)\subset\mb{C}$ for some~$\rho>0$. This enables one to define the complex power~$A^{-z}$ using the Cauchy integral representation
\begin{equation}
  A^{-z}\defeq \frac{\ii}{2\pi}\int_{\gamma}^{}\me^{-z\log \lambda}(A-\lambda)^{-1}\dd\lambda,
  \label{eqn-complex-pow-cauchy}
\end{equation}
where~$\gamma$ is the contour (with parametrization traversing in order)
\begin{equation}
  \gamma=\{r\me^{\ii\pi}~|~r>\rho\}\cup\{\rho\me^{\ii\theta}~|~-\pi<\theta<\pi\}\cup\{r\me^{-\ii \pi}~|~r>\rho\},
  \label{}
\end{equation}
and~$\log\lambda$ taken to be the principal branch defined on~$\mb{C}\setminus (-\infty,0]$ with $\log 1=0$.

\begin{prop}[\cite{Shu} proposition 10.1, theorems 10.1, 13.1, 13.2, also \cite{Seeley}]\label{prop-def-det-zeta} We have
\begin{enumerate}[(i)]
    \item For the operators~$A$ under consideration, we have the bound
  \begin{equation}
    \nrm{(A-\lambda)^{-1}}_{L^2}\le c|\lambda|^{-1}
    \label{eqn-resolve-op-norm-bound}
  \end{equation}
  for~$\lambda\in \gamma$, and the integral (\ref{eqn-complex-pow-cauchy}) defines~$A^{-z}$ as a holomorphic function valued in bounded~$L^2$ operators, for~$\fk{Re}(z)>0$. It continues as such a holomorphic operator function to all~$z\in\mb{C}$ via
  \begin{equation}
    A^{-z}\defeq A^k A^{-z-k},
    \label{eqn-complex-pow-cont}
  \end{equation}
  where~$k$ is any integer with~$\fk{Re}(z)>-k$ so that~$A^{-z-k}$ is defined by (\ref{eqn-complex-pow-cauchy}), and the definition (\ref{eqn-complex-pow-cont}) does not depend on~$k$.
  \item If~$A$ has order~$s$ then~$A^{-z}$ defined as above is a classical~$\Psi$DO of order~$-zs$. In particular it is trace class on~$L^2(M)$ when~$\fk{Re}(z)>d/s$, $d=\dim M$, for which~$\zeta_A(z)$ is well-defined by (\ref{eqn-zeta-def}). Moreover, for these~$z$,
\begin{equation}
  \zeta_A(z)=\sum_{j=0}^{\infty}\lambda_j^{-z},
  \label{eqn-zeta-eigenval}
\end{equation}
where~$\{\lambda_j\}_{j=0}^{\infty}$ are the eigenvalues of~$A$, and the sum converges absolutely, and uniformly in~$z$ over~$\{\fk{Re}(z)>d/s+\varepsilon\}$ for any~$\varepsilon>0$. 
\item Finally,~$\zeta_A(z)$ can be meromorphically continued over~$\mb{C}$ with simple poles possible at~$\{\frac{d}{s},\frac{d-1}{s},\frac{d-2}{s},\cdots\}\setminus \mb{Z}_{\le 0}$, and holomorphic elsewhere. In particular, it is holomorphic at~$z=0$.\hfill~$\Box$
\end{enumerate}
\end{prop}

\begin{deef}
  For an operator~$A$ under consideration, we define its \textsf{zeta-regularized determinant} as
  \begin{equation}
    \detz A\defeq \exp\left( -\partial_z\zeta_A(0) \right),
    \label{}
  \end{equation}
  where~$\zeta_A(0)$ is the \textsf{zeta function} of~$A$ given by (\ref{eqn-zeta-def}) and (\ref{eqn-complex-pow-cauchy}).
\end{deef}

\begin{def7}
  An alternative way of defining the zeta function and its meromorphic continuation is to use the heat kernel and the Mellin transform. See Gilkey \cite{Gilkey} section 1.12. This way of definition also gives (\ref{eqn-zeta-eigenval}) over the same region, defining therefore the same function as ours.
\end{def7}

\begin{def7}\label{rem-zeta-det-positive}
  From (ii) of proposition \ref{prop-def-det-zeta} we see that if~$A$ is self-adjoint and strictly positive, then~$\zeta_A(z)$ is real-valued for~$z\in (d/s,+\infty)$. But~$\fk{Im}(\zeta_A)$ is real analytic and hence~$\zeta_A$ remains real-valued on~$\mb{R}$ before crossing a pole, and by (iii) it is in particular real-valued on an interval around~$0$. Thus~$\partial_z \zeta_A(0)$ is real and~$\detz A$ is positive.
\end{def7}

Now we move on to the second notion of determinant. Let~$\mathcal{H}$ be a Hilbert space and~$A\in \mathcal{L}(\mathcal{H})$. Denote by~$\Lambda^k \mathcal{H}$ and~$\Lambda^k A$, respectively, the~$k$-th exterior product of~$\mathcal{H}$ and~$A$ (see Simon \cite{Sim1} section 1.5).

\begin{prop}[\cite{Sim1} lemma 3.3]\label{prop-def-det-fred}
  If~$A$ is trace class on~$\mathcal{H}$, then~$\Lambda^k A$ is also trace class on~$\Lambda^k \mathcal{H}$ with bound on trace norm
  \begin{equation}
    \bnrm{\Lambda^k A}_{\mm{tr}}\le \frac{1}{k!}\bnrm{A}_{\mm{tr}}^k.
    \label{}
  \end{equation}
  In particular, putting
  \begin{equation}
    \detf(\one +zA)\defeq \sum_{k=0}^{\infty} z^k\ttr_{\Lambda^k\mathcal{H}}(\Lambda^k A)
    \label{eqn-det-fred-func}
  \end{equation}
  for~$z\in\mb{C}$ defines an entire function, and
  \begin{equation}
    |\detf(\one+z A)|\le \exp(|z|\nrm{A}_{\mm{tr}}).\quad \Box
    \label{}
  \end{equation}
\end{prop}

\begin{deef}
  Let~$A$ be a trace class operator on the Hilbert space~$\mathcal{H}$. Then the determinant~$\detf(\one+A)$ given by (\ref{eqn-det-fred-func}) for~$z=1$ is called the \textsf{Fredholm determinant} of~$\one+A$.
\end{deef}

\begin{lemm}[\cite{Sim1} theorems 3.4, 3.7, 3.8] \label{lemm-det-fred-cont} ~
\begin{enumerate}[(i)]
    \item The map~$A\longmapsto \detf(\one+A)$ defines a continuous function on the trace ideal~$\mathcal{J}_1$ with~$\nrm{\cdot}_{\mm{tr}}$. More precisely,
  \begin{equation}
    |\detf(\one+A)-\detf(\one+B)|\le \nrm{A-B}_{\mm{tr}}\exp(\nrm{A}_{\mm{tr}}+\nrm{B}_{\mm{tr}}+1).
    \label{}
  \end{equation}
  \item If~$A$,~$B\in \mathcal{L}(\mathcal{H})$ are such that both~$AB$ and~$BA$ are of trace class, then we have
  \begin{equation}
    \detf(\one+AB)=\detf(\one +BA). \quad~\Box
    \label{}
  \end{equation}
\end{enumerate}
\end{lemm}

\subsubsection{Factorization Lemma}

\begin{lemm}\label{lemm-det-factor}
  Suppose~$A$ and~$K$ are~$\Psi$DOs such that both~$A$ and~$A(\one+K)$ satisfy the assumptions of proposition \ref{prop-def-det-zeta} and that~$\detz(A)$ and~$\detz(A(\one+K))$ are defined. Suppose moreover~$K$ is trace class and there exists smoothing operators~$\{K_i\}_{i=1}^{\infty}$ such that~$AK_i\to AK$ in~$\nrm{\cdot}_{\mm{tr}}$ (in particular, $AK$ is also trace class). Then
  \begin{equation}
    \detz(A(\one+K))=\detz(A)\detf(\one +K).
    \label{eqn-det-factor}
  \end{equation}
\end{lemm}

\begin{proof}
  We start from Kontsevich and Vishik \cite{KV} proposition 6.4 and take for granted that (\ref{eqn-det-factor}) holds with~$K_i$ in place of~$K$. Our assumptions are tailor-made so that as~$i\to \infty$,
  \begin{equation}
    \detz(A(\one+K_i))\lto \detz(A(\one+K)).
    \label{}
  \end{equation}
  Indeed, by (\ref{eqn-resolve-op-norm-bound}) we have
\begin{equation}
  \nrm{(A(\one+K)-\lambda)^{-1}A(K_i-K)(A(\one+K_i)-\lambda)^{-1}}_{\mm{tr}}\le c|\lambda|^{-2}\bnrm{A(K_i-K)}_{\mm{tr}},
  \label{eqn-diff-resolv-bound}
\end{equation}
with~$c$ independent of~$i$ since a fortiori~$AK_i\to AK$ under~$\nrm{\cdot}_{L^2}$. This in particular shows when~$\fk{Re}(z)>-1$ the integral expression for~$A(\one+K)^{-z}-A(\one+K_i)^{-z}$ is a converging Bochner integral valued in the trace class ideal~$\mathcal{J}_1$ (note~$|\lambda^{-z}|\asymp |\lambda|^{-\fk{Re}(z)}$ as~$\lambda\to -\infty$) and since~$\ttr_{L^2}$ is a continuous functional on~$\mathcal{J}_1$,
\begin{equation}
  \big|\ttr_{L^2}(A(\one+K)^{-z}-A(\one+K_i)^{-z})\big|\lesssim \int_{\gamma}^{} |\lambda|^{-\fk{Re}(z)-2}\bnrm{A(K_i-K)}_{\mm{tr}}\dd\lambda.
  \label{eqn-diff-zeta-bound}
\end{equation}
Now by (iii) of proposition \ref{prop-def-det-zeta} there is~$\frac{1}{2}>\delta>0$ so that~$\zeta_{A(\one+K_i)}$,~$\zeta_{A(\one+K)}$ are both holomorphic over~$\ol{\mb{B}}_{\delta}(0)$ (for example,~$\delta<1/|s|$ where~$s$ is the order of~$A$). Thus by (\ref{eqn-diff-zeta-bound}) and Cauchy's estimate
\begin{align*}
  \big|\zeta_{A(\one+K)}'(0)-\zeta_{A(\one+K_i)}'(0)\big|&\le \frac{1}{\delta}\sup_{|z|=\delta}\big|\zeta_{A(\one+K)}(z)-\zeta_{A(\one+K_i)}(z)\big| \\
  &\lesssim \frac{1}{\delta}\sup_{|z|=\delta}\int_{\gamma}^{} |\lambda|^{-\fk{Re}(z)-2}\bnrm{A(K_i-K)}_{\mm{tr}}\dd\lambda \\
  &\lesssim \bnrm{A(K_i-K)}_{\mm{tr}}\int_{\gamma}^{}|\lambda|^{-3/2}\dd\lambda \\
  &\lesssim \bnrm{A(K_i-K)}_{\mm{tr}}.
\end{align*}
This shows~$|\zeta_{A(\one+K)}'(0)-\zeta_{A(\one+K_i)}'(0)|\to 0$, as we have desired.
\end{proof}

\subsubsection{The Gluing Formula of Burghelea-Friedlander-Kappeler}\label{sec-BFK}

\noindent Let~$(M,g)$ is a closed Riemannian manifold and~$\Sigma\subset M$ an embedded closed hypersurface with induced metric. Assume proposition \ref{prop-stoc-decomp-closed} and decompose~$\phi=\phi_{\Sigma}+\phi_{M\setminus\Sigma}^D$ corresponding to~$\mu_{\mm{GFF}}^M=\mu_{\mm{GFF}}^{M\setminus\Sigma,D}\otimes \mu_{\DN}^{\Sigma,M}$. In view of equation (\ref{eqn-heu-gaussian-mea-det-volume}), and in parallel
\begin{equation}
  \me^{-\frac{1}{2}\ank{\varphi,\DN \varphi}_{L^2}}[\mathcal{L}\varphi]\heueq\detz(\DN_M^{\Sigma})^{-\frac{1}{2}}\dd\mu_{\DN}^{\Sigma,M},
  \label{}
\end{equation}
if we assume a ``formal Fubini theorem'' with respect to the heuristic expressions involving~$\mathcal{L}$, namely
\begin{equation}
  \int_{}^{}\me^{-\frac{1}{2}\sank{\phi,(\Delta+m^2)\phi}_{L^2}}[\mathcal{L}\phi] \heueq\iint \me^{-\frac{1}{2}\sank{\phi_{\Sigma},(\Delta+m^2)\phi_{\Sigma}}_{L^2}}\me^{-\frac{1}{2}\sank{\phi_{M\setminus\Sigma}^D,(\Delta+m^2)\phi_{M\setminus\Sigma}^D}_{L^2}} [\mathcal{L}\phi_{\Sigma}]\otimes [\mathcal{L}\phi_{M\setminus\Sigma}^D],
  \label{}
\end{equation}
then we are led to the following relation of the corresponding determinants (volumes) which were first rigorously proved by Burghelea, Friedlander and Kappeler \cite{BFK}.

\begin{prop}
  [\cite{BFK} theorem B, \cite{Lee} theorem 1.1] Let~$(M,g)$ is a closed Riemannian surface and~$\Sigma\subset M$ an embedded closed hypersurface with induced metric. Then
  \begin{equation}
    \detz(\Delta_M+m^2)=\detz(\Delta_{M\setminus\Sigma,D}+m^2)\detz(\DN_M^{\Sigma}). \quad\Box
    \label{eqn-bfk-non-disec}
  \end{equation}
\end{prop}

The following version where~$\Sigma$ dissects~$M$ such that~$M\setminus \Sigma=M^{\circ}_+\sqcup M^{\circ}_-$ is also useful.

\begin{corr}
  In the situation as above, we have
  \begin{equation}
    \detz(\Delta_M+m^2)=\detz(\Delta_{M_+,D}+m^2)\detz(\Delta_{M_-,D}+m^2)\detz(\DN_M^{\Sigma}).\quad\Box
    \label{eqn-bfk-disec}
  \end{equation}
\end{corr}

As we base our analysis on background Gaussian probability measures, the formulae (\ref{eqn-bfk-non-disec}) and (\ref{eqn-bfk-disec}) constitute separate ingredients (constants) that needs to be ``tuned'' for the final gluing result to hold exactly. In fact, it is also reasonable to consider ``projective gluing'' which allows the freedom for an arbitrary (nonzero) constant to appear in the equation (see Segal \cite{Segal} page 460).

\subsection{Quadratic Perturbation = Radon-Nikodym Density}\label{sec-quad-pert-rad-niko-main}

\noindent Let~$C$ be a Gaussian covariance operator of order $-s$ on a closed Riemannian manifold~$\Sigma$, and denote by~$\mu_C$ the Gaussian measure on~$\mathcal{D}'(\Sigma)$ with covariance~$\ank{-, C-}_{L^2(\Sigma)}$. Let~$V$ be another bounded formally self-adjoint operator on~$L^2(\Sigma)$ (it could be given by a real symmetric Schwartz kernel). In this section we look at the Gibbs measure
\begin{equation}
  \dd\mu (\varphi) \defeq \frac{\me^{-\frac{1}{2}\ank{\varphi,V\varphi}_{L^2}}\dd\mu_C(\varphi)}{\int_{}^{}\me^{-\frac{1}{2}\ank{\varphi,V\varphi}_{L^2}}\dd\mu_C},
  \label{eqn-quadratic-gibbs-meas}
\end{equation}
which is a Gaussian measure (see proposition \ref{prop-quadratic-perturbation}).

From another perspective we consider Radon-Nikodym densities between mutually absolutely continuous Gaussian measures on~$\mathcal{D}'(\Sigma)$. See Bogachev \cite{Bogachev} section 6.4 for a general treatment from this perspective. We shall reproduce a proof following Glimm and Jaffe \cite{GJ} section 9.3 for reader's convenience and adaptation to the current situation.

A principal corollary of the results of this section is the following.
\begin{corr}\label{corr-rad-niko-dense}
  Let~$\Sigma$ be the disjoint union of Riemannian circles, embedded in an ambient Riemannian surface~$M$ (with or without boundary). Let~$\mu_{\mm{DN}}^{\Sigma,M}$ and~$\mu_{2\mn{D}}^{\Sigma}$ be the two Gaussian measures constructed on~$\mathcal{D}'(\Sigma)$ with covariance operators~$(\DN_M^{\Sigma})^{-1}$ and~$(2\mn{D}_{\Sigma})^{-1}$ (if~$M$ has boundary, specify the boundary condition to be $B$ as in section \ref{sec-dn-map}). Then~$\mu_{\mm{DN}}^{\Sigma,M}$ and~$\mu_{2\mn{D}}^{\Sigma}$ are mutually absolutely continuous with Radon-Nikodym density given by
  \begin{equation}
    \frac{\dd\mu_{\mm{DN}}^{\Sigma,M}}{\dd\mu_{2\mn{D}}^{\Sigma}}(\varphi)=({\tts \det_{\zeta}} (2\mn{D}_{\Sigma}))^{-\frac{1}{2}}({\tts \det_{\zeta}} \DN_M^{\Sigma})^{\frac{1}{2}} \me^{-\frac{1}{2}\sank{\varphi,(\DN_M^{\Sigma}-2\mn{D}_{\Sigma})\varphi}_{L^2(\Sigma)}}.
    \label{}
  \end{equation}
\end{corr}

The proof is at the end of this section. First we come back to the general case.

\begin{prop}\label{prop-quadratic-perturbation}
  Let~$V:L^2(\Sigma)\lto L^2(\Sigma)$ be as above and moreover assume
  \begin{equation}
    C^{-1}+V\textrm{ is positive,}
    \label{eqn-quadratic-perturbation-cov}
  \end{equation}
  and that
  \begin{equation}
    \wh{V}\defeq C^{\frac{1}{2}}V C^{\frac{1}{2}} \textrm{ is trace class.}
    \label{eqn-quadratic-perturbation-trace}
  \end{equation}
  Then
  \begin{enumerate}[(i)]
  \item the random variable~$\ank{\varphi,V\varphi}_{L^2}$ can be defined in $L^1(\mu_C)$ and $\mb{E}_C[\ank{\varphi,V\varphi}_{L^2}]=\ttr(\wh{V})$,
  \item $Z:= \mb{E}_C[\me^{-\frac{1}{2}\ank{\varphi,V\varphi}_{L^2}}] ={\tts \det_{\mm{Fr}}}(\one+ \wh{V})^{-\frac{1}{2}}$, and
  \item the Gibbs measure (\ref{eqn-quadratic-gibbs-meas}) is Gaussian with covariance~$(C^{-1}+V)^{-1}$.
\end{enumerate}
\end{prop}

Note that since~$C^{-1}+V$ is positive and~$C$ is also positive,~$\one + \wh{V}=C^{\frac{1}{2}}(C^{-1}+V)C^{\frac{1}{2}}$ is positive.

\begin{lemm}\label{lemm-quadratic-pert-trace}
  There exist an orthonormal basis~$\{f_j\}_{j=1}^{\infty}$ of the Gaussian Hilbert space~$W^{-s}(\Sigma)$ of~$\mu_C$ equipped with~$\ank{-,C-}_{L^2}$, such that
  \begin{equation}
    \ank{\varphi,V \varphi}_{L^2}=\sum_{j=1}^{\infty} \lambda_j \varphi(f_j)^2,
    \label{eqn-quadratic-expression}
  \end{equation}
  for all~$\varphi$ belonging to the Cameron-Martin space~$W^s(\Sigma)$, where~$\{\lambda_j\}$ are the eigenvalues of~$\wh{V}$ on~$L^2(\Sigma)$, and the series converges absolutely in~$L^1(\mu_C)$. Thus we \textsf{define} the random variable~$\ank{\varphi,V\varphi}_{L^2}$ with this converging series. Consequently, (i) of proposition \ref{prop-quadratic-perturbation} holds.
\end{lemm}

\begin{proof}
  The key is to seek~$\{C^{\frac{1}{2}}f_j\}$ as complete~$L^2$-orthonormal eigenfunctions of~$\wh{V}$ with eigenvalues~$\{\lambda_j\}$, which exist since~$\wh{V}$ is self-adjoint and trace class on~$L^2(\Sigma)$. Now~$\{f_j\}\subset W^{-s}(\Sigma)$ and is complete orthonormal with respect to~$\ank{-,C-}_{L^2}$. This also means the random variables~$\{\varphi(f_j)\}$ are mutually independent. Note for~$\varphi\in W^s(\Sigma)$,~$C^{-\frac{1}{2}}\varphi \in L^2(\Sigma)$, and
  \begin{align*}
    \bank{\varphi,V \varphi}_{L^2}&=\bank{C^{-\frac{1}{2}}\varphi,\wh{V} C^{-\frac{1}{2}}\varphi}_{L^2} \\
    &=\sum_{j=1}^{\infty} \lambda_j \big|\bank{C^{-\frac{1}{2}}\varphi,C^{\frac{1}{2}}f_j}_{L^2}\big|^2 \\
    &=\sum_{j=1}^{\infty} \lambda_j \varphi(f_j)^2.
  \end{align*}
  Since~$\wh{V}$ is trace class,
  \begin{equation}
    \sum_{j=1}^{\infty} \mb{E}_C\big[ |\lambda_j| \varphi(f_j)^2\big]=\sum_{j=1}^{\infty}|\lambda_j|<\infty,
    \label{}
  \end{equation}
  and we obtain the result.
\end{proof}

\begin{def7}
  We thus defined~$\ank{\varphi,V\varphi}_{L^2}$ as a \textsf{Wiener quadratic form} and it lies in the second Wiener chaos of~$\mu_C$, a fortiori in~$L^2(\mu_C)$. See Bogachev \cite{Bogachev} pages 257-261 for more information and in particular proposition 5.10.16 for the same result in the context of Malliavin calculus. In fact,
  \begin{equation}
    VCf_j=C^{-\frac{1}{2}} \wh{V} C^{\frac{1}{2}}f_j=\lambda_j f_j,
    \label{}
  \end{equation}
  and hence~$V$ is trace class on~$W^s(\Sigma)$ with eigenbasis~$\{Cf_j\}$. 
\end{def7}

\begin{def7}
  We have~$\lambda_j>-1$ by the assumption (\ref{eqn-quadratic-perturbation-cov}).
\end{def7}

To treat (ii) and (iii) of proposition \ref{prop-quadratic-perturbation} we adopt some approximations.

\begin{lemm}
  Proposition \ref{prop-quadratic-perturbation} (ii) and (iii) is true in the case~$V$ has finite rank.
\end{lemm}

\begin{proof}
  In this case the series in (\ref{eqn-quadratic-expression}) is finite with~$\{f_j\}_{j=1}^N$ for some~$N\in\mb{N}$. Thus
  \begin{align*}
    Z&=\frac{1}{(2\pi)^{N/2}}\int_{\mb{R}^N}^{}\me^{-\frac{1}{2}\sum_{j} \lambda_j x_j^2} \me^{-\frac{1}{2}\sum_j x_j^2}\dd^N x \\
    &=\prod_j (1+\lambda_j)^{-\frac{1}{2}}=\det(\one +\wh{V})^{-\frac{1}{2}},
  \end{align*}
  by projecting onto~$\mb{R}^N$ via~$\varphi\mapsto (\varphi(f_1),\dots,\varphi(f_N))=:(x_1,\dots,x_N)$. 
   For the covariance, we orthogonally decompose~$W^{-s}(\Sigma)$ as
  \begin{equation}
    W^{-s}(\Sigma)=\spn\{f_j~|~ 1\le j\le N\}\oplus \spn\{f_j~|~ 1\le j\le N\}^{\perp},
    \label{eqn-quadratic-orth-gauss}
  \end{equation}
  let~$\Pi_0$ and~$\Pi_1$ be the corresponding orthogonal projections (in order), and for any~$f\in C^{\infty}(\Sigma)$ write
  \begin{equation}
    \varphi(f)=\alpha_1\varphi(f_1)+\cdots +\alpha_N\varphi(f_N)+\varphi(\Pi_1 f),
  \end{equation}
  then~$\varphi(\Pi_1 f)$ is independent of both~$\varphi(f_j)$, $1\le j\le N$, and~$\ank{\varphi,V\varphi}_{L^2}$. It is now clear that (\ref{eqn-quadratic-gibbs-meas}) is Gaussian because we could now express~$\mu(\{\phi(f)\in A\})$ for any Borel set $A\subset \mb{R}$ as a Gaussian integral over~$\mb{R}^{N+1}$. Apply $C^{\frac{1}{2}}$ to (\ref{eqn-quadratic-orth-gauss}) we get the~$L^2$-orthogonal decomposition
  \begin{equation}
    L^2(\Sigma)=\underbrace{C^{\frac{1}{2}}(\spn\{f_j~|~ 1\le j\le N\})}_{=:L^2(\Sigma)_0}\oplus \underbrace{C^{\frac{1}{2}}(\spn\{f_j~|~ 1\le j\le N\}^{\perp})}_{=:L^2(\Sigma)_1}.
    \label{}
  \end{equation}
  Clearly~$\wh{V}$ leaves this decomposition invariant and is zero on~$L^2(\Sigma)_1$. Hence~$(\one+\wh{V})^{-1}$ is block-diagonal,
  \begin{equation}
    (\one+\wh{V})^{-1}=
    \def\arraystretch{1.3}
    \begin{blockarray}{ccl}
      \{C^{\frac{1}{2}}f_j\}_1^N & L^2(\Sigma)_1& \\
      \begin{block}{(cc)c}
	\diag_{j=1}^N\{(\lambda_j+1)^{-1}\}&0& \{C^{\frac{1}{2}}f_j\}_1^N \\
	0&\one &L^2(\Sigma)_1\\ \end{block}\end{blockarray}
    \label{}
  \end{equation}
   Thus
  \begin{align*}
    Z^{-1}\mb{E}_C[\varphi(f)^2 \me^{-\frac{1}{2}\ank{\varphi,V \varphi}}]&=Z^{-1}\mb{E}_C[\varphi(\Pi_0 f)^2 \me^{-\frac{1}{2}\sum_j \lambda_j \varphi(f_j)^2}] +Z^{-1}\mb{E}_C[\varphi(\Pi_1 f)^2]\mb{E}_C[\me^{-\frac{1}{2}\ank{\varphi,V\varphi}}] \\
    &=\sum_j \alpha_j^2(\lambda_j+1)^{-1}+\bank{\Pi_1 f,C \Pi_1 f}_{L^2} \\
    &=\bank{C^{\frac{1}{2}}\Pi_0 f, (\one +\wh{V})^{-1}C^{\frac{1}{2}}\Pi_0 f}_{L^2}+\bank{\Pi_1 f,C^{\frac{1}{2}}(\one +\wh{V})^{-1}C^{\frac{1}{2}} \Pi_1 f}_{L^2} \\
        &=\ank{f,(C^{-1}+V)^{-1}f}_{L^2},
  \end{align*}
  where we perform again a Gaussian integral on~$\mb{R}^N$ in the second line.
\end{proof}

\begin{proof}[Proof of proposition \ref{prop-quadratic-perturbation} (ii) and (iii).]
For general~$V$ we impose spectral cut-off at~$N$,
\begin{equation}
  V_N(\varphi)\defeq \sum_{j=1}^N \lambda_j \ank{\varphi,f_j}_{L^2}f_j.
  \label{eqn-quadratic-V-spectral-cutoff}
\end{equation}
Then $C^{\frac{1}{2}}V_N C^{\frac{1}{2}}=:\wh{V}_N \to \wh{V}$ under the trace norm~$\nrm{\cdot}_{\ttr}$ (acting on~$L^2(\Sigma)$). Hence~$\ank{\varphi,V_N \varphi}_{L^2}\to \ank{\varphi,V\varphi}_{L^2}$ in~$L^1(\mu_C)$ by lemma \ref{lemm-quadratic-pert-trace} and after passing to a subsequence
\begin{equation}
  \me^{-\frac{1}{2}\ank{\varphi,V_N \varphi}_{L^2}}\to \me^{-\frac{1}{2}\ank{\varphi,V\varphi}_{L^2}}
  \label{}
\end{equation}
in~$L^1(\mu_C)$ and (ii) for $V$ follows, since $\detf(\one+\wh{V}_N)\to \detf(\one+\wh{V})$ by lemma \ref{lemm-det-fred-cont}.
  To prove (iii), we note that~$\wh{V_N}\to \wh{V}$ a fortiori under the operator norm, then~$C^{\frac{1}{2}}(\one+\wh{V}_N)^{-1}C^{\frac{1}{2}}\to C^{\frac{1}{2}}(\one +\wh{V})^{-1}C^{\frac{1}{2}}$ in norm and
  \begin{align*}
    \frac{\mb{E}_C[\me^{\ii\varphi(f)}\me^{-\frac{1}{2}\ank{\varphi,V\varphi}_{L^2}}]}{\mb{E}_C[\me^{-\frac{1}{2}\ank{\varphi,V\varphi}_{L^2}}]}&=\lim_{N\to\infty} \frac{\mb{E}_C[\me^{\ii\varphi(f)}\me^{-\frac{1}{2}\ank{\varphi,V_N\varphi}_{L^2}}]}{\mb{E}_C[\me^{-\frac{1}{2}\ank{\varphi,V_N\varphi}_{L^2}}]}\\
    &=\lim_{N\to\infty} \exp\Big( -\frac{1}{2}\ank{f,(C^{-1}+V_N)^{-1}f}_{L^2} \Big) \\
    &=\exp\Big( -\frac{1}{2}\ank{f,(C^{-1}+V)^{-1}f}_{L^2} \Big),
  \end{align*}
  for~$f\in C^{\infty}(\Sigma)$, showing that (\ref{eqn-quadratic-gibbs-meas}) is Gaussian with the right covariance.
\end{proof}

\begin{proof}[Proof of corollary \ref{corr-rad-niko-dense}.]
Remember now that~$\dim \Sigma=1$.
  Write for short~$\mn{D}:=2\mn{D}_{\Sigma}$ and~$\DN:=\DN_M^{\Sigma}$. Setting~$C=\mn{D}^{-1}$ and~$V=\DN-\mn{D}$ in proposition \ref{prop-quadratic-perturbation}, we are left to prove the determinant identity
  \begin{equation}
    \detz(\mn{D})\detf(\one+\mn{D}^{-\frac{1}{2}}(\DN-\mn{D})\mn{D}^{-\frac{1}{2}})=\detz(\DN).
    \label{}
  \end{equation}
  Indeed, this is now immediate as
  \begin{equation}
    \textrm{LHS}=\detz(\mn{D})\detf(\one+\mn{D}^{-1}(\DN-\mn{D}))=\textrm{RHS}
    \label{}
  \end{equation}
  by lemma \ref{lemm-det-fred-cont} (ii) and lemma \ref{lemm-det-factor}. We point out~$\mn{D}\mn{D}^{-1}(\DN-\mn{D})=V$ can be approximated in~$\nrm{\cdot}_{\mm{tr}}$ by smoothing operators (by lemma \ref{lemm-dn-prop} (iv),~$V=\DN-\mn{D}$ is~$L^2$-trace class), so the conditions of lemma \ref{lemm-det-factor} is satisfied. Indeed, as above,~$\{Cf_j\}\subset W^1(\Sigma)$ are $W^1$-complete eigenfunctions of~$V$. If the corresponding eigenvalue~$\lambda_j\ne 0$, then the bootstrap argument shows~$Cf_j\in C^{\infty}(\Sigma)$ since~$V\in \Psi^{<0}(\Sigma)$. Thus~$V_N$, where~$V_N$ is as in (\ref{eqn-quadratic-V-spectral-cutoff}), approximates~$V$ in~$\nrm{\cdot}_{\mm{tr}}$ and is smoothing. This concludes the proof.
\end{proof}

\begin{def7}
  Though it is probably true, we do not claim~$\DN-\mn{D}$ is elliptic.
\end{def7}

  \section{Variants of Nelson's Argument}\label{sec-nelson-main}

 \noindent The goal of this section is to define part~$A$ of (\ref{eqn-def-mes-gibbs-heu}) which culminates in \hyperref[thrm-nelson]{Nelson's theorem}. The principal obstacle in achieving this is the fact that powers of a distribution such as~$\phi^2$,~$\phi^4$ etc., are generally not defined. Even as a random variable under~$\mu_{\mm{GFF}}^M$, we have (formally)~$\mb{E}_{\mm{GFF}}^M[\phi(x)\phi(y)]\heueq G_{(\Delta+m^2)}(x,y)$ from (\ref{eqn-gff-cov-closed}) for~$x\ne y$ but this implies~$\mb{E}_{\mm{GFF}}^M[\phi(x)\phi(x)]= G_{(\Delta+m^2)}(x,x)=\infty$. This necessitates a procedure of \textsf{renormalization} which \textit{subtracts away}~$\infty$ and makes~$\mb{E}_{\mm{GFF}}^M[\phi(x)\phi(x)]<\infty$.

Here the natural renormalization strategy is provided by the Gaussian probability theory. Let~$\{K_{\varepsilon}~|~\varepsilon>0\}$ be a family of \textit{smoothing} operators\footnote{that is, each $K_{\varepsilon}$ maps $\mathcal{D}'(M)\lto C^{\infty}(M)$.} on~$M$, for which~$K_{\varepsilon}\to \one$ as~$\varepsilon\to 0$ (in a sense to be specified later), and consider the mollified random field~$\phi_{\varepsilon}:=K_{\varepsilon}\phi$. Instead of~$\phi_{\varepsilon}(x)^4$, say, we look at
\begin{align*}
  {:}\phi_{\varepsilon}(x)^4{:}&\defeq \phi_{\varepsilon}(x)^4-6\mb{E}[\phi_{\varepsilon}(x)^2] \phi_{\varepsilon}(x)^2+ 3\mb{E}[\phi_{\varepsilon}(x)^2]^2 \\
  &\defeq \phi_{\varepsilon}(x)^4-6C_{\varepsilon}(x) \phi_{\varepsilon}(x)^2+ 3C_{\varepsilon}(x)^2.
\end{align*}
It happens that for any~$\chi\in C^{\infty}(M)$, the integral~$\int_{M}^{}\chi(x){:}\phi_{\varepsilon}(x)^4{:}\dd V_M(x)$ converges as a random variable in~$L^2(\mu_{\mm{GFF}}^M)$ to a definitive limit, and defines~$\int_{M}^{}\chi(x){:}\phi(x)^4{:}\dd V_M(x)$ as a random variable in~$L^2(\mu_{\mm{GFF}}^M)$. Note that~$C_{\varepsilon}(x)\to \infty$ as~$\varepsilon\to 0$, so we have subtracted ``infinities''.

\begin{def7}\label{rem-locality-crucial}
\textbf{The crucial point} in our adaptation of Nelson's argument is to realize the \textit{locality} of the interaction~$\int_{M}^{}\chi(x){:}P(\phi(x)){:}\dd V_M(x)$ (see section \ref{sec-locality}). To this end we must allow a sufficiently large class of \textit{regulators}~$K_{\varepsilon}$ (in particular, \textit{local} ones) and show that they define the same interaction (proposition \ref{prop-nelson-main}). In addition, the \textit{Wick ordering} also needs to be local (see section \ref{sec-change-wick}), so that the interaction on a domain with boundary could be defined without reference to the ambient closed manifold where this domain ``caps''. See Brunetti, Fredenhagen, Verch \cite{BFV} and Guo, Paycha, Zhang \cite{GPZ} for more information and perspective on locality.
\end{def7}

\begin{def7}
The method adopted here is restricted to dimension two. In three dimensions, the target measure bearing the heuristic form (\ref{eqn-def-mes-gibbs-heu}) becomes mutually singular with respect to~$\mu_{\mm{GFF}}^M$ and hence cannot be expressed as an integrable function multiplied by~$\mu_{\mm{GFF}}^M$. A recent phenomenal method to treat this case is developed in the framework of stochastic PDEs, called \textsf{stochastic quantization}, providing an alternative to older results outlined in \cite{GJ} section 23.1. See the introductions in \cite{GH}, \cite{HS}, \cite{AK}, \cite{MWX} and \cite{BDFT2} for reviews of the literature and pedagogical discussions.
\end{def7}

\subsection{Regularizations}

\noindent In this subsection we describe an admissible class of regulators~$K_{\varepsilon}$ which would eventually produce the same random variable~$\int_{M}^{}\chi(x){:}\phi(x)^4{:}\dd V_M(x)$ as will be proved in the next subsection. Basically, they are smoothing operators such that~$K_{\varepsilon}\to\one$ in~$\Psi^{\delta}(M)$ in the \textit{symbol sense} for any~$\delta>0$ (see definition below). A compact notation is to say~$K_{\varepsilon}\to \one $ in~$\Psi^{0+}(M)$.

\begin{deef}
  Let~$r\in \mb{R}$ we say that operators~$K_{\varepsilon}\to K$ in~$\Psi^r(M)$ \textsf{in the symbol sense} if for any coordinate chart~$\kappa: U\lto \mb{R}^d$ and cut-off~$\chi\in C_c^{\infty}(U)$, the full symbol of~$\chi K_{\varepsilon}\chi$ (considered acting on~$C_c^{\infty}(\kappa(U))$), converges to that of~$\chi K\chi$ in the~$\mathcal{S}^r_{1,0}(\kappa(U)\times \mb{R}^d)$ topology.
\end{deef}

Now we describe the first candidate for~$K_{\varepsilon}$ satisfying the above assumption (the proof is in appendix \ref{app-symbol}). This was introduced in Dyatlov and Zworski \cite{DZ} and has the advantage of being \textsf{local}, realizing the locality of the~$P(\phi)$ interaction eventually in section \ref{sec-locality}. Consider~$\psi\in C_c^{\infty}((-1,1))$ with~$0\le \psi\le 1$ and equal to~$1$ near~$0$. For~$\varepsilon>0$ we define the operator
\begin{equation}
  E_{\varepsilon}u(x)\defeq \int_{M}^{} E_{\varepsilon}(x,y)u(y)\dd V_g(y),\quad \textrm{with }E_{\varepsilon}(x,y)=\frac{1}{F_{\varepsilon}(x)}\psi\left( \frac{d_g(x,y)}{\varepsilon} \right).
  \label{eqn-dyatlov-zwors-reg}
\end{equation}
Here~$F_{\varepsilon}(x)=\int_{}^{}\psi(d(x,y)/\varepsilon) \dd y$ so that~$\int_{}^{}E_{\varepsilon}(x,y)\dd y=1$, and~$d_g$ denotes the Riemannian distance. One observes that~$E_{\varepsilon}(x,y)$ is smooth for each~$\varepsilon>0$ so~$E_{\varepsilon}:\mathcal{D}'(M)\lto C^{\infty}(M)$. Observe also that~$\varepsilon^d/C\le F_{\varepsilon}(x)\le C\varepsilon^d$ for some~$C>0$, and this~$C$ could be made dependent neither on~$\varepsilon$ nor on~$x$ as~$M$ is compact.

\begin{lemm}\label{lemm-dyat-zwor-symbol-conv}
  For any~$\delta>0$ we have~$E_{\varepsilon}\to \one$ in~$\Psi^{\delta}(M)$ in the symbol sense.
\end{lemm}

\begin{proof}
    See appendix \ref{app-symbol}.
\end{proof}

Note that~$K^*_{\varepsilon}(x,y)=K_{\varepsilon}(y,x)$ for real smoothing operators (and their symbols are related in a simple manner),~$K_{\varepsilon}\to\one$ in the symbol sense is equivalent to~$K_{\varepsilon}^*\to \one$ in the symbol sense.

\begin{lemm}
  \label{lemm-conv-symbol-top}
  Let~$\{E'_{\varepsilon'}\}_{\varepsilon'>0}$ be another family of smoothing operators such that~$E'_{\varepsilon'}\to \one$ in~$\Psi^{\delta'}(M)$ in the symbol sense for any~$\delta'>0$. Then the net~$E_{\varepsilon}^*(\Delta+m^2)^{-1}(E'_{\varepsilon'}-E_{\varepsilon})$,~$(\varepsilon',\varepsilon)\in\mb{R}_+\times\mb{R}_+$ (we say~$(\varepsilon',\varepsilon)\prec (\varepsilon'_1,\varepsilon_1)$ iff~$\varepsilon'>\varepsilon'_1$ and~$\varepsilon>\varepsilon_1$), converges to zero in~$\Psi^{-2+\delta}(M)$ in the symbol sense for any~$\delta>0$.
\end{lemm}

\begin{proof}
  Note that following essentially the same arguments as above the~$\Psi^{\delta/2}(M)$ seminorms of~$E_{\varepsilon}$ can be bounded uniformly in~$\varepsilon$. This said, the result follows essentially from the continuity of the twisted product (composition product) of symbols as a map~$\mathcal{S}^r_{1,0}\times \mathcal{S}^{r'}_{1,0}\lto \mathcal{S}^{r+r'}_{1,0}$ with respect to the symbol topologies (see Folland \cite{Folland2} page 105 theorem 2.47).
\end{proof}

We shall consider another set of seminorms on~$\Psi^r(M)$ in the case~$-d<r<0$ which suits better our purposes. They are defined as follows. Let~$\mathcal{M}\subset C^{\infty}(M\times M,T(M\times M))$ denote the~$C^{\infty}(M\times M)$-module of smooth vector fields tangent to the diagonal in~$M\times M$. We fix a finite coordinate cover~$\{U_i\}_{i\in 1}^N$ of~$M$, with charts~$\kappa_i:U_i\lto \mb{R}^d$, and a partition of unity~$\{\chi_i\}$ subordinate to this cover.

\begin{deef}\label{def-kernel-diag-top}
  For any~$K\in C^{\infty}(M\times M\setminus \mm{diag})$,~$1\le i\le N$ and~$L_1$, \dots,~$L_p\in \mathcal{M}$, we define the seminorms
  \begin{equation}
    p_{i,L_1,\dots,L_p}(K)\defeq \sup_{(x,y)\in U_i\times U_i} \big|(\kappa_i\times \kappa_i)_*\left( (\chi_i\otimes \chi_i)L_1\cdots L_p K \right)(x,y)\big|\cdot d_g(x,y)^{d+r},
    \label{}
  \end{equation}
  while on~$U_i\times U_j$,~$i\ne j$, which does not touch the diagonal, we use the~$C^{\infty}(U_i\times U_j)$ seminorms. By the \textsf{kernel topology} on~$\Psi^r(M)$,~$-d<r<0$, we mean the topology induced by these seminorms on the Schwartz kernels~$K_A$ of~$A\in \Psi^r(M)$. Here~$d_g$ is the distance function.
\end{deef}

\begin{prop}\label{prop-kernel-top-equiv-symbol}
  In the case~$-d<r<0$, the above kernel topology is equivalent to the topology induced by symbols~$\mathcal{S}^r_{1,0}(T^*M)$ on~$\Psi^r(M)$. In particular, if~$A_{\varepsilon}\to A$ as~$\varepsilon\to 0$ in~$\Psi^r(M)$ in the symbol sense then~$A_{\varepsilon}\to A$ also in the above kernel topology.
\end{prop}

\begin{proof}
    Essentially in Taylor \cite{Taylor2} page 6, proposition 2.2, page 7, proposition 2.4 and page 10, proposition 2.7. See also Bailleul, Dang, Ferdinand and Tô \cite{BDFT} proposition 6.9 for a more detailed treatment.
\end{proof}

Finally, we observe that the heat operator~$K_{\varepsilon}=\me^{-\varepsilon(\Delta+m^2)}$ is also a valid candidate:

\begin{lemm}
  [\cite{Dang} lemma 4.15] We have~$\me^{-\varepsilon(\Delta+m^2)}\to \one$ in~$\Psi^{\delta}(M)$ in the symbol sense for any~$\delta>0$. \hfill~$\Box$
\end{lemm}

Some properties of the heat operator is summed up in appendix \ref{app-symbol}.

\subsection{Integrability of Interaction and Regularization Independence}\label{sec-nelson-reg-indep}

\begin{prop}
  \label{prop-nelson-main} Let~$(K_{\varepsilon})_{\varepsilon>0}$ be any family of real smoothing operators such that~$K_{\varepsilon}\to \one$ in~$\Psi^{\delta}(M)$ in the symbol sense for any~$\delta>0$. Define~$\phi_{\varepsilon}(x)$ and~${:}P(\phi_{\varepsilon}(x)){:}$ as above and let $\chi\in C_c^\infty(M)$ be a test function. Put
  \begin{equation}
    S_{M,\varepsilon,\chi}(\phi):=\int_{M}^{} \chi(x) {:}P(\phi_{\varepsilon}(x)){:} \dd V_M.
    \label{}
  \end{equation}
  This is a random variable on~$\mathcal{D}'(M)$ equipped with~$\mu_{\mm{GFF}}^M$. Then~$\{S_{M,\chi,\varepsilon}\}$ converges in~$L^2(\mu_{\mm{GFF}}^M)$ as~$\varepsilon\to 0$, and the limit is independent of the specific smoothing~$(K_{\varepsilon})$ chosen, provided they have the convergence property described above.
  
  More precisely, for any other smoothing family~$(\tilde{K}_{\varepsilon'})$ satisfying the same condition and defining the random variable~$\tilde{S}_{M,\varepsilon',\chi}$, we have a quantitative bound of the form
\begin{eqnarray*}
  \mathbb{E}\big[\big| S_{M,\varepsilon,\chi}-\tilde{S}_{M,\varepsilon^\prime,\chi}\big|^2\big]\le \fk{O}(\vert \varepsilon-\varepsilon^\prime\vert) C_n  \Vert \chi\Vert_{L^4}^2,
\end{eqnarray*}  
where~$\fk{O}(|\varepsilon-\varepsilon'|)$ is a function going to zero as~$|\varepsilon-\varepsilon'|\to 0$, depending only on~$M$.
\end{prop}

A particularity of the dimension two is seen in the following elementary lemma.

\begin{lemm}\label{lemm-2d-log-blow-up-can}
  If~$\dim M=2$ and~$g$ is a smooth Riemannian metric on~$M$ then
  \begin{equation}
    \iint_{M\times M}|\log(d_g(x,y))|^p \dd V_{M\times M}\le C_p<\infty
    \label{}
  \end{equation}
  for any~$1\le p<\infty$ and
  \begin{equation}
    \iint_{M\times M} d_g(x,y)^{-\delta}\dd V_{M\times M}\le C_{\delta}<\infty
    \label{}
  \end{equation}
  for any~$0<\delta<1$,~$d_g$ denoting the distance function.\hfill~$\Box$
\end{lemm}

Consequently, since in two dimensions~$G(x,y)=\mathcal{O}(\log(d(x,y)))$ as~$y\to x$ (see lemma \ref{lemm-tadpole-mass}), a simple argument with partition of unity shows
\begin{equation}
  \iint_{M\times M} \chi(x)\chi(y) |G(x,y)|^p \dd V_{M\times M}\le C_p \Vert \chi\Vert_{L^\infty}^2\vol(\text{supp}(\chi))^2 <\infty
  \label{eqn-upper-bound-green-p}
\end{equation}
for some~$C_p>0$, for all~$1\le p<\infty$. When~$p=2$ this is the familiar fact that in two dimensions the Green operator~$(\Delta+m^2)^{-1}$ is Hilbert-Schmidt, which may also be shown using Weyl's law.

We will denote
\begin{align}
  \mb{E}_{\mm{GFF}}^M\big[\phi_{\varepsilon}(x)\phi_{\varepsilon}(y)\big]&=\bank{\delta_x,K_{\varepsilon}^*(\Delta+m^2)^{-1}K_{\varepsilon}\delta_y}_{L^2}\defeq G_{\varepsilon,\varepsilon}(x,y),\\
  \mb{E}_{\mm{GFF}}^M\big[\phi_{\varepsilon}(x)\phi_{\varepsilon'}(y)\big]&=\bank{\delta_x,K_{\varepsilon}^*(\Delta+m^2)^{-1}\tilde{K}_{\varepsilon'}\delta_y}_{L^2}\defeq G_{\varepsilon,\varepsilon'}(x,y), \\
  \mb{E}_{\mm{GFF}}^M\big[\phi_{\varepsilon'}(x)\phi_{\varepsilon'}(y)\big]&=\bank{\delta_x,\Tilde{K}_{\varepsilon'}^*(\Delta+m^2)^{-1}\tilde{K}_{\varepsilon'}\delta_y}_{L^2}\defeq G_{\varepsilon',\varepsilon'}(x,y).
  \label{}
\end{align}
Note that~$G_{\varepsilon,\varepsilon'}$ has a \textit{different} regulator on the second variable! For fixed~$\varepsilon>0$, we have by lemma \ref{lemm-wick-feyn},
\begin{equation}
    {:}\phi_{\varepsilon}(x)^{2n}{:}=G_{\varepsilon,\varepsilon}( x,x)^n h_{2n}\big(\phi_{\varepsilon}(x)\big/G_{\varepsilon,\varepsilon}(x,x)^{\frac{1}{2}} \big)\ge -b_1 G_{\varepsilon,\varepsilon}(x,x)^n,
    \label{eqn-lower-bound-hermite}
  \end{equation}
  for some constant~$b_1>0$ independent of~$\varepsilon$, since the even-degree Hermite polynomial~$h_{2n}$ is bounded below.

\begin{proof}[Proof of proposition \ref{prop-nelson-main}.]
   We shall prove the proposition for the case~$P(\theta)=\theta^{2n}$,~$n\in\mb{N}$, the general case is similar. For fixed~$\varepsilon$,~$\varepsilon'>0$, we compute
  \begin{align*}
    \nrm{S_{M,\varepsilon}-S_{M,\varepsilon'}}_{L^2(\mu_{\mm{GFF}})}^2&=\mb{E}\bigg[ \left| \int_{M}^{} \chi {:}\phi_{\varepsilon}(x)^{2n}{:}\dd V_M -\int_{M}^{} \chi  {:}\phi_{\varepsilon'}(x)^{2n}{:} \dd V_M \right|^2 \bigg] \\
    &=\mb{E}\left[  \iint_{M\times M}^{}{:}\chi(x) \phi_{\varepsilon}(x)^{2n}{:}~{:}\phi_{\varepsilon}(y)^{2n}{:}\chi(y)\dd V_M\otimes \dd V_M \right] \\
    &\quad \quad-2\mb{E}\left[  \iint_{M\times M}^{}\chi(x)  {:}\phi_{\varepsilon}(x)^{2n}{:}~{:}\phi_{\varepsilon'}(y)^{2n}{:} \chi(y) \dd V_M\otimes \dd V_M \right] \\
    &\quad\quad +\mb{E}\left[  \iint_{M\times M}^{}  \chi(x){:}\phi_{\varepsilon'}(x)^{2n}{:}~{:}\phi_{\varepsilon'}(y)^{2n}{:}\chi(y)  \dd V_M\otimes \dd V_M \right] \\
    &=\iint_{M\times M}^{}\Big( \mb{E}\left[\chi(x) {:}\phi_{\varepsilon}(x)^{2n}{:}~{:}\phi_{\varepsilon}(y)^{2n}{:}\chi(y) \right] -2  \mb{E}\left[\chi(x) {:}\phi_{\varepsilon}(x)^{2n}{:}~{:}\phi_{\varepsilon'}(y)^{2n}{:} \chi(y)\right] \\
    &\quad\quad+\mb{E}\left[\chi(x) {:}\phi_{\varepsilon'}(x)^{2n}{:}~{:}\phi_{\varepsilon'}(y)^{2n}{:} \chi(y)\right] \Big)
    \dd V_{M\times M} \tag{Tonelli}\\
    &=(2n)!\iint_{M\times M}\chi(x)\chi(y) \left( G_{\varepsilon,\varepsilon}(x,y)^{2n}-2G_{\varepsilon,\varepsilon'}(x,y)^{2n}+G_{\varepsilon',\varepsilon'}(x,y)^{2n} \right) \dd V_{M\times M} \tag{lemma \ref{lemm-wick-feyn}}
  \end{align*}
  We will control the integral~$\iint_{M\times M}\chi(x)\chi(y)\left( G_{\varepsilon,\varepsilon}(x,y)^{2n}-G_{\varepsilon,\varepsilon'}(x,y)^{2n} \right)\dd V_{M\times M}$ for $\chi\in C^\infty(M)$ and $\chi\geqslant 0$. 
  Indeed,
\begin{align*}
  |\textrm{this integral}|&\le \iint_{M\times M}\chi(x)\chi(y) \left|G_{\varepsilon,\varepsilon}(x,y)-G_{\varepsilon,\varepsilon'}(x,y)\right| \cdot \big|G_{\varepsilon,\varepsilon}(x,y)^{2n-1} \\
  &\quad\quad +G_{\varepsilon,\varepsilon}(x,y)^{2n-2}G_{\varepsilon,\varepsilon'}(x,y)+\dots+G_{\varepsilon,\varepsilon'}(x,y)^{2n-1} \big| \dd V_{M\times M} \\
  &\le 2n  \Vert\chi\Vert^2_{L^4} \left( C_{8n-4} \right)^{\frac{1}{4}}
  \times \left( \iint_{M\times M} \left|G_{\varepsilon,\varepsilon}(x,y)-G_{\varepsilon,\varepsilon'}(x,y)\right|^2 \right)^{\frac{1}{2}} 
\end{align*}
Remember that~$G_{\varepsilon,\varepsilon'}(x,y)$ is the kernel of~$K^*_{\varepsilon}(\Delta+m^2)^{-1}\tilde{K}_{\varepsilon'}$ and~$G_{\varepsilon,\varepsilon'}(x,y)-G_{\varepsilon,\varepsilon}(x,y)$ is the kernel of~$K^*_{\varepsilon}(\Delta+m^2)^{-1}(\tilde{K}_{\varepsilon'}-K_{\varepsilon})$. By lemma \ref{lemm-conv-symbol-top}, definition \ref{def-kernel-diag-top} and proposition \ref{prop-kernel-top-equiv-symbol},
\begin{equation}
  |G_{\varepsilon,\varepsilon'}(x,y)|\le C_{M,\delta}d_g(x,y)^{-\delta},
  \label{}
\end{equation}
uniformly in~$(\varepsilon,\varepsilon')$ and by lemma \ref{lemm-conv-symbol-top},
\begin{equation}
  |G_{\varepsilon,\varepsilon'}(x,y)-G_{\varepsilon,\varepsilon}(x,y)| \le \fk{O}_{M,\delta}(|\varepsilon-\varepsilon'|)d_g(x,y)^{-\delta}
  \label{}
\end{equation}
for any~$\delta>0$. If we restrict moreover to $\delta<1$, then we prove our result, thanks to lemma \ref{lemm-2d-log-blow-up-can}.
\end{proof}

\subsection{Integrability of the Exponential of Interaction}

\noindent In this subsection we will adopt the heat regulator~$K_{\varepsilon}=\tilde{K}_{\varepsilon}:=\me^{-\varepsilon(\Delta+m^2)}$.

\begin{lemm}\label{lemmB6}
  Let~$\deg P=2n$. Then~$S_{M,\chi,\varepsilon}$,~$\varepsilon>0$, and hence the resulting limit~$S_{M,
  \chi}$, is in the~$(2n)$-th Wiener chaos of the GFF, that is,~$\ol{\mathcal{P}}_{2n}(\mathcal{H})$, where~$\mathcal{H}=W^{-1}(M)$ is the Gaussian Hilbert space of the GFF.
\end{lemm}

\begin{proof}
  We shall show it for~$P(\theta)=\theta^4$ and the general case is similar. Here we use the spectral representation of remark \ref{rem-mass-gau-field-four} (using the notations thereof) and take~$K_{\varepsilon}:=\me^{-\varepsilon(\Delta+m^2)}$. Thus we can write
  \begin{align*}
    \phi_{\varepsilon}(x)^4&=\bigg( \sum_{j=0}^{\infty}\ank{\varphi_j, K_{\varepsilon}\delta_x}_{L^2} \xi_j\bigg)^4 =\bigg( \sum_{j=0}^{\infty} \me^{-\varepsilon(\lambda_j+m^2)}\varphi_j(x) \xi_j\bigg)^4 \\
    &=\sum_{j,k,\ell,p} \me^{-\varepsilon(\lambda_j+\lambda_k+\lambda_{\ell}+\lambda_p+4m^2)} \varphi_j(x)\varphi_k(x)\varphi_{\ell}(x)\varphi_p(x)\xi_j\xi_k\xi_{\ell}\xi_p,
  \end{align*}
  the series converging absolutely in~$L^2(\mu_{\mm{GFF}}^M)$. Now each individual term is clearly in~$\ol{\mathcal{P}}_4(\mathcal{H})$. Since
  \begin{equation}
    \int_{}^{}|\varphi_j(x)\varphi_k(x)\varphi_{\ell}(x)\varphi_p(x)|\dd V_M(x) \lesssim \vol(M)\cdot [\textrm{polynomial in }\lambda_j,\lambda_k,\lambda_{\ell},\lambda_p\textrm{ with fixed degree}]
    \label{}
  \end{equation}
  as one has~$\sup_M |\varphi_j|\lesssim (1+\lambda_j)^2$ in two dimensions which follows essentially from the Sobolev embedding (see Sogge \cite{Sogge2} page 43 equation (3.1.12)), thus
  \begin{equation}
    \int_{}^{}\chi(x)\phi_{\varepsilon}(x)^4 \dd V_M(x)=\sum_{j,k,\ell,p} \me^{-\varepsilon(\lambda_j+\lambda_k+\lambda_{\ell}+\lambda_p+4m^2)} C_{\chi,j,k,\ell,p}\xi_j\xi_k\xi_{\ell}\xi_p
    \label{}
  \end{equation}
  with the series converging absolutely in~$L^2(\mu_{\mm{GFF}}^M)$ and the result is in~$\ol{\mathcal{P}}_4(\mathcal{H})$.
\end{proof}

\begin{prop}
  [hypercontractivity, \cite{Janson} theorem 5.10, \cite{Sim2} theorem I.22] \label{propB7} Let~$\mathcal{H}\subset L^2(Q,\mathcal{O},\mb{P})$ be a Gaussian Hilbert space on some probability space~$(Q,\mathcal{O},\mb{P})$, and let~$n\ge 1$,~$2\le p<\infty$. Then 
  \begin{equation}
    \mb{E}[|X|^p]^{\frac{1}{p}}\le (p-1)^{\frac{n}{2}}\mb{E}[X^2]^{\frac{1}{2}}
    \label{}
  \end{equation}
  for all~$X\in\ol{\mathcal{P}}_n (\mathcal{H})$.  \hfill~$\Box$
\end{prop}

Combining lemma \ref{lemmB6} and proposition \ref{propB7}, we have

\begin{corr}\label{cor-nelson-all-Lp-cutoff-limit}
  The convergence of~$\{S_{M,\chi,\varepsilon}\}$, as well as the limit~$S_{M,\chi}$, is in~$L^p(\mu_{\mm{GFF}})$ for all~$1\le p<\infty$. Moreover, if $\chi\rightarrow 0$ in $L^4(M)$ then 
$$\lim_{\varepsilon\rightarrow 0}\lim_{\chi\rightarrow 0} S_{M,\varepsilon,\chi}=0  $$
as random variable in $L^p(\mu_{\mm{GFF}}^M)$ for all $1\le p<\infty$.
Moreover $\mathbb{E}[\me^{-S_{M,\varepsilon,\chi}}] $ remains uniformly bounded along the limit. In particular $S_{M,\chi}$ is defined for $\chi\in L^4(M)$. \hfill~$\Box$
\end{corr}

We single out a calculus computation which will be used in the sequel:

\begin{lemm}\label{lemmB9}
  Let~$a$,~$b$ be positive real numbers. Then the real function~$\alpha(x):=x^{(bx)} a^x$,~$x>0$, attains its minimum value~$\me^{-be^{-1}a^{-1/b}}$ at~$x=e^{-1}a^{-1/b}$. \hfill ~$\Box$
\end{lemm}

\begin{thrm}
  [Nelson] \label{thrm-nelson} We have for $\chi\in L^4(M)$,
  \begin{equation}
    \me^{-S_{M,\chi}}\in L^1(\mu_{\mm{GFF}})
    \label{}
  \end{equation}
  and hence
  \begin{equation}
  Z_M={\tts \det_{\zeta}}(\Delta_g+m^2)^{-\frac{1}{2}}\int_{\mathcal{D}'(M)}^{}\me^{-\int_{M}^{}{:}P(\phi){:} \dd V_M}\dd\mu_{\mm{GFF}}^M(\phi)<\infty.
    \label{eqn-def-part-func-rigor}
  \end{equation}
\end{thrm}

\begin{proof}
  For any~$\varepsilon>0$, by (\ref{eqn-upper-bound-green-p}) and (\ref{eqn-lower-bound-hermite}),
  \begin{equation}
    S_{M,\chi,\varepsilon}\ge -b\vol(\text{supp}(\chi)) \Vert \chi\Vert_{L^\infty} \sup_x (|G_{\varepsilon,\varepsilon}(x,x)|^n).
    \label{}
  \end{equation}
From formula (\ref{eqnB9}), for~$\varepsilon$ small,
  \begin{equation}
    G_{\varepsilon,\varepsilon}(x,x)=\int_{2\varepsilon}^{\infty}p_t(x,x)\dd t=\bigg(\underbrace{\int_{2\varepsilon}^{1}}_{A}+\underbrace{\int_{1}^{\infty}}_{B}\bigg) p_t(x,x)\dd t.
    \label{}
  \end{equation}
  Now by (iii) of lemma \ref{lemm-heat} part~$A$ is~$\mathcal{O}(\log(2\varepsilon))$; since our field is \textit{massive} ($m>0$), by (iv) of lemma \ref{lemm-heat} part~$B$ is \textit{bounded}. Therefore one has overall~$G_{\varepsilon,\varepsilon}(x,x)=\mathcal{O}(\log(2\varepsilon))$.
  As a result~$S_{M,\chi,\varepsilon}\ge -b_2 |\log(2\varepsilon)|^n$ for~$\varepsilon$ small.

  Now we compute that
  \begin{align*}
    \mb{P}\big( \me^{-S_{M,\chi}}\ge \me^{b_2|\log(2\varepsilon)|^n +1}\big)&=\mb{P}\left( S_{M,\chi}\le -b_2|\log(2\varepsilon)|^n-1 \right) \\
    &\le \mb{P}\left( |S_{M,\chi}-S_{M,\chi,\varepsilon}|\ge 1 \right) \\
    &\le \nrm{S_{M,\chi}-S_{M,\chi,\varepsilon}}_{L^p(\mu_{\mm{GFF}})}^p \tag{Chebyshev} \\
    &\le (p-1)^{\frac{np}{2}} C_1^p \varepsilon^{\frac{p}{2}}\Vert \chi\Vert_{L^4}^p \tag{proposition \ref{propB7} and \ref{prop-nelson-main}} \\
    &\lesssim \Vert \chi\Vert_{L^4}^pp^{\frac{n}{2}p}(C_1\varepsilon^{\frac{1}{2}})^p,
  \end{align*}
  for all $2\le p <\infty$. The last line as a function of~$p$ has the form dealt with in lemma \ref{lemmB9} and attains a minimum of~$\exp\big( -C_2(\varepsilon^{\frac{1}{2}}\Vert \chi\Vert_{L^4})^{-1/n} \big)$, with some absorbed constant~$C_2>0$ which does not depend on $\chi$. Thus we obtain
  \begin{equation}
    \mb{P}\big( \me^{-S_{M,\chi}}\ge \me^{b_2|\log(2\varepsilon)|^n +1}\big)\lesssim \exp\big( -C_2\big(\varepsilon^{\frac{1}{2}}\Vert \chi\Vert_{L^4}\big)^{-\frac{1}{n}} \big).
    \label{}
  \end{equation}
  Now
we may conclude with the formula
\begin{eqnarray*}
\mathbb{E}\big[e^{-S_{M,\chi}} \big]=\int_0^\infty \mb{P}\left( \me^{-S_{M,\chi}}\ge  t\right)\dd t=\int_0^1 \mb{P}\left( \me^{-S_{M,\chi}}\ge  t\right)\Big| \frac{\dd t}{\dd\varepsilon} \Big| \dd\varepsilon
\end{eqnarray*}  
where the last integral involves a change of variable~$t:=\me^{b_2|\log(2\varepsilon)|^n +1}$. This gives
\begin{eqnarray*}
\mathbb{E}\big[e^{-S_{M,\chi}} \big]\lesssim 1+\int_0^1 \exp\big( -C_2\big(\varepsilon^{\frac{1}{2}}\Vert \chi\Vert_{L^4}\big)^{-\frac{1}{n}} \big) C_3\varepsilon^{-1}n\vert\log(\varepsilon)\vert^{n-1}
e^{c\vert\log(\varepsilon)\vert^n-1}
  \dd\varepsilon
\end{eqnarray*}  
which is finite since integrable near $\varepsilon=0$.   
  Moreover, we see that the bound is uniform when $\Vert \chi\Vert_{L^4}\leqslant C_0$ for some given $C_0>0$.
\end{proof}

\subsection{Change of Wick Ordering}\label{sec-change-wick}

\noindent In order for the proof of proposition \ref{prop-nelson-main} to work as it is written one has to insist on the Wick ordering~${:}\bullet{:}$ provided by~$\mu_{\mm{GFF}}^M$, since we desire convergence in~$L^2(\mu_{\mm{GFF}}^M)$ and with a different Wick ordering the Feynman rules (lemma \ref{lemm-wick-feyn}) are not exact. Nevertheless, in order to define the interaction over a domain~$\Omega$ independently of its embedding in an ambient manifold~$M$, one must employ a Wick ordering independent of~$M$, or in order words, one that is \textit{local}.

Let~$d$ denote the Riemannian distance function of~$M$. This function is local in the sense that~$d(x,y)$ (as $y\to x$) depends only on the restriction of the Riemannian metric on any geodesic convex neighborhood containing~$x$ and~$y$. The local Wick ordering~${:}\bullet{:}_0$ is provided by the \textsf{log-measure}~$\mu_{\log}^M$ which is the Gaussian measure on~$\mathcal{D}'(M)$ with covariance
\begin{equation}
  \mb{E}_{\log}[\phi(f)\phi(h)]=\int_{M}^{}f(x)\Big(-\frac{1}{2\pi}\log(m\cdot d(x,y))\Big) h(y)\dd V_M(x)\dd V_M(y),
  \label{}
\end{equation}
for~$f$,~$h\in C^{\infty}(M)$, thanks to lemma \ref{lemm-2d-log-blow-up-can}. Here~$m$ is the mass used for~$(\Delta+m^2)^{-1}$. We denote~$-\frac{1}{2\pi}\log(md(x,y)):=C_0(x,y)$ and the corresponding operator by~$C_0$. We emphasize here that~$\mu_{\log}^M$ is used only as a tool to produce a linear change of random variables with deterministic coefficients, no random variables will be actually defined on~$\mu_{\log}^M$.

\begin{lemm}
  If~$C_1$,~$C_2$ are two covariance operators on~$M$, then
  \begin{equation}
    {:}\phi(f)^n{:}_{C_1} =\sum_{j=0}^{\lfloor n/2 \rfloor} \frac{n!}{(n-2j)!j!2^j} \ank{f,(C_2-C_1)f}_{L^2(M)}^j {:}\phi(f)^{n-2j}{:}_{C_2},
    \label{}
  \end{equation}
  for~$f\in C^{\infty}(M)$.
\end{lemm}

\begin{proof}
  Follows readily from Wick's theorem.
\end{proof}

The reason why the new Wick ordering works is the following. Let~$G_{(\Delta+m^2)}(x,y)$ denote the integral kernel of~$(\Delta+m^2)^{-1}$.

\begin{lemm}\label{lemm-tadpole-mass}
  For each~$x\in M$, the limit
  \begin{equation}
    \lim_{y\to x}\left( G_{(\Delta+m^2)}(x,y)-C_0(x,y) \right)\defeq \delta G(x)
    \label{}
  \end{equation}
  exists, and that~$\delta G\in L^p(M)$ for all~$1\le p<\infty$.
\end{lemm}

\begin{proof}
  [Remarks for proof.] The function~$\delta G$ is called in our context the (point-splitting) \textsf{tadpole function} (see Kandel, Mnev and Wernli \cite{KMW} section 5.4, in particular lemma 5.20 for a precise expression), which can be seen as a \textsf{renormalized} diagonal value of the Green function~$G_{(\Delta+m^2)}(x,x)$. The asymptotic of the Green function along the diagonal is a classical subject and we have in fact~$G_{(\Delta+m^2)}(x,y)-C_0(x,y)\in C^1(M\times M)$. The function~$\delta G$ is also important in the context of conformal geometry where it is called the \textsf{mass function}, if more precisely we do not include the constant~$-\log m/2\pi$ in~$C_0$ but rather in~$\delta G$. See Hermann and Humbert \cite{HH}, Ludewig \cite{Ludewig} or Schoen and Yau \cite{SY} for more information.
\end{proof}

It follows that
\begin{equation}
  \ank{E_{\varepsilon}\delta_x,\left( (\Delta+m^2)^{-1}-C_0 \right)E_{\varepsilon}\delta_x}_{L^2(M)} \lto \delta G(x)
  \label{}
\end{equation}
as~$\varepsilon\to 0$ and through the limiting process of proposition \ref{prop-nelson-main} we find
\begin{equation}
  \int_{M}^{}\chi(x){:}\phi(x)^{2n}{:}_{0}\dd V_M(x) =\sum_{j=0}^{ n } \frac{(2n)!}{(2n-2j)!j!2^j} \int_{}^{}\chi(x)\delta G(x)^j {:}\phi(x)^{2n-2j}{:}_{\mm{GFF}}\dd V_M(x),
  \label{}
\end{equation}
which exist as a random variable in~$L^p(\mu_{\mm{GFF}}^M)$ for all~$2\le p<\infty$ since each term on the RHS are such by proposition \ref{prop-nelson-main} and corollary \ref{cor-nelson-all-Lp-cutoff-limit}.

  \section{Geometric Operators and Induced Laws}\label{sec-trace-poisson}

  \noindent In this section we obtain a series of rather elementary relations between various geometric-analytic operators on $M$ and on $\Omega$. The moral is that, the so-called ``sharp-time localization'' map $j_{\Sigma}$ (see lemma \ref{lemm-DN-trick}), induced probability laws of Gaussian fields under the trace $\tau_{\Sigma}$ (see section \ref{sec-ind-law}), and finally the Green-Stokes formula, are largely different aspects of the same thing.
  
  \subsection{Summary of Operators Concerned}  \label{sec-dn-map}

   Let~$(M,g)$ be a closed Riemannian manifold, and~$\Sigma\subset M$ a smooth embedded hypersurface (codimension one submanifold). The map
\begin{equation}
  \left.
  \begin{array}{rcl}
    \tau_{\Sigma}:C^{\infty}(M) & \lto & C^{\infty}(\Sigma),\\
    f &\longmapsto & f|_{\Sigma},
  \end{array}
  \right.
  \label{}
\end{equation}
is called the \textsf{trace map} from~$M$ onto~$\Sigma$. If~$(\Omega,g)$ is a compact Riemannian manifold with boundary~$\partial\Omega$, we know that for~$f\in C^{\infty}(\partial\Omega)$, the (Helmholtz) boundary value problem
\begin{equation}
  \left\{
  \begin{array}{ll}
    (\Delta_{\Omega}+m^2)u=0,&\textrm{in }\Omega,\\
    u|_{\partial\Omega}=f,&\textrm{on }\partial\Omega,
  \end{array}
  \right.
  \label{eqn-PI-pde-bdy}
\end{equation}
admits a unique solution~$u\in C^{\infty}(\ol{\Omega})$ which is extendably smooth upto~$\partial\Omega$. The solution operator
  \begin{equation}
    \left.
    \begin{array}{rcl}
      \PI_{\Omega}^{\partial\Omega}:C^{\infty}(\partial\Omega) &\lto &C^{\infty}(\ol{\Omega}),\\
      f&\longmapsto & u\textrm{ solving (\ref{eqn-PI-pde-bdy})},
    \end{array}
    \right.
    \label{}
  \end{equation}
  is called in this paper the \textsf{Poisson integral operator} from~$\partial\Omega$ to~$\Omega$, \textit{with mass}~$m>0$. We also need a variant of this operator that works for embedded hypersurfaces in closed manifolds, as $M$ and $\Sigma$ above. Pick~$f\in C^{\infty}(\Sigma)$. This time we look at the boundary value problem
\begin{equation}
  \left\{
  \begin{array}{ll}
    (\Delta_{M}+m^2)u=0,&\textrm{in }M\setminus\Sigma,\\
    u|_{\Sigma}=f,&\textrm{on }\Sigma.
  \end{array}
  \right.
  \label{eqn-PI-pde-hyp}
\end{equation}
Indeed, one views~$M\setminus \Sigma$ as a manifold with two boundaries~$\Sigma\sqcup \Sigma$, and as a result one obtains a unique solution~$u$ which is smooth on~$M\setminus \Sigma$ and one-sidedly smooth upto~$\Sigma$ respectively on its two sides. In this case we denote the solution operator by
  \begin{equation}
    \left.
    \begin{array}{rcl}
      \PI_{M}^{\Sigma}:C^{\infty}(\Sigma) &\lto &C^{\infty}(\ol{M\setminus\Sigma}),\\
      f&\longmapsto & u\textrm{ solving (\ref{eqn-PI-pde-hyp})},
    \end{array}
    \right.
    \label{eqn-def-pi-emb-hyp-in-closed-case}
  \end{equation}
  also called the \textsf{Poisson integral operator}, from~$\Sigma$ to~$M$, \textit{with mass}~$m>0$.

\begin{lemm}
  [\cite{Taylor1} page 334, 361, \cite{Eskin} page 57, example 13.3] \label{lemm-trace-prop} Let~$(M,g)$,~$\Sigma$, $\Omega$ and $\partial\Omega$ be as above. Then
  \begin{enumerate}[(i)]
    \item the map~$\tau_{\Sigma}$ extends uniquely to a continuous operator~$\tau_{\Sigma}:W^s(M)\lto W^{s-\frac{1}{2}}(\Sigma)$ for each~$s>\frac{1}{2}$;
    \item the map~$\tau_{\Sigma}:W^s(M)\lto W^{s-\frac{1}{2}}(\Sigma)$,~$s>\frac{1}{2}$, is surjective;
    \item the map~$\PI_{\Omega}^{\partial\Omega}$ extends uniquely to a continuous operator $\PI_{\Omega}^{\partial\Omega}:W^s(\partial \Omega)\lto W^{s+\frac{1}{2}}(\Omega)$ for each~$s\ge -\frac{1}{2}$.
  \end{enumerate}
\end{lemm}

\begin{def7}\label{rem-Dir-Neu-cond-mean}
  In this paper, a function~$u\in C^{\infty}(\ol{\Omega})$ is said to satisfy the \textsf{Dirichlet condition} (respectively \textsf{Neumann}) if~$u|_{\partial\Omega}=0$ (respectively~$(\partial_{\nu}u)|_{\partial\Omega}=0$,~$\nu$ the outward unit normal). The~$f$ appearing in (\ref{eqn-PI-pde-bdy}) is called a \textsf{Dirichlet datum}.
\end{def7}

\begin{deef}\label{def-pi-more-general}
   More generally if~$\Omega$ has boundaries and~$\Sigma$ is \textit{either} one component of~$\partial\Omega$ \textit{or} embedded in the \textit{interior} of~$\Omega$, then we denote by~$\PI_{\Omega}^{\Sigma,B}f$ the solution with its restriction equal to~$f$ on~$\Sigma$ and boundary condition ``$B$'' (Dirichlet or Neumann) on all components of~$\partial\Omega$ \textit{except}~$\Sigma$. Such notations raise no ambiguity when the situation is understood from context.
\end{deef}

\begin{lemm}\label{lemm-pi-hyp-reg}
  Let~$(M,g)$ and~$\Sigma$ be as above. Then~$\PI_{M}^{\Sigma}$ extends uniquely to a continuous operator
  \begin{equation}
    \PI_{M}^{\Sigma}:W^{\frac{1}{2}}(\Sigma)\lto W^{1}(M).
    \label{}
  \end{equation}
\end{lemm}

\begin{proof}
  Let~$f\in W^{\frac{1}{2}}(\Sigma)$ and~$u:=\PI_M^{\Sigma}f$. We know that~$u\in W^1(M\setminus\Sigma)\subset \mathcal{D}'(M\setminus\Sigma)$, this means~$u\in L^2(M\setminus \Sigma)=L^2(M)$, and~$\nabla u\in L^2(T(M\setminus \Sigma))$, as a distribution over~$M\setminus \Sigma$. The problem is to compute~$\nabla u$ as a distribution over~$M$. For this, one picks a testing vector field~$X\in C^{\infty}(M,TM)$ and applies the Green-Stokes formula to get
  \begin{align*}
    \ank{\nabla u,X}_{L^2(M,TM)}&\defeq -\ank{u,\ddiv X}_{L^2(M)}\\
    &=\int_{M\setminus\Sigma}^{}\ank{\nabla u,X}_g \dd V_M -\int_{\Sigma}^{} f\ank{X,\nu}_g \dd V_{\Sigma}-\int_{\Sigma}^{}f\ank{X,-\nu}_g\dd V_{\Sigma} \\
    &=\int_{M\setminus\Sigma}^{}\ank{\nabla u,X}_g \dd V_M.
  \end{align*}
  Here~$\nu$ is any one of the two possible unit normal vector fields along~$\Sigma$. This shows that, nevertheless,
  \begin{equation}
    \nabla^M u=\nabla^{M\setminus\Sigma}u,
    \label{}
  \end{equation}
  and hence~$\nrm{u}_{W^1(M)}=\nrm{\nabla u}_{L^2(T(M\setminus\Sigma))}+\nrm{u}_{L^2(M)}\approx\nrm{u}_{W^1(M\setminus\Sigma)}$. 
\end{proof}

  \begin{deef}
  Let~$(\Omega,g)$ be a compact Riemannian manifold with boundary~$\partial \Omega\ne\varnothing$, and~$\PI_{\Omega}^{\partial \Omega}:C^{\infty}(\partial \Omega)\to C^{\infty}(\ol{\Omega})$ the Poisson operator defined previously. Put
  \begin{equation}
    \left.
    \begin{array}{rcl}
       \DN_{\Omega}^{\partial \Omega}:C^{\infty}(\partial \Omega)&\lto & C^{\infty}(\partial \Omega),\\
       f&\longmapsto & \partial_{\nu}(\PI_{\Omega}^{\partial \Omega} f),
    \end{array}
    \right.
    \label{}
  \end{equation}
  where~$\nu=$ outward unit normal along $\partial \Omega$, called the \textsf{Dirichlet-to-Neumann operator} on $\partial\Omega$ with respect to $\Omega$.
\end{deef}

  \begin{deef} \label{def-jp-DN-map}
  Let~$(M,g)$ be a closed Riemannian manifold,~$\Sigma\subset M$ an embedded hypersurface, and~$\PI_{M}^{\Sigma}:C^{\infty}(\partial \Omega)\to C^{\infty}(M\setminus \Sigma)$ the hypersurface Poisson operator. Put
  \begin{equation}
    \left.
    \begin{array}{rcl}
       \DN_{M}^{\Sigma}:C^{\infty}(\partial \Omega)&\lto & C^{\infty}(\partial \Omega),\\
       f&\longmapsto & \partial_{\nu}(\PI_{M}^{\Sigma} f)|_{\Sigma_-}+\partial_{-\nu}(\PI_{M}^{\Sigma} f)|_{\Sigma_+},
    \end{array}
    \right.
    \label{}
  \end{equation}
  where~$\nu$ is any one of the two unit normal vector fields along~$\Sigma$, extended over a cylindrical neighborhood of~$\Sigma$. Here~$\Sigma_{-}$ and~$\Sigma_+$ means that we are taking one-sided derivatives, respectively, from the backward-time and forward-time directions with regard to the flow of~$\nu$. We call~$\DN_{M}^{\Sigma}$ the \textsf{jumpy Dirichlet-to-Neumann operator} on $\Sigma$ with respect to $M$.
\end{deef}

\begin{deef}\label{def-DN-boundary}
   Similarly if~$\Omega$ has boundaries and~$\Sigma$ is \textit{either} one component of~$\partial\Omega$ \textit{or} embedded in the \textit{interior} of~$\Omega$, then we denote by~$\DN_{\Omega}^{\Sigma,B}$ the corresponding operator with~$\PI_{M}^{\Sigma}$ replaced by~$\PI_{\Omega}^{\Sigma,B}$ in the definition (see definition \ref{def-pi-more-general}).
\end{deef}

\begin{def7}\label{rmk-dn-jump}
  If we see~$M\setminus\Sigma$ as a manifold with boundary~$\Sigma\sqcup \Sigma$, then~$\DN_M^{\Sigma}f$ is also the sum over~$\Sigma$ of the two \textit{outward} unit normal derivatives of~$\PI_M^{\Sigma}f$ along~$\Sigma\sqcup\Sigma$. Intuitively,~$\DN_M^{\Sigma}f$ describes the ``jump'' of~$\nabla\PI_M^{\Sigma}f$ across~$\Sigma$.
\end{def7}

We summarize in the following lemma the essential properties of~$\DN_{\Omega}^{\partial\Omega}$ and~$\DN_M^{\Sigma}$. Parallel results also hold for~$\DN_{\Omega}^{\Sigma,B}$ ($\Sigma$ being either one component of boundary or embedded in interior) but we shall not discuss them in order to simplify the presentation. The same applies to everything below this section.

\begin{lemm}\label{lemm-dn-prop}
  Under their respective settings,~$\DN_{\Omega}^{\partial\Omega}$ and~$\DN_M^{\Sigma}$ are such that
  \begin{enumerate}[(i)]
    \item their quadratic forms are given respectively by the Dirichlet energies of their harmonic extensions:
      \begin{align}
	\bank{f,\DN_{\Omega}^{\partial\Omega}f}_{L^2(\partial\Omega)}&=\int_{\Omega}^{}\big(|\nabla\PI_{\Omega}^{\partial\Omega}f|_g^2+m^2(\PI_{\Omega}^{\partial\Omega}f)^2\big) \dd V_{\Omega}, \\
	\bank{f,\DN_{M}^{\Sigma}f}_{L^2(\Sigma)}&=\int_{M}^{}\big(|\nabla\PI_{M}^{\Sigma}f|_g^2+m^2(\PI_{M}^{\Sigma}f)^2\big) \dd V_{M}, 
	\label{eqn-dn-qua-hyp}
      \end{align}
      for~$f\in C^{\infty}(\partial\Omega)$;
    \item they are formally self-adjoint, strictly positive, and~$L^2$-invertible;
    \item they are elliptic~$\Psi$DOs of order~$1$, with principal symbols being~$|\xi|_g$ and~$2|\xi|_g$ respectively;
    \item they afford a finer comparison with~$\mn{D}_{\partial\Omega}=(\Delta_{\partial\Omega}+m^2)^{\frac{1}{2}}$ or~$2\mn{D}_{\Sigma}=2(\Delta_{\Sigma}+m^2)^{\frac{1}{2}}$: the operators
      \begin{equation}
	\DN_{\Omega}^{\partial\Omega}-\mn{D}_{\partial\Omega},\quad \DN_{M}^{\Sigma}-2\mn{D}_{\Sigma},\quad \mn{D}_{\partial\Omega}^{-1}\DN_{\Omega}^{\partial\Omega}-\one, \quad \textrm{and}\quad (2\mn{D}_{\Sigma})^{-1}\DN_{M}^{\Sigma}-\one,
	\label{}
      \end{equation}
      are~$\Psi$DOs of orders at most~$-2$,~$-2$,~$-3$, and~$-3$ respectively. A fortiori, they are all of trace class when~$\dim \Omega=\dim M=2$ and~$\dim \Sigma=\dim \partial\Omega=1$.
  \end{enumerate}
\end{lemm}

\begin{proof}
  See Taylor \cite{Taylor3} and the references therein. For (iv) see Kandel, Mnev and Wernli \cite{KMW} proposition A.3.
\end{proof}

\subsection{Two Consequences of the Green-Stokes Formula}\label{sec-green-stokes}

\noindent Let~$(M,g)$ be a closed Riemannian manifold and~$\Sigma\subset M$ an embedded hypersurface. Formula (\ref{eqn-dn-qua-hyp}) in a slightly more general form allows one to obtain an expression for the ``distributional adjoint'' of the trace map onto~$\Sigma$.

\begin{lemm}\label{lemm-DN-trick}
  Let~$\tau_{\Sigma}:C^{\infty}(M)\to C^{\infty}(\Sigma)$,~$\phi\mapsto \phi|_{\Sigma}$ be the trace map. One has, for any~$\phi\in C^{\infty}(M)$ and~$f\in C^{\infty}(\Sigma)$,
  \begin{equation}
    \ank{\tau_{\Sigma}\phi,f}_{L^2(\partial\Omega)}=\ank{\phi,j_{\Sigma}f}_{L^2(M)}
    \label{eqn-adj-trace}
  \end{equation}
  where~$j_{\Sigma}=(\Delta+m^2)\PI_M(\DN_M^{\Sigma})^{-1}$. Moreover, this equality can be extended to~$\phi\in W^1(M)$ and~$f\in W^{-\frac{1}{2}}(\Sigma)$.
\end{lemm}

\begin{proof}
  First suppose~$h\in C^{\infty}(\Sigma)$, then applying the Green-Stokes formula to~$M\setminus\Sigma$ with boundary~$\Sigma\sqcup\Sigma$ gives
\begin{align*}
  \ank{\tau\phi,\DN_M^{\Sigma}h}_{L^2(\Sigma)}&=\int_{M\setminus\Sigma}^{} (\langle\nabla \phi,\nabla \PI_M h\rangle+ m^2 \phi(\PI_M h)) \dd V_{M} \\
&=\int_{M}^{} (\langle\nabla \phi,\nabla \PI_M h\rangle+ m^2 \phi(\PI_M h)) \dd V_{M} \tag{lemma \ref{lemm-pi-hyp-reg}}\\
&=\ank{\phi,(\Delta_M+m^2)\PI_M^{\Sigma} h}_{L^2(M)}. \tag{\#}
\end{align*}
We remark that step (\#) is the \textit{definition} of the action of~$(\Delta_M+m^2)$ on the distribution~$\PI_M^{\Sigma} h$. By lemma \ref{lemm-trace-prop}, lemma \ref{lemm-pi-hyp-reg}, and (iii) of lemma \ref{lemm-dn-prop}, then, this equality can be extended to~$\phi\in W^1(M)$ and~$h\in W^{\frac{1}{2}}(\Sigma)$. Finally, replacing~$h$ by~$(\DN_{M}^{\Sigma})^{-1}f$, with~$(\DN_M^{\Sigma})^{-1}$ being a~$\Psi$DO of order~$-1$, yields the desired relation (\ref{eqn-adj-trace}) as well as its domain.
\end{proof}

\begin{def7}
  One is advised to compare lemma \ref{lemm-DN-trick} with the fact in one dimensions that the distributional derivative of the Heaviside function~$H_a=a\cdot 1_{(0,\infty)}$ ($a\in\mb{R}$) is the delta function multiplied by the jump of~$H_a$ across~$0$, that is,
  \begin{equation}
  \ank{H_a',\varphi}_{L^2(\mb{R})}=[H_a(0+)-H_a(0-)]\cdot \varphi(0)
    \label{}
  \end{equation}
  for any~$\varphi\in \mathcal{S}(\mb{R})$. In our case the role of the Heaviside function is played by the vector field~$\nabla \PI_M^{\Sigma}(\DN_M^{\Sigma})^{-1}f$. Indeed, following remark \ref{rmk-dn-jump}, the ``jump'' of~$\nabla \PI_M^{\Sigma}(\DN_M^{\Sigma})^{-1}f$ across $\Sigma$ is exactly $f$, as the directions tangential to~$\Sigma$ does not contribute to the jump with~$(\DN_M^{\Sigma})^{-1}f$ being smooth. This comparison in mind, it is also customary to write~$j_{\Sigma}f$ as~$f\otimes \delta_{\Sigma}$, as for example, in Carron \cite{Carron}.
\end{def7}

\begin{corr}\label{cor-dn-is-conj-of-green}
  For~$f$,~$h\in C^{\infty}(\Sigma)$, we have
  \begin{equation}
    \ank{f,(\DN_M^{\Sigma})^{-1}h}_{L^2(\Sigma)}=\ank{j_{\Sigma}f,(\Delta+m^2)^{-1}j_{\Sigma}h}_{L^2(M)}.
    \label{}
  \end{equation}
  In other words,~$(\DN_M^{\Sigma})^{-1}=\tau(\Delta+m^{2})^{-1}j_{\Sigma}~(=\tau(\Delta+m^{2})^{-1}\tau^*)$.
\end{corr}

\begin{proof}
  This is immediate by noting that~$\tau_{\Sigma}\PI_M^{\Sigma}$ is the identity on~$L^2(\Sigma)$.
\end{proof}

\begin{def7}
  Indeed, noting that the Schwartz kernel~$K_{\tau}$ of~$\tau_{\Sigma}$ is the delta distribution on the diagonal~$\{(x,x)\}\subset \Sigma\times M$, and that $j_{\Sigma}$ is the distributional adjoint of $\tau_{\Sigma}$, corollary \ref{cor-dn-is-conj-of-green} allows one to deduce immediately the Schwartz kernel of~$(\DN_M^{\Sigma})^{-1}$, denoted $G_{\DN}^{\Sigma}$:
  \begin{align*}
    \ank{f,(\DN_M^{\Sigma})^{-1}h}_{L^2(\Sigma)}&=\iint_{\Sigma\times M}\iint_{\Sigma\times M} f(x)K_{\tau}(x,z)G_{(\Delta+m^2)}(z,w)K_{\tau}(y,w)h(y)\dd x\dd z\dd y\dd w \\
    &=\iint_{\Sigma\times \Sigma}f(x)G_{(\Delta+m^2)}(x,y)h(y)\dd x\dd y,
  \end{align*}
  that is,~$G_{\DN}^{\Sigma}=G_{(\Delta+m^2)}|_{\Sigma\times \Sigma}$, where $G_{(\Delta+m^2)}$ is the Helmholtz Green function on $M$, which is a well-known result. Of course, assuming this result, one could also work backwards to give lemma \ref{lemm-DN-trick} another proof, using the Poisson integral formula (lemma \ref{lemm-adj-PI-bdy} below) for $\PI_M^{\Sigma}$.
\end{def7}

Now we move to the second consequence of the Green-Stokes formula. Let~$(\Omega,g)$ be a compact Riemannian manifold with boundary~$\partial\Omega$. Recall that~$(\Delta_{\Omega,D}+m^2)^{-1}$ denotes the Helmholtz Green operator with \textit{Dirichlet} conditions on~$\partial\Omega$.

\begin{lemm}[\cite{Taylor2} page 46] \label{lemm-adj-PI-bdy}
  We have, for~$\varphi\in C^{\infty}(\partial\Omega)$ and~$f\in C_c^{\infty}(\Omega^{\circ})$,
  \begin{equation}
    \bank{\PI_{\Omega}^{\partial\Omega}\varphi,f}_{L^2(\Omega)}=
    \bank{\varphi,-\partial_{\nu}(\Delta_{\Omega,D}+m^2)^{-1}f|_{\partial\Omega}}_{L^2(\partial\Omega)},
    \label{}
  \end{equation}
  where~$\nu$, again, denotes the \textit{outward} unit normal vector field along~$\partial\Omega$.\hfill~$\Box$
\end{lemm}

\begin{corr}
  For~$(M,g)$ a closed Riemannian manifold and~$\Sigma\subset M$ an embedded hypersurface, for~$\varphi\in C^{\infty}(\Sigma)$ and~$f\in C^{\infty}_c(M\setminus\Sigma)$,
  \begin{equation}
    \ank{\PI_M^{\Sigma}\varphi,f}_{L^2(M)}=\ank{\varphi,-(\partial_{\nu}u|_{\Sigma_-}+\partial_{-\nu}u|_{\Sigma_+})}_{L^2(\Sigma)},
    \label{}
  \end{equation}
  where~$u=(\Delta_{M\setminus\Sigma,D}+m^2)^{-1}f$, and the notations~$\Sigma_-$ and~$\Sigma_+$ have the same meanings as in definition \ref{def-jp-DN-map}.
\end{corr}

\begin{proof}
  See~$M\setminus\Sigma$ as a manifold with two boundaries~$\Sigma\sqcup\Sigma$, and we note
  \begin{equation}
    \PI_M^{\Sigma}\varphi=\PI_{M\setminus\Sigma}^{\Sigma\sqcup\Sigma}\bnom{\varphi}{\varphi},
    \label{eqn-two-PI-same-law}
  \end{equation}
  as well as
  \begin{equation}
    \Bank{\bnom{\varphi}{\varphi},-\bnom{\partial_{\nu}u|_{\Sigma_-}}{\partial_{-\nu}u|_{\Sigma_+}}}_{L^2(\Sigma\sqcup\Sigma)}=
\bank{\varphi,-(\partial_{\nu}u|_{\Sigma_-}+\partial_{-\nu}u|_{\Sigma_+})}_{L^2(\Sigma)},
    \label{}
  \end{equation}
  while applying lemma \ref{lemm-adj-PI-bdy}.
\end{proof}

\subsection{Induced Laws}\label{sec-ind-law}

Results of this section rely on the possibility of extending linear functionals or operators measurably from the Cameron-Martin space, and the (almost sure) uniqueness of such extensions, a detailed treatment of which we refer to \cite{Bogachev} sections 2.10 and 3.7.

Now let~$(M,g)$ be a closed Riemannian manifold and~$\Sigma\subset M$ an embedded hypersurface. Lemma \ref{lemm-DN-trick} then says that, for each~$f\in C^{\infty}(\Sigma)$, the random variables
\begin{equation}
  \tau_{\Sigma}\phi(f)\quad\textrm{and}\quad \phi(j_{\Sigma}f),
  \label{}
\end{equation}
while~$\phi\sim \mu_{\mm{GFF}}^M$, are (surely) equal on the Cameron-Martin space~$W^1(M)$. By \cite{Bogachev} theorem 2.10.11, they are almost surely equal over~$\mathcal{D}'(M)$. Subsequently from corollary \ref{cor-dn-is-conj-of-green} we deduce
\begin{equation}
  \mb{E}_{\mm{GFF}}^M[\tau_{\Sigma}\phi(f)\tau_{\Sigma}\phi(h)]=\mb{E}_{\mm{GFF}}^M[\phi(j_{\Sigma}f)\phi(j_{\Sigma}h)]=\ank{f,(\DN_M^{\Sigma})^{-1}h}_{L^2(\Sigma)}.
  \label{}
\end{equation}
Taking into account lemma \ref{lemm-GHS-cam-mar-of-cov-C}, (iii) of lemma \ref{lemm-dn-prop} and \cite{Bogachev} theorem 3.7.6, we have proved the following.
\begin{prop}\label{prop-induce-law-DN}
  If~$\phi\in \mathcal{D}'(M)$ follows the law of~$\mu_{\mm{GFF}}^M$, then the random field~$\tau_{\Sigma}\phi\in \mathcal{D}'(\Sigma)$ can equivalently be realized as the (centered) Gaussian field~$\tilde{\varphi}$ on~$\Sigma$ with covariance
  \begin{equation}
    \mb{E}\big[\tilde{\varphi}(f)\tilde{\varphi}(h)\big]=\ank{f,(\DN_{M}^{\Sigma})^{-1}h}_{L^2(\Sigma)},
    \label{eqn-cov-dn-field}
  \end{equation}
  for~$f$,~$h\in C^{\infty}(\Sigma)$. In other words, the measure image~$\wh{\tau_{\Sigma}}_*(\mu_{\mm{GFF}}^M)$ of~$\mu_{\mm{GFF}}^M$ under any measurable linear extension~$\wh{\tau_{\Sigma}}$ of~$\tau_{\Sigma}:W^1(M)\lto W^{\frac{1}{2}}(\Sigma)$ coincides with the measure~$\mu_{\DN}^{\Sigma,M}$ on any~$\mathcal{D}'(\Sigma)$ for the field~$\tilde{\varphi}$ satisfying (\ref{eqn-cov-dn-field}). \hfill~$\Box$
\end{prop}

Next we study induced random fields in the other direction, by the Poisson integral operator. Namely, for~$\Omega$,~$\partial\Omega$ as in lemma \ref{lemm-adj-PI-bdy}, given a Gaussian random field~$\varphi$ on~$\partial\Omega$, what is the law of the field~$\PI_{\Omega}^{\partial\Omega}\varphi$? From another perspective one solves the Helmholtz (Laplace) equation with random boundary conditions. We write in shorthand
\begin{equation}
    (\PI_{\Omega}^{\partial\Omega})^*\defeq -\partial_{\nu}(\Delta_{\Omega,D}+m^2)^{-1}(-)|_{\partial\Omega}.
\end{equation}
Thus lemma \ref{lemm-adj-PI-bdy} says
\begin{equation}
\bank{\PI_{\Omega}^{\partial\Omega}\varphi,f}_{L^2(\Omega)}=
  \bank{\varphi,(\PI_{\Omega}^{\partial\Omega})^* f}_{L^2(\partial\Omega)},
  \label{}
\end{equation}
for~$\varphi\in C^{\infty}(\partial\Omega)$,~$f\in C_c^{\infty}(\Omega^{\circ})$. Suppose~$\varphi$ has covariance operator~$C$ of order~$-s$ ($s>0$). Note~$\PI_{\Omega}^{\partial\Omega}$ is always well-defined on the Cameron-Martin space~$W^s(\partial\Omega)$. By the same token as above, for~$f\in C_c^{\infty}(\Omega^{\circ})$,
\begin{equation}
  (\PI_{\Omega}^{\partial\Omega}\varphi)(f)\quad\textrm{and}\quad \varphi((\PI_{\Omega}^{\partial\Omega})^*f)
  \label{}
\end{equation}
are surely equal on~$W^s(\partial\Omega)$. Moreover,
\begin{equation}
  \mb{E}_{C}^{\partial\Omega}\big[(\PI_{\Omega}^{\partial\Omega}\varphi)(f)(\PI_{\Omega}^{\partial\Omega}\varphi)(h)\big]=\bank{f,\PI_{\Omega}^{\partial\Omega} C(\PI_{\Omega}^{\partial\Omega})^*h}_{L^2(\partial\Omega)}.
\end{equation}
We deduce
\begin{prop}\label{prop-induce-PI-law}
  If~$\varphi\in \mathcal{D}'(\partial\Omega)$ is a (centered) Gaussian random field with covariance operator~$C$, then the random field~$\PI_{\Omega}^{\partial\Omega}\varphi\in \mathcal{D}'(\Omega^{\circ})$ can equivalently be realized as the (centered) Gaussian field~$\tilde{\phi}$ on~$\Omega$ with covariance
  \begin{equation}
    \mb{E}\big[\tilde{\phi}(f)\tilde{\phi}(h)\big]=\bank{f,\PI_{\Omega}^{\partial\Omega} C(\PI_{\Omega}^{\partial\Omega})^*h}_{L^2(\partial\Omega)},
    \label{eqn-cov-pi-ind-field}
  \end{equation}
  for~$f$,~$h\in C^{\infty}(\partial\Omega)$. In other words, the measure image~$\wh{\PI_{\Omega}^{\partial\Omega}}_*(\mu_{C}^{\partial\Omega})$ of~$\mu_{C}^{\partial\Omega}$ under any measurable linear extension~$\wh{\PI_{\Omega}^{\partial\Omega}}$ of~$\PI_{\Omega}^{\partial\Omega}:W^s(\partial\Omega)\lto W^{s+\frac{1}{2}}(\Omega)$ coincides with the measure for the field~$\tilde{\phi}$ satisfying (\ref{eqn-cov-pi-ind-field}). \hfill~$\Box$
\end{prop}

 \begin{def7}\label{rem-induce-law-equality}
    What is strictly needed for showing Segal axioms is not the full proposition \ref{prop-induce-PI-law} but rather this innocent observation: by (\ref{eqn-two-PI-same-law}), if the random field~$\varphi\in \mathcal{D}'(\Sigma)$ follows a fixed probability law, then the induced random fields~$\PI_M^{\Sigma}\varphi$ and~$\PI_{M\setminus\Sigma}^{\Sigma\sqcup\Sigma}[\smx{\varphi\\ \varphi}]$ in~$\mathcal{D}'(M\setminus\Sigma)$ follows the same law.
  \end{def7}

\subsection{A Remark on Reflection Positivity}  
\label{sec-RP-first}

\noindent As the names would suggest, the positivity of the Dirichlet-to-Neumann map (itself the consequence of the positivity of the Dirichlet energy) gives an interesting inequality comparing the resolvants of Laplacians with Dirichlet and Neumann conditions (corollary \ref{cor-rp-op-ineq}), via the Poisson integral formula (lemma \ref{lemm-adj-PI-bdy}). Let~$\Omega$,~$\partial\Omega$ be as in lemma \ref{lemm-adj-PI-bdy}.  We adopt the shorthand notations
\begin{equation}
  \quad C_D\defeq (\Delta_{\Omega,D}+m^2)^{-1},\quad C_N\defeq (\Delta_{\Omega,N}+m^2)^{-1},
  \label{}
\end{equation}
where, similar to~$C_D$,~$C_N f$ for~$f\in C_c^{\infty}(\Omega^{\circ})$ solves the Neumann boundary value problem
\begin{equation}
  \left\{
    \begin{array}{ll}
      (\Delta+m^2)(C_N f)=f,&\textrm{in }\Omega,\\
      \partial_{\nu}(C_N f)|_{\partial\Omega}=0,&\textrm{on }\partial\Omega.
    \end{array}
    \right.
  \label{}
\end{equation}
There is the following simple, elementary relation:
\begin{lemm}
  In the situation as above, we have the operator equality on~$C_c^{\infty}(\Omega^{\circ})$,
  \begin{equation}
    \PI_{\Omega}^{\partial\Omega}(\DN_{\Omega}^{\partial\Omega})^{-1}(\PI_{\Omega}^{\partial\Omega})^*=C_N-C_D.
    \label{}
  \end{equation}
\end{lemm}

\begin{proof}
  Pick~$f\in C_c^{\infty}(\Omega^{\circ})$, let~$u:=C_D f$ and put~$w:=\PI_{\Omega}^{\partial\Omega}(\DN_{\Omega}^{\partial\Omega})^{-1}(-\partial_{\nu} u|_{\partial\Omega})$. Then, by the definition of~$\DN_{\Omega}^{\partial\Omega}$, $w$ solves the following boundary value problem:
  \begin{equation}
    \left\{
    \begin{array}{ll}
      (\Delta+m^2)w=0,&\textrm{in }\Omega,\\
      \partial_{\nu}w|_{\partial\Omega}=-\partial_{\nu}u|_{\partial\Omega},&\textrm{on }\partial\Omega.
    \end{array}
    \right.
    \label{}
  \end{equation}
  However,
  \begin{equation}
    \left\{
    \begin{array}{ll}
      (\Delta+m^2)u=f,&\textrm{in }\Omega,\\
     u|_{\partial\Omega}=0,&\textrm{on }\partial\Omega,
    \end{array}
    \right.\quad\textrm{implying}\quad
    \left\{
    \begin{array}{ll}
      (\Delta+m^2)(u+w)=f,&\textrm{in }\Omega,\\
      \partial_{\nu}(u+w)|_{\partial\Omega}=0,&\textrm{on }\partial\Omega,
    \end{array}
    \right.
    \label{}
  \end{equation}
  namely~$u+w=C_N f$, that is,~$w=C_N f-C_D f$. We obtain the result.
\end{proof}

Now the positivity of~$\DN_{\Omega}^{\partial\Omega}$ (lemma \ref{lemm-dn-prop}) implies~$C_N\ge C_D$, namely~$\ank{f,(C_N-C_D)f}_{L^2(\Omega)}\ge 0$ for all~$f\in C_c^{\infty}(\Omega^{\circ})$. One step further,

\begin{corr}\label{cor-rp-op-ineq}
  We have~$C_N\ge C_D$ as operators on~$L^2(\Omega)$. \hfill~$\Box$
\end{corr}

It is emphasized in Jaffe and Ritter \cite{JR} section 3 that~$C_N\ge C_D$ is the crucial relation that leads to the so-called \textit{reflection positivity} (RP) of the GFF. The geometric set-up is as follows. Let now~$\partial\Omega\subset\Omega$ be totally geodesic,~$\Omega^*$ a copy of~$\Omega$ (reversing the coorientation of~$\partial\Omega$),~$|\Omega|^2:=\Omega^*\cup_{\partial\Omega}\Omega$, the \textsf{isometric double} which is a closed Riemannian manifold, and~$\Theta:|\Omega|^2\lto |\Omega|^2$ an isometric involution fixing~$\partial\Omega$, such that~$\Theta(\Omega)=\Omega^*$ and~$\Theta(\Omega^*)=\Omega$.

\begin{def7}\label{rem-real-tun-geo}
  Such~$\Omega$ is named in Gibbons \cite{Gibbons} (in the 4-dimensional case) as a \textsf{real tunnelling geometry} (see \cite{Gibbons} section 4). The isometric double $|\Omega|^2$ and the involution~$\Theta$ exist, the latter being a \textsf{reflection}, \textit{dissecting} $|\Omega|^2$ into $\Omega$ and $\Omega^*$, its fixed point set being~$\partial\Omega$ (see also Ritter \cite{Ritter} section 2.1.1).
\end{def7}

The action of~$\Theta$ extends in the usual way to~$C^{\infty}(|\Omega|^2)$ and~$\mathcal{D}'(|\Omega|^2)$ by pulling-back. This set-up brings in another resolvant operator which is
\begin{equation}
  C\defeq (\Delta_{|\Omega|^2}+m^2)^{-1}.
  \label{}
\end{equation}
Denote also by~$\Pi_+:L^2(|\Omega|^2)\lto L^2(\Omega)$ the orthogonal projection. Then we have

\begin{lemm}[\cite{JR} lemma 3]
  Let~$|\Omega|^2$,~$\Theta$,~$\Pi_+$ and~$C$ be as above. Then
  \begin{equation}
    \Pi_+ \Theta C=\frac{1}{2}(C_N-C_D)
    \label{}
  \end{equation}
  on~$C_c^{\infty}(\Omega^{\circ})$ and~$L^2(\Omega)$.\hfill~$\Box$
\end{lemm}

In summary,

\begin{corr}
  [equivalent formulations of RP] \label{cor-equiv-form-rp} In the situation as above, we have
  \begin{equation}
    2\bank{f,\Theta Cf}_{L^2(|\Omega|^2)}=\bank{f,(C_N-C_D)f}_{L^2(\Omega)}=\bank{(\PI_{\Omega}^{\partial\Omega})^*f,(\DN_{\Omega}^{\partial\Omega})^{-1}(\PI_{\Omega}^{\partial\Omega})^* f}_{L^2(\partial\Omega)}
    \label{eqn-rp-equiv-form}
  \end{equation}
  for all~$f\in C_c^{\infty}(\Omega^{\circ})$, and all of the above quantities are nonnegative. \hfill~$\Box$
\end{corr}

\begin{def7}
  Note that the 3 quantities in (\ref{eqn-rp-equiv-form}) make sense respectively on the spaces~$L^2(|\Omega|^2)$, $L^2(\Omega)$ and~$L^2(\partial\Omega)$. While the first quantity is the original view of RP, the second quantity offers a \textit{one-sided} view and the third provides a view \textit{within the boundary}~$\partial\Omega$. Nevertheless, the map~$\DN_{\Omega}^{\partial\Omega}$ reflects the geometry of the bulk, see for example Paternain, Salo and Uhlmann \cite{PSU} section 11.5.
\end{def7}

\section{Markov Property and Consequences}\label{sec-markov-main}

\noindent Materials in sections \ref{sec-sobo-decomp}, \ref{sec-stoc-decomp-gff} are largely classical with an excellent source being Dimock \cite{Dimock} which we follow roughly. See also Simon \cite{Sim2} section III.3 and see Powell and Werner \cite{powell-werner} section 4.2 for a probabilistic point of view.

  \subsection{Decompositions of Sobolev Spaces}\label{sec-sobo-decomp}

   \noindent Let~$(M,g)$ be a closed Riemannian manifold. Recall from Appendix \ref{sec-app-sobo} the definitions of the spaces~$W^s_A(M)$ and~$W^s_U(M)$ for~$A\subset M$ closed and~$U\subset M$ open. Recall also (lemma \ref{lemm-sobo-inner-prod}) that we have the isometric isomorphism~$(\Delta+m^2):W^1(M)\xlongrightarrow{\sim}W^{-1}(M)$ as we endow~$W^1(M)$ and~$W^{-1}(M)$ respectively with the inner products~$\sank{-,(\Delta+m^2)-}_{L^2}$ and~$\sank{-,(\Delta+m^2)^{-1}-}_{L^2}$.

\begin{lemm}\label{lemm-sobo-decomp}
  Let~$A\subset M$ be a closed set. Then~$W^{1}(M)$ and~$W^{-1}(M)$ decompose orthogonally as
  \begin{equation}
  \begin{tikzcd}
    W^1(M) \ar[d,"(\Delta+m^2)"',"\sim"] \ar[rr, "p_{M\setminus A}^{\perp}" near start, bend left=60] \ar[rrrr, "p_{M\setminus A}" near end, bend left=30] &[-30pt] = &[-30pt] W^1_{M\setminus A}(M)^{\perp} \ar[d,"(\Delta+m^2)"',"\sim"] &[-30pt] \oplus &[-30pt] W^1_{M\setminus A}(M)\ar[d,"(\Delta+m^2)"',"\sim"]\\ [+20pt]
    W^{-1}(M) \ar[rr, "P_{ A}"' near start, bend right=60] \ar[rrrr, "P_{A}^{\perp}"' near end, bend right=30] &[-30pt] = &[-30pt] W^{-1}_{ A}(M) & \oplus & W^{-1}_{ A}(M)^{\perp} ,
  \end{tikzcd}
  \label{eqn-sobo-decomp}
\end{equation}
which is \textsf{preserved} by the isometric isomorphism~$\Delta+m^2$. In particular, we have
\begin{equation}
  P_{ A}(\Delta+m^2)=(\Delta+m^2)p_{M\setminus A}^{\perp},\quad\textrm{and}\quad P_{ A}^{\perp}(\Delta+m^2)=(\Delta+m^2)p_{M\setminus A},
  \label{eqn-sobo-proj-comm}
\end{equation}
where~$p_{M\setminus A}^{\perp}$,~$p_{M\setminus A}$,~$P_{ A}$, and~$P_{ A}^{\perp}$ are the corresponding orthogonal projections as indicated in the diagram.
\end{lemm}

\begin{proof}
  We just need to show that the image of~$W^1_{M\setminus A}(M)$ under~$\Delta+m^2$ is precisely~$W^{-1}_{ A}(M)^{\perp}$. Indeed, by our definition of the inner products we have
  \begin{equation}
    \sank{(\Delta+m^2)u,v}_{W^{-1}(M)}=\ank{u,v}_{L^2(M)}
    \label{eqn-equal-sobo-distri-pair}
  \end{equation}
  for all~$u\in W^1(M)$,~$v\in W^{-1}(M)$, where the RHS denotes the duality (distributional) pairing. However, the \textsf{annihilator} of~$W^{-1}_{ A}(M)$ under the duality pairing is exactly~$W^1_{M\setminus A}(M)$, see lemma \ref{lemm-sobo-dual}. This translates as
  \begin{equation}
    \big\{u\in W^1(M)~\big|~\sank{(\Delta+m^2)u,v}_{W^{-1}(M)}=0\textrm{ for all }v\in W^{-1}_{ A}(M)\big\}=W^1_{M\setminus A}(M),    \label{}
  \end{equation}
  which is what we desired.
\end{proof}

\begin{corr}
  (\ref{eqn-sobo-decomp}) and (\ref{eqn-sobo-proj-comm}) holds for~$A=\Sigma\subset M$ an embedded closed hypersurface (codimension one submanifold). \hfill~$\Box$
\end{corr}

Indeed, the crucial relation (\ref{eqn-equal-sobo-distri-pair}) together with (\ref{eqn-sobo-proj-comm}) gives

\begin{corr}
  [adjoints] \label{cor-adj-sobo-proj} We have
  \begin{equation}
    \bank{p_{M\setminus A}^{\perp}u,v}_{L^2}=\bank{u,P_A v}_{L^2},\quad\textrm{and}\quad \bank{p_{M\setminus A} u,v}_{L^2}=\bank{u,P_A^{\perp} v}_{L^2},
    \label{}
  \end{equation}
  for~$u\in W^1(M)$ and~$v\in W^{-1}(M)$. \hfill~$\Box$
\end{corr}

\begin{def7}
\label{rem-sobo-proj=dist-res}
Denote~$U:=M\setminus A$. Recall from (\ref{eqn-sobo-open-set-ortho}) that~$W^{-1}(U):=W^{-1}_{ A}(M)^{\perp}$. In fact, corollary \ref{cor-adj-sobo-proj} shows~$P_{M\setminus U}^{\perp}$ coincides with the restriction map~$\rho_{M|U}:\mathcal{D}'(M)\lto \mathcal{D}'(U)$, since if already~$u\in W^1_U(M)$ then~$p_U u=i_*u$,~$i_*:W^1_U(M)\lto W^1(M)$ being inclusion.
\end{def7}

We are thus lead naturally to the following and eventually corollary \ref{cor-sobo-decomp-boundary}.

\begin{corr}\label{cor-sobo-decomp-open}
  Let~$U\subset M$ be an open set and~$F\subset U$ be a closed set. Then~$W^{1}_U(M)$ and~$W^{-1}(U)$ decompose further as
  \begin{equation}
  \begin{tikzcd}
    W^1_U(M) \ar[d,"(\Delta+m^2)"',"\sim"] &[-30pt] = &[-30pt] W^1_{U\setminus F}(M)^{\perp} \ar[d,"(\Delta+m^2)"',"\sim"] &[-30pt] \oplus &[-30pt] W^1_{U\setminus F}(M)\ar[d,"(\Delta+m^2)"',"\sim"]\\ [+20pt]
    W^{-1}(U)&[-30pt] = &[-30pt] P_{M\setminus U}^{\perp}(W^{-1}_{ F}(M)) & \oplus & W^{-1}_{ F}(M)^{\perp} ,
  \end{tikzcd}
\end{equation}
where the orthogonal complements are taken respectively \textsf{inside}~$W^1_U(M)$ and~$W^{-1}(U)$, which is \textsf{preserved} by the isometric isomorphism~$\Delta+m^2$. Also, we have commutation relations similar to (\ref{eqn-sobo-proj-comm}).
\end{corr}

\begin{proof}
  The only point needing explanation is
  \begin{equation}
    (\Delta+m^2)( W^1_{U\setminus F}(M)^{\perp}\cap W_U^1(M))=P_{M\setminus U}^{\perp}(W^{-1}_{ F}(M)).
    \label{}
  \end{equation}
  Indeed, since~$\Delta+m^2$ is bijective, by (\ref{eqn-sobo-decomp}) we have
  \begin{align*}
    \textrm{LHS}&=W^{-1}_{F\cup(M\setminus U)}(M)\cap W_{M\setminus U}^{-1}(M)^{\perp}\\
    &=(W^{-1}_{M\setminus U}(M)\oplus W^{-1}_{F}(M)) \cap W_{M\setminus U}^{-1}(M)^{\perp} \tag{$F$ and~$M\setminus U$ disjoint}\\
    &=\textrm{RHS},
  \end{align*}
  where the direct sum in the second line is non-orthogonal.
\end{proof}

 Now let~$(\Omega,g)$ be compact with (smooth) boundary, smoothly and isometrically embedded in closed~$M$ (if $\Omega$ is a \textit{real tunnelling geometry}, then one choice for~$M$ is the ``isometric double'', see remark \ref{rem-real-tun-geo}) and~$\Sigma$ be embedded in~$ \Omega^{\circ}$. Then in particular

 \begin{corr}\label{cor-sobo-decomp-boundary}
  Corollary \ref{cor-sobo-decomp-open} holds for~$U=\Omega^{\circ}$ and~$\Sigma\subset \Omega^{\circ}$ an embedded closed hypersurface. \hfill~$\Box$
\end{corr}

 \subsection{The Markov Stochastic Decomposition of GFF}\label{sec-stoc-decomp-gff}
 
 \noindent Again suppose~$(M,g)$ is a closed Riemannian manifold and~$\Sigma\subset M$ an embedded closed hypersurface with induced metric. A probabilistic point of view of the results in this section is provided in Powell and Werner \cite{powell-werner} section 4.1.

 \begin{lemm}\label{lemm-sobo-proj-trace-harmonic}
  We have
  \begin{equation}
    p_{M\setminus\Sigma}^{\perp}=\PI_M^{\Sigma}\circ \tau_{\Sigma}
    \label{}
\end{equation}
on~$W^1(M)$, where $p_{M\setminus\Sigma}^{\perp}:W^1(M)\lto W^1_{M\setminus A}(M)^{\perp}$ is the orthogonal projection as in lemma \ref{lemm-sobo-decomp}.
\end{lemm}

\begin{proof}
  We show that for~$f\in C^{\infty}(M)$,~$p_{M\setminus\Sigma}^{\perp} f$ solves the boundary value problem
  \begin{equation}
    \left\{
    \begin{array}{ll}
      (\Delta+m^2)u=0,&\textrm{in }M\setminus\Sigma,\\
      u|_{\Sigma}=\tau_{\Sigma}f,&\textrm{on }\Sigma.
    \end{array}
    \right.
    \label{}
  \end{equation}
  Indeed, the first condition holds since~$\supp((\Delta+m^2)p_{M\setminus\Sigma}^{\perp} f)\subset \Sigma$ by lemma \ref{lemm-sobo-decomp}. To show the second, notice~$\tau_{\Sigma}=0$ on~$W_{M\setminus\Sigma}^1(M)$ by its definition (\ref{eqn-def-sobo-open-closure}). Hence
  \begin{equation}
    \tau_{\Sigma}p_{M\setminus\Sigma}^{\perp}=\tau_{\Sigma}(\one-p_{M\setminus\Sigma})=\tau_{\Sigma},
    \label{}
  \end{equation}
  on~$W^1(M)$. We conclude the proof.
\end{proof}

\begin{prop}\label{prop-stoc-decomp-closed}
  Suppose~$\phi\in \mathcal{D}'(M)$ follows~$\mu_{\mm{GFF}}^M$. Then there is a stochastic decomposition
  \begin{equation}
    \phi=\wh{p_{M\setminus\Sigma}^{\perp}}\phi+\wh{p_{M\setminus\Sigma}}\phi\defeq \phi_{\Sigma}+\phi_{M\setminus\Sigma}^D,
    \label{eqn-stoc-decomp-gff-closed}
  \end{equation}
  into independent random fields~$\phi_{\Sigma}$ and~$\phi_{M\setminus\Sigma}^D$. More precisely, for~$f\in C^{\infty}(M)$ we define the random variables
  \begin{equation}
    \phi_{\Sigma}(f)\defeq \phi(P_{\Sigma}f),\quad\textrm{and}\quad \phi_{M\setminus\Sigma}^D(f)\defeq \phi(P_{\Sigma}^{\perp}f).
    \label{eqn-def-markov-decomp-rv}
  \end{equation}
  Then,~$\phi_{M\setminus\Sigma}^D$ follows the law of~$\mu_{\mm{GFF}}^{M\setminus\Sigma,D}$, while~$\phi_{\Sigma}$ solves the boundary value problem
  \begin{equation}
    \left\{
    \begin{array}{ll}
      (\Delta+m^2)\phi_{\Sigma}=0,&\textrm{in }M\setminus\Sigma,\\
      \phi_{\Sigma}|_{\Sigma}=\tau_{\Sigma}\phi,&\textrm{on }\Sigma,
    \end{array}
    \right.
    \label{}
  \end{equation}
  almost surely.
\end{prop}

\begin{def7}
  The term ``stochastic decomposition'' is borrowed from Bogachev \cite{Bogachev} remark 3.7.7.
\end{def7}

\begin{proof}
  Firstly,~$\phi_{\Sigma}$ and~$\phi_{M\setminus\Sigma}^D$ are independent since by (\ref{eqn-def-markov-decomp-rv}) their Gaussian Hilbert spaces are respectively~$W_{\Sigma}^{-1}(M)$ and~$W_{\Sigma}^{-1}(M)^{\perp}=W^{-1}(M\setminus\Sigma)$. The latter also means~$\phi_{M\setminus\Sigma}^D$ follows~$\mu_{\mm{GFF}}^{M\setminus\Sigma,D}$. The last fact about~$\phi_{\Sigma}$ follows from lemma \ref{lemm-sobo-proj-trace-harmonic} and the uniqueness of the measurable linear extension (\cite{Bogachev} theorem 3.7.6).
\end{proof}

In the same vein using corollary \ref{cor-sobo-decomp-boundary}, we have also a decomposition for the Dirichlet GFF on a domain~$\Omega$ with smooth boundary, with respect to a hypersurface~$\Sigma$ (isometrically) embedded in the interior~$\Omega^{\circ}$.

\begin{prop}\label{prop-stoc-decomp-domain}
  Suppose~$\phi_{\Omega}^D\in \mathcal{D}'(\Omega^{\circ})$ follows~$\mu_{\mm{GFF}}^{\Omega,D}$. Then there is a stochastic decomposition
  \begin{equation}
    \phi=\wh{p_{\Omega^{\circ}\setminus\Sigma}^{\perp}}\phi+\wh{p_{\Omega^{\circ}\setminus\Sigma}}\phi\defeq (\phi_{\Omega}^D)_{\Sigma}+\phi_{\Omega\setminus\Sigma}^D,
    \label{}
  \end{equation}
  into independent random fields~$(\phi_{\Omega}^D)_{\Sigma}$ and~$\phi_{\Omega\setminus\Sigma}^D$. More precisely, for~$f\in W^{-1}(\Omega^{\circ})$ we define the random variables
  \begin{equation}
    (\phi_{\Omega}^D)_{\Sigma}(f)\defeq \phi_{\Omega}^D(P_{\Sigma}f),\quad\textrm{and}\quad \phi_{\Omega\setminus\Sigma}^D(f)\defeq \phi_{\Omega}^D(P_{\Sigma}^{\perp}f).
    \label{eqn-def-markov-decomp-rv-domain-diri}
  \end{equation}
  Then,~$\phi_{\Omega\setminus\Sigma}^D$ follows the law of~$\mu_{\mm{GFF}}^{\Omega\setminus\Sigma,D}$, while~$(\phi_{\Omega}^D)_{\Sigma}$ solves the boundary value problem
  \begin{equation}
    \left\{
    \begin{array}{ll}
      (\Delta+m^2)(\phi_{\Omega}^D)_{\Sigma}=0,&\textrm{in }\Omega^{\circ}\setminus\Sigma,\\
      (\phi_{\Omega}^D)_{\Sigma}|_{\Sigma}=\tau_{\Sigma}\phi_{\Omega}^D,&\textrm{on }\Sigma,\\
      (\phi_{\Omega}^D)_{\Sigma}|_{\partial\Omega}=0, &\textrm{on }\partial\Omega,
    \end{array}
    \right.
    \label{}
  \end{equation}
  almost surely. \hfill~$\Box$
\end{prop}

 The following version of proposition \ref{prop-stoc-decomp-closed} in the case of a dissecting (smooth isometrically embedded) hypersurface~$\Sigma\subset M$ such that~$M\setminus \Sigma=M^{\circ}_+\sqcup M^{\circ}_-$, will be useful in definition \ref{def-inter-over-domain}.

 \begin{lemm}
    We have~$W^{-1}_{M_-}(M)^{\perp}\subset W^{-1}_{M_+}(M)$ and thus~$W^{-1}_{M_+}(M)=W^{-1}_{M_-}(M)^{\perp}\oplus W^{-1}_{\Sigma}(M)$. Similarly~$W^{-1}_{M_-}(M)=W^{-1}_{M_+}(M)^{\perp}\oplus W^{-1}_{\Sigma}(M)$.
  \end{lemm}

  \begin{proof}
    By lemma \ref{lemm-sobo-decomp} with~$A=M_-$ we see that~$W^{-1}_{M_-}(M)^{\perp}=(\Delta+m^2)(W^1_{M_+^{\circ}}(M))$. These distributions are supported in~$M_+$ since~$\Delta+m^2$ is local.
  \end{proof}

\begin{corr}\label{cor-stoc-decomp-disec}
  Suppose~$\phi\in \mathcal{D}'(M)$ follows~$\mu_{\mm{GFF}}^M$. Then there is a stochastic decomposition
 \begin{equation}
   \phi=\wh{p_{M\setminus\Sigma}^{\perp}}\phi+\wh{p_{M_+^{\circ}}}\phi +\wh{p_{M_-^{\circ}}}\phi\defeq \phi_{\Sigma}+\phi_{M_+^{\circ}}^D +\phi_{M_-^{\circ}}^D,
    \label{}
  \end{equation} 
  into independent random fields. More precisely, for~$f\in C^{\infty}(M)$ the random variables are defined by
  \begin{equation}
    \phi_{\Sigma}(f)\defeq \phi(P_{\Sigma}f),\quad \phi_{M_+^{\circ}}^D(f)\defeq \phi(P_{M_-}^{\perp}f), \quad\textrm{and}\quad \phi_{M_-^{\circ}}^D(f)\defeq \phi(P_{M_+}^{\perp}f),
  \end{equation}
  and with~$\phi_{\Sigma}$ the same as in proposition \ref{prop-stoc-decomp-closed},~$\phi_{M_+^{\circ}}^D$ and~$\phi_{M_-^{\circ}}^D$ follows respectively~$\mu_{\mm{GFF}}^{M_+^{\circ},D}$ and~$\mu_{\mm{GFF}}^{M_-^{\circ},D}$. Moreover,~$(\phi_{\Sigma}+\phi_{M_+^{\circ}}^D)(f)=\phi(P_{M_+}f)$,~$(\phi_{\Sigma}+\phi_{M_-^{\circ}}^D)(f)=\phi(P_{M_-}f)$. \hfill~$\Box$
\end{corr}

\subsection{The Bayes Principle Applied to GFF}\label{sec-bayes-gff}

\noindent Let now~$M$ be a closed Riemannian manifold and~$\Sigma_1$,~$\Sigma_2\subset M$ \textit{non-intersecting} isometrically embedded smooth closed hypersurfaces. The goal of this section is to derive the Bayes principle relating the probability laws of the two random fields~$\tau_{\Sigma_1} \phi$ and~$\tau_{\Sigma_2}\phi$ where~$\phi$ is the GFF on~$M$. To avoid convolving with nomenclatures of conditional probabilities we prefer a direct measure theoretic argument, although these are clearly equivalent. Throughout this section we identify continuous linear maps of Cameron-Martin spaces with their measurable extensions to distributional~$Q$-spaces, as well as their induced actions on measures.

To begin with, by proposition \ref{prop-stoc-decomp-closed} we have a stochastic decomposition
\begin{equation}
  \phi=\phi_{\Sigma_1}+\phi_{M\setminus\Sigma_1}^D
  \label{}
\end{equation}
for~$\phi\sim \mu_{\mm{GFF}}^M$, the two components \textit{independent} of each other. Then, apply~$\tau_{\Sigma_2}$ we get
\begin{align}
  \tau_{\Sigma_2}\phi&=\tau_{\Sigma_2}\phi_{\Sigma_1}+\tau_{\Sigma_2}\phi_{M\setminus \Sigma_1}^D=\tau_{\Sigma_2}\PI_{M}^{\Sigma_1}\tau_{\Sigma_1}\phi+\tau_{\Sigma_2}\phi_{M\setminus \Sigma_1}^D \\
  &\defeq \mathcal{M}_{M,2}^1 \tau_{\Sigma_1}\phi +\tau_{\Sigma_2}\phi_{M\setminus \Sigma_1}^D.
  \label{eqn-stoc-decomp-with-M}
\end{align}
Here we define
\begin{equation}
  \mathcal{M}_{M,2}^1\defeq \tau_{\Sigma_2}\PI_{M}^{\Sigma_1}:\left\{
    \begin{array}{rcl}
      C^{\infty}(\Sigma_1)&\lto& C^{\infty}(\Sigma_2),\\
      W^{\frac{1}{2}}(\Sigma_1) &\lto &W^{\frac{1}{2}}(\Sigma_2),
    \end{array}
    \right.
  \label{eqn-trans-op-def}
\end{equation}
called the \textsf{transition operator/propagator}. Define accordingly
\begin{equation}
  \left.
  \begin{array}{rcl}
    \mathcal{G}_{M,2}^1:W^{\frac{1}{2}}(\Sigma_1)\times W^{\frac{1}{2}}(\Sigma_2) &\lto& W^{\frac{1}{2}}(\Sigma_1)\times W^{\frac{1}{2}}(\Sigma_2),\\
    (\varphi_1,h)&\longmapsto & (\varphi_1,h+\mathcal{M}_{M,2}^1\varphi_1).
  \end{array}
  \right.
  \label{}
\end{equation}
called the \textsf{graph operator}.
\begin{lemm}\label{lemm-pre-bayes}
  We have
  \begin{equation}
    \tau_{\Sigma_1\sqcup \Sigma_2}(\mu_{\mm{GFF}}^{M})=(\mathcal{G}_{M,2}^1)_*( \mu_{\DN}^{\Sigma_1,M} \otimes \mu_{\DN,D}^{\Sigma_2,M\setminus \Sigma_1}),
    \label{}
  \end{equation}
  where~$\mu_{\DN}^{\Sigma_1,M}$ and~$\mu_{\DN,D}^{\Sigma_2,M\setminus \Sigma_1}$ are Gaussian measures on $\mathcal{D}'(\Sigma_1)$ and $\mathcal{D}'(\Sigma_2)$ with covariances $(\DN_M^{\Sigma_1})^{-1}$ and $(\DN_{M\setminus\Sigma_1}^{\Sigma_2,D})^{-1}$, respectively.
\end{lemm}

\begin{proof}
  Following (\ref{eqn-stoc-decomp-with-M}),~$\tau_{\Sigma_2}\phi_{M\setminus \Sigma_1}^D=\tau_{\Sigma_2}\phi-\mathcal{M}_{M,2}^1 \tau_{\Sigma_1}\phi$ is independent of~$\tau_{\Sigma_1}\phi$, and follows the law~$\mu_{\DN,D}^{\Sigma_2,M\setminus \Sigma_1}$, while~$\tau_{\Sigma_1}\phi$ follows the law~$\mu_{\DN}^{\Sigma_1,M}$, both by proposition \ref{prop-induce-law-DN}.
\end{proof}

We could also define the operators~$\mathcal{M}_{M,1}^2$ and~$\mathcal{G}_{M,1}^2$ with the roles of~$\Sigma_1$ and~$\Sigma_2$ switched. Note that lemma \ref{lemm-pre-bayes} is entirely symmetric under the switching of~$\Sigma_1$ and~$\Sigma_2$. Recall the measures~$\mu_{2\mn{D}}^{\Sigma}$ from corollary \ref{corr-rad-niko-dense}.

  \begin{prop}[Bayes Principle for GFF]\label{prop-bayes-gff}
    Let now~$M$ be a closed Riemannian manifold and~$\Sigma_1$,~$\Sigma_2\subset M$ \textit{non-intersecting} isometrically embedded smooth closed hypersurfaces. We have equality of Radon-Nikodym densities
    \begin{align}
      \frac{\dd\tau_{\Sigma_1\sqcup \Sigma_2}(\mu_{\mm{GFF}}^{M}) }{\dd \mu_{2\mn{D}}^{\Sigma_1\sqcup \Sigma_2}}(\varphi_1,\varphi_2)
  &=\frac{\dd  \mu_{\DN}^{\Sigma_1,M}}{\dd\mu_{2\mn{D}}^{\Sigma_1}} (\varphi_1)
  \frac{\dd \big((\mathcal{M}_{M,2}^1 \varphi_1)_*\mu_{\DN,D}^{\Sigma_2,M\setminus \Sigma_1}\big)}{\dd\mu_{2\mn{D}}^{ \Sigma_2}}(\varphi_2) \label{eqn-bayes-gff-1}\\
  &=\frac{\dd  \mu_{\DN}^{\Sigma_2,M}}{\dd\mu_{2\mn{D}}^{\Sigma_2}} (\varphi_2)
  \frac{\dd \big((\mathcal{M}_{M,1}^2 \varphi_2)_*\mu_{\DN,D}^{\Sigma_1,M\setminus \Sigma_2}\big)}{\dd\mu_{2\mn{D}}^{ \Sigma_1}}(\varphi_1). \label{eqn-bayes-gff-2}
    \end{align}
    Here~$(\varphi)_*$ denotes the shift induced by~$\varphi$ on measures defined by $((\varphi)_*\mu)(A):=\mu(A-\varphi)$.
  \end{prop}

  \begin{proof}
    We just need to prove (\ref{eqn-bayes-gff-1}). In other words,~$(\mathcal{M}_{M,2}^1 \varphi_1)_*\mu_{\DN,D}^{\Sigma_2,M\setminus \Sigma_1}$ is the conditional law of~$\tau_{\Sigma_2}\phi$ provided ``$\tau_{\Sigma_1}=\varphi_1$''. The proof is straightforward. For positive (Borel) measurable functionals~$F\in L^+(\mathcal{D}'(\Sigma_1))$,~$G\in L^+(\mathcal{D}'(\Sigma_2))$, we have by lemma \ref{lemm-pre-bayes}, the change of variables formula, and Fubini's theorem,
\begin{align*}
    \int_{}^{} (F\otimes G)(\varphi_1,\varphi_2)\dd\tau_{\Sigma_1\sqcup \Sigma_2}(\mu_{\mm{GFF}}^{M})&= \iint_{}^{}(F\otimes G)\circ \mathcal{G}_{M,2}^1(\varphi_1,\varphi_2)\dd\mu_{\DN}^{\Sigma_1,M} \otimes \dd \mu_{\DN,D}^{\Sigma_2,M\setminus \Sigma_1} \\
    &=\int_{}^{}F(\varphi_1)\dd  \mu_{\DN}^{\Sigma_1,M}(\varphi_1)\int G(\varphi_2+\mathcal{M}_{M,2}^1 \varphi_1)\dd \mu_{\DN,D}^{\Sigma_2,M\setminus \Sigma_1}(\varphi_2)  \\
  &=\int_{}^{}F(\varphi_1)\dd  \mu_{\DN}^{\Sigma_1,M}\int_{}^{}G(\varphi_2)\dd(\mathcal{M}_{M,2}^1 \varphi_1)_*\mu_{\DN,D}^{\Sigma_2,M\setminus \Sigma_1}. 
\end{align*}
as~$F$ and~$G$ range over all positive measurable functionals, we obtain the result.
  \end{proof}

\begin{def7}
  Applying corollary \ref{corr-rad-niko-dense} and the Cameron-Martin-Girsanov formula (\cite{Bogachev} corollary 2.4.3), one obtains a relation between quadratic forms (the corresponding relation for determinants also works out by BFK),
  \begin{align*}
    &\quad \Bank{\bnom{\varphi_1}{x},\DN_{M}^{\Sigma_1\sqcup\Sigma_2}\bnom{\varphi_1}{x}}_{L^2(\Sigma_1\sqcup\Sigma_2)}\\
    =&\quad \ank{\varphi_1,\DN_M^{\Sigma_1}\varphi_1}_{L^2(\Sigma_1)}+\ank{x,\DN_{M\setminus\Sigma_1}^{\Sigma_2,D}x}_{L^2(\Sigma_2)}\\
    &\quad -2
    \ank{x,\DN_{M\setminus\Sigma_1}^{\Sigma_2,D}\mathcal{M}_{M,2}^1\varphi_1}_{L^2(\Sigma_2)}+\ank{\mathcal{M}_{M,2}^1\varphi_1,\DN_{M\setminus\Sigma_1}^{\Sigma_2,D}\mathcal{M}_{M,2}^1\varphi_1}_{L^2(\Sigma_2)}
  \end{align*}
  which, of course, can also be obtained directly noting
  \begin{align*}
    \ank{\varphi_1,\DN_M^{\Sigma_1}\varphi_1}_{L^2(\Sigma_1)}&=\Bank{\bnom{\varphi_1}{\mathcal{M}_{M,2}^1\varphi_1},\DN_{M}^{\Sigma_1\sqcup\Sigma_2}\bnom{\varphi_1}{\mathcal{M}_{M,2}^1\varphi_1}}_{L^2(\Sigma_1\sqcup\Sigma_2)},\\
\ank{x,\DN_{M\setminus\Sigma_1}^{\Sigma_2,D}x}_{L^2(\Sigma_2)}&=\Bank{\bnom{0}{x},\DN_{M}^{\Sigma_1\sqcup\Sigma_2}\bnom{0}{x}}_{L^2(\Sigma_1\sqcup\Sigma_2)},
  \end{align*}
  et cetera, and expanding the left hand side by writing~$[\smx{\varphi_1\\x}]=[\smx{\varphi_1\\ \mathcal{M}_{M,2}^1\varphi_1}]+[\smx{0\\ x-\mathcal{M}_{M,2}^1\varphi_1}]$.
\end{def7}

\begin{def7}
  In the case~$M=\Sigma\times \mb{R}$ with~$\Sigma_1:=\Sigma\times \{0\}$,~$\Sigma_2:=\Sigma\times \{t\}$, and Dirichlet condition at~$\Sigma\times \{\pm \infty\}$ (namely decay at~$\infty$), we have the explicit expression
  \begin{equation}
    \mathcal{M}^1_{M,2}=\me^{-t\mn{D}_{\Sigma}},
    \label{}
  \end{equation}
  where~$\mn{D}_{\Sigma}=(\Delta_{\Sigma}+m^2)^{1/2}$. Also, in this case~$\mn{D}_{\Sigma}=\DN_{M_+}^{\Sigma}$,~$M_+=\Sigma\times [0,+\infty)$.
\end{def7}

\subsection{Locality of the $P(\phi)$ interaction and the Restricted GFF}\label{sec-locality}

In this subsection we pronounce what is meant by the \textsf{locality} of the $P(\phi)$-interaction and prove it. Let~$M$ be a closed Riemannian surface and~$\Omega^{\circ}\subset M$ a domain with smooth boundary~$\partial\Omega$. Denote~$\Omega=\ol{\Omega^{\circ}}$. The formulation will be based on the domain $\sigma$-algebras defined by the GFF random varibales $\ank{-,f}_{L^2}$, $f\in C^{\infty}(M)$. That is, for~$A\subset M$ a closed set, we define the~$\sigma$-algebra
\begin{equation}
  \mathcal{O}_A\defeq \sigma\left( \phi(f)~|~f\in W_A^{-1}(M) \right).
  \label{}
\end{equation}

\begin{def7}
  In general one could consider \textsf{regular} domains~$\Omega^{\circ}$. A domain is called \textsf{regular} if~$W^{-1}_{\Omega^{\circ}}(M)=W^{-1}_{\Omega}(M)$, the two spaces defined by (\ref{eqn-def-sobo-open-closure}) and (\ref{eqn-def-sobo-closed-support}). See also Simon \cite{Sim2} page 267.
\end{def7}
Suppose~$\chi \in C^{\infty}(M)$. We define
\begin{equation}
  S_{\Omega,\varepsilon,\chi}(\phi)\defeq \int_{M}^{}\chi(x) {:}P(\phi_{\Omega,\varepsilon}^{D}(x)+\phi_{\partial\Omega,\varepsilon}(x)){:}\dd V_M(x),
  \label{}
\end{equation}
where
\begin{equation}
  \phi_{\Omega,\varepsilon}^{D}(x)\defeq \phi(P_{\ol{\Omega^c}}^{\perp} E_{\varepsilon}\delta_x),\quad \phi_{\partial\Omega,\varepsilon}(x)\defeq\phi(P_{\partial\Omega}E_{\varepsilon}\delta_x),\quad \phi\sim \mu_{\mm{GFF}}^M.
  \label{}
\end{equation}

\begin{prop}[locality]\label{prop-locality}
  Whenever~$\chi\in C_c^{\infty}(\Omega^{\circ})$, we have, for any fixed Wick ordering~${:}\bullet{:}$,
  \begin{equation}
    S_{M,\chi}(\phi)=\int_{M}^{}\chi(x) {:}P(\phi(x)){:}\dd V_M(x)=\lim_{\varepsilon\to 0}S_{\Omega,\varepsilon,\chi}(\phi),
    \label{eqn-decoupling-eqv-rv}
  \end{equation}
  and hence~$S_{M,\chi}$ is~$\mathcal{O}_{\Omega}$-measurable. In particular,
  \begin{equation}
    \int_{\Omega}^{}{:}P(\phi(x)){:}\dd V_{\Omega}= \int_{\Omega}^{} {:}P(\phi^{D}_{\Omega}(x)+\phi_{\partial\Omega}(x)){:}\dd V_{\Omega}
    \label{}
  \end{equation}
  and is~$\mathcal{O}_{\Omega}$-measurable.
\end{prop}

\begin{proof}
  Without loss of generality we suppose~$\supp \chi$ stays~$\delta$-away from~$\partial\Omega$ for some~$\delta>0$. Then corollary \ref{cor-nelson-all-Lp-cutoff-limit} and approximation in~$L^4$ of a general~$\chi$ as well as~$1_{\Omega}$ gives the result. We note that the support of~$E_{\varepsilon}\delta_x$ is contained in the~$\varepsilon$-ball around~$x$. Then whenever~$\varepsilon<\delta$ for~$x\in \supp\chi$ we have~$\supp(E_{\varepsilon}\delta_x)\subset \Omega$, hence
  \begin{equation}
    P_{\Omega}^{\perp}E_{\varepsilon}\delta_x =0,
    \label{}
  \end{equation}
  and
  \begin{equation}
    \phi_{\varepsilon}(x)\equiv \phi(P_{\ol{\Omega^c}}^{\perp} E_{\varepsilon}\delta_x) +\phi(P_{\partial\Omega}E_{\varepsilon}\delta_x),
    \label{}
  \end{equation}
  both~$\mathcal{O}_{\Omega}$-measurable. Thus
  \begin{equation}
    {:}P(\phi_{\Omega,\varepsilon}^{D}(x)+\phi_{\partial\Omega,\varepsilon}(x)){:}\equiv {:}P(\phi_{\varepsilon}(x)){:}
    \label{}
  \end{equation}
  for~$x\in \supp \chi$ and~$\varepsilon<\delta$, and is~$\mathcal{O}_{\Omega}$-measurable. Therefore (\ref{eqn-decoupling-eqv-rv}) holds by proposition \ref{prop-nelson-main}.
\end{proof}

We emphasize here that for our purposes we employ the Wick ordering~${:}\bullet{:}_0$ of section \ref{sec-change-wick} in defining~$S_{M,\chi}$. This is to say
\begin{equation}
  {:}\phi_{\varepsilon}(x)^{2n}{:}=\sum_{j=0}^{n} \frac{(-1)^j (2n)!}{(2n-2j)! j! 2^j}C_{\varepsilon}(x)^j \phi_{\varepsilon}(x)^{2n-2j},
  \label{}
\end{equation}
where
\begin{equation}
  C_{\varepsilon}(x)=\iint E_{\varepsilon}(x,y)C_0(y,z)E_{\varepsilon}(z,x)\dd V_M(y)\dd V_M(z).
  \label{}
\end{equation}
Since~$E_{\varepsilon}(x,-)$ is supported in an~$\varepsilon$-ball around~$x$, and~$C_0(y,z)$ depends only on the geometry of~$M$ resctricted to a convex neighborhood of~$y$,~$z$, the term~$C_{\varepsilon}(x)$ depends only on the geometry of~$M$ locally near~$x$. This means that under~${:}\bullet{:}_0$, once~$\supp \chi\subset \Omega$, the limiting (integrated) random variable~$S_{M,\chi}$, in addition to being~$\mathcal{O}_{\Omega}$-measurable, is in fact fully determined with knowledge of the metric~$g|_{\Omega}$ restricted to~$\Omega$. This allows the freedom of choosing the ambient manifold~$M$ where~$\Omega$ isometrically embeds in defining the interaction over~$\Omega$ (see definition \ref{def-inter-over-domain}).

Now we repeat proposition \ref{prop-locality} in a different way:
\begin{corr}\label{cor-interact-locality}
   Whenever~$\chi\in C_c^{\infty}(\Omega^{\circ})$, for any fixed Wick ordering~${:}\bullet{:}$,
   \begin{equation}
     \mb{E}_{\mm{GFF}}^M[S_{M,\chi}|\mathcal{O}_{\Omega}]=S_{M,\chi},
     \label{}
   \end{equation}
   and~$S_{M,\chi}$ is independent of~$\phi_{\Omega^c}^D$, that is, of~$\phi(P_{\Omega}^{\perp} f)$ for all~$f\in C^{\infty}(M)$ or~$\phi((\Delta+m^2)f)$ for all~$f\in C_c^{\infty}(\Omega^c)$. Moreover, since~$\Omega$ is regular,~$S_{M,\chi}$ is the limit in~$L^2(\mu_{\mm{GFF}}^M)$ of polynomials of random variables of the form~$\phi(f)$ with~$f\in C_c^{\infty}(\Omega^{\circ})$ via the It\^o-Wiener-Segal isomorphism.\hfill~$\Box$
\end{corr}

This result motivates

\begin{deef}\label{def-res-gff}
  Let~$\Omega^{\circ}\subset M$ be a regular domain. The \textsf{(massive) Gaussian Free Field} over~$\Omega^{\circ}$ \textsf{restricted from~$M$} is the centered Gaussian process indexed by~$C_c^{\infty}(\Omega^{\circ})$, with covariance
  \begin{equation}
    \mb{E}[\phi(f)\phi(h)]=\ank{f,h}_{W^{-1}_{\Omega}(M)}
    \label{}
  \end{equation}
  for any~$f$,~$h\in C_c^{\infty}(\Omega^{\circ})$. We denote it by $\phi|_{\Omega}$.
\end{deef}

Since the inner product of~$W^{-1}_{\Omega}(M)$ depends on~$M$, the restricted GFF does not make sense on~$\Omega^{\circ}$ alone. However, by corollary \ref{cor-stoc-decomp-disec}, it is equal in law to~$\phi_{\partial\Omega}+\phi_{\Omega}^D$, where~$\phi_{\Omega}^D$ does make sense over~$\Omega^{\circ}$ and the law of~$\phi_{\partial\Omega}=\PI_{\Omega}^{\partial\Omega}(\tau_{\partial\Omega}\phi)$ is determined via~$\PI_{\Omega}^{\partial\Omega}$ by that of a boundary data over~$\partial\Omega$.

Going back to corollary \ref{cor-interact-locality} we see that for~$\chi\in C_c^{\infty}(\Omega^{\circ})$, the interaction~$S_{M,\chi}$ can now be defined as a random variable of the sample~$\phi|_{\Omega}\in\mathcal{D}'(\Omega^{\circ})$. Equally, it is also a random variable over~$\mathcal{D}'(\Omega^{\circ})\times \mathcal{D}'(\partial\Omega)$ equipped with~$\mu_{\mm{GFF}}^{\Omega,D}\otimes \mu_C^{\partial\Omega}$ where~$\mu_C^{\partial\Omega}$ is some probability measure on~$\mathcal{D}'(\partial\Omega)$ mutually absolutely continuous with respect to~$\tau_{\partial\Omega}(\mu_{\mm{GFF}}^M)$ and on the same~$\sigma$-algebra. We consider only a few very specific candidates for~$\mu_C^{\partial\Omega}$.

Now suppose we start with a Riemannian surface~$(\Omega,g)$ with totally geodesic boundary~$\partial\Omega$. With the help of the isometric double~$|\Omega|^2$ discussed in remark \ref{rem-real-tun-geo} we define
\begin{deef}\label{def-inter-over-domain}
  The \textsf{interaction over~$\Omega$} is
  \begin{equation}
    S_{\Omega}(\phi_{\Omega}^D|\varphi)=S_{\Omega}(\phi_{\Omega}^D+\PI_{\Omega}^{\partial\Omega}\varphi)\defeq S_{|\Omega|^2,1_{\Omega}}(\phi_{\Omega}^D+\PI_{\Omega}^{\partial\Omega}\varphi)
    \label{}
  \end{equation}
  as an~$L^2$-random variable over~$\mathcal{D}'(\Omega^{\circ})\times \mathcal{D}'(\partial\Omega)$ equipped with~$\mu_{\mm{GFF}}^{\Omega,D}\otimes \mu_{2\DN}^{\partial\Omega,\Omega}$, where the latter is the Gaussian measure with covariance operator~$\frac{1}{2}(\DN_{\Omega}^{\partial\Omega})^{-1}$.
\end{deef}

\begin{def7}\label{rem-interp-inter-over-omega}
    We emphasize here that the variable~$S_{\Omega}(\phi_{\Omega}^D|\varphi)$ must be understood for~$\phi_{\Omega}^D$ and~$\varphi$ \textit{both} random, with the law of~$\varphi$ mutually absolutely continuous with respect to~$\mu_{2\DN}^{\partial\Omega,\Omega}$ (thus~$\mu_{2\mn{D}}^{\partial\Omega}$ is a valid candidate). Thus if one asks how much can one ``fix'' $\varphi$ as one interprets~$S_{\Omega}(\phi_{\Omega}^D|\varphi)$, the answer is that it makes sense as a random variable of~$\phi_{\Omega}^D$ alone for almost every fixed~$\varphi$ under $\mu_{2\mn{D}}^{\partial\Omega}$ (but this full-measure set with respect to $\mu_{2\mn{D}}^{\partial\Omega}$ is generally unknown). Said from a slightly different perspective, the expression~$S_{\Omega}(\phi)$ for a generic~$\phi\in \mathcal{D}'(\Omega^{\circ})$ makes sense as a random variable only when~$\phi$ follows a law mutually absolutely continuous with respect to~$\mu_{\mm{GFF}}^{\Omega,D} * (\PI_{\Omega}^{\partial\Omega})_*\mu_{2\mn{D}}^{\partial\Omega}$, where ``$*$'' denotes convolution product. Statements involving~$S_{\Omega}(\phi_{\Omega}^D|\varphi)$ in the sequel should be understood in this sense. 
  \end{def7}

  The following result will be useful in lemma \ref{lemm-amplitude-express}.

\begin{prop}\label{prop-decoupling}
  For any fixed Wick ordering~${:}\bullet{:}$ we have
  \begin{equation}
    \int_{M}^{}{:}P(\phi(x)){:}\dd V_M=\int_{\Omega}^{} {:}P(\phi^{D}_{\Omega}(x)+\phi_{\partial\Omega}(x)){:}\dd V_{\Omega} +\int_{\Omega^c}^{} {:}P(\phi_{\Omega^c}^{D}(x)+\phi_{\partial\Omega}(x)){:}\dd V_{\Omega^c}
    \label{eqn-decoup-inter}
  \end{equation}
  in~$L^p(\mu_{\mm{GFF}}^M)$ for all~$1\le p<\infty$ and in particular pointwise almost surely, as a result
  \begin{equation}
    \me^{- \int_{M}^{}{:}P(\phi(x)){:}\dd V_M}=\me^{-\int_{\Omega}^{} {:}P(\phi^{D}_{\Omega}(x)+\phi_{\partial\Omega}(x)){:}\dd V_{\Omega}} \me^{-\int_{\Omega^c}^{} {:}P(\phi_{\Omega^c}^{D}(x)+\phi_{\partial\Omega}(x)){:}\dd V_{\Omega^c}}
    \label{}
  \end{equation}
  pointwise almost surely and thus also in~$L^1(\mu_{\mm{GFF}}^M)$.
\end{prop}

\begin{def7}
  Note that the equality (\ref{eqn-decoup-inter}) says nothing about pointwise almost sure convergence of the mollified sequences~$S_{M,\varepsilon}$,~$S_{\Omega,\varepsilon}$, and~$S_{\Omega^c,\varepsilon}$. In fact, only separate subsequences in~$\varepsilon$ will converge pointwise almost surely to the three respective terms in (\ref{eqn-decoup-inter}). 
\end{def7}

\begin{proof}
  Consider a smooth partition of unity
  \begin{equation}
    1_M=\chi_{\Omega}+\chi_{\partial\Omega}+\chi_{\Omega^c},
    \label{eqn-decoup-pf-part-unity}
  \end{equation}
  where~$\chi_{\Omega}$ and~$\chi_{\Omega^c}$ are supported in the interiors of~$\Omega$ and~$\Omega^c$ respectively and~$\chi_{\partial\Omega}$ is supported near~$\partial\Omega$. Note then that (\ref{eqn-decoup-pf-part-unity}) holds as well in~$L^4(M)$ so by corollary \ref{cor-nelson-all-Lp-cutoff-limit} we have
  \begin{equation}
    \int_{M}^{}{:}P(\phi(x)){:}\dd V_M= \int_{M}^{}\chi_{\Omega}{:}P(\phi(x)){:}\dd V_M+\int_{M}^{}\chi_{\partial\Omega}{:}P(\phi(x)){:}\dd V_M+\int_{M}^{}\chi_{\Omega^c}{:}P(\phi(x)){:}\dd V_M
    \label{}
  \end{equation}
  in~$L^p(\mu_{\mm{GFF}}^M)$ for~$1\le p<\infty$. Now take
  \begin{equation}
    \chi_{\Omega}\to 1_{\Omega},\quad \chi_{\Omega^c}\to 1_{\Omega^c},\quad\textrm{and }\chi_{\partial\Omega}\to 0,
    \label{}
  \end{equation}
  in~$L^4(M)$, keeping the equality (\ref{eqn-decoup-pf-part-unity}) in the process. Then
  \begin{align*}
    \int_{M}^{}{:}P(\phi(x)){:}\dd V_M &=\int_{M}^{}1_{\Omega}{:}P(\phi(x)){:}\dd V_M+\int_{M}^{}1_{\Omega^c}{:}P(\phi(x)){:}\dd V_M \\
    &=\int_{\Omega}^{} {:}P(\phi^{D}_{\Omega}(x)+\phi_{\partial\Omega}(x)){:}\dd V_{\Omega} +\int_{\Omega^c}^{} {:}P(\phi_{\Omega^c}^{D}(x)+\phi_{\partial\Omega}(x)){:}\dd V_{\Omega^c}
  \end{align*}
  in~$L^p(\mu_{\mm{GFF}}^M)$ for~$1\le p<\infty$, and we finish the proof.
\end{proof}

\section{Presenting the $P(\phi)_2$ Model as a Segal Theory}\label{sec-segal-main}

\subsection{Description of Segal's Rules Tailored to the Model}\label{sec-segal-descript}

First we describe the (semi-)category\footnote{There are no identity morphisms. We will not discuss this point here. Some discussion can be found on \cite{KS} pp.\ 687.}~$\mathcal{C}_2^{\mm{riem}}$ of 2-dimensional Riemannian cobordisms used in this paper. There is yet no definitive choice for the exact definition of this category in the literature, in the context of functorial QFTs. Different versions (by all means related) serve different needs. See e.g.\ \cite{KS} Section 3, \cite{KMW} Section 7.1, \cite{ST} Section 2.3, \cite{freed-2019-lecturesonfieldtheoryandtopology} Section 1.1-2, \cite{GKR} Section 2.6, to name a few. What we consider in this paper is a special case that can be thought of as lying in the intersection of these more general versions in the references. It is closest in spirit to the one in \cite{KMW}.
\begin{description}
  \item[Objects] in~$\mathcal{C}_2^{\mm{riem}}$ are closed oriented Riemannian 1-manifolds with two-sided collars,~$\mf{\Sigma}=\Sigma\times (-\varepsilon,\varepsilon)$, i.e.\ disjoint unions of thin annuli. We equip~$\mf{\Sigma}$ with the product orientation of the one on~$\Sigma$ and the positive orientation on~$(-\varepsilon,\varepsilon)$, and the \textit{flat} metric~$\dd \theta^2 +\dd t^2$, with~$\theta$ being the arc length parameter on a component of~$\Sigma$. We also include the empty manifold $\varnothing$ as a special object. 
  \textbf{We identify $\Sigma$ with} $\Sigma\times\{0\}$, called the \textsf{core}, when the context is clear.
  \item[Morphisms]  from~$\mf{\Sigma}_{1}:=\Sigma_1\times (-\varepsilon_1,\varepsilon_1)$ to~$\mf{\Sigma}_{2}:=\Sigma_2\times (-\varepsilon_2,\varepsilon_2)$ are triples~$(\Omega,\mathsf{I}_{\Omega},\mathsf{O}_{\Omega})$ where~$\Omega$ is an oriented Riemannian 2-manifold with boundary,~$\mathsf{I}_{\Omega}:\Sigma_1\times[0,\varepsilon_1)\lto \Omega$ and~$\mathsf{O}_{\Omega}:\Sigma_2\times (-\varepsilon_2,0]\lto \Omega$, respectively, orientation preserving isometries such that~$\mathsf{I}_{\Omega}(\Sigma_1)$,~$\mathsf{O}_{\Omega}(\Sigma_2)\subset \partial\Omega$ and~$\mathsf{I}_{\Omega}(\Sigma_1)\sqcup \mathsf{O}_{\Omega}(\Sigma_2)=\partial\Omega$. In particular, the metric on $\Omega$ is required to be \textit{flat at the boundary}. 
  If one of~$\mf{\Sigma}_{1}$ and~$\mf{\Sigma}_{2}$, or both, is~$\varnothing$, then we formally replace the corresponding~$\mathsf{I}_{\Omega}$ and/or~$\mathsf{O}_{\Omega}$ by~$\varnothing$. We write $(\Omega,\mathsf{I}_{\Omega},\mathsf{O}_{\Omega})\in\Mor(\mf{\Sigma}_{1}, \mf{\Sigma}_{2})$ or simply $\Omega\in\Mor(\mf{\Sigma}_{1}, \mf{\Sigma}_{2})$ where the parametrizations are implied.
  Thus $\Mor(\varnothing,\varnothing)$ are simply closed oriented Riemannian surfaces.
  \item[Composition]  of~$\Omega_{12}\in\Mor(\mf{\Sigma}_1,\mf{\Sigma}_2)$ with~$\Omega_{23}\in \Mor(\mf{\Sigma}_2,\mf{\Sigma}_3)$, where $\mf{\Sigma}_2\ne \varnothing$, is the surface~$\Omega_{23}\circ\Omega_{12}$ obtained by gluing~$\Omega_{12}$ and~$\Omega_{23}$ along~$\mf{\Sigma}_2$ using~$\mathsf{O}_{\Omega_{12}}$ and~$\mathsf{I}_{\Omega_{23}}$, with~$\mathsf{I}_{\Omega_{12}}$ and~$\mathsf{O}_{\Omega_{23}}$ as its parametrization data. If~$\mf{\Sigma}_2=\varnothing$, then~$\Omega_{23}\circ\Omega_{12}:=\Omega_{23}\sqcup\Omega_{12}$. 
  Moreover if~$\Omega\in\Mor(\varnothing,\mf{\Sigma}'\sqcup \mf{\Sigma}_i\sqcup\mf{\Sigma}_j)$ and~$\rho:\mf{\Sigma}_i\lto \mf{\Sigma}_j$ is an orientation \textit{reversing} isometry sending cores to cores, then
we obtain~$\Omega/\rho\in\Mor(\varnothing,\mf{\Sigma}')$ by gluing~$ \mf{\Sigma}_i$ with~$ \mf{\Sigma}_j$ along~$\rho$ (see figure \ref{fig-segal-gluing}).
Possibly $\mf{\Sigma}'=\varnothing$.
  Note that the Riemannian metric on~$\Omega_{23}\circ\Omega_{12}$ or $\Omega/\rho$ is automatically smooth since we glued along collars.
\end{description}

The following aspects of this set-up should be understood before proceeding.
\begin{description}
  \item[Different Collars.] Clearly any morphism~$\Omega\in\Mor(\mf{\Sigma}_1,\mf{\Sigma}_2)$ can be canonically identified with a morphism from~$\mf{\Sigma}_1':=\Sigma_1\times (-\varepsilon_1',\varepsilon_1')$ to~$\mf{\Sigma}_2':=\Sigma_2\times (-\varepsilon_2',\varepsilon_2')$ for any~$\varepsilon_i'\le \varepsilon_i$. We will apply this identification when convenient. Moreover, we define the disjoint union of objects by~$\mf{\Sigma}_1\sqcup \mf{\Sigma}_2:=(\Sigma_1\sqcup \Sigma_2)\times (-\varepsilon_{\min},\varepsilon_{\min})$ where~$\varepsilon_{\min}=\min\{\varepsilon_1,\varepsilon_2\}$. These points are related to the ``germs'' considered in \cite{KS} Section 3.
  \item[Crossing and Reflections.] On~$\mf{\Sigma}=\Sigma\times (-\varepsilon,\varepsilon)$ there is the isometric reflection~$\Theta_{\mf{\Sigma}}(\theta,t):=(\theta,-t)$. Denote by~$\mf{\Sigma}^*$ the same manifold as~$\mf{\Sigma}$ but with orientation reversed by~$\Theta_{\mf{\Sigma}}$. Then any morphism~$\Omega\in \Mor(\mf{\Sigma}_1,\mf{\Sigma}_2)$ is identified e.g.\ with a morphism~$\Omega\in\Mor(\mf{\Sigma}_1\setminus \mf{S},\mf{\Sigma}_2\sqcup \mf{S}^*)$ (same notation), for all~$\mf{S}\subset \mf{\Sigma}_1$ which is a union of components, in particular with one in~$\Mor(\varnothing,\mf{\Sigma}_1^*\sqcup \mf{\Sigma}_2)$ or~$\Mor(\mf{\Sigma}_1\sqcup \mf{\Sigma}_2^*,\varnothing)$. This is called \textsf{crossing} in \cite{Segal}. We use a particular notation for the ``adjoint''~$\Omega^*\in \Mor(\mf{\Sigma}_2^*,\mf{\Sigma}_1^*)$. 
  \item[Isometry of Morphisms.] Two morphisms~$\Omega$,~$\Omega'\in\Mor(\mf{\Sigma}_1,\mf{\Sigma}_2)$ are identified if there is an isometry~$\kappa:\Omega\lto \Omega'$ such that~$\kappa\circ\mathsf{I}_{\Omega}=\mathsf{I}_{\Omega'}$,~$\kappa\circ \mathsf{O}_{\Omega}=\mathsf{O}_{\Omega'}$.
\end{description}

Now we are ready to define \textbf{Segal's axioms}.

\begin{deef}
  [\cite{Segal} page 460] \label{def-segal-2} A 2d Riemannian (unitary) QFT is a correspondence where
  \begin{enumerate}[(i)]
    \item to each object $\mf{\Sigma}$ in~$\mathcal{C}_2^{\mm{riem}}$ we associate a real Hilbert space~$\mathcal{H}_{\mf{\Sigma}}$, such that $\mathcal{H}_{\mf{\Sigma}_1\sqcup\mf{\Sigma}_2}=\mathcal{H}_{\mf{\Sigma}_1}\otimes \mathcal{H}_{\mf{\Sigma}_2}$ (completed tensor product), and~$\mathcal{H}_{\varnothing}:=\mb{R}$;
    if~$\rho:\mf{\Sigma}\xlongrightarrow{\sim} \mf{\Sigma}'$ is an orientation and core preserving isometry, there is a corresponding isometry~$\rho_*:\mathcal{H}_{\mf{\Sigma}}\xlongrightarrow{\sim}\mathcal{H}_{\mf{\Sigma}'}$; if~$\rho$ is orientation reversing, then~$\rho_*:\mathcal{H}_{\mf{\Sigma}}\xlongrightarrow{\sim}\mathcal{H}_{\mf{\Sigma}'}^*\cong \mathcal{H}_{\mf{\Sigma}'}$;
    \item to each $\Omega\in\Mor(\varnothing,\mf{\Sigma})$ we associate an element~$\mathcal{A}_{\Omega}\in \mathcal{H}_{\mf{\Sigma}}$; if~$M\in\Mor(\varnothing,\varnothing)$ is closed, we associate a real number~$Z_{M}\in \mb{R}$, called the \textsf{partition function};
    \item let~$\Omega\in\Mor(\varnothing,\mf{\Sigma}'\sqcup \mf{\Sigma}_i\sqcup\mf{\Sigma}_j)$,~$\rho$ and $\Omega/\rho$ be as described in the ``composition'' item above. Then
      \begin{equation}
	\mathcal{A}_{\Omega/\rho}=\ttr_{\rho}(\mathcal{A}_{\Omega}),
	\label{eqn-trace-axiom-general-segal}
      \end{equation}
      where~$\ttr_{\rho}$ is (formally) the map
      \begin{equation}
	\left.
	\begin{array}{rcl}
	  \ttr_{\rho}:\mathcal{H}_{\mf{\Sigma}_i} \otimes \mathcal{H}_{\mf{\Sigma}_j} \otimes \mathcal{H}_{\mf{\Sigma}'}  &\lto&\mathcal{H}_{\mf{\Sigma}'}  ,\\
	  F\otimes G\otimes H&\longmapsto &\ank{\rho_* F, G}_{\mf{\Sigma}_j} H,
	\end{array}
	\right.
	\label{eqn-segal-general-trace}
      \end{equation}
      with~$\ank{-,-}_{\mf{\Sigma}_j} $ being the (real) inner product on~$\mathcal{H}_{\mf{\Sigma}_j} $.
    \end{enumerate}
\end{deef}

In our construction later, the map $\ttr_{\rho}$ must be \textbf{interpreted differently} than (\ref{eqn-segal-general-trace}). More precisely, each $\mathcal{H}_{\mf{\Sigma}}$ will be the $L^2$ space over some measure space with measure $\mu_{\Sigma}$, 
and~$\ttr_{\rho}$ should be understood as
  \begin{equation}
    \ttr_{\rho}:\mathcal{A}_{\Omega}(x,y,z) \longmapsto \int_{}^{}\mathcal{A}_{\Omega}(x,x,z)\dd \mu_{\Sigma_j}(x),
    \label{eqn-def-flat-trace}
  \end{equation}
  and we will show in sections \ref{sec-trace-axiom} and \ref{sec-gluing-final} that the latter integral converges and indeed gives (\ref{eqn-trace-axiom-general-segal}). On the other hand (\ref{eqn-segal-general-trace}) is an algebraic definition which in general may not equal (\ref{eqn-def-flat-trace}).

In Definition \ref{def-segal-2} we focused on morphisms from (or to)~$\varnothing$ since the general case can be deduced from this case via crossing. We re-state this fact as a lemma below, tailored to the case where the Hilbert spaces are~$L^2$ spaces and where it follows from Theorem 3.8.4 of \cite{Sim3}.
\begin{lemm}\label{lemm-segal-transfer}
  Consider a 2D Riemannian QFT in the sense of Definition \ref{def-segal-2} where the Hilbert spaces are~$L^2$ spaces of real measurable functions and where~$\ttr_{\rho}$ is interpreted as (\ref{eqn-def-flat-trace}). Given~$\Omega\in \Mor(\mf{\Sigma}_{1},\mf{\Sigma}_{2})$, we define the \textsf{Segal transfer operator}~$U_{\Omega}$ by
 \begin{equation}
    \left.\def\arraystretch{1.3}
    \begin{array}{rcl}
      U_{\Omega}: \mathcal{H}_{\mf{\Sigma}_{1}} &\lto& \mathcal{H}_{\mf{\Sigma}_{2}},\\
      F(x) &\longmapsto & \ddp\int_{}^{}\mathcal{A}_{\Omega}(y,x)
      F(x)\,\dd \mu_{\Sigma_1}(x),
    \end{array}
    \right.
    \label{eqn-def-segal-transfer}
  \end{equation} 
  where~$\mathcal{A}_{\Omega}\in \mathcal{H}_{\mf{\Sigma}_1}\otimes \mathcal{H}_{\mf{\Sigma}_2}$ viewing~$\Omega\in \Mor(\varnothing,\mf{\Sigma}_1^*\sqcup \mf{\Sigma}_2)$. Then $U_{\Omega}$ is Hilbert-Schmidt. Moreover,
  \begin{enumerate}[(i)]
    \item if~$\Omega_{12}\in\Mor(\mf{\Sigma}_1,\mf{\Sigma}_2)$,~$\Omega_{23}\in \Mor(\mf{\Sigma}_2,\mf{\Sigma}_3)$, then
      \begin{equation}
    U_{\Omega_{23}\circ  \Omega_{12}}=U_{\Omega_{23}}\circ U_{\Omega_{12}}.
    \label{eqn-def-segal-sew}
  \end{equation}
  \item If~$\Omega\in \Mor(\mf{\Sigma},\mf{\Sigma})$, $\mf{\Sigma}$ identified with itself along an isometry $\rho$, and let~$\check{\Omega}:=\Omega/\rho$ the closed surface, then
  \begin{equation}
    Z_{\check{\Omega}}=\ttr_{\rho} (U_{\Omega}).
    \label{}
  \end{equation}
\item If~$\Omega\in \Mor(\mf{\Sigma}_1,\mf{\Sigma}_2)$, recall from ``crossing and reflections'' the interpretation~$\Omega^*\in \Mor(\mf{\Sigma}_2^*,\mf{\Sigma}_1^*)$, then
  \begin{equation}
    U_{\Omega^*}=U_{\Omega}^{\dagger},
    \label{}
  \end{equation}
  with~$U_{\Omega}^{\dagger}$ denoting the (real) adjoint of~$U_{\Omega}$. \hfill~$\Box$
  \end{enumerate}
\end{lemm}

\begin{figure}
    \centering
    \includegraphics[width=0.8\linewidth]{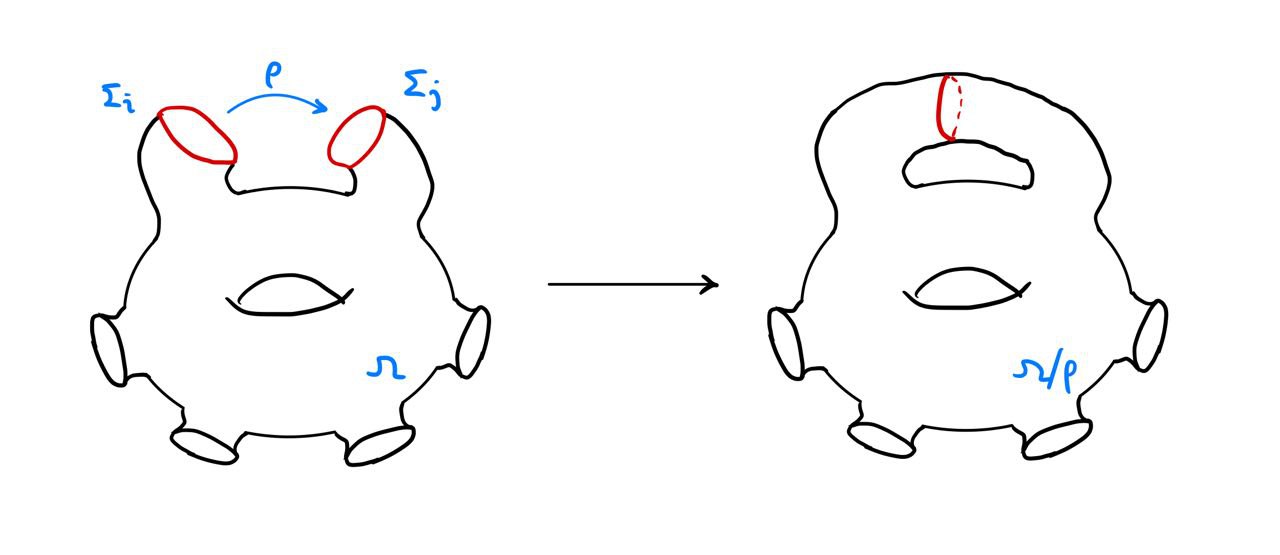}
    \caption{sewing}
    \label{fig-segal-gluing}
\end{figure}

All these said, we now restate theorem \ref{thrm-intro-main-1} in the introduction in the following precise manner.

\begin{thrmbis}{thrm-intro-main-1}
\label{thrm-main-1-precise}
  The Hilbert spaces~$\mathcal{H}_R$ given by (\ref{eqn-def-hilb-space-circ}), the amplitudes~$\mathcal{A}_{\Omega}^P$ given by (\ref{eqn-def-amp-omega}) and the partition functions~$Z_M$ given by (\ref{eqn-def-part-func-rigor}) together gives a 2d Riemannian QFT satisfying the requirements of Definition \ref{def-segal-2}.
\end{thrmbis}

The proof of this theorem constitutes the main body of subsections \ref{sec-trace-axiom} and \ref{sec-gluing-final}.

\begin{def7}
  [difference with CFT setting] Segal's axioms were originally intended to describe Conformal Field Theories (CFTs), to which the $P(\phi)_2$-model does not belong. Compared to Definition \ref{def-segal-2}, the \textit{main extra ingredient} required for a CFT is a transformation rule of the objects~$Z_M$ and~$\mathcal{A}_{\Omega}$ (resp.\ $U_{\Omega}$) under \textit{conformal change} of the underlying metrics on~$M$ and~$\Omega$, written in terms of the so-called \textit{Weyl anomaly} (see e.g.\ \cite{GKR} Sections 2.4, 2.6). With these rules one is allowed to scale the geometric objects and obtain relations between the corresponding~$\mathcal{A}_{\Omega}$s. See \cite{GKR} Section 2.7(1)-(2).
\end{def7}

\subsection{The Hilbert Spaces}\label{sec-segal-hilb-space}

\noindent Now we define the Hilbert space~$\mathcal{H}_{\mf{\Sigma}}$ associated to the object~$\mf{\Sigma}$ in $\mathcal{C}_2^{\mm{riem}}$.  
In fact, we will define a Hilbert space~$\mathcal{H}_{\Sigma}$ for the \textit{core}~$\Sigma$ of each~$\mf{\Sigma}$ and~$\mathcal{H}_{\mf{\Sigma}}:=\mathcal{H}_{\Sigma}$. The core~$\Sigma$ is a closed oriented Riemannian 1-manifold.
Let
\begin{equation}
    \mn{D}_{\Sigma}\defeq(\Delta_{\Sigma}+m^2)^{1/2}
\end{equation}
where~$\Delta_{\Sigma}:=-\partial_{\theta}^2$ is the positive Laplacian on $\Sigma$ with an arc length parameter~$\theta$ on each component. Denote by~$\mu_{(2\mn{D}_{\Sigma})^{-1}}$ the Gaussian measure on~$\mathcal{D}'(\Sigma)$ with covariance~$(2\mn{D}_{\Sigma})^{-1}$ (for convenience, we just write~$\mu_{2\mn{D}}$). Then we define
\begin{equation}
  \mathcal{H}_{\Sigma} \defeq L^2_{\mb{R}}(\mathcal{D}'(\Sigma),\wh{\mu}_{2\mn{D}}).
  \label{eqn-def-hilb-space-circ}
\end{equation}

 Here we take~$\wh{\mu}_{2\mn{D}}:=\mu_{2\mn{D}}\cdot\detz(2\mn{D}_{\Sigma})^{-\frac{1}{2}}$, namely the probability measure scaled by the positive finite constant~$\detz(2\mn{D}_{\Sigma})^{-\frac{1}{2}}$. We observe that
\begin{equation}
  \mathcal{D}'(\Sigma_1\sqcup \Sigma_2)=\mathcal{D}'(\Sigma_1)\oplus \mathcal{D}'(\Sigma_2)
  \label{}
\end{equation}
and the operator~$\mn{D}_{\Sigma}$ acts diagonally, and clearly the Gaussian fields on $\Sigma_1$ and $\Sigma_2$ are independent, that is,~$\mu_{2\mn{D}}^{\Sigma_1\sqcup\Sigma_2}=\mu_{2\mn{D}}^{\Sigma_1}\otimes \mu_{2\mn{D}}^{\Sigma_2}$. The real~$L^2$ space is hence the tensor product (\cite{RSim} page 52).

Finally, if~$\rho:\Sigma\lto \Sigma'$ is an isometry (of Riemannian 1-manifolds), then define~$\rho_*:\mathcal{H}_{\Sigma}\lto \mathcal{H}_{\Sigma'}$ to be the obvious map induced by the correspondence between~$\mathcal{D}'(\Sigma)$ and~$\mathcal{D}'(\Sigma')$. Up to this point we have verified (i) of Definition \ref{def-segal-2}. By the nature of our definitions, in the sequel we do not distinguish $\mathcal{H}_{\mf{\Sigma}}$ from $\mathcal{H}_{\Sigma}$.

\subsection{Amplitudes = Schwartz Kernels} \label{sec-amplitude}

\noindent Let $\Omega\in\Mor(\varnothing,\mf{\Sigma})$, now we define~$\mathcal{A}_{\Omega}\in \mathcal{H}_{\Sigma}$. Consider $\Omega^*\in \Mor(\mf{\Sigma}^*,\varnothing)$ as described in the paragraph ``crossing and reflections'' of subsection \ref{sec-segal-descript}, and let~$|\Omega|^2:=\Omega^*\circ\Omega$, a closed surface. Clearly $|\Omega|^2$ is also the \textit{isometric double} of $\Omega$, equipped with a \textit{reflection} $\Theta_{|\Omega|^2}$ with $\Theta_{|\Omega|^2}|_{\mf{\Sigma}}=\Theta_{\mf{\Sigma}}$ (\cite{JR08} Section 2). Now suppose we have defined~$\mathcal{A}_{\Omega}$ satisfying Definition \ref{def-segal-2}, then (iii) would imply
\begin{align*}
  Z_{|\Omega|^2}&=\int_{}^{}|\mathcal{A}_{\Omega}(\varphi)|^2 \dd\mu_{2\mn{D}}^{\Sigma}(\varphi)\detz (2\mn{D}_{\Sigma})^{-\frac{1}{2}} \\
  &=\int_{\mathcal{D}'(M)}^{}\me^{-S_{|\Omega|^2}(\phi)}\dd\mu_{\mm{GFF}}^{|\Omega|^2}(\phi)\detz(\Delta_{|\Omega|^2}+m^2)^{-\frac{1}{2}} \\
  &=\int_{\mathcal{D}'(\Sigma)}^{}\dd\tau_{\Sigma}\big(\me^{-S_{|\Omega|^2}}\cdot\mu_{\mm{GFF}}^{|\Omega|^2}\big)(\varphi) \detz(\Delta_{|\Omega|^2}+m^2)^{-\frac{1}{2}}.
\end{align*}
In the case~$\Omega\in \Mor(\mf{\Sigma}_1,\mf{\Sigma}_2)$ we have in other words~$Z_{|\Omega|^2}=\ttr_{\rho}(U_{\Omega^*}U_{\Omega})$. This motivates the following definition which is also inspired by related considerations in \cite{Pickrell}.

\begin{deef}\label{def-segal-amp-omega}
  Let~$\Omega\in\Mor(\varnothing,\mf{\Sigma})$ and $P$ the polynomial defining the interaction. We define the \textsf{amplitude} associated to~$\Omega$ to be the quantity
\begin{equation}
  \mathcal{A}^P_{\Omega}(\varphi)\defeq \bigg[\frac{\detz(\Delta_{|\Omega|^2}+m^2)^{-\frac{1}{2}}}{\detz (2\mn{D}_{\Sigma})^{-\frac{1}{2}}}\cdot \frac{\dd\big[\tau_{\Sigma}\big(\me^{-S_{|\Omega|^2}}\cdot\mu_{\mm{GFF}}^{|\Omega|^2}\big)\big]}{\dd\mu_{2\mn{D}}^{\Sigma}}(\varphi) \bigg] ^{\frac{1}{2}},
  \label{eqn-def-amp-omega}
\end{equation}
where the second ratio denotes the Radon-Nikodym density.
\end{deef}

This is well-defined since~$\tau_{\Sigma}(\me^{-S_{|\Omega|^2}}\cdot\mu_{\mm{GFF}}^{|\Omega|^2}) \ll \mu_{2\mn{D}}^{\Sigma}$ as~$\tau_{\Sigma}(\mu_{\mm{GFF}}^{|\Omega|^2}) =\mu_{2\DN}^{\Sigma,\Omega} \ll \mu_{2\mn{D}}^{\Sigma}$ by Corollary \ref{corr-rad-niko-dense}, and both are finite positive measures. We see also that~$\mathcal{A}_{\Omega}>0$ almost surely with respect to~$\wh{\mu}_{2\mn{D}}^{\Sigma}$. We have automatically~$\mathcal{A}_{\Omega}\in L^2(\mathcal{D}'(\Sigma),\wh{\mu}_{2\mn{D}}^{\Sigma})=\mathcal{H}_{\Sigma}$ since~$|\mathcal{A}_{\Omega}|^2\in L^1(\mathcal{D}'(\Sigma),\wh{\mu}_{2\mn{D}}^{\Sigma})$ by definition.

\begin{lemm}
  \label{lemm-amp-iso-inv} Let~$\Omega$,~$\Omega'\in\Mor(\varnothing,\mf{\Sigma})$ and~$\kappa:\Omega\lto \Omega'$ an isometry of morphisms in the sense of Subsection \ref{sec-segal-descript}. Then~$\mathcal{A}_{\Omega}^P=\mathcal{A}_{\Omega'}^P$.
\end{lemm}

\begin{proof}
  The isometry~$\kappa$ may be extended to an isometry of the double~$|\Omega|^2$,~$\kappa^*\sqcup \kappa$, which fixes a collar of~$\Sigma$. The lemma follows because all objects involved in the definition (\ref{eqn-def-amp-omega}) are isometry invariant. In particular the interaction~$S_{|\Omega|^2}$ is so as there exist isometry invariant regularization procedures (e.g. heat or (\ref{eqn-dyatlov-zwors-reg})).
\end{proof}

\begin{exxx}
  Let us derive the amplitude~$\mathcal{A}_{\Omega}^0$ for the free field, that is, with $S_{|\Omega|^2}=0$. Indeed,
  \begin{equation}
    |\mathcal{A}^0_{\Omega}(\varphi)|^2=\frac{\dd\tau_{\Sigma}\big(\mu_{\mm{GFF}}^{|\Omega|^2}\detz(\Delta_{|\Omega|^2}+m^2)^{-\frac{1}{2}}\big)}{\dd\mu_{2\mn{D}}^{\Sigma}\detz (2\mn{D})^{-\frac{1}{2}}}(\varphi).
    \label{}
  \end{equation}
  Since
  \begin{equation}
    \tau_{\Sigma}(\mu_{\mm{GFF}}^{|\Omega|^2})=\mu_{2\DN}^{\Sigma,\Omega},
    \label{}
  \end{equation}
  namely the Gaussian measure on~$\mathcal{D}'(\Sigma)$ with covariance~$\frac{1}{2}(\DN_{\Omega}^{\Sigma})^{-1}$, taking into account the BFK formula
  \begin{equation}
    \detz(\Delta_{|\Omega|^2}+m^2)=[\detz(\Delta_{\Omega,D}+m^2)]^2\detz(2\DN_{\Omega}^{\Sigma}),
    \label{}
  \end{equation}
  we obtain, by Corollary \ref{corr-rad-niko-dense},
  \begin{equation} \mathcal{A}_{\Omega}^0(\varphi)=\detz(\Delta_{\Omega,D}+m^2)^{-\frac{1}{2}}\me^{-\frac{1}{2}\sank{\varphi,(\DN_{\Omega}^{\Sigma}-\mn{D})\varphi}_{L^2(\Sigma)}}.
    \label{}
  \end{equation}
\end{exxx}

\begin{lemm}\label{lemm-amplitude-express}
  We have
  \begin{equation}
    \mathcal{A}_{\Omega}^P(\varphi)=\mathcal{A}_{\Omega}^0(\varphi)\int_{}^{}\me^{-S_{\Omega}(\phi^{D}_{\Omega}|\varphi)}\dd\mu_{\mm{GFF}}^{\Omega,D}(\phi^{D}_{\Omega}),
    \label{}
  \end{equation}
  for almost every~$\varphi$ under~$\mu_{2\mn{D}}^{\Sigma}$. Here~$S_{\Omega}(\phi^{D}_{\Omega}|\varphi)$ is as defined in Definition \ref{def-inter-over-domain}.
\end{lemm}

\begin{proof}
  Indeed, by definition of the measure image,
  \begin{equation}
    \mathcal{A}_{\Omega}^P(\varphi)^2=\mathcal{A}_{\Omega}^0(\varphi)^2\iint_{}^{}\me^{-S_{|\Omega|^2}(\phi^{D}_{\Omega}+\phi^D_{\Omega^*}+\PI_{|\Omega|^2}^{\Sigma}\varphi)}\dd\mu_{\mm{GFF}}^{\Omega,D}(\phi^{D}_{\Omega})\otimes \dd\mu_{\mm{GFF}}^{\Omega^*,D}(\phi^{D}_{\Omega^*}).
    \label{}
  \end{equation}
  However, by Proposition \ref{prop-decoupling} which decouples the interaction into a sum over complementary regions,
  \begin{equation}
    \me^{-S_{|\Omega|^2}(\phi)}=\me^{-S_{\Omega}(\phi^{D}_{\Omega}|\varphi)}\me^{-S_{\Omega^*}(\phi^D_{\Omega^*}|\varphi)},
    \label{}
  \end{equation}
  almost surely againt~$\mu_{\mm{GFF}}^{\Omega,D}\otimes \mu_{\mm{GFF}}^{\Omega^*,D}\otimes \tau_{\Sigma}(\mu_{\mm{GFF}}^{|\Omega|^2})$, and hence also against~$\mu_{\mm{GFF}}^{\Omega,D}\otimes \mu_{\mm{GFF}}^{\Omega^*,D}\otimes \mu_{2\mn{D}}^{\Sigma}$. Thus by reflection symmetry between $\Omega$ and $\Omega^*$ and independence between~$\phi^{D}_{\Omega}$ and~$\phi^D_{\Omega^*}$ we obtain the result.
\end{proof}

\subsection{Trace Axiom and its Consequences}\label{sec-trace-axiom}

\noindent In this subsection we treat separately (ii) of Lemma \ref{lemm-segal-transfer} as part of (iii) of Definition \ref{def-segal-2}. Let~$\Omega\in\Mor(\mf{\Sigma},\mf{\Sigma})$, \textit{not necessarily connected}. We will consider two closed surfaces:~$\check{\Omega}$ (from Lemma \ref{lemm-segal-transfer}(ii)) and~$|\Omega|^2:=\Omega^*\circ\Omega$.
\begin{prop}
  [pre-trace]\label{prop-pre-trace}
  We have
  \begin{equation}
    \frac{\dd\tau_{\Sigma}(\mu_{\mm{GFF}}^{\check{\Omega}})}{\dd\mu_{2\mn{D}}^{\Sigma}}(\varphi)=\frac{\detz(\DN_{\check{\Omega}}^{\Sigma})^{\frac{1}{2}}}{\detz(2\DN_{\Omega}^{\Sigma\sqcup\Sigma})^{\frac{1}{4}}}\bigg( 
    \frac{\dd\tau_{\Sigma\sqcup\Sigma}(\mu_{\mm{GFF}}^{|\Omega|^2})}{\dd\mu_{2\mn{D}}^{\Sigma\sqcup\Sigma}}(\varphi,\varphi)
    \bigg)^{\frac{1}{2}}.
    \label{eqn-pre-trace}
  \end{equation}
\end{prop}

\begin{proof}
  This boils down to comparing the explicit expressions for the densities as given by Corollary \ref{corr-rad-niko-dense}. Indeed, the relation that we need is
  \begin{equation}
    \Bank{\bnom{\varphi}{\varphi},(2\DN_{\Omega}^{\Sigma\sqcup\Sigma} -2\mn{D}_{\Sigma\sqcup\Sigma})\bnom{\varphi}{\varphi}}_{L^2(\Sigma\sqcup\Sigma)}=
    2\bank{\varphi,(\DN_{\check{\Omega}}^{\Sigma}-2\mn{D}_{\Sigma})\varphi}_{L^2(\Sigma)}.
    \label{}
  \end{equation}
  This is true for~$\varphi\in W^{\frac{1}{2}}(\Sigma)$ because~$\mn{D}_{\Sigma\sqcup\Sigma}=\mn{D}_{\Sigma}\oplus \mn{D}_{\Sigma}$ and~$\sank{[\smx{\varphi\\ \varphi}], \DN_{\Omega}^{\Sigma\sqcup\Sigma}[\smx{\varphi\\ \varphi}]}_{L^2(\Sigma\sqcup\Sigma)}$ and~$\sank{\varphi,\DN_{\check{\Omega}}^{\Sigma}\varphi}_{L^2(\Sigma)}$ are both the Dirichlet energy of the harmonic extension over~$\Omega$ with boundary condition~$(\varphi,\varphi)$. The equality then extends to~$\varphi\in W^{-\delta}(\Sigma)$ for small enough~$\delta$ by continuity (see Lemma \ref{lemm-ghs-for-gff}(iii)).
\end{proof}

\begin{corr}
  [trace for free field] \label{cor-free-trace} In the situation as above, we have
  \begin{equation}
    \int_{}^{} \mathcal{A}_{\Omega}^0(\varphi,\varphi)\dd\mu_{2\mn{D}}^{\Sigma}(\varphi) \detz(2\mn{D}_{\Sigma})^{-\frac{1}{2}}=\detz(\Delta_{\check{\Omega}}+m^2)^{-\frac{1}{2}}.
    \label{}
  \end{equation}
\end{corr}

\begin{proof}
  Indeed,
  \begin{align*}
    \textrm{LHS}&=\frac{\detz(\Delta_{|\Omega|^2}+m^2)^{-\frac{1}{4}}}{\detz(2\mn{D}_{\Sigma}\oplus 2\mn{D}_{\Sigma})^{-\frac{1}{4}}}
    \int_{}^{}\bigg( 
    \frac{\dd\tau_{\Sigma\sqcup\Sigma}(\mu_{\mm{GFF}}^{|\Omega|^2})}{\dd\mu_{2\mn{D}}^{\Sigma\sqcup\Sigma}}(\varphi,\varphi)
    \bigg)^{\frac{1}{2}} \dd\mu_{2\mn{D}}^{\Sigma}(\varphi) \detz(2\mn{D}_{\Sigma})^{-\frac{1}{2}} \\
    &=\frac{\detz(2\DN_{\Omega}^{\Sigma\sqcup\Sigma})^{\frac{1}{4}}}{\detz(\DN_{\check{\Omega}}^{\Sigma})^{\frac{1}{2}}}
    \frac{\detz(\Delta_{|\Omega|^2}+m^2)^{-\frac{1}{4}}}{\detz(2\mn{D}_{\Sigma})^{-\frac{1}{2}}}
    \underbrace{\int_{}^{}\frac{\dd\tau_{\Sigma}(\mu_{\mm{GFF}}^{\check{\Omega}})}{\dd\mu_{2\mn{D}}^{\Sigma}}(\varphi) \dd\mu_{2\mn{D}}^{\Sigma}(\varphi)}_{=~ 1} \detz(2\mn{D}_{\Sigma})^{-\frac{1}{2}} \\
    &=\detz(\DN_{\check{\Omega}}^{\Sigma})^{-\frac{1}{2}}\detz(\Delta_{\Omega,D}+m^2)^{-\frac{1}{2}} \tag{BFK for~$|\Omega|^2$}\\
    &=\textrm{RHS}, \tag{BFK for~$\check{\Omega}$}
  \end{align*}
  finishing the proof.
\end{proof}

Note that we do not insist that~$\Omega$ be connected. This leads to the following important consequence of Proposition \ref{prop-pre-trace}.

\begin{corr}\label{cor-disec-dens}
  Let~$\Omega_1\in\Mor(\varnothing,\mf{\Sigma})$,~$\Omega_2\in \Mor(\mf{\Sigma},\varnothing)$. Denote~$M:=\Omega_2\circ\Omega_1\in \Mor(\varnothing,\varnothing)$. Then
  \begin{equation}
    \frac{\dd\tau_{\Sigma}(\mu_{\mm{GFF}}^{M})}{\dd\mu_{2\mn{D}}^{\Sigma}}(\varphi)=
    \frac{\detz(\DN_{M}^{\Sigma})^{\frac{1}{2}}}{\detz(2\DN_{\Omega_1}^{\Sigma})^{\frac{1}{4}}\detz(2\DN_{\Omega_2}^{\Sigma})^{\frac{1}{4}}}\bigg( 
    \frac{\dd\tau_{\Sigma}(\mu_{\mm{GFF}}^{|\Omega_1|^2})}{\dd\mu_{2\mn{D}}^{\Sigma}}(\varphi)
    \bigg)^{\frac{1}{2}}
    \bigg( 
    \frac{\dd\tau_{\Sigma}(\mu_{\mm{GFF}}^{|\Omega_2|^2})}{\dd\mu_{2\mn{D}}^{\Sigma}}(\varphi)
    \bigg)^{\frac{1}{2}}.
    \label{eqn-disect-density}
  \end{equation}
\end{corr}

\begin{proof}
  Indeed, in this case~$\Omega:=\Omega_1\sqcup\Omega_2=\Omega_1\circ\Omega_2$ can be seen as an element of~$\Mor(\mf{\Sigma},\mf{\Sigma})$. Then~$|\Omega|^2=|\Omega_1|^2\sqcup |\Omega_2|^2$ and the GFFs on the two components are independent. In addition,~$\DN_{\Omega}^{\Sigma\sqcup\Sigma}$ is the direct sum~$\DN_{\Omega_1}^{\Sigma}\oplus \DN_{\Omega_2}^{\Sigma}$. Thus (\ref{eqn-pre-trace}) gives (\ref{eqn-disect-density}).
\end{proof}

\begin{corr}
  [dissection gluing] \label{cor-disect-free} In the same situation as Corollary \ref{cor-disec-dens}, we have
  \begin{equation}
    \int_{}^{}\mathcal{A}_{\Omega_1}^0(\varphi)\mathcal{A}_{\Omega_2}^0(\varphi)\dd\mu_{2\mn{D}}^{\Sigma}(\varphi) \detz(2\mn{D}_{\Sigma})^{-\frac{1}{2}}=\detz(\Delta_M+m^2)^{-\frac{1}{2}}.
    \label{}
  \end{equation}
\end{corr}

\begin{proof}
  Again, we apply Corollary \ref{cor-free-trace} directly to the case~$\Omega=\Omega_1\sqcup\Omega_2$ and~$|\Omega|^2=|\Omega_1|^2\sqcup |\Omega_2|^2$. We note that~$\mathcal{A}_{\Omega_1}^0(\varphi)\mathcal{A}_{\Omega_2}^0(\varphi)= \mathcal{A}_{\Omega}^0(\varphi,\varphi)$ because
  \begin{equation}
    \frac{\dd\tau_{\Sigma\sqcup\Sigma}(\mu_{\mm{GFF}}^{|\Omega|^2})}{\dd\mu_{2\mn{D}}^{\Sigma\sqcup\Sigma}}(\varphi,\varphi)=\frac{\dd\tau_{\Sigma}(\mu_{\mm{GFF}}^{|\Omega_1|^2})}{\dd\mu_{2\mn{D}}^{\Sigma}}(\varphi) 
    \frac{\dd\tau_{\Sigma}(\mu_{\mm{GFF}}^{|\Omega_2|^2})}{\dd\mu_{2\mn{D}}^{\Sigma}}(\varphi)
    \label{}
  \end{equation}
  and~$\detz(\Delta_{|\Omega|^2}+m^2)=\detz(\Delta_{|\Omega_1|^2}+m^2)\detz(\Delta_{|\Omega_2|^2}+m^2)$.
\end{proof}

Next we deal with the trace axiom in the interacting case. Again let~$\Omega\in\Mor(\mf{\Sigma},\mf{\Sigma})$ and consider~$\check{\Omega}$. One has in this case the decomposition
\begin{equation}
  \mu_{\mm{GFF}}^{\check{\Omega}}=\mu_{\mm{GFF}}^{\Omega,D}\otimes \tau_{\Sigma}(\mu_{\mm{GFF}}^{\check{\Omega}})
  \label{}
\end{equation}
and against which
\begin{equation}
  S_{\check{\Omega}}(\phi)=S_{\check{\Omega}}(\phi^{D}_{\Omega}+\PI_{\check{\Omega}}^{\Sigma}\varphi)=S_{\Omega}(\phi^{D}_{\Omega}|\varphi,\varphi)
  \quad\textrm{and}\quad
  \me^{-S_{\check{\Omega}}(\phi^{D}_{\Omega}+\PI_{\check{\Omega}}^{\Sigma}\varphi)}=\me^{-S_{\Omega}(\phi^{D}_{\Omega}|\varphi,\varphi)}
  \label{}
\end{equation}
for~$\phi^{D}_{\Omega}\sim \mu_{\mm{GFF}}^{\Omega,D}$ and~$\varphi\sim \tau_{\Sigma}(\mu_{\mm{GFF}}^{\check{\Omega}})$ almost surely.

\begin{corr}
  [trace for~$P(\phi)$ field] \label{cor-P(phi)-trace} In the situation as above, we have
  \begin{equation}
    \int_{}^{} \mathcal{A}_{\Omega}^P(\varphi,\varphi)\dd\mu_{2\mn{D}}^{\Sigma}(\varphi) \detz(2\mn{D}_{\Sigma})^{-\frac{1}{2}}=Z_{\check{\Omega}}.
    \label{}
  \end{equation}
\end{corr}

\begin{proof}
  Indeed,
  \begin{equation}
    \mathcal{A}_{\Omega}^P(\varphi,\varphi)=\mathcal{A}_{\Omega}^0(\varphi,\varphi)\int_{}^{}\me^{-S_{\Omega}(\phi^{D}_{\Omega}|\varphi,\varphi)}\dd\mu_{\mm{GFF}}^{\Omega,D}(\phi^{D}_{\Omega}),
    \label{}
  \end{equation}
  therefore, with the constants involving determinants working out in exactly the same way as Corollary \ref{cor-free-trace}, one has
  \begin{align*}
    \textrm{LHS}&=\detz(\Delta_{\check{\Omega}}+m^2)^{-\frac{1}{2}} \int_{}^{}\frac{\dd\tau_{\Sigma}(\mu_{\mm{GFF}}^{\check{\Omega}})}{\dd\mu_{2\mn{D}}^{\Sigma}}(\varphi) \dd\mu_{2\mn{D}}^{\Sigma}(\varphi)\int_{}^{}\me^{-S_{\Omega}(\phi^{D}_{\Omega}|\varphi,\varphi)}\dd\mu_{\mm{GFF}}^{\Omega,D}(\phi^{D}_{\Omega}) \\
    &=\detz(\Delta_{\check{\Omega}}+m^2)^{-\frac{1}{2}}
    \int_{}^{}\me^{-S_{\Omega}(\phi^{D}_{\Omega}|\varphi,\varphi)} \dd\mu_{\mm{GFF}}^{\Omega,D}\otimes \dd\tau_{\Sigma}(\mu_{\mm{GFF}}^{\check{\Omega}})(\phi^{D}_{\Omega},\varphi)\\
    &=\detz(\Delta_{\check{\Omega}}+m^2)^{-\frac{1}{2}} \mb{E}_{\mm{GFF}}^{\check{\Omega}}[\me^{-S_{\check{\Omega}}(\phi)}]=\textrm{RHS}.
  \end{align*}
  We arrive at the proof.
\end{proof}

\begin{corr}
  [dissection gluing for $P(\phi)$ field] In the same situation as Corollary \ref{cor-disec-dens}, we have
  \begin{equation}
    \int_{}^{}\mathcal{A}_{\Omega_1}^P(\varphi)\mathcal{A}_{\Omega_2}^P(\varphi)\dd\mu_{2\mn{D}}^{\Sigma}(\varphi) \detz(2\mn{D}_{\Sigma})^{-\frac{1}{2}}=Z_M.
    \label{}
  \end{equation}
\end{corr}

\begin{proof}
  One verifies that for~$\Omega=\Omega_1\sqcup \Omega_2$,
  \begin{equation}
    \mathcal{A}_{\Omega}^P(\varphi,\varphi)=\mathcal{A}_{\Omega_1}^P(\varphi)\mathcal{A}_{\Omega_2}^P(\varphi).
    \label{}
  \end{equation}
  This in fact also comes directly from the definition as
  \begin{equation}
    S_{|\Omega|^2}(\phi_{\mm{GFF}}^{|\Omega|^2})=S_{|\Omega_1|^2}(\phi_{\mm{GFF}}^{|\Omega_1|^2})+S_{|\Omega_2|^2}(\phi_{\mm{GFF}}^{|\Omega_2|^2})
    \label{}
  \end{equation}
  pointwise almost surely with~$\phi_{\mm{GFF}}^{|\Omega|^2}=\phi_{\mm{GFF}}^{|\Omega_1|^2}+\phi_{\mm{GFF}}^{|\Omega_2|^2}$. One then follows a verbatim reasoning as for Corollary \ref{cor-disect-free}.
\end{proof}

\subsection{Gluing in the General Case}\label{sec-gluing-final}

\noindent Now let us put ourselves in the general situation of (iii) of Definition \ref{def-segal-2}. Suppose~$\Omega\in \Mor(\varnothing,\mf{\Sigma}_2\sqcup \mf{\Sigma}_2^*\sqcup\mf{\Sigma}_1)$ and that~$\rho:\mf{\Sigma}_2\lto \mf{\Sigma}_2^*$ is an orientation reversing isometry. Recall~$\Omega/\rho$ from the paragraph ``composition'' from subsection \ref{sec-segal-descript}. Denote by~$\Sigma_4$ the reflected copy of~$\Sigma_2$ (core) in~$|\Omega/\rho|^2$, see top-right picture in figure \ref{fig-general-trace}.

\begin{notation}
  In general, for~$M$ a closed manifold and~$\Sigma_1$, \dots,~$\Sigma_{\ell}\subset M$ finitely many nonintersecting closed embedded hypersurfaces, we denote
  \begin{align*}
    \tau_{i\sqcup j\sqcup\cdots}&\defeq \textrm{the trace map }C^{\infty}(M)\lto C^{\infty}(\Sigma_i)\times C^{\infty}(\Sigma_j)\times\cdots,\textrm{ and its extensions,}\\
    \PI_M^{i\sqcup j\sqcup\cdots}&\defeq \textrm{the Poisson integral operator }W^{\frac{1}{2}}(\Sigma_i)\times W^{\frac{1}{2}}(\Sigma_j)\times\cdots \lto W^1(M), \\
    \mathcal{M}_{M,k\sqcup\cdots}^{i\sqcup j\sqcup\cdots}&\defeq \tau_{k\sqcup\cdots}\PI_M^{i\sqcup j\sqcup\cdots},\textrm{ the transition operator }\mathcal{D}'(\Sigma_i)\times \mathcal{D}'(\Sigma_j)\times\cdots \lto \mathcal{D}'(\Sigma_k)\times\cdots\textrm{ as in (\ref{eqn-trans-op-def}).}
  \end{align*}
  We also remind the reader of the notations set up in definitions \ref{def-DN-boundary} and \ref{def-pi-more-general}.
\end{notation}

\begin{prop}
  [Bayes with free ends] \label{prop-bayes-free-end} We have
  \begin{equation}
    \frac{\dd\tau_{1}(\mu_{\mm{GFF}}^{|\Omega/\rho|^2})}{\dd \mu_{2\mn{D}}^{\Sigma_1}}(\varphi_1)
    \bigg( 
    \frac{\dd \big((\mathcal{M}_{|\Omega/\rho|^2,2}^{1}\varphi_1)_*\mu_{\DN}^{\Sigma_2,\Omega/\rho,D}\big)}{\dd\mu_{2\mn{D}}^{ \Sigma_2}}(x)
    \bigg)^2 =\frac{\detz(\DN_{|\Omega/\rho|^2}^{\Sigma_2\sqcup\Sigma_4})^{\frac{1}{2}}}{\detz(\DN_{|\Omega|^2}^{\Sigma_2\sqcup\Sigma_2^*})^{\frac{1}{2}}}
    \frac{\dd\tau_{2\sqcup 2^*\sqcup 1}(\mu_{\mm{GFF}}^{|\Omega|^2})}{\dd \mu_{2\mn{D}}^{2\sqcup 2^*\sqcup 1}}(x,x,\varphi_1)
  \end{equation}
\end{prop}

\begin{figure}
    \centering
    \includegraphics[width=0.8\linewidth]{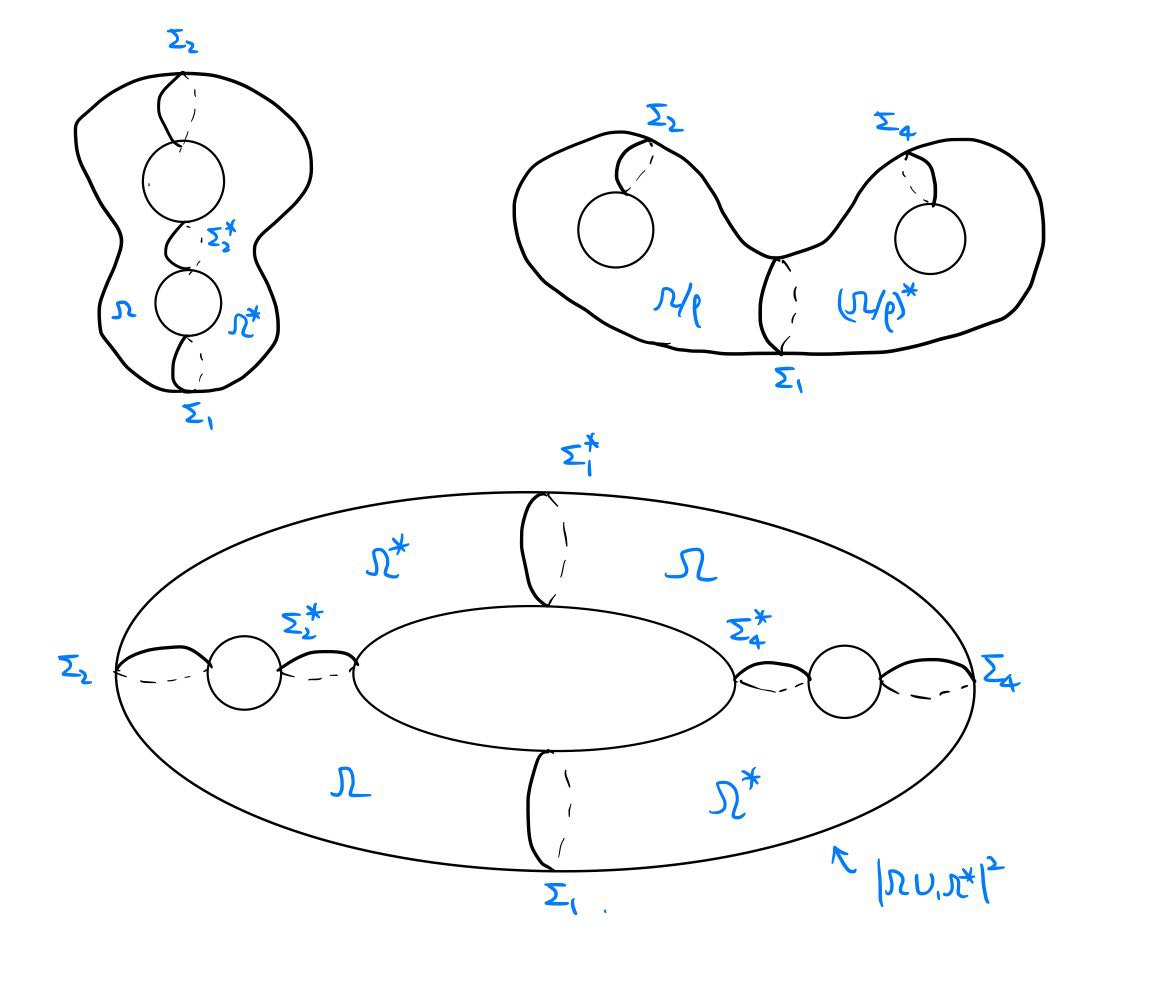}
    \caption{two ways of taking trace}
    \label{fig-general-trace}
\end{figure}

\begin{proof}
  Indeed, by reflection symmetry across~$\Sigma_1$ where~$\Sigma_2$ gets identified with~$\Sigma_4$, we have
  \begin{equation}
  \mathcal{M}_{|\Omega/\rho|^2,2}^{1}\varphi_1 =
\mathcal{M}_{|\Omega/\rho|^2,4}^{1}\varphi_1,\quad\textrm{and}\quad
\mu_{\DN}^{\Sigma_2,\Omega/\rho,D}=\mu_{\DN}^{\Sigma_4,\Omega/\rho,D},
    \label{}
  \end{equation}
  under the symmetry. Also by symmetry (method of images)   
 , we have
  \begin{equation}
  \mathcal{M}_{|\Omega/\rho|^2,1}^{2\sqcup 4}\bnom{x}{ x}
    =\mathcal{M}_{|\Omega|^2,1}^{2\sqcup 2^*}\bnom{x}{ x},
    \label{}
  \end{equation}
  because both of them expresses the associated Dirichlet data on~$\Sigma_1$ of the solution of the Helmholtz boundary value problem over $\Omega$ with Dirichlet data equal to~$x$ on both~$\Sigma_2$,~$\Sigma_2^*$ and Neumann condition (zero normal derivative) on~$\Sigma_1$. Thus,
  \begin{align*}
    \textrm{LHS}&=\frac{\dd\tau_{2\sqcup 4\sqcup 1}(\mu_{\mm{GFF}}^{|\Omega/\rho|^2})}{\dd \mu_{2\mn{D}}^{2\sqcup 4\sqcup 1}}(x,x,\varphi_1) \tag{formula (\ref{eqn-bayes-gff-1}) backward}\\
    &=\frac{\dd\tau_{2\sqcup 4}(\mu_{\mm{GFF}}^{|\Omega/\rho|^2})}{\dd \mu_{2\mn{D}}^{\Sigma_2\sqcup \Sigma_4}}(x,x)
    \frac{\dd \big((\mathcal{M}_{|\Omega/\rho|^2,1}^{2\sqcup 4}[\smx{x\\ x}])_*\mu_{2\DN}^{\Sigma_1,\Omega,D}\big)}{\dd\mu_{2\mn{D}}^{ \Sigma_1}}(\varphi_1) \tag{formula (\ref{eqn-bayes-gff-2}) forward}\\
    &=\frac{\detz(\DN_{|\Omega/\rho|^2}^{\Sigma_2\sqcup\Sigma_4})^{\frac{1}{2}}}{\detz(\DN_{|\Omega|^2}^{\Sigma_2\sqcup\Sigma_2^*})^{\frac{1}{2}}}
   \frac{\dd\tau_{2\sqcup 2^*}(\mu_{\mm{GFF}}^{|\Omega|^2})}{\dd \mu_{2\mn{D}}^{\Sigma_2\sqcup \Sigma_2^*}}(x,x)
   \frac{\dd \big((\mathcal{M}_{|\Omega|^2,1}^{2\sqcup 2^*}[\smx{x\\ x}])_*\mu_{2\DN}^{\Sigma_1,\Omega,D}\big)}{\dd\mu_{2\mn{D}}^{ \Sigma_1}}(\varphi_1) \tag{\#}\\
  &=\textrm{RHS}.
  \end{align*}
We see that with the help of Bayes principle we ``glued away'' the free end and we are reduced to the situation of Proposition \ref{prop-pre-trace}. Indeed,  at step (\#) we use Proposition \ref{prop-pre-trace} twice with respect to~$|\Omega\cup_1\Omega^*|^2$ which is~$\Omega$ ``reflected twice'' by first gluing~$\Omega$ with a reflected copy along~$\Sigma_1$ to get~$\Omega\cup_1\Omega^*$, and then reflect and glue again along~$\Sigma_2\sqcup\Sigma_2^*\sqcup\Sigma_4\sqcup\Sigma_4^*$, with~$\Sigma_4$,~$\Sigma_4^*$ being the reflected copies of~$\Sigma_2$,~$\Sigma_2^*$ across~$\Sigma_1$. From this we get
\begin{align*}
  \frac{\dd\tau_{2\sqcup 4}(\mu_{\mm{GFF}}^{|\Omega/\rho|^2})}{\dd \mu_{2\mn{D}}^{\Sigma_2\sqcup \Sigma_4}}(x,x)&=
  \frac{\detz(\DN_{|\Omega/\rho|^2}^{\Sigma_2\sqcup\Sigma_4})^{\frac{1}{2}}}{\detz(2\DN_{|\Omega\cup_1 \Omega^*|^2}^{2\sqcup 2^*\sqcup 4\sqcup 4^*})^{\frac{1}{4}}}
  \bigg( 
    \frac{\dd\tau_{2\sqcup 2^*\sqcup 4\sqcup 4^*}(\mu_{\mm{GFF}}^{|\Omega\cup_1 \Omega^*|^2})}{\dd\mu_{2\mn{D}}^{2\sqcup 2^*\sqcup 4\sqcup 4^*}}(x,x,x,x)
    \bigg)^{\frac{1}{2}}\\ 
    &=\frac{\detz(\DN_{|\Omega/\rho|^2}^{\Sigma_2\sqcup\Sigma_4})^{\frac{1}{2}}}{\detz(\DN_{|\Omega|^2}^{\Sigma_2\sqcup\Sigma_2^*})^{\frac{1}{2}}}
   \frac{\dd\tau_{2\sqcup 2^*}(\mu_{\mm{GFF}}^{|\Omega|^2})}{\dd \mu_{2\mn{D}}^{\Sigma_2\sqcup \Sigma_2^*}}(x,x),
\end{align*}
  finishing the proof.
\end{proof}

\begin{corr}
  [free-end gluing for free field] \label{cor-real-glue-free} We have
  \begin{equation}
    \int_{}^{}\mathcal{A}_{\Omega}^0(\varphi,\varphi,\psi)\dd\mu_{2\mn{D}}^{\Sigma_2}(\varphi)\detz(2\mn{D}_{\Sigma_2})^{-\frac{1}{2}}=\mathcal{A}_{\Omega/\rho}^0(\psi),
    \label{}
  \end{equation}
  for almost every $\psi\sim \mu_{2\mn{D}}^{\Sigma_1}$.
\end{corr}

\begin{proof}
  The key point is to disintegrate $\mathcal{A}_{\Omega}^0(\varphi,\varphi,\psi)$ into a part involving only $\varphi$, which one could ``integrate out'' cleanly, multiplied by a part independent of $\varphi$, using Proposition \ref{prop-bayes-free-end}. Indeed,
  \begin{align*}
    \textrm{LHS}&=\frac{\detz(\Delta_{|\Omega|^2}+m^2)^{-\frac{1}{4}}}{\detz(2\mn{D}_{\Sigma_2})^{-\frac{1}{2}}\detz( 2\mn{D}_{\Sigma_1})^{-\frac{1}{4}}}
    \int_{}^{}\bigg( 
    \frac{\dd\tau_{2\sqcup 2^*\sqcup 1}(\mu_{\mm{GFF}}^{|\Omega|^2})}{\dd \mu_{2\mn{D}}^{2\sqcup 2^*\sqcup 1}}(\varphi,\varphi,\psi)
    \bigg)^{\frac{1}{2}} \dd\mu_{2\mn{D}}^{\Sigma_2}(\varphi) \detz(2\mn{D}_{\Sigma_2})^{-\frac{1}{2}} \\
    &=\frac{\detz(\DN_{|\Omega|^2}^{\Sigma_2\sqcup\Sigma_2^*})^{\frac{1}{4}}}{\detz(\DN_{|\Omega/\rho|^2}^{\Sigma_2\sqcup\Sigma_4})^{\frac{1}{4}}}
    \frac{\detz(\Delta_{|\Omega|^2}+m^2)^{-\frac{1}{4}}}{\detz(2\mn{D}_{\Sigma_1})^{-\frac{1}{4}}}
    \bigg( \frac{\dd\tau_{1}(\mu_{\mm{GFF}}^{|\Omega/\rho|^2})}{\dd \mu_{2\mn{D}}^{\Sigma_1}}(\psi) \bigg)^{\frac{1}{2}}
    \underbrace{\int_{}^{}\frac{\dd \big((\mathcal{M}_{|\Omega/\rho|^2,2}^{1}\psi)_*\mu_{\DN}^{\Sigma_2,\Omega/\rho,D}\big)}{\dd\mu_{2\mn{D}}^{ \Sigma_2}}(\varphi)
    \dd\mu_{2\mn{D}}^{\Sigma_2}(\varphi)}_{=~ 1}  \\
    &=\detz(\DN_{|\Omega/\rho|^2}^{\Sigma_2\sqcup\Sigma_4})^{-\frac{1}{4}} \detz(\Delta_{|\Omega|^2\setminus \Sigma_2\sqcup\Sigma_2^*,D}+m^2)^{-\frac{1}{4}}
    \detz(2\mn{D}_{\Sigma_1})^{\frac{1}{4}} 
    \bigg( \frac{\dd\tau_{1}(\mu_{\mm{GFF}}^{|\Omega/\rho|^2})}{\dd \mu_{2\mn{D}}^{\Sigma_1}}(\psi) \bigg)^{\frac{1}{2}}\tag{BFK for~$|\Omega|^2$ and~$\Sigma_2\sqcup\Sigma_2^*$}\\
    &=\textrm{RHS}, \tag{BFK for~$|\Omega/\rho|^2$ and~$\Sigma_2\sqcup\Sigma_4$}
  \end{align*}
  finishing the proof.
\end{proof}

We proceed to the interacting case.

\begin{lemm}\label{lemm-equal-law-2-PI}
  For fixed~$\psi\in W^{\frac{1}{2}}(\Sigma_1)$, then the random fields
  \begin{equation}
    \phi\defeq\PI_{\Omega/\rho}^{1}\psi+\PI_{\Omega/\rho}^{2,D}\tau_{2}\phi_{\Omega/\rho}^D,\quad\textrm{and}\quad \tilde{\phi}\defeq\PI_{\Omega}^{2\sqcup 2^*\sqcup 1}[\varphi,\varphi,\psi]
    \label{}
  \end{equation}
  with~$\phi_{\Omega/\rho}^D\sim \mu_{\mm{GFF}}^{\Omega/\rho,D}$ and~$\varphi\sim (\mathcal{M}_{|\Omega/\rho|^2,2}^{1}\psi)_*\mu_{\DN}^{\Sigma_2,\Omega/\rho,D}$, follow the same law on~$\mathcal{D}'(\Omega^{\circ})$.
\end{lemm}

\begin{proof}
  Here~$\Omega^{\circ}$ is also~$(\Omega/\rho)\setminus\Sigma_2$. Indeed, both~$\phi$ and~$\tilde{\phi}$ solves the stochastic boundary value problem
  \begin{equation}
    \left\{
    \begin{array}{ll}
      (\Delta+m^2)u=0&\textrm{in }\Omega^{\circ},\\
      u|_{\Sigma_2}=u|_{\Sigma_2^*}\sim(\mathcal{M}_{|\Omega/\rho|^2,2}^{1}\psi)_*\mu_{\DN}^{\Sigma_2,\Omega/\rho,D}, & \textrm{on }\Sigma_2\textrm{ or }\Sigma_2\sqcup\Sigma_2^* \\
      u|_{\Sigma_1}=\psi,&\textrm{on }\Sigma_1,
    \end{array}
    \right.
    \label{}
  \end{equation}
  with equalities holding almost surely.
\end{proof}

\begin{corr}
  [free-end gluing for~$P(\phi)$ field]  In the situation as above, we have
  \begin{equation}
    \int_{}^{} \mathcal{A}_{\Omega}^P(\varphi,\varphi,\psi)\dd\mu_{2\mn{D}}^{\Sigma_2}(\varphi) \detz(2\mn{D}_{\Sigma_2})^{-\frac{1}{2}}=\mathcal{A}_{\Omega/\rho}^P(\psi),
    \label{}
  \end{equation}
  for almost every $\psi\sim \mu_{2\mn{D}}^{\Sigma_1}$.
\end{corr}

\begin{proof}
  Indeed,
  \begin{equation}
    \mathcal{A}_{\Omega}^P(\varphi,\varphi,\psi)=\mathcal{A}_{\Omega}^0(\varphi,\varphi,\psi)\int_{}^{}\me^{-S_{\Omega}(\phi^{D}_{\Omega}|\varphi,\varphi,\psi)}\dd\mu_{\mm{GFF}}^{\Omega,D}(\phi^{D}_{\Omega}),
    \label{}
  \end{equation}
  therefore, with the constants involving determinants working out in exactly the same way as Corollary \ref{cor-real-glue-free}, one has
  \begin{align*}
    \textrm{LHS}&=\frac{\detz(\Delta_{|\Omega/\rho|^2}+m^2)^{-\frac{1}{4}}}{\detz(2\mn{D}_{\Sigma_1})^{\frac{1}{4}}}
    \bigg( \frac{\dd\tau_{1}(\mu_{\mm{GFF}}^{|\Omega/\rho|^2})}{\dd \mu_{2\mn{D}}^{\Sigma_1}}(\psi) \bigg)^{\frac{1}{2}}
    \int_{}^{}\dd \big((\mathcal{M}_{|\Omega/\rho|^2,2}^{1}\psi)_*\mu_{\DN}^{\Sigma_2,\Omega/\rho,D}\big)(\varphi) 
    \int_{}^{} \me^{-S_{\Omega}(\phi^{D}_{\Omega}|\varphi,\varphi,\psi)}\dd\mu_{\mm{GFF}}^{\Omega,D}(\phi^{D}_{\Omega}) \\
    &=\mathcal{A}_{\Omega/\rho}^0(\psi)
    \int_{}^{}\me^{-S_{\Omega/\rho}(\phi^{D}_{\Omega}+(\phi_{\Omega/\rho}^D)_{\Sigma_2}|\psi)} \dd\mu_{\mm{GFF}}^{\Omega,D}\otimes \dd \big[(\PI_{\Omega/\rho}^{2,D}\tau_{2})_*(\mu_{\mm{GFF}}^{\Omega/\rho,D})\big](\phi^{D}_{\Omega},(\phi_{\Omega/\rho}^D)_{\Sigma_2}) \tag{Lemma \ref{lemm-equal-law-2-PI}}\\
    &=\mathcal{A}_{\Omega/\rho}^0(\psi)
    \int_{}^{}\me^{-S_{\Omega/\rho}(\phi_{\Omega/\rho}^D|\psi)} \dd\mu_{\mm{GFF}}^{\Omega/\rho,D}(\phi_{\Omega/\rho}^D)=\textrm{RHS},
  \end{align*}
  where the notation $(\phi_{\Omega/\rho}^D)_{\Sigma_2}$ is as in Proposition \ref{prop-stoc-decomp-domain}. We arrive at the proof.
\end{proof}

Now let~$\Omega_{12}\in\Mor(\mf{\Sigma}_1,\mf{\Sigma}_2)$,~$\Omega_{23}\in\Mor(\mf{\Sigma}_2,\mf{\Sigma}_3)$ and consider the double~$|\Omega_{23}\circ\Omega_{12}|^2=:|\Omega_{23}\Omega_{12}|^2$. In this case the result could be seen as a special case of the previous one, where $\Omega=\Omega_{12}\sqcup\Omega_{23}\in\Mor(\varnothing, \mf{\Sigma}_1^*\sqcup \mf{\Sigma}_2\sqcup \mf{\Sigma}_2^*\sqcup \mf{\Sigma}_3)$ has two disjoint components. We shall re-state these results without proofs.

\begin{corr}\label{prop-pre-glu}
  We have
  \begin{align*}
    &\quad\frac{\dd\tau_{1\sqcup 3}(\mu_{\mm{GFF}}^{|\Omega_{23}\Omega_{12}|^2})}{\dd \mu_{2\mn{D}}^{\Sigma_1\sqcup \Sigma_3}}(\varphi_1,\varphi_3)
    \left( \frac{\dd \big((\mathcal{M}_{|\Omega_{23}\Omega_{12}|^2,2}^{1\sqcup 3}[\smx{\varphi_1\\ \varphi_3}])_*\mu_{2\DN}^{\Sigma_2,\Omega_{23}\Omega_{12},D}\big)}{\dd\mu_{2\mn{D}}^{ \Sigma_2}}(x) \right)^2 \\
    =&\quad \frac{\detz(\DN_{|\Omega_{23}\Omega_{12}|^2}^{\Sigma_2\sqcup\Sigma_4})^{\frac{1}{2}}}{\detz(2\DN_{\Omega_{12}}^{\Sigma_2,N})^{\frac{1}{2}}\detz(2\DN_{\Omega_{23}}^{\Sigma_2,N})^{\frac{1}{2}}}
    \frac{\dd\tau_{\Sigma_1\sqcup \Sigma_2}(\mu_{\mm{GFF}}^{|\Omega_{12}|^2}) }{\dd \mu_{2\mn{D}}^{\Sigma_1\sqcup \Sigma_2}}(\varphi_1,x) 
    \frac{\dd\tau_{\Sigma_2\sqcup \Sigma_3}(\mu_{\mm{GFF}}^{|\Omega_{23}|^2}) }{\dd \mu_{2\mn{D}}^{\Sigma_2\sqcup \Sigma_3}}(x,\varphi_3). \quad\Box
  \end{align*}
\end{corr}

\begin{corr}
  [composition for~$P(\phi)$ field]  In the situation as above, we have
  \begin{equation}
    \int_{}^{} \mathcal{A}_{\Omega_{12}}^P(\psi_1,\varphi)\mathcal{A}_{\Omega_{23}}^P(\varphi,\psi_3)\dd\mu_{2\mn{D}}^{\Sigma_2}(\varphi) \detz(2\mn{D}_{\Sigma_2})^{-\frac{1}{2}}=\mathcal{A}_{\Omega_{23}\circ\Omega_{12}}^P(\psi_1,\psi_3),
    \label{}
  \end{equation}
  for almost every $(\psi_1,\psi_3)\sim \mu_{2\mn{D}}^{\Sigma_1\sqcup\Sigma_3}$. \hfill~$\Box$
\end{corr}

\section{Periodic Cover = Spin Chain} \label{sec-period-cover}

\subsection{Geometric setting}
\label{s:abeliancovers}

\noindent Our set-up corresponds to example 1 in Bergeron \cite{Bergeron}. Let~$M$ be a closed oriented Riemannian surface\footnote{In particular its~$\mb{Z}$-homology groups are non-torsion.} of genus~$g\ge 1$ and~$\Sigma\subset M$ an embedded \textit{primitive} closed geodesic whose $\mb{Z}$-homology class is non-trivial (exists by a classical theorem of E. Cartan). Necessarily,~$\Sigma$ is nondissecting.\footnote{If~$\Sigma$ dissects~$M$ into~$M_+^{\circ}\sqcup M_-^{\circ}$, then~$\Sigma=\partial M_+$; now closed 1-forms integrate to zero over $\Sigma$ by Stokes theorem, thus $\Sigma$ is null-homologous via de Rham's theorem. Alternatively, note~$\partial:H_2(M_+,\Sigma;\mb{Z})\lto H_1(\Sigma;\mb{Z})$ in the exact sequence of the pair $(M_+,\Sigma)$ is surjective, producing the fundamental class.} We consider the covering space~$M_{\infty}^{\Sigma}\lto M$ over~$M$ given by a (normal) subgroup~$\ker\rho$ of the fundamental group~$\pi_1(M)$ where~$\rho$ is the map
\begin{equation}
  \begin{tikzcd}
    \rho:\pi_1(M) \ar[r,"\mm{Ab}"] &[+10pt] H_1(M;\mb{Z}) \ar[r, " I(-{,}{[}\Sigma{]})"] &[+10pt] \mb{Z},
  \end{tikzcd}
  \label{}
\end{equation}
where the first is Abelianization and~$I(-,[\Sigma])$ is the \textsf{oriented intersection number}\footnote{$I([\gamma],[\Sigma])=D(D^{-1}[\gamma]\smile D^{-1}[\Sigma])=\int_{\gamma}^{}\eta_{\Sigma}\in\mb{Z}$, where~$D^{-1}$ is the Poincar\'e dual map and~$\eta_{\Sigma}$ is a smooth bump 1-form supported in a tubular neighborhood of~$\Sigma$ such that~$\int_{\Sigma}^{}\alpha=\int_{M}^{}\alpha\wedge\eta_{\Sigma}$ for any 1-form~$\alpha$.} with~$\Sigma$, which is surjective (since~$\Sigma$ is primitive). In other words, we put~$M_{\infty}^{\Sigma}=\ker \rho\setminus \tilde{M}$ where~$\tilde{M}$ is the universal cover of~$M$ and~$\ker\rho$ acts on~$\tilde{M}$ as deck transformations. Equip~$M_{\infty}^{\Sigma}$ with the covering metric (thus deck transformations act by isometries).

Geometrically,~$M_{\infty}^{\Sigma}$ can be understood as first cutting~$M$ along~$\Sigma$ and obtaining the surface~$\Omega:=M\setminus\Sigma$ with boundaries~$\Sigma_{\mm{in}}\sqcup \Sigma_{\mm{out}}$ where~$\Sigma_{\mm{in}} \cong \Sigma_{\mm{out}}\cong \Sigma$, and gluing~$\Omega$ periodically where each~$\Sigma_{\mm{out}}$ is glued to the ``next''~$\Sigma_{\mm{in}}$. Indeed, the class~$[\gamma]\in \pi_1(M)$ of a loop~$\gamma$ is in~$\ker \rho$ iff~$I([\gamma],[\Sigma])=0$; in other words, these loops are exactly those which are ``not cut'', i.e. lifts to a loop on~$M_{\infty}^{\Sigma}$, and loops which do intersect~$\Sigma$ are lifted to segments whose end points are related by a deck transformation, i.e. they are ``cut''.

Now, for~$N\in\mb{N}$, compose~$\rho$ further with the mod~$N$ map~$\mb{Z}\lto \mb{Z}_N=:\mb{Z}/N\mb{Z}$ and denote it by~$\rho_N$, and let the covering space of~$M$ corresponding to~$\ker\rho_N$ be~$M_N^{\Sigma}$. Since~$\ker\rho\subset \ker\rho_N$,~$M_{\infty}^{\Sigma}$ also covers~$M_N^{\Sigma}$. Geometrically, this corresponds to closing the surface after gluing~$N$ copies of~$\Omega$ --- loops that intersect~$\Sigma$~$N$-times are now lifted to a ``big loop'' in~$M_N^{\Sigma}$.

\begin{figure}
    \centering
    \includegraphics[width=0.8\linewidth]{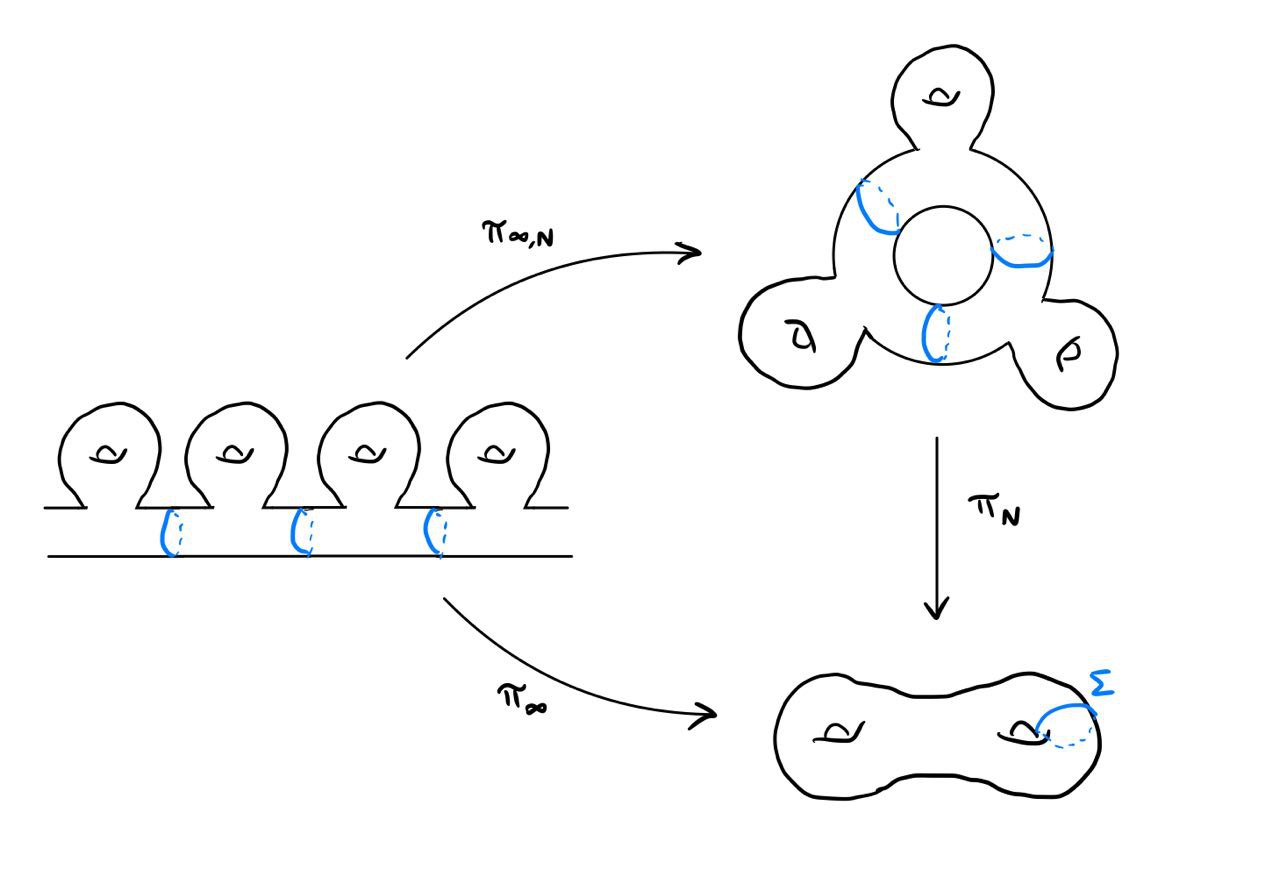}
    \caption{periodic covering and cyclic covering}
    \label{fig-period-cover}
\end{figure}

\begin{def7}
  We also say that the sequence of covers~$(M_N^{\Sigma})_N$ \textsf{converges} to~$M_{\infty}^{\Sigma}$.
\end{def7}

That~$\Omega:=M\setminus\Sigma$ should be understood for what follows.

\subsection{Continued Introduction of Spin Chain Example}\label{sec-spin-chain-cont}

\noindent Here we continue our discussion of the circular spin chain proposed in section \ref{sec-intro-spin-chain}, in particular the equation (\ref{eqn-spin-chain-part-func-trace}).  More generally, one can insert ``nice'' functionals~$F_1$, \dots,~$F_k$ in between at the sites~$1\le i_1<\cdots <i_k\le N$, then
\begin{equation}
  \int_{\mb{R}^N}^{}F_k(\sigma(i_k))\cdots F_1(\sigma(i_1))\me^{-S(\sigma)} \dd^N\sigma =\ttr_{L^2(\mb{R})}\left( T^{N+1-i_k}F_k T^{i_k-i_{k-1}}\cdots F_1 T^{i_1-1} \right),
  \label{}
\end{equation}
where~$F_j$ are thought of as multiplication operators. The evaluation
\begin{equation}
  \left.
  \def\arraystretch{2.5}
  \begin{array}{rcl}
     F_k\otimes\cdots\otimes F_1&\longmapsto& \ddp \frac{1}{\mathcal{Z}(N)}\int_{\mb{R}^N}^{}F_k(\sigma(i_k))\cdots F_1(\sigma(i_1))\me^{-S(\sigma)} \dd^N\sigma \\
     &= & \ddp \frac{\ttr_{L^2(\mb{R})}\left( T^{N+1-i_k}F_k T^{i_k-i_{k-1}}\cdots F_1 T^{i_1-1} \right)}{\ttr_{L^2(\mathbb{R})}( T^{N} )}
  \end{array}
  \right.
  \label{eqn-spin-chain-gibbs-state}
\end{equation}
is said to define a \textsf{Gibbs state} of our spin chain on~$\mb{Z}_N$. Alternatively this is the expected value functional under the discrete Gibbs measure (\ref{eqn-spin-chain-gibbs-part-func}).

\begin{deef}\label{def-thermo-limit}
  We say that a Gibbs state exists in the \textsf{thermodynamic limit} if the second expression in (\ref{eqn-spin-chain-gibbs-state}) tends to a limit as~$N\to \infty$ for bounded functionals~$F_1$, \dots,~$F_k\in L^{\infty}(\mb{R})$.
\end{deef}

We can see from (\ref{eqn-spin-chain-trans-kernel}) that the operator~$T$ is not just bounded, smoothing, but its kernel is strictly positive. Such an operator has a special property which we call the \textsf{Perron-Frobenius} property, referring to the consequence of proposition \ref{prop-perron-frob} below. In particular, this would ensure that the thermodynamic limit does exist, as we shall see in corollaries \ref{cor-segal-trans-exp-decay}, \ref{cor-log-trace-asymp}, and \ref{cor-gibbs-state-limit}.

 As announced in the introduction,
   Corollary \ref{cor-segal-trans-exp-decay} could be equivalently understood as saying that for compactly supported observables $(F,G)$, and $(\tau\sigma)(i):=\sigma(i+1) $ the shift operator, then
\begin{eqnarray*}
  \mathbb{E}[F(\tau^k \sigma) G(\sigma)]= \mathbb{E}[F] \mathbb{E}[G] + \mathcal{O}(\alpha^k)  
\end{eqnarray*}
with $\alpha<1$ the same as in corollary \ref{cor-segal-trans-exp-decay}. Here the expected value should be thought of as coming from a Gibbs measure on the infinite path space~$\mb{R}^{\mb{Z}}$ over~$\mb{Z}$. Indeed, by definition \ref{def-thermo-limit}, this is exactly the vague limit of the finite dimensional (periodic) Gibbs measures over~$\mb{Z}_N$. We say that with respect to this Gibbs measure the shift operator  is \textsf{exponentially mixing} which is exactly saying the $P(\phi)_2$ Gibbs state constructed has a \textbf{mass gap}.

With transfer operator being Perron-Frobenius, the partition functions as in (\ref{eqn-spin-chain-gibbs-part-func}) also enjoy explicit asymptotics. We will explore a consequence in the case of periodic surfaces in the last section.

\begin{exxx}
  In the case $P(\sigma)=m^2\sigma^2$, the spin chain is the discrete massive GFF. We have an exact formula for the partition function $\mathcal{Z}(N)=\prod_{k=0}^{N-1} (1+m^2 -\cos(2\pi\frac{k}{N}  ))^{-\frac{1}{2}}$
hence
\begin{eqnarray*}
  \lim_{N\rightarrow+\infty}    \frac{1}{N} \log \left( \mathcal{Z}(N)\right) = -\frac{1}{2} \int_0^1 \log(1+m^2-\cos(2\pi x))\dd x.
\end{eqnarray*}
\end{exxx}

\subsection{Perron-Frobenius Property and Gibbs State}\label{sec-perron-frob-gibbs}

  A large part of this section follow from the ``properties of a Perron-Frobenius operator''. But we include them here as they form an integral part of the discussion of physical phenomena.

\begin{deef}
  An operator~$A$ on~$L^2(Q,\mu)$ of some measure space~$(Q,\mu)$ has \textsf{strictly positive kernel} if for any nonnegative~$F\in L^2(Q,\mu)$ such that~$\nrm{F}_{L^2}\ne 0$ we have~$AF>0$ almost surely.
\end{deef}

\begin{prop}
  [Perron-Frobenius, \cite{GJ} page 51] \label{prop-perron-frob} If~$A$ on~$L^2(Q,\mu)$ has strictly positive kernel, and~$\lambda=\nrm{A}$ is an eigenvalue of~$A$, then~$\lambda$ is simple, and the corresponding eigenvector can be chosen to be strictly positive almost surely. \hfill~$\Box$
\end{prop}

By our definition of the amplitudes~$\mathcal{A}_{\Omega}$ as the square-root of the Radon-Nikodym density between two mutually absolutely continuous positive finite measures ($\me^{-S_{\Omega}}$ is almost surely positive since~$S_{\Omega}$ is a real-valued random variable, recall also remark \ref{rem-zeta-det-positive} that the zeta-determinants are positive), we get immediately

\begin{lemm}
  For any traceable cobordism~$\Omega\in \Mor(\Sigma,\Sigma)$, the Segal transfer operator~$U_{\Omega}:L^2(\mathcal{D}'(\Sigma),\mu_{2\mn{D}})\lto L^2(\mathcal{D}'(\Sigma),\mu_{2\mn{D}})$ has strictly positive kernel. \hfill~$\Box$
\end{lemm}

We deduce that~$U_{\Omega}$ has a simple top eigenvalue~$\lambda_0=\nrm{U_{\Omega}}$ spanned by a normalized, almost surely strictly positive eigenvector~$\Omega_0\in L^2(\mathcal{D}'(\Sigma),\mu_{2\mn{D}})$. Alternatively speaking $U_{\Omega}$ has a \textsf{spectral gap}. We get
\begin{corr}\label{cor-segal-trans-exp-decay}
  Denote~$\wh{U}_{\Omega}:=\lambda_0^{-1}U_{\Omega}$, let~$\lambda_1$ be the eigenvalue of~$U_{\Omega}$ with next largest modulus, thus~$\lambda_0>|\lambda_1|$, and put~$\alpha:=|\lambda_1|/\lambda_0 <1$. Then for any~$F$,~$G\in L^2(\mathcal{D}'(\Sigma),\mu_{2\mn{D}})$, we have
  \begin{equation}
    \big|\bank{F,\wh{U}_{\Omega}^N G}-\bank{F,\Omega_0}\bank{\Omega_0,G}\big| \le \alpha^N \nrm{F}\nrm{G}.
    \label{}
  \end{equation}
\end{corr}

\begin{proof}
  Note~$\sank{F,\wh{U}_{\Omega}^N G}-\ank{F,\Omega_0}\ank{\Omega_0,G}=\sank{\Pi_0^{\perp}F,\wh{U}_{\Omega}^N \Pi_0^{\perp}G}$ where~$\Pi_0^{\perp}$ is the orthogonal projection onto the complement of~$\spn\{\Omega_0\}$, where~$\wh{U}_{\Omega}$ has norm~$\alpha<1$.
\end{proof}

The above corollary is exactly stating that the Gibbs state we defined has a mass gap, in perfect analogy with $1$D spin chains with local interactions which are well known to have no phase transitions.

\begin{corr}\label{cor-log-trace-asymp}
  We have
  \begin{equation}
    \lim_{N\to\infty}\frac{1}{N}\log\ttr (U_{\Omega}^N) =\log\lambda_0.
    \label{}
  \end{equation}
\end{corr}

\begin{proof}
  Without loss of generality let~$N\ge 2$ so each~$U_{\Omega}^N$ is trace class. On one hand
  \begin{equation}
    \ttr(U_{\Omega}^N)\le \bnrm{U_{\Omega}^2}_{\mm{tr}}\bnrm{U_{\Omega}^{N-2}}\le \bnrm{U_{\Omega}^2}_{\mm{tr}}\lambda_0^{N-2}.
    \label{eqn-log-trace-upper-bound}
  \end{equation}
  On the other, we decompose~$\mathcal{H}_{\Sigma}=\spn\{\Omega_0\}\oplus \spn\{\Omega_0\}^{\perp}$ where
  \begin{align*}
    \ttr(U_{\Omega}^N)&=\lambda_0^N +\ttr(U_{\Omega}^N|_{\spn\{\Omega_0\}^{\perp}}) \\
    &\ge \lambda_0^N -\bnrm{U_{\Omega}^2|_{\spn\{\Omega_0\}^{\perp}}}_{\mm{tr}}|\lambda_1|^{N-2}\\
    &=\lambda_0^N\left( 1-C_1 \alpha^{N-2}\lambda_0^{-2} \right).
  \end{align*}
  This and (\ref{eqn-log-trace-upper-bound}) gives the result after taking~$N\to\infty$.
\end{proof}

Corollary \ref{cor-segal-trans-exp-decay} and the proof of corollary \ref{cor-log-trace-asymp} implies

\begin{corr}\label{cor-gibbs-state-limit}
  For any bounded operator~$F\in \mathcal{L}(\mathcal{H}_{\Sigma})$ we have
  \begin{equation}
    \lim_{N\to \infty}\frac{\ttr(U_{\Omega}^{N-L}F)}{\ttr(U_{\Omega}^N)}=\frac{1}{\lambda_0^L}\ank{\Omega_0,F\Omega_0}=\frac{\ank{\Omega_0,F\Omega_0}}{\ank{\Omega_0,U_{\Omega}^L \Omega_0}}.
    \label{}
  \end{equation}
\end{corr}
\begin{proof}
  Remember that~$N\gg L$. Thus by~$\mathcal{H}_{\Sigma}=\spn\{\Omega_0\}\oplus \spn\{\Omega_0\}^{\perp}$ and \cite{Sim1} theorem 2.14 we have
  \begin{equation}
    \lambda_0^{-N+L} \ttr\big(U_{\Omega}^{N-L}F\big)=\bank{\Omega_0,\wh{U}_{\Omega}^{N-L}F\Omega_0}+ \ttr\big(\wh{U}_{\Omega}^{N-L}F|_{\spn\{\Omega_0\}^{\perp}}\big)
  \end{equation}
  with
  \begin{equation}
    \big|\ttr\big(\wh{U}_{\Omega}^{N-L}F|_{\spn\{\Omega_0\}^{\perp}}\big)\big| \le \nrm{F} \bnrm{\wh{U}_{\Omega}^{N-L}|_{\spn\{\Omega_0\}^{\perp}}} \to 0.
    \label{}
  \end{equation}
  Since $\sank{\Omega_0,\wh{U}_{\Omega}^{N-L}F\Omega_0}\to \sank{\Omega_0,F\Omega_0}$ by corollary \ref{cor-segal-trans-exp-decay} and~$\ttr(U_{\Omega}^N)\asymp \lambda_0^N$ by corollary \ref{cor-log-trace-asymp}, we obtain the result.
\end{proof}

We arrive at the conclusion that, for functionals~$F_1$, \dots,~$F_k\in L^{\infty}(\mathcal{D}'(\Sigma),\mu_{2\mn{D}})$ and integers~$1\le i_1<\cdots<i_k$, the evaluation
\begin{equation}
  \left.
  \def\arraystretch{2.5}
  \begin{array}{rcl}
     F_k\otimes\cdots\otimes F_1 &\longmapsto& \ddp \lim_{N\to \infty}\frac{\ttr\big( U_{\Omega}^{N+1-i_k}F_k U_{\Omega}^{i_k-i_{k-1}}\cdots F_1 U_{\Omega}^{i_1-1} \big)}{\ttr\big( U_{\Omega}^{N} \big)} \\
     &= & \ddp \frac{\bank{\Omega_0,F_k U_{\Omega}^{i_k-i_{k-1}}\cdots F_1 U_{\Omega}^{i_1-1}\Omega_0}}{\bank{\Omega_0,U_{\Omega}^{i_k-1} \Omega_0}}
  \end{array}
  \right.
  \label{}
\end{equation}
defines a Gibbs state in the thermodynamic limit on a~$\mathcal{D}'(\Sigma)$-valued~$\mb{Z}$-spin chain.

\begin{def7}
  In fact, a valid functional~$F\in L^{\infty}(\mathcal{D}'(\Sigma),\mu_{2\mn{D}})$ could be given by
  \begin{equation}
    F(\psi)\defeq \int_{}^{}\mathcal{A}_{\Omega}^0(\psi,\varphi)\dd\mu_{2\mn{D}}^{\Sigma}(\varphi)\int_{}^{}\tilde{F}(\phi^{D}_{\Omega},\psi,\varphi)\me^{-S_{\Omega}(\phi^{D}_{\Omega}|\psi, \varphi)}\dd\mu_{\mm{GFF}}^{\Omega,D}(\phi^{D}_{\Omega}),
    \label{}
  \end{equation}
  for~$\tilde{F}\in L^{\infty}(\mathcal{D}'(\Omega^{\circ}),\mu_{\mm{GFF}}^{\Omega,D}\otimes \mu_{2\mn{D}}^{\Sigma\sqcup\Sigma})$, with~$\mu_{2\mn{D}}^{\Sigma\sqcup\Sigma}$ considered as the induced measure on~$\mathcal{D}'(\Omega^{\circ})$ via~$\PI_{\Omega}^{\Sigma\sqcup\Sigma}$. Hence the Gibbs state actually extends to the continuum~$M_{\infty}^{\Sigma}$.
\end{def7}

\subsection{Asymptotic of Partition Function}\label{sec-segal-end-asymp-part-func}

\noindent Now we assume the surface~$\Omega\in\Mor(\mf{\Sigma},\mf{\Sigma})$ be \textsf{reflection symmetric}, which means~$\Omega=\tilde{\Omega}^*\circ\tilde{\Omega}$ for some~$\tilde{\Omega}\in\Mor(\mf{\Sigma},\mf{\Sigma'})$, in the notations of Lemma \ref{lemm-segal-transfer}. In other words there is an isometric involution~$\Theta:\Omega\lto \Omega$ whose fixed point set is exactly~$\Sigma'$, and exchanges the two components of~$\partial\Omega$. In this case~$U_{\Omega}=U_{\tilde{\Omega}^*}\circ U_{\tilde{\Omega}}$ and \cite{Sim3} theorem 3.8.5 applies, namely~$\ttr(U_{\Omega})=\ttr_{\rho}(U_{\Omega})$, as well as for~$U_{\Omega}^N$. From Corollary \ref{cor-P(phi)-trace} and Corollary \ref{cor-log-trace-asymp} we then deduce that
\begin{equation}
  \lim_{N\to\infty}\frac{1}{N}\log(Z_{M_N^{\Sigma}})=\log \lambda_0.
  \label{}
\end{equation}
Thus we have proved theorem \ref{thrm-intro-main-2}. This applies in particular to the free case where~$Z_{M_N^{\Sigma}}=\detz(\Delta_{M_N^{\Sigma}}+m^2)^{-\frac{1}{2}}$.

\begin{def7}\label{rem-final-open-ques}
  There arises the interesting question of how~$\lambda_0$ (or~$\log \lambda_0$) would depend on the geometry of~$\Omega$. Very crudely one would expect~$\log\lambda_0 \propto \vol(\Omega)$, since in our case the corresponding~$\tilde{\lambda}_0$ for~$N\Omega$ ($N$ copies of~$\Omega$ glued) is just~$\lambda_0^N$. A more precise formula in the general, non-periodic case seems desirable.
\end{def7}

\chapter[Spectral Cut-off, Locality, and RP]{Spectral Cut-off, Locality, and Reflection Positivity}\label{chap-cut-off}

This Chapter is adapted form \cite{BDFL2} joint with Ismael Bailleul, Nguyen Viet Dang, Léonard Ferdinand, Gaëtan Leclerc.

\section{Introduction and context}
\subsection{Markov property and reflection positivity of a random field}
Let $(M,g)$ be a smooth, closed, compact Riemannian manifold or $M=\mb{R}\times \Sigma$ where $\Sigma$ is a complete Riemannian manifold endowed with the product metric $\dd  t^2+g_\Sigma$. We  denote by $\Delta_g$ the positive Laplace-Beltrami operator acting on $C^\infty(M)$. In case $M$ is compact, we write $(e_\lambda)_{\lambda\in \mathrm{spec}(\Delta_g)}$ for the $L^2$ basis of eigenfunctions of $\Delta_g$ and 
$$
\mathrm E_{\leqslant \Lambda} \defeq \mathrm{span}(e_\lambda)_{\lambda\leqslant\Lambda}\subset C^\infty(M)\,.
$$
The massive Gaussian Free Field (GFF in the sequel) $\Phi$ of law $\mu_{\mathrm{GFF}}$ is defined
on $M$ compact as
the random series
$$ \Phi\defeq\sum_{\lambda\in\mathrm{spec}(\Delta_g)}\frac{c_\lambda}{\sqrt{1+\lambda}}  e_\lambda\,,  $$
where the coefficients $c_\lambda$ are independent, identically distributed, random variables with Gaussian distribution $\CB{N}(0,1)$. On a cylinder $M=\mb{R}\times \Sigma$ the massive Gaussian free field is defined as the unique Gaussian process $\Phi$ indexed by $H^{-1}(M)$ with covariance
$$\mb{E}\big[\Phi(f)\Phi(h)\big] = \left\langle f,h\right\rangle_{H^{-1}(M)}.$$
For every $\varepsilon>0$, one can realize the Gaussian process $\Phi$ as a random variable with values in $H^{({2-d})/{2}-\varepsilon}(M)$. 
This follows essentially from the Weyl law by an argument analogous to \cite{Berestycki-Powell-2025} Theorem 1.45.
In the sequel, we thus fix $\varepsilon>0$ and work with the canonical probability space $\big(H^{(2-d)/2-\varepsilon}(M), \CB{O}, \mu_{\text{GFF}}\big)$, where $\CB{O}$ stands for the Borel $\sigma$-algebra of $\Omega=H^{(2-d)/2-\varepsilon}(M)$. One can take $\Phi(\omega)=\omega$ for every $\omega\in\Omega$.

\begin{deef}
In case of $M$ compact or a cylinder, $\Delta_g:C^\infty_c(M)\subset L^2(M)\rightarrow L^2(M) $ is essentially self-adjoint, and therefore admits a well-defined functional calculus. 
For $\Lambda>0$ we write
$$
\Pi_\Lambda\defeq1_{[0,\Lambda]}(\Delta_g)
$$ 
for the sharp spectral projector. In particular, when $M$ is compact, one has $\Pi_\Lambda:\CB D'(M)\rightarrow \EL$. Note that $\Pi_\Lambda$ is self-adjoint. Finally, we define the spectrally cut-off GFF $\Phi_\Lambda$ GFF 
by
$$\Phi_\Lambda\defeq\Pi_\Lambda(\Phi)\,, $$ 
which is a random smooth function. 
\end{deef}
Except in Section \ref{sec2}, throughout the rest of the article, without loss of generality, in our choice of cut-off free field we could replace $\Pi_\Lambda$ by $\Psi(  \Delta_g/\Lambda)$ where $\Psi$ is smooth and compactly supported. 

We will use the notion of smooth functionals  on a locally convex space such as $\CB{D}^\prime(M),C^\infty(M)$ or on some Sobolev space $H^s(M)$ as discussed in~\cite{DBLR}, along with the notion of support of a smooth functional. Both notions are recalled hereafter, we start with smooth functionals~:
\begin{deef}
Let $E$ be a locally convex space and $U\subset E$ an open subset. A function $F:U\rightarrow \mb{R}$ is \textbf{smooth} if for every $k\in \mb{N}$, $x\in U$, every $(h_1,\dots,h_k)\in E^k$, the limit
$$  \frac{\partial^kF\big(x+t_1h_1+\dots+t_kh_k\big)}{\partial t_1 \cdots\partial t_k}\Big|_{(t_1,\dots,t_k)=0}=\D^k_xF\big(h_1,\cdots,h_k\big) $$ exists and $\D^kF:U\times E^k\rightarrow \mb{R}$ is jointly continuous as a function of its arguments and linear in $h_1,\dots,h_k$.   
\end{deef}
Next, we discuss the notion of support:
\begin{deef}\label{def:suppfonc}
Take $E$ to be $C^\infty(M)$, $\CB{D}^\prime(M)$, or some Sobolev space $H^s(M)$, and let $U\subset E$ be an open subset. 
The \textbf{support of some smooth functional} $F: U\rightarrow \mb R$, denoted by $\mathrm{supp}(F)$, is the set containing those points $x\in M$ such that for every neighbourhood $A_x\subset M$ of $x$, there exists $\phi,\psi\in U$ that verify $F(\phi)\neq F(\psi)$ and $\mathrm{supp}(\phi-\psi)\subset A_x$.
\end{deef} 

We next introduce the notion of reflection positive manifold  following the work of Jaffe \& Ritter~\cite{JR2,JR08,JR}:
\begin{deef}  \label{def-ref-positive-manifold}
We say that the $d$ dimensional manifold $M$ is \textbf{reflection positive} if there exists an \textbf{isometric involution} $\Theta\colon M\rightarrow M$ which admits as invariant subset a submanifold $\Sigma$ of dimension $d-1$ such that $M\setminus \Sigma$ is the disjoint union of two open manifolds $M_+$ and $M_-$ that both verify $\partial \overline{M_+}=\partial \overline{M_-}=\Sigma$. One can therefore think of $M$ as the disjoint union
\begin{align*}
    M=M_+\cup\Sigma\cup M_-\,,\,
    \Theta\upharpoonright\Sigma=\mathrm{Id}_\Sigma\,,\, \Theta (M_\pm) = M_\mp\,.
\end{align*}
\end{deef}

\begin{figure}[ht]
\centering
\begin{tikzpicture}[scale=1.5]
\draw[black] (-3.5,2.5)--(0.5,2.5);
\draw[teal] (0,0)..controls(.3,.3)and(.3,.7)..(0,1);
\draw[teal] (0,1)..controls(-.2,1.15)and(-.8,1.15)..(-1,1);

\draw[blue] (-2,1)..controls(-2.2,1.15)and(-2.8,1.15)..(-3,1);
\draw[blue] (-3,1)..controls(-3.3,.7)and(-3.3,.3)..(-3,0);
\draw[blue] (-2,0)..controls(-2.2,-.15)and(-2.8,-.15)..(-3,0);

\draw[teal] (-0,0)..controls(-.2,-.15)and(-.8,-.15)..(-1,0);
\draw[red] (-1.5,.22)..controls(-1.61,.4)and(-1.61,.6)..(-1.5,.78);
\draw[red,densely dotted] (-1.5,.22)..controls(-1.39,.4)and(-1.39,.6)..(-1.5,.78);
\node[red] at (-1.5,.0){$\Sigma$};
\draw[teal] (-.8,.5)..controls(-.6,.61)and(-.4,.61)..(-.2,.5);
\node[blue] at (-2.5,.8){$M_-$};
\node[teal] at (-.5,.8){$M_+$};
\draw[teal] (-.8,.5)..controls(-.6,.39)and(-.4,.39)..(-.2,.5);
\draw[blue] (-2.8,.5)..controls(-2.6,.61)and(-2.4,.61)..(-2.2,.5);
\draw[blue] (-2.8,.5)..controls(-2.6,.39)and(-2.4,.39)..(-2.2,.5);
\draw[red] (-1.5,2.22)..controls(-1.61,2.4)and(-1.61,2.6)..(-1.5,2.78);
\draw[red,densely dotted] (-1.5,2.22)..controls(-1.39,2.4)and(-1.39,2.6)..(-1.5,2.78);

\draw[teal] (-1,0)..controls(-1.2,.15)and(-1.35,.25)..(-1.5,0.22);
\draw[blue] (-2,0)..controls(-1.8,.15)and(-1.65,.25)..(-1.5,0.22);

\draw[teal] (-1,1)..controls(-1.2,.85)and(-1.35,.75)..(-1.5,0.78);
\draw[blue] (-2,1)..controls(-1.8,.85)and(-1.65,.75)..(-1.5,0.78);

\draw[gray] (0,2.22)..controls(-.11,2.4)and(-.11,2.6)..(-0,2.78);
\draw[gray] (-0,2.22)..controls(.11,2.4)and(.11,2.6)..(0,2.78);
\draw[gray] (-3,2.22)..controls(-3.11,2.4)and(-3.11,2.6)..(-3,2.78);
\draw[gray,densely dotted] (-3,2.22)..controls(-2.89,2.4)and(-2.89,2.6)..(-3,2.78);
\draw[gray] (-3,2.22)--(0,2.22);
\draw[gray] (-3,2.78)--(0,2.78);
\node[red] at (-1.5,2.0){$\Sigma$};
\node[black] at (.7,2.5){$\mb R$};
\draw[black] (.5,2.5)--(.45,2.45);
\draw[black] (.5,2.5)--(.45,2.55);

\draw[black] (-2,2)..controls(-1.8,1.7)and(-1.2,1.7)..(-1,2);
\draw[black] (-2,1.2)..controls(-1.8,1.5)and(-1.2,1.5)..(-1,1.2);
\draw[black] (-2,1.2)--(-1.93,1.2);
\draw[black] (-2,1.2)--(-2,1.27);
\draw[black] (-1,1.2)--(-1.07,1.2);
\draw[black] (-1,1.2)--(-1,1.27);
\draw[black] (-2,2)--(-1.93,2);
\draw[black] (-2,2)--(-2,1.93);
\draw[black] (-1,2)--(-1.07,2);
\draw[black] (-1,2)--(-1,1.93);

\node[black] at (-1.5,1.6){$\Theta$};
\end{tikzpicture}
\caption{A cylinder $\mb R\times\color{red}{\Sigma}$ and a reflection positive manifold $M=\color{teal}{M_+}\color{black}{\;\cup\;}\color{red}{\Sigma} \color{black}{\;\cup\;} \color{blue}{M_-}$.}
\label{f:RP_space}
\end{figure} 
\noindent For clarity we write below $\overline{M}_\pm$ for $\overline{M_\pm}$ and refer to Figure~\ref{f:RP_space}. The spaces of constant curvature $\mb R^d$, $\mb S^d$, and $\mb H^d$ are reflection positive. In the sequel, we will mostly be interested in the cases $M=\mb R\times\Sigma$ and $M=\mb S^d$, for which we precise the setting. The cylinder $\mb R\times\Sigma$ is reflection positive with respect to $\{0\}\times\Sigma$. The sphere $\mb S^d$ is reflection positive with respect to any of its equators.

Whenever a manifold $M$ is reflection positive with isometric involution $\Theta$, the map $\Theta$ acts via pull-back on $C^\infty(M)$ and $\CB{D}'(M)$ and hence on functionals on $M$ and $\CB{D}'(M)$-valued random variables. Note that for any distribution $\omega\in\CB D'(M)$,
\begin{equation} \label{EqSymmetrySpectralProjectors}
\Pi_\Lambda(\omega)\upharpoonright
{\overline{M}_-} = \big(\Pi_\Lambda(\omega) \circ\Theta\big)\upharpoonright{\overline{M}_+} = \Pi_\Lambda(\omega\circ\Theta)\upharpoonright{\overline{M}_+}\,.
\end{equation}


Let $A$ be a closed subset of $M$. We denote by $W^{-1}_A(M) $ the subset of elements in the Sobolev space $W^{-1}(M)$ which are supported in $A$. Let $\Psi$ be a random distribution defined on our probability space $\Omega$, with law $\mu_\Psi$ and associated expectation operator $\mb{E}_\Psi[\bigcdot]$. We introduce a sub $\sigma$-algebra of $\CB{O}$
\begin{align*}    
\sigma(\Psi;A) \defeq \sigma\big(  \Psi(f) ; f\in W^{-1}_A(M)\big)\,.
\end{align*}
We implicitly assume here that we only consider those functions $f\in  W^{-1}_A(M)$ for which $\Psi(f)$ is well-defined. One can think of $\sigma(\Psi;A)$ as the $\sigma$-algebra generated by the observables that are only sensitive to the fluctuation of the field $\Psi$ on the set $A$.

\begin{deef}
A random process $\Psi$ as above has the \textbf{Markov property} if for \textit{any} closed subsets $A, B$ of $M$ such that $A^{\circ}\cap B=\varnothing $, namely $B$ does not intersect the interior of $A$ (Figure~\ref{f:Markov} below),  and any random variable $F\in L^2$ that is $\sigma(\Psi;B)$-measurable, one has
\begin{equation*}
  \mb{E}\big[F | \sigma(\Psi;A)\big] = \mb{E}\big[F | \sigma(\Psi;\partial A)\big].
  \label{}
\end{equation*}
\end{deef}

\begin{figure}[ht]
\centering
\begin{tikzpicture}[scale=1.5]
\draw[gray] (0,0)..controls(.3,.3)and(.3,.7)..(0,1);
\draw[gray] (0,1)..controls(-.2,1.15)and(-.8,1.15)..(-1,1);

\draw[blue] (-2,1)..controls(-2.2,1.15)and(-2.8,1.15)..(-3,1);
\draw[blue] (-3,1)..controls(-3.3,.7)and(-3.3,.3)..(-3,0);
\draw[blue] (-2,0)..controls(-2.2,-.15)and(-2.8,-.15)..(-3,0);

\draw[gray] (-1,0)..controls(-1.2,.15)and(-1.35,.25)..(-1.5,0.22);
\draw[blue] (-2,0)..controls(-1.8,.15)and(-1.65,.25)..(-1.5,0.22);

\draw[gray] (-1,1)..controls(-1.2,.85)and(-1.35,.75)..(-1.5,0.78);
\draw[blue] (-2,1)..controls(-1.8,.85)and(-1.65,.75)..(-1.5,0.78);

\draw[gray] (-0,0)..controls(-.2,-.15)and(-.8,-.15)..(-1,0);
\draw[blue] (-1.5,.22)..controls(-1.61,.4)and(-1.61,.6)..(-1.5,.78);
\draw[blue,densely dotted] (-1.5,.22)..controls(-1.39,.4)and(-1.39,.6)..(-1.5,.78);
\draw[gray] (-.8,.5)..controls(-.6,.61)and(-.4,.61)..(-.2,.5);
\node[black] at (-2.5,.8){$A$};
\node[black] at (-.5,.8){$B$};
\node[blue] at (-3.5,.5){$\partial A$};
\draw[gray] (-.8,.5)..controls(-.6,.39)and(-.4,.39)..(-.2,.5);
\draw[blue] (-2.8,.5)..controls(-2.6,.61)and(-2.4,.61)..(-2.2,.5);
\draw[blue] (-2.8,.5)..controls(-2.6,.39)and(-2.4,.39)..(-2.2,.5);
\end{tikzpicture}
\caption{Two closed subsets $A$ and $B$ of $M$ such that $A^{\circ}\cap B=\varnothing$.}
\label{f:Markov}
\end{figure}

The Markov property was introduced by Nelson in \cite{Nel73a}, \cite{Nel2} (pp.\ 224-225). He also proved that $\mu_{\text{GFF}}$ verifies this property ~\cite{Nel2} (pp.\ 225).

\begin{deef}
A random process $\Psi$ on a reflection positive manifold $M$ is \textbf{reflection positive} if for any $\mb{C}$-valued function $F$ in $ L^2\big(\CB{D}^\prime(M),\sigma(\Psi;\overline{M}_+),\mu_\Psi\big)$ it holds
\begin{align*}
    \mb E_\Psi\big[(\overline{\Theta F}) F\big] \geqslant 0\,.
\end{align*}
\end{deef}

\noindent Reflection positivity is one of the axioms introduced by Osterwalder and Schrader \cite{OS}; it is a crucial condition to recover a Lorentzian Quantum Field Theory from a Euclidean Quantum Field Theory. 

Lastly we indicate the relation between the Markov property and reflection positivity. Let $M$ be a reflection positive manifold with reflection $\Theta$ and decomposition $M=M_+\cup\Sigma\cup M_-$ as in Definition \ref{def-ref-positive-manifold}. The random process $\Psi$ as above with law $\mu_{\Psi}$ is called \textbf{reflection invariant} if $\mb E_\Psi [\ol{F}G]=\mb E_\Psi[\ol{\Theta F}\Theta G]$ for any $F$, $G\in L^2(\CB{D}^\prime(M),\mu_\Psi)$.
\begin{lemm} [\cite{Dimock}, Theorem 2] \label{LemMarkovImpliesRP}
On a reflection positive Riemannian manifold $(M,g)$, a random field $\Psi$ indexed by $W^{-1}(M)$ that is reflection invariant and has the Markov property is reflection positive.
\end{lemm}

\begin{proof}
     First observe that if $G\in L^2\big(\CB{D}^\prime(M),\sigma(\Psi;\Sigma),\mu_\Psi\big)$ then $\Theta G=G$, since functionals of the form $\varphi(\Psi(f_1),\cdots,\Psi(f_n))$, for $f_j\in W^{-1}_\Sigma(M)$, $\varphi\in \CB{S}(\mb{R}^n)$, $n\in\mb{N}$, are dense in $L^2\big(\CB{D}^\prime(M),\sigma(\Psi;\Sigma),\mu_\Psi\big)$ (\cite{Sim2} Lemma I.5), and,
    \begin{equation}
        \Theta\cdot\varphi(\Psi(f_1),\cdots)=\varphi(\Theta\Psi(f_1),\cdots)=\varphi(\Psi(\Theta f_1),\cdots)=\varphi(\Psi(f_1),\cdots)\,,
    \end{equation}
    as each $f_j\in W^{-1}_\Sigma(M)$ is invariant by $\Theta$. This with reflection invariance tells that $\mb E_\Psi[H|\sigma(\Psi;\Sigma)]=\mb E_\Psi[\Theta H|\sigma(\Psi;\Sigma)]$ for any $H\in L^2$. Indeed, for any $G\in L^2\big(\CB{D}^\prime(M),\sigma(\Psi;\Sigma),\mu_\Psi\big)$,
    \begin{align*}
       \mb E_\Psi[\mb E_\Psi[\Theta H|\sigma(\Psi;\Sigma)]\cdot G]&=
       \mb E_\Psi[\mb E_\Psi[\Theta H\cdot \Theta G|\sigma(\Psi;\Sigma)]]=
       \mb E_\Psi[\Theta H\cdot \Theta G]\\
       &\overset{\mathrm{reflection}}{=} \mb E_\Psi[ H G]
       =\mb E_\Psi[\mb E_\Psi[H|\sigma(\Psi;\Sigma)]\cdot G].
    \end{align*}
    Finally we proceed with the main argument: pick $F\in L^2\big(\CB{D}^\prime(M),\sigma(\Psi;\overline{M}_+),\mu_\Psi\big)$. It is clear that the functional $\Theta F$ is then $\sigma(\Psi;\overline{M}_-)$-measurable. Hence we have
    \begin{equation*}
        \mb E_\Psi [\ol{\Theta F}\cdot F]=
        \mb E_\Psi [\mb E_\Psi[\ol{\Theta F}|\sigma(\Psi;\ol{M}_+)]\cdot F]
        \overset{\mathrm{Markov}}{=}
        \mb E_\Psi [\mb E_\Psi[\ol{\Theta F}|\sigma(\Psi;\Sigma)]\cdot F]=
        \mb E_\Psi \big|\mb E_\Psi[F|\sigma(\Psi;\Sigma)]\big|^2\geqslant 0,
    \end{equation*}
    proving the result.
\end{proof}

\subsection{The regularized $\Phi^4_3$ measure built from the cut-off interaction}
 The $\Phi^4_3$ measure $\nu$ is a probability measure on $\CB D'(M)$ with $(M,g,\Theta)$ a closed three dimensional reflection positive manifold that heuristically reads as
\begin{align} \label{EqPhi43Density}
\nu(\dd \Phi)\propto e^{- c\Vert\Phi\Vert^4_{L^4(M)}} \, \mu_{\text{GFF}}(\dd \Phi)\,.
\end{align}
 We use the non-conventional notation $c$ for the coupling constant. Since the dimension of $M$ is greater then or equal to 2 the GFF is not supported on $L^p$ functions and the $L^4$ norm of $\Phi$ is almost surely infinite. The formal expression \eqref{EqPhi43Density} is thus meaningless. The probability measure $\nu$ can therefore only be defined as the weak limit of a sequence of approximations $( \nu_{\rho,\Lambda})_{\Lambda\geqslant0}$. A choice which is often made is to consider
\begin{align}\label{eqnu}     
 \nu_{\rho,\Lambda}(\dd \Phi)\propto e^{- c\Vert \rho\Phi_\Lambda\Vert^4_{L^4(M)} - c a_\Lambda\Vert \rho\Phi_\Lambda\Vert^2_{L^2(M)}} \, \mu_{\text{GFF}}(\dd \Phi)\,,
\end{align}
where $\rho\in C^\infty_c(M)$ is identically equal to $1$ on a large compact set of $M$ and $(a_\Lambda)_{\Lambda\geqslant0}$ is a suitably chosen sequence of real number that is divergent at large $\Lambda$. 

On $M$, the $\Phi^4_3$ measure $\nu$ is a good starting point to construct a EQFT, provided $M$ is reflection positive and $\nu$ is proved to be reflection positive too. When $\nu$ is constructed as a limit of a sequence of approximations, reflection positivity of $\nu$ is in general obtained by showing that the measures $\nu_{\rho,\Lambda}$ are reflection positive for every $\Lambda$, in which case $\nu$ inherits the reflection positivity property of its approximations.

 However, in any attempt to prove that property for $\nu_{\rho,\Lambda}$, we quickly face a problem that comes from the following fact, which is proved in Section \ref{sec1}.

\begin{thrm} \label{thm1}
Let $M$ be a smooth complete Riemannian manifold of dimension $d$. Let $A\subsetneq M$ be a closed subset with non-empty interior, and $G$ a smooth functional on $C^\infty(M)$ such that the random variable $G\circ \Pi_\Lambda(\Phi)$ is $\sigma(\Phi;A)$-measurable. Then the function $G\circ \Pi_\Lambda$ on $H^{{(2-d)/}{2}-\varepsilon}(M)$ is constant.
\end{thrm}

 This fact prevents us from transfering the reflection positivity property of the GFF to its spectrally cutoff version $\Phi_\Lambda$. As a matter of fact, we prove in Theorem \ref{ThmNotRFRegularizedPhi43} in Section \ref{SectionCounterexample} that the regularized $\Phi^4_3$ measures $\nu_{\rho,\Lambda}$ are not reflection positive for a small enough coupling constant $c>0$ in \eqref{eqnu}. Moreover we give hereafter three arguments that show that the cut-off free field $\Phi_\Lambda$ can by no means be neither Markov nor reflection positive, neither on cylinders like $\mb R^d$ nor on compact manifolds. This also confirms that any approach involving spectral cut-offs is unlikely to be efficient in order to prove the reflection positivity of the $\Phi^4_3$ measure.

To conclude this introduction, we would like to mention several works 
where the authors are able to construct reflection positive interacting EQFT measures.
In the seminal work of Glimm-Jaffe, RP of the regularized theory comes from the RP property of the lattice approximations~\cite{GJ}.  The Wick rotation of the $P(\Phi)_2$ model on the cylinder and the sphere was first treated in the seminal articles \cite{HoeghKrohn-1974,Figari-HoeghKrohn-Nappi-1975}. We also refer to \cite{Gerard-Jakel-2005a,Gerard-Jakel-2005b,Barata-Jakel-Mund-2023} for detailed constructions of some quantum fields in 2 space-time dimensions.
 
In~\cite{Abdesselam} (pp.\ 237), the author suggests to mollify the field by convolution with some compactly supported function and put a position dependent coupling in front of the interaction $\tilde{g}_r\phi^4$ where the coupling function $\tilde{g}_r$, which depends on the cut--off $r>0$, vanishes in some neighborhood of the reflection hyperplane. This ensures that the regularized measures are RP. To implement this rigorously on say the $\Phi^4_3$ model, the difficulty would be to show that one recovers the same EQFT measure as limit EQFT measure where the coupling function is constant.
  In the very interesting
recent preprint~\cite{Duch}, the authors use a very similar strategy to construct a reflection positive $P(\Phi)_2$ measure on $\mb{R}^2$. Finally, a similar strategy is used by the last author~\cite{Lin} to prove that the $P(\Phi)_2$ theory on every Riemannian surface satisfies both the spatial Markov property and the Segal gluing axioms which is stronger than just RP.

\medskip

 \subsection{Organization and Comments}
 We prove Theorem~\ref{thm1} in Section~\ref{sec1}. This somehow shows that the spatial Markov property for the cut-off GFF is not well-posed since we prove that the only $\sigma(\Phi;\overline{M}_+)$-measurable random variables of the form $G(\Phi_\Lambda)$ where $G$ is smooth are constants. In Section \ref{sec2} we give a counter-example showing that the cut-off free field on the cylinder $\mb R\times \Sigma$ is not reflection positive, by constructing a function on which the bilinear form associated to the cut-off covariance acts negatively, see Theorem~\ref{thm2} below. In Section \ref{sec3} we show that the cut-off free field on compact manifolds also has a covariance which is not reflection positive, by constructing a counter-example, see Theorem~\ref{thm41} below. Given Lemma \ref{LemMarkovImpliesRP} the results of these sections show that the cut-off regularized Gaussian free field does not have the Markov property in these settings.

\begin{def7}
 We point out that we may also see that the cut-off GFF cannot be reflection positive non-constructi- vely, using the Källén–Lehmann representation (\cite{GJ} Theorem 6.2.4), which states that the covariance must be a weighted sum of inverse massive laplacians of different masses. The cut-off GFF cannot have a Källén–Lehmann representation since, if so, it would have non-zero covariance at arbitrarily high frequency, i.e. one would have $\mb E|\Phi_{\Lambda}(e_\lambda)|^2>0$ for $\lambda$ arbitrarily large. This comment was suggested to us by an anonymous referee.
\end{def7}

 Note that in the flat case $M=\mb R^3$, it is a well-established fact that the $\Phi^4_3$ measure is reflection positive, since it can be constructed as the limit of a sequence of Gibbs measure on the lattice~\cite{GH}, where regularized measures are reflection positive. However this construction does not generalize easily to the case of a compact Riemannian manifold $M$ where lattice regularization is a more difficult option. On the other hand, we prove in Section~\ref{SectionCounterexample} that the regularized $\Phi^4_3$ measures $\nu_{\rho,\Lambda}$ are not reflection positive for a small enough coupling constant $c>0$. Hence the proof of the reflection positivity of the $\Phi^4_3$ measure on the sphere $\mb S^3$ is still an open problem.

 Regarding the Markov property, contrary to the two dimensional case, not much is known about the $\Phi^4_3$ measure. Note that the original construction of the $\Phi^4_3$ QFT by Glimm and Jaffe relies on the study of the Hamiltonian \cite{Glimm-1968,Glimm-Jaffe-1973}, which indicates that the $\Phi^4_3$ measure should have a semi-group property on the plane.

\medskip

\subsection{Acknowledgements}

We would like to thank A. Abdesselam, D. Benedetti, M. Gubinelli, A. Mouzard, T.D.T\^o, M. Wrochna for their interest and comments on the present note. N.V.D acknowledges the support of the Institut Universitaire de France.
The authors would like to
thank the ANR grant SMOOTH ”ANR-22-CE40-0017” and QFG ”ANR-20-CE40-0018” for support.

\section[Local Observables under Cut-off]{Triviality of Smooth Localized Observables under Cut-off}\label{sec1}

We prove Theorem \ref{thm1} in this section. Our proof is based on the following result due to Lebeau-Robbiano~\cite{LR95}, Jerison-Lebeau~\cite{JL99}, Lebeau-Zuazua~\cite{LZ98} which asserts that the zero set (also called \textbf{nodal set}) of linear combinations of Laplace eigenfunctions always has empty interior:
\begin{lemm}[Nodal sets of linear combinations of eigenfunctions]\label{thmnodal}
Let $(M,g)$ be a smooth closed compact Riemannian manifold.
Then for any $\Lambda\in(0,\infty)$ and any non-trivial finite 
linear combination of the eigenfunctions of $\Delta_g$ that we denote by $f\in\EL$, the zero set $Z_f=\{x\in M|f(x)=0\}$ of the function $f$ has \textbf{empty interior}. 
\end{lemm}

\begin{proof}
Assume there exists $f\in \EL$ such that $\Vert f\Vert_{L^2(M)}=1$ and $Z_f$ contains an open subset $U$.
Then the Lebeau-Robbiano spectral inequality~\cite{JL99} Theorem 14.6,
states that given $U\subset M$, there exists constants $C,K>0$ such that for all $\varphi\in \EL$, we have an inequality of the form:
$$ \Vert \varphi\Vert_{L^2(M)} \leqslant Ce^{K\Lambda}\Vert \varphi\Vert_{L^2(U)}\,, $$
where $C,K$ do not depend on $\varphi$.
Setting $\varphi=f$ yields
$$\Vert f\Vert_{L^2(M)} \leqslant Ce^{K\Lambda}\Vert f\Vert_{L^2(U)}=0\,, $$
since we assumed $f\upharpoonright U=0$. This yields a contradiction with the assumption $\Vert f\Vert_{L^2(M)}=1$.
\end{proof}

\begin{corr} \label{Lemm21}
We assume $(M,g)$ is a smooth closed compact Riemannian manifold or $M=\mb{R}^d$ with the flat metric and $B$ is some open subset $B\subsetneq M $.
Let $T\in \CB{D}^\prime(M)$ be a distribution such that for every $f\in C_c^\infty(B)$, $T(\Pi_\Lambda f)=0$. Then $\Pi_\Lambda T=0$.
\end{corr}

\begin{proof}
    Since $\Pi_\Lambda$ is self-adjoint, we have
\begin{align*}
    T(\Pi_\Lambda f)=\langle \Pi_\Lambda T,f\rangle_{L^2(M)}.
\end{align*}
Therefore, $\Pi_\Lambda T$ is a smooth function in  $ \mathrm E_{\leqslant \Lambda}$ which vanishes on the interior of $B$. Hence, in view of Lemma~\ref{thmnodal}, it is identically equal to zero. 

In the case $M=\mb{R}^d$, we repeat the exact same argument and conclude using the fact that $\Pi_\Lambda T$ is analytic by the Paley-Wiener Theorem hence null if it vanishes on the open set $B$.
\end{proof}

We are now ready to give the proof of Theorem~\ref{thm1}. 

\begin{proof}[Proof of Theorem~\ref{thm1}]
We work by contradiction. Set $F=G\circ \Pi_\Lambda$ which is a smooth functional on $H^{{(2-d)/}{2}-\varepsilon}(M)$. Assuming that $F$ is $\sigma(\Phi;A)$-measurable, for every direction $h\in C^\infty_c(A^c)$, we have the identity
$$ 
F(\Phi+h)=F(\Phi) 
$$
for $\Phi$ in a set $\Omega_h\subset\Omega$ of probability $1$ that depends on $h$. Now we would like to go from an $h$-dependent almost sure statement to a deterministic statement, namely we aim to prove that $$\tau_hF(\Phi)\defeq F(\Phi+h)-F(\Phi)=0$$ for every distribution $\Phi \in H^{{(2-d)/}{2}-\varepsilon}(M)$. 

Assume by contradiction that there exists some $\varphi\in H^{{(2-d)/}{2}-\varepsilon}(M)$ such that $\tau_h F(\varphi)\neq 0$, and assume without loss of generality that $\tau_h F(\varphi)=L >0$. The functional $$ H^{{(2-d)/}{2}-\varepsilon}(M)\ni\Phi\mapsto\tau_h F(\Phi)$$ is smooth. Hence, by continuity, there is some open subset $U_\varphi\ni\varphi $ of $H^{{(2-d)/}{2}-\varepsilon}(M)$ such that $0<\frac{L}{2}\leqslant  F \upharpoonright{U_\varphi}\leqslant \frac{3L}{2}$. There, there exists $R>0$ such that $B_\varphi(R)\subset U_\varphi$ for the $H^{{(2-d)/}{2}-\varepsilon}(M)$ topology. One then has
$$
0=\mb{E}_{\mathrm{GFF}}\left[ \tau_hF(\bigcdot)1_{B_\varphi(R)} \right] \geqslant \frac{L}{2}\mu_{\mathrm{GFF}}\left( B_\varphi(R)\right) \,.
$$
The first equality follows from the almost sure vanishing of $\tau_hF(\bigcdot)$, which implies that 
$$
\mu\big( B_\varphi(R)\big) = 0\,.
$$ 
This is a contradiction, because $\mu_{\text{GFF}}$ has \textbf{full support} in $H^{{(2-d)}{/2}-\varepsilon}(M)$ since the Cameron-Martin space $H^1(M)$ is everywhere dense in $H^{{(2-d)/}{2}-\varepsilon}(M)$ for the $H^{{(2-d)/}{2}-\varepsilon}(M)$ topology (\cite{Bogachev} Theorem 3.6.1, \cite{HairerSPDE} Proposition 3.68). Therefore, $F(\Phi+h)=F(\Phi) $ for every distributions $\Phi\in H^{{(2-d)/}{2}-\varepsilon}(M)$ and every $h\in C^\infty_c(A^c)$ which is much stronger than the almost sure statement. It means that $G\circ\Pi_\Lambda$ should not depend on $h\in C_c^\infty(A^c)$. More precisely, by definition of the support of a functional (see Definition~\ref{def:suppfonc}), for any such $h$ we have $G(\Phi_\Lambda+\Pi_\Lambda h)=G(\Phi_\Lambda)$, which implies that $\D_\Phi(G\circ\Pi_\Lambda)(h)=0$. Observing that 
\begin{align*}
    \D_\Phi(G\circ\Pi_\Lambda)(h)=\lim_{t\downarrow0}\frac{G(\Phi_\Lambda+t\Pi_\Lambda h)-G(\Phi_\Lambda)}{t}=\D_{\Phi_\Lambda}(G)(\Pi_\Lambda h)\,,
\end{align*}
we conclude that $\D_{\Phi_\Lambda}(G)(\Pi_\Lambda h)=0$. In view of Corollary~\ref{Lemm21} taking $B=A^c$ and $T=\D_{\Phi_\Lambda}(G)$, we thus have that $\Pi_\Lambda\D_{\Phi_\Lambda}(G)=0$. In particular, for all $f\in C^\infty(M)$, 
\begin{align*}
   \Pi_\Lambda\D_{\Phi_\Lambda}(G)(f)=\D_{\Phi_\Lambda}(G)(\Pi_\Lambda f)=\D_\Phi(G\circ\Pi_\Lambda)(f)=0\,,
\end{align*}
so $G\circ\Pi_\Lambda $ is indeed constant. 
\end{proof}

\section{Non-Reflection Positivity}

\subsection{The cut-off GFF on Cylinders}\label{sec2}

In this section, we give an explicit counter-example contradicting the assertion that $\Phi_\Lambda$ could be reflection positive on $\mb R^d$ and more generally on any Riemannian cylinder $M$ of the form $M=\mb{R}\times \Sigma$ where $\Sigma$ is complete Riemannian and the cylinder is endowed with the split metric $\dd  t^2+g_\Sigma$. This is based on the following lemma.
\begin{lemm}[\cite{GJ}, Theorem 6.2.2]\label{lemm31}
A Gaussian random field $\Psi$ with covariance $C$ is reflection positive if and only if its covariance is reflection positive, in the sense that for any $f\in C_c^\infty(\mb{R}_{\geqslant 0}\times\Sigma)$,
\begin{align*}
    \langle \Theta f,Cf\rangle_{L^2(M)}\geqslant0\,.
\end{align*}
\end{lemm}
We construct hereafter such a function $f$ such that $\sank{ \Theta f,\frac{\Pi_\Lambda}{\Delta_g+1}f}_{L^2(M)}<0$, which implies the following result.
\begin{thrm}\label{thm2}
    The spectrally cut-off massive Gaussian free field on $\mb R^d$ or on any Riemannian cylinder $\Phi_\Lambda$ is not reflection positive. In particular, it is not Markov either. 
\end{thrm}
To construct our counter-example, we need the following intermediate lemma.
\begin{lemm}\label{lemm2}
Fix $\kappa>0$. There exists a function~$h\in C_c^{\infty}(\mb{R}_{\geqslant 0})$ such that, denoting by $\CB F h(\xi)\defeq\int_{\mb R} h(x) e^{\ii\xi x}\dd  x$ the Fourier transform of $h$, it holds
  \begin{equation*}
    \int_{-1}^{1}\ol{\CB F{h}(-\xi)} \,\CB F{h}(\xi) \, \frac{\dd \xi}{\xi^2+\kappa} < 0\,.
    \label{}
  \end{equation*}
\end{lemm}
\begin{proof}
  Note that because $h$ is real, one has~$\ol{\CB F{h}(-\xi)}=\CB F{h}(\xi)$. Thus, $h$ needs to satisfy
  \begin{equation*}
    \int_{-1}^{1}\CB F{h}^2(\xi) \frac{\dd \xi}{\xi^2+\kappa}<0\,.
    \label{}
  \end{equation*}
  Now write
  \begin{equation*}
    \CB F{h}(\xi)=A_h(\xi)+\ii B_h(\xi)\,,
    \label{}
  \end{equation*}
  where
  \begin{equation*}
    A_h(\xi)\defeq\int_{\mb{R}} h(x)\cos(\xi x)\dd  x\,,\quad B_h(\xi)\defeq\int_{\mb{R}}h(x)\sin(\xi x)\dd  x\,,
    \label{}
  \end{equation*}
  so that~$A_h$ is even and~$B_h$ is odd. This gives
  \begin{equation*}
   \CB F{h}^2(\xi)=A^2_h(\xi)-B^2_h(\xi)+C_h(\xi)\,,\quad\text{where}\quad C_h(\xi)\defeq2\ii {A_h(\xi)B_h(\xi)}\,.
    \label{}
  \end{equation*}
Here, observe that $C_h(\xi)$ is an odd function, which does not contribute to the integral. 
  
Moreover, the important idea is that since~$\CB F{h}^2(0)=\big(\int_{\mb{R}} h(x)\dd  x\big)^2\geqslant0$, one needs to shift the weight of integration~$(\xi^2+\kappa)^{-1}$ ``away from zero''. This can be done by taking derivatives: setting $h=\varphi^{(2n)}$, we have
  \begin{align*}
    \int_{-1}^{1}\big(\CB F{\varphi^{(2n)}}(\xi)\big)^2\frac{\dd \xi}{\xi^2+\kappa}&=2\int_{0}^{1}\frac{(\ii \xi)^{4n}}{\xi^2+\kappa} \big(A^2_{\varphi}(\xi)-B^2_{\varphi}(\xi)\big) \dd \xi  \,.
  \end{align*}
Note that we have taken care of choosing a number of derivatives that enforces $\ii ^{4n}=1$. Now, observe that
  \begin{equation*} \frac{ \int_{0}^{1}\frac{\xi^{4n}}{\xi^2+\kappa} \big(A^2_{\varphi}(\xi)-B^2_{\varphi}(\xi)\big) \dd \xi}{\int_{0}^{1}\frac{\xi^{4n}}{\xi^2+\kappa} \dd \xi} \lto A^2_{\varphi}(1)-B^2_{\varphi}(1)\,\,\, \textrm{as }n\uparrow \infty\,.
  \end{equation*}
Choosing $\varphi$ supported near $\frac{\pi}{2}$ ensures that $\cos(x)\approx 0$ and $\sin(x)\approx 1$, so that $A^2_{\varphi}(1)-B^2_{\varphi}(1)<0$. Taking $n$ large enough but finite and rescaling $\varphi$, one finally obtains the desired result. Moreover, since $\varphi$ is compactly supported on $\mb{R}_{\geqslant 0}$, so is $h=\varphi^{(2n)}$. 
\end{proof}
\begin{proof}[Proof of Theorem \ref{thm2}]
As alluded to, the proof follows from the fact that we provide a counter-example of the positivity of the covariance, that is to say, we construct $f\in  C_c^\infty(\mb{R}_{\geqslant 0}\times\Sigma)$ such that  $ \langle \Theta f,\frac{\Pi_\Lambda}{\Delta_g+1}f\rangle_{L^2(M)}<0$. 

We denote by $\CB E_\lambda\dd \lambda$ the projection valued measure of the slice Laplacian $\Delta_\Sigma$ which is well--known to be self-adjoint by the completeness of $\Sigma$. For any function $f\in L^2(M)$ 
on the cylinder $M=\mb{R}_t\times \Sigma_x$, we denote by 
$$
\CB T_{(\tau,\lambda)}f(\bigcdot)\defeq \int_{\mb{R}} e^{-\ii t\tau} \CB E_\lambda f(t,\bigcdot) \dd  t  
$$ 
its Fourier transform w.r.t.\ the time variable $t$ and its spectral transform with respect to the spectral measure of $\Delta_\Sigma$. It is a convenient expression because we know that the action of the Laplacian on $\CB T_{(\tau,\lambda)}f$ boils down to the multiplication by $\tau^2+\lambda$.

Then, denoting by $\mathrm{v}_\Sigma$ the volume form on $\Sigma$, we rewrite the pairing as
\begin{align}\nonumber\label{eq:step}
\big\langle \Theta f,\frac{\Pi_\Lambda}{\Delta_g+1}f\big\rangle_{L^2(M)}
&=\int_{\{\tau^2\vee\lambda\leqslant\Lambda\}\subset\mb{R}\times \mb{R}_{\geqslant 0}}\Big(\int_\Sigma \overline{\CB T_{(-\tau,\lambda)}{f}}\CB T_{(\tau,\lambda)}{f}\dd  \mathrm{v}_\Sigma \Big) \frac{\dd \tau \dd \lambda}{\tau^2+\lambda+1}   \\
&=\int_{0}^{\Lambda} \int_{-\sqrt\Lambda}^{\sqrt\Lambda}\Big(\int_\Sigma \overline{\CB T_{(-\tau,\lambda)}{f}}  \CB T_{(\tau,\lambda)}{f}  \dd \mathrm v_\Sigma \Big)
\frac{ \dd \tau\dd \lambda }{\tau^2+\lambda+1}\,.
\end{align}
We will reduce to the one variable case by a scaling argument. First, fix $\varphi\in C^\infty_c(\mb{R}_{\geqslant 0})$ such that
$ \int_{-1}^1\CB F{\varphi}^2(\tau) (\tau^2+\lambda + \Lambda^{-1})^{-1}\dd \tau<0 $ for every $\lambda\in [0,1]$. The existence
of such $\varphi$ comes from Lemma~\ref{lemm2}. 
Then, choose $\chi\in L^2(\Sigma)$ such that $\Vert\chi\Vert_{L^2(\Sigma)}^2=1$ and $\chi$ has non trivial spectral measure in the interval $[0,1]$, that is to say
$$
\int_0^1 \Vert \CB  E_\lambda\chi\Vert^2_{L^2(\Sigma)} \dd \lambda > 0\,.
$$ 
Finally, we set
$$
f_{\varphi,\chi}(t,x)\defeq\int_0^\infty \sqrt\Lambda \varphi(\sqrt\Lambda t)\CB E_{{\lambda}/{\Lambda}}\chi(x)\dd \lambda\,,
$$
which verifies 
\begin{align*}
    \CB T_{(\tau,\lambda)}f_{\varphi,\chi}(x)=\CB F \varphi(\tau/\sqrt{\Lambda})\CB E_{\lambda/\Lambda}\chi(x)\,.
\end{align*}
Plugging this expression inside \eqref{eq:step}, we obtain
\begin{align*}
\big\langle \Theta f_{\varphi,\chi},\frac{\Pi_\Lambda}{\Delta_g+1}f_{\varphi,\chi}\big\rangle_{L^2(M)} &= \int_{0}^{\Lambda}  \int_{-\sqrt\Lambda}^{\sqrt\Lambda}\Vert \CB E_{\lambda/\Lambda} \chi\Vert_{L^2(\Sigma)}^2    \CB F{\varphi}^2({\tau}/{\sqrt\Lambda})\frac{\dd \tau   \dd \lambda}{\tau^2+\lambda+1}    \\
&=\Lambda^{3/2} \int_{0}^{1} \int_{-1}^1\Vert\CB E_\lambda \chi\Vert_{L^2(\Sigma)}^2  \CB F{\varphi}^2(\tau)\frac{\dd \tau \dd \lambda}{\Lambda(\tau^2+\lambda)+1}   \\
&=\sqrt{\Lambda}\int_{0}^{1}  \int_{-1}^1 \Vert \CB E_\lambda \chi\Vert_{L^2(\Sigma)}^2  \CB F {\varphi}^2(\tau) \frac{ \dd \tau   \dd \lambda}{\tau^2+\lambda+\frac{1}{\Lambda}}   \\
&<0.
\end{align*}
We have thus proven that the cut-off covariance is not reflection positive, and by Lemma~\ref{lemm31}, $\Phi_\Lambda$ is therefore not reflection positive.
\end{proof}

\subsection{The cut-off GFF on Compact Manifolds}\label{sec3}

We now turn to the proof that the cut-off GFF on compact manifolds is not reflection positive. This will also imply that it is not Markov.
\begin{thrm}\label{thm41}
Assume that $(M,g)$ is a smooth compact Riemannian manifold which is reflection positive. 
Then there exists $L>0$ such that for any $\Lambda\geqslant L$, there exists a function $f\in C^\infty_c(M_+)$
verifying 
\begin{eqnarray*}
\big\langle \Theta f, \frac{\Pi_\Lambda}{\Delta_g+1}f \big\rangle_{L^2(M)}<0\,.
\end{eqnarray*}
\end{thrm} 
\begin{proof}
First, because $\Delta_g$ and $\Theta$ commute, we can choose the $L^2$ basis of $\Delta_g$ in such a way that any eigenfunction $e_\lambda$ of the Laplacian $\Delta_g$ satisfies either 
$\Theta e_\lambda=e_\lambda$ (we say that $e_\lambda$ is even) or $\Theta e_\lambda=-e_\lambda$ (we say that it is odd). Observe that the set of odd eigenfunctions is not
empty, since the span of eigenfunctions is dense in $C^\infty(M)$, which contains
odd functions. We choose $L$ large enough so that there exists $\lambda\leqslant L$ such that $e_\lambda$ is odd. Given $\Lambda \geqslant L$, we denote by $\lambda_*$ the largest eigenvalue smaller than $ \Lambda$ with odd eigenfunction $e_{\lambda_*}$.

Next, we consider the linear map
\begin{eqnarray*}
\CB P_\Lambda\colon  C^\infty(M_+)\ni f\mapsto \left(\left\langle f,e_{\lambda} 
\right\rangle \right)_{\lambda\leqslant \Lambda}\in \mb{R}^{\dim(\EL)}\,.
\end{eqnarray*}
We have to show that $\CB P_\Lambda$ is onto. Assume by contradiction that it is not onto. Then $\CB P_\Lambda(C^\infty(M_+) )$ is a strict vector subspace of $\mb{R}^{\dim(\EL)}$, and we can choose some vector $(c_\lambda)_{\lambda\leqslant \Lambda}$ in $\CB P_\Lambda(C^\infty(M_+) )^\perp$. In other words there exists a non-trivial linear combination $\varphi=\sum_{\lambda\leqslant \Lambda} c_\lambda e_\lambda \in \EL$ such that for every $f\in C^\infty_c(M_+)$ it holds
$$
\sum_{\lambda\leqslant \Lambda} c_\lambda \left\langle e_\lambda,f\right\rangle = 0\,. 
$$
This entails that $\varphi\in\EL$ vanishes on $M_+$. However, since $M_+$ has non-empty interior, this would contradict Lemma~\ref{thmnodal}, which concludes the proof that the map $\CB P_\Lambda$ is surjective. The surjectivity of $\CB P_\Lambda$ allows us to find $f_*\in C^\infty_c(M_+)$ such that $\left\langle f_*,e_\lambda\right\rangle = 0$ if $\lambda\neq \lambda_*$ and $\left\langle f_*,e_{\lambda_*}\right\rangle=1$. For this $f_*$, we have
\begin{align*}
\big\langle \Theta f_*, \frac{\Pi_\Lambda}{\Delta_g+1}f_* \big\rangle_{L^2(M)}&=
\sum_{\lambda\leqslant \Lambda} \left\langle f_*,\Theta e_\lambda\right\rangle \frac{1}{\lambda+1} \left\langle f_*,e_\lambda\right\rangle\\
&=
-\left\langle f_*,e_{\lambda_*}\right\rangle \frac{1}{\lambda_*+1} \left\langle f_*,e_{\lambda_*}\right\rangle\\
&=-\frac{1}{\lambda_*+1}<0\,,
\end{align*}
which concludes the proof.
\end{proof}

\subsection{The Regularized $\Phi_3^4$ Measure on $\mb{R}^{3}$}
\label{SectionCounterexample}

We conclude by giving a counterexample on flat space. To do so, given $f\in L^2(\mb{R}^d)$, we write
$$
(\rho\Pi_\Lambda\Phi)(f) \defeq\int_{\mb{R}^d} f(x)  (\Pi_\Lambda\Phi)(x) \rho(x) \dd  x\,.
$$ 
We would like to summarize in one key lemma the central idea behind our counterexample.
\begin{lemm}\label{keylemma}
Let $A\subset \mb{R}^d$ be a closed subset with nonempty interior and 
$B\subset \mb{R}^{d*}$ a closed compact ball in Fourier space. Denote by $L^2_A(\mb{R}^d)$ the subspace of $L^2$ functions whose supported is contained in $A$.
Then the Fourier restriction map defined as 
$$\CB R_B: L^2_A(\mb{R}^d)\ni\varphi \longmapsto \CB F{\varphi}\upharpoonright{B}\in L^2(B)$$ 
has everywhere dense image.
\end{lemm}
The idea of proof is very similar to the one of Theorem~\ref{thm41}. The simple but powerful idea is that any function on the ball $B$ can be approximated by the Fourier transform of some function supported over $A$.

\begin{proof}
By contradiction, if $\text{Ran}(\CB R_B)$ were not dense, then $\text{Ran}(\CB R_B)^\perp$ would be a non-empty closed vector subspace and there would be some non-null element $g\in L^2(B)$ such that for all $f\in L^2_A$
\begin{eqnarray*}
0 = \big\langle \CB F{f},g\big\rangle _{L^2(B)} = \big\langle \CB F{f},1_{B}g\big\rangle_{L^2(\mb{R}^d)} = \left\langle  f, \CB{F}^{-1}\left( 1_{B}g\right)\right\rangle_{L^2(\mb{R}^d)},
\end{eqnarray*}
from Plancherel identity. The function $\CB{F}^{-1}\left( 1_{B}g\right)$ is therefore not supported on $A$ hence vanishes on some open subset. This contradicts the fact that $\CB{F}^{-1}\left( 1_{B}g\right)$ is a non-trivial analytic function by the Paley-Wiener theorem.
\end{proof}

We now fix $d=3$, since it is the most interesting case. Recall the definition \eqref{eqnu} of the approximate $\Phi_3^4$ measure $\nu_{\rho,\Lambda}$. We emphasize here the dependence of the measures $\nu_{\rho,\Lambda}$ on the coupling constant $c$ by writing $\nu(c,\rho,\Lambda)$, and we set
$$
G_{c,\rho,\Lambda} \defeq e^{- c\Vert\rho\Phi_\Lambda\Vert^4_{L^4} -  c a_\Lambda\Vert\rho\Phi_\Lambda\Vert^2_{L^2}} \leqslant 1.
$$

Denoting by $(x_1,x_2,x_3)$ the canonical coordinates of a point $x\in\mb{R}^3$, we write $f\in L^2(\mb{R}^3_+)$ to say that $f\in L^2(\mb{R}^3)$ and $\text{supp}(f)\subset\{x_1>0\}$.

\begin{thrm} \label{ThmNotRFRegularizedPhi43}
For $c>0$ small enough, there exists a function $f\in L^2(\mb{R}^{3}_+)$ such that
\begin{eqnarray}
\mb{E}_{\nu(c,\rho,\Lambda)}\big[  \overline{(\rho\Pi_\Lambda\Phi)(\Theta f)} \, (\rho\Pi_\Lambda\Phi)(f) \big] < 0.
\end{eqnarray}
This implies that the regularized $\Phi_3^4$ measures $\nu(c,\rho,\Lambda)$ are not reflection positive for $c>0$ small enough, depending on $\Lambda$ and $\rho$.
\end{thrm}

\begin{proof}
Since $0\leqslant G_{c,\rho,\Lambda}(\varphi)$ goes to $1$ as $c\downarrow0$ for every distributions $\varphi$, one has by dominated convergence
\begin{equation} \label{EqConvergenceSmallCoupling}
\mb{E}_{\nu(c,\rho,\Lambda)}\big[  \overline{(\rho\Pi_\Lambda\Phi)(\Theta f)} \, (\rho\Pi_\Lambda\Phi)(f) \big] \overset{c\downarrow0}{\longrightarrow} \mb{E}_{\text{GFF}}\big[ \overline{(\rho\Pi_\Lambda\Phi)(\Theta f)} \, (\rho\Pi_\Lambda\Phi)(f) \big].
\end{equation}
To justify the use of the dominated convergence, at some fixed cut-offs $\rho$ and $\Lambda$, we use the lower bound on the interaction
$$ 
c\Vert\rho\Phi_\Lambda \Vert^4_{L^4(\mb{R}^3)}+ca_\Lambda \Vert\rho\Phi_\Lambda \Vert^2_{L^2(\mb{R}^3)} \geqslant c\int_{\mb{R}^3} \rho^4\Phi_\Lambda^4\,\dd x -c\frac{\vert a_\Lambda \vert\delta^2}{2}\int_{\mb{R}^3} \Phi_\Lambda^4\rho^4\,\dd x- c\frac{\vert a_\Lambda\vert}{2\delta^2} \vol(\text{supp}(\rho)),  
$$
for any $\delta>0$, using Young's inequality and the compactness of the support of $\rho$. Choosing $\delta$ small enough yields a lower bound of the form
$$
c\Vert\rho\Phi_\Lambda \Vert^4_{L^4(\mb{R}^3)}+ca_\Lambda \Vert\rho\Phi_\Lambda \Vert^2_{L^2(\mb{R}^3)}  \geqslant c \frac{2}{3} \Vert\rho\Phi_\Lambda \Vert^4_{L^4(\mb{R}^3)}-cK 
$$
for some $K>0$.

From the convergence result in \eqref{EqConvergenceSmallCoupling} it suffices to find $f\in L^2(\mb{R}^{3}_+)$ such that
$$ 
\mb{E}_{\text{GFF}}\big[ \overline{(\rho\Pi_\Lambda\Phi)(\Theta f)} \, (\rho\Pi_\Lambda\Phi)(f) \big] =  \left\langle \Pi_\Lambda (\rho \overline{\Theta f}) \, , \, (\Delta+1)^{-1}\Pi_\Lambda (\rho f) \right\rangle_{L^2} < 0
$$ 
and choose $c>0$ small enough. Set the Fourier restriction map 
\begin{eqnarray*}
\CB R_\Lambda:  L^2(\mb{R}^{3}_+)\ni f\longmapsto \CB F(\rho f) 1_{B_0(\Lambda)}\in L^2(B_0(\Lambda))\,,
\end{eqnarray*}
where on the right hand side we consider the restriction of the Fourier transform $\CB F(\rho f)$ to the centered Fourier ball of radius $\Lambda$. By a similar argument as in Lemma~\ref{keylemma}, the key idea is to prove that the image of $\CB T$ is dense in $L^2(B_0(\Lambda))$. If it were not true, then $\text{Ran}(\CB R_\Lambda)^\perp$ would be a non-empty closed vector subspace and there would be some non-null element $g\in L^2(B_0(\Lambda))$ such that for every $f\in L^2(\mb{R}^{3}_+)$
\begin{eqnarray*}
0 = \big\langle \CB F(\rho f),g\big\rangle _{L^2(B_0(\Lambda))} = \big\langle \CB F(\rho f),1_{B_0(\Lambda)}g\big\rangle_{L^2(\mb{R}^3)} = \left\langle  f, \rho\CB{F}^{-1}( 1_{B_0(\Lambda)}g)\right\rangle_{L^2(\mb{R}^3)},
\end{eqnarray*}
from Plancherel identity. The function $\rho\CB{F}^{-1}( 1_{B_0(\Lambda)}g)$ is therefore supported outside the half-space $\{x_1>0\}$. Since $\rho>0$ on some non-empty region of $\{x_1>0\}$ this implies that the function $\CB{F}^{-1}( 1_{B_0(\Lambda)}g)$ vanishes in some open subset contained in the half-space $\{x_1>0\}$. This contradicts the fact that $\CB{F}^{-1}( 1_{B_0(\Lambda)}g)$ is a non-trivial analytic function by the Paley-Wiener theorem.

Therefore, the provisional conclusion is that there exists a sequence $f_n$ in $ L^2(\mb{R}^{3}_+)$ such that the sequence $\CB R_\Lambda f_n$ converges in $L^2(B_0(\Lambda))$ to $\xi_1 1_{B_0(\Lambda)}(\xi)$. Therefore for this sequence $(f_n)_{n\geqslant0}$ we find that
\begin{align*} 
{}&\left\langle\Pi_\Lambda (\rho \overline{\Theta f_n})  ,  (\Delta+1)^{-1}\Pi_\Lambda(\rho f_n) \right\rangle_{L^2} = \int_{\mb{R}^{3*}} 1_{[0,\Lambda]}(|\xi|^2) \langle\xi\rangle^{-2} \, \overline{\CB R_\Lambda f_n}(-\xi_1,\xi_2,\xi_3) \CB R_\Lambda f_n(\xi_1,\xi_2,\xi_3) \dd \xi   
\end{align*} 
converges as $n\uparrow\infty$ to
$$
-\int_{{B_0(\Lambda)}} 1_{[0,\Lambda]}(|\xi|^2) \langle\xi\rangle^{-2}\xi_1^2\dd \xi < 0\,.
$$
This shows that the heat regularized (and space localized) Gaussian free field measure is not reflection positive. Of course the Gaussian free field measure itself is reflection positive and the above proof breaks down when $\Lambda\uparrow\infty$ as the function $\xi_1$ alone is not an element of $L^2$.
\end{proof}
Note that Theorem \ref{ThmNotRFRegularizedPhi43} \textbf{does not exclude} the fact that $\nu(c,\rho,\Lambda)$ may be reflection positive for some large enough coupling constant $c$.

\chapter[Entanglement Entropy and Polyakov]{Entanglement Entropy and Cauchy-Hadamard Renormalization}\label{chap-entangle}

This chapter is adapted from \cite{BL} joint with Benoit Estienne.

\section{Introduction}

An (Euclidean) Conformal Field Theory (CFT) is a theory which, in one way, emerges as the \textit{scaling limit} of a lattice statistical mechanics model with \textit{critical parameters} (at phase transition), as the lattice spacing goes to zero. Due to its conformal covariance properties, such a theory can, in two dimensions, naturally be defined on a surface with its conformal structure ($=$ complex structure for oriented surfaces). 
Broadly speaking and for the purposes of this paper,  the data of such a theory consists of assigning to each (closed) Riemann surface~$\Sigma$ with metric~$g$ a number~$\mathcal{Z}(\Sigma,g)$ (the partition function) and a collection of correlation functions defined on~$\Sigma^n_{\ne}$ (subspace of the~$n$-fold product~$\Sigma^n$ with non-coincident points), which satisfy well-defined transformation rules under conformal changes of the metric~$g$ (see definition \ref{def-cft}). In comparison to the original lattice problem, the number~$\mathcal{Z}(\Sigma,g)$ corresponds to the normalization constant for the Gibbs measure, while the functions on~$\Sigma^n_{\ne}$ describe the scaling limit of the probabilistic correlation functions of~$n$ local observables of the spin configurations under the Gibbs measure. 

It has long been recognized in the physics literature that the presence of \textsf{conical singularities} in the underlying metric alters the conformal covariance properties of the partition function $\mathcal{Z}(\Sigma,g)$ in a way that mimics the behavior of correlation functions.  This paper aims to provide a rigorous framework for these observations, with a particular emphasis on partition functions on branched coverings, motivated by considerations of entanglement entropy (which we introduce in Section \ref{sec-ent-ent-main}). We remark that a CFT in the sense of the previous paragraph makes sense, \textit{a priori}, only for smooth metrics and smooth conformal changes. The first task of this paper thus lies in defining in a simple and natural manner the number~$\mathcal{Z}(\Sigma,\tilde{g})$ when the metric~$\tilde{g}$ admits finitely many isolated \textsf{conical singularities}. Then, applying this definition to metrics coming from pull-backs of ramified holomorphic maps, we observe that certain ratios of these numbers behave exactly like what is expected of correlation functions on the target surface where the variables are the critical values. More precisely, the main result is the following.

\begin{thrm}
  [proposition \ref{prop-main-app-ram-cov}] \footnote{See also ``Riemann surface terminologies'' on the next page. Moreover, we will assume that all Riemann surfaces in this paper are connected, unless otherwise specified.} Let~$\Sigma_d$,~$\Sigma$ be closed Riemann surfaces, with a smooth conformal metric~$g$ on~$\Sigma$, and~$f:\Sigma_d\lto \Sigma$ a ramified $d$-sheeted holomorphic map, whose critical values are~$w_1$, \dots,~$w_p$. Consider a conformal field theory with central charge~$c$ whose partition function is denoted~$\mathcal{Z}$. Pick~$h\in C^{\infty}(\Sigma)$ on~$\Sigma$. Then under the definition \ref{def-renom-part-func} for $\mathcal{Z}(\Sigma_d, f^* e^{2h} g)$ and $\mathcal{Z}(\Sigma_d, f^* g)$ we have
	  \begin{equation}
	    \frac{\mathcal{Z}(\Sigma_d, f^* e^{2h} g)}{\mathcal{Z}(\Sigma, e^{2h} g)^d} = e^{- \sum_j  h(w_j)\Delta_{j}} \frac{\mathcal{Z}(\Sigma_d, f^* g)}{\mathcal{Z}(\Sigma, g)^d}, 
	    \label{eqn-intro-conf-weights}
	  \end{equation}
	  where
	  \begin{equation}
	    \Delta_{j} \defeq \frac{c}{12}  \sum_{z \in f^{-1}(w_j)}\Big(\ord_f(z) - \frac{1}{\ord_f(z)} \Big)
	    \label{}
	  \end{equation}
	  are the \textsf{conformal weights} or \textsf{scaling dimensions}. Here $\ord_f(z)$ denotes the \textsf{order} or \textsf{multiplicity} of $f$ at $z\in \Sigma_d$.
\end{thrm}

The relation between CFT and conical singularities has been explored in various physical contexts. Early investigations, particularly within string theory, focused on the construction of orbifold CFTs \cite{Dixon}. In this framework, Knizhnik \cite{Knizhnik} examined the behavior of CFTs defined on flat branched coverings of $\mathbb{CP}^1$. More generally, the emergence of universal logarithmic divergences in the free energy due to conical singularities was recognized by Cardy and Peschel \cite{Cardy_Peschel}. Motivated by both black hole entropy and quantum information theory, entanglement entropy in CFTs was explored in seminal works \cite{Holzhey} and \cite{Cardy_Calabrese}, where it was shown that the universal contributions from conical singularities play a crucial role in understanding the scaling behavior of entanglement entropy in one-dimensional quantum critical systems. In this context, the Rényi entropy was identified as the free energy on $d$-sheeted branched covering of certain flat Riemann surface (typically a cylinder or a torus). Further studies, such as in \cite{Wiegmann}, extended the analysis of conical singularities to hyperbolic surfaces, with applications in the context of the quantum Hall effect \cite{KMMW}. 

Conical singularities is also a relatively well-studied subject mathematically. They are one of the simplest types of singularities that can appear on a Riemann surface. The main motivations include spectral geometry (hearing the shape of a drum) \cite{Cheeger} and the Berger-Nirenberg problem of finding metrics with prescribed curvature \cite{HT, Troy}. Motivated by its higher dimensional analogue in complex geometry called K\"ahler-Einstein edge metrics, people have also studied the Ricci flow on surfaces with such singularities \cite{MRS}. Particularly relevant to the present work are recent investigations on the~$\zeta$-determinant of the Laplacian under these conical metrics and the Polyakov formulas \cite{AKR, Kalvin}, on which we shall make a more detailed comment in remarks \ref{rem-rel-kal} and \ref{rem-rel-akr}. We would also like to point out the interesting recent work \cite{JV} on Coulomb gas and the Grunsky operator where conical singularities (more precisely, ``corners'' on the boundary) also play a role. For more literature from these various perspectives we refer to the introductions of \cite{AKR, Kalvin}, section 1.2 of \cite{JV}, and section 2.E of \cite{MRS}. Finally, one precise relation between conical singularities and the probabilistically constructed Liouville CFT has been discussed in \cite{HRV,TNA}.

One important feature of the present work lies in its \textit{simplicity} and \textit{naturalness}. Essentially, the method involves only a close look at the geometry near the cone points and integration by parts (Green-Stokes formula). Consequently, the geometric meaning of each term that shows up in the result is transparent (see also remark \ref{rem-extra-quad-term}). Moreover, we only need rather weak regularity to be imposed on the regular metric potential at the cone points compared to other related works in the literature (e.g.\ to \cite{Kalvin} definition 2.1), and this is basically the assumption adopted in \cite{Troy} (see also remark \ref{rem-reg-met-pot-cone}). 

\paragraph{Organization.} In section \ref{sec-def} we define the three main objects dealt with by this paper: the CFT correlation functions, conical singularities, and the renormalized Polyakov anomaly; these are accompanied by a few pivotal lemmas. Section \ref{sec-ent-ent-main} explains the physical motivation which leads us to (\ref{eqn-intro-conf-weights}). In these sections (and hence the whole paper) no prior knowledge of QFT, CFT or statistical mechanics is assumed, as this paper serves also to introduce these ideas to the mathematics community. Only some notions of quantum mechanics are required to make sense of entanglement entropy. Then in section \ref{sec-geo-lemm}, we collect a few known facts about conical metrics (subsection \ref{sec-conic-remark}), prove a crucial scaling lemma (subsection \ref{sec-geo-scaling}), and compute the asymptotics of several logarithmically divergent integrals using integration by parts (subsection \ref{sec-int-by-part}). They are important as pointed out by remarks \ref{rem-well-def-renorm-anom}, \ref{rem-const-def-part-func}, and also in obtaining the final result. In section \ref{sec-renorm-tech} we tie up some loose ends around the definitions of the renormalized anomaly and partition function for conical metrics. In subsection \ref{sec-main-proof} we prove the main result and apply it to entanglement entropies, and in \ref{sec-literature} we comment on the two closely related work \cite{AKR, Kalvin}. Finally in appendix \ref{sec-app-poin-lelong} we recall a few things around the so-called ``Poincar\'e-Lelong lemma'' concerning the Laplacian of the log of the distance function on a Riemannian surface.

\paragraph{Future work.} In the present version of the work we have focused on the simplest case dealing only with partition functions to illustrate our (already simple) methods. It is clear that the same arguments could apply with minor modifications to obtain analogous results for Segal's amplitudes on surfaces with boundary (see \cite{Gaw} section 2.6) when the metric is ``flat-at-the-boundary'', and for the case of general boundaries with possible presence of ``corners'' (i.e.\ ``polygons'') by including the boundary term in the anomaly, as considered for example by \cite{AKR}. These considerations may be included in a future update of the present manuscript. However, another main aim of that work will be to present an equivalent rigorous construction of the quantities~$\mathcal{Z}(\Sigma_d,f^*g)/\mathcal{Z}(\Sigma,g)^d$ in (\ref{eqn-intro-conf-weights})  for the GFF by constructing the so-called ``twist fields'' in the physics literature. This is related to a singular version of the ``twisted Laplacians'' used e.g.\ in Phillips and Sarnak \cite{PSa}. Lastly, less apparent but interesting future investigations include obtaining the precise relation between this work and \cite{Kalvin} in light of Segal's gluing axioms, and exploring further connections with the Quantum Hall Effect and Coulomb gas in a rigorous manner.

\paragraph{Riemann surface terminology and asymptotic notations.} Let~$f:\Sigma'\lto \Sigma$ be a holomorphic map of closed (connected\footnote{We will assume that all Riemann surfaces in this paper are connected, unless otherwise specified.}) Riemann surfaces. If at~$z\in \Sigma'$ we have~$\dd f|_z=0$ we say that~$z$ is a \textsf{critical point} of~$f$ and~$w=f(z)$ is a \textsf{critical value} of~$f$. 
If in local holomorphic charts around~$z$ and~$w$ respectively~$f$ looks like a~$k$-th power, then we say~$k$ is the \textsf{order} of~$f$ at~$z\in \Sigma'$ (sometimes referred to as the \textsf{local degree} or \textsf{multiplicity}),  denoted~$k=:\ord_f(z)$. 
We say that~$f$ is a \textsf{ramified} (or \textsf{branched}) holomorphic map when we want to emphasize that it is not a strict covering map.
The number~$\sum_{z\in f^{-1}(w)}\ord_f(z)=:d$, which is same for every~$w\in \Sigma$, is called the \textsf{number of sheets} of~$f$ and we say~$f$ is~\textsf{$d$-sheeted}. Finally, the notation $u\asymp v$ under some limit process means $u/v\to C>0$ in that limit, and we write $u\sim v$ for the case $C=1$.

  \section{Definition of Main Objects}\label{sec-def}

 \subsection{Conformal Field Theory and Correlation Functions}

\begin{deef}\label{def-cft}
  In this paper, by a \textsf{2d Conformal Field Theory} we mean a rule that associates to each compact Riemannian surface~$\Sigma$ with metric~$g$ (a priori smooth) a complex number~$\mathcal{Z}(\Sigma,g)$ called the \textsf{partition function}, and a family~$\{\ank{\phi_{\alpha_1}(\cdot)\cdots}_{\Sigma,g}\}$ of functions of finite tuples of non-coincident points on~$\Sigma$, called \textsf{correlation functions of primary fields} (labelled by the~$\alpha$'s), such that the following two conditions hold:
  \begin{enumerate}[(i)]
    \item \textsf{diffeomorphism invariance:} if~$\Psi:\Sigma'\lto \Sigma$ is a diffeomorphism of smooth surfaces, then
      \begin{align}
	\mathcal{Z}(\Sigma',\Psi^*g)&=\mathcal{Z}(\Sigma,g), \label{eqn-part-func-diffeo-inv}\\
	\ank{\phi_{\alpha_1}(\Psi(x_1))\cdots \phi_{\alpha_n}(\Psi(x_n))}_{\Sigma,g}&=\ank{\phi_{\alpha_1}(x_1)\cdots \phi_{\alpha_n}(x_n)}_{\Sigma',\Psi^*g},
	\label{eqn-main-diffeo-inv}
      \end{align}
      for any~$x_1$, \dots,~$x_n\in \Sigma$ non-coincident, and as well
    \item \textsf{local scale (conformal) covariance:} if~$\sigma\in C^{\infty}(\Sigma)$ then
      \begin{align}
	\mathcal{Z}(\Sigma,\me^{2\sigma}g)&=\exp\Big( \frac{c}{24\pi}\int_{\Sigma}^{}(|\nabla_g \sigma|_g^2+2K_g\cdot\sigma)\dd V_g \Big)\cdot \mathcal{Z}(\Sigma,g),\label{eqn-ord-polya-anom}\\
	\ank{\phi_{\alpha_1}(x_1)\cdots \phi_{\alpha_n}(x_n)}_{\Sigma,\me^{2\sigma}g}&=\prod_{j=1}^n \me^{-\sigma(x_j)\Delta_{\alpha_j}}\ank{\phi_{\alpha_1}(x_1)\cdots \phi_{\alpha_n}(x_n)}_{\Sigma,g},
	\label{eqn-main-conformal-cov}
      \end{align}
      where~$K_g$ denotes the Gauss curvature of~$g$ (half the scalar curvature), the constant~$c\in \mb{R}$ is called the \textsf{central charge}, charateristic of the specific theory at hand, and constants~$\Delta_{\alpha}\in\mb{R}$ the \textsf{conformal weights}, charateristic of the theory as well as the fields~$\phi_{\alpha}$. 
      \end{enumerate}
\end{deef}

\begin{deef}
  The quantity
  \begin{equation}
    A_{\Sigma}(\me^{2\sigma}g,g) \defeq  \frac{1}{24\pi}\int_{\Sigma}^{}(|\nabla_g \sigma|_g^2+2K_g\cdot\sigma)\,\dd V_g 
    \label{eqn-def-ord-poly-anomaly}
  \end{equation}
  for ~$\sigma\in C^{\infty}(\Sigma)$, that appears in the exponential in (\ref{eqn-ord-polya-anom}) is usually called the \textsf{Weyl} or \textsf{Polyakov Anomaly} of $\me^{2\sigma}g$ against $g$. 
\end{deef}

\begin{def7}\label{rem-ord-cocycle}
  The quantities~$A_{\Sigma}$ has the so-called \textsf{cocycle property}, that is, if~$g_3=\me^{2h}g_2$ and~$g_2=\me^{2\sigma}g_1$ where~$h$,~$\sigma\in C^{\infty}(\Sigma)$, then
  \begin{equation}
    A_{\Sigma}(g_3,g_2)+A_{\Sigma}(g_2,g_1)=A_{\Sigma}(g_3,g_1),
    \label{}
  \end{equation}
  as one can check using the relations (\ref{eqn-scale-vol}) --- (\ref{eqn-scale-lap}). One essential ingredient of this work is to define a ``renormalized'' version of the anomaly of a conically singular metric against a smooth one (definition \ref{def-renorm-polya-conic}). With these definitions the cocycle property will be modified accordingly when involving the conical metrics, see proposition \ref{prop-main-conical-scaling}. 
\end{def7}

\begin{exxx}[Gaussian Free Field] \label{exp-gff}
  Let~$(\Sigma,g)$ be a Riemannian surface with smooth metric~$g$ whose Laplacian is denoted~$\Delta_g$ (negative). Consider
  \begin{equation}
    \mathcal{Z}(\Sigma,g)\defeq \detz'(-\Delta_g)^{-\frac{1}{2}},
    \label{}
  \end{equation}
  where
  \begin{equation}
    \detz'(-\Delta_g)\defeq \exp\Big( -\frac{\dd}{\dd s}\Big|_{s=0}\sum_{j=1}^{\infty}\lambda_j^{-s} \Big),
    \label{}
  \end{equation}
  with~$0=\lambda_0<\lambda_1\le \lambda_2\le\cdots$ being the eigenvalues of~$-\Delta_g$ counted with multiplicity, called the~\textsf{$\zeta$-regularized determinant} \cite{RS}. Then~$\mathcal{Z}$ satisfies (\ref{eqn-ord-polya-anom}) (and (\ref{eqn-part-func-diffeo-inv}) trivially) with~$c=1$, by the \textsf{Polyakov formula} (see \cite{peltola-wang} appendix B). This~$\mathcal{Z}$ corresponds to the ``total mass'' of the formal measure~$\exp(-\frac{1}{2}\int_{\Sigma}^{}|\nabla_g \phi|^2 \dd V_g)\dd \mathcal{L}(\phi)$ on the space of zero-average distributions~$\mathcal{D}'_0(\Sigma)$ on~$\Sigma$, with~$\mathcal{L}$ being the non-existent Lebesgue measure there. In this case, an actual Gaussian probability measure~$\mu_{\mm{GFF}}^{\Sigma}=:\mu$ can indeed be defined on~$\mathcal{D}'_0(\Sigma)$ that corrresponds to~$\mathcal{Z}(\Sigma,g)^{-1}\exp(-\frac{1}{2}\int_{\Sigma}^{}|\nabla_g \phi|^2 \dd V_g)\dd \mathcal{L}(\phi)$, called the (massless) \textsf{Gaussian Free Field} \cite{powell-werner}. This is characterized by the formal covariance property (pretending that the point values $\phi(x)$, $x\in\Sigma$, are legitimate real random variables)
  \begin{equation}
    \mb{E}_{\mu}\big[\phi(x)\phi(y)\big]=G_{\Sigma}(x,y),\quad\textrm{and}\quad\mb{E}_{\mu}\big[\phi(x)\big]\equiv0,\quad\quad x,y\in \Sigma,
    \label{}
  \end{equation}
  with~$G_{\Sigma}$ being the Green function which is the integral kernel of~$(-\Delta_g)^{-1}P_{\ker \Delta_g}^{\perp}$. Now, after an appropriate renormalization process which we do not detail here (related to obtaining a finite value for~$G_{\Sigma}(x,x)$), denoted~``$\mathcal{R}$'', the quantities
  \begin{equation}
    \bank{\phi_{\alpha_1}(x_1)\cdots \phi_{\alpha_n}(x_n)}_{\Sigma,g}\defeq \textrm{``}\mathcal{R}\textrm{''}\mb{E}_{\mu}\big[\me^{\ii \alpha_1\phi(x_1) }\cdots \me^{\ii \alpha_n\phi(x_n) }\big],
    \label{}
  \end{equation}
  where now~$\alpha_j\in\mb{R}$, transform according to (\ref{eqn-main-conformal-cov}) with~$\Delta_{\alpha_j}:=\alpha_j^2/4\pi$.
\end{exxx}

The rules (i), (ii) actually determine, say, the two-point function up to a constant for some simple geometries.

\begin{lemm}\label{lemm-two-pt-func}
  Consider a CFT defined on the Riemann sphere~$\mb{S}^2=\mb{C}\cup\{\infty\}$, equipped with the Fubini-Study metric~$g_{\mm{FS}}(z):=4(1+|z|^2)^{-2}|\dd z|^2$. Take two primary fields~$\phi_1$,~$\phi_2$ with conformal weights~$\Delta_1$,~$\Delta_2$. Then
  \begin{equation}
  \bank{\phi_1(u)\phi_2(v)}_{\mb{S}^2,g_{\mm{FS}}}=
  \left\{
  \begin{array}{ll}
     C\sin\big( \frac{1}{2}d_{\mm{FS}}(u,v) \big)^{-2\Delta},&\textrm{when }\Delta_1=\Delta_2=\Delta,\\
    0 &\textrm{otherwise,}
  \end{array}
  \right.
  \label{}
\end{equation}
where $u\ne v$ and~$C$ is a non-zero constant that cannot be determined from the rules (i) and (ii) alone.
\end{lemm}

\begin{proof}
Without loss of generality we suppose $u$, $v\ne \infty$. Given a M\"obius transformation~$\psi\in \mm{PSL}(2,\mb{C})$ that sends~$0\mapsto u$,~$\infty\mapsto v$, we have by (\ref{eqn-main-diffeo-inv}) and (\ref{eqn-main-conformal-cov}),
	\begin{align}
	  \bank{\phi_1(u)\phi_2(v)}_{\mb{S}^2,g_{\mm{FS}}}
	  &=\bank{\phi_1(0)\phi_2(\infty)}_{\mb{S}^2,\psi^* g_{\mm{FS}}}\nonumber\\
	  &=\me^{-\Delta_1 \sigma(0)-\Delta_2\sigma(\infty)}\bank{\phi_1(0)\phi_2(\infty)}_{\mb{S}^2,g_{\mm{FS}}},\label{eqn-two-point-examp-calc}
	\end{align}
	where~$\psi^*g_{\mm{FS}}=\me^{2\sigma}g_{\mm{FS}}$. The above must hold for \textit{all} such M\"obius maps. These are of the form
	\begin{equation}
	\psi(z)=\frac{z v - z_0 u}{z - z_0}.
	\label{}
      \end{equation} 
      where $z_0  = \psi^{-1}(\infty) \in \mathbb{C}\setminus \{0 \}$ parametrizes the preimage of infinity. 
      This gives
      \begin{equation}
	(\psi^* g_{\mm{FS}})(z)=\frac{4|\psi'(z)|^2}{(1+|\psi(z)|^2)^2}|\dd z|^2=\frac{(1+|z|^2)^2}{(1+|\psi(z)|^2)^2}\frac{|v-u|^2 |z_0|^2}{|z - z_0|^4} g_{\mm{FS}}(z).
	  \label{}
	\end{equation}
 Therefore, according to (\ref{eqn-two-point-examp-calc}),
\begin{equation}
  \bank{\phi_1(u)\phi_2(v)}_{\mb{S}^2,g_{\mm{FS}}}=C\Big(\frac{1+|u|^2}{|u -v|}\Big)^{\Delta_1} \Big(\frac{1+|v|^2}{|u -v|}\Big)^{\Delta_2} |z_0|^{\Delta_2- \Delta_1}
\end{equation}
where $C =: \ank{\phi_1(0)\phi_2(\infty)}_{\mb{S}^2,g_{\mm{FS}}}$.
Since the correlation function should not depend on $z_0$, one finds 
\begin{equation}
  \textrm{either } \Delta_1 = \Delta_2 \quad \textrm{or} \quad C=0.
  \label{}
\end{equation}
Summing up,
\begin{equation}
  \bank{\phi_1(u)\phi_2(v)}_{\mb{S}^2,g_{\mm{FS}}}=
  \left\{
  \begin{array}{ll}
    C \left(1+|u|^2\right)^{\Delta} \left(1+|v|^2\right)^{\Delta}  |u - v|^{-2 \Delta}
    ,&\textrm{when }\Delta_1=\Delta_2=\Delta,\\
    0 &\textrm{otherwise.}
  \end{array}
  \right.
  \label{}
\end{equation}
where the constant $C  \in \mathbb{C}$ is $C= \bank{\phi_1(0)\phi_2(\infty)}_{\mb{S}^2,g_{\mm{FS}}}$.
We note here that the quantity
	\begin{equation}
	  d_{\mm{ch}}(u,v)\defeq \frac{2|u-v|}{\sqrt{(1+|u|^2)(1+|v|^2)}}
	  \label{}
	\end{equation}
	is the so-called \textsf{spherical chordal distance} between~$u$ and~$v$ (see \cite{Ahlfors} page 20), and is related to the actual spherical distance~$d_{\mm{FS}}(u,v)$ by
	\begin{equation}
	  d_{\mm{ch}}(u,v)=2\sin\Big( \frac{d_{\mm{FS}}(u,v)}{2} \Big).
	  \label{}
	\end{equation}
 This gives us the result.
\end{proof}

      \subsection{Conical Metrics}
\label{sec-singular-conic}

In this section we describe precisely our geometric set-up by adopting some of the terminologies of Troyanov \cite{Troy}.
	\begin{deef}
	  Let~$\Sigma$ be a closed Riemann surface. A \textsf{generalized conformal metric} on~$\Sigma$ is a distributional Riemannian metric~$\tilde{g}$ on~$\Sigma$ such that for any local complex coordinate~$z_U:U\lto \mb{C}$ defined on~$U\subset \Sigma$ we have
	  \begin{equation}
	    \tilde{g}=\varrho_U(z_U)\,|\dd z_U|^2
	    \label{}
	  \end{equation}
	  for some positive measurable function~$\varrho_U$ on~$U$. We say~$\tilde{g}$ is a \textsf{smooth conformal metric} if the functions~$\rho_U$ are all smooth and $0<c_{U,z}\le \varrho_U\le C_{U,z}$, for some constants $c_{U,z}$, $C_{U,z}$ depending on $U$ and the coordinate.
	\end{deef}

 Smooth conformal metrics exist on any Riemann surface by simple constructions using partitions of unity.

	\begin{deef}
	  A \textsf{real divisor} on a closed Riemann surface~$\Sigma$ is a finite formal sum
	  \begin{equation}
	    D=\sum_{j=1}^p \gamma_j z_j,\quad\quad \gamma_j\in\mb{R},~z_j\in\Sigma,~z_i\ne z_j\textrm{ for }i\ne j.
	    \label{}
	  \end{equation}
	  The number~$|D|:=\sum_j \gamma_j$ is the \textsf{degree} of~$D$. The set $\supp D:=\{z_j~|~\gamma_j\ne 0\}$ is the \textsf{support} of $D$.
	\end{deef}

    	\begin{deef}\label{def-singularities}
	  Let~$\Sigma$ be a closed Riemann surface. We say that a generalized conformal metric~$\tilde{g}$ has an (admissible, isolated) \textsf{conical singularity} (or \textsf{cone point}) of \textsf{order}~$d>0$ or \textsf{exponent} $\gamma:=d-1>-1$ at~$z_0\in\Sigma$, if~$z_0$ has a neighborhood~$U\subset \Sigma$, such that
\begin{equation}
  \tilde{g}|_{U\setminus z_0} =d_g(\bullet,z_0)^{2\gamma}\cdot \me^{2\varphi_{U,g}}g,
  \label{eqn-conic-sing-local-form}
\end{equation}
for some smooth conformal metric~$g$ on~$\Sigma$ and some function~$\varphi_{U,g}$ on~$U$, called the \textsf{(regular) metric potential} (against $g$), which is continuous on $U$, smooth on $U\setminus z_0$, and $\Delta_g \varphi_{U,g}\in L^1(U,g)$. Moreover, we require that the Gauss curvature of~$\tilde{g}$ on~$U\setminus z_0$ is bounded.
	\end{deef}

    \begin{def7}
  If~$\gamma=0$ (resp.\ $d=1$) then the point is called \textsf{regular}. In a neighborhood~$U$ of a regular point, the regularity of the metric potential~$\varphi_{U,g}$ is closely related to that of the Gauss curvature (given by (\ref{eqn-liouville-eqn}) as an~$L^1(U,g)$ function) of~$\tilde{g}$ on~$U$. For example, if the Gauss curvature is smooth then one could deduce that~$\varphi_{U,g}$ is also smooth. For details see \cite{Troy2} proposition 1.2. Our results will be consistent with the presence of regular points upon setting the corresponding exponents to zero, as one could check. 
\end{def7}

	\begin{deef}
	  Let~$\Sigma$ be a closed Riemann surface,~$D=\sum_{j=1}^p \gamma_j z_j$ be a real divisor with $\gamma_j>-1$. We say that a generalized conformal metric~$\tilde{g}$ \textsf{represents} the divisor~$D$ if it has an isolated conical singularity of exponent $\gamma_j$ at $z_j$ respectively for each $j$, and smooth on $\Sigma\setminus \supp D$.
	\end{deef}

\begin{def7}
	  Our presentation differs from what usually happens in the literature where the background metric~$g$ is the local flat one coming from some compatible complex coordinate on the neighborhood~$U$. We are faced with the natural question of whether~$\tilde{g}$ would have the same form (\ref{eqn-conic-sing-local-form}) with the same regularity on the metric potential if we switched to another smooth conformal background metric~$g_1=\me^{2h}g$,~$h\in C^{\infty}(\Sigma)$. Indeed, in this case
	  \begin{equation}
	    d_g(\bullet,z_0)^{2\gamma}\cdot \me^{2(\varphi_{U,g}-h)}g_1=d_{g_1}(\bullet,z_0)^{2\gamma}\Big( \frac{d_{g}(\bullet,z_0)}{d_{g_1}(\bullet,z_0)} \Big)^{2\gamma}\cdot \me^{2(\varphi_{U,g}-h)}g_1,
	    \label{}
	  \end{equation}
	  and the ratio of distances~$d_{g}(\bullet,z_0)/d_{g_1}(\bullet,z_0)$ is continuous on~$U$ with a limit~$\me^{-h(z_0)}$ at~$z_0$, smooth on~$U\setminus z_0$ and~$\Delta_{g_1}\log d_g -\Delta_{g_1}\log d_{g_1}$ is in $L^1(U,g_1)$ as a distribution by lemma \ref{lemm-bound-lap-log-ratio}. In the traditional setting if both~$g$ and~$g_1$ are locally~$|\dd z|^2$ and~$|\dd w|^2$ for some complex coordinates~$z$ and~$w$, then the regular metric potentials differ by a harmonic function, as is well known.
In particular, if we do test against such a coordinate metric then our conditions on the regular metric potential satisfy what Troyanov calls that of an \textsf{admissible metric} in \cite{Troy2}. 
\end{def7}

The following lemma describes the source of the conical singularities that the main application of the main result of this article tries to target.

\begin{lemm}
     Let~$\Sigma'$,~$\Sigma$ be closed Riemann surfaces and~$f:\Sigma'\lto \Sigma$ a holomorphic map. Let~$z_0\in \Sigma'$ be a critical point of order~$k$ and~$w_0=f(z_0)\in \Sigma$ the corresponding critical value. Let~$g$ be a smooth conformal metric on~$\Sigma$. Then~$f^*g$ has a conical singularity at~$z_0\in \Sigma'$ with order~$k$.
  \end{lemm}

 \begin{proof}
    There exists holomorphic charts $(U,z)$ and $(V,w)$ around $z_0$ and $w_0 = f(z_0)$ in which $f$ is represented  by $z\mapsto z^k$. The metric $g$ is locally of the form $g\big|_{V} = e^{2 h} |\dd w|^2$ for some $h \in C^{\infty}(V)$, and 
     \begin{equation}
f^*g\big|_{U} =  e^{2 f^*h} \,  |\dd f|^2 = e^{2 f^*h} \, k^2  |z|^{2(k-1)} |\dd z|^2  = d_{g'}(\bullet,z_0)^{2(k-1)} \, k^2e^{2 f^* h}  \, g'\big|_{U}
 \label{eqn-proof-pull-back-met-con}
    \end{equation}  
where $g'$ is any smooth conformal metric on $\Sigma'$ whose restriction to $U$ is $g'\big|_{U} = |\dd z|^2$ (such a metric can be manufactured using smooth bump functions, and taking $U$ smaller if necessary). Now (\ref{eqn-proof-pull-back-met-con})  is of the form (\ref{eqn-conic-sing-local-form}) with metric potential~$\varphi= f^* h + \log k$. Note that $\varphi$ is actually smooth across~$U$ in this situation.
  \end{proof}

      \subsection{Renormalized Anomaly}\label{sec-def-renorm-anom-part}

\begin{deef}\label{def-main-renorm-anomaly}
    Let~$\Sigma$ be a closed Riemann surface and~$\tilde{g}$ a generalized conformal metric representing~$D=\sum_{j=1}^p \gamma_j z_j$ with~$\gamma_j>-1$. Suppose~$g$ is a smooth conformal metric on~$\Sigma$ and~$\tilde{g}=\me^{2\sigma}g$ for some~$\sigma\in C^{\infty}(\Sigma\setminus \supp D)$. We define the \textsf{renormalized Polyakov anomaly} against metric~$g$ to be
  \begin{equation}
   \mathcal{R}A_{\Sigma}(\tilde{g},g)\defeq \frac{1}{24\pi}\lim_{\varepsilon\to 0^+}\Big[ 
      \int_{\Sigma\setminus \bigcup_{i=1}^p B_{\varepsilon}(z_i,\tilde{g})}(|\nabla_g \sigma|_g^2+2K_g\sigma)\dd V_g +2\pi\sum_{i=1}^p \frac{\gamma_i^2}{1+\gamma_i}\log(\varepsilon)
    \Big],
    \label{eqn-def-renorm-polya}
  \end{equation}
where~$B_{\varepsilon}(z_i,\tilde{g})$ is the metric disk of radius~$\varepsilon$ centered at~$z_i$ under the metric~$\tilde{g}$.
\end{deef} 

\begin{def7} \label{rem-well-def-renorm-anom} To show that $\mathcal{R}A_{\Sigma}(\tilde{g},g)<\infty$ we first show that it is true for~$g$ locally equal to~$|\dd w_j|^2$ near each~$z_j$ with the coordinate~$w_j$ coming from lemma \ref{lemm-reg-metric-pot} (Troyanov form). In this case the result follows from corollary \ref{cor-quad-blow-tilde}. Then we use lemma \ref{lemm-consist} with~$g_0$ being the special Troyanov form and~$g_1$ generic, to extend to the case of a generic smooth~$g$.
\end{def7}

\begin{def7}
    This method of renormalization dates back to Cauchy and Hadamard as they defined ``principal values'' of divergent integrals.
  \end{def7}

	\begin{deef}\label{def-renom-part-func}
	  Consider a conformal field theory with central charge~$c\in\mb{R}$ on the Riemann surface~$\Sigma$ and let~$\tilde{g}$ be a generalized conformal metric representing~$D=\sum_{j=1}^p \gamma_j z_j$ with~$\gamma_j>-1$. We define the \textsf{renormalized partition function} of $(\Sigma,\tilde{g})$ to be
	  \begin{equation}
	    \mathcal{Z}(\Sigma,\tilde{g})\defeq \exp\big(c\mathcal{R}A_{\Sigma}(\tilde{g},g)\big)\mathcal{Z}(\Sigma,g),
	    \label{eqn-def-renom-part-func}
	  \end{equation}
	  where~$g$ is any smooth conformal metric on~$\Sigma$ such that~$\tilde{g}=\me^{2\sigma}g$ for some~$\sigma\in C^{\infty}(\Sigma\setminus \supp D)$. 
	\end{deef}

\begin{def7}\label{rem-const-def-part-func} Lemma \ref{lemm-consist} ensures that $\mathcal{Z}(\Sigma,\tilde{g})$ is independent from the choice of reference metric $g$.
\end{def7}

      \section{Physical Interpretation}\label{sec-ent-ent-main}

In this section we explain in detail the so-called \textsf{entanglement entropies} of quantum systems and a technique for computing them using \textsf{path integrals}, called the \textsf{replica trick} in the physics community \cite{Cardy_Calabrese, Hea, RT, Witten}. This leads us to consider precisely the ratios of the form shown in (\ref{eqn-intro-conf-weights}) and to expect that they behave like correlation functions. We remark that what we present here is more of a conjectural, heursitic framework than rigorous proofs.

\subsection{Partial Trace and Entanglement}
In this subsection we discuss some generalities on describing entanglement between quantum systems. We start with a generic quantum system associated to a Hilbert space~$\mathcal{H}$. Recall that a \textsf{(quantum) statistical ensemble} (or \textsf{statistical state}) on~$\mathcal{H}$ is represented by a nonnegative trace class operator~$\rho$ on~$\mathcal{H}$ with~$\ttr_{\mathcal{H}}(\rho)=1$. This is usually called a \textsf{density operator}. Denote by~$\mathcal{J}_1(\mathcal{H})$,~$\mathcal{J}_{\infty}(\mathcal{H})$ and~$\mathcal{L}(\mathcal{H})$ the trace class, compact and bounded operators on~$\mathcal{H}$.

Now suppose the quantum system can be decomposed into two subsystems~$A$ and~$B$. In other words, assume that the Hilbert space~$\mathcal{H}$ is a tensor product:
      \begin{equation}
	\mathcal{H}=\mathcal{H}_A\otimes \mathcal{H}_B.
	\label{eqn-hilb-tens-decomp-gen}
      \end{equation}
      In order to describe the relation of the subsystems to the whole system, a useful operation is called the \textsf{partial trace}. This corresponds to taking trace ``over one component'' and one is left with an operator acting on the ``remainder component''. It could be defined rigorously as follows.
For each~$\rho\in \mathcal{J}_1(\mathcal{H})$, consider the linear functional
    \begin{equation}
      C\longmapsto \ttr_{\mathcal{H}}(\rho(C\otimes \one_B))
      \label{}
    \end{equation}
    for~$C\in \mathcal{J}_{\infty}(\mathcal{H}_A)$. We have
    \begin{equation}
      \big|\ttr_{\mathcal{H}}\big(\rho(C\otimes \one_B)\big)\big|\le \bnrm{\rho}_{\mathcal{J}_1(\mathcal{H})}\bnrm{C\otimes \one_B}_{\mathcal{L}(\mathcal{H})}= \bnrm{\rho}_{\mathcal{J}_1(\mathcal{H})}\bnrm{C}_{\mathcal{L}(\mathcal{H}_A)}.
      \label{eqn-gen-part-trace-bound}
    \end{equation}
    Thus (\ref{eqn-gen-part-trace-bound}) defines a bounded linear functional on~$\mathcal{J}_{\infty}(\mathcal{H}_A)$ with norm~$\nrm{\rho}_{\mathcal{J}_1(\mathcal{H})}$. Since~$\mathcal{J}_1(\mathcal{H}_A)$ is the Schatten dual of the ideal~$\mathcal{J}_{\infty}(\mathcal{H}_A)$ we obtain a unique representative~$\ttr_B(\rho)\in \mathcal{J}_1(\mathcal{H}_A)$.

    \begin{deef}
	  Let the Hilbert space~$\mathcal{H}$ be a tensor product as in (\ref{eqn-hilb-tens-decomp-gen}). Then the \textsf{partial trace over~$\mathcal{H}_B$} is the linear map
	  \begin{equation}
	    \left.
	    \begin{array}{rcl}
	       \ttr_B:\mathcal{J}_1(\mathcal{H}) &\lto& \mathcal{J}_1(\mathcal{H}_A),\\
	       \rho &\longmapsto &\ttr_B(\rho),
	    \end{array}
	    \right.
	    \label{}
	  \end{equation}
	  such that~$\ttr_B(\rho)$ is the unique operator on~$\mathcal{H}_A$ satisfying
	  \begin{equation}
	    \ttr_{\mathcal{H}_A}\big(\ttr_B(\rho)C\big) =\ttr_{\mathcal{H}}\big(\rho(C\otimes \one_B)\big)
	    \label{eqn-part-trace-dual-def}
	  \end{equation}
	  for all bounded operators~$C\in \mathcal{L}(\mathcal{H}_A)$.
	\end{deef}

    	\begin{def7}\label{rem-property-partial-trace}
	  For more details see Simon \cite{Sim4} pages 269-270, and also \cite{Sim1} theorem 3.2. In particular, from (\ref{eqn-gen-part-trace-bound}) one can see that the partial trace is continuous with respect to the~$\mathcal{J}_1(\mathcal{H})$- and~$\mathcal{J}_1(\mathcal{H}_A)$-norms (trace norms), and by taking~$C=|\psi\rangle\langle \psi|$,~$\psi\in \mathcal{H}_A$, in (\ref{eqn-part-trace-dual-def}) one could also see that~$\ttr_B(\rho)$ is nonnegative if~$\rho$ is.
	\end{def7}
      
      \begin{def7}
  Restricted to rank-one operators, we have
  \begin{equation}
    \ttr_B\big( |\varphi_A\otimes \varphi_B\rangle\langle \psi_A\otimes \psi_B| \big)=\langle \psi_B|\varphi_B\rangle |\varphi_A\rangle\langle \psi_A|,
    \label{}
  \end{equation}
  with~$|\varphi_A\rangle$,~$|\psi_A\rangle\in \mathcal{H}_A$,~$|\varphi_B\rangle$,~$|\psi_B\rangle\in \mathcal{H}_B$. Here the notation~$|w\rangle\langle v|$ for~$w$,~$v\in \mathcal{H}$ denotes the operator mapping any~$u\in \mathcal{H}$ to~$\ank{v,u}_{\mathcal{H}}w$. See also example \ref{exp-mat-part-trace}. 
\end{def7}

\begin{deef}
	Suppose~$\rho$ is a density operator over~$\mathcal{H}=\mathcal{H}_A\otimes \mathcal{H}_B$. Then the \textsf{reduced density operator over~$A$} is defined by
	\begin{equation}
	  \rho_A\defeq \ttr_B(\rho).
	  \label{}
	\end{equation}
      \end{deef}
It follows from remark \ref{rem-property-partial-trace} that $\rho_A$ is a density operator over~$\mathcal{H}_A$.
      \begin{deef}
Suppose~$\rho$ is a density operator over~$\mathcal{H}=\mathcal{H}_A\otimes \mathcal{H}_B$. We define
\begin{align}
  \textrm{\textsf{the entanglement entropy (over }}A\textrm{\textsf{)}} && \mathcal{S}_A&\defeq -\ttr_{\mathcal{H}_A}(\rho_A \log \rho_A), \\
  \textrm{\textsf{the }}n\textrm{\textsf{-th R\'enyi entropy (over }}A\textrm{\textsf{)}} && \mathcal{S}^{(n)}_A &\defeq \frac{1}{1-n}\log\ttr_{\mathcal{H}_A}(\rho_A^n),\quad \textrm{for }n>1.
  \label{eqn-def-renyi-ent}
\end{align}
      \end{deef}

\begin{def7}
	In general the quantities~$\mathcal{S}_A$ and~$\mathcal{S}_A^{(n)}$ defined above do not equal~$\mathcal{S}_B$ and~$\mathcal{S}_B^{(n)}$, the corresponding quantities defined the other way round by first taking partial trace over~$A$. However,~$\mathcal{S}_A=\mathcal{S}_B$ and~$\mathcal{S}_A^{(n)}=\mathcal{S}_B^{(n)}$ do hold if~$\rho$ is a \textsf{pure state}, namely a rank-one projector written as~$\rho=|\psi\rangle\langle \psi|$ for some~$|\psi\rangle\in \mathcal{H}$. Therefore, for pure states it makes sense to define \textit{the} entanglement and R\'enyi entropies~$\mathcal{S}:=\mathcal{S}_A=\mathcal{S}_B$ and~$\mathcal{S}^{(n)}:=\mathcal{S}_A^{(n)}=\mathcal{S}_B^{(n)}$. In the general case, one could then consider another quantity called the \textsf{mutual information}, defined as
	\begin{equation}
	  \mathcal{I}_{A,B}^{(n)}\defeq \mathcal{S}_A^{(n)}+\mathcal{S}_B^{(n)}-\mathcal{S}_{\mathcal{H}}^{(n)},
	  \label{}
	\end{equation}
	where~$\mathcal{S}_{\mathcal{H}}^{(n)}$ is defined the same way as (\ref{eqn-def-renyi-ent}) by just removing the subscript~$A$. 
      \end{def7}

\begin{def7}
	  If~$\rho$ is trace class, then~$\rho_A$ is trace class and the complex power~$\rho_A^z$ is well-defined for~$\fk{Re}(z)>1$. If moreover~$\mathcal{S}_A<\infty$, one then has the relation
	  \begin{equation}
	    \mathcal{S}_A=\lim_{z\to 1^+}\mathcal{S}_A^{(z)}=-\lim_{z\to 1^+}\frac{\partial}{\partial z}\ttr_{\mathcal{H}_A}(\rho_A^z).
	    \label{}
	  \end{equation}
	  However, for continuum QFTs the tensor product decomposition~$\mathcal{H}=\mathcal{H}_A\otimes \mathcal{H}_B$ is generally not well-defined, and the only approach available which produces reasonable answers seems to be employing the so-called \textsf{replica} interpretation to be explained in the next subsection, or the equivalent approach by \textsf{twist fields} (both avoid defining the tensor product). These concerns the quantities~$\ttr_{\mathcal{H}_A}(\rho_A^n)$ for positive integral~$n$ and is also what we will focus on for this paper. 
	\end{def7}

    \begin{exxx}\label{exp-mat-part-trace}
  Let~$\mathcal{H}=\mb{C}^2\otimes \mb{C}^2 =:\mathcal{H}_A\otimes \mathcal{H}_B$, in this order. Consider the standard basis~$e_1\otimes e_1$,~$e_1\otimes e_2$,~$e_2\otimes e_1$ and~$e_2\otimes e_2$ (in this order) of~$\mathcal{H}$ where~$e_1$,~$e_2$ are standard base vectors of~$\mb{C}^2$. Then for~$\rho\in \End(\mb{C}^2\otimes \mb{C}^2)$ written
  \begin{equation}
    \rho=\left(
    \begin{array}{cccc}
      a_{11} &a_{12} &a_{13} &a_{14} \\
      a_{21} &a_{22} &a_{23} &a_{24} \\ 
      a_{31} &a_{32} &a_{33} &a_{34} \\
      a_{41} &a_{42} &a_{43} &a_{44} 
    \end{array}
    \right),
    \label{}
  \end{equation}
  we have
  \begin{equation}
    \ttr_B( \rho)=\left(
    \begin{array}{cc}
      a_{11}+a_{22} & a_{13}+a_{24} \\
      a_{31}+a_{42} & a_{33}+a_{44}
    \end{array}
    \right),\quad\quad 
    \ttr_A( \rho)=\left(
    \begin{array}{cc}
      a_{11}+a_{33} & a_{12}+a_{34} \\
      a_{21}+a_{43} & a_{22}+a_{44}
    \end{array}
    \right).
    \label{}
  \end{equation}
  Also, if~$\rho=\rho_A\otimes \rho_B$ is a Kronecker product, we have~$\ttr_B(\rho)=\ttr(\rho_B)\rho_A$,~$\ttr_A(\rho)=\ttr(\rho_A)\rho_B$.
\end{exxx}

\subsection{Path Integrals and Replica}\label{sec-main-replica}
      Now let us describe the geometric framework where the quantities~$\mathcal{S}_A$ and~$\mathcal{S}_A^{(n)}$ get represented by \textit{path integrals}, as presented in \cite{Cardy_Calabrese}. Here we shall describe the pictures heuristically and \textbf{not attempt at any rigor}. For this paper we focus on~$1+1$ dimensional field theories where the space is either the real line or a circle with a specific perimeter (as a Riemannian manifold), denoted generally by~$X$, and space-time a 2-dimensional Riemannian surface with or without boundary denoted~$\Sigma$ (for example,~$\Sigma=X\times [0,T]$ or~$X\times \mb{R}$). Each field theory comes with specified \textsf{field configuration spaces} over space and space-time, denoted~$\conf(X)$ and~$\conf(\Sigma)$, as well as over their sub-regions. Typically~$\conf(X)=\mm{Map}(X,\mss{V})$, the space of \textit{maps} from~$X$ to a \textsf{spin value space}/\textsf{target space}~$\mss{V}$. Such configuration spaces should allow
    \begin{enumerate}[(a)]
      \item restriction of a configuration onto a subregion~$A\subset X$ (or~$\Sigma$), heuristically a map
	\begin{equation}
	  \left.
	  \begin{array}{rcl}
	     \conf(X) &\lto &\conf(A),\\
	     \phi &\longmapsto &\phi|_A.
	  \end{array}
	  \right.
	  \label{}
	\end{equation}
	In other words one is allowed to \textit{localize} the field;
      \item each configuration~$\phi\in \conf(X)$ to be recovered from the pair~$(\phi|_A,\phi|_{A^c})$ of its restrictions onto~$A$ and~$A^c=X\setminus A$, and moreover any pair of configurations on complementary subregions should combine into a global configuration; this is to say heuristically,
	\begin{equation}
	  \conf(X)\cong \conf(A)\times \conf(A^c).
	  \label{eqn-decomp-local-complemen}
	\end{equation}
    \end{enumerate}
    
    \begin{def7}
	  The above statements are rigorous in the case~$X=\Lambda$ is a discrete lattice, and~$\conf(\Lambda)=\mm{Map}(\Lambda,\mss{V})$. Here ``subregions'' correspond to subsets of lattice sites.
	\end{def7}

     \begin{def7}
    Suppose~$\Sigma=X\times [0,T]$. Then~$\conf(\Sigma)$ could be considered as the space of ``paths'' through which a configuration over~$X$ evolves across the time interval~$[0,T]$. Indeed, if we have taken~$\conf(\Sigma)$ to be~$\mm{Map}(\Sigma,\mss{V})$, then we would have simply~$\mm{Map}(X\times[0,T],\mss{V})\cong \mm{Map}([0,T],\mm{Map}(X,\mss{V}))$. The integral of (\ref{eqn-heu-gen-amp}) below is thus an integral over a space of ``paths''.
  \end{def7}

    We now describe heuristically the Hilbert space and the time evolution. For the Hilbert space one takes
	\begin{equation}
	  \mathcal{H}_X\defeq L^2(\conf(X),\mathcal{L}),
	  \label{}
	\end{equation}
	where~$\mathcal{L}$ denotes the non-existent \textit{Lebesgue measure} on the configuration space~$\conf(X)$. Under the path-integral formalism (Euclidean\footnote{See the first paragraph of subsection \ref{sec-dens-op-imag-time}.}) time evolution across~$\Sigma=X\times [0,T]$, namely over time~$T$, is represented by the integral operator
	\begin{equation}
	  \left.
	  \def\arraystretch{1.3}
	  \begin{array}{rcl}
	     U_T: \mathcal{H}_X &\lto &\mathcal{H}_X,\\
	     F&\longmapsto & (U_T F)(\psi)\defeq \ddp\int_{\conf(X)}^{}\mathcal{A}_T(\psi,\varphi)F(\varphi)\,\dd \mathcal{L}(\varphi),
	  \end{array}
	  \right.
	  \label{}
	\end{equation}
	with the integral kernel
	\begin{equation}
	  \mathcal{A}_T(\psi,\varphi)\defeq \int_{
	     \left\{  \phi\in \conf(\Sigma)~\middle|~ \substack{
	     \phi|_{X\times\{0\}}=\varphi,\\ 
	     \phi|_{X\times \{T\}}=\psi
	   }
      \right\}} \me^{-S_{\mm{EQFT}}(\phi)}\,\dd \mathcal{L}(\phi),
	  \label{eqn-heu-gen-amp}
	\end{equation}
	where now we integrate against the still non-existent Lebesgue measure on~$\conf(\Sigma)$ (with the indicated boundary conditions), and where~$S_{\mm{EQFT}}$ is the \textsf{action functional} (``E'' stands for ``Euclidean'') of the specific theory at hand.

    The above recipe for time evolution does not involve the fact that~$\Sigma$ is~$X\times [0,T]$, and indeed it makes sense for space-times~$\Sigma$ having any kind of geometry as long as~$\partial \Sigma=\partial_{\mm{ini}}\Sigma\sqcup \partial_{\mm{ter}}\Sigma$ and~$\partial_{\mm{ini}}\Sigma\cong \partial_{\mm{ter}}\Sigma\cong X$, namely its boundary has two components (called \textsf{initial} and \textsf{terminal}) both isometric to~$X$. Generalizing still further, Segal \cite{Segal} proposed a set of axioms that defines a QFT abstractly as a rule that associates Hilbert spaces to space manifolds and evolution operators to space-time manifolds that ``connects'' them (cobordisms). 
    
    Now denote by~$U_{\Sigma}$ and~$\mathcal{A}_{\Sigma}$ the evolution operator and its integral kernel corresponding to the space-time piece~$\Sigma$. Two important axioms of Segal (and Atiyah) are written as
	\begin{description}
	  \item[(composition)] if one has two space-time pieces~$\Sigma_1$,~$\Sigma_2$ and one glues them together by identifying~$\partial_{\mm{ter}}\Sigma_1$ with~$\partial_{\mm{ini}}\Sigma_2$ via~$X$, obtaining the piece~$\Sigma_2\circ \Sigma_1$, then we have
	    \begin{equation}
	      U_{\Sigma_2\circ \Sigma_1}=U_{\Sigma_2}\circ U_{\Sigma_1};
	      \label{}
	    \end{equation}
	  \item[(trace)] if one glues the space-time~$\Sigma$ with itself by identifying~$\partial_{\mm{ini}}\Sigma$ and~$\partial_{\mm{ter}}\Sigma$ via~$X$, obtaining~$\check{\Sigma}$, then
	    \begin{equation}
	      \ttr_{\mathcal{H}}(U_{\Sigma})=\mathcal{Z}(\check{\Sigma}),
	      \label{}
	    \end{equation}
	    where the number
	    \begin{equation}
	      \mathcal{Z}(\check{\Sigma})\defeq \int_{\conf(\check{\Sigma})}^{}\me^{-S_{\mm{EQFT}}(\phi)}\,\dd \mathcal{L}(\phi)
	      \label{eqn-heu-gen-part-func}
	    \end{equation}
	    is the \textsf{partition function} (same as what appears in definition \ref{def-cft}).
	\end{description}

\begin{def7}
	  Comparing (\ref{eqn-heu-gen-part-func}) with (\ref{eqn-heu-gen-amp}), one sees that the trace axiom corresponds formally to the fact that the ``trace'' of an integral operator is the integral along the diagonal of its kernel. For rigorous results concerning this statement see Simon \cite{Sim3} section 3.11.
	\end{def7}

    Now we try to incorporate considerations of subregions~$A\subset X$ into the framework described above. For simplicity we will take~$A$ to be a finite interval (remember~$X$ is either a circle or the real line). First of all, following (\ref{eqn-decomp-local-complemen}), we have formally
	\begin{equation}
	  \mathcal{H}_X \overset{\mm{heu}}{\cong} L^2\big(\conf(A)\times \conf(A^c),\mathcal{L}_{\conf(A)}\otimes \mathcal{L}_{\conf(A^c)}\big)\overset{\mm{heu}}{\cong} \mathcal{H}_{A}\otimes \mathcal{H}_{A^c}.
	  \label{}
	\end{equation}
	To include the partial trace, the trace axiom is now slightly extended by adding that taking the partial trace (over~$A^c$) corresponds to gluing ``partially'' (along~$A^c$ but not~$A$), leaving out a ``slit'' with two sides denoted~$A_-$ and~$A_+$. In terms of the integral kernels, if we denote by~$\mathcal{A}_{\Sigma}^A$ the kernel of~$\ttr_{A^c}(U_{\Sigma})$ acting on~$\mathcal{H}_{A}$, then we have
	\begin{equation}
	  \mathcal{A}_{\Sigma}^A(\psi_{A},\varphi_A)\heueq \int_{\conf(A^c)}^{} \mathcal{A}_{\Sigma}(\psi_A,\sigma_{A^c},\varphi_A,\sigma_{A^c})\,\dd \mathcal{L}(\sigma_{A^c}).
	  \label{}
	\end{equation}
    Here~$\ttr_{A^c}(U_{\Sigma})$ as an operator acting on~$\mathcal{H}_A$, the Hilbert space associated to an interval, is represented by a surface (space-time) with a ``slit'' (or \textsf{branch cut}/\textsf{defect line}) where the two ``sides'' of the slit are identified with two copies of the interval. Accordingly, we must also extend the composition axiom to incorporate this situation, namely to allow gluing of surfaces with slits along sides which represents composition of operators acting on Hilbert spaces over intervals. We shall assume that this has been done in the obvious manner.

    Finally we arrive at an interpretation of a quantity of the form~$\ttr_{\mathcal{H}_A}(\ttr_{A^c}(U_{\Sigma})^n)$ which appears in the expression for the R\'enyi entropy (\ref{eqn-def-renyi-ent}). Indeed, one starts with the surface~$\Sigma$ and glue its two ends ``partially'' along the parts corresponding to~$A^c$ in~$X$. We denote the resulting ``surface with slit'' by~$\check{\Sigma}\setminus A$. 
	Then the two sides of the slit gets identified with two copies of~$A$, which we denote by~$A_{\pm}$, that corresponds to approaching~$A$ from the two sides within~$\check{\Sigma}\setminus A$. Next, we take~$n$ copies of~$\check{\Sigma}\setminus A$ and glue them in a \textit{cyclic manner}, that is, we glue~$A_+^{(j)}$ on the~$j$-th copy to~$A_-^{(j+1)}$ on the~$(j+1)$-th copy,~$1\le j<n$, and finally~$A_+^{(n)}$ to~$A_-^{(1)}$. We denote the surface thus obtained by~$\check{\Sigma}_n$, called the $n$-th \textsf{replica}. 
    Importantly, $\check{\Sigma}_n$ comes equipped with a \textbf{metric} which is induced from the metric on the original space-time $\Sigma$. Equivalently, there is an obvious map~$f_n:\check{\Sigma}_n\lto \check{\Sigma}$, sending the end points of~$A$ to themselves, and any other point to its counterpart on the original copy. 
    Then~$f_n$ is a branched~$n$-sheeted cyclic covering with critical points being the two end points of the copies of~$A$, and the metric on~$\check{\Sigma}_n$ is the pull-back under~$f_n$ of the original metric on~$\check{\Sigma}$ (on~$\Sigma$). This will be a metric with \textit{conical singularities} at the two end points of $A$ which are not duplicated. If we assume the \textit{extended} versions of the composition and trace axioms discussed above, then we find
    \begin{equation}
	  \ttr_{\mathcal{H}_A}\big(\ttr_{A^c}(U_{\Sigma})^n\big) =\mathcal{Z}\big(\check{\Sigma}_n\big).
	  \label{}
	\end{equation}
The geometric correspondence is summarized in the table below.
    
      \begin{center}
  \def\arraystretch{1.3}
  \begin{tabular}[]{|c|c|}
    \hline
   Operation & Picture \\ \hline
    evolution operator $U_{\Sigma}$ & 
	 \raisebox{-0.5\height}{\includegraphics[width=5cm]{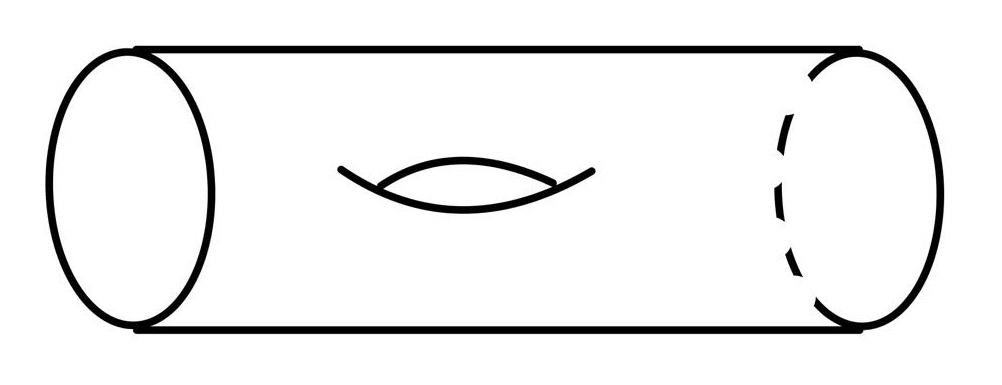}} 
    \\ \hline
	  $\ttr_{\mathcal{H}}(U_{\Sigma})$ & 
      \raisebox{-0.5\height}{\includegraphics[width=5cm]{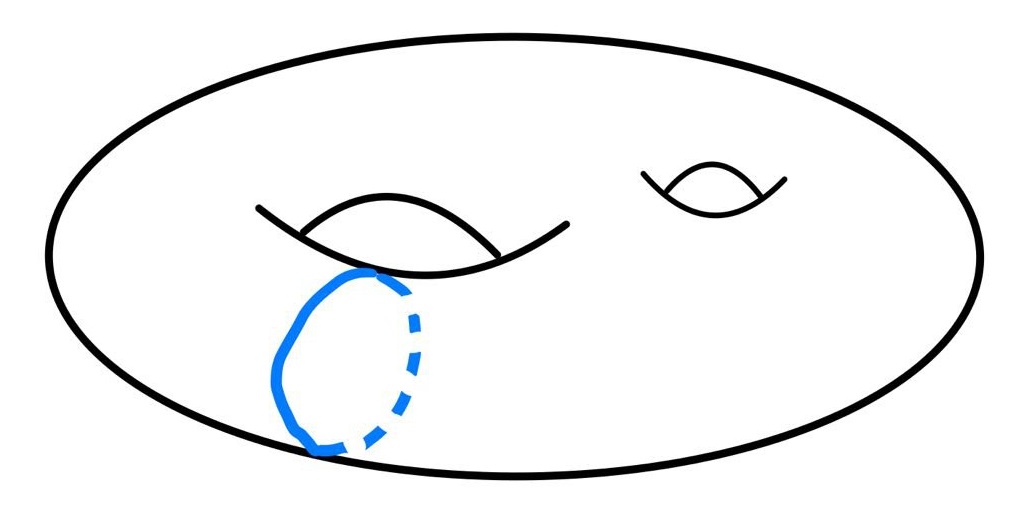}}
      \\ \hline
	   $\ttr_{A^c}(U_{\Sigma})$ & 
       \raisebox{-0.5\height}{\includegraphics[width=5cm]{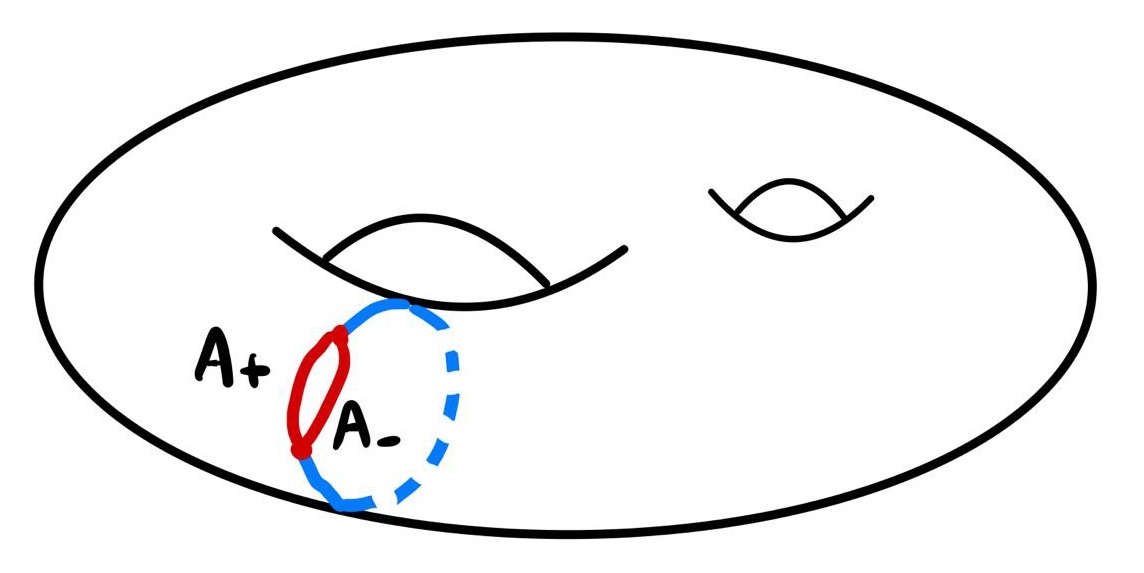}}
       \\ \hline
       $\ttr_A\big(\ttr_{A^c}(U_{\Sigma})^n\big)$ & 
       \raisebox{-0.5\height}{\includegraphics[width=5cm]{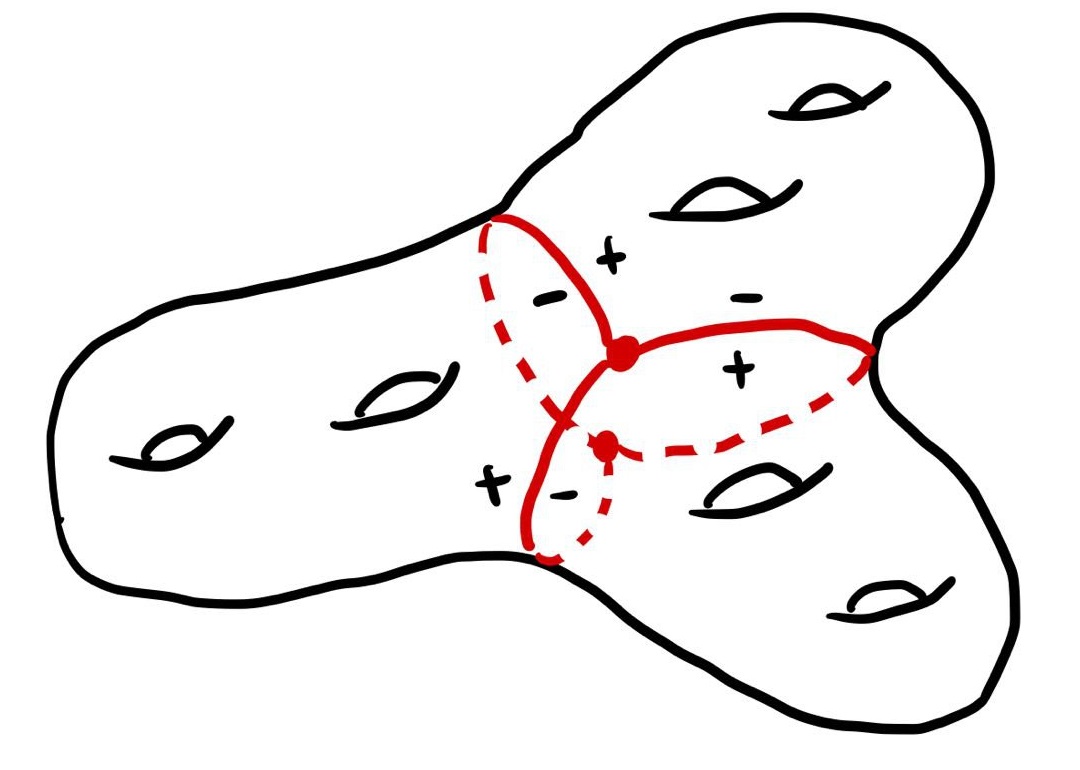}}
       \\ \hline
	\end{tabular}\end{center}

\subsection{Density Operators Represented by Imaginary-time Evolution}\label{sec-dens-op-imag-time}

We are left with the question: which density operators are also evolution operators across some space-time? To begin with, we reemphasize: our space-time metrics all have Euclidean signature and we work with \textsf{Wick rotated}/\textsf{Euclidean} QFT. More precisely, the action functional~$S_{\mm{EQFT}}(\phi)$, which usually involves itself the underlying space-time metric, is real-valued and the exponential weight~$\exp(-S_{\mm{EQFT}}(\phi))$ positive in (\ref{eqn-heu-gen-amp}) and (\ref{eqn-heu-gen-part-func}). The evolution operators defined in this way is said to be in \textsf{imaginary time}. For actual Lorentzian QFT one must use~$\exp(\ii S_{\mm{QFT}})$ instead of~$\exp(-S_{\mm{EQFT}})$, where the action~$S_{\mm{QFT}}$ is also written in terms of the Lorentzian metric.

Next, which imaginary-time evolution operators are trace class and nonnegative, serving as candidates for a density operator? In this regard it would be instructive to introduce the third axiom of Segal's:
	\begin{description}
	  \item[(adjoint)] let~$\Sigma^*$ denote the same space-time as~$\Sigma$ but with the identification of the initial and terminal boundaries \textit{reversed}, then
	    \begin{equation}
	      U_{\Sigma^*}=U_{\Sigma}^{\dagger},
	      \label{}
	    \end{equation}
	    where~$\dagger$ denotes the operator adjoint\footnote{More precisely, this is a requirement for \textit{unitary} QFTs, and we restrict to this case throughout our discussion.}.
	\end{description}
	From this, together with the two previous axioms, one could see that Euclidean evolution across a cobordism naturally corresponds to \textit{Hilbert-Schmidt} operators. Indeed,
	\begin{equation}
	  \ttr_{\mathcal{H}}\big(U_{\Sigma}^{\dagger}U_{\Sigma}\big)=\ttr_{\mathcal{H}}\big(U_{\Sigma^*\circ \Sigma}\big)=\mathcal{Z}\big((\Sigma^*\circ \Sigma)^{\vee}\big)\quad (<\infty),
	  \label{}
	\end{equation}
	where~$(\Sigma^*\circ \Sigma)^{\vee}$ is the ``double'' of~$\Sigma$. This ``shows'' that~$U_{\Sigma}$ is Hilbert-Schmidt. Immediately from the same expression one also sees that a cobordism of the form~$\Sigma^*\circ \Sigma$ naturally gives a nonnegative trace class operator which is what we wanted. Such a cobordism has, in other words, a \textit{reflection symmetry} exchanging its initial and terminal boundaries, and whose invariant set is exactly isomorphic to a copy of~$X$ as a Riemannian submanifold.

\begin{def7}
	  A somewhat technical requirement is that the boundaries of~$\Sigma$ be geodesic so that its double would be smooth.
	\end{def7}

	\begin{def7}
	  For certain models the (ordinary) Segal axioms have been fully and rigorously constructed. See, for example, the work \cite{Lin} of the second author and also \cite{GKRV}. In these cases~$U_{\Sigma}$ is rigorously shown to be Hilbert-Schmidt. The main difficulty lies in defining (\ref{eqn-heu-gen-amp}) and (\ref{eqn-heu-gen-part-func}) rigorously, in which case it implies~$\mathcal{Z}(\check{\Sigma})<\infty$, and also in showing the generic composition axiom out of these definitions.
	\end{def7}

    Now we give some examples.

    \begin{exxx}
	  Assume that one is given a \textsf{Hamiltonian operator}~$H$ acting on~$\mathcal{H}_X$, and let~$\beta>0$. Then the so-called \textsf{thermal state} (or \textsf{Gibbs state}) at \textsf{inverse temperature}~$\beta$ is given by the density operator
	  \begin{equation}
	    \rho_{\beta}\defeq \frac{\me^{-\beta H}}{\ttr_{\mathcal{H}}(\me^{-\beta H})}.
	    \label{}
	  \end{equation}
	  The operator~$\me^{-\beta H}$ gives the evolution over imaginary time~$\ii\beta$ and could be represented by a path integral (\ref{eqn-heu-gen-amp}) over the cylinder~$\Sigma_{\beta}:=X\times [0,\beta]$ with a specific action related to~$H$. Now we denote by~$\check{\Sigma}_{n,\beta}$ the closed surface obtained by gluing~$n$ copies of~$\check{\Sigma}_{\beta}\setminus A$ cyclicly along the slit~$A$, as explained at the end of the last subsection. Then the replica trick yields
	  \begin{equation}
	    \ttr_A\big(\ttr_{A^c}(\rho_{\beta})^n\big)=\frac{\mathcal{Z}(\check{\Sigma}_{n,\beta})}{\mathcal{Z}(\check{\Sigma}_{\beta})^n}.
	    \label{}
	  \end{equation}
	  We remark again that~$\check{\Sigma}_{n,\beta}$ comes equipped with the ``replica metric'' which is induced from~$\Sigma_{\beta}$.
	\end{exxx}

    \begin{exxx}\label{exp-ground-state-disj-cob}
	  Here we consider cobordisms which are disconnected,~$\Sigma=\Sigma_1\sqcup \Sigma_2$, with~$\partial\Sigma_1\cong \partial\Sigma_2\cong X$. We mark~$\partial\Sigma_1$ as initial and~$\partial\Sigma_2$ as terminal. In this case
	  \begin{equation}
	    \conf(\Sigma)=\conf(\Sigma_1)\times \conf(\Sigma_2),\quad\quad \mathcal{L}_{\conf(\Sigma)}=\mathcal{L}_{\conf(\Sigma_1)}\otimes \mathcal{L}_{\conf(\Sigma_2)},
	    \label{}
	  \end{equation}
	  and assume (a very weak form of) \textsf{locality} of the action (no interaction between disjoint space-times), that is,
	  \begin{equation}
	    S_{\Sigma}(\phi)=S_{\Sigma_1}(\phi_1)+S_{\Sigma_2}(\phi_2),
	    \label{}
	  \end{equation}
	  with~$\phi=(\phi_1,\phi_2)$, then we have formally
	  \begin{align}
	    \mathcal{A}_{\Sigma}(\psi,\varphi)&=\iint_{ \left\{ \substack{\phi_1\in\conf(\Sigma_1), \\ \phi_1|_{\partial\Sigma_1}=\varphi} \right\}\times \left\{ \substack{\phi_2\in\conf(\Sigma_2), \\ \phi_2|_{\partial\Sigma_2}=\psi} \right\}}
	    \me^{-S_{\Sigma_1}(\phi_1)}\me^{-S_{\Sigma_2}(\phi_2)}
	    \,\dd\mathcal{L}_{\conf(\Sigma_1)}(\phi_1)\,\dd \mathcal{L}_{\conf(\Sigma_2)}(\phi_2)\nonumber \\
	    &=\mathcal{A}_{\Sigma_1}(\varphi)\mathcal{A}_{\Sigma_2}(\psi), \label{eqn-disj-union-axiom}
	  \end{align}
	  where~$\mathcal{A}_{\Sigma_1}$ and~$\mathcal{A}_{\Sigma_2}$ are the amplitudes associated to~$\Sigma_1$ and~$\Sigma_2$ individually, defined by the same formula (\ref{eqn-heu-gen-amp}) in the case of having just one boundary component (sometimes (\ref{eqn-disj-union-axiom}) itself is taken as an axiom, see \cite{Gaw} page 766). We see that in this case~$U_{\Sigma}$ is a rank-one operator. Written more suggestively,
	\begin{equation}
	  U_{\Sigma}=\big|\mathcal{A}_{\Sigma_2}\big\rangle \big\langle \mathcal{A}_{\Sigma_1}\big|.
	  \label{}
	\end{equation}
    	We normalize by putting~$\rho_{\Sigma}:=U_{\Sigma}/\ttr(U_{\Sigma})$, then following the replica trick we have
	\begin{equation}
	  \ttr_A\big(\ttr_{A^c}(\rho_{\Sigma})^n\big)=\frac{\mathcal{Z}(\check{\Sigma}_n)}{\mathcal{Z}(\check{\Sigma})^n},
	  \label{}
	\end{equation}
	where~$\check{\Sigma}$ is the closed surface obtained by gluing~$\Sigma_1$ with~$\Sigma_2$ along~$X$, and~$\check{\Sigma}_n$ by gluing~$n$ copies of~$\check{\Sigma}\setminus A$ along the slit~$A$, equipped with the induced metric. In particular, if~$\Sigma_1=\Sigma_2=\mb{D}$, the unit disk, equipped with a flat-at-the-boundary metric (namely it could be written~$|z|^{-2}|\dd z|^2$ in some complex coordinate on an annulus around~$\partial\mb{D}$), then~$\rho_{\Sigma}$ defined as above is metric independent (provided flat at the boundary) and represents (projection onto) the \textsf{vacuum state}, usually written~$|0\rangle\langle 0|$ (see \cite{Gaw} page 768).
	\end{exxx}

\newpage
\section{Geometric Lemmas}\label{sec-geo-lemm}
	To begin with, we record some basic conformal relations that will be used repeatedly throughout the paper. These relations are local in nature as they concern only what happens in an arbitrarily small neighborhood around each point. Namely, let~$\Sigma$ be a smooth \textit{surface} and denote by~$\dd V_g$,~$\nabla_g$,~$K_g$,~$\Delta_g$ and $\nabla_v^g$ respectively the volume (area) form, gradient, Gauss curvature (half the scalar curvature), Laplacian and the Levi-Civita covariant derivative in the direction~$v\in T_z\Sigma$ under the Riemannian metric~$g$, and let~$h\in C^{\infty}(\Sigma)$, then we have
\begin{align}
      \dd V_{\me^{2h}g}&=\me^{2h}\dd V_g, \label{eqn-scale-vol}\\
      \nabla_{\me^{2h}g}&=\me^{-2h}\nabla_g, \label{eqn-scale-nabla}\\
      K_{\me^{2h}g}&=\me^{-2h}\left( -\Delta_g h+K_g \right),
      \label{eqn-liouville-eqn}\\
      \Delta_{\me^{2h}g}&=\me^{-2h}\Delta_g. \label{eqn-scale-lap} \\
      \nabla_v^{\me^{2h}g} X&=\nabla_v^g X+ (vh)X+(Xh)v-\ank{v,X}_g \nabla_g h. \label{eqn-scale-cov-deri}
    \end{align}
The relation (\ref{eqn-liouville-eqn}) is also called \textsf{Liouville's equation}.

\subsection{Technical Remarks on Conical Singularities}\label{sec-conic-remark}

The principal local regularity result around a conical singularity in our sense (definition \ref{def-singularities}) is the following obtained by Troyanov \cite{Troy2}.

  \begin{lemm}
    [\cite{Troy2} proposition 3.2] \label{lemm-reg-metric-pot} Let~$\Sigma$ be a closed Riemann surface and~$\tilde{g}$ a generalized conformal metric on~$\Sigma$ with a conical singularity at~$z_0\in \Sigma$ of exponent~$\gamma$. Then there exists a complex coordinate~$w$ defined on a neighborhood~$U\ni z_0$, with~$w(z_0)=0$, such that
    \begin{equation}
      \tilde{g}|_{U\setminus z_0}=|w|^{2\gamma}\me^{2 \varphi_{U,w}}|\dd w|^2,
      \label{}
    \end{equation}
    such that~$\varphi_{U,w}$ satisfies the conditions of definition \ref{def-singularities} as well as
    \begin{equation}
      \varphi_{U,w}(z)=\varphi_{U,w}(z_0)+\mathcal{O}(|w|^{2\gamma+2}),\quad\textrm{and}\quad \partial_z \varphi_{U,w},~\partial_{\ol{z}} \varphi_{U,w}=\mathcal{O}(|w|^{2\gamma+1}).
      \label{eqn-scale-metric-pot}
    \end{equation}
    In particular, $\partial_r \varphi_{U,w}=\mathcal{O}(|w|^{2\gamma+1})$ where $r$ is the radial coordinate~$r(z)=|w(z)|$ under $|\dd w|^2$. \hfill~$\Box$
  \end{lemm}

   Troyanov showed further that a neighborhood (still denoted~$U$) of the conical singularity could in fact be equivalently and more intrinsically described by a set of \textsf{polar coordinates}. That is, there is a map
   \begin{equation}
	  h:[0,\varrho]\times \mb{S}_{\Theta}^1 \lto U,
	  \label{}
	\end{equation}
	where~$\mb{S}_{\Theta}^1$ denotes the Riemannian circle with perimeter~$\Theta=2\pi(\gamma+1)$, such that
	\begin{enumerate}[(a)]
	  \item $h(r,\theta)=z_0$ iff~$r=0$;
	\item $h|_{(0,\varrho]\times \mb{S}^1_{\Theta}}$ is a locally bi-Lipschitz homeomorphism onto~$U\setminus z_0$; 
	\item we have
	  \begin{equation}
	    h^*(\tilde{g}|_{U\setminus z_0})=\dd r^2+\omega(r,\theta)^2\dd \theta^2,
	    \label{}
	  \end{equation}
	where~$\omega:(0,\varrho]\times \mb{S}_{\Theta}^1\lto \mb{R}$ is a function such that~$0<c_1\le \omega(r,\theta)\le c_2$ for some constants~$c_1$,~$c_2$ for all~$(r,\theta)$ and~$\lim_{r\to 0}\omega(r,\theta)/r =1$ for all~$\theta\in \mb{S}_{\Theta}^1$.
	\end{enumerate}
By the last requirement we have
\begin{equation}
  \Theta(r)\defeq \int_{\mb{S}^1}^{}\omega(r,\theta)\,\dd\theta\asymp r\Theta=2\pi r(\gamma+1),\quad\quad r\to 0,
  \label{}
\end{equation}
and thus~$2\pi(\gamma+1)$ is sometimes called the \textsf{cone angle}. Moreover, the regularity of the function~$\omega$ as well as~$h$ itself could be deduced from the regularity of the curvature of~$\tilde{g}$ on~$U\setminus z_0$ (\cite{Troy2} theorem 4.1). There are certain recent works in the literature which begin naturally with this latter point of view (\cite{AKR} for example). But we stick to definition \ref{def-singularities} throughout this article although utilization of polar coordinates may (or may not) aide with certain proofs.

\begin{def7}\label{rem-reg-met-pot-cone}
  In the literature people have considered more restrictive classes of allowable conical metrics by assuming more regularity on the regular metric potential (see \cite{Kalvin} section 2.1). The principle is to respect the correct scaling property under dilations centered at the cone points. Especially for the conical metrics with \textit{constant curvature}, the corresponding regular metric potentials (against coordinate metrics) are shown to be \textsf{dilation analytic}. This translates, in our notations, roughly into saying that~sufficiently near~$z_0$, $\varphi_{U,w}$ would be a real analytic function of~$|w|^{\gamma+1}$, extendable over a neighborhood of zero, in each direction respectively (this is clearly an enhancement of (\ref{eqn-scale-metric-pot})). However, due to the simplicity of our method we only need a very rough scaling property that is already deducible from lemma \ref{lemm-reg-metric-pot}, which we treat in the section below.
\end{def7}

\subsection{Dilation Properties at the Cone Point}\label{sec-geo-scaling}

     \begin{lemm}\label{lemm-scale-dist-sing}
Let~$\Sigma$ be a closed Riemann surface and~$\tilde{g}$ a generalized conformal metric on~$\Sigma$ with a conical singularity at~$z_0\in \Sigma$ of exponent~$\gamma$. Then 
	  \begin{equation}
	    \tilde{r}(z)=\frac{\me^{\varphi(z_0)}}{\gamma+1}r(z)^{\gamma+1}+\mathcal{O}\left(r^{3\gamma+3}\right),\quad\quad  z\to z_0,
	    \label{}
	  \end{equation}
	  where ~$\tilde{r}(z):=d_{\tilde{g}}(z,z_0)$,~$r(z):=|w(z)|$,~$z\ne z_0$, $\varphi:=\varphi_{U,w}$ as in lemma \ref{lemm-reg-metric-pot} and $w$ is the complex coordinate from lemma \ref{lemm-reg-metric-pot}. In particular, we see that~$\tilde{r}(z)\asymp r(z)^{\gamma+1}$ as~$z\to z_0$.
	\end{lemm}

	\begin{proof}
	  By (\ref{eqn-scale-metric-pot}) we have
	  \begin{equation}
	    0<\me^{\varphi(z_0)}-\alpha r(z)^{2\gamma+2}\le \me^{\varphi(z)}\le \me^{\varphi(z_0)}+\alpha r(z)^{2\gamma+2}
	    \label{eqn-conf-factor-sandwich}
	  \end{equation}
	  for some~$\alpha>0$, in some smaller neighborhood~$U_1\subset U$. Now fix~$z\in U_1$, let~$c(s)=w^{-1}(s w(z)/|w(z)|)$ be defined for~$s\in [0,r(z)]$ (the unit speed geodesic under~$|\dd w|^2$). Then~$\tilde{r}(z)$ is majorized by the length of~$c$ under~$\tilde{g}$, namely
	  \begin{equation}
	    \tilde{r}(z)\le \int_{0}^{r}|c'(s)|_{\tilde{g}}\,\dd s=\int_{0}^{r} s^{\gamma}\me^{\varphi(c(s))}\,\dd s\le \frac{\me^{\varphi(z_0)}r^{\gamma+1}}{\gamma+1}+\alpha \frac{ r^{3\gamma+3}}{3\gamma+3},
	    \label{eqn-con-dist-finite}
	  \end{equation}
	  by (\ref{eqn-conf-factor-sandwich}). On the other hand, the length under~$\tilde{g}$ of any curve inside~$U_1$ is bounded below by its length under the metric
	  \begin{equation}
	    g_1\defeq r^{2\gamma}\big( \me^{\varphi(z_0)}-\alpha r^{2\gamma+2} \big)^2\cdot |\dd w|^2 =f(r)^2\cdot |\dd w|^2.
	    \label{}
	  \end{equation}
	  defined on~$U_1$. By (\ref{eqn-scale-cov-deri}) we know that~$c$ would now reparametrize into a geodesic under~$g_1$ (minimizing since $g_1$ is radial with respect to the geodesic coordinates of $g$). Therefore we obtain
	  \begin{equation}
	    \tilde{r}(z)\ge \int_{0}^{r} |c'(s)|_{g_1}\,\dd s=\frac{\me^{\varphi(z_0)}r^{\gamma+1}}{\gamma+1}-\frac{\alpha r^{3\gamma+3}}{3\gamma+3}.
	    \label{}
	  \end{equation}
	  This gives the result.
	\end{proof}

    \begin{def7}
  In particular, (\ref{eqn-con-dist-finite}) also shows~$\tilde{r}(z)<\infty$, which is not necessarily true \textit{a priori}.
\end{def7}

 From this lemma it follows immediately the next two corollaries.

 	\begin{corr}
	  With the same set-up and notation as lemma \ref{lemm-scale-dist-sing}, we have
	  \begin{equation}
	    r(z)=(\gamma+1)^{\frac{1}{\gamma+1}}\me^{-\frac{\varphi(z_0)}{\gamma+1}}\tilde{r}(z)^{\frac{1}{\gamma+1}} +\mathcal{O}(\tilde{r}^{2+\frac{1}{\gamma+1}}),
	    \label{}
	  \end{equation}
	  as~$z\to z_0$.
	\end{corr}

    \begin{proof}
        This follows from a simple order analysis using Newton's binomial formula.
    \end{proof}

	\begin{corr}\label{cor-ratio-radius}
	  Pick~$h\in C^{\infty}(\Sigma)$ and denote by~$\tilde{r}_h(z)$ the distance from~$z$ to~$z_0$ under the scaled metric
	  \begin{equation}
	    \tilde{g}_h|_U\defeq d_g(\bullet,z_0)^{2\gamma}\me^{2(\varphi+h)}\cdot g|_U.
	    \label{}
	  \end{equation}
	  Denote by~$\delta(\varepsilon)$ and~$\delta_h(\varepsilon)$ metric radii under~$g$ of the~$\varepsilon$-metric disks under~$\tilde{g}$ and~$\tilde{g}_h$ respectively centered at~$z_0$. Then we have
	  \begin{equation}
	    \frac{\delta(\varepsilon)}{\delta_h(\varepsilon)}\lto \exp\left(\frac{h(z_0)}{\gamma+1}\right),
	    \label{}
	  \end{equation}
	  as~$\varepsilon\to 0^+$. \hfill~$\Box$
	\end{corr}

\subsection{Log-divergent Integrals}\label{sec-int-by-part}
	In this section we collect some basic computations of ``toy integrals'' involving functions with log-divergent singularities, which will nevertheless play a fundamental role in the proofs of the main results of this paper.

	Throughout this section, we assume~$(\Sigma,g)$ is a Riemannian surface with smooth metric~$g$, let~$z_0\in \Sigma$ and let~$\sigma\in C^{\infty}(\Sigma\setminus\{z_0\})$ such that in a neighbourhood $U$ of $z_0$,
	\begin{equation}
	  \sigma(z)=\gamma \log d_g(z,z_0) +\varphi(z)
	  \label{eqn-local-form-conf-factor}
	\end{equation}
	for some~$\gamma>-1$, where $\varphi$ is a function satisfying the requirements for the metric potential in definition \ref{def-singularities} as well as (\ref{eqn-scale-metric-pot}) with~$|\dd w|^2$ replaced by~$g$. Formula (\ref{eqn-limit-lap-smooth-against-con}), however, needs much weaker assumptions. We start by recalling the following basic fact.
	\begin{lemm}\label{lemm-bounded-lap-log}
	  In the set-up as above, the function~$\Delta_g\log d_g(\bullet,z_0)$ is smooth and uniformly bounded near but not coincident with $z_0$.
	\end{lemm}

    \begin{proof}
        See appendix \ref{sec-app-poin-lelong}.
    \end{proof}

 \begin{def7}
	  A related classical result says that the perimeter~$\ell_g(\partial B_r(z_0))$ of a geodesic circle centered at~$z_0\in \Sigma$ of radius~$r$ has an asymptotics
	  \begin{equation}
	    \ell_g(\partial B_r(z_0))=2\pi r-\frac{\pi}{3}r^3 K_g(z_0)+o(r^3),\quad r\to 0^+.
	    \label{eqn-asymp-peri-circle}
	  \end{equation}
	  Here as above,~$K_g(z_0)$ is the Gauss curvature of~$g$ at~$z_0$. See, for example, \cite{dC2} page 296.
	\end{def7}

 \begin{lemm}\label{lemm-basic-renorm-log}
	  In the set-up as above, let~$\delta>0$. Then we have
	  \begin{equation}
	    \int_{\Sigma\setminus B_{\delta}(z_0,g)}^{} |\nabla_g \sigma|_g^2\,\dd V_g=-2\pi\gamma^2\log(\delta)-2\pi\gamma \varphi(z_0)
        -\int_{\Sigma\setminus B_{\delta}}^{} \sigma\Delta_g \sigma\,\dd V_g
        +\mathcal{O}(\delta^{2\min\{\gamma,0\}+2}\log \delta)
	    \label{}
	  \end{equation}
	  as~$\delta\to 0^+$, where~$B_{\delta}(z_0,g)$ denotes the geodesic disk around~$z_0$ of radius~$\delta$ under the metric~$g$.
	\end{lemm}

	\begin{proof}
	  Denote~$r(z):=d_g(z,z_0)$. Integrating by parts (Green-Stokes formula) we have
	  \begin{align*}
	    \textrm{LHS}&=\int_{\partial B_{\delta}}^{} \sigma(-\partial_r \sigma)\,\dd\ell_{g} -\int_{\Sigma\setminus B_{\delta}}^{} \sigma\Delta_g \sigma\,\dd V_g \\
	    &=-\gamma^2\underbrace{\int_{\partial B_{\delta}}^{} \log r(\partial_r \log r)\,\dd\ell_{g}}_{A}
	    -\gamma\underbrace{\int_{\partial B_{\delta}}^{} \varphi(\partial_r \log r)\,\dd\ell_{g}}_{B}
	    -\gamma\underbrace{\int_{\partial B_{\delta}}^{} \log r(\partial_r \varphi)\,\dd\ell_{g}}_{C} \\
	    &\quad\quad -\underbrace{\int_{\partial B_{\delta}}^{} \varphi(\partial_r \varphi)\,\dd\ell_{g}}_{D}-\int_{\Sigma\setminus B_{\delta}}^{} \sigma\Delta_g \sigma\,\dd V_g.
	  \end{align*}
	  By the assumptions (\ref{eqn-scale-metric-pot}) on $\varphi$ and by (\ref{eqn-asymp-peri-circle}), we have
	  \begin{align}
	    A&=\log\delta\cdot\frac{1}{\delta}(2\pi \delta+\mathcal{O}(\delta^3))=2\pi\log\delta+\mathcal{O}(\delta^2\log \delta),\\
	    B&=\frac{1}{\delta}\int_{\partial B_{\delta}}^{} \varphi\,\dd\ell_{g} = 2\pi \varphi(z_0)+\mathcal{O}(\delta^2),\\
	    C&=\log\delta\int_{\partial B_{\delta}}^{} \partial_r\varphi\,\dd\ell_{g}=\mathcal{O}(\delta^{2\gamma+2}\log \delta),\\
	    D&=\mathcal{O}(\delta^{2\gamma+2}).
	    \label{}
	  \end{align}
	  Adding them all up, we obtain the result. Note $2\gamma+2$ may be small but is positive, so $\mathcal{O}(\delta^{2\gamma+2}\log \delta)=o(1)$.
	\end{proof}

  \begin{corr}\label{cor-quad-blow-tilde}
  In the set-up as above, let~$\tilde{g}$ be a generalised conformal metric representing $D=\gamma z_0$ with $\gamma>-1$, written in the form (\ref{eqn-conic-sing-local-form}). Suppose~$g$ is a smooth conformal metric on~$\Sigma$ and~$\tilde{g}=\me^{2\sigma}g$ for some~$\sigma\in C^{\infty}(\Sigma\setminus \{z_0\})$, so that $\sigma$ is of the form (\ref{eqn-local-form-conf-factor}). Then
  \begin{align}
    \int_{\Sigma\setminus B_{\varepsilon}(z_0,\tilde{g})}^{} |\nabla_g \sigma|^2\,\dd V_g&= -\frac{2\pi \gamma^2}{\gamma+1}\big[\log(\varepsilon)+\log(\gamma+1)-\varphi(z_0)\big]
    -2\pi\gamma \varphi(z_0)
        -\int_{\Sigma\setminus z_0}^{} \sigma\Delta_g \sigma\,\dd V_g
    +o(1)
    \label{}
  \end{align}
  as~$\varepsilon\to 0$, where~$B_{\varepsilon}(z_0,\tilde{g})$ denotes the geodesic ball around~$z_0$ of radius~$\varepsilon$ under the metric~$\tilde{g}$.
\end{corr}

	\begin{proof}
	  This is because~$B_{\varepsilon}(z_0,\tilde{g})$ has radius~$\delta(\varepsilon)\sim((\gamma+1)\varepsilon \me^{-\varphi(z_0)})^{1/(\gamma+1)}$ under the metric~$g$ by lemma \ref{lemm-scale-dist-sing}.
	\end{proof}

\begin{lemm}\label{lemm-linear-sing-green-sto}
	  In the set-up as lemma \ref{lemm-basic-renorm-log}, pick further~$h\in C^{\infty}(\Sigma)$. Then we have
	  \begin{align}
	    \lim_{\delta\to 0^+}\int_{\Sigma\setminus  B_{\delta}(z_0,g)}^{} \big(
    \ank{\nabla_g h,\nabla_g\sigma}_g +(\Delta_g h)\sigma\big)\,\dd V_g &=0, \label{eqn-limit-lap-smooth-against-con}\\
\lim_{\delta\to 0^+}\int_{\Sigma\setminus  B_{\delta}(z_0,g)}^{} \big(
    \ank{\nabla_g h,\nabla_g\sigma}_g +h(\Delta_g\sigma)\big)\,\dd V_g &=-2\pi\gamma h(z_0).
	    \label{}
	  \end{align}
   Here for (\ref{eqn-limit-lap-smooth-against-con}) we only need the weaker assumption that~$\varphi$ be bounded over~$U$ in the expression (\ref{eqn-local-form-conf-factor}) for~$\sigma$. Moreover, the 2-form~$(
    \ank{\nabla_g h,\nabla_g\sigma}_g +h(\Delta_g\sigma))\dd V_g$ is conformally invariant, that is,
    \begin{equation}
      \big(
      \sank{\nabla_{\me^{2\varphi}g} h,\nabla_{\me^{2\varphi}g}\sigma}_{\me^{2\varphi}g} +h(\Delta_{\me^{2\varphi}g}\sigma)\big)\dd V_{\me^{2\varphi}g} =\big(
    \ank{\nabla_g h,\nabla_g\sigma}_g +h(\Delta_g\sigma)\big)\dd V_g
      \label{eqn-conf-inv-2-form}
    \end{equation}
    for any~$\varphi\in C^{\infty}(\Sigma\setminus \{z_0\})$. Therefore we have in particular
    \begin{equation}
      \lim_{\varepsilon\to 0^+}\int_{\Sigma\setminus  B_{\delta}(z_0,g)}^{} \big(
      \sank{\nabla_{\tilde{g}} h,\nabla_{\tilde{g}}\sigma}_{\tilde{g}} +h(\Delta_{\tilde{g}}\sigma)\big)\,\dd V_{\tilde{g}} =-2\pi\gamma h(z_0),
      \label{}
    \end{equation}
    where~$\tilde{g}$ is as defined in corollary \ref{cor-quad-blow-tilde}.
	\end{lemm}

	\begin{proof}
	  We apply Green-Stokes. The first integral boils down to
	  \begin{align*}
	    \int_{\partial B_{\delta}}^{}(-\partial_r h)\sigma\,\dd\ell_g &=\gamma\int_{\partial B_{\delta}}^{}(-\partial_r h)\log r\,\dd\ell_g+\int_{\partial B_{\delta}}^{}(-\partial_r h)\varphi\,\dd\ell_g= \mathcal{O}(\delta\log \delta)+\mathcal{O}(\delta),
	  \end{align*}
	  and the second to
	  \begin{align*}
	    \int_{\partial B_{\delta}}^{}h(-\partial_r \sigma)\,\dd\ell_g &=-\gamma\int_{\partial B_{\delta}}^{}h\cdot\frac{1}{r}\,\dd\ell_g+\int_{\partial B_{\delta}}^{}h(-\partial_r \varphi)\,\dd\ell_g\\
	    &=-2\pi\gamma h(z_0)+\mathcal{O}(\delta^{2\min\{\gamma,0\}+2}),
	  \end{align*}
	  as~$\delta\to 0^+$, by (\ref{eqn-scale-metric-pot}) and (\ref{eqn-asymp-peri-circle}). Equality (\ref{eqn-conf-inv-2-form}) follows directly from (\ref{eqn-scale-vol}), (\ref{eqn-scale-nabla}) and (\ref{eqn-scale-lap}) and we obtain the rest of the lemma.
	\end{proof}

The following lemma is not used in the main proof, but is important regarding remark \ref{rem-extra-quad-term}.

\begin{lemm}\label{lemm-quadratic-sing}
    Under the assumptions at the beginning of this section, let~$0<\delta_1(\varepsilon)<\delta_2(\varepsilon)$ be two positive functions of~$\varepsilon$ such that~$\delta_{i}(\varepsilon)\to 0$ as~$\varepsilon\to 0^+$,~$i=1$,~$2$, and~$\delta_2/\delta_1\to Q>0$.  Then
    \begin{equation}
      \lim_{\varepsilon\to 0^+}\int_{B_{\delta_2(\varepsilon)}(z_0,g)\setminus B_{\delta_1(\varepsilon)}(z_0,g)}^{}|\nabla_g \sigma|_g^2\,\dd V_g=2\pi \gamma^2\log Q.
      \label{}
    \end{equation}
  \end{lemm}

\begin{proof}
    We abbreviate the disks~$B_{\delta_i(\varepsilon)}(z_0,g)$ as~$B(\delta_i)$,~$i=1$,~$2$. Again integrating by parts we have,
    \begin{equation}
      \int_{B(\delta_2)\setminus B(\delta_1)}^{}|\nabla_g \sigma|_g^2\,\dd V_g =-\int_{B(\delta_2)\setminus B(\delta_1)}^{} \sigma \Delta_g \sigma\,\dd V_g
      +\int_{\partial B(\delta_2)}^{}\sigma(\partial_r \sigma)\,\dd\ell_g
      -\int_{\partial B(\delta_1)}^{}\sigma(\partial_r \sigma)\,\dd\ell_g.
      \label{}
    \end{equation}
    Since~$\sigma$ is locally integrable and~$\Delta_g \sigma$ is bounded on~$\Sigma\setminus\{z_0\}$, we have
    \begin{equation}
      \int_{B(\delta_2)\setminus B(\delta_1)}^{} \sigma \Delta_g \sigma\,\dd V_g \quad\xlongrightarrow{\varepsilon\to 0^+}\quad 0.
      \label{}
    \end{equation}
    Next by (\ref{eqn-local-form-conf-factor}) and (\ref{eqn-scale-metric-pot}),
    \begin{align*}
      &\int_{\partial B(\delta_2)}^{}\varphi(\partial_r \varphi)\,\dd\ell_g,~
      \int_{\partial B(\delta_1)}^{}\varphi(\partial_r \varphi)\,\dd\ell_g= \mathcal{O}(\delta^{2\gamma+2})&&\xlongrightarrow{\varepsilon\to 0^+}\quad 0,\\
      &\int_{\partial B(\delta_2)}^{}\varphi(\partial_r \log r)\,\dd\ell_g,~
      \int_{\partial B(\delta_1)}^{}\varphi(\partial_r \log r)\,\dd\ell_g&&\xlongrightarrow{\varepsilon\to 0^+}\quad 2\pi \varphi(z_0).\\
      &\int_{\partial B(\delta_2)}^{}\log r(\partial_r \varphi)\,\dd\ell_g,~
      \int_{\partial B(\delta_1)}^{}\log r(\partial_r \varphi)\,\dd\ell_g
      = \mathcal{O}(\delta^{2\gamma+2}\log \delta)&&\xlongrightarrow{\varepsilon\to 0^+}\quad 0.
    \end{align*}
    Therefore the only thing left is
    \begin{equation}
     \int_{\partial B(\delta_2)}^{}\frac{\log r}{r}\,\dd\ell_g -\int_{\partial B(\delta_1)}^{} \frac{\log r}{r}\,\dd\ell_g   \sim 2\pi \log\Big(\frac{\delta_2}{\delta_1}\Big)
      \quad\xlongrightarrow{\varepsilon\to 0^+}\quad 2\pi \log Q.
      \label{}
    \end{equation}
    Adding up all the above, we obtain the result. 
  \end{proof}

\section{Renormalization Procedure}\label{sec-renorm-tech}

\subsection{Consistency}

As we could see from definitions \ref{def-main-renorm-anomaly} and \ref{def-renom-part-func} that a particular reference smooth metric $g$ was chosen to make the definition. Now we show that a different conformal reference metric would in fact give the same partition function for the target metric $\tilde{g}$ and hence $\mathcal{Z}(\Sigma,\tilde{g})$ is invariantly defined.

\begin{lemm}\label{lemm-consist}
	  Consider a conformal field theory with central charge~$c\in\mb{R}$ on the Riemann surface~$\Sigma$ and let~$\tilde{g}$ be a generalized conformal metric representing~$D=\sum_{j=1}^p \gamma_j z_j$ with~$\gamma_j>-1$. Also let~$g_1$ and~$g_0$ be two smooth conformal metrics such that~$g_1=\me^{2h}g_0$ for some~$h\in C^{\infty}(\Sigma)$, and~$\tilde{g}=\me^{2\sigma}g_1$ for~$\sigma\in C^{\infty}(\Sigma\setminus \supp D)$. Then if one of $\mathcal{R}A_{\Sigma}(\tilde{g},g_0)$ and $\mathcal{R}A_{\Sigma}(\tilde{g},g_1)$ defined by (\ref{eqn-def-renorm-polya}) is finite, so is the other, and we have
	  \begin{equation}
	    \mathcal{R}A_{\Sigma}(\tilde{g},g_0)-\mathcal{R}A_{\Sigma}(\tilde{g},g_1)=   A_{\Sigma}(g_1,g_0),
	    \label{}
	  \end{equation}
   	 where $A_{\Sigma}(g_1,g_0)$ is defined by (\ref{eqn-def-ord-poly-anomaly}). Hence definition (\ref{eqn-def-renom-part-func}) is independent of the reference metric chosen.
	\end{lemm}

 \begin{proof}
	  Without loss of generality suppose~$D=\gamma z_0$ with~$\gamma>-1$. Denote, for simplicity, by~$\nabla_i$,~$|\cdot|_i$,~$K_i$ and~$\dd V_i$ respectively the gradient, metric norm, Gauss curvature and area form under the metric~$g_i$,~$i=0$,~$1$. Then by definition and the relations (\ref{eqn-scale-vol}) --- (\ref{eqn-liouville-eqn}), valid on~$\Sigma\setminus  B_{\varepsilon}(z_0,\tilde{g})$ for any~$\varepsilon>0$, we have
	  \begin{align}
	    \mathcal{R}A_{\Sigma}(\tilde{g},g_1)&=\frac{1}{24\pi}\lim_{\varepsilon\to 0^+}\Big[ 
	      \int_{\Sigma\setminus  B_{\varepsilon}(z_0,\tilde{g})}^{}(|\nabla_{1} \sigma|_{1}^2+2K_{1}\sigma)\dd V_{1} +\frac{2\pi\gamma^2}{\gamma+1}\log(\varepsilon)
    \Big]\\
    &=\frac{1}{24\pi}\lim_{\varepsilon\to 0^+}\Big[ 
      \int_{\Sigma\setminus  B_{\varepsilon}(z_0,\tilde{g})}^{}( |\nabla_0 \sigma|_0^2 -2(\Delta_0 h)\sigma +2K_0 \cdot\sigma
      )\dd V_{0} +\frac{2\pi\gamma^2}{\gamma+1}\log(\varepsilon)
    \Big],\\
    \mathcal{R}A_{\Sigma}(\tilde{g},g_0)&=\frac{1}{24\pi}\lim_{\varepsilon\to 0^+}\Big[ 
      \int_{\Sigma\setminus  B_{\varepsilon}(z_0,\tilde{g})}^{}\big(|\nabla_{0}(h+ \sigma)|_{0}^2+2K_{0}(h+\sigma)\big)\dd V_{0} +\frac{2\pi\gamma^2}{\gamma+1}\log(\varepsilon)
    \Big].
  \end{align}
  Therefore
  \begin{align*}
    \mathcal{R}A_{\Sigma}(\tilde{g},g_0)-\mathcal{R}A_{\Sigma}(\tilde{g},g_1)&=\frac{1}{24\pi}\int_{\Sigma}^{}(|\nabla_{0}h|_{0}^2+2K_{0}\cdot h)\,\dd V_{0}\\
    &\quad\quad +\frac{1}{12\pi}\underbrace{\lim_{\varepsilon\to 0^+}\int_{\Sigma\setminus  B_{\varepsilon}(z_0,\tilde{g})}^{} \big(
    \ank{\nabla_0 h,\nabla_0\sigma}_0 +(\Delta_0 h)\sigma\big)\,\dd V_0}_{=~0}\\
    &=\frac{1}{24\pi}\int_{\Sigma}^{}(|\nabla_{0}h|_{0}^2+2K_{0}\cdot h)\,\dd V_{0}
  \end{align*}
  by (\ref{eqn-limit-lap-smooth-against-con}) since $\sigma$ has the form (\ref{eqn-local-form-conf-factor}) with $\varphi$ bounded with respect to $g_0$. 
	\end{proof}

\subsection{Regularized Curvature and Anomaly}
In this subsection we note another kind of regularized anomaly which concerns smooth (bounded) scaling of a singular metric, rather than singular scaling of a smooth metric which was in some sense what we did above. This quantity will also play a role in the main result proposition \ref{prop-main-conical-scaling}.

	\begin{lemm}\label{lemm-smooth-scale-con}
	  Let~$\Sigma$ be a closed Riemann surface and~$\tilde{g}$ a generalized conformal metric representing~$D=\sum_{j=1}^p \gamma_j z_j$ with~$\gamma_j>-1$. Now suppose~$h\in C^{\infty}(\Sigma)$ and consider the scaled metric~$\me^{2h}\tilde{g}$ on~$\Sigma\setminus\supp D$. Then
  \begin{equation}
    \mathcal{R}A_{\Sigma}(\me^{2h}\tilde{g},\tilde{g})\defeq \frac{1}{24\pi}
      \int_{\Sigma\setminus \supp D}\big(|\nabla_{\tilde{g}} h|_{\tilde{g}}^2+2K_{\tilde{g}}h\big)\,\dd V_{\tilde{g}} <\infty.
    \label{eqn-def-renorm-polya-conic}
  \end{equation}
	\end{lemm}

	\begin{proof}
	  Choose a smooth background conformal metric~$g$ and write~$\tilde{g}=\me^{2\sigma}g$,~$\sigma\in C^{\infty}(\Sigma\setminus\supp D)$. As in lemma \ref{lemm-linear-sing-green-sto}, the 2-form~$|\nabla_{\bullet}h|_{\bullet}^2\,\dd V_{\bullet}$ is conformally invariant, and the Gauss curvature transforms as
	  \begin{equation}
	    K_{\tilde{g}}=\me^{-2\sigma}(-\Delta_g \sigma+K_g),\quad \textrm{on }\Sigma\setminus\supp D.
	    \label{}
	  \end{equation}
	  Therefore
	  \begin{align}
	    \int_{\Sigma\setminus \supp D}(|\nabla_{\tilde{g}} h|_{\tilde{g}}^2+2K_{\tilde{g}}h)\,\dd V_{\tilde{g}} =
	    \int_{\Sigma\setminus \supp D}(|\nabla_{g}h|_{g}^2+2(-\Delta_g \sigma+K_g)h)\,\dd V_{g}<\infty, \label{eqn-conf-change-to-smooth}
	  \end{align}
	  because of lemma \ref{lemm-bounded-lap-log}.
	\end{proof}

\begin{def7}
	  As a by-product, we also see that the RHS of (\ref{eqn-conf-change-to-smooth}), which is expressed in terms of the background smooth metric~$g$, is independent of~$g$.
	\end{def7}

\begin{deef}\label{def-renorm-polya-conic}
	  In the situation as lemma \ref{lemm-smooth-scale-con}, we call (\ref{eqn-def-renorm-polya-conic}) the \textsf{regular Polyakov anomaly} of~$\me^{2h}\tilde{g}$ against~$\tilde{g}$.
	\end{deef}

\begin{def7}
  The distribution (or current, more precisely) on~$\Sigma$ which we denote by~$K_{\tilde{g}}\dd V_{\tilde{g}}$, defined by setting
  \begin{equation}
    \ank{K_{\tilde{g}}\dd V_{\tilde{g}}, \psi}\defeq \int_{\Sigma\setminus \supp D}^{}\psi K_{\tilde{g}}\,\dd V_{\tilde{g}}
    \label{}
  \end{equation}
  for all~$\psi\in C^{\infty}(\Sigma)$, is usually called the \textsf{regularized Gauss curvature}. This is in $L^1(\Sigma, g)\,\dd V_g$ for any smooth conformal metric $g$ for just the same reason as in (\ref{eqn-conf-change-to-smooth}).
\end{def7}

\subsection{Comparison with Cut-off Methods}

   There are existing methods in the literature that essentially define the same quantity as (\ref{eqn-def-renorm-polya}), but instead of removing the balls~$B_{\varepsilon}(z_i,\tilde{g})$ entirely, they introduce a regularized metric inside. Such ideas are found in the incomplete notes of Zamolodchikov et al. \cite{ZZ} (page 86), as well as in Eskin, Kontsevich, and Zorich \cite{EKZ}, section 3.6. The cut-off method used in the latter is more mathematically rigorous, being smooth, and it avoids generating additional singularities along each~$\partial B_{\varepsilon}$ for~$\nabla \sigma$. In this subsection, we note that our definition \ref{def-renom-part-func} yields the same result as in \cite{EKZ}, and arguably, it also aligns with other regularization methods discussed in \cite{EKZ}, section 3.6.\\

For each fixed~$\delta> \delta'> 0$ we consider a smooth function~$\sreg_{\delta,\delta'}(r)$ such that
	\begin{equation}
	  \sreg_{\delta,\delta'}(r)=\left\{
	  \begin{array}{ll}
	   \log  r &\textrm{for }r\ge \delta,\\
	    \log \delta &\textrm{for }0\le r\le \delta',
	  \end{array}
	  \right.
	  \label{}
	\end{equation}
as well as
 	\begin{equation}
 \log \delta' \le  \sreg_{\delta,\delta'}(r) \le \log \delta, \quad \quad\textrm{for }0\le r\le \delta.
\end{equation}
and 
 	\begin{equation}
 \left| \partial_r \sreg_{\delta,\delta'}(r)  \right| \le C \delta'^{-1}, \quad\quad\textrm{for }0\le r\le \delta,
\end{equation}
for some constant $C$ independent of $\delta, \delta'$. 
Such a function can be constructed using a smooth cut-off function~$f:\mb{R}\lto\mb{R}$ for which~$f(t)\equiv 0$ for~$t\le 0$,~$0<f(t)<1$ for~$0<t<1$,~$f(t)\equiv 1$ for~$t\ge 1$, and posing
 	\begin{equation}
\sreg_{\delta,\delta'}(r) \defeq  \log \delta + f \Big(\frac{r- \delta'}{\delta - \delta'} \Big) \log \frac{r}{\delta} .
\end{equation}
Now given a closed Riemann surface~$\Sigma$ and~$\tilde{g}$ a generalized conformal metric representing~$D=\sum_{j=1}^p \gamma_j z_j$ with~$\gamma_j>-1$, as well as a smooth background conformal metric~$g$ such that around each singularity~$z_j$ holds (\ref{eqn-conic-sing-local-form}), and the regular metric potentials~$\varphi_j$ satisfy the Troyanov conditions (\ref{eqn-scale-metric-pot}), we introduce the regularized smooth metric~$g_{\varepsilon,\varepsilon'}$ (for~$\varepsilon> \varepsilon'> 0$ sufficiently small) 
   	\begin{equation}
	  g_{\varepsilon,\varepsilon'}\defeq \left\{
	  \begin{array}{ll}
	    \tilde{g}&\textrm{on }\Sigma\setminus\bigcup_1^p U_j,\\
	   \exp \big( 2 \gamma_j \cdot \sreg_{\delta_j,\delta'_j}(r_j) \big) \me^{2\varphi_j}g& \textrm{on }U_j,
	  \end{array}
	  \right.
	  \label{}
	\end{equation}
    where $r_j = d_g(\bullet,z_j)$, the~$U_j$'s given in definition \ref{def-singularities}, and we choose~$\delta_j = \delta_j(\varepsilon)$ to be the radius of the ball~$B_{\varepsilon}(z_j,\tilde{g})$ under the metric~$g$, and likewise $\delta'_j = \delta_j(\varepsilon')$.
	\begin{lemm}
	  In the setting as above, write~$g_{\varepsilon,\varepsilon'}=\me^{2\sigma_{\varepsilon,\varepsilon'}}g$ now with~$\sigma_{\varepsilon,\varepsilon'}\in C^{\infty}(\Sigma)$ for each~$\varepsilon>0$, and assuming $\varepsilon' \sim \varepsilon$ as $\varepsilon\to 0^+$, we have
	  \begin{equation}
	     \frac{1}{24\pi}\lim_{\varepsilon\to 0^+}\Big[ 
	     \int_{\Sigma}(|\nabla_g \sigma_{\varepsilon,\varepsilon'}|_g^2+2K_g\sigma_{\varepsilon,\varepsilon'})\,\dd V_g 
         +2\pi\sum_{i=1}^p \frac{\gamma_i^2}{1+\gamma_i}\log(\varepsilon)
         \Big] =\mathcal{R}A_{\Sigma}(\tilde{g},g).
	    \label{}
	  \end{equation}
	\end{lemm}

\begin{proof}
	  Since~$g_{\varepsilon,\varepsilon'}=\tilde{g}$ on~$\Sigma\setminus \bigcup_{i=1}^p B_{\varepsilon}(z_i,\tilde{g})$ we just need to show that for each~$j$,
	  \begin{equation}
	    \int_{B_{\varepsilon}(z_j,\tilde{g})}(|\nabla_g \sigma_{\varepsilon,\varepsilon'}|_g^2+2K_g\sigma_{\varepsilon,\varepsilon'})\,\dd V_g \quad\xlongrightarrow{\varepsilon\to 0^+}\quad 0.
	    \label{}
	  \end{equation}
  We suppress the subscript~$j$ as we treat each ball individually. Now
	  \begin{equation}
	    \sigma_{\varepsilon,\varepsilon'}(r)= \gamma \cdot\sreg_{\delta,\delta'}(r) +\varphi,
	    \label{}
	  \end{equation}
	  with~$r=d_g(\bullet,z_j)$. From the fact that 
 	\begin{equation}
 \log \delta' \le  \sreg_{\delta,\delta'}(r) \le \log \delta 
\end{equation}
one finds that 
 	\begin{equation}
\sigma_{\epsilon,\epsilon'}(r) = \mathcal{O}( \log \epsilon)
\end{equation}
  and hence the integral of~$K_g \sigma_{\varepsilon,\varepsilon'}$ over the ball of radius $\epsilon$ goes to zero.  Since~$\varphi$ satisfies the Troyanov conditions (\ref{eqn-scale-metric-pot}) we have
  \begin{equation}
    \int_{B_{\varepsilon}(z_j,\tilde{g})}|\nabla_g \varphi|_g^2\,\dd V_g \lesssim \int_0^{2\pi}\int_{0}^{\delta}r^{4\gamma+2}\cdot r\,\dd r\,\dd \theta=\mathcal{O}(\delta^{4\gamma+4})\quad\xlongrightarrow{\varepsilon\to 0^+}\quad 0,
    \label{}
  \end{equation}
  taking $(r,\theta)$ to be the geodesic polar coordinates with respect to $g$. Now it is straightforward to check that 
  	 \begin{equation}
    \sup_{r\le \delta}\big|\nabla_g \sreg_{\delta,\delta'}(r)\big|_{g} =\mathcal{O}(\delta^{-1}),\quad\quad \delta\to 0^+.
    \label{}
  \end{equation}
  provided $\delta' \sim \delta$.  Now together with (\ref{eqn-scale-metric-pot}) we have
  \begin{equation}
    \int_{B_{\varepsilon}(z_j,\tilde{g})}\ank{\nabla_g \sreg_{\delta,\delta'},\nabla_g\varphi}_g\,\dd V_g \lesssim \delta^{-1}\int_0^{2\pi}\int_{0}^{\delta}r^{2\gamma+1}\cdot r\,\dd r\,\dd \theta=\mathcal{O}(\delta^{2\gamma+2})\quad\xlongrightarrow{\varepsilon\to 0^+}\quad 0,
    \label{}
  \end{equation}
	  and finally, again by the fact that $\delta'\sim\delta$,
	  \begin{equation}
	    \int_{B_{\varepsilon}(z_j,\tilde{g})} \big|\nabla_g \sreg_{\delta,\delta'}(r)\big|^2\,\dd V_g \lesssim \frac{1}{\delta^2}\int_{\{\delta'\le r\le \delta\}}^{} \dd V_g
	    =\frac{1}{\delta}\mathcal{O}(\delta - \delta')\quad\xlongrightarrow{\delta \to 0^+}\quad 0. \label{eqn-comparison-cut-off-conv}
	  \end{equation}
	  which concludes the proof.
	\end{proof}

    \begin{def7}
    This lemma specifically explains why no boundary terms are included in the anomaly formula from definition \ref{def-renorm-polya-conic}. We are working over~$\Sigma\setminus \bigcup_{i=1}^p B_{\varepsilon}(z_i,\tilde{g})$, which is a surface with boundary, and the boundary is generally not geodesic. The reason we don’t include boundary terms is that we could have alternatively opted for the method described in this subsection, which only considers the closed surface~$\Sigma$, and still arrived at the same outcome. Moreover, as mentioned in the "future work" section of the introduction, we intend to investigate various settings involving surfaces with boundary in a future revision of the manuscript.
	\end{def7}

\begin{def7}
	  Imaginably, one could allow even more flexibility in the choice of the cut-off function~$\sreg_{\delta,\delta'}$ to gain even better convergence in (\ref{eqn-comparison-cut-off-conv}). Clearly one just needs to make the infinitesimal annuli with nonzero gradient~$\nabla \sigma_{\varepsilon,\varepsilon'}$ thin enough. Nevertheless our choice suffices.
	\end{def7}

\section{Main Proof}

\subsection{Proof}\label{sec-main-proof}

\begin{prop}\label{prop-main-conical-scaling}
	  Let~$\Sigma$ be a closed Riemann surface and~$\tilde{g}$ a generalized conformal metric representing~$D=\sum_{j=1}^p \gamma_j z_j$ with~$\gamma_j>-1$. Now suppose~$h\in C^{\infty}(\Sigma)$ and consider the scaled metric~$\me^{2h}\tilde{g}$ on~$\Sigma\setminus\supp D$. Then 
      \begin{equation}
	 \mathcal{R}A_{\Sigma}(e^{2h}\tilde{g},g) -\mathcal{R}A_{\Sigma}\left(\tilde{g},g\right)= \mathcal{R}A_{\Sigma}(e^{2h}\tilde{g},\tilde{g})  - \frac{1}{12} \sum_j   \frac{\gamma_j(\gamma_j+2)}{\gamma_j+1} h(z_j),
     \label{eqn-main-gen-conic-interp-cocycle}
	  \end{equation}
	  where $\mathcal{R}A_{\Sigma}(e^{2h}\tilde{g},\tilde{g}) $ is defined by (\ref{eqn-def-renorm-polya-conic}) and $\mathcal{R}A_{\Sigma}(e^{2h}\tilde{g},g) $, $\mathcal{R}A_{\Sigma}\left(\tilde{g},g\right)$ by definition \ref{def-main-renorm-anomaly}.
	\end{prop}

    \begin{def7}\label{rem-conic-cocycle}
  From this we deduce that if we consider a conformal field theory with central charge~$c$ defined on $\Sigma$ equipped with the metric $\tilde{g}$ as above, whose partition function $\mathcal{Z}(\Sigma,\tilde{g})$ is defined in definition \ref{def-renom-part-func}, then we have
  \begin{equation}
	    \frac{\mathcal{Z}(\Sigma, e^{2h} \tilde{g})}{\mathcal{Z}(\Sigma, \tilde{g})} = 
	    \exp\Big(c \, \mathcal{R}A_{\Sigma}(\me^{2h}\tilde{g},\tilde{g}) -\sum_{j=1}^p h(z_j)\Delta_{j} \Big), 
	    \label{eqn-main-general-conic}
	  \end{equation}
	  where
	  \begin{equation}
	    \Delta_{j} \defeq \frac{c}{12} \frac{\gamma_j(\gamma_j+2)}{\gamma_j+1}
	    \label{}
	  \end{equation}
	  are the \textsf{conformal weights} or \textsf{scaling dimensions} associated to conical singularities.
\end{def7}

    First we point out (directly from corollary \ref{cor-quad-blow-tilde})

\begin{lemm}\label{lemm-comp-renorm-anom}
  Let~$g$ be a smooth conformal metric on~$\Sigma$ so that near the cone points~$z_j$ the regular metric potentials~$\varphi_j$ of~$\tilde{g}$ against~$g$ satisfy (\ref{eqn-scale-metric-pot}).  Then
  \begin{equation}
    \mathcal{R}A_{\Sigma}(\tilde{g},g)=\frac{1}{12}\Big[-\sum_{j=1}^p\frac{ \gamma_j^2\log(\gamma_j+1)}{\gamma_j+1}+\sum_{j=1}^p \Big( \frac{\gamma_j^2}{\gamma_j+1}-\gamma_j \Big)\varphi_j(z_j)\Big]
    +\frac{1}{24\pi}\int_{\Sigma\setminus \supp D}^{} \sigma(-\Delta_g \sigma+2K_g)\,\dd V_g,
    \label{}
  \end{equation}
  where~$\tilde{g}=\me^{2\sigma}g$ with~$\sigma\in C^{\infty}(\Sigma\setminus \supp D)$.\hfill $\Box$
\end{lemm}

\begin{proof}
  [Proof of proposition \ref{prop-main-conical-scaling}.]  It follows from lemma \ref{lemm-consist} that it is sufficient to prove (\ref{eqn-main-gen-conic-interp-cocycle}) for a specific reference smooth metric $g$. Choose a smooth conformal metric~$g$ with the same conditions as in lemma \ref{lemm-comp-renorm-anom}. Write~$\tilde{g}=\me^{2\sigma}g$ then~$\me^{2h}\tilde{g}=\me^{2(h+\sigma)}g$,~$\sigma\in C^{\infty}(\Sigma\setminus \supp D)$. Using lemma \ref{lemm-comp-renorm-anom} and relation (\ref{eqn-conf-change-to-smooth}) one obtains immediately 
  \begin{equation}
	 \mathcal{R}A_{\Sigma}(e^{2h}\tilde{g},g) - \mathcal{R}A_{\Sigma}(e^{2h}\tilde{g},\tilde{g}) - \mathcal{R}A_{\Sigma}\left(\tilde{g},g\right)  = - \frac{1}{12}\sum_j   \frac{\gamma_j}{\gamma_j+1} h(z_j) + \frac{1}{24 \pi} \int_{\Sigma\setminus \supp D} \left( h \Delta_g \sigma -  \sigma \Delta_g h \right) \,\dd V_g.
	  \end{equation}
What we desire follows then from applying lemma \ref{lemm-app-poincare-lel} to $\sigma$ which has the local form (\ref{eqn-local-form-conf-factor}).
\end{proof}

\begin{def7}\label{rem-extra-quad-term}
  In view of the Poincar\'e-Lelong lemma it would be instructive to compare the result with a ``naive guess'' for the Polyakov formula by pretending that the conical singularities did no more than just putting a log term into the metric potential and hence creating a delta function in the curvature. This is to say we suppose
  \begin{equation}
    \log\Big(
    \frac{\mathcal{Z}(\Sigma, e^{2h} \tilde{g})}{\mathcal{Z}(\Sigma, \tilde{g})} \Big)\textrm{``}=\textrm{''}\frac{c}{24\pi}\int_{\Sigma}(|\nabla_{\tilde{g}} h|_{\tilde{g}}^2+2K_{\tilde{g}}h)\,\dd V_{\tilde{g}}.
    \label{}
  \end{equation}
  Now taking into account the delta functions produced by~$-\Delta_g \sigma$ in (\ref{eqn-conf-change-to-smooth}) over~$\supp D$, we find
  \begin{equation}
    \int_{\Sigma}(|\nabla_{\tilde{g}} h|_{\tilde{g}}^2+2K_{\tilde{g}}h)\,\dd V_{\tilde{g}} =\int_{\Sigma\setminus\supp D}(|\nabla_{\tilde{g}} h|_{\tilde{g}}^2+2K_{\tilde{g}}h)\,\dd V_{\tilde{g}} -4\pi\sum_{j=1}^p \gamma_j h(z_j).
    \label{}
  \end{equation}
  Therefore the extra piece that we get in the actual formula is the quadratic term~$\sum_j \frac{2\pi\gamma_j^2}{\gamma_j+1}h(z_j)$. Going back to the definitions, we could see that this term originates as the asymptotic Dirichlet energy of the metric potential~$\sigma$ on a shrinking ``scaled annulus'', which we singled out as lemma \ref{lemm-quadratic-sing} above (cf.\ corollary \ref{cor-ratio-radius}).
\end{def7}

\begin{prop}\label{prop-main-app-ram-cov}
	  Let~$\Sigma_d$,~$\Sigma$ be closed Riemann surfaces, with a smooth conformal metric~$g$ on~$\Sigma$, and~$f:\Sigma_d\lto \Sigma$ a ramified $d$-sheeted holomorphic map, whose critical values are~$w_1$, \dots,~$w_p$. Consider a conformal field theory with central charge~$c$ whose partition function is denoted~$\mathcal{Z}$. Pick~$h\in C^{\infty}(\Sigma)$ on~$\Sigma$. Then under the definition \ref{def-renom-part-func} for $\mathcal{Z}(\Sigma_d, f^* e^{2h} g)$ and $\mathcal{Z}(\Sigma_d, f^* g)$ we have
	  \begin{equation}
	    \frac{\mathcal{Z}(\Sigma_d, f^* e^{2h} g)}{\mathcal{Z}(\Sigma, e^{2h} g)^d} = e^{- \sum_j  h(w_j)\Delta_{j}} \frac{\mathcal{Z}(\Sigma_d, f^* g)}{\mathcal{Z}(\Sigma, g)^d}, 
	    \label{}
	  \end{equation}
	  where
	  \begin{equation}
	    \Delta_{j} \defeq \frac{c}{12}  \sum_{z \in f^{-1}(w_j)}\Big(\ord_f(z) - \frac{1}{\ord_f(z)} \Big)
	    \label{}
	  \end{equation}
	  is the \textsf{conformal weight} or \textsf{scaling dimension} associated to the point $w_j$.
	\end{prop}
\begin{proof}
	  Following proposition \ref{prop-main-conical-scaling}, the only thing we need to check is
	  \begin{equation}
      \mathcal{R}A_{\Sigma_d}\big(f^*\me^{2h}g,f^*g\big)= d \cdot A_{\Sigma}\big(\me^{2h}g,g \big),
	    \label{}
	  \end{equation}
      where we recall that the left hand side is defined by (\ref{eqn-def-renorm-polya-conic}). Indeed, let $\{z_1,\dots,z_q\}$ denote the preimage of the critical values $\{w_1,\dots,w_p\}$ of $f$. Then on~$\Sigma_d\setminus \{z_1,...,z_q\}$ the map~$f$ is an unbranched~$d$-fold covering, we have
\begin{equation}
  \int_{\Sigma_d\setminus \{z_1,...,z_q\}}^{}
  \big(|\nabla_{f^*g}f^*h|_{f^*g}^2 +2K_{f^*g}\cdot f^*h \big)\,\dd V_{f^*g}=d\cdot \int_{\Sigma\setminus \{w_1,...,w_p\}}^{}
\big(|\nabla_{g}h|_{g}^2 +2K_{g}\cdot h \big)\,\dd V_{g},
  \label{}
\end{equation}
 and we finish the proof.
	\end{proof}

    \begin{def7}\label{rem-diffeo-inv}
  Note that~$\mathcal{Z}(\Sigma_d, f^* g)/\mathcal{Z}(\Sigma, g)^d$ is diffeomorphism invariant as a function of the critical values of~$f$ because~$\mathcal{Z}(\Sigma_d, f^*(\Psi^* g))/\mathcal{Z}(\Sigma,\Psi^* g)^d=\mathcal{Z}(\Sigma_d, (\Psi\circ f)^* g)/\mathcal{Z}(\Sigma, g)^d$ if~$\Psi:\Sigma\lto \Sigma$ is a diffeomorphism, for the ordinary~$\mathcal{Z}(\Sigma,g)$ is so invariant by definition.
\end{def7}

    \begin{corr}
  Consider a quantum system in the sense of subsection \ref{sec-main-replica} defined on a circle~$\mb{S}^1_L$ of perimeter~$L$, with central charge~$c$. Let~$A\subset \mb{S}^1_L$ be an interval of length~$\ell$ and let~$\rho$ be the \textsf{vacuum state} associated to this system. Then under the \textsf{replica interpretation} we have
  \begin{equation}
    \ttr_{\mathcal{H}_A}\big( \ttr_{A^c}(\rho)^d\big) =C \Big( \frac{L}{\pi}\sin \frac{\pi\ell}{L} \Big)^{-\frac{c}{6}\left( d-\frac{1}{d} \right)},
    \label{}
  \end{equation}
  where~$C$ is the same constant as in lemma \ref{lemm-two-pt-func}.
\end{corr}

\begin{proof}
  We adopt the interpretation of example \ref{exp-ground-state-disj-cob}. Then the surface~$\check{\Sigma}$ corresponds to a re-scaled Fubini-Study Riemann sphere (of radius~$L/2\pi$) where~$\mb{S}^1_L$ embeds isometrically as the equator. By proposition \ref{prop-main-app-ram-cov} we could first assume~$L=2\pi$. Suppose the end points of~$A$ correspond to~$0$ and~$z_0\in \wh{\mb{C}}$. Then the surface~$\check{\Sigma}_d$ is topologically still~$\wh{\mb{C}}$ but equipped with the pull-back of the Fubini-Study metric under a holomorphic map~$f_d:\wh{\mb{C}}\lto \wh{\mb{C}}$ which ramifies exactly at~$0$ and~$z_0$ both with order~$d$. Such a map could be given by (but not related to final result)
  \begin{equation}
    f_d(z)\defeq \frac{z_0\big( \frac{z}{z-z_0} \big)^d}{\big( \frac{z}{z-z_0} \big)^d -1}=\frac{z_0 z^d}{z^d-(z-z_0)^d}.
    \label{}
  \end{equation}
  Now we apply a smooth scaling~$\me^{2h}$ with~$h=\log(L/ 2\pi)$ on the equator, while making $\me^{2h}g_{\mm{FS}}$ flat in a thin neighborhood of the equator, in accordance with the requirement for vacuum state explained in example \ref{exp-ground-state-disj-cob}. The latter is always possible for Fubini-Study since~$4(1+|z|^2)^{-2}=|z|^{-2}+\mathcal{O}((1-|z|)^2)$. Then by proposition \ref{prop-main-app-ram-cov}, lemma \ref{lemm-two-pt-func} (and remark \ref{rem-diffeo-inv}) we obtain
  \begin{align*}
    \ttr_{\mathcal{H}_A}\big( \ttr_{A^c}(\rho)^d\big)&=\frac{\mathcal{Z}(\check{\Sigma}_d)}{\mathcal{Z}(\check{\Sigma})^d}=C\me^{-h(0)\Delta_d}\me^{-h(z_0)\Delta_d}\sin\Big( \frac{1}{2}d_{\mm{FS}}(0,z_0) \Big)^{-2\Delta_d}\\
    &=C\Big( \frac{L}{\pi}\sin \frac{\pi\ell}{L} \Big)^{-2\Delta_d},
  \end{align*}
  with~$\Delta_d=\frac{c}{12}(d-\frac{1}{d})$, giving us the result.
\end{proof}

\subsection{Comments on Relation to Literature}\label{sec-literature}
\begin{def7}
  [comparison with a result of V.\ Kalvin \cite{Kalvin}] \label{rem-rel-kal} We note here that V.\ Kalvin has obtained a closely related result in \cite{Kalvin} corollary 1.3(1) for the~$\zeta$-determinant of the Friederichs Laplacian under the conical metric with zero ``boundary condition'' at the cone points. As far as the Polyakov formula is concerned, the~$\zeta$-determinant is no different from a general CFT partition function in the sense of definition \ref{def-cft}, except for a constant. So our proposition \ref{prop-main-conical-scaling} coincides with the case of his result when the two divisors coincide. If they don't coincide, one could readily derive the corresponding result out of our method following lemma \ref{lemm-comp-renorm-anom}. However, in that case, the constants (last term in \cite{Kalvin} eqn.\ (1.8)) would be different, as our lemma \ref{lemm-comp-renorm-anom} only agrees with \cite{Kalvin} eqn.\ (1.1) upto the integral term and the term involving the regular metric potential. Supposedly, the extra constants in \cite{Kalvin} come from the specific Friederichs extension.

  On a deeper level, the proof of \cite{Kalvin} proceeds also by removing disks but then using the BFK decomposition formula. One key point is to show that the blow-up order of~$\log\detz \Delta_{D_{\varepsilon}(z_j)}$ of the Friederichs Dirichlet Laplacian~$\Delta_{D_{\varepsilon}(z_j)}$ on the shrinking disks~$D_{\varepsilon}(z_j)$ centered at the cone points with the singular metric, added up, matches exactly with the blow-up order of the~$\log\detz$ of the Dirichlet Laplacian on their complement, as the disks shrink (as the contributions from D-to-N maps cancel). In this regard we keenly observe that the cancellation mechanism as shown in \cite{Kalvin} eqn.\ (2.28) is \textit{different} from ours. The radii of Kalvin's disks were measured under the background smooth metric. For us, on the other hand, it is important to align the rate of shrinking with the singular metric itself (i.e.\ measure the disks using the singular metric) in two respects. First to obtain the invariance of the definition of the renormalized partition function as in lemma \ref{lemm-consist}, and secondly to obtain the correct scaling dimensions in the final result. A more detailed understanding of the relation of the two methods seems desirable, especially considering that the BFK formula is closely related to Segal's gluing axioms for QFT/CFT \cite{Lin, GKRV}, and that the correlation functions are (conjecturally) the ``limits'' of propagators associated to the complements of the shrinking disks (\cite{KS} section 3).
\end{def7}

\begin{def7}
  [relation to Aldana, Kirsten and Rowlett \cite{AKR}] \label{rem-rel-akr} As the authors have themselves remarked in relation to \cite{Kalvin}, they also obtained the same result as above for the~$\zeta$-determinant (\cite{AKR} eqn.\ (1.5)). Their method is based on asymptotic analysis of the heat trace for which related results were obtained earlier \cite{MR}. For the latter, they employed a blow-up (geometric) technique and the cone contribution term (main interest of this article) seems to ultimately come from computations done to an infinite exact cone which has a longer history (see their appendix A for further references). Last but not least, we note that \cite{AKR} seems to have restricted their cone angles to~$(0,2\pi)$, which is in the complement of the main concern of this article (proposition \ref{prop-main-app-ram-cov}).
\end{def7}

\chapter[Zeta Determinants on Cyclic Covers]{Zeta Determinants of Laplacians on Cyclic Covers}\label{chap-covers}

This Chapter is based on joint work with Nguyen Viet Dang and Fr\'ed\'eric Naud.

\section{Introduction}

In Euclidean Quantum Field Theory over some compact Riemannian 
manifold $M$, for $m\geqslant 0$, the partition function of the massive Gaussian free field is formally denoted by
\begin{align*}
Z_M=\int_{\mm{Map}(M,\mb{R})} \exp\Big( -\int_M \vert \nabla \varphi \vert^2 + m^2\varphi^2 \,\dd V_g \Big) \dd\mathcal{L}(\varphi)   
\end{align*}
where $\mm{Map}(\mb{S}^1,\mb{R})$ denotes loosely functions on which we deliberately do not specify the regularity, and $\mathcal{L}$ is non-existent Lebesgue measure on $\mm{Map}(\mb{S}^1,\mb{R})$, see section \ref{sec-intro-class-field-theory}. This formula is analogous to the case of discrete Gaussian free field which is some ferromagnetic spin system (cf.\ example \ref{exam-intro-discre-gff}) on a discrete box $\Lambda\subset \mathbb{Z}^d$ of finite size whose partition function
reads $$ Z_\Lambda= \int_{ \mathbb{R}^\Lambda}  \exp\Big( -\sum_{i\sim j} \vert \sigma_i-\sigma_j\vert^2 -\sum_{i\in \Lambda} m^2\sigma_i^2 \Big) \,\dd^\Lambda\sigma. $$

In the discrete case, the Dirichlet action functional $S(\sigma)$ is a quadratic form on $\mathbb{R}^\Lambda$, $S(\sigma)=\sum_{i\in \Lambda} \sigma_i(\Delta_\Lambda \sigma)_i+m^2\sigma_i^2 $ where $\Delta_\Lambda$ is the discrete Laplacian and therefore, the partition function $Z_\Lambda$ of the discrete GFF is a simple Gaussian integral which is easily calculated and equals
\begin{align*}
 Z_\Lambda= \det(\Delta_\Lambda+m^2)^{-\frac{1}{2}}   
\end{align*}
where $\Delta_\Lambda$ is the discrete Laplacian for discrete functions on $\Lambda$.
On the manifold $M$, since the space of functions is infinite and the Laplace Beltrami operator is no longer a finite dimensional matrix, so one has to regularize the  determinant or equivalently one has to regularize the infinite product $\prod_{\lambda\in \sigma(\Delta)\setminus \{0\}} (\lambda+m^2)$.
Inspired by the seminal work of Ray--Singer~\cite{RS} and also by the more physically oriented paper
of Hawking~\cite{Hawking},
a common choice in theoretical physics is to use zeta regularization to give a mathematical meaning to free fields partition functions. So $Z_M$ is defined to be 
\begin{align*}
Z_M:=\text{det}_\zeta\left(\Delta_g+m^2 \right)^{-\frac{1}{2}}    
\end{align*}
where we refer to paragraph~\ref{sss:statement} for a precise definition of 
zeta determinants.

From the point of view of statistical physics, 
the free energy $\log Z_M$ is expected to grow as the size of the system, here the natural system being our Riemannian manifold $M$. So the natural \textbf{intensive quantity} should be the 
density of free energy defined as~\cite[eq (24.28) p.~144]{Schwarz}
\begin{equation}\label{eq:freeenergy}
 F\defeq \frac{\log (Z_M)}{\text{Vol}(M)} .
\end{equation}

It is expected, as the volume of our Riemannian manifold gets larger in a certain limiting sense, that 
the above quantity converges as the volume of $M$ goes to infinity. 
It means there exists a \textbf{density of free energy in the thermodynamic limit}, which is also interpreted as the energy density of the ground state of a quantum system described by the GFF. 
Inspired by the recent results of the third author~\cite{Nauddet} as well as several studies on the spectral theory of Laplacians on random surfaces~\cite{MonkMarklov} and motivated by problems of taking large volume limits in
quantum field theory~\cite{Lin},
the goal of the present paper is to establish the existence of a limit for the quantity $F$ defined by equation~(\ref{eq:freeenergy}) for suitable sequences of cyclic covers of large degree that \emph{converge} in some appropriate sense to some $\mathbb{Z}^p$-periodic Riemannian manifolds obtained as abelian cover of some initial closed compact Riemannian manifold $M$.

\subsubsection{Geometric Setting}
\label{sss:setting}
Our geometric set-up here will be a direct generalization of that in section \ref{s:abeliancovers} to general dimensions.
Let $(M,g)$ denotes a smooth, closed, compact Riemannian manifold. 
We shall start by giving some way to manufacture towers of abelian covers of $M$ following~\cite{JaNaSo}.
Consider any group homomorphism of the form
$$\rho: \pi_1(M)\lto H_1(M,\mathbb{Z})\lto \mathbb{Z}$$
where the last arrow $H_1(M,\mathbb{Z})\lto \mathbb{Z}$ has infinite elements in its image.
Then $\ker(\rho)$ is a subgroup of the fundamental group
$\pi_1(M)$. We denote by $\widehat{M}$
the universal cover of $M$, $\ker(\rho)$ is a subgroup of the deck group of $\widehat{M}\lto M$ (deck group are automorphisms of the cover which descend trivially on the base) and the quotient 
$M_\infty:= \widehat{M}/\ker(\rho) $
is an abelian cover of $M$ with deck group $\mathbb{Z}$, in fact the quotient
$  \widehat{M}/\ker(\rho) $ is a periodic Riemannian manifold endowed with an 
isometric action of $\mathbb{Z}$.
We denote by $\rho_N:\pi_1(M)\lto \mathbb{Z}\lto \mathbb{Z}/N\mathbb{Z}$
the composition of $\rho$ with mod $N$ reduction.
Then the quotient $M_N:= \widehat{M}/\ker(\rho_N)$ is an abelian cover of $M$ with deck group $\mathbb{Z}/N\mathbb{Z}$ and when $N\rightarrow +\infty$ we can think of the sequence $(M_N)_N$ as an increasing sequence of abelian covers that \emph{converges} to the periodic cover $M_\infty$. Denoting by $\gamma$ the generator of the $\mathbb{Z}$-action on $M_\infty$, one can also think of $M_N$ as the quotient $M_\infty/\left\langle \gamma^N \right\rangle$ where $\left\langle \gamma^N \right\rangle$ denotes the subgroup of $\mathbb{Z}$ generated by $N$.

Like in \ref{s:abeliancovers}, we now make a \textbf{crucial assumption} which will also related our set-up to Segal's composition picture as described in \ref{s:abeliancovers}. Let us concretely describe how the abelian cover appears from the manifold $M$ we started with. 
We suppose a closed hypersurface $\Sigma\subset M$ whose dual class
$[\Sigma]\in H^1(M,\mathbb{Z})$ is non-trivial and non-torsion which means that the class $[\Sigma]\in H^1(M,\mathbb{R})$ with coefficients over the reals is non-trivial. 
The dual class is defined by the intersection of cycles
$$ [\Sigma]:= [p]\in H_1(M)\longmapsto \left\langle [p],[\Sigma] \right\rangle \in \mathbb{Z}$$
which we assume to be \textit{surjective}. 
The intersection pairing defines a morphism
$\rho:\pi_1(M)\lto H_1(M,\mathbb{Z})\lto \mathbb{Z}  $ 
whose kernel $\ker(\rho)$ generates a $\mathbb{Z}$-cover of $M$ as $\widehat{M}/\ker(\rho)$. Geometrically, 
cut the manifold $M$ along $\Sigma$, this creates a smooth manifold with even number of boundary components, identify one part as ingoing and the other part as outgoing yielding a cobordism $U$. Then make an infinite iterate 
$U\circ U\circ\dots \circ U$ of this cobordism generates the $\mathbb{Z}$-cover $\widehat{M}/\ker(\rho)$.    
%
To generate cyclic covers, just repeat the previous operation
and compose the previous $\rho$ with reduction mod $N$.

\subsubsection{Statement of the Main Result}
\label{sss:statement}

In this chapter we will focus on the limits as $N\to +\infty$ of quantities (\ref{eq:freeenergy}) \textbf{in the case} $m=0$. The case $m>0$ was in a sense (for surfaces) obtained in subsection \ref{sec-segal-end-asymp-part-func} by exploiting the composition axiom in the context of a Segal QFT. Moreover the $m>0$ case may also be proved following a similar argument as in subsection \ref{sec-covers-dom-converge} of the present chapter without resorting further to the lemmas. However for $m=0$ the accumulation of the eigenvalues near $0$ of the Laplacian on the cover as the degree grows large creates an extra difficulty, and a main part of the present chapter, therefore, is to deal with this difficulty. 

We keep the notations from the previous paragraph, denote by $\Delta_N$ (resp.\ $\Delta_\infty$) the Laplace-Beltrami operator on $M_N$ (resp.\ $M_\infty$) we can state the main results of our note. We consider $M_N\lto M$ a sequence of cyclic cover of degree $N$ as defined in the introduction~\ref{s:abeliancovers}. One would like to think that the sequence of covers $M_N$ converges in a suitable sense to a limit cover $M_\infty$. We denote by $\Delta_{N}$ the natural Laplace-Beltrami operator on the cover $M_N$, $\det_\zeta(\Delta_{N})$ denotes the zeta regularized determinant of $\Delta_{N}$ defined
as in \ref{exp-gff} and also~\cite{RS},
\begin{eqnarray}
\text{det}'_\zeta(\Delta_{M_N})\defeq\exp(-\zeta^\prime(0)),\quad\quad  \zeta_{\Delta_N}(s)\defeq\sum_{\lambda\in\sigma(\Delta_N)\setminus \{0\}}\lambda^{-s}
\end{eqnarray}
where the spectral zeta function $\zeta_{\Delta_N}(s)
$ has holomorphic continuation near $s=0$.
 We show in the present note the following:
\begin{thrm}\label{thm:convergencewithmass}
Under the above assumptions, we find that the sequence
\begin{eqnarray}
\boxed{\frac{\log\left(\det'_\zeta\left( \Delta_{M_N}\right) \right)}{\mm{vol}(M_N)}}
\end{eqnarray}
has a limit when $N\rightarrow +\infty$ and we give a formula for the limit in terms of the heat kernel on $M_\infty$.
\end{thrm}

\section{Main Strategy and Proof}
We start by recalling the Mellin transform formula for the zeta function of a positive semi-definite elliptic differential operator~$P$, with strictly positive principal symbol\footnote{This means the principal symbol~$\sigma_P(x,\xi)$ is positive definite for~$\xi\ne 0$ (\cite{Gilkey} section 1.6.2).}, acting on a Hermitian vector bundle over a Riemannian manifold (\cite{Gilkey} section 1.12):
\begin{equation}
  \zeta_P(s)=\frac{1}{\Gamma(s)}\int_{0}^{\infty}t^{s-1}\big(\ttr_{L^2}(\me^{-t P})-\dim\ker P \big)\,\dd t.
  \label{eqn-general-mellin-zeta-func}
\end{equation}
We denote by~$\me^{-tP}(x,y)$ the heat kernel which is smooth on~$M\times M$. In the first part of this paper we consider~$P=\Delta_{\infty}$,~$\Delta_N$, or~$\Delta$, the Laplacians on~$M_{\infty}$,~$M_N$ or~$M$ respectively, acting on functions.

By definition of the~$\zeta$-determinant, to prove the existence of the limit in theorem \ref{thm:convergencewithmass} it suffices to prove that~$\{\zeta_{\Delta_N}(s)/\vol(M_N)\}$ is a sequence of holomorphic functions converging in the compact-open topology on a neighborhood of~$s=0$ as~$N\to +\infty$. To do this, we use formula (\ref{eqn-general-mellin-zeta-func}) and then express the heat trace as the integral along diagonal of the heat kernel. 
We use the crucial fact that the heat kernels on~$M_N$ and~$M_{\infty}$ are related in a simple way by summing over the deck orbits (\cite{Buser} section 7.5). 
Applying a Li-Yau type estimate on the heat kernel over $M_{\infty}$, we obtain the convergence of~$N^{-1}\ttr_{L^2(M_N)}(\me^{-t\Delta_N})$ as~$N\to +\infty$ for fixed~$t>0$. 
The main technical part then boils down to obtaining a bound on the heat trace in~$t$, uniform in~$N$, to justify dominated convergence.

\subsection{Relation of Heat Kernels and a Li-Yau Estimate}

We will use in a crucial way the $\mathbb{Z}$-cover $M_\infty \lto M_{N}  \lto M$ which covers all the degree $N$ cyclic covers $M_N$.
The group $\mathbb{Z}$ acts on $M_\infty$ by isometries (the reader should think translations), the generator is denoted by $\gamma:M_\infty\lto M_\infty$ that we identify with an isometry of $M_\infty$.
The infinite cover $M_\infty$ endowed with the metric induced from $(M,g)$ is a geodesically complete Riemannian manifold by the Hopf-Rinow Theorem since every closed bounded subset in it is compact.

Denote by $\Delta_{\infty}$ the Laplacian on $M_\infty$, $\pi_N:M_\infty\lto M_N$. The geodesic completeness implies that $\Delta_{\infty}:C^\infty_c(M_\infty)\subset L^2(M_\infty)\lto L^2(M_\infty) $ has a self-adjoint extension by the work of Gaffney. 
Then observe we have the identity
\begin{eqnarray*}
\left(\pi_N\times \pi_N\right)^* e^{-t\Delta_N}=\sum_{k\in \mathbb{Z}} e^{-t\Delta_\infty}(\cdot,\gamma^{kN}\cdot).
\end{eqnarray*}

We make the crucial observation that
\begin{eqnarray*}
\ttr_{L^2(M_N)}(e^{-t\Delta_N})=\sum_{k\in \mathbb{Z}}\int_{\Omega_N}e^{-t\Delta_\infty}(x,\gamma^{kN}x)\,\dd V_g(x)\,=
N\sum_{k\in \mathbb{Z}}\int_{\Omega}e^{-t\Delta_\infty}(x,\gamma^{kN}x)\,\dd V_g(x)\,
\end{eqnarray*}
where $\Omega$, $\Omega_N$ are the respective fundamental domains of $M,M_N$ on the cover $M_\infty$, we also used the invariance equation $e^{-t\Delta_\infty}(x,x)=e^{-t\Delta_\infty}(\gamma(x),\gamma(x))$.

In the sequel, we will make an extensive use of the following Lemma
which is a particular case of Milnor-Schwarz~\cite[Prop 8.19 p.~140]{BrHa}:
\begin{lemm}[Milnor-Schwarz type Lemma]
Let $(M,g)$ be a smooth compact Riemannian manifold and $M_\infty\lto M$ a $\mathbb{Z}$ cover of $M$. Then
there exists $L>0$ such that for all $x\in M_\infty$:
$$ \textbf{d}(x,\gamma^px)\geqslant pL. $$
A similar kind of estimate holds true in the case of $\mathbb{Z}^d$ covers.
\end{lemm}

The second ingredient we shall need are Gaussian bounds on the heat kernel $e^{-t\Delta_\infty}$ which are due to Li-Yau~\cite[Thm 4.6 p.~169]{SchoenYau}:
\begin{thrm}[Li-Yau heat kernel bounds]
Let $(M_\infty,g)$ be a complete Riemannian manifold of dimension $d$ with $\mm{Ric}(M)\ge -K$, $K\ge 0$. Then for all $0<t\le 1$, $\delta\in (0,1)$ and $(x,y)\in M_\infty\times M_{\infty}$, the heat kernel $e^{-t\Delta_\infty}$
satisfies the bound of the form
\begin{eqnarray}
\vert e^{-t\Delta_\infty}(x,y)\vert\leqslant C_1(\delta,d)t^{-\frac{d}{2}} e^{-\textbf{d}^2(x,y)/(4+\delta)t}\me^{C_3 K\delta t}
\label{eqn-li-yau-bound}
\end{eqnarray}
where $\textbf{d}$ is the Riemannian distance and $C_1(\delta,d)\rightarrow +\infty$ when $\delta\rightarrow 0$, and $C_3$ depends only on the dimension.
\end{thrm}

\subsection{Dominated Convergence}\label{sec-covers-dom-converge}
Here we prove theorem \ref{thm:convergencewithmass} while leaving details on the estimate of the heat trace to section \ref{sec-long-time-control}.

\begin{proof}
    (Proof of theorem \ref{thm:convergencewithmass}.) Using (\ref{eqn-general-mellin-zeta-func}), we decompose
    \begin{equation}
        \frac{\zeta_{\Delta_N}(s)}{\vol(M_N)}=\frac{1}{\Gamma(s)\vol(M_N)}\Big( \underbrace{\int_0^1}_A +\underbrace{\int_1^{\infty}}_B\Big) \big(\ttr_{L^2(M_N)}(e^{-t\Delta_N})-1\big)t^{s-1}\,\dd t.
    \end{equation}
    For part $A$ we have
     \begin{equation}
    \frac{1}{N}\int_{0}^{1}\big(\ttr_{L^2(M_N)}(e^{-t\Delta_N})-1\big)t^{s-1}\,\dd t =
    \int_{0}^{1}\Big[\int_{\Omega}^{} \sum_{k\in\mb{Z}} \me^{-t\Delta_{\infty}}(x,\gamma^{kN}x)\,\dd V_g(x) -\frac{1}{N}\Big] t^{s-1}\,\dd t.
    \label{}
  \end{equation}
  We observe by the Milnor-Schwarz lemma and Li-Yau estimate (and positivity of the heat kernel) that
    \begin{align*}
|\mathcal{T}_{\ne 0}(t,s)|&\defeq \bigg| t^{s-1} \int_{\Omega} 
\sum_{\substack{k\in \mathbb{Z}\\ k\neq 0}}e^{-t\Delta_\infty}(x,\gamma^{kN}x)\,\dd V_g(x)\bigg| \leqslant 
   \vert t^{s-1} \vert \,\int_{\Omega}
  \sum_{\substack{k\in \mathbb{Z}\\ k\neq 0}}e^{-t\Delta_\infty}(x,\gamma^{kN}x)\,\dd V_g(x)\\
  &\lesssim 
  t^{\fk{Re}(s)-1-\frac{d}{2}} \int_{\Omega}  \sum_{\substack{k\in \mathbb{Z}\\ k\neq 0}}e^{-C_2 \mathbf{d}^2(x,\gamma^{kN}(x))/t} \,\dd V_g(x) \\
& \lesssim  \vol(M) \cdot t^{\fk{Re}(s)-1-\frac{d}{2}}\Big(\sum_{\substack{k\in \mathbb{Z}\\ k\neq 0}} e^{-C_2 (kNL)^2 /t}\Big)    \le  \vol(M) \cdot t^{\fk{Re}(s)-1-\frac{d}{2}}   \Big(2\sum_{k=1}^{\infty} e^{-C_2 kNL /t}\Big)   \\
& = \vol(M) \frac{ 2e^{-C_2 NL/t}}{1-e^{-C_2 NL/t}}\cdot t^{\fk{Re}(s)-1-\frac{d}{2}}  
\lesssim e^{-C_2 NL/t} \cdot t^{\fk{Re}(s)-1-\frac{d}{2}}.\tag{\#}
\end{align*}
 Since $ e^{-\frac{C}{t}} t^{\fk{Re}(s)-1-\frac{d}{2}} =\mathcal{O}(t^\infty)$ for $C>0$ as $t\downarrow 0$, the r.h.s.\ of the above is obviously integrable on $(0,1)$, regardless of $s$. This shows firstly $\int_0^1 \mathcal{T}_{\ne 0}(t,s)\,\dd t$ in fact defines an entire function of $s$, and secondly from a simple change of variables we have $\int_0^1 e^{-C_2 NL/t} \cdot t^{\fk{Re}(s)-1-\frac{d}{2}}\,\dd t \lesssim N^{\fk{Re}(s)-1-\frac{d}{2}}$, and we conclude that
$$  \int_0^1 \mathcal{T}_{\ne 0}(t,s) \,\dd t
\lesssim N^{\fk{Re}(s)-1-\frac{d}{2}}\quad \xlongrightarrow{N\to +\infty} \quad 0$$
for $\fk{Re}(s)\ll 0$ and hence as entire functions of $s$ (Montel-Vitali theorem). Since
  \begin{equation}
    \frac{1}{N\Gamma(s)}\int_{0}^{1}t^{s-1}\,\dd t=\frac{1}{N\Gamma(s+1)}\quad \xlongrightarrow{N\to +\infty} \quad 0,
    \label{}
  \end{equation}
  we deduce
  \begin{equation}
    \frac{1}{N\Gamma(s)}\int_{0}^{1}\big(\ttr_{L^2(M_N)}(e^{-t\Delta_N})-1\big)t^{s-1}\,\dd t\quad \xlongrightarrow{N\to +\infty} \quad  
    \frac{1}{\Gamma(s)}\int_{0}^{1}t^{s-1}\int_{\Omega}^{} \me^{-t\Delta_{\infty}}(x,x) \,\dd V_g(x)\,\dd t,
    \label{}
  \end{equation}
  the latter defining a meromorphic function of~$s$ which is holomorphic near~$s=0$ by the asymptotic expansion of the heat kernel (\cite{BGV} theorem 2.30) and that~$\vol(\Omega)=\vol(M)<\infty$.

  To treat part~$B$, we first note that by the same estimates (\#) we have shown
\begin{equation}
  \lim_{N\to +\infty} \frac{1}{N}\big(\ttr_{L^2(M_N)}(e^{-t\Delta_N})-1\big) =\int_{\Omega}^{} \me^{-t\Delta_{\infty}}(x,x) \,\dd V_g(x)
  \label{eqn-heat-trace-pointwise-limit}
\end{equation}
for each fixed~$t>0$. Now by lemma \ref{lemm-long-time-bound-heat-trace} and dominated convergence we obtain
\begin{equation}
     \lim_{N\to +\infty}\frac{1}{N\Gamma(s)}\int_{1}^{\infty}\big(\ttr_{L^2(M_N)}(e^{-t\Delta_N})-1\big)t^{s-1}\,\dd t=  \frac{1}{\Gamma(s)}\int_{1}^{\infty}t^{s-1}\int_{\Omega}^{} \me^{-t\Delta_{\infty}}(x,x) \,\dd V_g(x)\,\dd t,
\end{equation}
as holomorphic functions of $s$ for $\fk{Re}(s)<1/2p$ and hence in particular near $s=0$. Combining the two parts we obtain the result with limit
\begin{equation}
  \lim_{N\to +\infty}\frac{\zeta_{\Delta_N}(s)}{\vol(M_N)}=\frac{\zeta_{\Delta_{\infty},\Omega}(s)}{\vol(M)}
  \label{}
\end{equation}
where
\begin{equation}
  \zeta_{\Delta_{\infty},\Omega}(s)\defeq \frac{1}{\Gamma(s)}\int_{0}^{\infty}t^{s-1}\int_{\Omega}^{} \me^{-t\Delta_{\infty}}(x,x) \,\dd V_g(x)\,\dd t.
  \label{}
\end{equation}
Note that we know a posteriori after (\ref{eqn-heat-trace-pointwise-limit}) and lemma \ref{lemm-long-time-bound-heat-trace} that~$\zeta_{\Delta_{\infty},\Omega}(s)$ is a holomorphic function of~$s$ near~$s=0$ as the sum of two such functions, one being a meromorphic continuation.
\end{proof}

\section[Accumulation of Eigenvalues and Heat Trace]{Accumulation of Eigenvalues and Long-time Behavior of Heat Trace}\label{sec-long-time-control}

The goal of this section is to prove
\begin{lemm}\label{lemm-long-time-bound-heat-trace}
  There exists a positive integer~$p$, a number~$\varepsilon_0>0$, and constants~$C_4$,~$C_5>0$ independent of~$N$, such that 
  \begin{equation}
    \frac{1}{N}\big(\ttr_{L^2(M_N)}(e^{-t\Delta_N})-1\big) \le C_4 t^{-1/2p} +C_5 \me^{-t \varepsilon_0},
    \label{}
  \end{equation}
  for all $N\in\mb{N}$ and all $t\ge 1$.
\end{lemm}

\begin{proof}
  We obtain from corollary \ref{cor-small-eigen-estima} the threshold~$\varepsilon_0>0$,~$p$ and~$C_4$ independent of~$N$ so that
  \begin{equation}
    \frac{1}{N}\sum_{\substack{\lambda\in \sigma(\Delta_{N}) \\ 0<\lambda<\varepsilon_0}} \me^{-t\lambda} \le C_4 t^{-1/2p}.
    \label{}
  \end{equation}
  Now since~$\vol(M_N)=N\vol (M)$, and~$\mm{Ric}(M_N)\ge \mm{Ric}(M)$, by lemma \ref{lemm-rough-counting} we have some~$C_6$ independent of~$N$ such that
  \begin{equation}
  \sharp\{\lambda\le \Lambda|\lambda\in \sigma(\Delta_N)\}\le C_6 N\Lambda^{d/2}.
  \label{}
\end{equation}
Now pick~$\varepsilon_0\ll\Lambda_0$, and put $\Lambda_j:=\Lambda_0+j$. Then we have for~$t\ge 1$,
\begin{align*}
  \sum_{\substack{\lambda\in \sigma(\Delta_N) \\ \lambda\ge \varepsilon_0}} \me^{-t\lambda}&\le \sharp\big([\varepsilon_0,\Lambda_0)\cap \sigma(\Delta_N)\big) \,\me^{-t\varepsilon_0}+\sum_{j=0}^{\infty}\sharp\big([\Lambda_j,\Lambda_{j+1})\cap \sigma(\Delta_N)\big)\, \me^{-t\Lambda_j}\\
  &\le C_6 N\Lambda_0^{d/2}\me^{-t\varepsilon_0}+\me^{-t\Lambda_0}\sum_{j=0}^{\infty} C_6 N\Lambda_{j+1}^{d/2} \me^{-tj}\\
  &\le C_6 N\Lambda_0^{d/2}\me^{-t\varepsilon_0}+ C_6 N \me^{-t\Lambda_0} \underbrace{\sum_{j=0}^{\infty} (\textrm{poly. in }j)\cdot  \me^{-j}}_{\textrm{converge}} \\
  &\le C_5 N\me^{-t\varepsilon_0}.
\end{align*}
We obtain the result.
\end{proof}

\subsection{Twisted Laplacians and Decomposition of Spectra}

The next step is to identify the action of the Laplace operator $\Delta_N$ on $M_N$ with the action of some family of Laplacians acting on functions on the base manifold $M$.

\subsubsection{Fourier Transform in the Fibers}

Consider $M_N\lto M
$ the $N$-th degree cyclic cover, we may think of the cover $M_N\lto M$ as a $\mathbb{Z}/N\mathbb{Z}$--principal bundle over $M$ so that we can use the discrete Fourier transform over the fibers. Then we
observe that
given any function $f$ on $M_N$, the deck group of $M_N$ acts on the function $f$ as some discrete rotation of the fibers. Therefore,
we can define the \textsf{fiberwise Fourier transform} $\wh{\bullet}:C^{\infty}(M_N)\lto C^{\infty}(M_N)^N$ by
\begin{eqnarray*}
\widehat{f}_p(x)\defeq \sum_{k=1}^N f(\gamma_0^k.x)e^{-2\pi \ii\frac{ p}{N}k},\quad\quad\textrm{for }0\le p\le N-1.
\end{eqnarray*}
where $\gamma_0$ is the \textit{generator} of the deck shifts. The Fourier inversion formula subsequently reads
\begin{eqnarray}
f(x)=\frac{1}{N}\sum_{p=1}^N \widehat{f}_p(x), \quad\quad x\in M_N.
\label{eqn-automorphic-fiber-fourier}
\end{eqnarray}
Now we define several function spaces as follows.
\begin{align}
  L^2_{p,N}(M_N)&\defeq \{u\in L^2(M_N)~|~ \gamma_0^* u=\me^{2\pi\ii \frac{p}{N}}u\},\\
  C^{\infty}_{p,N}(M_N)&\defeq \{u\in C^{\infty}(M_N)~|~ \gamma_0^* u=\me^{2\pi\ii \frac{p}{N}}u\},\\
  C^{\infty}_{p,N}(\Omega)&\defeq \{u|_{\Omega}~|~ u\in C^{\infty}_{p,N}(M_N)\},
  \label{}
\end{align}
where~$\Omega$ is the fundamental domain fixed in the introduction. Observe that $\wh{f}_p\in L_{p,N}^2(M_N)$ and
\begin{lemm}
  We have~$L^2_{p,N}(M_N)\perp L^2_{q,N}(M_N)$ in~$L^2(M_N)$ for~$p\ne q$.
\end{lemm}

\begin{proof}
    We have for $u\in L^2_{p,N}$ and $v\in L^2_{q,N}$, since the deck shifts are transitive,
    \begin{align*}
\left\langle u,v\right\rangle_{L^2(M_N)}&= \sum_{k=1}^N  \int_{\Omega} \overline{u}(\gamma^k.x)v(\gamma^k.x)\,\dd V_g(x)\,\\
&=
\sum_{k=1}^N e^{-2\pi \ii\frac{ p}{N}k}e^{2\pi \ii\frac{ q}{N}k} \int_{\Omega} \overline{u}(
x)v( x)\,\dd V_g(x)=0,
\end{align*}
by the automorphy condition and since $\sum_{k=1}^N e^{2\pi \ii\frac{ p-q}{N}k}=0 $.
\end{proof}

The Fourier inversion formula thus gives the orthogonal decomposition
\begin{equation}
  L^2(M_N)=\bigoplus_{p=0}^{N-1}L_{p,N}^2(M_N).
  \label{}
\end{equation}
Note that since~$\gamma_0$ acts by isometry we have~$\Delta_N:C^{\infty}_{p,N}(M_N)\lto C^{\infty}_{p,N}(M_N)$ for each~$p$. Since~$\Delta_N$ is essentially self-adjoint on~$C^{\infty}(M_N)$, we have also decomposed
\begin{equation}
  \Delta_N=\bigoplus_{p=0}^{N-1} \Delta_N|_{L_{p,N}^2(M_N)}.
  \label{}
\end{equation}

Next we observe that by the automorphy condition, the action of~$\Delta_N$ on each~$L_{p,N}^2(M_N)$ boils down to its action on~$L^2(\Omega)$ with a specified boundary condition. More precisely, given~$f\in L^2(\Omega)$, define its \textsf{(quasi-)periodization}~$\mathcal{P}_{p,N}f\in L^2(M_N)$, such that as a distribution
\begin{equation}
  \bank{\mathcal{P}_{p,N}f, \varphi}_{L^2(M_N)}\defeq \bank{f,\wh{\varphi}_p|_{\Omega}}_{L^2(\Omega)},
  \label{}
\end{equation}
for test functions~$\varphi\in C^{\infty}(M_N)$. One can subsequently verify that indeed~$\mathcal{P}_{p,N}f\in L^2_{p,N}$ by resorting to the adjoint of~$\gamma_0^*$. Moreover,~$\mathcal{P}_{p,N}:C^{\infty}_{p,N}(\Omega)\lto C^{\infty}_{p,N}(M_N)$. We arrive at the main observation of this subsection.
\begin{lemm}\label{lemm-unit-equiv-domain-bdy-cond}
  For each integer~$0\le p\le N-1$ we have the commutative diagram
  \begin{equation}
    \begin{tikzcd}
      C^{\infty}_{p,N}(\Omega) \ar[r," \Delta_{\Omega}"] \ar[d,"\mathcal{P}_{p,N}"'] & L^2(\Omega) \ar[d,"\mathcal{P}_{p,N}"] \\[+10pt]
      C^{\infty}_{p,N}(M_N) \ar[r," \Delta_N"] & L^2_{p,N}(M_N).
    \end{tikzcd}
    \label{}
  \end{equation}
  Denote by~$\Delta_{p,N}$ the self-adjoint extension of the Laplacian on~$L^2(\Omega)$ with core~$C^{\infty}_{p,N}(\Omega)$. Then~$\Delta_{p,N}$ is unitarily equivalent to~$\Delta_N|_{L_{p,N}^2(M_N)}$ through~$\frac{1}{N}\mathcal{P}_{p,N}$ for each~$0\le p\le N-1$. 
\end{lemm}

\begin{proof}
  We only point out that in this case the existence and uniqueness of the self-adjoint extension~$\Delta_{p,N}$ can be obtained thanks to the unitary map~$\frac{1}{N}\mathcal{P}_{p,N}$ and the corresponding result for~$\Delta_N|_{L_{p,N}^2(M_N)}$. A more general dituation is considered in lemma \ref{lemm-twist-self-adjoint-domain}.
\end{proof}

\subsubsection{Characters and Unitary Equivalence}

We remember that our cover~$M_{\infty}$ was constructed to correspond to the normal subgroup of~$\pi_1(M)$ which is the kernel of the map
  \begin{equation}
  \begin{tikzcd}
    \rho:\pi_1(M) \ar[r,"\mm{Ab}"] &[+10pt] H_1(M;\mb{Z}) \ar[r, " I(-{,}{[}\Sigma{]})"] &[+10pt] \mb{Z},
  \end{tikzcd}
  \label{}
\end{equation}
where the first is Abelianization and~$I(-,[\Sigma])$ is the \textsf{oriented intersection number} with a compact oriented hypersurface~$\Sigma$, which is surjective. 

By Poincar\'e duality and Hodge theory on the closed oriented manifold~$M$, there exists a unique harmonic 1-form~$\alpha_{\Sigma}$ such that
\begin{equation}
  I([\gamma],[\Sigma])=\int_{\gamma}^{}\alpha_{\Sigma}
  \label{eqn-def-1-form-inter-number}
\end{equation}
for all (smooth) loops~$\gamma$ in~$M$. By the construction of the cover~$\pi_{\infty}:M_{\infty}\lto M$, there is the relation
\begin{equation}
  (\pi_{\infty})_*(\pi_1(M_{\infty},x_0))=\ker \rho,
  \label{}
\end{equation}
upon picking a base point~$x_0\in M_{\infty}$ and interpreting~$\pi_1(M)$ as~$\pi_1(M,\pi_{\infty}(x_0))$. This translates into saying that
\begin{equation}
  \int_{(\pi_{\infty})_*\tilde{\gamma}}^{}\alpha_{\Sigma}=\int_{\tilde{\gamma}}^{} \pi_{\infty}^* \alpha_{\Sigma}=0
  \label{}
\end{equation}
for all loops~$\tilde{\gamma}$ in~$M_{\infty}$ based at~$x_0$. This means~$\pi_{\infty}^* \alpha_{\Sigma}\in \Omega^1(M_{\infty})$ is exact ($M_{\infty}$ is connected), and a primitive is given by
\begin{equation}
  \psi_{\Sigma}(x)\defeq \int_{x_0}^{x} \pi_{\infty}^* \alpha_{\Sigma}.
  \label{eqn-def-harmonic-primitive}
\end{equation}
Observe moreover that~$\psi_{\Sigma}$ is harmonic since~$\alpha_{\Sigma}$ is harmonic (thus co-closed). The most important property of~$\psi_{\Sigma}$ for us is the following which we single out as a lemma.

\begin{lemm}
  For~$\gamma\in \pi_1(M)$, use the same notation to denote the deck action of the class $\ol{\gamma}\in \pi_1(M)/\ker\rho$ on~$M_{\infty}$. Then we have
  \begin{equation}
    \psi_{\Sigma}(\gamma.x)=I([\gamma],[\Sigma])+\psi_{\Sigma}(x)=\rho(\gamma)+\psi_{\Sigma}(x),
    \label{eqn-deck-auto-primitive-potential}
  \end{equation}
  for any~$x\in M_{\infty}$.
\end{lemm}

\begin{proof}
  This comes from lifting loops, (\ref{eqn-def-harmonic-primitive}) and the property (\ref{eqn-def-1-form-inter-number}).
\end{proof}

Now we introduce several helpful functions. Let~$\theta\in [-\pi,\pi)$ (identified with the torus) be a parameter. Define the \textsf{character}
\begin{equation}
  \left.
  \begin{array}{rcl}
    \chi_{\theta}: \pi_1(M)&\lto & \mb{C},\\
    \gamma &\longmapsto & \me^{\ii \theta \int_{\gamma}^{}\alpha_{\Sigma}}= \me^{\ii \theta\cdot  \rho(\gamma)}.
  \end{array}
  \right.
  \label{}
\end{equation}
Next we define for $\theta\in [-\pi,\pi)$ the simple automorphic function~$G_{\theta}:M_{\infty}\lto \mb{C}$,
\begin{equation}
  G_{\theta}(x)\defeq\me^{\ii \theta\psi_{\Sigma}}.
  \label{}
\end{equation}
By (\ref{eqn-deck-auto-primitive-potential}) we find that
\begin{equation}
  G_{\theta}(\gamma.x)=\chi_{\theta}(\gamma)G_{\theta}(x),
  \label{}
\end{equation}
for deck shifts~$\gamma\in \pi_1(M)$,~$x\in M_{\infty}$. Now we restrict~$G_{\theta}$ to the fundamental domain~$\Omega$. Similarly as before we define
\begin{equation}
  C_{\theta}^{\infty}(\Omega)\defeq \left\{\tilde{f}|_{\Omega}~ \middle|~ \tilde{f}\in C^{\infty}(M_{\infty}),~\tilde{f}(\gamma.x)=\chi_{\theta}(\gamma)\tilde{f}(x)\textrm{ for }\gamma\in \pi_1(M) \right\}.
  \label{}
\end{equation}
Denote by~$\Delta_{\Omega,\theta}$ the self-adjoint Laplacian acting on~$L^2(\Omega)$ with core~$C_{\theta}^{\infty}(\Omega)$, sometimes called a \textsf{Born-von K\'arm\'an Laplacian} \cite{Lewin}. Note that~$C_0^{\infty}(\Omega)$ is identified with~$C^{\infty}(M)$, and trivially~$L^2(\Omega)$ with~$L^2(M)$. 

\begin{lemm}\label{lemm-twist-self-adjoint-domain}
  The Laplacian is indeed essentially self-adjoint on~$C_{\theta}^{\infty}(\Omega)$. The map~$G_{\theta}\mn{\cdot}:L^2(\Omega)\lto L^2(\Omega)$,~$f\mapsto G_{\theta}f$ is a unitary isomorphism. The operator~$\Delta_{\theta}$ on~$L^2(M)$ defined by~$\Delta_{\theta}f:=G_{\theta}^{-1}\Delta_{\Omega,\theta}(G_{\theta}f)$ is essentially self-adjoint on the core~$C^{\infty}(M)$ and unitarily equivalent to~$\Delta_{\Omega,\theta}$.
\end{lemm}

Note that if we identify~$M\setminus \Sigma$ with~$\Omega^{\circ}$, then~$G_{\theta}$ has monodromy crossing the cut~$\Sigma$. However, for~$f\in C^{\infty}(M)$ it follows from local computations with~$\Delta=-\ddiv \nabla$ that
\begin{equation}
  \Delta_{\theta}f=\Delta_M f -2\ii\theta \alpha_{\Sigma}( \nabla f)+\theta^2|\alpha_{\Sigma}|_g^2 f,
  \label{eqn-express-twist-lap}
\end{equation}
where~$\alpha_{\Sigma}=\dd \psi_{\Sigma}$ is the harmonic 1-form on~$M$ obtained in (\ref{eqn-def-1-form-inter-number}). Thus we do have~$\Delta_{\theta}f\in C^{\infty}(M)$.

\begin{proof}(Proof of lemma \ref{lemm-twist-self-adjoint-domain}.)
  We point out that~$\Delta_{\Omega,\theta}$ will have domain
  \begin{equation}
    \dom(\Delta_{\Omega,\theta})=\left\{ 
      f\in H^2(\Omega)~\middle| 
      \left.
      \begin{array}{l}
	f|_{\Sigma_+}=\me^{\ii\theta}f|_{\Sigma_-},\\
	(\partial_{\nu_+}f)|_{\Sigma_+}=\me^{\ii\theta}(\partial_{\nu_-}f)|_{\Sigma_-}
      \end{array}
      \right.
    \right\},
    \label{eqn-domain-boundary-cond-von-karman}
  \end{equation}
  where~$H^2(\Omega)$ denotes the second order Sobolev space (\cite{Lewin} page 85 for the torus case). We also see from (\ref{eqn-express-twist-lap}) that~$\dom(\Delta_{\theta})=H^2(M)$. Consequently it's clear that~$G_{\theta}\cdot \dom(\Delta_{\theta})=\dom(\Delta_{\Omega,\theta})$. The rest of the lemma follows.
\end{proof}

Combining lemmas \ref{lemm-unit-equiv-domain-bdy-cond} and \ref{lemm-twist-self-adjoint-domain}, we obtain
\begin{equation}
  \sigma(\Delta_N)=\bigcup_{p=0}^{N-1}\sigma(\Delta_{2\pi p/N})=
  \bigcup_{p=0}^{N-1}\sigma(\Delta_{p,N}).
  \label{eqn-main-spectral-decomp}
\end{equation}

Such operators already appeared in the works~\cite{PS, katsuda-sunada-90} whose purpose was to count geodesics with homological constraints, as well as in the more recent works~\cite[Appendix]{AnETDS}, \cite[Remark 1.10 p.~600]{AnGAFA} where $\Delta_{p,N}$ is called \emph{twisted Laplacians}~\cite[p.~599, 610]{AnGAFA}.

\subsubsection{Vanishing of Bottom of Spectrum}

%
%
%
%

\begin{lemm}\label{lemm:minimaslambda0}
  Let~$\Delta_{\theta}$ be the operator defined in lemma \ref{lemm-twist-self-adjoint-domain} with~$\theta\in\mb{R}$. Then $\ker(\Delta_\theta)\neq \{0\}$ iff $\theta\in 2\pi\mathbb{Z}$.
\end{lemm}

\begin{proof}
  If~$\theta=2\pi k$, then~$G_{\theta}^{-1}=\me^{-2\pi\ii k\psi_{\Sigma}}\in C^{\infty}(M)\subset \dom(\Delta_{\theta})$. One could verify directly with (\ref{eqn-express-twist-lap}) that in this case~$G_{\theta}^{-1}\in\ker(\Delta_{\theta})$. Conversely, suppose now~$\ker(\Delta_{\theta})\ne \{0\}$. Here it turns out more convenient to work with the original operators~$\Delta_{\Omega,\theta}$ on the fundamental domain~$\Omega$. Indeed, for general~$s\in H^2(\Omega)$ we have by Green-Stokes formula ($\ol{s}$ denoting the complex conjugate of $s$)
  \begin{equation}
    \int_{\Omega}^{}\ol{s}\Delta_{\Omega}s\,\dd V_g=\int_{\Omega}^{}\sank{\ol{\nabla s},\nabla s}_g\,\dd V_g -\int_{\Sigma_+}^{} \ol{s}(\partial_{\nu_+}s)\,\dd V_g -\int_{\Sigma_-}^{}\ol{s}(-\partial_{\nu_-} s)\,\dd V_g,
    \label{eqn-green-stokes-funda-domain}
  \end{equation}
  where we remember that~$\nu_+$ is \textit{outward} pointing whereas~$\nu_{-}$ is \textit{inward} pointing. In particular, for~$s$ satisfying the boundary conditions in (\ref{eqn-domain-boundary-cond-von-karman}), the boundary terms of (\ref{eqn-green-stokes-funda-domain}) vanish. This shows that, under \textit{these boundary conditions},~$s\in \ker(\Delta_{\Omega,\theta})$ iff~$\nabla s=0$, which implies~$s$ is a constant. However, looking at the boundary conditions again, this could only happen with $s\ne 0$ if~$\theta\in 2\pi\mb{Z}$, proving the result.
\end{proof}

From (\ref{eqn-green-stokes-funda-domain}) and the boundary conditions one can also see that~$\Delta_{\Omega,\theta}$ are all nonnegative. On the other hand, from the expression (\ref{eqn-express-twist-lap}) for~$\Delta_{\theta}$ we see that it is a second order elliptic differential operator which by the usual methods has compact resolvent and hence discrete spectrum. Thus we denote the spectrum of~$\Delta_{\theta}$ by~$0\le \lambda_0(\theta)\le \lambda_1(\theta)\le \cdots$, counting multiplicity. Lemma \ref{lemm:minimaslambda0} then translates into saying that~$\lambda_0(\theta)=0$ iff~$\theta\in 2\pi\mb{Z}$, and~$\lambda_0(\theta)>0$ otherwise.

\subsection{Kato Perturbations and Uniform Spectral Gap}

Our goal is to control the bottom of the spectrum of the sequence $\Delta_N$ uniformly in $N$ when $N$ gets large enough. We have followed the approach of Phillips-Sarnak~\cite{PS} and obtained a family of operators $\Delta_{\theta}$ indexed by the parameter $\theta\in [-\pi,\pi)$ on the ``Jacobian torus'', which are related to the spectrum of $\Delta_N$ by (\ref{eqn-main-spectral-decomp}).

We have also shown that lowest eigenvalue $\lambda_0(\theta)$ of $\Delta_{\theta}$ vanishes exactly when $\theta=0$ for $\theta\in [-\pi,\pi)$. In this subsection we will apply Kato's perturbation theory to see that $\lambda_0(\theta)$ in fact depends analytically on $\theta$ near~$\theta=0$ and $\lambda_0(0)=0$ is a \textbf{strict minimum}.

\subsubsection{The Family of Twisted Laplacians Parametrized by the Angle}

\begin{deef}
    An (unbounded) operator-valued function~$T(\beta)$ on a complex (resp.\ real) domain~$\mathcal{R}$ is called an \textsf{analytic family (in the sense of Kato)} if
\begin{enumerate}[(i)]
  \item for every~$\beta\in \mathcal{R}$,~$T(\beta)$ is closed and has non-empty resolvent set~$\mb{C}\setminus \sigma(T(\beta))$;
  \item for every~$\beta_0\in \mathcal{R}$ there exists~$z_0\in \mb{C}\setminus \sigma(T(\beta_0))$ and~$\delta>0$ such that~$z_0$ remains in the resolvent of~$T(\beta)$ for~$|\beta-\beta_0|<\delta$, and~$(T(\beta)-z_0)^{-1}$ is a (resp.\ real) analytic operator-valued function of~$\beta$ for~$|\beta-\beta_0|<\delta$.
\end{enumerate}
\end{deef}

\begin{lemm}
  The function~$\theta\mapsto \Delta_{\theta}$ is a real analytic family on~$\mb{R}$.
\end{lemm}

For every $s\in \mathbb{R}$, we will denote by $\Psi^s(M)$ the pseudodifferential operators of order $s$ acting on distributions on the manifold $M$.

\begin{proof}
    Essentially, we see from the expression (\ref{eqn-express-twist-lap}) that~$H^2(M)$ is a \textit{common domain} for all of~$\Delta_{\theta}$,~$\theta\in \mb{R}$. We could verify with that expression also that for~$u\in H^2(M)$,~$\Delta_{\theta}u$ is a~$L^2(M)$-valued (real) analytic function of~$\theta$. These ensure that~$(\Delta_{\theta})_{\theta}$ is an \textit{analytic family of type (A)} in the sense of \cite{RS4} page 16 and is consequently a Kato-analytic family (see \cite{RS4} and Kato \cite{Kato} pages 375-381). 
    
    Nevertheless, we indicate the argument for our case to make the paper more self-contained. We would like to prove, say, the real analyticity near a fixed $\theta_0\in\mb{R}$. Clearly for any $L>0$, $-L\notin \sigma(\Delta_\theta)$ for all $\theta$ since $\Delta_\theta:H^2(M)\lto L^2(M)$ are self-adjoint with nonnegative spectrum. By definition of the operators $\Delta_\theta$, we have an identity of the form
\begin{eqnarray*}
\Delta_\theta=\Delta_M-\ii \theta X+\theta^2 V\defeq \Delta_M+A_\theta
\end{eqnarray*}
where $X$ is a smooth vector field (differential operator of order $1$) and $V$ a smooth potential on $M$. We would have
\begin{align*}
(\Delta_\theta+L)^{-1}=(\one+(L+\Delta)^{-1}A_\theta)^{-1}(L+\Delta)^{-1}
\end{align*} 
provided $\Vert
(L+\Delta)^{-1}A_\theta \Vert_{L^2\to L^2}<1$. Indeed, by the composition theorem for pseudodifferential operators, $(L+\Delta)^{-1}\circ A_\theta$ is bounded in $\Psi^{-1}(M)$ uniformly in $\theta$ near $\theta_0$ and therefore as operators in $\mathcal{L}(L^2(M))$ since $\Psi^{-1}(M)\subset \Psi^0(M)\subset \mathcal{L}(L^2(M))$ by the Calder\'on-Vaillancourt theorem. Thus by choosing $L$ large enough (possibly depending on $\theta_0$), we can make sure that
$ \Vert
(L+\Delta)^{-1}A_\theta \Vert_{L^2\to L^2}\leqslant \frac{1}{2}$, uniformly in $\theta$ near $\theta_0$ and the result follows by explicit power series expansion of $(\Delta_{\theta}+L(\theta_0))^{-1}$ in $\theta$ near $\theta_0$.
\end{proof}

\subsubsection{Analytic Vanishing of Bottom of Spectrum}

We use the following two results which can be found in \cite{RS4} theorems XII.8 and XII.13.

\begin{prop}
  [Kato-Rellich, partial statement] \label{prop-kato-rellich} Let~$T(\beta)$ be an analytic family in the sense of Kato on a domain~$\mathcal{R}$. Pick~$\beta_0\in \mathcal{R}$ and let~$E(\beta_0)$ be an isolated \textsf{simple} eigenvalue of~$T(\beta_0)$. Then there exists~$\eta(\beta_0)>0$ such that for~$|\beta-\beta_0|<\eta(\beta_0)$, there is exactly one point~$E(\beta)$ of~$\sigma(T(\beta))$ near~$E(\beta_0)$ which is an isolated simple eigenvalue. Moreover~$E(\beta)$ is an analytic function of~$\beta$ in the neighborhood~$|\beta-\beta_0|<\eta(\beta_0)$.
\end{prop}

\begin{prop}
  [Kato-Rellich, finite multiplicity version] \label{prop-kato-rellich-fini-mult} Let~$T(\beta)$ be an real analytic family of \textsf{self-adjoint operators} in the sense of Kato on a domain~$\mathcal{R}\subset \mb{R}$. Pick~$\beta_0\in \mathcal{R}$ and let~$E(\beta_0)$ be an isolated eigenvalue of~$T(\beta_0)$ with finite multiplicity $m$. Then there exists~$\eta'(\beta_0)>0$ and $m$ not necessarily distinct (real) analytic functions $E^{(1)}(\beta)$, \dots, $E^{(m)}(\beta)$, with $E^{(k)}(0)=E_0$, such that for~$|\beta-\beta_0|<\eta'(\beta_0)$, $E^{(1)}(\beta)$, \dots, $E^{(m)}(\beta)$ are the only eigenvalues of $T(\beta)$ near $E_0$.
\end{prop}

\begin{lemm}\label{lemm:Kato}
Let $\Delta_\theta$ be the analytic family of operators defined in lemma \ref{lemm-twist-self-adjoint-domain}. There exists a threshold~$\varepsilon_0>0$ and an integer~$p>0$ such that
\begin{enumerate}[(i)]
  \item $[0,\varepsilon_0)$ contains no~$\lambda_i(\theta)$ for any~$i\ge 1$ and~$\theta\in[-\pi,\pi]$;
  \item there exists $a>b>0$ as well as~$\eta>0$ such that~$|\lambda_0(\theta)|=|\lambda_0(\theta)-\lambda_0(0)|<\varepsilon_0$ implies
  \begin{equation}
      b\theta^{2p}\le \lambda_0(\theta)=a\theta^{2p}+\mathcal{O}(\theta^{2p+1}), 
      \label{eqn-ground-energy-anal-local}
  \end{equation}
  as well as $|\theta|<\eta$, with $\lambda_0(0)=0$ being a strict minimum of $\lambda_0(\theta)$.
\end{enumerate}
\end{lemm}

\begin{proof}
  First of all we remark that~$\lambda_0(\theta)=\inf \sigma(\Delta_{\theta})$ is continuous in~$\theta$ since at every~$\theta_0$ we can choose~$L>0$ large enough so that~$-L$ is in the resolvent set of~$\Delta_{\theta}$ for all~$\theta$ near~$\theta_0$ and, as each $(\Delta_{\theta}+L)^{-1}$ is compact,~$(\lambda_0(\theta)+L)^{-1}=\sup \sigma((\Delta_{\theta}+L)^{-1})=\nrm{(\Delta_{\theta}+L)^{-1}}_{L^2\to L^2}$, the latter being continuous in~$\theta$ since $\Delta_{\theta}$ is a Kato-analytic family. This tells us that in fact
  \begin{equation}
    \inf_{\theta\in [-\pi,\pi]\setminus (-\delta,\delta)}\lambda_0(\theta)=\varepsilon_1>0
    \label{}
  \end{equation}
  for any~$0<\delta<\pi$ since otherwise it would contradict lemma \ref{lemm:minimaslambda0}. 
  Since we know that $\lambda_0(0)=0$ is a simple eigenvalue, each~$\Delta_{\theta}$ is nonnegative and we have lemma \ref{lemm:minimaslambda0}, we deduce by the Kato-Rellich theorem \ref{prop-kato-rellich} that in some neighborhood~$|\theta|<\eta_1$,~$\theta=0$ is a strict minimum of the function~$\theta\mapsto \lambda_0(\theta)$, which is analytic in the neighborhood.
  Therefore in this neighborhood $\lambda_0(\theta)$ will have the local form on the r.h.s.\ of (\ref{eqn-ground-energy-anal-local}). Consequently for any $b<a$ the inequality on the l.h.s.\ of (\ref{eqn-ground-energy-anal-local}) will be satisfied when $|\theta|<\eta_3$ for some $\eta_3<\eta_1$.
  Next, since we also know that~$\lambda_1(0)>0$ is an isolated eigenvalue of finite multiplicity, by the finite multiplicity version of Kato-Rellich \ref{prop-kato-rellich-fini-mult} we know that for some other~$\eta_2>0$,~$\lambda_1(\theta)$ is continuous in~$\theta$ for~$|\theta|<\eta_2$. Summing up,
   by necessarily choosing some
   \begin{equation}
       0<\varepsilon_0<\min\Big\{ \varepsilon_1, \inf_{|\theta|\le \eta_2}\lambda_1(\theta)\Big\},\quad \textrm{and}\quad
       \eta=\min\{\eta_2,\eta_3\},
   \end{equation}
   we will have (i) and (ii) satisfied for some~$p>0$ and we obtain the result.
\end{proof}

The spectral decomposition (\ref{eqn-main-spectral-decomp}) together with lemma \ref{lemm:Kato} tells us that for all~$N\in\mb{N}$, every~$\lambda\in [0,\varepsilon_0)\cap \sigma(\Delta_N)$ must be~$\lambda_0(2\pi k/N)$ for some integer~$k<\frac{1}{2}N$. This will give us the following important corollary.

\begin{corr}\label{cor-small-eigen-estima}
      Let~$\varepsilon_0$ and~$p$ be the numbers obtained in lemma \ref{lemm:Kato}. Then there exists~$C_4>0$ independent of~$N$, such that
  \begin{equation}
    \frac{1}{N}\sum_{\substack{\lambda\in \sigma(\Delta_{N}) \\ 0<\lambda<\varepsilon_0}} \me^{-t\lambda} \le C_4 t^{-1/2p}
    \label{}
  \end{equation}
  for all~$N\in\mb{N}$ and all $t>0$.
\end{corr}

 We point out the following formula related to the Gamma function:
  \begin{lemm}
    Let~$\alpha$,~$\beta>0$ then we have
    \begin{equation}
      \int_{0}^{\infty}\me^{-\beta x^{\alpha}}\,\dd x=\frac{1}{\alpha}\Gamma\Big( \frac{1}{\alpha}\Big) \beta^{-\frac{1}{\alpha}}.
      \label{}
    \end{equation}
  \end{lemm}

  \begin{proof}
      (Proof of corollary \ref{cor-small-eigen-estima}.) 
      By (ii) of lemma \ref{lemm:Kato} and the above formula, we have 
\begin{align*}
\frac{1}{N} \sum_{\substack{\lambda\in \sigma(\Delta_{N}) \\ 0<\lambda<\varepsilon_0}}  e^{-t \lambda}&\leqslant \frac{2\pi}{N} \sum_{0<2\pi|k|\leqslant \eta N}  e^{-t b(2\pi k/N)^{2p}}
\leqslant 2\int_{0}^\eta e^{-tb\theta^{2p}}\,\dd\theta \le \frac{1}{p}\Gamma\Big( \frac{1}{2p}\Big) (tb)^{-\frac{1}{2p}},
\end{align*}
the second inequality holds true since $e^{-u}$ is decreasing. We obtain the result.
  \end{proof}

\subsection{Rough Eigenvalue Counting}

The author sincerely acknowledge Junrong Yan for suggesting the following elegant estimate.

\begin{lemm}\label{lemm-rough-counting}
  Let~$\mathcal{M}$ be a smooth compact Riemannian manifold of dimension~$d$ with~$\mm{Ric}(\mathcal{M})\ge -K$,~$K\ge 0$. Let~$\Delta_{\mathcal{M}}$ be the Laplacian on~$\mathcal{M}$. Then for all~$\Lambda\ge 1$ we have
  \begin{equation}
    \sharp\{\lambda\le \Lambda|\lambda\in \sigma(\Delta_{\mathcal{M}})\} \le 
    C_6\me^{1+C_3 K/2}\vol(\mathcal{M}) \Lambda^{d/2},
    \label{}
  \end{equation}
  where~$C_6$ is a constant depending only on dimension, and~$C_3$ the same constant appearing in (\ref{eqn-li-yau-bound}).
\end{lemm}

\begin{proof}
  Indeed, we have,
  \begin{equation}
    \sharp\{\lambda\le \Lambda|\lambda\in \sigma(\Delta_{\mathcal{M}})\}  =\sum_{\substack{\lambda\in \sigma(\Delta_{\mathcal{M}}) \\ 0\le \lambda\le \Lambda}} 1 \le \sum_{\substack{\lambda\in \sigma(\Delta_{\mathcal{M}}) \\ 0\le \lambda\le \Lambda}} \me^{1-\frac{1}{\Lambda}\lambda}
    \le \me \cdot \ttr_{L^2(\mathcal{M})}(\me^{-\Delta_{\mathcal{M}}/\Lambda}).
  \end{equation}
  Thus the result follows by integrating the Li-Yau estimate (\ref{eqn-li-yau-bound}) choosing $\delta=1/2$.
\end{proof}

\newpage

\appendix

\chapter{Probability and Gaussian Measures}

In this appendix we collect some more background materials in probability theory and Gaussian measures. Let us start with the following remark.

Let~$\mathcal{X}$ be an infinite dimensional vector space equipped with a translation invariant metric~$d_{\mathcal{X}}$. In most practical situations, since~$\mathcal{X}$ is infinite dimensional, there exists a sequence~$\{x_n\}_{n=1}^{\infty}$ on the unit ball~$\mb{B}(0,1)$ such that~$d_{\mathcal{X}}(x_i,x_j)\ge \frac{1}{2}$ for any~$i\ne j$. This implies that~$\mb{B}(0,2)$ contains infinitely many disjoint balls of radii~$\frac{1}{6}$, namely~$\mb{B}(x_n,\frac{1}{6})$. Therefore, if we would like a translation invariant Borel measure~$\mu$ on~$\mathcal{X}$ such that~$\mu(\mb{B}(0,\frac{1}{6}))>0$, then~$\mu(\mb{B}(0,2))=\infty$! By a scaling, we see that if the balls were to have positive measure, then necessarily~$\mu(U)=\infty$ for all open set~$U$. So a \textit{working} Lebesgue measure cannot really exist in infinite dimensions! This creates the basic problem for expressions such as (\ref{eqn-1d-free-path-int-heu-kernel}). The substitute for a nicely behaved measure in infinite dimensions turns out to be a \textit{Gaussian probability measure}, which we have discussed in subsection \ref{sec-intro-general-gauss-meas} and continue to discuss in this appendix.

\section{$\sigma$-algebras and Conditioning}\label{sec-app-proba-background}

This appendix discuss some basic probability notion centered around the concept of \textit{conditioning}. We refer to Dudley \cite{Dudley} chapter 10 for more details.

A \textsf{probability space}~$(Q,\mathcal{O},\mb{P})$ consists of a set~$Q$, a family~$\mathcal{O}$ of subsets of~$Q$ called the \textsf{$\sigma$-algebra}, and a function~$\mb{P}:\mathcal{O}\lto [0,1]$ called the \textsf{probability measure}. The $\sigma$-algebra is required by definition, roughly speaking, to preserve the ``natural'' set operations (union, intersection, etc.) and hence could be thought of as an algebra made of subsets. A function~$X:Q\lto \mb{R}$ (or~$\mb{C}$) is called $\mathcal{O}$\textsf{-measurable} if~$X^{-1}(U)\in \mathcal{O}$ for all~$U$ open, and such is also called a \textsf{random variable} on $(Q,\mathcal{O},\mb{P})$. With~$X$, we could always let it \textsf{generate} a sub-$\sigma$-algebra~$\sigma(X)\subset \mathcal{O}$ which is the smallest one containing every~$X^{-1}(U)$. We give below a basic example illustrating the intuition behind the $\sigma$-algebra.

\begin{exxx}
    We drop a pin randomly onto the square~$[0,1]\times [0,1]$ and assume that it always hits a unique point. In this model we take~$Q=[0,1]\times [0,1]$, which is the ``collection of outcomes'' of the experiment, and let $\mathcal{O}$ be the $\sigma$-algebra generated by the open sets in $[0,1]^2$. In the usual circumstance where one has no trouble in observing precisely the outcome of the experiment, one can, in principle, ``measure'' or estimate the probability $\mb{P}(A)$ that the pin falls into any open set $A\subset[0,1]^2$ via a large number of trails.
    
    Now let us divide the square~$[0,1]^2$ into 4 parts~$[0,\frac{1}{2}]\times[0,\frac{1}{2}]$,~$[0,\frac{1}{2}]\times[\frac{1}{2},1]$,~$[\frac{1}{2},1]\times[0,\frac{1}{2}]$, and~$[\frac{1}{2},1]\times[\frac{1}{2},1]$, and label them by~$1$,~$2$,~$3$ and~$4$. Suppose now that somebody has placed this square at the bottom of a dark well so that we could no longer see ourselves the results of our pin-dropping experiment, but there is a robot at the bottom telling us which of the 4 smaller squares our pin has fallen into. The ``robot observer'' thus defines a random variable~$X_{\mm{Rob}}:Q\lto \{1,2,3,4\}\subset\mb{R}$. One could see in this case that the~$\sigma$-algebra~$\sigma(X_{\mm{Rob}})$ generated by~$X_{\mm{Rob}}$ would then consist only of unions of the 4 smaller squares as subsets of~$Q$, and nothing more. In a sense, therefore, a~$\sigma$-algebra represents a certain ``resolution'' at which one observes the experiment: a ``finer''~$\sigma$-algebra gives higher resolution, and a more ``coarse'' one, lower resolution. 
\end{exxx}

  Let~$(Q,\mathcal{O},\mb{P})$ be a probability space and let~$X\in L^1(Q,\mathcal{O},\mb{P})$, that is, assume~$X$ is integrable against~$\mb{P}$. Let~$\mathcal{A}\subset \mathcal{O}$ be a sub-$\sigma$-algebra. Then the \textsf{conditional expectation} of~$X$ with respect to~$\mathcal{A}$ is an~$\mathcal{A}$-measurable (integrable) random variable denoted~$\mb{E}[X|\mathcal{A}]$, defined uniquely by the property that
 \begin{equation}
   \mb{E}[XY]=\mb{E}[\mb{E}[X|\mathcal{A}]Y]
   \label{}
 \end{equation}
 for any other~$\mathcal{A}$-measurable bounded random variable~$Y$. For~$B\in \mathcal{O}$ we define~$\mb{P}(B|\mathcal{A}):=\mb{E}[1_B|\mathcal{A}]$.

 \begin{exxx}
   [continued] If~$X:[0,1]^2\lto\mb{R}$ is a random variable then one can show
   \begin{equation}
     \mb{E}[X|\sigma(X_{\mm{Rob}})]=\sum_{j=1}^4 \frac{\mb{E}[X 1_{A_j}]}{\mb{P}(A_j)} 1_{A_j},
     \label{}
   \end{equation}
   where~$A_j$,~$1\le j\le 4$, are the 4 smaller sub-squares constructed previously. Here~$\mb{E}[X 1_{A_j}]/\mb{P}(A_j)$ is in other words just the average value of~$X$ on~$A_j$.
 \end{exxx}

   \begin{lemm}[\cite{Dudley} page 338]\label{lemm-total-prob}
   Let~$(Q,\mathcal{O},\mb{P})$ be a probability space and~$\mathcal{A}_1\subset \mathcal{A}_2\subset \mathcal{O}$ two sub-$\sigma$-algebras. Then for any~$X\in L^1(Q,\mathcal{O},\mb{P})$ one has
   \begin{equation}
     \mb{E}[~\mb{E}[X|\mathcal{A}_2]~|\mathcal{A}_1]=\mb{E}[X|\mathcal{A}_1].
     \label{}
   \end{equation}
 \end{lemm}

The above lemma is sometimes called the \textsf{law of total probability}. The conditional expectations of $L^2$ random variables will be important for us in the context of the It\^o-Wiener-Segal theorem (proposition \ref{prop-ito-wiener-segal-iso}).
 
 \begin{lemm}
    Let~$(Q,\mathcal{O},\mb{P})$ be a probability space and~$\mathcal{A}\subset \mathcal{O}$ a sub-$\sigma$-algebra. Then~$L^2(\Omega,\mathcal{A},\mb{P})$ is a closed subspace of~$L^2(Q,\mathcal{O},\mb{P})$ and for~$X\in L^2(Q,\mathcal{O},\mb{P})$,~$\mb{E}[X~|~\mathcal{A}]$ is the orthogonal projection of~$X$ onto~$L^2(\Omega,\mathcal{A},\mb{P})$.
  \end{lemm}

 \begin{corr}[\cite{Janson} theorem 4.9]\label{cor5.7}
   Let~$(Q,\mathcal{O},\mb{P})$ be a probability space and~$\mathcal{H}$,~$\mathcal{K}\subset L^2(Q,\mathcal{O},\mb{P})$ two Gaussian Hilbert spaces. Denote by~$P_{\mathcal{K}|\mathcal{H}}$ the restriction of the orthogonal projection~$L^2(Q,\mathcal{O},\mb{P})\to \mathcal{K}$ to~$\mathcal{H}$, and by~$\mathcal{O}(\mathcal{H})$, $\mathcal{O}(\mathcal{K})$ the~$\sigma$-algebras generated by variables in~$\mathcal{H}$ and $\mathcal{K}$ respectively. Then
   \begin{equation}
     \left.
     \begin{array}{rcl}
       \Gamma(P_{\mathcal{K}|\mathcal{H}}):L^2(\Omega,\mathcal{O}(\mathcal{H}),\mb{P}) &\lto & L^2(\Omega,\mathcal{O}(\mathcal{K}),\mb{P}),\\
       X &\longmapsto & \mb{E}\big[ X\big|\mathcal{O}(\mathcal{K})\big],
     \end{array}
     \right.
     \label{}
   \end{equation}
   where~$\mb{E}\big[ X\big|\mathcal{O}(\mathcal{K})\big]$ is the conditional expectation of~$X$ with respect to~$\mathcal{O}(\mathcal{K})$.
 \end{corr}

 \section{Cameron-Martin Space}

 \begin{deef}\label{def3.7}
  Let~$\mathcal{X}$ be a real separable Fr\'echet (or Banach) space, and~$\gamma$ a Borel Gaussian measure on~$\mathcal{X}$. Then the closure of~$\mathcal{X}^*$ in~$L^2(\mathcal{X},\gamma)$ is called the \textsf{Gaussian Hilbert space of~$\gamma$}, denoted~$\mathcal{H}_{\gamma}$. 
\end{deef}

 Let~$\mathcal{X}$ be reflexive and equipped with a Borel Gaussian measure~$\gamma$. Starting from the Gaussian Hilbert space~$\mathcal{H}_{\gamma}$ we define the \textsf{covariance operator}
 \begin{equation}
   \left.
   \begin{array}{rcl}
     C: \mathcal{H}_{\gamma} &\lto & \mathcal{X}^{**}= \mathcal{X},\\
     f &\longmapsto & \big(C(f):h\longmapsto \mb{E}_{\gamma}[\phi(f)\phi(h)] \big).
   \end{array}
   \right.
   \label{eqn-def-cov-op-general}
 \end{equation}
 In other words~$C(f)$ is the linear functional~$\ank{f,-}_{L^2(\mathcal{X},\gamma)}$ on~$\mathcal{X}^*$. Then the image~$C(\mathcal{H}_{\gamma})\subset \mathcal{X}$, which is a Hilbert space with norm denoted~$\nrm{\cdot}_{C(\mathcal{H}_{\gamma})}$, is called the \textsf{Cameron-Martin space}. It does the following job. For two measures~$\mu$ and~$\nu$, we write ``$\mu\approx \nu$'' to mean~$\mu$ and~$\nu$ are \textsf{mutually absolutely continuous}, and ``$\mu\perp \nu$'' to mean~$\mu$ and~$\nu$ are mutually singular.

\begin{prop}
  [Cameron-Martin-Girsanov, \cite{Bogachev} corollary 2.4.3] \label{prop-cam-mar} Let~$\mathcal{X}$ be a real separable Fr\'echet space,~$\gamma$ a Gaussian measure on~$(\mathcal{X},\mathcal{B}_{\mathcal{X}})$, and~$C(\mathcal{H}_{\gamma})\subset \mathcal{X}$ the Cameron-Martin space. For any~$\phi_0\in \mathcal{X}$, define the measure~$(\phi_0)_* \gamma$ on~$(\mathcal{X},\mathcal{B}_{\mathcal{X}})$ by setting
  \begin{equation}
    (\phi_0)_* \gamma(B)\defeq\gamma(B-\phi_0),
    \label{}
  \end{equation}
  where~$B-\phi_0=\{\phi-\phi_0~|~\phi\in B\}$. Then
  \begin{equation}
    \left\{\def\arraystretch{1.2}
    \begin{array}{ll}
      (\phi_0)_* \gamma \approx \gamma,&\textrm{if }\phi_0\in C(\mathcal{H}_{\gamma}),\\
      (\phi_0)_* \gamma \perp \gamma,&\textrm{if }\phi_0\not\in C(\mathcal{H}_{\gamma}).
    \end{array}
    \right.
    \label{}
  \end{equation}
  Moreover, in the first case, one has the Radon-Nikodym density
  \begin{equation}
    \frac{\dd (\phi_0)_* \gamma}{\dd \gamma}(\phi)=\me^{(C^{-1}\phi_0)(\phi)-\frac{1}{2}\nrm{\phi_0}_{C(\mathcal{H}_{\gamma})}^2},
    \label{}
  \end{equation}
  where~$C$ is the operator defined in (\ref{eqn-def-cov-op-general}).
\end{prop}

  \section{More Wick Calculus and Feynman Graphs}
  
  This appendix is a sequel to subsection \ref{sec-intro-fock-wiener-chaos}. We list more Feynman rules, define Feynman diagrams and give an example of computation using Feynman diagrams.

\begin{lemm}
  [\cite{Sim2} propositions I.2, I.3, I.4, \cite{Janson} theorems 1.28, 3.9, 3.19] \label{lemm-wick-feyn} ~
  \begin{enumerate}[(i)]
  \item For~$X_1$, \dots,~$X_n\in \mathcal{H}$ (not necessarily distinct) jointly Gaussian random variables,
  \begin{equation}
    \mb{E}[X_1\cdots X_n]=\sum_{\mss{P}}\prod_k \mb{E}[X_{i_k}X_{j_k}],
    \label{}
  \end{equation}
  where the sum is over all partitions~$\mss{P}$ of the set~$\{1,\dots,n\}$ into disjoint pairs~$\{i_k,j_k\}$,~$1\le k\le n/2$. In particular, the result is zero if~$n$ is odd.
    \item For~$X\in \mathcal{H}$,
      \begin{equation}
	{:}X^n{:}=\sum_{j=0}^{\lfloor n/2 \rfloor} \frac{(-1)^j n!}{(n-2j)! j! 2^j}\mb{E}[X^2]^j X^{n-2j}=\mb{E}[X^2]^{\frac{n}{2}} h_n\big(X\big/\mb{E}[X^2]^{\frac{1}{2}}\big).
	\label{}
      \end{equation}
    \item Let~$X_1$, \dots,~$X_n\in \mathcal{H}$ and~$Y_1$, \dots,~$Y_m\in \mathcal{H}$ be jointly Gaussian random variables. Then
   \begin{equation}
     \mb{E}\big[{:}X_1\dots X_n{:}~{:}Y_1\dots Y_m{:}\big]=\left\{
       \begin{aligned}
	 &\sum_{\sigma\in\fk{S}_n}\prod_{i=1}^n \mb{E}[X_i Y_{\sigma(i)}], && m=n,\\
	 &~~~0, && m\ne n.
       \end{aligned}
     \right.
     \label{eqn-wick-theorem-exp}
   \end{equation}
   In particular,~$\mb{E}[{:}X^n{:}~{:}Y^m{:}]=\delta_{nm}n!\mb{E}[XY]^n$.
  \end{enumerate}
\end{lemm}

 There is a general way of associating random variables (or numbers) to \textsf{Feynman diagrams}. A \textsf{Feynman diagram} consists of vertices, legs (segments with only one end attached to a vertex), and edges (contracted legs). Two legs are \textsf{contracted} means they are connected to form an edge. A Feynman diagram is called \textsf{fully contracted} if there is no unconnected legs. Given~$n$ random variables~$X_1$, \dots,~$X_n$, a Feynman diagram \textsf{labelled by}~$(X_1,\dots,X_n)$ is simply any Feynman diagram whose vertices are in bijection with~$\{X_1,\dots,X_n\}$. A Feynman diagram labelled by~$(X_1,\dots,X_n)$ can now be associated with \textit{either} a random variable \textit{or} a number using the following rules:
\begin{enumerate}[(i)]
  \item for each leg attached to a vertex~$j$, write down the corresponding random variable~$X_j$, and multiply them all together;
  \item whenever two legs are contracted, enclose the corresponding two random variables by~$\mb{E}[\bullet]$.
\end{enumerate}
Thus fully contracted diagrams are always associated with numbers. If~$\gamma$ is a Feynman diagram labelled by~$X_1$, \dots,~$X_n$, denote by~$v(\gamma)$ the associated object following the above rules.

\begin{exxx}[\cite{PS} section 4.4]
  In a physical context Feynman diagrams are used in rather formal calculations. We would like to find
  \begin{equation}
    \mb{E}_{\mm{GFF}}^M\big[\phi(x)\phi(y)\me^{-\frac{\lambda}{4!}\int_{M}^{}\phi(x)^4\dd V_M(x)}\big]
    \label{eqn-feynman-calc}
  \end{equation}
  by Taylor expanding~$\exp(-\frac{\lambda}{4!}\int_{M}^{}\phi(x)^4\dd V_M(x))$. This gives
  \begin{align*}
    (\textrm{\ref{eqn-feynman-calc}})&\heueq \mb{E}_{\mm{GFF}}^M\big[\phi(x)\phi(y)\big] +
    \mb{E}_{\mm{GFF}}^M\Big[\phi(x)\phi(y)\Big( -\frac{\lambda}{4!} \Big)\int_{}^{}\phi(z)^4\dd V_M(z) \Big] \\
    &\quad\quad + \mb{E}_{\mm{GFF}}^M\Big[\phi(x)\phi(y)\frac{1}{2}\Big( -\frac{\lambda}{4!} \Big)^2\iint_{}^{}\phi(z)^4\phi(w)^4\dd V_M(z)\dd V_M(w) \Big]  +\cdots\\
    &=\mb{E}_{\mm{GFF}}^M\big[\phi(x)\phi(y)\big]  -\frac{\lambda}{4!}
   \underbrace{ \int_{}^{}\mb{E}_{\mm{GFF}}^M\big[\phi(x)\phi(y)\phi(z)^4\big]\dd V_M(z) }_{A} \\
    &\quad\quad + \frac{1}{2}\Big( -\frac{\lambda}{4!} \Big)^2 \iint_{}^{} \mb{E}_{\mm{GFF}}^M\big[\phi(x)\phi(y)\phi(z)^4\phi(w)^4\big]\dd V_M(z)\dd V_M(w) +\cdots
  \end{align*}
  Treating~$\phi(x)$,~$\phi(y)$,~$\phi(z)$ as (jointly Gaussian!) random variables, by (i) of lemma \ref{lemm-wick-feyn} and the above rules we write
  \begin{align*}
    A&\heueq 3\int_{}^{} \dd V_M(z) v\bigg( \raisebox{-.5\height}{\includegraphics[height=2cm]{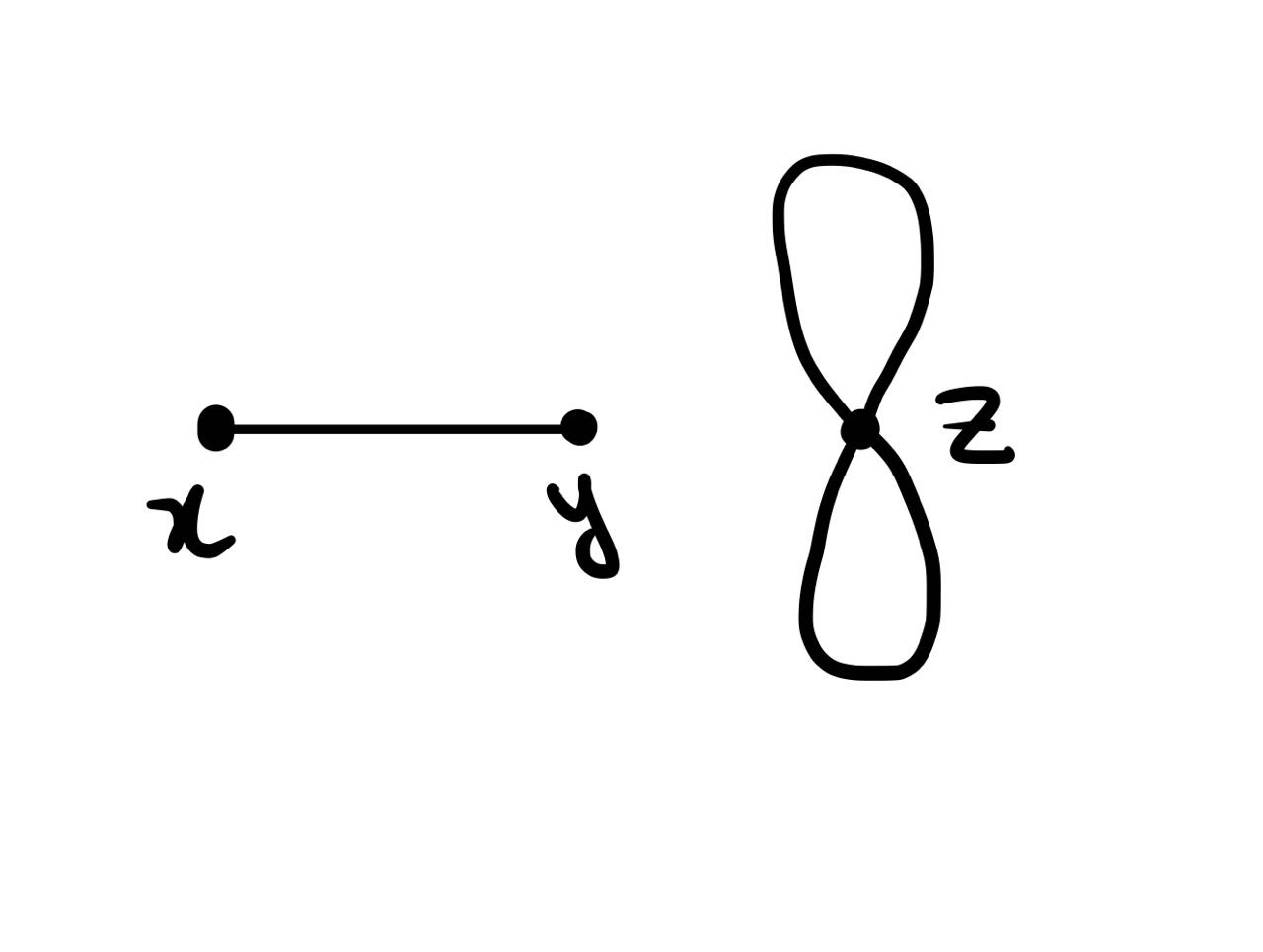}} \bigg) +12 \int_{}^{}\dd V_M(z) v\bigg( \raisebox{-.5\height}{\includegraphics[height=2cm]{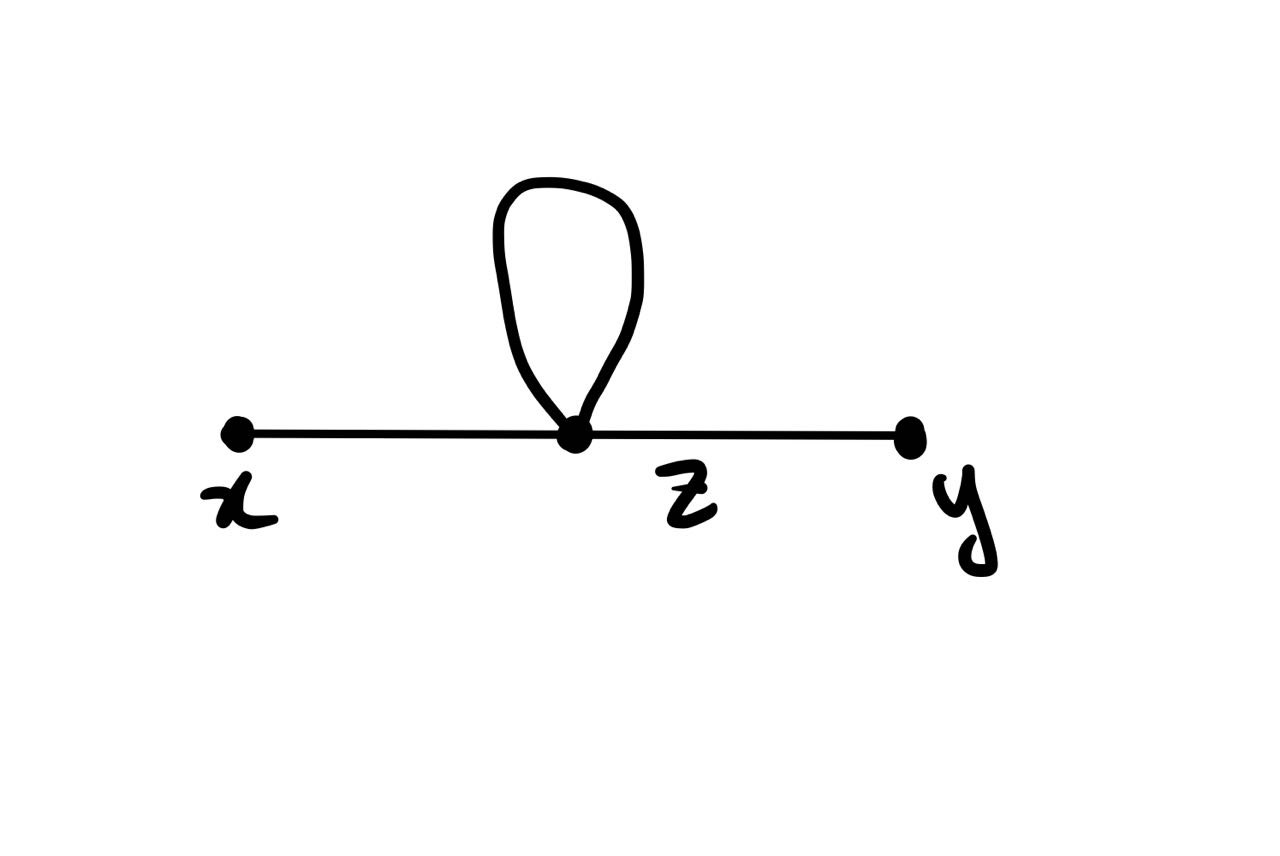}} \bigg) \\
    &=3\int_{}^{}G(x,y)G(z,z)^2 \dd V_M(z) +12 \int_{}^{} G(x,z)G(y,z)G(z,z)\dd V_M(z) ,
  \end{align*}
  where the factors correspond to the number of ways of getting the same contraction starting from 6 legs (4 on~$z$, 1 on~$x$,~$y$ each), and~$G$ denotes the Green function (of~$\Delta+m^2$). One can represent higher order terms using these diagrams in a similar manner. More than that, the diagrams also represent actual physical processes. See Peskin and Schroeder \cite{PS}.
\end{exxx}

\section{Remarks on $Q$-spaces for Random Fields}\label{app-q-space}
\noindent A~$Q$\textsf{-space} is in other words a probability sample space. The GFF on~$M$, say, can be seen (abstractly) as a way of associating a Gaussian random variable~$\phi(f)$ to each~$f\in W^{-1}(M)$ so that (\ref{eqn-gff-cov-closed}) holds. This says nothing about the sample space on which these random variables are actually defined, and naturally there exist many choices. For example,
\begin{enumerate}[(i)]
  \item the Bochner-Minlos construction, mentioned as proposition \ref{prop-boch-min};
  \item the formal Fourier series construction, mentioned in remark \ref{rem-mass-gau-field-four}; this can be identified with the previous one by appealing to the condition under which a formal Fourier series represents an actual distribution, see Shubin \cite{Shu} page 92 proposition 10.2;
  \item the \textsf{abstract Wiener space} construction. While this is not used essentially in this paper, it is a way of constructing (recovering) a separable Banach~$Q$-space~$\mathcal{X}$ starting from (knowing) the Cameron-Martin space, and taking closure with respect to a carefully defined norm weaker than~$\nrm{\cdot}_{C(\mathcal{H}_{\gamma})}$. See Sheffield \cite{Shef} and Bogachev \cite{Bogachev} section 3.9 for details.
\end{enumerate}
Some other models are discussed in Simon \cite{Sim2} section I.2, to which we refer for details in general. We shall discuss the question of in what sense two models of~$Q$-spaces are equivalent, though all the three examples above could eventually by realized in~$\mathcal{D}'(M)$. This is useful concerning the decomposition (\ref{eqn-stoc-decomp-gff-closed}), and we give the precise sense in which the original~$\mu_{\mm{GFF}}^M$ could be recovered from the decomposed measure~$\mu_{\DN}^{\Sigma,M}\otimes \mu_{\mm{GFF}}^{M\setminus\Sigma,D}$.

\begin{deef}[\cite{Sim2} page 4]
  Two probability measure spaces~$(Q,\mathcal{O},\mu)$ and~$(Q',\mathcal{O}',\mu')$ are called \textsf{isomorphic} if there is an isomorphism of \textsf{measure algebras}
  \begin{equation}
    T:\mathcal{O}/\mathcal{I}_{\mu}\overset{\sim}{\lto} \mathcal{O}'/\mathcal{I}_{\mu'},
    \label{}
  \end{equation}
  here~$\mathcal{I}_{\mu}$,~$\mathcal{I}_{\mu'}$ being the ideals of measure zero sets of~$\mu$ and~$\mu'$, such that~$\mu'(T(A))=\mu(A)$ for all~$A\in \mathcal{O}/\mathcal{I}_{\mu}$ (we do not distinguish an event~$A$ from its class in~$\mathcal{O}/\mathcal{I}_{\mu}$).
\end{deef}
\begin{def7}
  Mutually absolutely continuous measures~$\mu$ on~$\mathcal{O}$ will define the same measure algebra~$\mathcal{O}/\mathcal{I}_{\mu}$ as they have the same measure zero sets.
\end{def7}

\begin{deef}[\cite{Sim2} page 5]
  If~$(Q,\mathcal{O},\mu)$ and~$(Q',\mathcal{O}',\mu')$ are isomorphic under~$T$, then two random variables~$f:Q\lto\mb{R}$ and~$f':Q'\lto \mb{R}$ \textsf{correspond under the isomorphism} if
  \begin{equation}
    T(f^{-1}(B))=(f')^{-1}(B)
    \label{}
  \end{equation}
  for all Borel sets~$B\subset \mb{R}$.
\end{deef}

For a real Hilbert space~$\mathcal{H}$, a (centered) \textsf{Gaussian process indexed by}~$\mathcal{H}$ is a family of centered Gaussian random variables~$\{\phi(f)~|~f\in \mathcal{H}\}$ so that
\begin{equation}
  \mb{E}[\phi(f)\phi(h)]=\ank{f,h}_{\mathcal{H}}
  \label{}
\end{equation}
for all~$f$,~$h\in \mathcal{H}$. Such a process is \textsf{defined on} the probability space~$(Q,\mathcal{O},\mu)$ if each~$\phi(f)$ is a random variable from~$Q$ and~$\{\phi(f)~|~f\in \mathcal{H}\}$ generates~$\mathcal{O}$.

\begin{prop}[\cite{Sim2} theorem I.6]
  Let~$\{\phi(f)\}$ and~$\{\phi'(f)\}$ be two Gaussian processes indexed by~$\mathcal{H}$ defined respectively on~$(Q,\mathcal{O},\mu)$ and~$(Q',\mathcal{O}',\mu')$. Then there is an isomorphism between the two probability spaces so that~$\phi(f)$ corresponds to~$\phi'(f)$ under the isomorphism for each~$f\in \mathcal{H}$.
\end{prop}

\begin{proof}
  [Remark for proof.] Note for probability spaces~$L^2(Q,\mathcal{O},\mu)\supset L^{\infty}(Q,\mathcal{O},\mu)\supset \{\textrm{indicators}\}$. Similarly for $(Q',\mathcal{O}',\mu')$. Thus~$\Gamma(U)$ where~$U$ takes each~$\phi(f)$ to~$\phi'(f)$ gives the isomorphism (corollary \ref{cor-IWS-induced}).
\end{proof}

\begin{prop}[\cite{Sim2} proposition I.7] \label{prop-gauss-hilb-space-cons-q-space}
  Let~$\{\phi_1(f)\}$ and~$\{\phi_2(f)\}$ be Gaussian processes, respectively, indexed by~$\mathcal{H}_1$ and~$\mathcal{H}_2$, defined on~$(Q_1,\mathcal{O}_1,\mu_1)$ and~$(Q_2,\mathcal{O}_2,\mu_2)$. Then a Gaussian process~$\{\phi(h)\}$ indexed by~$\mathcal{H}=\mathcal{H}_1\oplus \mathcal{H}_2$ can be defined on~$Q:=Q_1\times Q_2$ equipped with~$\mathcal{O}_1\wh{\otimes} \mathcal{O}_2$ and~$\mu_1\otimes \mu_2$, by putting
  \begin{equation}
    \phi(f_1\oplus f_2)\defeq \phi_1(f_1)+\phi_2(f_2),
    \label{}
  \end{equation}
  for all~$f_1\in \mathcal{H}_1$,~$f_2\in \mathcal{H}_2$.
\end{prop}

\chapter{Functional Analysis}

 \section{Operator Ideals}

 Let~$A:\mathcal{H}\lto \mathcal{K}$ be a compact operator and~$A^*:\mathcal{K}\lto \mathcal{H}$ its adjoint. Then~$A^*A$ is a compact self-adjoint nonnegative operator whose spectrum consists of positive eigenvalues~$\mu_0(A)^2\ge\mu_1(A)^2\ge \cdots\ge 0$ with the only possible accumulation point being zero, which may or may not be in the spectrum or an eigenvalue. The numbers~$\mu_n(A)$ are called \textsf{singular values} of~$A$. Define
 \begin{align}
   \nrm{A}_{\mathcal{J}_p}&\defeq \Big( \sum_n \mu_n(A)^p \Big)^{1/p},&\quad \textrm{for }1\le p<\infty,\\
   \nrm{A}_{\mathcal{J}_{\infty}}&\defeq \nrm{A},&\quad\textrm{for }p=\infty.
 \end{align}
 Put~$\mathcal{J}_p(\mathcal{H},\mathcal{K}):=\{A\textrm{ compact}~|~\nrm{A}_{\mathcal{J}_p}<\infty\}$. We shall be concerned only with~$\mathcal{J}_1$, where the norm is also denoted~$\nrm{\cdot}_{\mm{tr}}$, called the \textsf{trace class} operators and~$\mathcal{J}_2$, called the \textsf{Hilbert-Schmidt} operators, where the norm is also written $\nrm{\cdot}_{\mm{HS}}$. When~$\mathcal{K}=\mathcal{H}$, define also
 \begin{equation}
   \ttr_{\mathcal{H}}(A)\defeq \sum_n\ank{e_n, A e_n}_{\mathcal{H}},
   \label{eqn-def-trace-H}
 \end{equation}
 where~$\{e_n\}$ is an orthonormal basis of~$\mathcal{H}$. Its convergence and (in)dependence on basis is discussed below. When~$\mathcal{H}=\mathcal{K}$ we use the notation~$\mathcal{J}_p(\mathcal{H})$.
 \begin{lemm}[\cite{Sim1} theorems 2.7, 2.8, 2.14, 3.1, \cite{GJ} page 132-133] \label{lemm-trace-ideals}
   We have
   \begin{enumerate}[(i)]
     \item whenever~$A:\mathcal{H}\lto \mathcal{H}$ is bounded and self-adjoint nonnegative, the sum (\ref{eqn-def-trace-H}) valuing in~$[0,+\infty]$ is independent of bases; it is finite iff~$A\in \mathcal{J}_1(\mathcal{H})$ in which case~$\ttr_{\mathcal{H}}(A)=\nrm{A}_{\mm{tr}}$;
     \item $A\in \mathcal{J}_p(\mathcal{H},\mathcal{K})$ iff~$(A^*A)^{p/2}\in \mathcal{J}_{1}(\mathcal{H})$, in which case~$\nrm{A}_{\mathcal{J}_p}=\ttr_{\mathcal{H}}((A^*A)^{p/2})^{1/p}$. In particular,~$A$ is Hilbert-Schmidt iff~$A^*A$ is trace class.
     \item Whenever~$A$,~$C\in \mathcal{L}(\mathcal{H})$ and~$B\in \mathcal{J}_p(\mathcal{H})$ we have~$\nrm{ABC}_{\mathcal{J}_p}\le \nrm{A}\nrm{C}\nrm{B}_{\mathcal{J}_p}$; thus each~$\mathcal{J}_p$ is a two-sided ideal in~$\mathcal{L}(\mathcal{H})$;
     \item if~$A\in \mathcal{J}_1(\mathcal{H})$ then the sum (\ref{eqn-def-trace-H}) converges absolutely, is independent of bases, and~$|\ttr_{\mathcal{H}}(A)|\le \nrm{A}_{\mm{tr}}$;
     \item whenever~$p^{-1}=q^{-1}+r^{-1}$ and~$A\in \mathcal{J}_q(\mathcal{H})$,~$B\in \mathcal{J}_r(\mathcal{H})$ then~$\nrm{AB}_{\mathcal{J}_p}\le \nrm{A}_{\mathcal{J}_q}\nrm{B}_{\mathcal{J}_r}$ and~$AB\in \mathcal{J}_p(\mathcal{H})$; in particular, the product of two Hilbert-Schmidt operators is trace class.
   \end{enumerate}
 \end{lemm}

 \section{Hilbert-Schmidt Operators on $L^2$}
  \begin{prop}
   [\cite{RSim} page 52] \label{prop-tensor-l2} Let~$(Q_1,\mu_1)$ and~$(Q_2,\mu_2)$ be measure spaces so that $L^2(Q_1,\mu_1)$, $L^2(Q_2,\mu_2)$ are separable. Then the map
   \begin{equation}
     \left.
     \begin{array}{rcl}
        L^2(Q_1,\mu_1)\times L^2(Q_2,\mu_2) &\lto & L^2(Q_1\times Q_2,\mu_1\otimes \mu_2),\\
	(f,g) &\longmapsto& fg,
     \end{array}
     \right.
     \label{}
   \end{equation}
   with~$(fg)(x,y):=f(x)g(y)$ extends to a unique isomorphism~$L^2(Q_1,\mu_1)\otimes L^2(Q_2,\mu_2) \cong L^2(Q_1\times Q_2,\mu_1\otimes \mu_2)$. 
 \end{prop}

 \begin{prop}
   [\cite{RSim} page 220, \cite{Sim3} theorem 3.8.4]
   Let~$\mathcal{H}$,~$\mathcal{K}$ be Hilbert spaces. Then the map
   \begin{equation}
     \left.
     \begin{array}{rcl}
       \mathcal{I}:\mathcal{H}\times \mathcal{K} &\lto & \mathcal{J}_2(\mathcal{H},\mathcal{K}),\\
       (v,w) &\longmapsto & \ank{v,-}_{\mathcal{H}}w
     \end{array}
     \right.
     \label{}
   \end{equation}
   extends to a unique isometric isomorphism~$\mathcal{H}\otimes \mathcal{K}\cong  \mathcal{J}_2(\mathcal{H},\mathcal{K})$. In particular, for every~$\eta\in \mathcal{H}\otimes \mathcal{K}$ there exist orthonormal sets~$\{\varphi_n\}\subset \mathcal{H}$,~$\{\psi_n\}\subset \mathcal{K}$ and real numbers~$\lambda_n>0$ so that
   \begin{equation}
     \eta=\sum_n \lambda_n(\varphi_n\otimes \psi_n),\quad\textrm{and}\quad \mathcal{I}(\eta)(v)=\sum_n \lambda_n\ank{\varphi_n,v}_{\mathcal{H}}\psi_n,
     \label{}
   \end{equation}
   and vice versa.
 \end{prop}

 \begin{corr}
   [\cite{Sim3} theorem 3.8.5]
   Let~$(Q_1,\mu_1)$ and~$(Q_2,\mu_2)$ be measure spaces so that~$L^2(Q_1,\mu_1)$, $L^2(Q_2,\mu_2)$ are separable. Then there is an isometric isomorphism
   \begin{equation}
     \left.
     \def\arraystretch{1.3}
     \begin{array}{rcl}
       \mathcal{I}:L^2(Q_1\times Q_2,\mu_1\otimes \mu_2) &\lto & \mathcal{J}_2(L^2(Q_1,\mu_1),L^2(Q_2,\mu_2)),\\
       K(x,y) &\longmapsto& \Big[f\longmapsto \ddp\int_{}^{}K(x,-)f(x)\dd\mu_1(x) \Big].
     \end{array}
     \right.
     \label{eqn-hilbert-schmidt-int-ker-rep}
   \end{equation}
   In particular, there exist orthonormal families~$\{\varphi_n\}\subset L^2(Q_1,\mu_1)$,~$\{\psi_n\}\subset L^2(Q_2,\mu_2)$ so that
   \begin{equation}
     K(x,y)=\sum_n \mu_n(\mathcal{I}(K)) \ol{\varphi_n(x)} \psi_n(y),
     \label{}
   \end{equation}
   with the latter series converging absolutely in~$L^2(Q_1,\mu_1)$ for almost every fixed~$y$, and~$\mu_n(\mathcal{I}(K))$ denote the singular values of~$\mathcal{I}(K)$.
 \end{corr}

  Now (ii) of lemma \ref{lemm-trace-ideals} implies that
 \begin{corr}\label{lemm-trace-int-diagonal-hs}
   Let~$L^2(Q_1,\mu_1)$,~$L^2(Q_2,\mu_2)$ be as above and~$A:L^2(Q_1,\mu_1)\lto L^2(Q_2,\mu_2)$ a Hilbert-Schmidt operator, represented via (\ref{eqn-hilbert-schmidt-int-ker-rep}) by the integral kernel~$K_A(x,y)$. Then
   \begin{equation}
     \nrm{A}_{\mm{HS}}^2=\ttr_{L^2(Q_1)}(A^*A)=\int_{}^{}|K_A(x,y)|^2 \dd\mu_1(x)\dd\mu_2(y). \quad\Box
     \label{}
   \end{equation}
 \end{corr}

 \begin{def7}\label{rem-trace-and-diagonal}
   For general~$A\in \mathcal{L}(L^2(Q))$ or~$\mathcal{J}_2(L^2(Q))$, the fact that~$\int_{}^{}|K_A(x,x)|\dd \mu(x)<\infty$ does not imply~$A$ is trace class, and nor does the trace equal $\int_{}^{}K_A(x,x)\dd \mu(x)$ when $A$ is trace class (since the diagonal has measure zero in $Q\times Q$, in reality one could let $K_A(x,x)$ be arbitrary without affecting $A$, thus an important condition is that $K_A$ be continuous in ``some sense''). See Simon \cite{Sim3} section 3.11, also Vershik, Petrov and Zatitskiy \cite{PVZ} section 3.3.
 \end{def7}

\section{Nuclear Spaces}\label{sec-app-nuclear-space}

We say that a topological vector space~$\mathcal{H}_{\infty}$ is a \textsf{countably Hilbert space} if it is a Fr\'echet space whose topology is generated by a countable family~$(\ank{\cdot,\cdot}_n)_{n=1}^{\infty}$ of \textit{inner products} (thus~$\ank{v,v}=0$ implies~$v=0$) whose norms~$\nrm{\cdot}_n:=\sqrt{\ank{\cdot,\cdot}_n}$ form an increasing sequence~$\nrm{\cdot}_1\le \nrm{\cdot}_2\le \cdots$. By completeness of~$\mathcal{H}_{\infty}$ in its topology it follows that
\begin{equation}
  \mathcal{H}_{\infty}=\bigcap_{n=1}^{\infty}\mathcal{H}_n,\quad\textrm{with}\quad \mathcal{H}_n\defeq \ol{\mathcal{H}}_{\infty}^{\nrm{\cdot}_n}.
  \label{}
\end{equation}
Each~$\mathcal{H}_n$ is a Hilbert space with~$\ank{\cdot,\cdot}_n$, and~$\mathcal{H}_1\supset \mathcal{H}_2\supset\cdots$. We denote by~$\iota_{n,m}:\mathcal{H}_n\lhook\joinrel\longrightarrow \mathcal{H}_m$ the natural inclusion for~$m\le n$. It follows from the definitions that~$\iota_{n,m}$ is continuous with respect to the~$\nrm{\cdot}_n$ and~$\nrm{\cdot}_m$ norms and has dense image. 

\begin{deef}
  A countably Hilbert space~$\mathcal{H}_{\infty}$ with inner products~$(\ank{\cdot,\cdot}_n)_{n=1}^{\infty}$ is called \textsf{nuclear} if for every~$m\in\mb{N}$ there exists~$n> m$ so that the injection~$\iota_{n,m}:\mathcal{H}_n\lhook\joinrel\longrightarrow \mathcal{H}_m$ is Hilbert-Schmidt.
\end{deef}
 
By definition of the topology of~$\mathcal{H}_{\infty}$, a bilinear form~$\mathcal{C}:\mathcal{H}_{\infty}\times \mathcal{H}_{\infty}\lto \mb{C}$ is continuous if and only if there exist~$C>0$ and~$m\in\mb{N}$ such that
\begin{equation}
  |\mathcal{C}(f,h)|\le C\nrm{f}_m \nrm{h}_m
  \label{}
\end{equation}
for all~$f$,~$h\in \mathcal{H}_{\infty}$. Denote by~$\mathcal{H}_{-\infty}:=\mathcal{H}_{\infty}'$ the topological dual of~$\mathcal{H}_{\infty}$. Countably Hilbert spaces are reflexive.

\chapter{Global Analysis}

\section{Sobolev Spaces over Domains}\label{sec-app-sobo}

\noindent In this paper we make essential use of the usual~$L^2$ Sobolev spaces over Riemannian manifolds. First let~$(M,g)$ be a closed Riemannian manifold and~$s\in \mb{R}$. Then the Sobolev space~$W^s(M)$ of order~$s$ is defined generally as the closure of~$C^{\infty}(M)$ under a norm~$\nrm{\cdot}_{W^s(M)}$, where the norm~$\nrm{\cdot}_{W^s(M)}$ could be defined in various equivalent ways. We refer to Taylor \cite{Taylor1} chapter 4 for a general discussion. For us,~$s=\pm 1$,~$\pm\frac{1}{2}$. We rely heavily on the following fact.
\begin{lemm}\label{lemm-sobo-inner-prod}
  Let~$\Lambda_{2s}$ be an elliptic \textsf{strictly positive} formally self-adjoint pseudodifferential operator on~$M$ with order~$2s$. Then the inner product
  \begin{equation}
    \ank{-,-}_{W^s}\defeq \ank{-,\Lambda_{2s}-}_{L^2}
    \label{eqn-sobo-inner-prod}
  \end{equation}
  induces an equivalent norm for~$W^s(M)$.
\end{lemm}

In particular, the real power~$(\Delta_M+m^2)^s$ of the Helmholtz operator (massive Laplacian)~$\Delta_M+m^2$ provides such a candidate for~$\Lambda_{2s}$. \textbf{Convention:} whenever we use the space~$W^s(M)$, the inner product (\ref{eqn-sobo-inner-prod}) with~$\Lambda_{2s}=(\Delta_M+m^2)^s$ is understood, unless otherwise specified.

\begin{def7}
  Various regimes of functional calculus can be used to define~$(\Delta_M+m^2)^s$. One of them is presented in section \ref{sec-det-def} which in fact defines \textit{complex} powers. We also mention a smooth functional calculus presented in Sogge \cite{Sogge} theorem 4.3.1.
\end{def7}

Next we discuss important subspaces of~$W^s(M)$. Let~$A\subset M$ be a closed set and~$U\subset M$ an open set. Define
\begin{align}
  W^s_A(M)&\defeq \{u\in W^s(M)~|~\supp u\subset A\textrm{ as a distribution}\},\label{eqn-def-sobo-closed-support}\\
  W^s_U(M) &\defeq \textrm{closure of }C_c^{\infty}(U)\textrm{ inside }W^s(M), \label{eqn-def-sobo-open-closure}\\
  W^s(U)&\defeq W^s_{M\setminus U}(M)^{\perp}\subset W^s(M).
  \label{eqn-sobo-open-set-ortho}
\end{align}
These are closed subspaces of~$W^s(M)$. 

\begin{def7}
  We point out right away that by definition, then,
  \begin{equation}
    W^s(U)\cong W^s(M)/W^s_{M\setminus U}(M),
    \label{}
  \end{equation}
  the latter equipped with the quotient norm, which is a more familiar characterization of~$W^s(U)$, see Taylor \cite{Taylor1} page 339. Our definition as in (\ref{eqn-sobo-open-set-ortho}) poses the obvious problem that in general~$C_c^{\infty}(U)\not\subset W^s(U)$, at least for~$s\not\in \mb{Z}_+$. We emphasize therefore that what is important in this definition is not the space~$W^s(U)$ \textit{per se} but the following choice for its inner product:
  \begin{equation}
    \bank{f,h}_{W^s(U)}\defeq \bank{P_{M\setminus U}^{\perp} f,P_{M\setminus U}^{\perp} h}_{W^s(M)},
    \label{eqn-def-inner-prod-sobo-open}
  \end{equation}
  for any~$f$,~$h\in W^s(M)$, in particular for~$f$,~$h\in C_c^{\infty}(U)$, which produces a norm equivalent to the quotient norm, where~$P_{M\setminus U}^{\perp}:W^s(M)\lto W^s(U)$ denotes the orthogonal projection.
\end{def7}

\begin{def7}
  Clearly~$W^s_U(M)\subset W^s_{\ol{U}}(M)$ by definition. In general the inclusion is strict (certainly if $U\ne (\ol{U})^{\circ}$!). See Taylor \cite{Taylor1} page 339 and section 4.7 for interesting discussions on conditions for~$s$ and~$U$ for which equality holds. In particular,~$W^k_{\Omega}(M)=W^k_{\Omega^{\circ}}(M)= W^k_{\ol{\Omega}}(M)$ if~$\Omega\subset M$ is a domain with smooth boundary~$\partial\Omega$ (a closed Riemannian manifold with one dimension less) and~$k\in \mb{Z}_+$. In this case, we use these notations interchangeably. 
\end{def7}

The rest of this appendix could be read along with section \ref{sec-sobo-decomp}. Let~$s=-1$. Although~$C_c^{\infty}(U)\not\subset W^{-1}(U)$, we have

\begin{lemm}\label{lemm-sobo-test-dense-in-open}
  Let~$U\subset M$ be an open set. Then~$P_{M\setminus U}^{\perp}(C_c^{\infty}(U))$ is dense in~$W^{-1}(U)$.
\end{lemm}

\begin{proof}
  We note~$\Delta_M+m^2$ is local and therefore~$(\Delta+m^2)(C_c^{\infty}(U))\subset C_c^{\infty}(U)$. It follows from lemma \ref{rem-sobo-proj=dist-res} and our definition of~$W_U^1(M)$ that~$(\Delta+m^2)(C_c^{\infty}(U))\subset W^{-1}_{M\setminus U}(M)^{\perp}$ and is dense there, proving the result.
\end{proof}

\begin{def7}\label{rem-natural-sobo-cam-mar}
  Clearly, the map~$P_{M\setminus U}^{\perp}$ is also injective on~$C_c^{\infty}(U)$; together with lemma \ref{lemm-sobo-test-dense-in-open} this shows~$P_{M\setminus U}^{\perp}$ is a good embedding of~$C_c^{\infty}(U)$ in~$W^{-1}(U)$. In fact, this is the same as the embedding of~$C_c^{\infty}(U)$ in~$\mathcal{D}'(U)$, by remark \ref{rem-sobo-proj=dist-res}. Nevertheless, the smaller class~$(\Delta+m^2)(C_c^{\infty}(U))$, as it is already dense in~$W^{-1}(U)$, suffices as a class of test functions to define the GFF with Dirichlet condition over a domain (see remark below lemma \ref{lemm-ghs-for-gff}). This reflects the fact that the Cameron-Martin pairing~$\ank{-,-}_{W^1}$ is more natural than~$\ank{-,-}_{L^2}$ in treating the GFF (see remark \ref{rem-use-cam-mar-pairing-not-l2}). We have stuck to~$\ank{-,-}_{L^2}$ only because this is more practical with functional analysis.
\end{def7}

\begin{def7}
For general~$s\in\mb{R}$, one could also define
  \begin{equation}
    W^s_0(U)\defeq \textrm{closure of }C_c^{\infty}(U)\textrm{ under (\ref{eqn-def-inner-prod-sobo-open})}.
    \label{}
  \end{equation}
  Then~$W^s_U(M)\subset W^s_0(U)$. But it cannot generally be compared with~$W^s_{\ol{U}}(M)$ (to the author's knowledge). See the exercises in Taylor \cite{Taylor1} pages 343-344 for more information.
\end{def7}

Next we state the duality results for the various spaces. Recall that~$\ank{-,-}_{L^2(M)}$ denotes both the inner product of~$L^2(M)$ and the distributional pairing between~$\mathcal{D}'(M)$ and~$C^{\infty}(M)$. Below, we extend it to denote also the pairing between dual Sobolev spaces (see (i) of the lemma below).
\begin{lemm}\label{lemm-sobo-dual}
  Let~$M$ be a closed Riemannian manifold,~$U\subset M$ an open set,~$A\subset M$ a closed set, and~$s\in\mb{R}$.
  \begin{enumerate}[(i)]
    \item $W^{-s}(M)$ is the dual Banach space, denoted~$W^s(M)^*$, of~$W^s(M)$ under~$\ank{-,-}_{L^2}$;
    \item the \textsf{annihilator} of~$W^s_U(M)$ under~$\ank{-,-}_{L^2}$ is~$W^{-s}_{M\setminus U}(M)$, that is,
      \begin{equation}
	W^{-s}_{M\setminus U}(M)=\{u\in W^{-s}(M)~|~\ank{u,f}_{L^2}=0\textrm{ for all }f\in W^s_U(M)\};
	\label{}
      \end{equation}
    the \textsf{annihilator} of~$W^s_A(M)$ is accordingly~$W^{-s}_{M\setminus A}(M)$;
  \item $W^s(U)^*\cong W^{-s}_U(M)$,~$W^s_U(M)^*\cong W^{-s}(U)$, these spaces being therefore reflexive.
  \end{enumerate}
\end{lemm}

Finally, when~$\Omega\subset M$ is a domain with smooth boundary~$\partial\Omega$, we define, in view of lemma \ref{lemm-sobo-decomp}, the \textsf{Dirichlet Green operator}~$(\Delta_{\Omega,D}+m^2)^{-1}:=(\Delta+m^2)^{-1}P_{M\setminus \Omega^{\circ}}^{\perp}:W^{-1}(\Omega^{\circ})\lto W_{\Omega^{\circ}}^1(M)$. Clearly this agrees with the usual definition. In terms of quadratic forms,
\begin{lemm}[\cite{Sim2} theorem VII.1] \label{lemm-diri-green-op-quad-form}
  Let~$\Omega\subset M$ be a domain with smooth boundary~$\partial\Omega$. We have
  \begin{equation}
    \bank{f,(\Delta_{\Omega,D}+m^2)^{-1}h}_{L^2}=\bank{P_{M\setminus \Omega^{\circ}}^{\perp}f, P_{M\setminus \Omega^{\circ}}^{\perp}h}_{W^{-1}},
    \label{}
  \end{equation}
  for~$f$,~$h\in C_c^{\infty}(\Omega^{\circ})$.
\end{lemm}

 \section{Symbol Convergence Lemma and Heat Kernel}\label{app-symbol}

 \begin{proof}[Proof of lemma \ref{lemm-dyat-zwor-symbol-conv}.]
  By coordinate invariance of the definition of~$\Psi^m(M)$ it suffices to pick~$x\in M$ and prove the result for a chart around~$x$ and~$\chi(x)=1$. Denote the kernel of~$\chi E_{\varepsilon}\chi$ by~$E_{\chi,\varepsilon}$ then in this chart we could write
  \begin{equation}
    E_{\chi,\varepsilon}(x,y)=\tilde{E}_{\chi,\varepsilon}(x,h)=\frac{1}{F_{\varepsilon}(x)}\tilde{\psi}\left( \frac{h}{\varepsilon} \right)\tilde{\chi}(h),
    \label{}
  \end{equation}
  where~$h=x-y$. Indeed, by definition of our function~$\psi$ and freedom of choosing~$\chi$ we could further assume that for small enough~$\varepsilon$ one has~$\tilde{\chi}(h)\equiv 1$ on the support of~$\tilde{\psi}(\cdot/\varepsilon)$. Thus under this condition
  \begin{equation}
    \sigma_{\chi E_{\varepsilon}\chi}(x,\xi)=\int_{\mb{R}^d}^{} \me^{-\ii h\cdot \xi}\frac{1}{F_{\varepsilon}(x)} \tilde{\psi}\left( \frac{h}{\varepsilon} \right) \dd h =\frac{\varepsilon^d}{F_{\varepsilon}(x)}\underbrace{F_1(x) \sigma_{\chi E_1\chi}(x,\varepsilon\xi)}_{\textrm{indep. of }x}.
    \label{eqn-reg-symbol-expression}
  \end{equation}
  Note that~$\sigma_{\chi E_1\chi}(x,\eta)$ is Schwartz in~$\eta$ and~$\sigma_{\chi E_1\chi}(x,0)=1$. On the other hand clearly~$\sigma_{\chi\one\chi}(x,\xi)\equiv 1$. Thus for some~$U'\subset U$ depending only on the chart and~$\chi$, one has
  \begin{equation}
    \sup_{x\in K\Subset U'}\sup_{\xi}\frac{|\sigma_{\chi (E_{\varepsilon}-\one)\chi}(x,\xi)|}{\ank{\xi}^{\delta}} \le \left\{
      \def\arraystretch{1.3}
    \begin{array}{l}
      \ddp C\sup_{x\in K\Subset U'}\sup_{|\xi|\le R} \ank{\xi}^{-\delta} |\sigma_{\chi E_{1}\chi}(x,\varepsilon\xi)-1| \le C_{K,\chi} \sqrt{\varepsilon},\\
      \ddp C\sup_{x\in K\Subset U'}\sup_{|\xi|\ge R} (\cdots)\le C_{K,\chi}\varepsilon^{\delta/2}\sup_{\eta}|\sigma_{\chi E_1\chi}(x,\eta)|,
    \end{array}
    \right.
    \label{eqn-reg-symbol-bound}
  \end{equation}
  with~$R=\varepsilon^{-1/2}$. Next we deal with derivatives. Note that by (\ref{eqn-reg-symbol-expression}) all the~$x$-derivatives fall on~$1/F_{\varepsilon}(x)$ and all~$\xi$-derivatives fall on~$\sigma_{\chi E_1\chi}(x,\varepsilon\xi)$. Indeed, one has~$|\partial_x^{\alpha}(1/F_{\varepsilon}(x))|\le C_{\alpha} \varepsilon^{-d}$ (see Dyatlov and Zworski \cite{DZ} page 28), and so when there are only~$x$-derivatives we obtain the same bounds as (\ref{eqn-reg-symbol-bound}) only with new constants depending on~$\alpha$. When there is at least one~$\xi$-derivative,
  \begin{equation}
    |\partial_{\xi}^{\beta}\left( \sigma_{\chi E_1\chi}(x,\varepsilon\xi) \right)|=|\varepsilon^{|\beta|}(\partial_{\xi}^{\beta} \sigma_{\chi E_1\chi})(x,\varepsilon\xi)|\le C_{\beta,K,\chi}\varepsilon^{|\beta|}\ank{\varepsilon\xi}^{\delta-|\beta|},\quad |\beta|\ge 1.
    \label{}
  \end{equation}
  Hence, on account of (\ref{eqn-reg-symbol-expression}) again,
  \begin{equation}
    \sup_{x\in K\Subset U'}\sup_{\xi}\frac{|\partial_x^{\alpha}\partial_{\xi}^{\beta}\sigma_{\chi (E_{\varepsilon}-\one)\chi}(x,\xi)|}{\ank{\xi}^{\delta-|\beta|}}\le C\varepsilon^d C_{\alpha}\varepsilon^{-d}C_{\beta,K,\chi}\varepsilon^{|\beta|}=C_{\alpha,\beta,K,\chi}\varepsilon^{|\beta|}.
    \label{}
  \end{equation}
  Consequently, all the~$\mathcal{S}^{\delta}_{1,0}$ seminorms of~$\sigma_{\chi(E_{\varepsilon}-\one)\chi}$ goes to zero as~$\varepsilon\to 0$. We obtain the result.
\end{proof}

In what follows we sum up some properties of the \textsf{heat operator}~$\me^{-t(\Delta+m^2)}$ of the massive Laplacian (Helmholtz operator) and its Schwartz kernel~$p_t(x,y)$ called the \textsf{heat kernel}.

 \begin{lemm}[\cite{BGV} theorems 2.30, 2.38 and pages 92-94] \label{lemm-heat} We have
    \begin{enumerate}[(i)]
      \item $p_t(x,y)\in C^{\infty}((0,\infty)\times M\times M)$;
      \item we have
  \begin{equation}
    (\Delta+m^2)^{-1}=\int_{0}^{\infty}\me^{-t(\Delta+m^2)}\dd t.
    \label{eqnB8}
  \end{equation}
  In particular, the kernel~$G_t$ of~$\me^{-t(\Delta+m^2)}(\Delta+m^2)^{-1}=(\Delta+m^2)^{-1}\me^{-t(\Delta+m^2)}$ is
  \begin{equation}
    G_{t}(x,y)=\int_{t}^{\infty} p_s(x,y)\dd s,
    \label{eqnB9}
  \end{equation}
  for~$x$,~$y\in M$.
\item Let~$\dim M=n$. There are asymptotic expansions
  \begin{align}
  p_t(x,y)&\sim \frac{1}{(4\pi t)^{n/2}}\me^{-\frac{1}{4t}d(x,y)^2}\sum_{i=0}^{\infty} f_i(x,y) t^i,\\
    \ttr_{L^2(M)}\big(\me^{-t(\Delta+m^2)}\big)&\sim \frac{1}{(4\pi t)^{n/2}}\sum_{i=0}^{\infty} a_i t^i
    \label{}
  \end{align}
  as~$t\to 0+$, for some real numbers~$a_i$ and functions $f_i\in C^{\infty}(M\times M)$,~$i=0$,~$1$,~$2$, \dots.
\item For~$t$ large and each~$\ell\in\mb{N}$,
  \begin{equation}
    \nrm{p_t(x,y)}_{C^{\ell}}\le C_{\ell}\me^{-tm^2/2}
    \label{}
  \end{equation}
  for some constant~$C_{\ell}$.
    \end{enumerate}
\end{lemm}

\section{Poincar\'e-Lelong Lemmas}\label{sec-app-poin-lelong}

\begin{proof}[Proof of lemma \ref{lemm-bounded-lap-log}.]
	  We denote by~$\nabla_v$ the Levi-Civita covariant derivative in the direction~$v$. Since we are concerned with only one metric~$g$ we omit it from the notations. Take the geodesic polar coordinates $(r, \theta)$ around $z_0$. Then
 $g = \dd r^2 + \omega^2(r) \dd\theta^2$ where $\omega(r) = | \partial_{\theta}|$. From the expression of $\Delta$ in polar coordinates we have
	  \begin{align*}
	    \Delta \log r =  \frac{1}{\omega(r)} \partial_r \left (\omega(r) \partial_r \log r \right) = -\frac{1}{r^2} + \frac{\partial_r \omega(r)}{r\, \omega(r)}.
	  \end{align*}
	  Either noting the fact that~$\partial_{\theta}$ is a Jacobi field along $\theta=\textrm{const}$ connecting~$z_0$ to~$z$ (\cite{dC} page 115), or otherwise, we have a Taylor expansion
	  \begin{equation}
	    |\partial_{\theta}|=r+\mathcal{O}(r^3),\quad r\to 0.
	    \label{}
	  \end{equation}
	  This gives
	  \begin{equation}
	    \frac{\partial_r|\partial_{\theta}|}{|\partial_{\theta}|}=\frac{1}{r}+\mathcal{O}(r),\quad r\to 0,
	    \label{}
	  \end{equation}
	  yielding the result.
	\end{proof}

\begin{lemm}\label{lemm-app-poincare-lel}
  Let~$(\Sigma,g)$ be a Riemannian surface with smooth metric~$g$, and fix~$z_0\in \Sigma$. Denote~$r(z):=d_g(z,z_0)$. Then for any~$\psi\in C_c^{\infty}(U)$ where the neighborhood~$U$ lies within the injectivity radius of~$z_0$,
  \begin{equation}
    \int_{\Sigma}^{}\log r (-\Delta_g \psi)\,\dd V_g=\int_{\Sigma\setminus z_0}^{}\psi (-\Delta_g \log r)\,\dd V_g -2\pi \psi(z_0).
    \label{}
  \end{equation}
\end{lemm}

\begin{proof}
  Pick the metric disk~$B_{\varepsilon}(z_0)$ with radius~$\varepsilon>0$ at~$z_0$. Then Green's formula gives
  \begin{align*}
    \int_{\Sigma\setminus B_{\varepsilon}(z_0)}^{}\big[\log r (\Delta_g \psi)-\psi (\Delta_g \log r)\big]\,\dd V_g&= -\int_{\partial B_{\varepsilon}}^{} \big[\log r (\partial_r \psi)-\psi (\partial_r \log r)\big]\,\dd \ell_g\\
    &=\mathcal{O}(\varepsilon\log \varepsilon)+\psi(z_0)\frac{1}{\varepsilon}2\pi \varepsilon +\mathcal{O}(\varepsilon),
  \end{align*}
  as~$\varepsilon\to 0$, since~$\psi$ is smooth and (\ref{eqn-asymp-peri-circle}).
\end{proof}

\begin{def7}
  More precisely we have the equality of currents~$\Delta_g \log r\,\dd V_g=\Delta_g \log r|_{U\setminus z_0}\dd V_g+2\pi \delta_{z_0}$ on~$U$.
\end{def7}

\begin{lemm}\label{lemm-bound-lap-log-ratio}
  Fix~$r(z):=d_g(z,z_0)$ and now let~$g_1=\me^{2h}g$,~$h\in C^{\infty}(\Sigma)$, and~$r_1(z):=d_{g_1}(z,z_0)$. Then
  \begin{equation}
    \Delta_{g_1}\log r -\Delta_{g_1}\log r_1 \in L^1(U,g_1).
    \label{}
  \end{equation}
  Namely, as a distribution, it agrees with an integrable function when paired with~$\dd V_{g_1}$.
\end{lemm}

\begin{proof}
  The point is that the current~$\Delta_{g_1}\log r\,\dd V_{g_1}$ agrees with~$\Delta_g \log r\,\dd V_g$, because of their definitions and conformal covariances (\ref{eqn-scale-vol}), (\ref{eqn-scale-lap}). Therefore the delta functions cancel out exactly and we are left with the regular parts~$\Delta_{g_1}\log r|_{U\setminus z_0} -\Delta_{g_1}\log r_1|_{U\setminus z_0}$ which is in fact~$L^{\infty}(U,g_1)$.
\end{proof}

\newpage

\printbibliography[heading=bibintoc, title={References}]

\end{document}